\documentclass[pdflatex,sn-mathphys-num]{sn-jnl}

\usepackage{ragged2e}
\usepackage{graphicx}
\usepackage{multirow}
\usepackage{mathtools}
\usepackage{amsmath,amssymb,amsfonts}
\usepackage{amsthm}
\usepackage{mathrsfs}
\usepackage[title]{appendix}
\usepackage{xcolor}
\usepackage{float}
\usepackage{textcomp}
\usepackage{manyfoot}
\usepackage{booktabs}
\usepackage{algpseudocode}
\usepackage{listings}
\usepackage{booktabs}  
\usepackage[linesnumbered, ruled, vlined]{algorithm2e}
\usepackage{caption}
\usepackage{subcaption}
\usepackage{tcolorbox}
\usepackage{tabularx}
\usepackage{adjustbox}
\usepackage{rotating}
\usepackage{lmodern}
\usepackage{lipsum} 
\usepackage[utf8]{inputenc}
\DeclareUnicodeCharacter{2212}{-}
\newtheorem{mytheorem}{Theorem}
\newtheorem{mylemma}{Lemma}[subsection]
\newtheorem{mydef}{Definition}

\begin{document}


\title[Article Title]{Profit Maximization for Viral Marketing in Online Social Networks using Two Phase Diffusion Approach}


\author{\fnm{Poonam} \sur{Sharma}}\email{poonam.sharma@iitjammu.ac.in}

\author*{\fnm{Suman} \sur{Banerjee}}\email{suman.banerjee@iitjammu.ac.in}

\equalcont{These authors contributed equally to this work.}





\affil{\orgdiv{Department of Computer Science and Engineering}, \orgname{Indian Institute of Technology Jammu}, \orgaddress{\postcode{181221}, \state{Jammu and Kashmir}, \country{India}}}


\abstract{Now-a-days, \emph{Online Social Networks} (OSNs) are extensively used by different commercial houses for viral marketing. The key problem that arises in this context is to choose a limited number of highly influential users as the initial adopters of a brand such that the influence regarding the brand in the network gets maximized. Deviating from this standard setting, in this paper, we study the problem where every user of the network is associated with a selection cost and a benefit value. This benefit value can be earned from the user if (s)he is influenced by the brand. A fixed amount of budget is allocated for selecting the seed users. The goal of initial adopters is to choose a set of seed users within the budget such that the profit is maximized. We propose a two phase diffusion model for this problem where the goal is to split the diffusion process into two phases, and hence, split the budget into two halves. First, we spend the first half budget to select seed users for the first phase and observe the diffusion for a few rounds and then deploy the seed users for the second phase and successively complete the diffusion process. We prove several properties of the two phase influence function. Three solution approaches have been proposed for our problem with detailed analysis and illustrative examples. We conduct a number of experiments with three real\mbox{-}world social network datasets. From the experiments, we observe that the two phase diffusion approach leads to more amount of profit compared to the single-phase diffusion. In particular, for most instances, this improvement is greater than $18 \%$ and reaching as high as $40\%$ by the proposed methodologies.}




\maketitle

	\section{Introduction} \label{sec:IntroSec}
	Online Social Networks are interconnected structures among a group of users, where two users are directly connected if they are friends or follower-followee to each other \cite{carrington2005models}. Now-a-days, due to the advancement of wireless internet, the use of online social networks by people has increased significantly. As per recent statistics, the active social media users were $970$ million in $2010$ and has increased to $ 4.48$ billion in July $2021$.
    E-Commerce houses started using online social media platforms for viral marketing and promotional advertisement for their brands among the people, and this process goes as follows \cite{domingos2005mining}. Initially, E-Commerce houses choose a limited number of users from the social networks of their customers by providing some amount of incentive (e.g., a certain percent discount, free home delivery etc.) with the hope that they will share positive feedback about the brand among the neighbors. Some of the neighbors will adopt the brand and share the feedback further. This way, the diffusion process will continue, and the brand will go viral among customers. Thus, a brand can be made viral using online social media platforms.  
	\paragraph{\textbf{Background}} Our problem relies on a natural phenomenon of social networks called \emph{diffusion of information} \cite{guille2013information}, which means that users of a social network tend to share their information with their neighbors. By this mechanism, information propagates from one node of the network to the other. Significant efforts have been made to understand the diffusion phenomenon in a social network. Several models have been proposed to study the diffusion process in a social network, and among them, the \emph{Independent Cascade Model} (henceforth mentioned as \emph{IC Model}) has been predominantly used in the literature \cite{wang2012scalable}.
	\par One important problem studied in the domain of social network analysis is the \emph{Problem of Social Influence Maximization}. Given a social network and a positive integer $k$ this problem asks to choose a subset of $k$ nodes as initial adopters such that the influence in the network gets maximized. Domingos and Richardson \cite{domingos2001mining,richardson2002mining} were the first to propose the problem in the context of viral marketing. However, Kempe et al. \cite{kempe2003maximizing} were the first to study this problem computationally and showed that this problem is \textsf{NP-hard} under the IC Model of diffusion. They also showed that the influence function under the diffusion of IC Model is non-negative, monotone, and submodular in nature. Hence, the incremental greedy approach based on marginal influence gain leads to $(1-\frac{1}{e})$-factor approximation algorithm \cite{nemhauser1978analysis}. Later, there are several studies on this problem in different settings, considering the presence of negative influence \cite{chen2011influence}, in competitive situations \cite{hong2020efficient}, in unknown social network \cite{kamarthi2019influence} and many more. Also, this problem has been studied in different variants such as Target Set Selection Problem, Budgeted Influence Maximization Problem (BIM), and many more. Look at \cite{banerjee2020survey,li2018influence} for a recent survey on this topic.
	
	\par One of the variants of the influence maximization problem relevant to our problem is the Budgeted Influence Maximization Problem. Nguyen et al. \cite{nguyen2013budgeted} were the first to introduce this variant, though later, there were several studies on this problem \cite{bian2020efficient,banerjee2019combim}. Along with the social network, this problem considers that every user of the network is associated with a selection cost. A fixed amount of budget has been allocated. The goal here is to choose a subset of the users as seed nodes within the budget to maximize the influence in the network. However, it is important to observe that in commercial campaigns, the goal is to maximize the profit, not just the influence. There are a few studies in the literature that considers maximization of profit instead of influence \cite{tang2017profit,zhang2016profit}. However, compared to the influence maximization problem, available literature for the profit maximization problem is much less. Next, we describe the motivation behind our study.
	
	\paragraph{\textbf{Motivation}} As mentioned previously, now-a-days, online social networks have been heavily used for viral marketing purposes. It has been mentioned in a study that $55 \%$ buyers do research via social media before buying a product.
    Also, even a small business house spends $7-8 \%$ of their gross revenue for marketing and advertising.
    Hence, selecting the high-quality seed nodes is the root in the context of viral marketing. Solving this problem (or at least making progress in this direction) brings benefits for the E-Commerce houses, and using the proposed solution approaches may be adopted by them for effective advertisement policies.
	\par One important point to note here is that the existing studies on the profit maximization problem consider that the selection of seed nodes has to be done in one go. That means, based on the allocated budget, all the seed nodes are selected initially, and the diffusion process starts subsequently. However, there are a few recent studies in the literature that shows instead of deploying all the seed nodes in one go, the budget is divided into more than one part, and accordingly, the diffusion phases are also divided into that many rounds \cite{dhamal2016information,sun2018multi,cautis2019adaptive}. This setting is called \emph{Multi-Phase Diffusion Approach}. Experimental results in these studies show that the multi-phase diffusion approach leads to more influenced nodes compared to the traditional single-phase diffusion. This triggers a natural question of whether a similar kind of seed selection strategy increases profit for viral marketing. In this paper, we do a detailed study to investigate this question. Next, we describe  our contributions.
	
	\paragraph{\textbf{Our Contributions}}
	
	In this paper, we study the \emph{Profit Maximization Problem} under the two phase diffusion approach, and to the best of our knowledge, this is the first study in this direction. The key contributions of this paper are as follows:
	\begin{itemize}
		\item  We formulate a mathematical model for the profit maximization problem under the two phase diffusion approach, analyze the proposed model, and describe it with an illustrated example.
		\item We propose three solution approaches, namely the simple greedy, the double greedy, and the stochastic greedy approach with detailed analysis.
		\item We implement the proposed solution methodologies with three publicly available social network datasets, and the results are reported.
		\item From the performed experiments, we resolve a number of research questions with detailed analysis and explanation.
	\end{itemize}
	
	\paragraph{\textbf{Organization of the Paper}} The rest of the paper is organized as follows. Section \ref{Sec:RW} describes some relevant studies from the literature. In Section \ref{Sec:PPD}, we describe some preliminary concepts and define our problem formally. Section \ref{Sec:PM} contains the formulated mathematical model of our problem along with its analysis and illustration with an example. Section \ref{Sec:Solution} contains the proposed solution methodologies for our problem along with their analysis, theoretical results, and examples. Section \ref{Sec:Experiments} contains the experimental evaluation of the proposed methodologies. Finally, Section \ref{Sec:Conclusions} concludes our study and gives future research directions.

	\section{Related Work} \label{Sec:RW}
	In this section, we present some relevant studies from the literature. Our study comes under the broad theme of Influence Maximization in Social Networks, and in particular, profit or benefit or revenue maximization using Online Social Networks. In the following two subsections, we describe them.
	\subsection{Influence Maximization in Social Network}
	The problem of influence maximization aims to identifying a limited number of highly influential users in a given social network such that the number of influenced nodes is maximized. Damingoes and Rechardson \cite{domingos2001mining,richardson2002mining} were the first to introduce the problem in the context of viral marketing. However, Kempe et al. studied the problem in a computational setting and showed the problem is \textsf{NP-hard} under the IC Model of diffusion. They also showed that the influence function arises under the IC Model of diffusion, and hence the incremental greedy approach based on the marginal influence gain leads to $(1-\frac{1}{e})$-factor approximation algorithm. This problem arises in many real-life situations, such as viral marketing, computational advertisement, feed ranking, influential tweet selection, and many more. Hence, the study done by Kempe et al. triggers a significant amount of research, and a vast number of solution methodologies are available. We can classify the solution methodologies into different groups such as Approximation Algorithms (e.g., CELF, CELF++ \cite{goyal2011celf++} and many more), Sketch-based Approaches (e.g., IMM, TIM \cite{tang2014influence} and many more), heuristic solutions (e.g., IRIE \cite{jung2012irie}, SIMPATH \cite{goyal2011simpath} and many more), different metaheuristic algorithms (e.g., Genetic Algorithm, Particle Swarm Optimization Algorithm and many more). Readers are requested to look into the surveys \cite{banerjee2020survey} on this problem.
	\par The problem of influence maximization has also been studied in different variants and moved into different directions. The first one is the \emph{Budgeted Influence Maximization Problem}. In this problem, the users of the network are associated with a selection cost, and a fixed amount of budget is allocated. The goal is to choose a seed set within the allocated budget to maximize the influence. Another one is the Target Set Selection Problem, which is a deterministic variant of the influence maximization problem. In this problem, every node is mapped with a positive integer, including zero, and a node will change its state from `inactive' to `active' if the threshold many neighboring nodes are already in active state. The goal of this problem is to choose a minimum cardinality seed set such that at the end of the diffusion process, all the nodes will be in the active state. There are several studies on this problem \cite{chiang2013target,chiang2013some,chopin2014constant}. Also, there are several other variants studied in the recent past, such as influence maximization with feedback \cite{tang2020influence}, influence maximization in unknown social network \cite{kamarthi2019influence}, influence maximization under competitive situation \cite{xie2021competitive}, influence maximization with multi-boosting stages \cite{9804846}, comparative influence maximization with adaptive policy \cite{11214355}, and many more.
	
	\subsection{Profit/Benefit/Revenue Maximization using Online Social Networks}
	As mentioned previously, in the context of viral marketing, profit is the more sought parameter rather than just influence. Hence, in recent past decade, there have been several studies on profit maximization and related problems. To the best of our knowledge, Lu and Lakshmanan \cite{lu2012profit} were the first to study and introduce the profit maximization problem. They extended the classical linear threshold model, considering the price and valuation of the product for decision-making. They proposed three different algorithms for solving this problem. Tang et al. \cite{tang2018towards} proposed a two phase framework for this problem and showed that it is a challenging problem as the function is neither monotone nor submodular. Gao et al. \cite{gao2021adaptive} studied the profit maximization problem in an adaptive setting where the seed nodes are chosen one by one. This problem leads to the profit function, which is non-adaptive submodular. Subsequently, they developed a solution methodology that provides a data-dependent approximation guarantee. Shi et al. \cite{shi2021profit} studied the profit maximization problem under competitive social advertising. In their study, they consider that a social media host runs commercial campaigns for different competitive clients, and each client comes up with different budgets and different spread requirements. Now, the task of the social media host is to allocate the seed nodes such that the respective spread requirements can be met, and if so, then the whole budget of the corresponding client can be earned. They showed that the problem is \textsf{NP-hard}, even hard to approximate within any constant factor. They proposed an iterative approach that selects one seed node at a time and obtains a partial allocation. Du et al. \cite{du2021nonsubmodular}  studied the same problem and showed that the function is non-submodular, and it is a difference between two submodular functions. They proposed a marginal increment-based prune and search approach for solving this problem. Chen et al. \cite{chen2020random} studied this problem in multiple product setting and called it as Profit Maximization with Multiple Adoptions Problem. They proposed an approach called \emph{reverse influence sampling} which achieves $(1-\frac{1}{e}-\epsilon)$\mbox{-}factor approximation guarantee. Salavati et al. \cite{salavati2019identifying} studied the profit maximization problem and adopted the \emph{ant colony optimization technique} for solving this problem. Gao et al. \cite{gao2019robust} studied the profit maximization problem and proposed double sandwich algorithm, and also they proposed a sampling strategy to improve the proposed approach. Fang et al. \cite{9861104} studied \textit{budgeted multiple-product profit maximization} in social networks by explicitly modeling users’ limited purchasing ability across products with different prices. Li et al. \cite{9772716} developed an analytical model to compare ad-sponsored, subscription-based, and hybrid revenue models for a monopolistic social networking platform considering service quality differences, user–service fit, and network effects. Du et al. \cite{10559622} formulated a non-submodular profit maximization problem and introduced a Shapley value–based algorithm that effectively selects influencers to maximize overall marketing profit on real social networks. Teng et al. \cite{10817614} introduced a multi-grade revenue maximization framework for viral marketing that simultaneously models promotional diffusion and competitive influence among multiple advertisers in social networks. 
	\section{Preliminaries and Problem Definition} \label{Sec:PPD}
	
		In this section, we describe the background and define our problem formally. We start by describing Social Networks in Definition \ref{Def:Soc_Net}.
	
	\begin{mydef}[Social Network] \cite{wasserman1994social} \label{Def:Soc_Net}
		A Social Network is generally modeled using a simple and weighted graph $G(V, E, \mathcal{P})$. Here, the vertex set $V(G)=\{v_1, v_2, \ldots, v_n\}$ are the users of the social network, the edge set $E(G)=\{e_1, e_2, \ldots, e_m\}$ represents the social relation among the users, and the edge weight function $\mathcal{P}$ maps each edge to its corresponding influence probability; i.e.; $\mathcal{P}: E(G) \longrightarrow (0,1]$.
	\end{mydef}
	Depending on the nature of the relationship, the graph may be directed or undirected. As an example, in the case of \texttt{Twitter} social network, the relationship is follower-followee, hence it is directed. On the other hand, if we consider the \texttt{Facebook} social network, the relationship is a friend, and hence it is undirected. We denote the number of vertices and edges of $G$ by $n$ and $m$, respectively. For any edge $e \in E(G)$, we denote its influence probability as $\mathcal{P}(e)$. If, $e \notin E(G)$ then $\mathcal{P}(e)=0$. The implication of the influence probability is if $e \equiv (uv) \in E(G)$, then the user $u$ can influence $v$ with the probability $\mathcal{P}(e)$. As in this study, we consider a social network represented by a graph; in the subsequent parts of this paper, we use the terms social network and graph interchangeably. Among many important properties of social networks, one relevant property to this study is the \emph{diffusion of information} which tells that information flows from one node to the other node through the intermediate nodes. To study this diffusion process in a social network, several diffusion models have been introduced and studied in the literature \cite{guille2013information}. Among them, the \emph{Independent Cascade Model} (\emph{IC Model}) and the \emph{Linear Threshold Model} (\emph{LT Model}) have been predominantly used in the social network analysis literature. In this study, we consider that the diffusion of information is happening by the rule of IC Model. Next, we state this in Definition \ref{Def:IC_Model}. By $[T]$, we define the set $\{1, 2, \ldots,T\}$.

	\begin{mydef}[Independent Cascade Model] \label{Def:IC_Model}
		Under the Independent Cascade Model, information propagates in discrete time steps. The basic rules of this model are as follows:
		\begin{itemize}
			\item a node can be either in an `active' (also called `influenced') or `inactive' (also called `uninfluenced') state,
			\item a node can change its state from `inactive' to `active' but not vice versa,
			\item every newly active node at time $t$ (denoted by $\mathcal{S}_{t}$) will get a single chance to activate its inactive neighbor and will be successful with the probability as given by respective edge weight.
			\item diffusion process stops when no more node activation is possible.
		\end{itemize}
	\end{mydef}
	
	From Definition \ref{Def:IC_Model}, it can be observed that in IC Model, diffusion starts from the initial set of nodes. Individually, we call them \emph{seed node}, and together we call them \emph{seed set}. As mentioned in Definition \ref{Def:Soc_Net}, a social network is a probabilistic graph as its edges are marked with a probability value. However, it is important to observe that at the end of diffusion, one deterministic graph will be generated, and this is called a live graph which is stated in Definition \ref{Def:Live_Graph}.
	
	\begin{mydef}[Live Graph] \label{Def:Live_Graph}
		Given a social network $G(V, E, \mathcal{P})$, a seed set $\mathcal{S} \subseteq V(G)$ a live graph is one which can be obtained at the end of the diffusion process by independent cascade model. We denote the set of all possible live graphs by $L(G)$. Every edge $e \in E(G)$, can either be present or not in a live graph. So, there are $2^{m}$ many possible live graphs where $m$ denotes the number of edges of the social network. We denote this set by $L(G)$. The probability that any live graph $\mathcal{G} \in L(G)$ will be generated can be given by the Equation \ref{Eq:Probability}. 
		\begin{equation} \label{Eq:Probability}
			P(\mathcal{G})= \underset{e \in E(\mathcal{G})}{\prod} \ \mathcal{P}(e) \ \underset{e \in E(G) \setminus E(\mathcal{G})}{\prod} \ (1- \mathcal{P}(e))
		\end{equation}
	\end{mydef}
	
	The influence of a seed set is defined as the number of active nodes at the end of the diffusion process. For any seed set $\mathcal{S}$, let $I(\mathcal{S})$ denote the set of nodes activated at the end of the diffusion process, hence $I(\mathcal{S})= \underset{i \in [T] \cup \{0\}}{\bigcup} \ \mathcal{S}_{i}$. As the diffusion of the information under the IC Model is a probabilistic process hence the influence of a seed set is measured in terms of expectation and denoted by $\sigma(\mathcal{S})$; i.e.; $\sigma(\mathcal{S})= \mathbb{E}[|I(\mathcal{S})|]$. Here, $\mathbb{E}[.]$ denotes the expectation operator of a random variable, and $\sigma()$ is the social influence function that maps every subset of users to its expected influence; i.e. $\sigma: 2^{V(G)} \longrightarrow \mathbb{R}_{0}^{+}$ with $\sigma(\emptyset)=0$. $\mathbb{R}_{0}^{+}$ is the set of positive real numbers including $0$. One well-studied problem in this context is to choose a subset of $k \in \mathbb{Z}^{+}$ (given as input) nodes to maximize the influence. This problem is referred to as the Influence Maximization Problem in the literature \cite{li2018influence,banerjee2020survey}, which has been stated in Definition \ref{Def:Inf_Max}.
	
	\begin{mydef}[Influence Maximization Problem] \label{Def:Inf_Max}
		Given social network $G(V, E, \mathcal{P})$ and a positive integer $k$, the goal of the influence maximization problem is to choose a subset of $k$ nodes such that the influence in the network is maximized. Mathematically, this problem can be written as follows:
		
		\begin{equation}
			\mathcal{S}^{OPT}=\underset{\mathcal{S} \subseteq V(G); |\mathcal{S}| \leq k}{argmax} \ \sigma(\mathcal{S})
		\end{equation}
	Here, $\mathcal{S}^{OPT}$ denotes the optimal seed set of size $k$ in $G$.
	\end{mydef}
	It has been established in the literature that even given a seed set counting the number of influenced nodes is \textsf{$\#P$-Complete}, and hence,  the Influence Maximization Problem is \textsf{NP-hard} \cite{kempe2003maximizing}. Recently, this problem has been studied with a slight variation by Nguyen et al.\cite{nguyen2013budgeted}. In this problem, instead of the budget as the number of nodes, the budget has been given as an amount of money. Each user of the network is associated with a cost. This problem has been referred to as the Budgeted Influence Maximization Problem stated in Definition \ref{Def:Bud_Inf_Max}.
	
	\begin{mydef}[Budgeted Influence Maximization Problem] \label{Def:Bud_Inf_Max}
		Given a social network $G= (V, E, \mathcal{P})$, a cost function $C: V(G) \longrightarrow \mathbb{R}^{+}$ and a budget $B$, the budgeted influence maximization problem asks to choose a subset of the nodes as seed within the allocated budget such that the influence in the network is maximized. Mathematically, this problem can be written as follows:
		\begin{equation}
			\mathcal{S}^{OPT}=\underset{\mathcal{S} \subseteq V(G); C(\mathcal{S}) \leq B}{argmax} \ \sigma(\mathcal{S})
		\end{equation}
		For any user $u \in V(G)$, let $C(u)$ denote its selection cost. $C(\mathcal{S})$ denotes the total selection of the users in $\mathcal{S}$, and hence $C(\mathcal{S})=\underset{u \in \mathcal{S}}{\sum} \ C(u)$. 
	\end{mydef}
	It can be observed from Definition \ref{Def:Bud_Inf_Max} that the BIM Problem exactly resembles the Influence Maximization Problem when for every user $u \in V(G)$, $C(u)=1$. As the BIM Problem is a generalization of the Influence Maximization Problem, hence it will also be \textsf{NP-hard} to solve optimally. In Influence Maximization Problem or BIM Problem, we want to maximize the number of active nodes with a seed set which is good enough for situations when one simply wants to inform or make aware a maximum number of people about a piece of information. Now, consider a different situation. Suppose a commercial house wants to promote a product through online social networks, so naturally, there is an associated cost with the product, and once a user is influenced, then (s)he will purchase the product, and thus the commercial house can earn some benefit out of it. The function $b$ associates each user to the corresponding benefit value, i.e., $b: V(G) \longrightarrow \mathbb{R}^{+}$ and for any user $u \in V(G)$, we denote this value by $b(u)$. Now, we state the benefit by a seed set in Definition \ref{Def:Benefit_Seed_Set}.
	\begin{mydef}[Benefit by a Seed Set]\label{Def:Benefit_Seed_Set}
		Given a social network $G(V, E, \mathcal{P})$ and a seed set $\mathcal{S}$, we denote the benefit obtained by the seed set as $\beta(\mathcal{S})$ and defined in terms of expectation over the set of all possible live graphs. Mathematically, this can be defined using Equation \ref{Eq:Expected_Benefit}. 
		\scriptsize
		\begin{equation} \label{Eq:Expected_Benefit}
			\beta(\mathcal{S})= \mathbb{E}[\underset{v \ \in I(\mathcal{S}) \cap V(G)}{\sum} \ b(v)]  = \underset{\mathcal{G} \in L(G)}{\sum} \ P(\mathcal{G}) \cdot \underset{v \ \in I(\mathcal{S}) \cap V(\mathcal{G})}{\sum} \ b(v)
		\end{equation}
	\end{mydef}
	
	From the commercial house point of view, the goal is to maximize the profit obtained by a seed set which is stated in Definition \ref{Def:Profit_Seed_Set}.
	
	\begin{mydef}[Profit of a Seed Set] \label{Def:Profit_Seed_Set}
		Given a seed set $\mathcal{S} \subseteq V(G)$, the profit earned by it is defined as the difference between the expected benefit that can be earned by influencing the users in $\mathcal{S}$ and the cost of the seed set. We denote this quantity by $\phi(\mathcal{S})$ and defined in Equation \ref{Eq:Profit}.
		\begin{equation} \label{Eq:Profit}
			\phi(\mathcal{S})= \beta(\mathcal{S}) \ - \ C(\mathcal{S})
		\end{equation}
	\end{mydef}
	Based on the definition of profit of a seed set, now we state the Profit Maximization Problem in Definition \ref{Def:Prof_Max}.

	\begin{mydef}[Profit Maximization Problem] \label{Def:Prof_Max}
		Given a social network $G = (V, E, \mathcal{P})$ along with the cost and benefit function $C$ and $b$, respectively, and a fixed budget $B$, the goal of the Profit Maximization Problem is to choose a subset of the nodes as seed such that the profit function as defined in Equation \ref{Eq:Profit} is maximized. Mathematically, this problem can be written using Equation \ref{Eq:Problem}.
		\begin{equation} 
		\label{Eq:Problem}
		\mathcal{S}^{OPT} = \underset{\mathcal{S} \subseteq V(G), C(\mathcal{S}) \leq B}  {argmax} \ \phi(\mathcal{S}) 
		\end{equation}
	\end{mydef}
	In the literature, it has been reported that instead of deploying all the seed nodes initially, we can deploy a small number initially and observe the diffusion process. Subsequently, we can deploy the seed nodes for the second phase, and thus it completes the diffusion process. This strategy is called influence maximization in two phases \cite{dhamal2016information}. Also, this method has been extended even for multi-round case \cite{sun2018multi}. One point to highlight here is that both these studies consider the allocated budget in terms of the number of nodes. However, in practice, the commercial house allocates some budget in terms of money. Also, real-world social networks are formed by rational human beings, and hence they require to be incentives if they act as initial adopters. Hence, it is more realistic to consider the budget in terms of incentives, not in terms of the number of nodes. In this direction, we define the two phase profit maximization problem in Definition \ref{Def:Two_Phase_Prof_Max}. 
	   
	\begin{mydef}[Two Phase Profit Maximization Problem] \label{Def:Two_Phase_Prof_Max}
		Given a social network $G(V,E,\mathcal{P})$, a cost function $C:V(G) \longrightarrow \mathbb{R}^{+}$, benefit function $b: V(G) \longrightarrow \mathbb{R}_{0}^{+}$ and a budget $B$, the \textsc{Two Phase Profit Maximization Problem} asks to divide the budget into two parts $B_{1}$ and  $B_{2}$ such that $B_{1}+ B_{2} \leq B$. Also, the budget $B_{1}$ and $ B_{2}$ will be used in the first and second phase of the diffusion process such that at the end of the second phase of the diffusion process, the number of active nodes is maximized.
	\end{mydef}
	From the algorithmic point of view, the Two Phase Profit Maximization Problem can be expressed as follows:
	
	\begin{center}
		\begin{tcolorbox}[title=\textsc{\textbf{\begin{center}
		Two Phase Profit Maximization Problem
		\end{center}}}, width=13.2cm] 
			\textbf{Input:} A Social Network $G(V,E, \mathcal{P})$ with $n$ users and $m$ edges, The Cost Function of the Users $C$, The Benefit Function $b$, One Information Diffusion Model, Budget $B$.
			\vspace{0.2 cm}
			\\
			\textbf{Problem:} Come up with a Two Phase Diffusion Model as well as Solution Strategy for the Profit Maximization Problem on Social Networks.
			\vspace{0.1 cm}
			\\
			\textbf{Output:} \justifying{The optimal splitting of the budget $B$ into two parts $B_1$ and $B_2$ for the first and second phase, respectively such that}
			\\ $B_1+B_2 \leq B$. 
		\end{tcolorbox}
	\end{center}

	\par The goal of the paper is to investigate how the profit made by a commercial house can be maximized using two phase diffusion approach. Symbols and notations used in this paper have been listed in Table \ref{Table1:Notations}. Next, we propose a mathematical model for this problem.
	\begin{table*}
		\caption{Symbols and Notations with their Interpretations}
		\label{Table1:Notations}
		\vspace*{3mm}
		\centering
		\begin{tabular}{p{0.20\linewidth}p{0.80\linewidth}}
			\hline
			\hline
			\textbf{Notation} & \textbf{Description}\\
			\hline
			\hline
			$G(V, E, \mathcal{P})$ & A simple and weighted directed graph $G$ with a node set $V$, Edge set $E$ and $\mathcal{P}$ is the edge weight function
			\\
			\hline $n,m$ & $n$ is the number of users of the social network $G$, and $m$ is the number of social relations among the users in the social network $G$
			\\
			\hline $u, e$ & $u$ denotes a user connected to another user with the edge $e$
			\\
			\hline $d$ & timestep
			\\ 
			\hline $\mathcal{P}(e)$ & Influence probability of the edge $e$
			\\
			\hline $L(\mathcal{G})$ & The set of live graphs which are obtained by sampling the edges of $G$
			\\
			\hline $\mathcal{G}$ & A live graph of $G(V, E, \mathcal{P})$\\
			\hline $Y$ & The partial observation of the live graph $\mathcal{G}$
			\\ 
			\hline $A_Y$ & Already activated nodes at timestep $d$
			\\
			\hline  $R_Y$ & Recently activated nodes at timestep $d$
			\\ 
			\hline $P(\mathcal{G})$ & Probability of generation of live graph $\mathcal{G}$
			\\
			\hline $P(Y)$ & Probability of generation of partial observation $Y$
			\\ 
			\hline $\mathcal{S}$ & Seed set; a set of initially activated nodes
			\\
			\hline $\mathcal{S}^{OPT}$ & Optimal seed set
			\\
			\hline $b$ & Function which associates each user $u$ to its corresponding benefit
			\\
			\hline $b(u)$ & Benefit from a node $u$ of $G$
			\\
			\hline $\beta(\mathcal{S})$ & Expected Benefit by a seed set $\mathcal{S}$
			\\
			\hline $\phi(\mathcal{S})$ & Profit earned from a seed set $\mathcal{S}$
			\\
			\hline $C$ & Cost function which associates a user $u$ with its corresponding cost
			\\
			\hline $C(u)$ & Cost of a node $u$ of $G$
			\\
			\hline $\mathcal{S}_{1}, \ \mathcal{S}_{2}$ & Seed set for Phase I, Seed set for Phase II
			\\
			\hline $C(\mathcal{S}_{1}), \ C(\mathcal{S}_{2})$ & Selection cost of seed set of Phase I, Selection cost of seed set of Phase II
			\\
			\hline $B$ & Total Budget
			\\
			\hline $B_{1}, \ B_{2}$ & Budget for Phase I, Budget for Phase II
			\\
			\hline $P(\frac{\mathcal{G}}{Y})$ & Probability of generation of live graph $\mathcal{G}$ when partial observation $Y$ is given
			\\
			\hline $\mathbb{R}^{+}_{0}$ & The set of positive real numbers including $0$
			\\
			\hline
			\hline\\
		\end{tabular}%
	\end{table*}

	\section{The Proposed Model} \label{Sec:PM}
	In this section, we describe the formulation of our proposed model and subsequently prove certain properties along with illustrative examples. Accordingly, we divide this section into three parts. In the first one, we describe the formulation of the objective function of our model. The second one contains an illustrative example, and the third one contains properties of the objective function.   
	
	\subsection{Formulation of the Objective Function}
	\label{sec:objfun}
	Here, we formulate one objective function that measures the expected profit at the end of the second phase. Now, let a live graph $\mathcal{G} \in L(G)$ with its generation probability $P(\mathcal{G})$ is destined to occur, and $\mathcal{S}_{1}$ is the seed set for the Phase I. Certainly, $C(\mathcal{S}_{1}) \leq B_{1}$.  The diffusion process starts on a live graph $\mathcal{G}$ with seed set $\mathcal{S}_{1}$ under the IC Model. Now, we observe this diffusion process till timestep $d$. If we do that, we will have the information regarding which nodes are influenced and which are not. We call this as the \emph{partial observation} till timestep $d$ and denoted as $Y$. So, at the end of the timestep $d$, we have \emph{already activated nodes} denoted by  $A_{Y}$ and \emph{recently activated nodes} (at timestep $d$) $R_{Y}$. These two sets $A_{Y}$ and $R_{Y}$ are determined from the partial observation $Y$. The benefit from the nodes in $R_{Y}$ is as follows:
	\begin{equation}
		b(R_{Y}) = \displaystyle{\sum_{v \in R_{Y}}}b(v)
	\end{equation}
	
	Now, as we have the partial observation $Y$, for a subset of the edges of the live graph $\mathcal{G}$, we are sure whether they have appeared or not. So, the generation probability of the live graph $\mathcal{G}$ is updated as $P(\frac{\mathcal{G}}{Y})$. Now, the second phase needs to begin, and assume that $\mathcal{S}_{2}^{OPT (Y,B_{2})}$ denotes the optimal seed set for the given partial observation $Y$ and second phase budget $B_2$. Assume that we know which nodes are in $\mathcal{S}_{2}^{OPT (Y,B_{2})}$. Now, both the sets $\mathcal{S}_{2}^{OPT(Y,B_{2})}$ and $R_{Y}$ will be used for the diffusion process of the second phase. Thus, Phase II proceeds from the timestep $d$ onward.   
	\par Now, our goal is to calculate the expected profit that can be earned in Phase II. Now, it is important to observe that in Phase II, the nodes from which the profit can be earned are from $V(G) \setminus A_{Y}$. So, for the given partial observation $Y$ (hence, recently activated nodes $R_{Y}$) along with an optimal seed set for Phase II $\mathcal{S}_{2}^{OPT(Y,B_{2})}$, the expected profit for Phase II will be equals to $\displaystyle{\sum_{\mathcal{G}}P(\frac{\mathcal{G}}{Y}})[\phi^{V(\mathcal{G})\setminus A_{Y}}(R_{Y} \cup \mathcal{S}_{2}^{OPT(Y,B_{2})})]$. Here, $\phi^{V(\mathcal{G})\setminus A_{Y}}(\mathcal{S})$ denotes the profit earned by the seed set $\mathcal{S}$ from the graph $V(\mathcal{G})\setminus A_{Y}$.
	\par Now, we can observe that given a live graph $\mathcal{G}$, seed set $\mathcal{S}_{1}$ of Phase I, and timestep $d$, we can obtain the partial observation $Y$. Hence, $\mathcal{S}_{2}^{OPT(Y,B_{2})}$ can be written as $\mathcal{S}_{2}^{OPT(X, \mathcal{S}_{1},d,B_{2})}$. Now, to develop an objective function where the decision variable will be the seed set for Phase I, we assume that given the partial observation $Y$, we will select an optimal seed set for Phase II. It is important to observe that at the start of Phase I, the partial observation $Y$ is not known. So let our objective function be $\mathbb{F}(\mathcal{S}_{1}, d, B_{2})$ as the expected profit with respect to all possible occurrences $Y$. Assuming that $d$ and $B_{2}$ are already given, so we can write $\mathbb{F}(\mathcal{S}_{1}, d, B_{2})$ as $f(\mathcal{S}_{1})$. Now, our goal is to develop a mathematical expression, 
	
	{\scriptsize  
	 \begin{align*}
	       f(\mathcal{S}_{1}) &= 
	       \displaystyle{\sum_{Y}P({Y})} \Bigg\{\Big\{ \displaystyle{\sum_{\mathcal{G}}P(\frac{\mathcal{G}}{Y}})\phi(A_{Y}) 
	       +\displaystyle{\sum_{\mathcal{G}}P(\frac{\mathcal{G}}{Y}})[\phi^{V(\mathcal{G})\setminus A_{Y}}(R_{Y} \cup \mathcal{S}_{2}^{OPT(Y,B_{2})})]\Big\}\Bigg\} \\
	       &= \displaystyle{\sum_{Y}P({Y})}\displaystyle{\sum_{\mathcal{G}}P(\frac{\mathcal{G}}{Y})}\Bigg\{\Big\{\phi(A_{Y}) 
	   + [\phi^{V(\mathcal{G})\setminus A_{Y}}(R_{Y} \cup \mathcal{S}_{2}^{OPT(X,\mathcal{S}_{1},d,B_{2})})]\Big\}\Bigg\}\\
	   &= \displaystyle{\sum_{Y}P({Y})}\displaystyle{\sum_{\mathcal{G}}P(\frac{\mathcal{G}}{Y})}\Bigg\{\Big\{[\beta(A_{Y}) - C(\mathcal{S}_{1})] 
	 + [\phi^{V(\mathcal{G})\setminus A_{Y}}(R_{Y} \cup \mathcal{S}_{2}^{OPT(X,\mathcal{S}_{1},d,B_{2})})]\Big\}\Bigg\} \\
	 		&= \displaystyle{\sum_{Y}P({Y}})\displaystyle{\sum_{\mathcal{G}}P(\frac{\mathcal{G}}{Y}})\Bigg\{{\phi^{\mathcal{G}}(\mathcal{S}_{1} \cup \mathcal{S}_{2}^{OPT(X,\mathcal{S}_{1},d,B_{2})})}\Bigg\}
	    \end{align*}
	   \begin{align*}
		\Bigg\{ \because \displaystyle{\sum_{Y}P({Y}})\displaystyle{\sum_{\mathcal{G}}P(\frac{\mathcal{G}}{Y}}) &= \displaystyle{\sum_{Y}}\displaystyle{\sum_{\mathcal{G}}} {\frac{P(\mathcal{G},Y)}{P(Y)}} P(Y)
		= \displaystyle{\sum_{Y}}\displaystyle{\sum_{\mathcal{G}}} P(\mathcal{G},{Y}) 
		= \displaystyle{\sum_{\mathcal{G}}}\displaystyle{\sum_{Y}} P(\mathcal{G},{Y}) \\
		&= \displaystyle{\sum_{\mathcal{G}}} P(\mathcal{G})\Bigg\}
	\end{align*}
	}
	
		The objective function formulated for our Two Phase Profit Maximization Problem is as follows:
	\begin{equation}
		\label{Eq: Obj_eq}
		\therefore f(\mathcal{S}_{1})= \displaystyle{\sum_{\mathcal{G}}}P(\mathcal{G}){\phi^{\mathcal{G}}(\mathcal{S}_{1} \cup \mathcal{S}_{2}^{OPT(X,\mathcal{S}_{1},d,B_{2})})}
	\end{equation}

	\subsection{Example}
	Following Figure \ref{Fig1: TPPMillustration} is the Graph $G$ whose live graph occurrences $\mathcal{G}$ are tabulated in Table \ref{Table2:TPPMExampleillustration}. We have seed set $\mathcal{S}_1 = \{u_{1}\}$ at timestep $d = 1$, with budget $B_1 = 2$ and a total budget $B = 5$. The calculated value of the objective function is $2.846$.
	\begin{figure}[h!]
		\begin{center}
			\includegraphics[height=1.50in,width=3.50in]{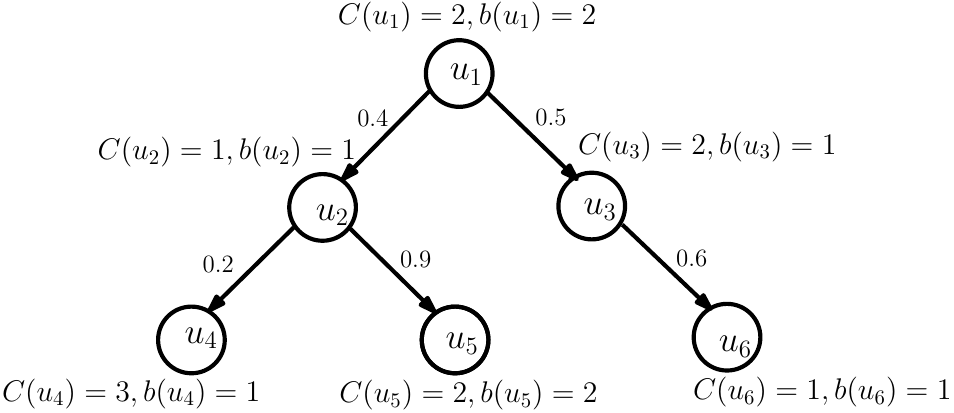}
			\caption{An example graph to illustrate Two phase Profit Maximization}
			\label{Fig1: TPPMillustration}
		\end{center}
	\end{figure}
	
	\begin{table*}[htbp]
	\centering
	\caption{Two Phase Profit Maximization Approach Illustration}
	\label{Table2:TPPMExampleillustration}
	\begin{tabular}{|lllll|l|}
	\hline
	\multicolumn{1}{|l|}{$\mathcal{G}$}              & \multicolumn{1}{l|}{$P(\mathcal{G})$}   & \multicolumn{1}{l|}{$\mathcal{S}_{1}/A_{Y}$} & \multicolumn{1}{l|}{$R_{Y}$}   & $\mathcal{S}_{2}^{OPT(Y, B_{2})}$   & $f(\mathcal{S}_{1})$  \\ \hline
	\multicolumn{1}{|l|}{\{\}} & \multicolumn{1}{l|}{0.0096} & \multicolumn{1}{l|}{$u_1$}     & \multicolumn{1}{l|}{-}    & $u_2$    & 0      \\ \hline
	\multicolumn{1}{|l|}{$u_3u_6$}             & \multicolumn{1}{l|}{0.0144} & \multicolumn{1}{l|}{$u_1$}     & \multicolumn{1}{l|}{-}    & $u_3$    & 0      \\ \hline
	\multicolumn{1}{|l|}{$u_2u_5$}             & \multicolumn{1}{l|}{0.0864} & \multicolumn{1}{l|}{$u_1$}     & \multicolumn{1}{l|}{-}    & $u_2$    & 0.1728 \\ \hline
	\multicolumn{1}{|l|}{$u_2u_5$, $u_3u_6$}         & \multicolumn{1}{l|}{0.1296} & \multicolumn{1}{l|}{$u_1$}     & \multicolumn{1}{l|}{-}    & $u_2$, $u_3$ & 0.2592 \\ \hline
	\multicolumn{1}{|l|}{$u_2u_4$}             & \multicolumn{1}{l|}{0.0024} & \multicolumn{1}{l|}{$u_1$}     & \multicolumn{1}{l|}{-}    & $u_2$    & 0.0024 \\ \hline
	\multicolumn{1}{|l|}{$u_2u_4$, $u_3u_6$}         & \multicolumn{1}{l|}{0.0036} & \multicolumn{1}{l|}{$u_1$}     & \multicolumn{1}{l|}{-}    & $u_2$, $u_3$ & 0.0036 \\ \hline
	\multicolumn{1}{|l|}{$u_2u_4$, $u_2u_5$}         & \multicolumn{1}{l|}{0.0216} & \multicolumn{1}{l|}{$u_1$}     & \multicolumn{1}{l|}{-}    & $u_2$    & 0.0648 \\ \hline
	\multicolumn{1}{|l|}{$u_2u_4$, $u_2u_5$, $u_3u_6$}     & \multicolumn{1}{l|}{0.0324} & \multicolumn{1}{l|}{$u_1$}     & \multicolumn{1}{l|}{-}    & $u_2, u_3$ & 0.0972 \\ \hline
	\multicolumn{1}{|l|}{$u_1u_3$}             & \multicolumn{1}{l|}{0.0096} & \multicolumn{1}{l|}{$u_1$}     & \multicolumn{1}{l|}{$u_3$}    & $u_2$    & 0.0096 \\ \hline
	\multicolumn{1}{|l|}{$u_1u_3, u_3u_6$}         & \multicolumn{1}{l|}{0.0144} & \multicolumn{1}{l|}{$u_1$}     & \multicolumn{1}{l|}{$u_3$}    & $u_2$    & 0.0288 \\ \hline
	\multicolumn{1}{|l|}{$u_1u_3, u_2u_5$}         & \multicolumn{1}{l|}{0.0864} & \multicolumn{1}{l|}{$u_1$}     & \multicolumn{1}{l|}{$u_3$}    & $u_2$    & 0.2592 \\ \hline
	\multicolumn{1}{|l|}{$u_1u_3, u_2u_5, u_3u_6$}     & \multicolumn{1}{l|}{0.1296} & \multicolumn{1}{l|}{$u_1$}     & \multicolumn{1}{l|}{$u_3$}    & $u_2$    & 0.5184 \\ \hline
	\multicolumn{1}{|l|}{$u_1u_3, u_2u_4$}         & \multicolumn{1}{l|}{0.0024} & \multicolumn{1}{l|}{$u_1$}     & \multicolumn{1}{l|}{$u_3$}    & $u_2$    & 0.0048 \\ \hline
	\multicolumn{1}{|l|}{$u_1u_3, u_2u_4, u_3u_6$}     & \multicolumn{1}{l|}{0.0036} & \multicolumn{1}{l|}{$u_1$}     & \multicolumn{1}{l|}{$u_3$}    & $u_2$    & 0.0108 \\ \hline
	\multicolumn{1}{|l|}{$u_1u_3, u_2u_4, u_2u_5$}     & \multicolumn{1}{l|}{0.0216} & \multicolumn{1}{l|}{$u_1$}     & \multicolumn{1}{l|}{$u_3$}    & $u_2$    & 0.0864 \\ \hline
	\multicolumn{1}{|l|}{$u_1u_3, u_2u_4, u_2u_5, u_3u_6$} & \multicolumn{1}{l|}{0.0324} & \multicolumn{1}{l|}{$u_1$}     & \multicolumn{1}{l|}{$u_3$}    & $u_2$    & 0.162  \\ \hline
	\multicolumn{1}{|l|}{$u_1u_2$}             & \multicolumn{1}{l|}{0.0064} & \multicolumn{1}{l|}{$u_1$}     & \multicolumn{1}{l|}{$u_2$}    & $u_3$    & 0      \\ \hline
	\multicolumn{1}{|l|}{$u_1u_2, u_3u_6$}         & \multicolumn{1}{l|}{0.0096} & \multicolumn{1}{l|}{$u_1$}     & \multicolumn{1}{l|}{$u_2$}    & $u_3$    & 0.0096 \\ \hline
	\multicolumn{1}{|l|}{$u_1u_2, u_2u_5$}         & \multicolumn{1}{l|}{0.0576} & \multicolumn{1}{l|}{$u_1$}     & \multicolumn{1}{l|}{$u_2$}    & $u_3$    & 0.1152 \\ \hline
	\multicolumn{1}{|l|}{$u_1u_2, u_2u_5, u_3u_6$}     & \multicolumn{1}{l|}{0.0864} & \multicolumn{1}{l|}{$u_1$}     & \multicolumn{1}{l|}{$u_2$}    & $u_3$    & 0.2592 \\ \hline
	\multicolumn{1}{|l|}{$u_1u_2, u_2u_4$}         & \multicolumn{1}{l|}{0.0016} & \multicolumn{1}{l|}{$u_1$}     & \multicolumn{1}{l|}{$u_2$}    & $u_3$    & 0.0016 \\ \hline
	\multicolumn{1}{|l|}{$u_1u_2, u_2u_4, u_3u_6$}     & \multicolumn{1}{l|}{0.0024} & \multicolumn{1}{l|}{$u_1$}     & \multicolumn{1}{l|}{$u_2$}    & $u_3$    & 0.0048 \\ \hline
	\multicolumn{1}{|l|}{$u_1u_2, u_2u_4, u_2u_5$}     & \multicolumn{1}{l|}{0.0144} & \multicolumn{1}{l|}{$u_1$}     & \multicolumn{1}{l|}{$u_2$}    & $u_3$    & 0.0432 \\ \hline
	\multicolumn{1}{|l|}{$u_1u_2, u_2u_4, u_2u_5, u_3u_6$} & \multicolumn{1}{l|}{0.0216} & \multicolumn{1}{l|}{$u_1$}     & \multicolumn{1}{l|}{$u_2$}    & $u_3$    & 0.0864 \\ \hline
	\multicolumn{1}{|l|}{$u_1u_2, u_1u_3$}         & \multicolumn{1}{l|}{0.0064} & \multicolumn{1}{l|}{$u_1$}     & \multicolumn{1}{l|}{$u_2$, $u_3$} & $u_4$    & 0      \\ \hline
	\multicolumn{1}{|l|}{$u_1u_2, u_1u_3, u_3u_6$}     & \multicolumn{1}{l|}{0.0096} & \multicolumn{1}{l|}{$u_1$}     & \multicolumn{1}{l|}{$u_2$, $u_3$} & $u_4$   & 0.0192 \\ \hline
	\multicolumn{1}{|l|}{$u_1u_2, u_1u_3, u_2u_5$}     & \multicolumn{1}{l|}{0.0576} & \multicolumn{1}{l|}{$u_1$}     & \multicolumn{1}{l|}{$u_2$, $u_3$} & $u_4$    & 0.1728 \\ \hline
	\multicolumn{1}{|l|}{$u_1u_2, u_1u_3, u_2u_5, u_3u_6$} & \multicolumn{1}{l|}{0.0864} & \multicolumn{1}{l|}{$u_1$}     & \multicolumn{1}{l|}{$u_2$, $u_3$} & $u_5$    & 0.2592 \\ \hline
	\multicolumn{1}{|l|}{$u_1u_2, u_1u_3, u_2u_4$}     & \multicolumn{1}{l|}{0.0016} & \multicolumn{1}{l|}{$u_1$}     & \multicolumn{1}{l|}{$u_2$, $u_3$} & $u_5$    & 0.0048 \\ \hline
	\multicolumn{1}{|l|}{$u_1u_2, u_1u_3, u_2u_4, u_3u_6$} & \multicolumn{1}{l|}{0.0024} & \multicolumn{1}{l|}{$u_1$}     & \multicolumn{1}{l|}{$u_2$, $u_3$} & $u_5$    & 0.0096 \\ \hline
	\multicolumn{1}{|l|}{$u_1u_2, u_1u_3, u_2u_4, u_2u_5$} & \multicolumn{1}{l|}{0.0144} & \multicolumn{1}{l|}{$u_1$}     & \multicolumn{1}{l|}{$u_2$, $u_3$} & $u_6$    & 0.072  \\ \hline
	\multicolumn{1}{|l|}{$u_1u_2, u_1u_3, u_2u_4, u_2u_5, u_3u_6$} & \multicolumn{1}{l|}{0.0216} & \multicolumn{1}{l|}{$u_1$} & \multicolumn{1}{l|}{$u_2$, $u_3$} & $u_6$ & 0.108 \\ \hline
	\multicolumn{5}{|c|}{\textbf{$\displaystyle{\sum_{\mathcal{G}}}P(\mathcal{G}){\phi^{\mathcal{G}}(\mathcal{S}_{1} \cup \mathcal{S}_{2}^{OPT(\mathcal{G},\mathcal{S}_{1},d,B_{2})})}$}}                                                                                                             & 2.8456 \\ \hline
	\end{tabular}%
	\end{table*}

	\subsection{Properties of the Objective Function}
	\begin{mylemma}
		$f(.)$ may be positive or negative.
		\label{Lem:First}
	\end{mylemma}
	\begin{figure}[h!]
	\centering
	\resizebox{\columnwidth}{!}{\begin{tabular}{cc}
	\includegraphics[scale=0.40]{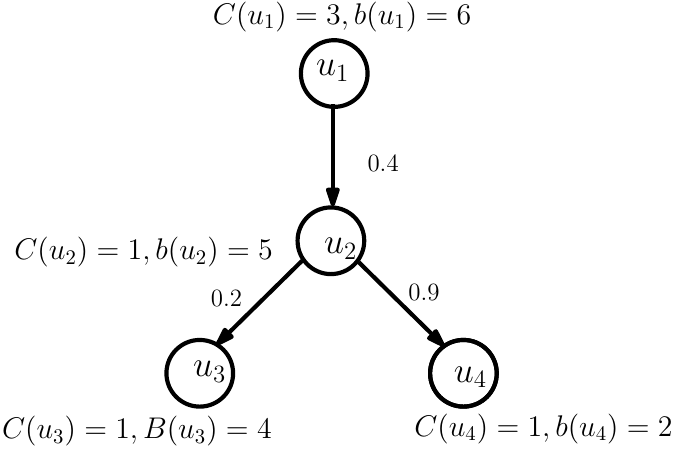} & \includegraphics[scale=0.40]{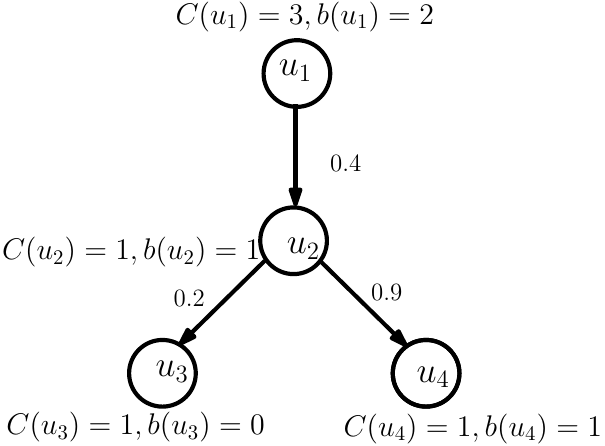} \\

	{\footnotesize (a) $f(\mathcal{S}_{1})$ is positive}  & {\footnotesize (b) $f(\mathcal{S}_{1})$ is negative} \\
	\end{tabular}}
	\caption{Example Graphs to illustrate $f(\cdot)$ }
	\label{Fig2:Lemma4.1_positive_negative}
	\end{figure}
	
	\begin{table*}
	\centering
	\caption{Positive and Negative value of Objective Function}
	\label{Table3:PositiveNegativeObj}
	\resizebox{\columnwidth}{!}{%
	\begin{tabular}{|lllllllll|lll|}
	\hline
	\multicolumn{6}{|l|}{} &
	  \multicolumn{3}{c|}{$f(\mathcal{S}_{1}) \geq 0$} &
	  \multicolumn{3}{c|}{$f(\mathcal{S}_{1}) < 0$} \\ \hline
	\multicolumn{1}{|l|}{$\mathcal{G}$} &
	  \multicolumn{1}{l|}{$P(\mathcal{G})$} &
	  \multicolumn{1}{l|}{$\mathcal{S}_{1}$} &
	  \multicolumn{1}{l|}{$R_{Y}$} &
	  \multicolumn{1}{l|}{$\mathcal{S}_{2}$} &
	  \multicolumn{1}{l|}{$C(\mathcal{S}_{1} \cup \mathcal{S}_{2})$} &
	  \multicolumn{1}{l|}{$\beta(\mathcal{S}_{1} \cup \mathcal{S}_{2})$} &
	  \multicolumn{1}{l|}{$\phi(\mathcal{S}_{1} \cup \mathcal{S}_{2})$} &
	  $f(\mathcal{S}_{1})$ &
	  \multicolumn{1}{l|}{$\beta(\mathcal{S}_{1} \cup \mathcal{S}_{2})$} &
	  \multicolumn{1}{l|}{$\phi(\mathcal{S}_{1} \cup \mathcal{S}_{2})$} &
	 $f(\mathcal{S}_{1})$ \\ \hline
	\multicolumn{1}{|l|}{\{\}} &
	  \multicolumn{1}{l|}{0.012} &
	  \multicolumn{1}{l|}{$u_1$} &
	  \multicolumn{1}{l|}{} &
	  \multicolumn{1}{l|}{$u_2$} &
	  \multicolumn{1}{l|}{4} &
	  \multicolumn{1}{l|}{11} &
	  \multicolumn{1}{l|}{7} &
	  0.084 &
	  \multicolumn{1}{l|}{3} &
	  \multicolumn{1}{l|}{-1} &
	  -0.012 \\ \hline
	\multicolumn{1}{|l|}{$u_2u_4$} &
	  \multicolumn{1}{l|}{0.108} &
	  \multicolumn{1}{l|}{$u_1$} &
	  \multicolumn{1}{l|}{} &
	  \multicolumn{1}{l|}{$u_2$} &
	  \multicolumn{1}{l|}{4} &
	  \multicolumn{1}{l|}{13} &
	  \multicolumn{1}{l|}{9} &
	  0.972 &
	  \multicolumn{1}{l|}{4} &
	  \multicolumn{1}{l|}{0} &
	  0 \\ \hline
	\multicolumn{1}{|l|}{$u_2u_3$} &
	  \multicolumn{1}{l|}{0.028} &
	  \multicolumn{1}{l|}{$u_1$} &
	  \multicolumn{1}{l|}{} &
	  \multicolumn{1}{l|}{$u_2$} &
	  \multicolumn{1}{l|}{4} &
	  \multicolumn{1}{l|}{15} &
	  \multicolumn{1}{l|}{11} &
	  0.308 &
	  \multicolumn{1}{l|}{3} &
	  \multicolumn{1}{l|}{-1} &
	  -0.028 \\ \hline
	\multicolumn{1}{|l|}{$u_2u_3, u_2u_4$} &
	  \multicolumn{1}{l|}{0.252} &
	  \multicolumn{1}{l|}{$u_1$} &
	  \multicolumn{1}{l|}{} &
	  \multicolumn{1}{l|}{$u_2$} &
	  \multicolumn{1}{l|}{4} &
	  \multicolumn{1}{l|}{17} &
	  \multicolumn{1}{l|}{13} &
	  3.276 &
	  \multicolumn{1}{l|}{4} &
	  \multicolumn{1}{l|}{0} &
	  0 \\ \hline
	\multicolumn{1}{|l|}{$u_1u_2$} &
	  \multicolumn{1}{l|}{0.018} &
	  \multicolumn{1}{l|}{$u_1$} &
	  \multicolumn{1}{l|}{$u_2$} &
	  \multicolumn{1}{l|}{$u_3$} &
	  \multicolumn{1}{l|}{4} &
	  \multicolumn{1}{l|}{15} &
	  \multicolumn{1}{l|}{11} &
	  0.198 &
	  \multicolumn{1}{l|}{3} &
	  \multicolumn{1}{l|}{-1} &
	  -0.018 \\ \hline
	\multicolumn{1}{|l|}{$u_1u_2, u_2u_4$} &
	  \multicolumn{1}{l|}{0.162} &
	  \multicolumn{1}{l|}{$u_1$} &
	  \multicolumn{1}{l|}{$u_2$} &
	  \multicolumn{1}{l|}{$u_3$} &
	  \multicolumn{1}{l|}{4} &
	  \multicolumn{1}{l|}{17} &
	  \multicolumn{1}{l|}{13} &
	  2.106 &
	  \multicolumn{1}{l|}{4} &
	  \multicolumn{1}{l|}{0} &
	  0 \\ \hline
	\multicolumn{1}{|l|}{$u_1u_2, u_2u_3$} &
	  \multicolumn{1}{l|}{0.042} &
	  \multicolumn{1}{l|}{$u_1$} &
	  \multicolumn{1}{l|}{$u_2$} &
	  \multicolumn{1}{l|}{$u_4$} &
	  \multicolumn{1}{l|}{4} &
	  \multicolumn{1}{l|}{17} &
	  \multicolumn{1}{l|}{13} &
	  0.546 &
	  \multicolumn{1}{l|}{4} &
	  \multicolumn{1}{l|}{0} &
	  0 \\ \hline
	\multicolumn{1}{|l|}{$u_1u_2, u_2u_3, u_2u_4$} &
	  \multicolumn{1}{l|}{0.378} &
	  \multicolumn{1}{l|}{$u_1$} &
	  \multicolumn{1}{l|}{$u_2$} &
	  \multicolumn{1}{l|}{$u_3$} &
	  \multicolumn{1}{l|}{4} &
	  \multicolumn{1}{l|}{17} &
	  \multicolumn{1}{l|}{13} &
	  4.914 &
	  \multicolumn{1}{l|}{4} &
	  \multicolumn{1}{l|}{0} &
	  0 \\ \hline
	\multicolumn{8}{|c|}{\textbf{$\displaystyle{\sum_{\mathcal{G}}}P(\mathcal{G}){\phi^{\mathcal{G}}(\mathcal{S}_{1} \cup \mathcal{S}_{2}^{OPT(\mathcal{G},\mathcal{S}_{1},d,B_{2})})}$}} &
	  12.404 &
	 \multicolumn{2}{l|}{\textbf{$\displaystyle{\sum_{\mathcal{G}}}P(\mathcal{G}){\phi^{\mathcal{G}}(\mathcal{S}_{1} \cup \mathcal{S}_{2}^{OPT(\mathcal{G},\mathcal{S}_{1},d,B_{2})})}$}} &
	  -0.058 \\ \hline
	\end{tabular}%
	}
	\end{table*}

	\begin{proof}
		In Figure \ref{Fig2:Lemma4.1_positive_negative}, we have two example graphs in Figure \ref{Fig2:Lemma4.1_positive_negative}(a) and Figure \ref{Fig2:Lemma4.1_positive_negative}(b) for illustrating the positive and negative value of the objective function, respectively. The parameters, timestep $d = 1$, budget $B_{1} = 3$ and a total budget $B = 4$ are the same for both the graphs.
		
		All the occurrences of generation of both graphs of Figure \ref{Fig2:Lemma4.1_positive_negative} are listed in Table \ref{Table3:PositiveNegativeObj} with their probabilities of generation, $P(\mathcal{G})$. The objective function $f(\mathcal{S}_{1}) = 12.404$ calculated from Figure \ref{Fig2:Lemma4.1_positive_negative} (a) is a positive value i.e. $f(\mathcal{S}_{1}) >$ 0. The objective function $f(\mathcal{S}_{1}) = -0.058$ calculated for Figure \ref{Fig2:Lemma4.1_positive_negative} (b) is a negative value i.e. $f(\mathcal{S}_{1}) <$ 0.  
		Therefore, Lemma \ref{Lem:First} is proved. For a given $G$, $f(\mathcal{S}_{1})$ may be positive or negative.
	\end{proof}
	
	\begin{mylemma}
		$f(\cdot)$ is neither monotonically increasing nor monotonically decreasing.
	\label{Lemma2:Monotonicity}
	\end{mylemma}
	\begin{proof}
		 \begin{figure}[h!]
		 	\centering
		 	\includegraphics[scale=0.55]{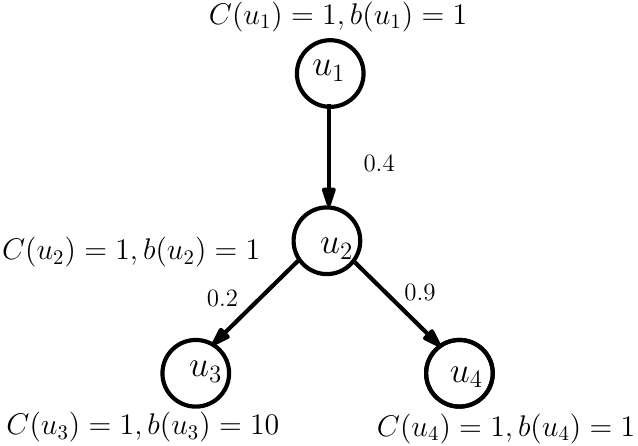}
		 	\caption{Example Graph to illustrate Non-monotone behaviour of $f(\cdot)$}
		 	\label{Fig3:Non-MonotoneProperty}
		 \end{figure}
		Let there be a $\mathcal{S}_{1}$, such that $\mathcal{S}_{1} \subset V(G)$.\\
		$f(\mathcal{S}_{1})$ is monotonically increasing, if
		\begin{equation}
			f(\mathcal{S}_{1}) \le f(\mathcal{S}_{1} \cup \{u\}), u \in V(G)\setminus S_{1}
			\label{eq:mon_increasing}
		\end{equation}
		
		$f(\mathcal{S}_{1})$ is monotonically decreasing, if
		\begin{equation}
			f(\mathcal{S}_{1}) \geq f(\mathcal{S}_{1} \cup \{u\}), u \in V(G)\setminus \mathcal{S}_{1}
			\label{eq:mon_decreasing}
		\end{equation}
	
	\begin{table*}
	\centering
	\caption{Calculating $f(\mathcal{S}_{1})$ to show its non-monotonicity}
	\label{Table4:Non-monotonicity}
	\resizebox{\columnwidth}{!}{%
	\begin{tabular}{|lllll|lll|lll|}
	\hline
	\multicolumn{2}{|l|}{} &
	\multicolumn{3}{|c|}{$f(\mathcal{S}_{1})$} &
	  \multicolumn{3}{c|}{$f(\mathcal{S}_{1}) < f(\mathcal{S}_{1} \cup \{u_3\})$} &
	  \multicolumn{3}{c|}{$f(\mathcal{S}_{1}) > f(\mathcal{S}_{1} \cup \{u_4\})$} \\ \hline
	\multicolumn{1}{|l|}{$\mathcal{G}$} &
	  \multicolumn{1}{l|}{$P(\mathcal{G})$} &
	  \multicolumn{1}{l|}{$\mathcal{S}_{1}$} &
	  \multicolumn{1}{l|}{$C(\mathcal{S}_{1})$} &
	  $f(\mathcal{S}_{1})$ &
	  \multicolumn{1}{l|}{$\mathcal{S}_{1} \cup \{u_3\}$} &
	  \multicolumn{1}{l|}{$C(\mathcal{S}_{1} \cup \mathcal{S}_{2})$} &
	  $f(\mathcal{S}_{1} \cup \{u_3\})$ &
	  \multicolumn{1}{l|}{$\mathcal{S}_{1} \cup \{u_4\}$} &
	  \multicolumn{1}{l|}{$C(\mathcal{S}_{1} \cup \mathcal{S}_{2})$} &
	  $f(\mathcal{S}_{1} \cup \{u_4\})$ \\ \hline
	\multicolumn{1}{|l|}{\{\}} &
	  \multicolumn{1}{l|}{0.012} &
	  \multicolumn{1}{l|}{$u_1$} &
	  \multicolumn{1}{l|}{2} &
	  0 &
	  \multicolumn{1}{l|}{$u_1$, $u_3$} &
	  \multicolumn{1}{l|}{3} &
	  0.108 &
	  \multicolumn{1}{l|}{$u_1$, $u_4$} &
	  \multicolumn{1}{l|}{3} &
	  0 \\ \hline
	\multicolumn{1}{|l|}{$u_2u_4$} &
	  \multicolumn{1}{l|}{0.108} &
	  \multicolumn{1}{l|}{$u_1$} &
	  \multicolumn{1}{l|}{2} &
	  0.108 &
	  \multicolumn{1}{l|}{$u_1$, $u_3$} &
	  \multicolumn{1}{l|}{3} &
	  1.08 &
	  \multicolumn{1}{l|}{$u_1$, $u_4$} &
	  \multicolumn{1}{l|}{3} &
	  0 \\ \hline
	\multicolumn{1}{|l|}{$u_2u_3$} &
	  \multicolumn{1}{l|}{0.028} &
	  \multicolumn{1}{l|}{$u_1$} &
	  \multicolumn{1}{l|}{2} &
	  0.28 &
	  \multicolumn{1}{l|}{$u_1$, $u_3$} &
	  \multicolumn{1}{l|}{3} &
	  0.252 &
	  \multicolumn{1}{l|}{$u_1$, $u_4$} &
	  \multicolumn{1}{l|}{3} &
	  0.28 \\ \hline
	\multicolumn{1}{|l|}{$u_2u_3, u_2u_4$} &
	  \multicolumn{1}{l|}{0.252} &
	  \multicolumn{1}{l|}{$u_1$} &
	  \multicolumn{1}{l|}{2} &
	  2.772 &
	  \multicolumn{1}{l|}{$u_1$, $u_3$} &
	  \multicolumn{1}{l|}{3} &
	  2.52 &
	  \multicolumn{1}{l|}{$u_1$, $u_4$} &
	  \multicolumn{1}{l|}{3} &
	  2.52 \\ \hline
	\multicolumn{1}{|l|}{$u_1u_2$} &
	  \multicolumn{1}{l|}{0.018} &
	  \multicolumn{1}{l|}{$u_1$} &
	  \multicolumn{1}{l|}{2} &
	  0.18 &
	  \multicolumn{1}{l|}{$u_1$, $u_3$} &
	  \multicolumn{1}{l|}{3} &
	  0.18 &
	  \multicolumn{1}{l|}{$u_1$, $u_4$} &
	  \multicolumn{1}{l|}{3} &
	  0.18 \\ \hline
	\multicolumn{1}{|l|}{$u_1u_2, u_2u_4$} &
	  \multicolumn{1}{l|}{0.162} &
	  \multicolumn{1}{l|}{$u_1$} &
	  \multicolumn{1}{l|}{2} &
	  1.782 &
	  \multicolumn{1}{l|}{$u_1$, $u_3$} &
	  \multicolumn{1}{l|}{3} &
	  1.62 &
	  \multicolumn{1}{l|}{$u_1$, $u_4$} &
	  \multicolumn{1}{l|}{3} &
	  1.62 \\ \hline
	\multicolumn{1}{|l|}{$u_1u_2, u_2u_3$} &
	  \multicolumn{1}{l|}{0.042} &
	  \multicolumn{1}{l|}{$u_1$} &
	  \multicolumn{1}{l|}{2} &
	  0.462 &
	  \multicolumn{1}{l|}{$u_1$, $u_3$} &
	  \multicolumn{1}{l|}{3} &
	  0.42 &
	  \multicolumn{1}{l|}{$u_1$, $u_4$} &
	  \multicolumn{1}{l|}{3} &
	  0.42 \\ \hline
	\multicolumn{1}{|l|}{$u_1u_2, u_2u_3, u_2u_4$} &
	  \multicolumn{1}{l|}{0.378} &
	  \multicolumn{1}{l|}{$u_1$} &
	  \multicolumn{1}{l|}{2} &
	  4.158 &
	  \multicolumn{1}{l|}{$u_1$, $u_3$} &
	  \multicolumn{1}{l|}{3} &
	  3.78 &  \multicolumn{1}{l|}{$u_1$, $u_4$} &
	  \multicolumn{1}{l|}{3} &
	  3.78 \\ \hline
	\multicolumn{4}{|l|}{\textbf{$\displaystyle{\sum_{\mathcal{G}}}P(\mathcal{G}){\phi^{\mathcal{G}}(\mathcal{S}_{1} \cup \mathcal{S}_{2}^{OPT(\mathcal{G},\mathcal{S}_{1},d,B_{2})})}$}} &
	  9.742 &
	  \multicolumn{2}{l|}{\textbf{$\displaystyle{\sum_{\mathcal{G}}}P(\mathcal{G}){\phi^{\mathcal{G}}(\mathcal{S}_{1} \cup \mathcal{S}_{2}^{OPT(\mathcal{G},\mathcal{S}_{1},d,B_{2})})}$}} &
	  9.96 &
	  \multicolumn{2}{l|}{\textbf{$\displaystyle{\sum_{\mathcal{G}}}P(\mathcal{G}){\phi^{\mathcal{G}}(\mathcal{S}_{1} \cup \mathcal{S}_{2}^{OPT(\mathcal{G},\mathcal{S}_{1},d,B_{2})})}$}} &
	  8.8 \\ \hline
	\end{tabular}%
	}
	\end{table*}
	
	Let us consider another example graph $G$ shown in Figure \ref{Fig3:Non-MonotoneProperty} to calculate $f(\mathcal{S}_{1})$ where $\mathcal{S}_{1} = \{u_1\}$ at timestep $d = 1$ unit, budget $B_{1} = 2$ units and total budget $B = 3$ units. Now, we calculate the value of our objective function as $9.742$ in Table \ref{Table4:Non-monotonicity}.
	Similarly, the value of objective function calculated for $f(\mathcal{S}_{1}) \cup \{u_4\})$ while keeping remaining parameter settings same, is $8.8$ in Table \ref{Table4:Non-monotonicity}. 
	
	Since, $f(\mathcal{S}_{1}) > f(\mathcal{S}_{1} \cup \{u_4\})$ and not following inequality \ref{eq:mon_increasing}. Therefore, it proves that our objective function $f(\mathcal{S}_{1})$ is not a monotonically increasing function.
	
	Now, we have calculated $f(\mathcal{S}_{1} \cup \{u_3\}) = 9.96$ in Table \ref{Table4:Non-monotonicity} with same values of $d$, $B_{1}$ and $B$. Since, $f(\mathcal{S}_{1}) < f(\mathcal{S}_{1} \cup \{u_3\})$ and following inequality \ref{eq:mon_decreasing}. Hence, it proves that our objective function $f(\mathcal{S}_{1})$ is not a monotonically decreasing function.
	
	Therefore, Lemma \ref{Lemma2:Monotonicity} is proved. For a given $G$, $f(\mathcal{S}_{1})$ is neither monotonically increasing nor monotonically decreasing.
	        
	\end{proof}
	
	\begin{mylemma}
		$f(\cdot)$ is neither submodular nor supermodular.
	            \label{Lemma:Submodular}
	\end{mylemma}
	\begin{proof}
		Let two sets $\mathcal{S}$ and $\mathcal{T}$, such that $\forall \mathcal{S} \subset \mathcal{T} \subset V(G)$. \\
		$f(\cdot)$ is submodular if 
		\begin{equation}
			f(\mathcal{S} \cup \{i\}) - f(\mathcal{S}) \geq f(\mathcal{T} \cup \{i\}) - f(\mathcal{T})
			\label{eq:submodular}
		\end{equation}
		$f(\cdot)$ is supermodular if 
		\begin{equation}
			f(\mathcal{S} \cup \{i\}) - f(\mathcal{S}) \leq f(\mathcal{T} \cup \{i\}) - f(\mathcal{T})
			\label{eq:nonsubmodular}
		\end{equation}

	\begin{table*}
	\centering
	\caption{Calculating $f(\mathcal{S}_{1})$ to show it's neither submodular nor supermodular}
	\label{Table5:Submodularity}
	\resizebox{\columnwidth}{!}{
	\begin{tabular}{|lllll|lll|lll|}
	\hline
	\multicolumn{2}{|l|}{} &
	  \multicolumn{3}{c|}{$\mathcal{S}_{1} = \{u_3\}$} &
	  \multicolumn{3}{c|}{$\mathcal{S}_{1} = \{u_1, u_2\}$} &
	  \multicolumn{3}{c|}{$\mathcal{S}_{1} = \{u_1, u_2, u_4\}$} \\ \hline
	\multicolumn{1}{|l|}{$\mathcal{G}$} &
	  \multicolumn{1}{l|}{$P(\mathcal{G})$} &
	  \multicolumn{1}{l|}{$\mathcal{S}_{1}$} &
	  \multicolumn{1}{l|}{$C(\mathcal{S}_{1} \cup \mathcal{S}_{2})$} &
	  $f(\mathcal{S}_{1})$ &
	  \multicolumn{1}{l|}{$\mathcal{S}_{1}$} &
	  \multicolumn{1}{l|}{$C(\mathcal{S}_{1} \cup \mathcal{S}_{2})$} &
	  $f(\mathcal{S}_{1})$ &
	  \multicolumn{1}{l|}{$\mathcal{S}_{1}$} &
	  \multicolumn{1}{l|}{$C(\mathcal{S}_{1} \cup \mathcal{S}_{2})$} &
	  $f(\mathcal{S}_{1})$ \\ \hline
	\multicolumn{1}{|l|}{\{\}} &
	  \multicolumn{1}{l|}{0.012} &
	  \multicolumn{1}{l|}{$u_3$} &
	  \multicolumn{1}{l|}{2} &
	  0.108 &
	  \multicolumn{1}{l|}{$u_1$, $u_2$} &
	  \multicolumn{1}{l|}{3} &
	  0.108 &
	  \multicolumn{1}{l|}{$u_1$, $u_2$, $u_4$} &
	  \multicolumn{1}{l|}{4} &
	  0.12 \\ \hline
	\multicolumn{1}{|l|}{$u_2u_4$} &
	  \multicolumn{1}{l|}{0.108} &
	  \multicolumn{1}{l|}{$u_3$} &
	  \multicolumn{1}{l|}{2} &
	  1.08 &
	  \multicolumn{1}{l|}{$u_1$, $u_2$} &
	  \multicolumn{1}{l|}{3} &
	  1.08 &
	  \multicolumn{1}{l|}{$u_1$, $u_2$, $u_4$} &
	  \multicolumn{1}{l|}{4} &
	  1.08 \\ \hline
	\multicolumn{1}{|l|}{$u_2u_3$} &
	  \multicolumn{1}{l|}{0.028} &
	  \multicolumn{1}{l|}{$u_3$} &
	  \multicolumn{1}{l|}{2} &
	  0.252 &
	  \multicolumn{1}{l|}{$u_1$, $u_2$} &
	  \multicolumn{1}{l|}{3} &
	  0.28 &
	  \multicolumn{1}{l|}{$u_1$, $u_2$, $u_4$} &
	  \multicolumn{1}{l|}{3} &
	  0.28 \\ \hline
	\multicolumn{1}{|l|}{$u_2u_3$, $u_2u_4$} &
	  \multicolumn{1}{l|}{0.252} &
	  \multicolumn{1}{l|}{$u_3$} &
	  \multicolumn{1}{l|}{2} &
	  2.52 &
	  \multicolumn{1}{l|}{$u_1$, $u_2$} &
	  \multicolumn{1}{l|}{2} &
	  2.772 &
	  \multicolumn{1}{l|}{$u_1$, $u_2$, $u_4$} &
	  \multicolumn{1}{l|}{3} &
	  2.52 \\ \hline
	\multicolumn{1}{|l|}{$u_1u_2$} &
	  \multicolumn{1}{l|}{0.018} &
	  \multicolumn{1}{l|}{$u_3$} &
	  \multicolumn{1}{l|}{2} &
	  0.162 &
	  \multicolumn{1}{l|}{$u_1$, $u_2$} &
	  \multicolumn{1}{l|}{3} &
	  0 & 
	  \multicolumn{1}{l|}{$u_1$, $u_2$, $u_4$} &
	  \multicolumn{1}{l|}{4} &
	  0.162 \\ \hline
	\multicolumn{1}{|l|}{$u_1u_2$, $u_2u_4$} &
	  \multicolumn{1}{l|}{0.162} &
	  \multicolumn{1}{l|}{$u_3$} &
	  \multicolumn{1}{l|}{2} &
	  1.62 &
	  \multicolumn{1}{l|}{$u_1$, $u_2$} &
	  \multicolumn{1}{l|}{3} &
	  1.62 &
	  \multicolumn{1}{l|}{$u_1$, $u_2$, $u_4$} &
	  \multicolumn{1}{l|}{4} &
	  1.458 \\ \hline
	\multicolumn{1}{|l|}{$u_1u_2$, $u_2u_3$} &
	  \multicolumn{1}{l|}{0.042} &
	  \multicolumn{1}{l|}{$u_3$} &
	  \multicolumn{1}{l|}{2} &
	  0.378 &
	  \multicolumn{1}{l|}{$u_1$, $u_2$} &
	  \multicolumn{1}{l|}{3} &
	  0.42 &
	  \multicolumn{1}{l|}{$u_1$, $u_2$, $u_4$} &
	  \multicolumn{1}{l|}{3} &
	  0.42 \\ \hline
	\multicolumn{1}{|l|}{$u_1u_2$, $u_2u_3$, $u_2u_4$} &
	  \multicolumn{1}{l|}{0.378} &
	  \multicolumn{1}{l|}{$u_3$} &
	  \multicolumn{1}{l|}{2} &
	  3.78 &
	  \multicolumn{1}{l|}{$u_1$, $u_2$} &
	  \multicolumn{1}{l|}{2} &
	  4.158 &
	  \multicolumn{1}{l|}{$u_1$, $u_2$, $u_4$} &
	  \multicolumn{1}{l|}{3} &
	  3.78 \\ \hline
	\multicolumn{4}{|l|}{\textbf{$\displaystyle{\sum_{\mathcal{G}}}P(\mathcal{G}){\phi^{\mathcal{G}}(\mathcal{S}_{1} \cup \mathcal{S}_{2}^{OPT(\mathcal{G},\mathcal{S}_{1},d,B_{2})})}$}} &
	  9.9 &
	  \multicolumn{2}{l|}{\textbf{$\displaystyle{\sum_{\mathcal{G}}}P(\mathcal{G}){\phi^{\mathcal{G}}(\mathcal{S}_{1} \cup \mathcal{S}_{2}^{OPT(\mathcal{G},\mathcal{S}_{1},d,B_{2})})}$}} &
	  10.438 &
	  \multicolumn{2}{l|}{\textbf{$\displaystyle{\sum_{\mathcal{G}}}P(\mathcal{G}){\phi^{\mathcal{G}}(\mathcal{S}_{1} \cup \mathcal{S}_{2}^{OPT(\mathcal{G},\mathcal{S}_{1},d,B_{2})})}$}} &
	  9.7 \\ \hline
	\end{tabular}%
	}
	\end{table*}
	
	Considering the example graph in Figure \ref{Fig3:Non-MonotoneProperty}, with its cost and benefit values, however, parameters are set as timestep $d = 1$, budget for phase one $B_{1} = 3$, and total budget $B = 4$.

	Let $\mathcal{S} = \{u_1\}$, $i = u_4$, and $f(\mathcal{S} \cup \{i\}) = f(\{u_1, u_4\})$.
	So, $f(\mathcal{S}) = 9.742$ and $f(\mathcal{S} \cup \{i\}) = 8.8$, both are calculated in Table \ref{Table4:Non-monotonicity}.
	
	Now, let $\mathcal{T} = \{u_1, u_2\}$, $i = u_4$, and $f(\mathcal{T} \cup \{i\}) = f(\{u_1, u_2\} \cup \{u_4\})$.
	So, $f(\mathcal{T}) = 10.438$ and $f(\mathcal{T} \cup \{i\}) = 9.7$, both are calculated in Table \ref{Table5:Submodularity}.	These values are not satisfying the inequality \ref{eq:submodular}. Hence, it proves that $f(\cdot)$ is not submodular.	
		
	Likewise, let $\mathcal{S} = \{\}$, $i = u_3$, and $f(\mathcal{S} \cup \{i\}) = f(\{u_3\})$. 
	So, $f(\mathcal{S}) = 0$ and $f(\mathcal{S} \cup \{i\}) = 9.9$ calculated in Table \ref{Table5:Submodularity}.
	
	Now, let $\mathcal{T} = \{u_1\}$, $i = u_3$, and $f(\mathcal{T} \cup \{i\}) = f(\{u_1 \cup u_3\})$. 
	So, $f(\mathcal{T}) = 9.742$ and $f(\mathcal{T} \cup \{i\}) = 9.96$ calculated in Table \ref{Table4:Non-monotonicity} and Table \ref{Table5:Submodularity}, respectively. These values are not satisfying the inequality \ref{eq:nonsubmodular}. Hence, it proves that $f(\cdot)$ is not supermodular.
		
	Therefore, Lemma \ref{Lemma:Submodular} is proved. For a given $G$, $f(\cdot)$ is neither submodular nor supermodular.
	\end{proof}
	
	\begin{mylemma}
		$f(\cdot)$ is subadditive.
		\label{lem: subadditive}
	\end{mylemma}
	\begin{proof}
		Let two sets $M$ and $N$, such that $\forall M, N \subseteq V(G)$.\\ $f(\cdot)$ is subadditive if 
		\begin{equation}
			f(M \cup N) \leq f(M) + f(N)
			\label{eq: subadditive}
		\end{equation}
		Let $V_{1} = M_{1} \cup N_{1}$ and an optimal set $V_{2}^{O(X, V_{1}, d, k_{2})}$ of $k_{2}$ nodes. We can start Phase I of the diffusion process by seed set $V_{1}$.\\
		{\scriptsize  
	 \begin{align}
	  f(M_{1}) + f(N_{1}) &=  \displaystyle{\sum_{\mathcal{G}}}P(\mathcal{G})\big\{{\phi^{\mathcal{G}}(M_{1} \cup M_{2}^{OPT(X,M_{1},d,B_{2})})} + {\phi^{\mathcal{G}}(N_{1} \cup N_{2}^{OPT(X,N_{1},d,B_{2})})}\big\} \label{eq:Subadditive1}\\
		& \geq \displaystyle{\sum_{\mathcal{G}}}P(\mathcal{G})\big\{{\phi^{\mathcal{G}}(M_{1} \cup V_{2}^{OPT(X,V_{1},d,B_{2})})} + {\phi^{\mathcal{G}}(N_{1} \cup V_{2}^{OPT(X,V_{1},d,B_{2})})}\big\} \label{eq:Subadditive2}\\
		& \geq \displaystyle{\sum_{\mathcal{G}}}P(\mathcal{G}){\phi^{\mathcal{G}}(M_{1} \cup N_{1} \cup V_{2}^{OPT(X,V_{1},d,B_{2})})} \label{eq:Subadditive3}\\
		& = \displaystyle{\sum_{\mathcal{G}}}P(\mathcal{G}){\phi^{\mathcal{G}}(V_{1} \cup V_{2}^{OPT(X,V_{1},d,B_{2})})} \label{eq:Subadditive4}\\
		& = f(V_{1}) \nonumber \\
	 & = f(M_{1} \cup N_{1}) \label{eq:Subadditive5}
	 \end{align}
	}	
	The inequality \ref{eq:Subadditive2} is from the optimality of the set of Phase II. 
	The inequality \ref{eq:Subadditive3} is from the subadditivity property of the set function. 
	\end{proof}

	\section{Proposed Solution Approaches} \label{Sec:Solution}
	In this section, we describe our proposed solution approaches for the problem. First, we introduce the notion of Marginal Profit Gain in Definition \ref{Def:Marginal_Profit_Gain} which has been used subsequently.  
	
	\begin{mydef}[Marginal Profit Gain]\label{Def:Marginal_Profit_Gain}
		Given a social network $G(V, E, \mathcal{P})$, a seed set $\mathcal{S}$, and a node $u \in V(G) \setminus \mathcal{S}$ we denote the Marginal Profit Gain for the node $u$ with respect to the seed set $\mathcal{S}$ as $\phi_u(\mathcal{S})$ and it is defined as the difference of profit earned when $u$ is added to $\mathcal{S}$ and when $u$ is not in $\mathcal{S}$. As shown in Lemma \ref{Lemma2:Monotonicity} that the profit function $\phi()$ may be non-monotone as well. As we are considering the `gain', for any seed set $\mathcal{S}$ and node $u \in V(G) \setminus \mathcal{S}$, $\phi_u(\mathcal{S})$ is only defined only when $\phi_u(\mathcal{S}) > \phi(\mathcal{S})$. Mathematically, this can be defined using Equation \ref{Eq:Marginal_Profit_Gain}.
		
		\begin{equation} \label{Eq:Marginal_Profit_Gain}
			\phi_u(\mathcal{S}) = \phi(\mathcal{S} \cup \{u\}) - \phi(\mathcal{S}) \quad such \; that \; \phi_u(\mathcal{S}) > \phi(\mathcal{S})
		\end{equation}
	\end{mydef}
	
	Next, we proceed to describe our first solution methodology, which is  the `Simple Greedy Approach'. 
	
	\subsection{Simple Greedy Approach}
	The key intuition of the simple greedy approach is to build the seed set for both phases in an iterative greedy manner. In every iteration, we pick the node that causes maximum marginal profit gain to its cost ratio and include it in the seed set. This process stops when no more node selection is possible. Algorithm \ref{Algo: Algorithm1_Simple Greedy} shows the pseudocode of this solution methodology. Next, we describe the working principle of this algorithm in detail.
	
	\paragraph{Description of Algorithm \ref{Algo: Algorithm1_Simple Greedy}}	This algorithm takes inputs as the social network $G$, the budget for both the phases $B_1$ and $B_2$, and the duration of first phase $d$. This algorithm outputs the seed set for both the phases and the seed set as a whole. Accordingly, we initialize three sets $\mathcal{S}$, $\mathcal{S}_{1}$, and $\mathcal{S}_{2}$ in Line Number $1$, $3$ and $19$, respectively. Now, we begin the first phase. Until the budget $B_{1}$ is exhausted, for every non-seed node $u$, we calculate the marginal profit gain to its cost ratio, and the node results in the maximum value of the quantity chosen. If its marginal profit gain is strictly positive, then it is included in the seed set of the first phase else; we just execute \texttt{break} statement. It may so happen that the allocated budget for the first phase has not been totally exhausted. If so, we calculate the remaining budget after the first phase from Line $13$ to $15$ to utilize it further in the Second Phase. The seed set $\mathcal{S}_{1}$ obtained from line $16$ completes the First Phase. Now, we conduct a diffusion process based on the IC Model starting from the seed set $\mathcal{S}_{1}$ till $d$-th timestep. Thus, at the end of the first phase, we have $A_{Y}$ and $R_{Y}$ as the influence of $\mathcal{S}_{1}$ at timestep $d$. Now, we begin our second phase with the updated budget as shown in Line $20$. The process of seed set selection is quite similar to the first phase, with just one difference. During the first phase, we deal with the entire social network, and hence, while computing the marginal profit gain, we consider the whole network. However, for the second phase, we deal with the network obtained by deleting the already activated nodes from the original network. Accordingly, during the second phase, while computing the marginal profit gain, we consider the remaining network.
	
	\paragraph{Complexity Analysis of Algorithm \ref{Algo: Algorithm1_Simple Greedy}}	Now, we analyze Algorithm \ref{Algo: Algorithm1_Simple Greedy} to understand its time and space requirement. As the statements mentioned in Line $1$ and $3$ are initialization statements, hence time requirement to execute them is $\mathcal{O}(1)$. Now, it is an important question how many times the \texttt{while} loop from Line $4$ to $12$ will be executed? One point to observe is that in each iteration, we select at most one node. Assume that the minimum cost of a node in the network is $C_{min}$, i.e., $C_{min}=\underset{u \in V(G)}{min} \ C(u)$. Now, it is easy to observe that with budget $B_{1}$, we can select at most $\mathcal{O}(\frac{B_{1}}{C_{min}})$ many nodes. So the \texttt{while} loop from Line $4$ to $12$ will execute at most $\mathcal{O}(\frac{B_{1}}{C_{min}})$ times. Now, the first step inside the loop is the marginal profit gain computation. It is easy to observe that in every iteration, the number of nodes for which the marginal profit gain needs to be computed is of $\mathcal{O}(n)$. Now, for a given seed node computing, the marginal gain is equivalent to traversing the graph, and this takes $\mathcal{O}(m+n)$ time. In the worst case, the size of $\mathcal{S}_{1}$ can be of $\mathcal{O}(n)$. Hence, for one marginal gain computation time requirement is $\mathcal{O}(n(m+n))$. As there are $\mathcal{O}(n)$ many marginal profit gain computations, hence the time requirement for this purpose is $\mathcal{O}(n^{2}\cdot(m+n))$. Now, choosing the node that causes the maximum marginal profit gain takes $\mathcal{O}(n)$ time. So, the time requirement to execute the Line $5$ is of $\mathcal{O}(n^{2} \cdot (m+n) + n)$. Now, it is easy to observe that the computations performed from Line $6$ to $12$ will take $\mathcal{O}(1)$ time. So, the requirement for the first phase is of $\mathcal{O}(1+ \frac{B_{1}}{C_{min}} \cdot (n^{2} \cdot (m+n)+n+1)) \approx \mathcal{O}(\frac{B_{1}}{C_{min}} \cdot n^{2} \cdot (m+n))$. Between the first and second phase, there are two steps. The first one is to conduct the diffusion process till timestep $d$ with the seed set of the first phase. In the worst case, this step may take $\mathcal{O}(m+n)$ time. The second step is to delete the already activated nodes from the network. In the worst case, the number of already activated nodes can be of $\mathcal{O}(n)$, and hence, the number of edges incident with these vertices can be of $\mathcal{O}(n^{2})$. Assuming the graph is given in the form of \emph{adjacency matrix}, deleting these nodes from the graph requires $\mathcal{O}(n^{2})$ time.  Now, for the second phase, the analysis will be the same as the first phase, with two exceptions. The first one is the minimum cost of the nodes. As in the second phase, we are dealing with the graph obtained after deleting the already activated nodes from the first phase, and it may so happen that the cost of all the nodes in the reduced graph is more than $C_{min}$. Assume that the minimum cost among all the nodes in the reduced graph is $C_{min}^{'}$, i.e., $C_{min}^{'}=\underset{u \in V(G) \setminus A^{Y}}{min} \ C(u)$. The second one is the number of nodes and edges in the reduced graph. However, in the worst case, these will be $n$ and $m$, respectively. Hence, the time requirement for the second phase is of $\mathcal{O}(\frac{B_{2}}{C_{min}^{'}} \cdot n^{2} \cdot (m+n))$. Finally, the total time requirement of Algorithm \ref{Algo: Algorithm1_Simple Greedy} is of $\mathcal{O}(\frac{B_{1}}{C_{min}} \cdot n^{2} \cdot (m+n) + n + n^{2}+ \frac{B_{2}}{C_{min}^{'}} \cdot n^{2} \cdot (m+n)) \approx \mathcal{O}( (\frac{B_{1}}{C_{min}} + \frac{B_{2}}{C_{min}^{'}}) \cdot n^{2} \cdot (m+n))$. Now, the extra space consumed by Algorithm \ref{Algo: Algorithm1_Simple Greedy} is to store the seed sets $\mathcal{S}$, $\mathcal{S}_{1}$, $\mathcal{S}_{2}$, already and recently activated nodes, and in the worst case all of them may contain $\mathcal{O}(n)$ many nodes. Hence, Theorem \ref{Th:Simple_Greedy_Approach} holds.

	\begin{mytheorem} \label{Th:Simple_Greedy_Approach}
	Running time and space requirement of the proposed simple greedy approach (Algorithm \ref{Algo: Algorithm1_Simple Greedy}) is of $\mathcal{O}( (\frac{B_{1}}{C_{min}} + \frac{B_{2}}{C_{min}^{'}}) \cdot n^{2} \cdot (m+n))$ and $\mathcal{O}(n)$, respectively.
	
	\end{mytheorem}
	{\scriptsize
	\begin{algorithm}[h]
	\SetAlgoLined
	\KwData{$G, B_{1}, B_{2}, d$}
	\KwResult{  $\mathcal{S}$}
	 Initialize $\mathcal{S} \leftarrow \emptyset$ \;
	 \textbf{Seed set selection for First Phase}\;
	 Initialize $\mathcal{S}_{1} \leftarrow \emptyset$\;
	 \While{$\texttt{TRUE}$}{
	 Find $u' \leftarrow {argmax}_{u \in V(G) \setminus \mathcal{S}_{1}} {\frac{\phi(\mathcal{S}_{1} \cup \{u\}) - \phi(\mathcal{S}_{1})}{C(u)}}$\;
	 \If{$(\phi_{u'}(\mathcal{S}_{1})) \leq 0$}{
	 Break\;
	 }
	 \If{$C(u^{'}) \leq B_{1}$}{
	 $\mathcal{S}_{1} \longleftarrow \mathcal{S}_{1} \cup \{u'\}$\;
	 $B_{1} \longleftarrow B_{1} - C(u^{'})$\;
	 }
	 }
	Return $\mathcal{S}_{1}$\;
	From the partial observation in $G$ using seed set $\mathcal{S}_{1}$ at timestep $d$, we have recently activated nodes, $R_{Y}$ and already activated nodes  $A_{Y}$\;
	
	\textbf{Seed set selection for Second Phase}\;
	
	Initialize $\mathcal{S}_{2} \leftarrow \emptyset$\;
	Update $B_{2} \leftarrow B_{2} + B_{1}$\;
	\While{$\texttt{TRUE}$}{
	
		Find $v' \leftarrow {argmax}_{v \in V(G) \setminus A_{Y}}{\frac{\phi_{V(G) \setminus A_{Y}}(\mathcal{S}_{2} \cup \{v\}) - \phi_{V(G) \setminus A_{Y}}(\mathcal{S}_{2})}{C(v)}}$\;
		\If{$(\phi_{v'}(\mathcal{S}_{2})) \leq 0$}{
			Break\;
		}
		\If{$C(v^{'}) \leq B_{2}$}{
		$S_{2} \longleftarrow S_{2} \cup \{v'\}$\;
		$B_{2} \longleftarrow B_{2} - C(v^{'})$\;
		}	
	}
	Return $\mathcal{S}_{2}$\;
	$\mathcal{S} \longleftarrow \mathcal{S}_{1} \cup  \mathcal{S}_{2}$\;
	Return $\mathcal{S}$\;
	 \caption{ Simple Greedy Algorithm for Two Phase Profit Maximization Problem}
	 \label{Algo: Algorithm1_Simple Greedy}
	\end{algorithm}	
	}	
	
	\subsection{Double Greedy Approach}	
	
	\paragraph{Description of Algorithm \ref{Algo: Algorithm2_Double Greedy}}
	First phase of this method goes like this. We initialize two sets $\mathcal{S}_{1}$ and $\mathcal{T}_{1}$. The first one is with $\emptyset$, and the second one is with $V(G)$. Now, for every node $u \in V(G)$, we compute two measures $r^{+}_{u}$ and $r^{-}_{u}$ which are mentioned in Equation \ref{Eq:7} and \ref{Eq:8}, respectively.
	 
	\begin{minipage}[t]{0.45\textwidth}
	\begin{equation} \label{Eq:7}
	r^{+}_{u} \leftarrow \frac{\phi(\mathcal{S}_{1} \cup \{u\}) - \phi(\mathcal{S}_{1})}{C(u)}
	\end{equation}
	\end{minipage}
	\begin{minipage}[t]{0.45\textwidth}
	\begin{equation} \label{Eq:8}
	r^{-}_{u} \leftarrow - \frac{\phi(\mathcal{T}_{1} \setminus \{u\}) - \phi(\mathcal{T}_{1})}{C(u)}
	\end{equation}
	\end{minipage} 
	
	 Now, if $r^{+}_{u} \geq r^{-}_{u}$ and the cost of the current nodes is less than the available budget, then the set $\mathcal{S}_{1}$ is updated as $\mathcal{S}_{1} \cup \{u\}$ and the budget $B_{1}$ is updated as $B_{1}- \mathcal{C}(u)$, though $\mathcal{T}_{1}$ remains the same. If the budget is not sufficient or if $r^{+}_{u} < r^{-}_{u}$, $\mathcal{T}_{1}$ is reduced by deleting the current node and in that case $\mathcal{S}_{1}$ remains the same. After repeating these steps, we obtain the seed set for Phase I, i.e.; $\mathcal{S}_{1}$. Now, we conduct the diffusion process and observe up to the timestep $d$ and thus obtain the already activated nodes and recently activated nodes. If there is any unutilized budget of Phase I, that has been added to the budget of Phase II. For the seed set selection of Phase II, we repeat the same process however on the reduced graph, i.e.; the graph obtained by deleting the already activated nodes in the first phase. Now, we proceed to describe the analysis of the double greedy approach. \\

	{\scriptsize
	\begin{algorithm}
	\SetAlgoLined
	\KwData{$G, B_{1}, B_{2}, d$}
	\KwResult{$\mathcal{S}$}
	 Initialize $\mathcal{S} \leftarrow \emptyset$ \;
	 \textbf{Seed set selection for First Phase}\;
	 Initialize $\mathcal{S}_{1} \leftarrow \emptyset$, $\mathcal{T}_{1} \leftarrow V(G)$\;
	 
	 \For{$\text{All } u \in V(G)$}{
	$ r^{+}_{u} \leftarrow \frac{\phi(\mathcal{S}_{1} \cup \{u\}) - \phi(\mathcal{S}_{1})}{C(u)}$\;
	$ r^{-}_{u} \leftarrow \: -{\frac{\phi(\mathcal{T}_{1} \setminus \{u\}) - \phi(\mathcal{T}_{1})}{C(u)}}$\;
	\eIf{$r^{+}_{u} \geq r^{-}_{u}$}{
	\eIf{$C(u) \leq B_{1}$}{
	$\mathcal{S}_{1} \leftarrow \mathcal{S}_{1} \cup \{u\}$; $\mathcal{T}_{1} \: remains \: same$\;
	$B_{1} \longleftarrow B_{1} - C(u)$\;
	}{
	
	$\mathcal{T}_{1} \leftarrow \mathcal{T}_{1} \setminus \{u\}$; $\mathcal{S}_{1} \: remains \: same$\;
	}
	
	}
	{
	$\mathcal{T}_{1} \leftarrow \mathcal{T}_{1} \setminus \{u\}$; $\mathcal{S}_{1} \: remains \: same$\;
	
	}

	}
	Return $\mathcal{S}_{1} (=\mathcal{T}_{1})$\;
	From the partial observation in $G$ using seed set $\mathcal{S}_{1}$ at timestep $d$, we have recently activated nodes, $R_{Y}$ and already activated nodes  $A_{Y}$\;
	
	\textbf{Seed set selection for Second Phase}\;
	
	Initialize $\mathcal{S}_{2} \leftarrow \emptyset$,  $\mathcal{T}_{2} \leftarrow V(G) \setminus A_{Y}$\;
	Update $B_{2} \leftarrow B_{2} + B_{1}$\;
	
	 \For{$\text{All } v \in V(G) \setminus A_{Y}$}{
	 
	$ r^{+}_{v} \leftarrow \frac{\phi_{V(G) \setminus A_{Y}}(\mathcal{S}_{1} \cup \{v\}) - \phi_{V(G) \setminus A_{Y}}(\mathcal{S}_{1})}{C(v)}$\;
	$ r^{-}_{v} \leftarrow -{\frac{\phi_{V(G) \setminus A_{Y}}(\mathcal{T}_{1} \setminus \{v\}) - \phi_{V(G) \setminus A_{Y}}(\mathcal{T}_{1})}{C(v)}}$\;
	
	\eIf{$r^{+}_{v} \geq r^{-}_{v}$}{
	\eIf{$C(v) \leq B_{2}$}{
	$\mathcal{S}_{2} \leftarrow \mathcal{S}_{2} \cup \{v\}$; $\mathcal{T}_{2} \: remains \: same$\;
	$B_{2} \longleftarrow B_{2} - C(v)$\;
	}{
	
	$\mathcal{T}_{2} \leftarrow \mathcal{T}_{2} \setminus \{v\}$; $\mathcal{S}_{2} \: remains \: same$\;
	
	}
	
	}
	{
	$\mathcal{T}_{2} \leftarrow \mathcal{T}_{2} \setminus \{v\}$; $\mathcal{S}_{2} \: remains \: same$\;
	
	}
	
	}
	Return $\mathcal{S}_{2} (=\mathcal{T}_{2})$\;
	$\mathcal{S} \longleftarrow \mathcal{S}_{1} \cup  \mathcal{S}_{2}$\;
	Return $\mathcal{S}$\;
	 \caption{ Double Greedy Algorithm for Two Phase Profit Maximization Problem}
	 \label{Algo: Algorithm2_Double Greedy}
	\end{algorithm}	
	}
	\paragraph{Complexity Analysis of Algorithm \ref{Algo: Algorithm2_Double Greedy}} The analysis is quite similar to the simple greedy approach in two\mbox{-}phase setting. As stated previously, computing one marginal profit gain computation requires $\mathcal{O}(m+n)$ time. It is important to observe that in this method, we are performing two marginal profit gain computations per node. Hence, the time requirement for the seed set selection of the first phase is $\mathcal{O}(n \cdot (m+n))$. Other than the marginal gain computations, all the remaining statements from Line $7$ to $17$ will take $\mathcal{O}(1)$ time. Now, in the worst case the size of $\mathcal{S}_{1}$, $R_{Y}$, and $A_{Y}$ can be of $\mathcal{O}(n)$. So, there can be $\mathcal{O}(n^{2})$ many edges associated with the vertices of $A_{Y}$. Hence, deleting the set $A_{Y}$ leads to the modification of the $\mathcal{O}(n^{2})$ many adjacency matrix entries of the input social network. Thus, performing the deletion step after Phase I requires $\mathcal{O}(n^{2})$ time. Like Phase I, it is easy to observe that the time requirement for seed set selection in Phase II will be of $\mathcal{O}(n \cdot (m+n))$. Hence, the total time requirement for the double greedy approach in two\mbox{-}phase setting will be of $\mathcal{O}(n \cdot (m+n) + n^{2})= \mathcal{O}(n(m+n))$. The extra space consumed by this algorithm is to store the sets $\mathcal{S}_{1}$, $\mathcal{T}_{1}$, $\mathcal{S}_{2}$, $\mathcal{T}_{2}$, and $\mathcal{S}$. In the worst case, all of them will consume $\mathcal{O}(n)$ space. Hence, Theorem \ref{Th:2} holds.
	
	\begin{mytheorem} \label{Th:2}
	 Running time and the space requirement of the double greedy approach in the two\mbox{-}phase setting is of $ \mathcal{O}(n(m+n))$ and $ \mathcal{O}(n)$, respectively.
	 \end{mytheorem}
\subsection{Improving Scalability of Our Algorithms: Stochastic Greedy Approach}
\label{SubSec:Stochastic}
Now, we describe an approach using which we can make the proposed algorithms more scalable. It can be observed that the key computational overheads are due to the computation of marginal profit gain for all non-seed nodes. To take care of this problem, we have adopted the following approach called \textit{Stochastic Greedy}. We fix a constant $\epsilon$ which takes value in between $0$ and $1$. Now, instead of computing the marginal profit gain for all the non-seed nodes, we randomly sample $\mathcal{O}(\frac{n}{k} \log \frac{1}{\epsilon})$ many nodes from the set of non-seed nodes and we compute the marginal profit gain and subsequent operations for these nodes only. As the number of nodes for which the marginal profit gain is computed is less than the original algorithm, hence the computational time requirement is much less in this case. This approach has been proposed by Mirzasoleiman et al. \cite{mirzasoleiman2015lazier} and has been used by many studies. In this paper, we have adopted this approach to make our algorithms more scalable. Here, the parameter $\epsilon$ is called the accuracy-efficiency trade off parameter. If we increase the value of $\epsilon$, then the sample size will decrease, and hence, the error will increase, though the computational time will decrease. On the other hand, if we increase the value of $\epsilon$, the sample size will increase, hence computational time increases, and the error in computation will decrease. Next, we proceed to describe the experimental evaluation of the  proposed solution approaches.
	
	\section{Experimental Evaluation}  \label{Sec:Experiments}
	In this section, we describe the experimental evaluation of the proposed solution approach. Initially, we start by describing the datasets that we have used in this study.
	\subsection{Datasets Used}

	\begin{itemize}
    \item \textbf{Les Miserables (LM)} 	\cite{dhamal2016information} This undirected network contains co-occurrences of characters in Victor Hugo's novel `Les Misérables'. A node represents a character, and an edge between two nodes shows that these two characters appeared in the same chapter of the book. The weight of each link indicates how often such a co-appearance occurred. The dataset is manipulated to make it directed. To do that, if there is an edge $(u, v)$ then an edge $(v, u)$ is added to the dataset.

	\item \textbf{Email-Eu-Core:} \cite{yin2017local,leskovec2007graph} The network was generated using email data from a large European Research Institution. This dataset has anonymized information about all incoming and outgoing emails between members of the research institution. There is an edge $(u, v)$ in the network if person $u$ sent at least one email to person $v$. 
	
	\item \textbf{Slashdot:} \cite{leskovec2009community} Slashdot is a technology-related news website know for its specific user community. The website features user-submitted and editor-evaluated current primarily technology-oriented news. In 2002 Slashdot introduced the Slashdot Zoo feature which allows users to tag each other as friends or foes. The network contains friend/foe links between the users of Slashdot. The network was obtained in February 2009.
    \end{itemize}
	Table \ref{tab:Datasets} gives the basic statistics of the datasets.

	\begin{table}[h]
	\centering
	\caption{Basic Statistics of the Datasets}
	\begin{tabular}{|c|c|c|c|c|c|}
	\hline
	\textbf{Dataset} & \textbf{Type of} & \textbf{Number of} & \textbf{Number of} & \textbf{Maximum} & \textbf{Average} \\ 
	\textbf{Name}           & \textbf{Graph}      & \textbf{Nodes} & \textbf{Edges}  & \textbf{Degree} & \textbf{Degree}       \\ \hline
	Les Miserables & Directed   & 77    & 508    &    72  & 13.19  \\ \hline	
	Email-Eu-Core & Directed   & 1005    & 25571    &  347    & 33.25 \\ \hline
	Slashdot  & Directed & 82168 & 948464 &   5064  &   23.09      \\ \hline
	\end{tabular}%
	\label{tab:Datasets}
	\end{table}

	\subsection{Experimental Setup}
	Our study has several parameters whose value needs to be set up. The first one of them is the \emph{influence probability}.
	\paragraph{\textbf{Influence Probability}:} 
    In this study, we consider trivalency influence probability setting. In this setting, the influence probability of every edge is uniformly at random from this set $\{0.1, 0.01, 0.001\}$. This probability setting has been extensively used in influence maximization in particular and social network analysis in general.

	\paragraph{\textbf{Cost and Benefit}} The cost and benefit of a node are assigned uniformly at random from the intervals $[50,100]$ and $[800, 1000]$, respectively. 
	\subsection{Algorithms in our Experimentation}
	In our experimentation, the following methods are involved. To make the comparison fair, for the baseline methods, also we use the two phase setting.
	\begin{itemize}
	\item \textbf{Simple Greedy Approach (SG):} This is exactly Algorithm \ref{Algo: Algorithm1_Simple Greedy} of this manuscript where the allocated budget is split into two parts and using the budget of each part based on the marginal gain of the nodes in an incremental way seed nodes are chosen in both the phases.
	
	\item \textbf{Double Greedy Approach (DG):} This is exactly Algorithm \ref{Algo: Algorithm2_Double Greedy} of this manuscript.

    \item \textbf{Stochastic Greedy Approach (StG):} This is explained in subsection \ref{SubSec:Stochastic} of this manuscript.
	
	\item \textbf{Single Discount (SD):} In this method, firstly, we divide the budget into two parts, and they are used for the respective phases. We take the maximum degree node from the graph and decrease the degree of its neighbors by one. Hence, it also means that the maximum degree node is deleted from the graph, which decreases the degree of its neighbors by a value of one. This is repeated til we have exhausted the budget or the remaining budget is not enough to include that node in the seed set.
	
	\item \textbf{Degree Discount (DD):} This method is the same as Single Discount except that the decrease in the degree of neighbors of a node in the seed set is by $dd_{v} = d_{v} - 2t_{v} - (d_{v} - t_{v})t_{v}p$ explained in \cite{chen2009efficient}.
	
	\item \textbf{High Degree (HD):} In this method, after dividing the budget into two phases, we take the maximum degree node from the graph and continue to take the next highest degree node into the seed set till our budget is exhausted, or we don't have enough budget to add one more node to the seed set. 
	
	\item \textbf{High Clustering Coefficient (HighCC):} In this method, divide the budget into two parts for the respective phases. Now, in these phases, as the name suggests, the clustering coefficient of each node is computed and sorted in non-increasing order based on the value. As long as the budget has not been exhausted, the nodes from the sorted list are chosen as seed nodes. 
	
	\item \textbf{Random:} In this method, the budget is divided into two parts, and they are used for the respective phases. As long as the budgets are not exhausted in both phases, the seed nodes are selected uniformly at random. This method has been used by many studies as a baseline method for comparison.
	\end{itemize}

\subsection{Research Questions}
In our experiments, we have addressed the following research questions:
\begin{itemize}
\item \textbf{RQ1. (Single Phase Vs. Two Phase)}: For any particular dataset, a fixed diffusion probability setting and a fixed budget, how much extra percentage of profit can be earned in two phase. 
\item \textbf{RQ2. (Budget Splitting Ratio)}: To make the two phase approach more effective, in what proportion the allocated budget should be splited across first and second phase.  
\item  \textbf{RQ3. (Commencement of Second Phase)}: To make the two phase approach more effective, from what time stamp on wards the second phase should be started. In other words, at what time stamp, the seed set for the second phase should by deployed.
\item \textbf{RQ4. (Size of Seed Set in Single Phase Vs. Two Phase)}:For any particular dataset, a fixed diffusion probability setting and a fixed budget, the size of the seed set selected by different methodologies for single phase Vs. two phase.

\item \textbf{RQ5. (Time Requirement for Diffusion in Single Phase Vs. Two Phase):} For any particular dataset, a fixed diffusion probability setting and a fixed budget, how much extra rounds of diffusion happens in two phases.
\item \textbf{RQ6. (Computational Time Requirement):}  For any particular dataset, a fixed diffusion probability setting and a fixed budget, comparison of computational time requirement for different methodologies. 
\item \textbf{RQ7. (The Impact of $\epsilon$ on Profit in Single Phase Vs. Two Phase)}: For any particular dataset, any diffusion probability setting, and a fixed budget value, what is the impact of $\epsilon$ value.
\end{itemize}
Next, we proceed to describe the experimental results to address these research questions.
\subsection{Experimental Results and Description}
In this section, to evaluate the effectiveness of the proposed two phase profit maximization framework, we analyze the experimental results with respect to seven key research questions (RQ1–RQ7), each addressing a distinct aspect of performance, strategy, or computational trade-off.

\subsubsection{Single Phase Vs. Two Phase in Terms of Profit}
The results for the \textit{LM} dataset using the trivalency probability setting highlight how different algorithms respond to the two phase setting. The baseline algorithms show varied performance, with some consistently improving and others showing mixed outcomes. Among the baselines, \textbf{SD} and \textbf{DD} stand out as strong performers. \textbf{SD} achieves profits ranging from $8292.84$ to a maximum of $37960.78$, with percentage gains reaching $73.10\%$ (Figure~\ref{Fig:RQ1LM_T3}(e)). \textbf{DD} also shows stable behavior, with a top profit of $34613$ and a $22.29\%$ gain at budget $1500$, split ratio $0.1$, and timestep $4$ (Figure~\ref{Fig:RQ1LM_T1}(d)). The \textbf{Random} algorithm, despite its simplicity, achieves the highest relative gain of $111.74\%$ at budget $500$, split ratio $0.7$, and timestep $10$ (Figure~\ref{Fig:RQ1LM_T4}(a)). It shows consistent improvement in all configurations, confirming the broad benefit of the two phase setting. The \textbf{HD} algorithm also performs reliably well, reaching up to $87.54\%$ gain (Figure~\ref{Fig:RQ1LM_T5}(b)). In contrast, the \textbf{HighCC} method is more sensitive to parameter settings. While it achieves a maximum gain of $30.23\%$, it also records a loss of $-26.54\%$ in some settings, as seen in Figures~\ref{Fig:RQ1LM_T1}–\ref{Fig:RQ1LM_T5}. This indicates that its performance depends heavily on the chosen budget, split ratio, and timestep. Among the proposed methods, the \textbf{DG} algorithm stands out. It achieves a maximum profit of $35491.86$ (Figure~\ref{Fig:RQ1LM_T2}(g)) and a profit gain of $74.41\%$. The \textbf{SG} algorithm shows a high-risk, high-reward pattern. It records the lowest profit of $2627.77$ in one case, but also the highest profit of $38312.00$ at budget $2500$, split ratio $0.9$, and timestep $2$, resulting in a gain of $43.26\%$ (Figure~\ref{Fig:RQ1LM_T5}(f)). The third proposed method, \textbf{StG0.1}, also performs well, reaching a maximum profit of $30128.44$ at budget $2500$, split ratio $0.9$, and timestep $6$ (Figure~\ref{Fig:RQ1LM_T5}(h)). In summary, while many baseline methods benefit from the two phase setting, the proposed algorithms—especially \textbf{SG} and \textbf{DG} show greater potential to achieve the highest profits under the right configurations. 

Figures~\ref{Fig:RQ1_T1} to \ref{Fig:RQ1_T5} compare the profit achieved in single-phase and two phase setting for the \textit{Email-Eu-Core} dataset using the trivalency probability setting. These figures show that, in most cases, the two phase approach leads to higher profits across all algorithms evaluated. For example, the baseline \textbf{Random} algorithm achieved its highest profit gain of $44.24\%$ at a budget of $500$, with a split ratio of $0.5$ and a timestep of $10$ (Figure~\ref{Fig:RQ1_T3}(a)). The \textbf{HD} algorithm showed a maximum profit increase of $20.89\%$ under the same split ratio and timestep, but at a budget of $1500$ (Figure~\ref{Fig:RQ1_T3}(b)). The \textbf{HighCC} algorithm showed even larger improvements, with its peak profit gain occurring at budget $500$, split ratio $0.9$, and timestep $10$ (Figure~\ref{Fig:RQ1_T5}(c)). Among other baselines, the \textbf{DD} method reached a maximum gain of $25.37\%$ at a budget of $2500$, split ratio $0.1$, and timestep $10$ (Figure~\ref{Fig:RQ1_T1}(d)). Similarly, \textbf{SD} achieved its best performance with a gain of $26.04\%$ at a budget of $2000$, split ratio $0.1$, and timestep $10$ (Figure~\ref{Fig:RQ1_T1}(e)). For the proposed \textbf{SG} method, one instance with budget $500$, split ratio $0.9$, and timestep $8$ (Figure~\ref{Fig:RQ1_T5}(f)) showed an increase in profit from $325320.1$ (single-phase) to $354941$ (two phase), giving a gain of $9.1\%$. The \textbf{DG} algorithm achieved its highest gain of $32.25\%$ at budget $500$, split ratio $0.9$, and timestep $10$ (Figure~\ref{Fig:RQ1_T5}(g)). The \textbf{StG0.1} approach recorded its best result at budget $2500$, split ratio $0.1$, and timestep $10$, achieving a profit gain of $14.45\%$ (Figure~\ref{Fig:RQ1_T1}(h)). Overall, these results confirm that the two phase approach consistently improves profit across algorithms, and higher budgets generally lead to higher profits in both single and two phase settings.


\begin{figure}[htbp]
\centering
\captionsetup[sub]{font=footnotesize}
\begin{tabular}{cccc}
    \begin{subfigure}[t]{0.22\textwidth}
        \includegraphics[width=\linewidth]{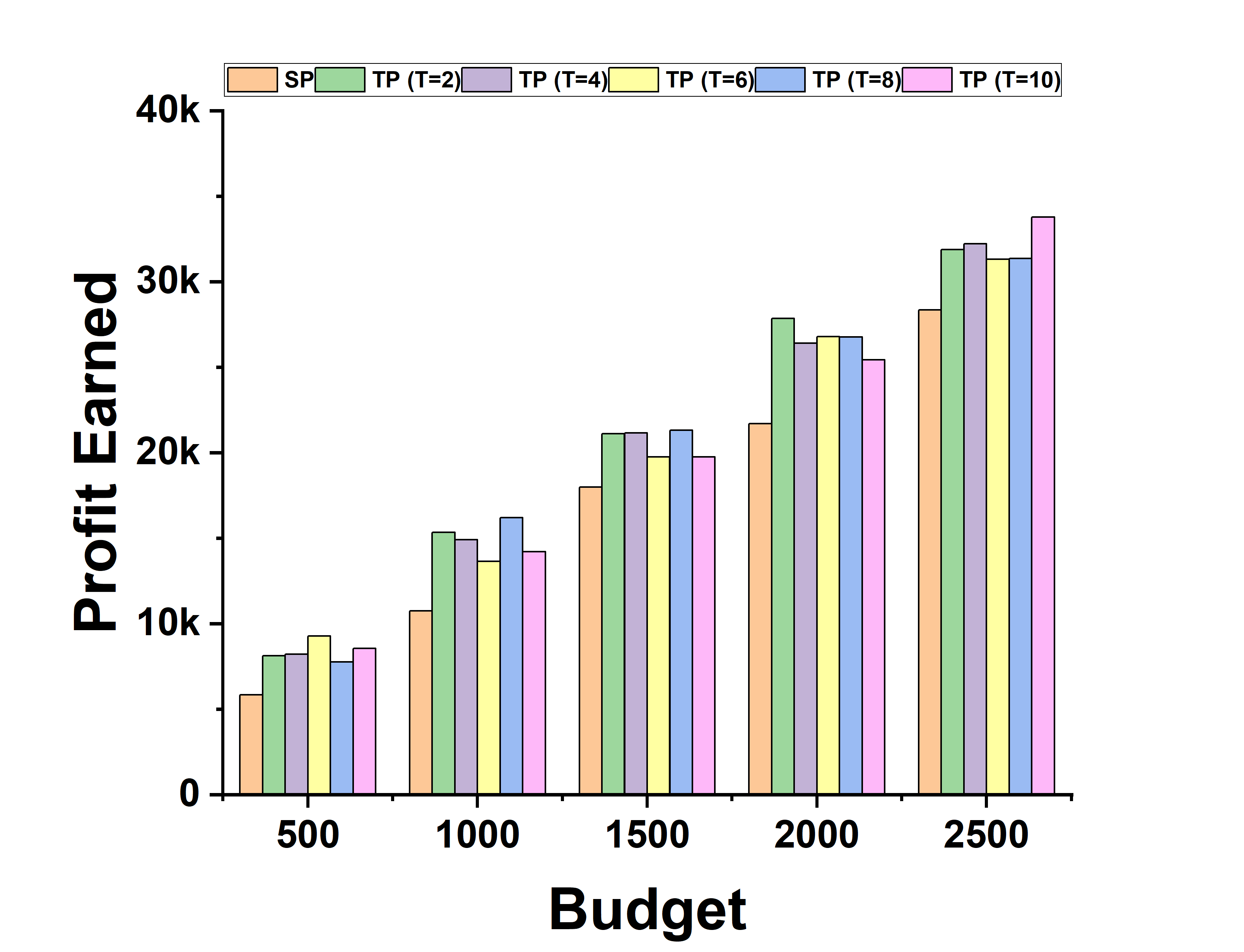}
        \caption{Random}
    \end{subfigure} &
    \begin{subfigure}[t]{0.22\textwidth}
        \includegraphics[width=\linewidth]{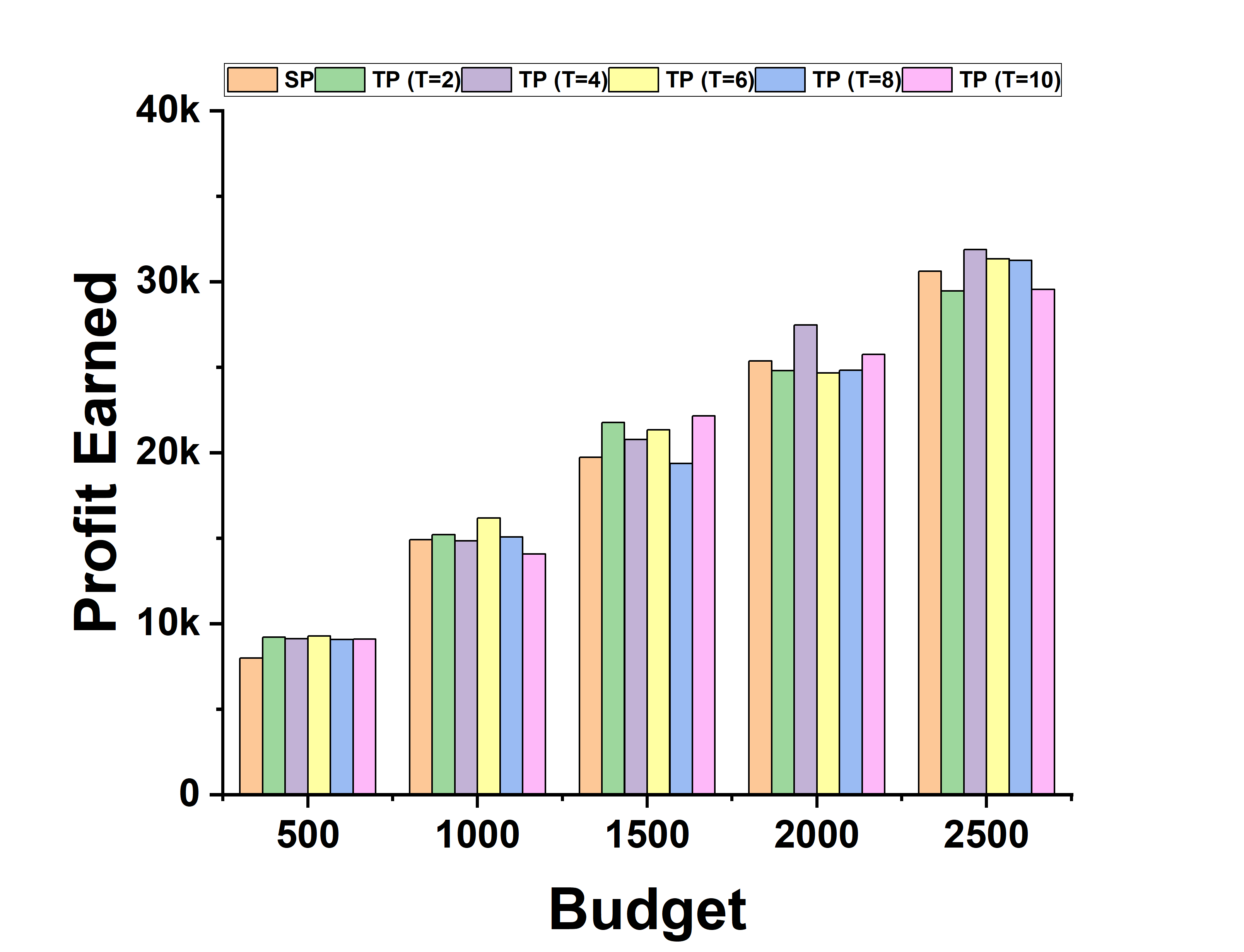}
        \caption{High Degree}
    \end{subfigure} &
    \begin{subfigure}[t]{0.22\textwidth}
        \includegraphics[width=\linewidth]{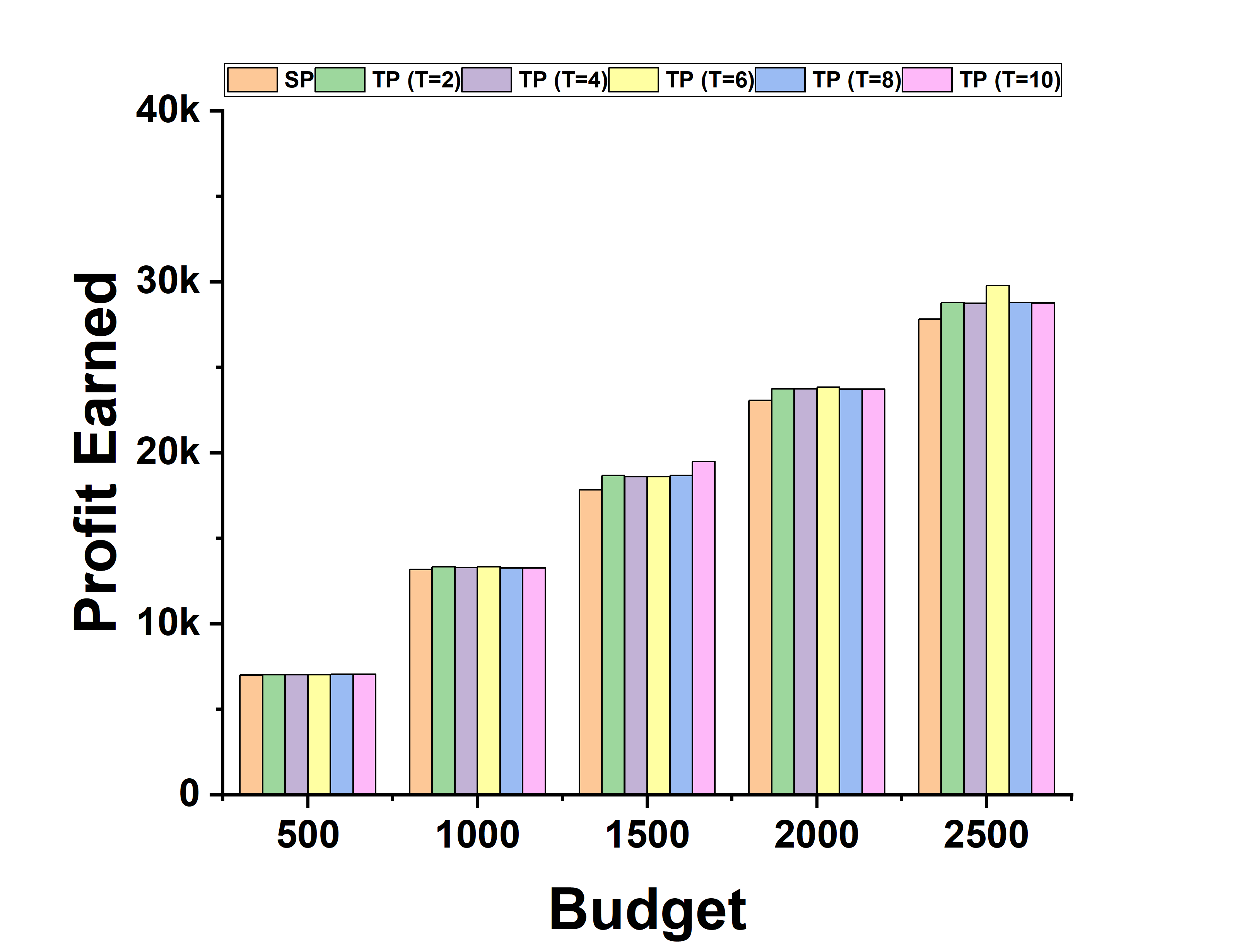}
        \caption{Clustering\\Coefficient}
    \end{subfigure} &
    \begin{subfigure}[t]{0.22\textwidth}
        \includegraphics[width=\linewidth]{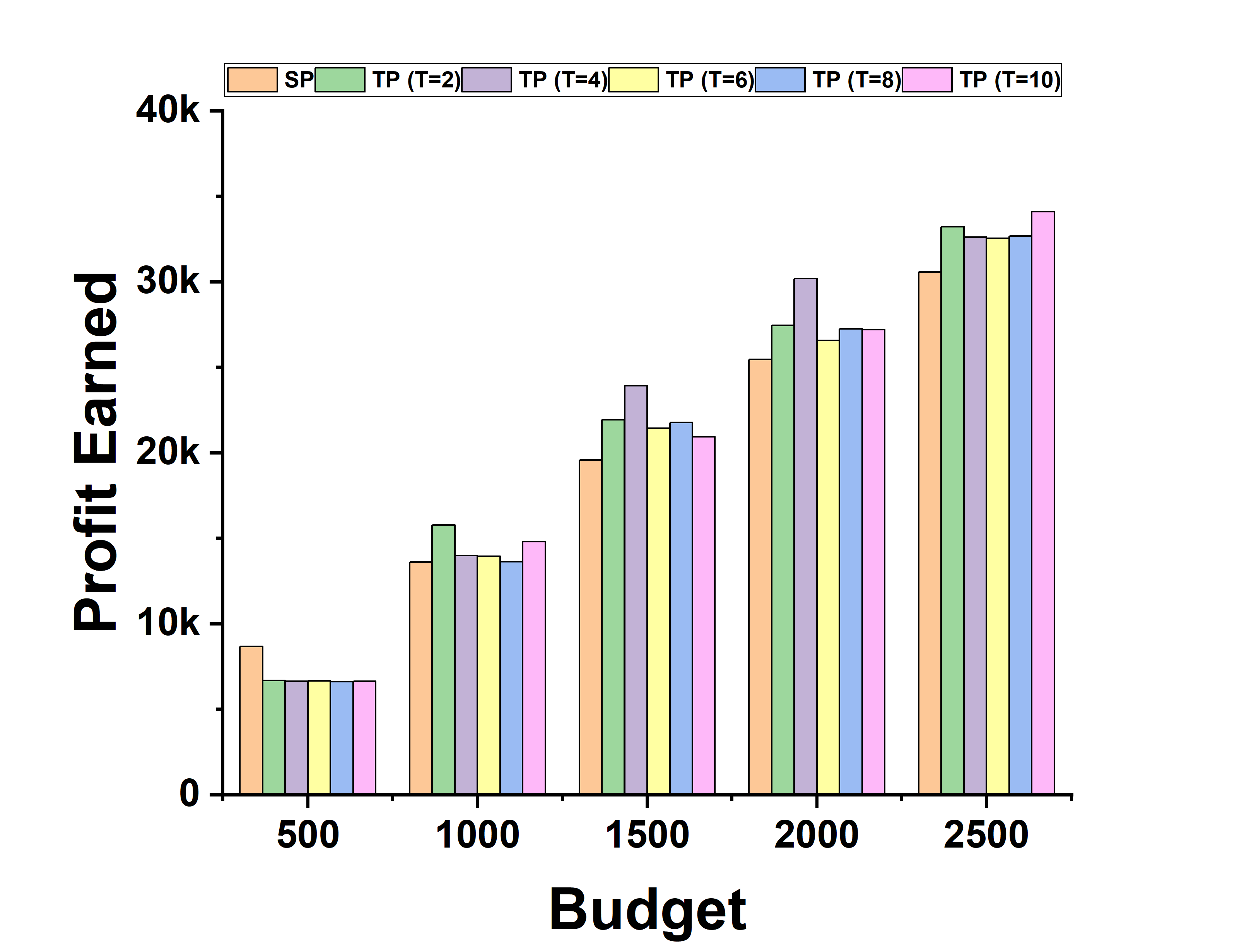}
        \caption{Degree Discount}
    \end{subfigure} \\[6pt]

    \begin{subfigure}[t]{0.22\textwidth}
        \includegraphics[width=\linewidth]{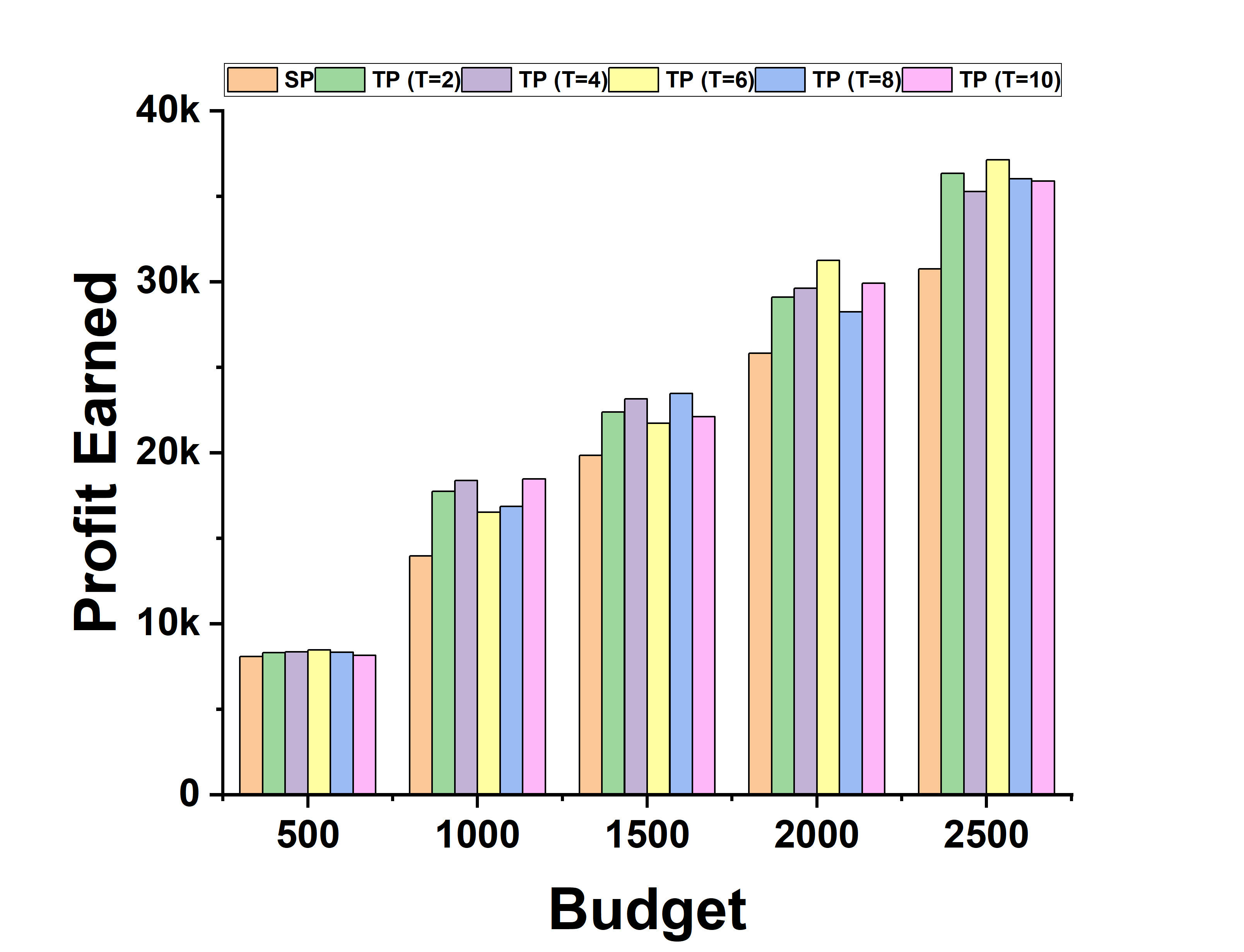}
        \caption{Single Discount}
    \end{subfigure} &
    \begin{subfigure}[t]{0.22\textwidth}
        \includegraphics[width=\linewidth]{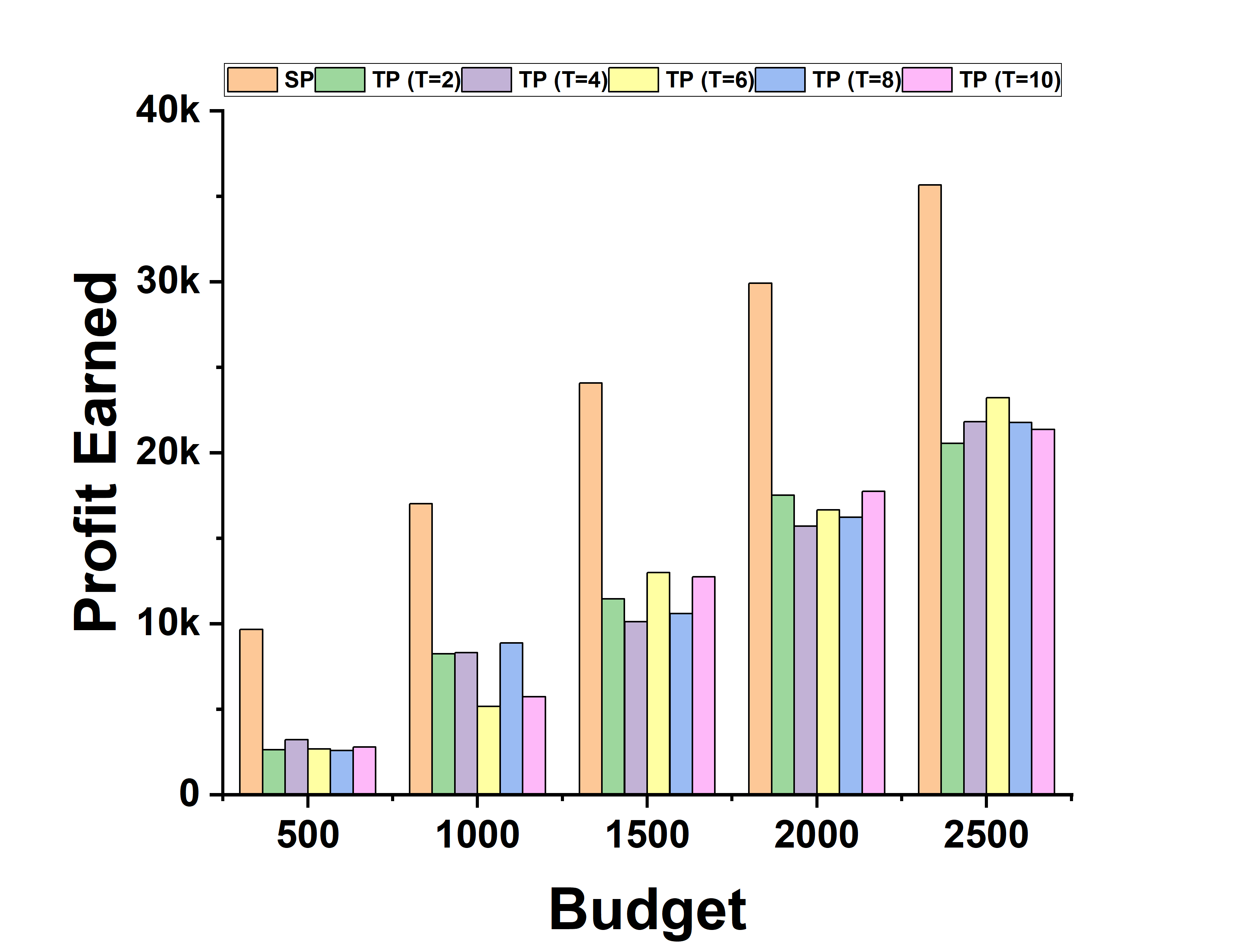}
        \caption{Simple Greedy}
    \end{subfigure} &
    \begin{subfigure}[t]{0.22\textwidth}
        \includegraphics[width=\linewidth]{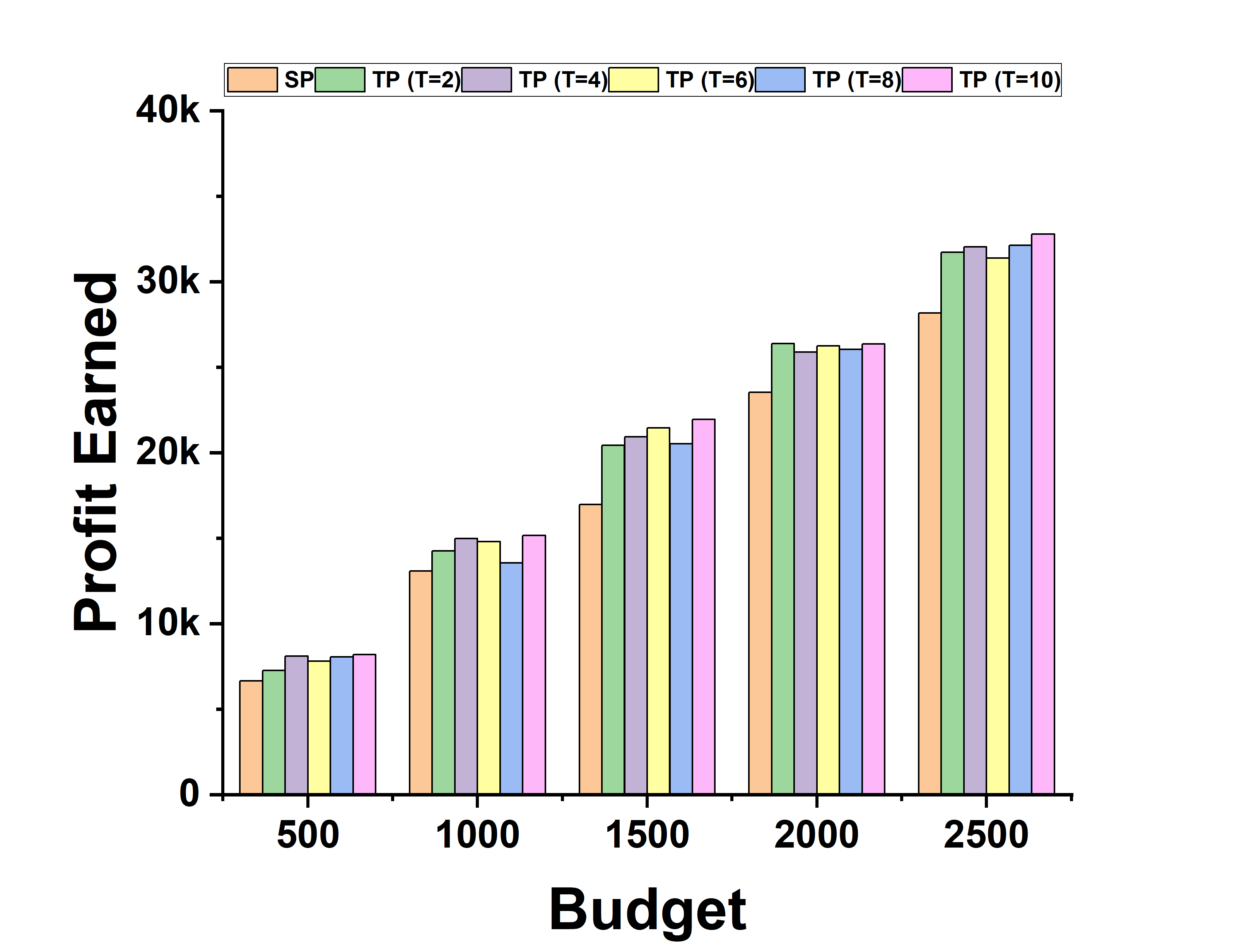}
        \caption{Double Greedy}
    \end{subfigure} &
    \begin{subfigure}[t]{0.22\textwidth}
        \includegraphics[width=\linewidth]{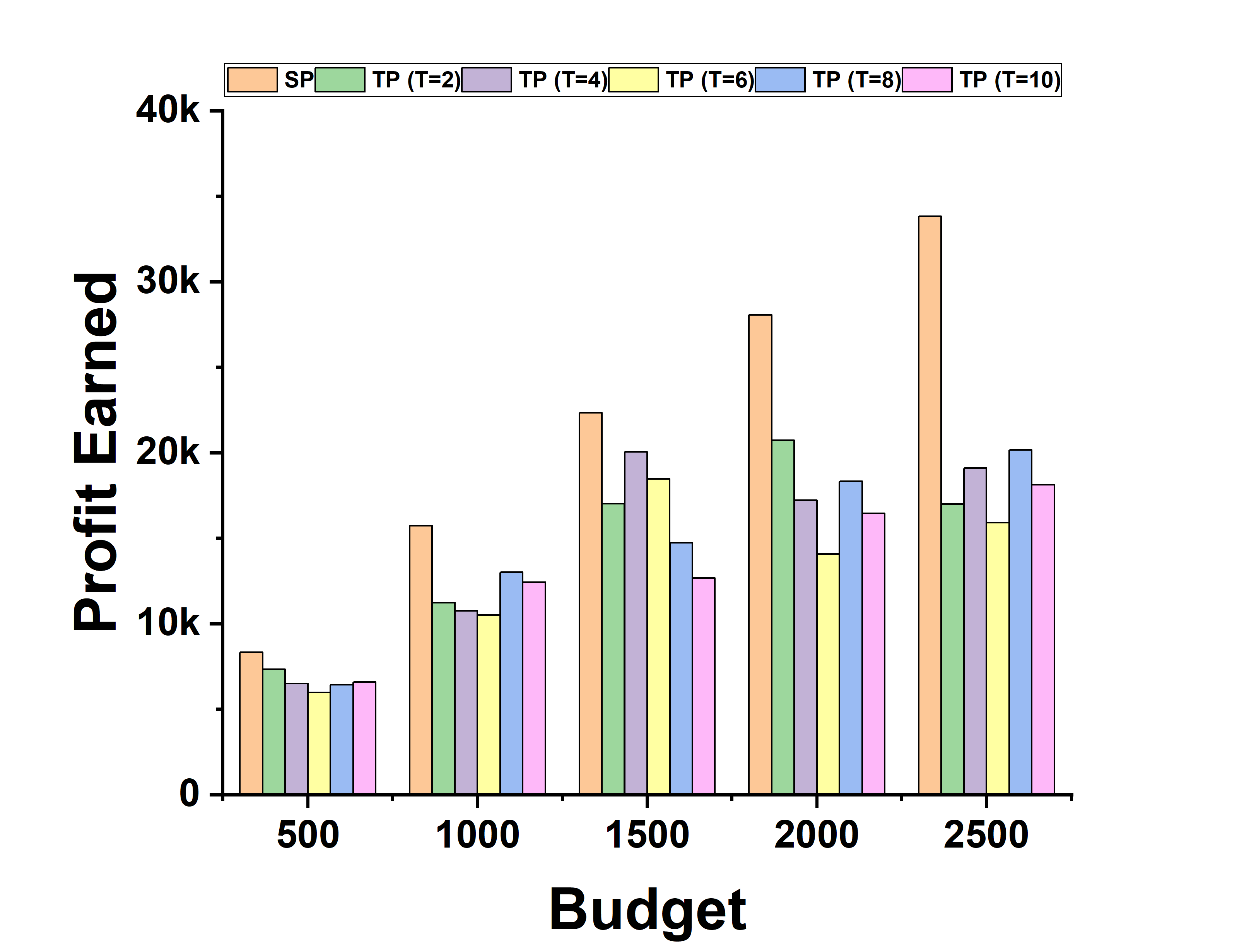}
        \caption{Stochastic Greedy}
    \end{subfigure}
\end{tabular}
\caption{Profit Earned in Single Phase Vs. Two Phase setting (split ratio 10\%, Probability Setting - Trivalency, \textit{LM} Dataset)}
\label{Fig:RQ1LM_T1}
\end{figure}

\begin{figure}[htbp]
\centering
\captionsetup[sub]{font=footnotesize}
\begin{tabular}{cccc}
    \begin{subfigure}[t]{0.22\textwidth}
        \includegraphics[width=\linewidth]{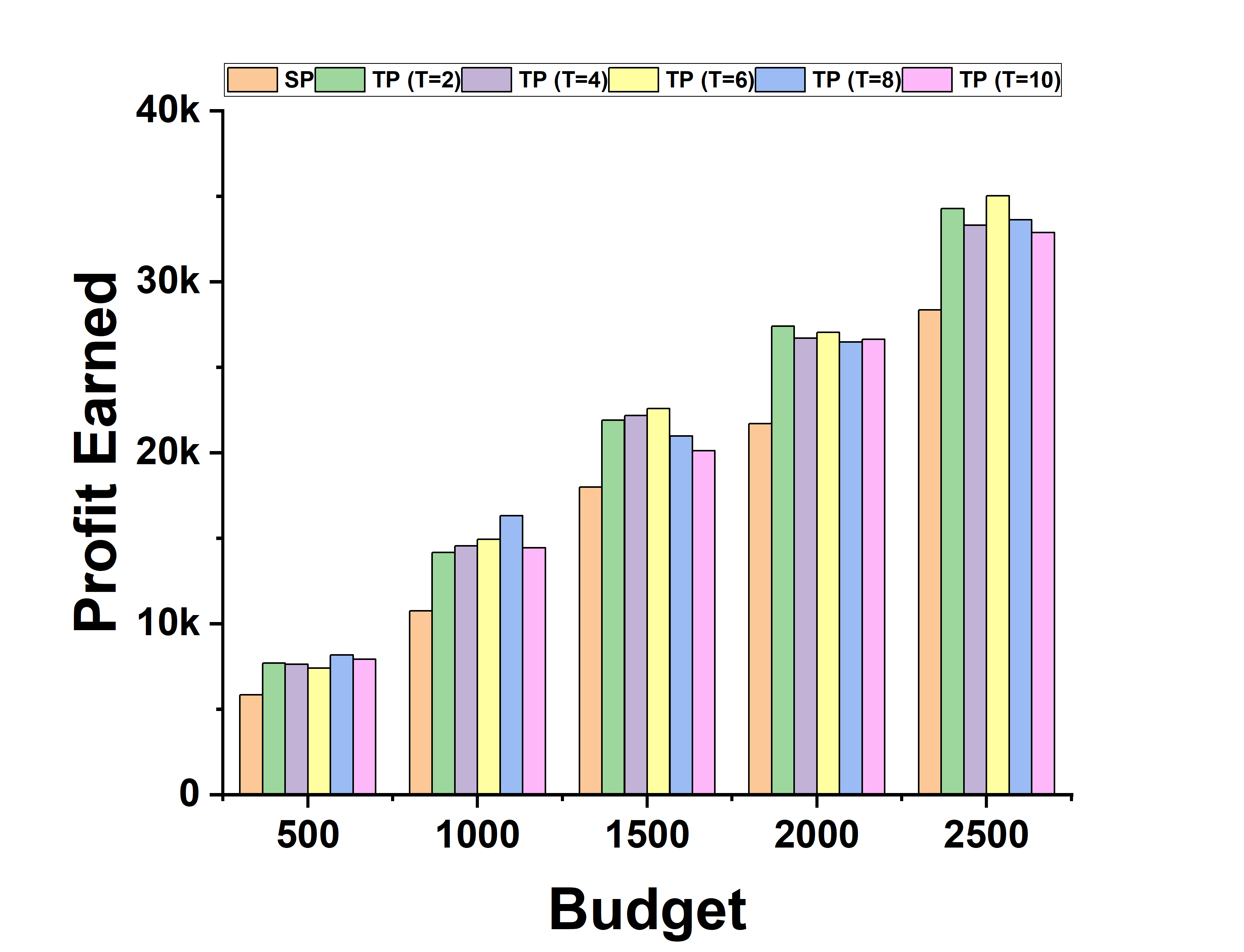}
        \caption{Random}
    \end{subfigure} &
    \begin{subfigure}[t]{0.22\textwidth}
        \includegraphics[width=\linewidth]{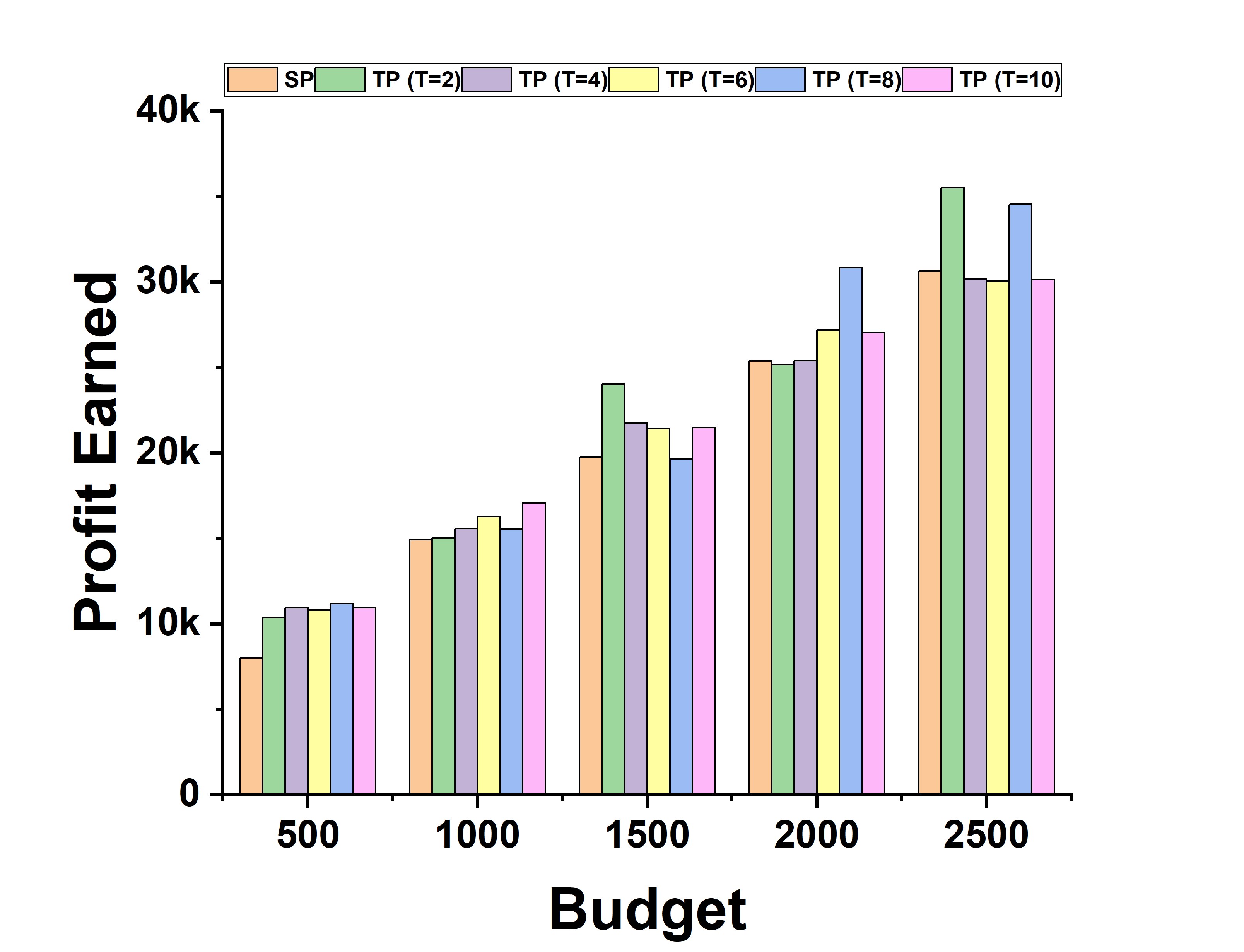}
        \caption{High Degree}
    \end{subfigure} &
    \begin{subfigure}[t]{0.22\textwidth}
        \includegraphics[width=\linewidth]{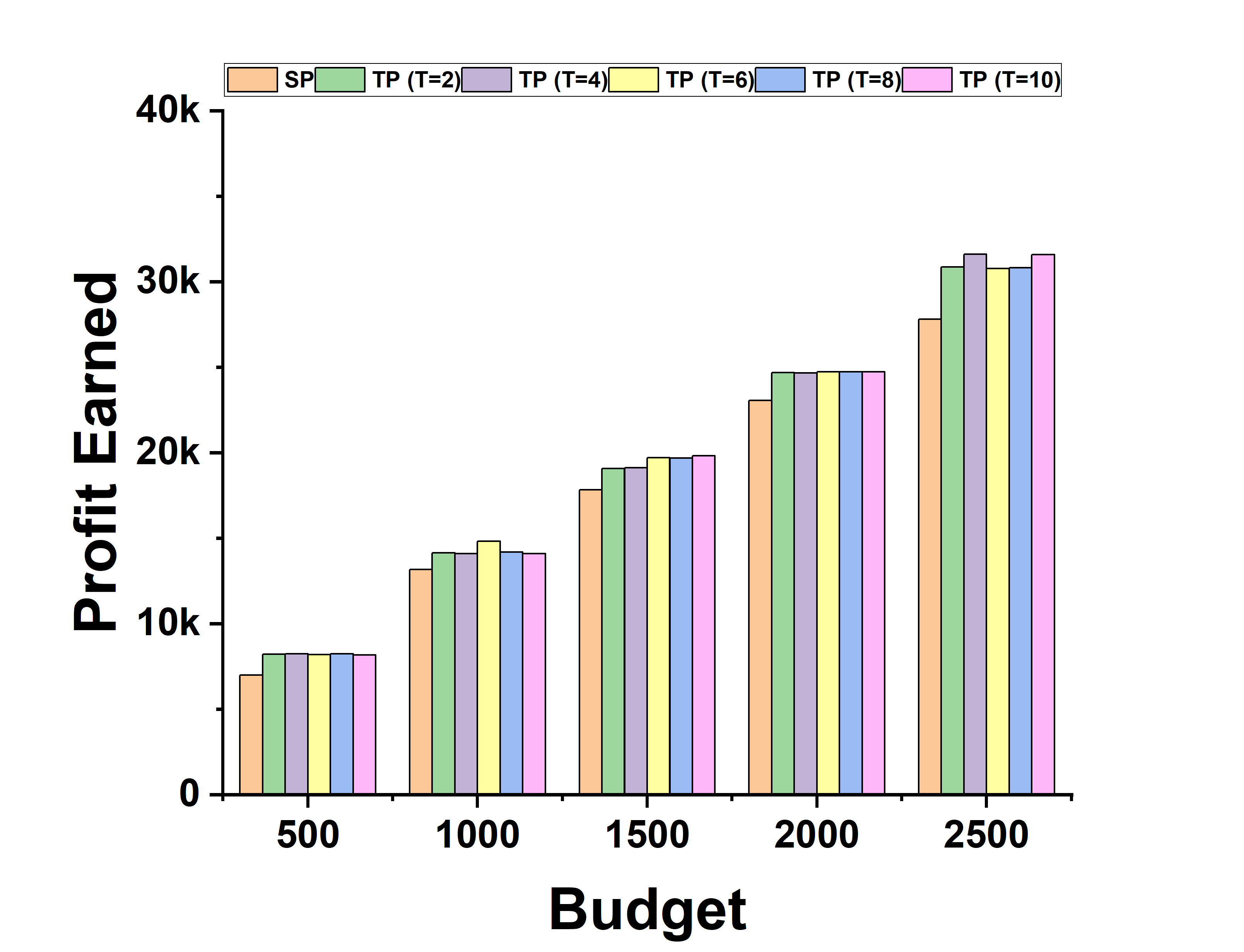}
        \caption{Clustering\\Coefficient}
    \end{subfigure} &
    \begin{subfigure}[t]{0.22\textwidth}
        \includegraphics[width=\linewidth]{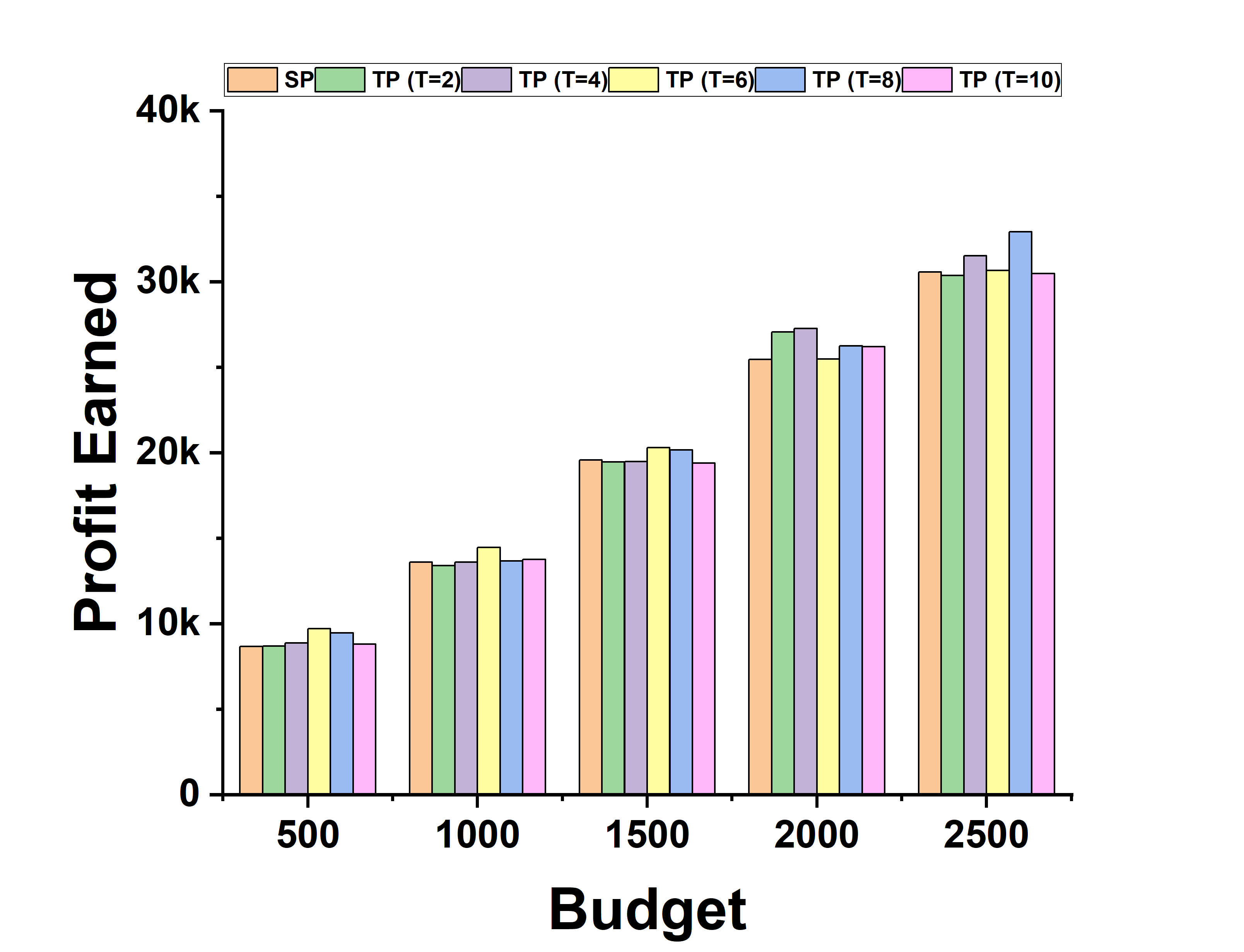}
        \caption{Degree Discount}
    \end{subfigure} \\[6pt]

    \begin{subfigure}[t]{0.22\textwidth}
        \includegraphics[width=\linewidth]{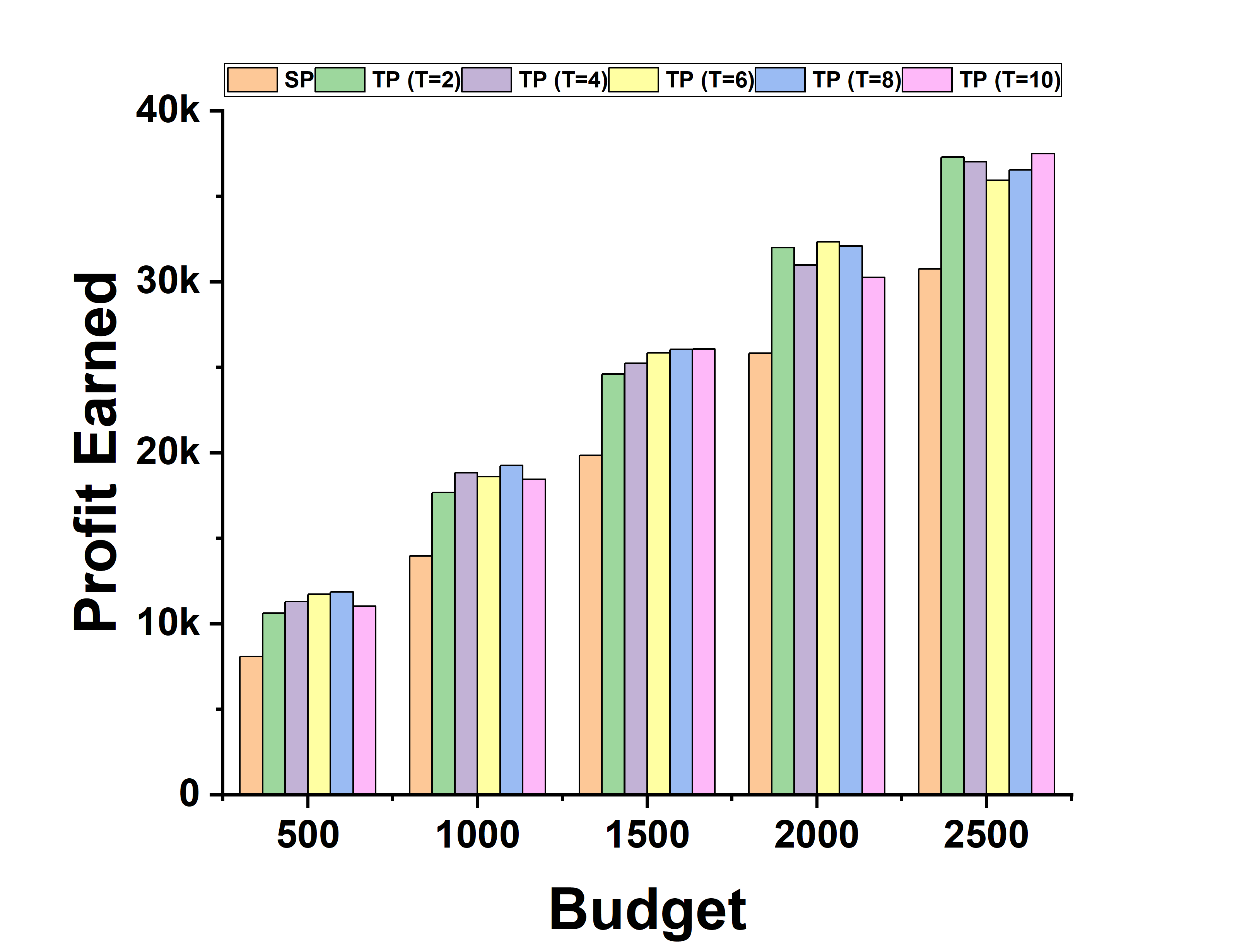}
        \caption{Single Discount}
    \end{subfigure} &
    \begin{subfigure}[t]{0.22\textwidth}
        \includegraphics[width=\linewidth]{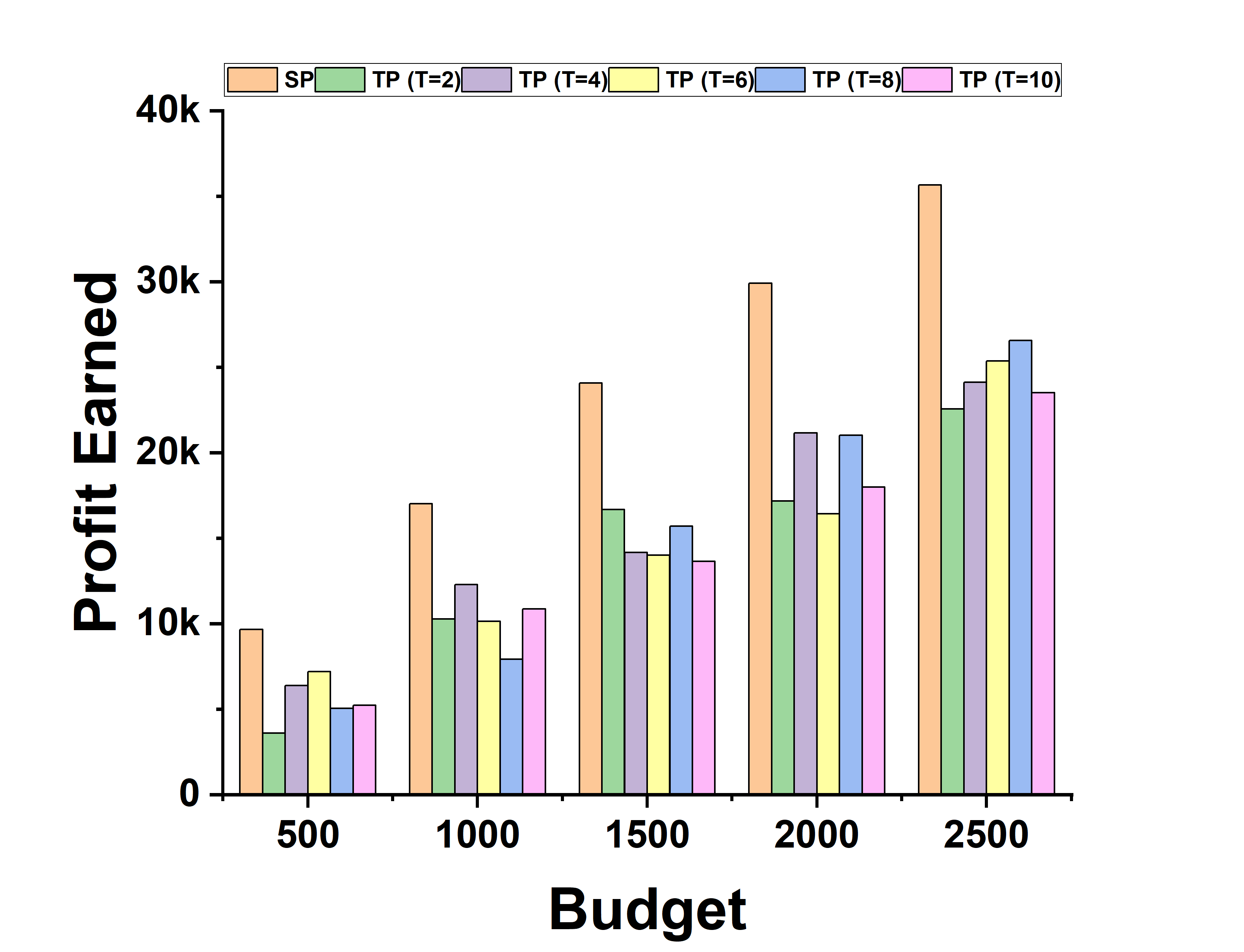}
        \caption{Simple Greedy}
    \end{subfigure} &
    \begin{subfigure}[t]{0.22\textwidth}
        \includegraphics[width=\linewidth]{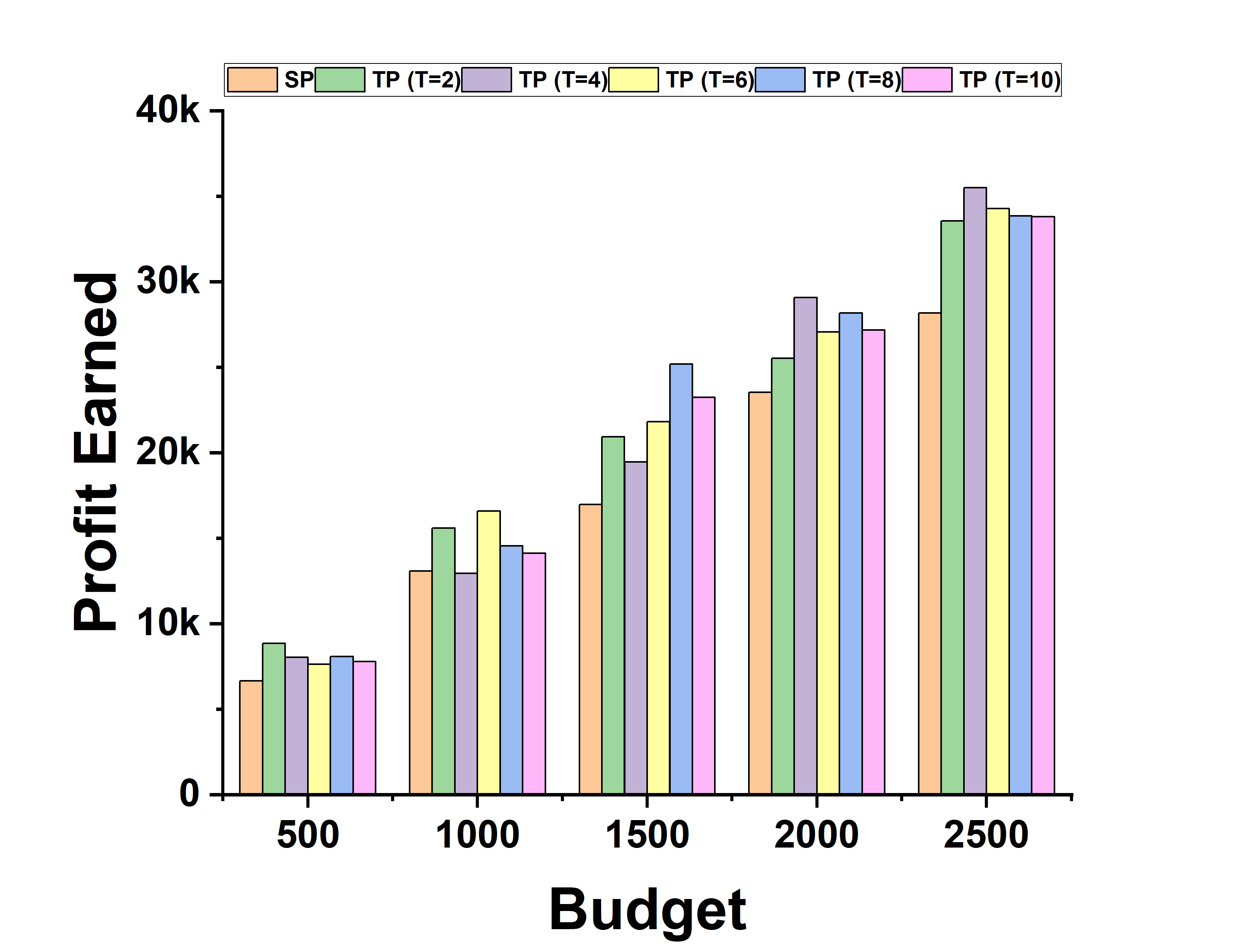}
        \caption{Double Greedy}
    \end{subfigure} &
    \begin{subfigure}[t]{0.22\textwidth}
        \includegraphics[width=\linewidth]{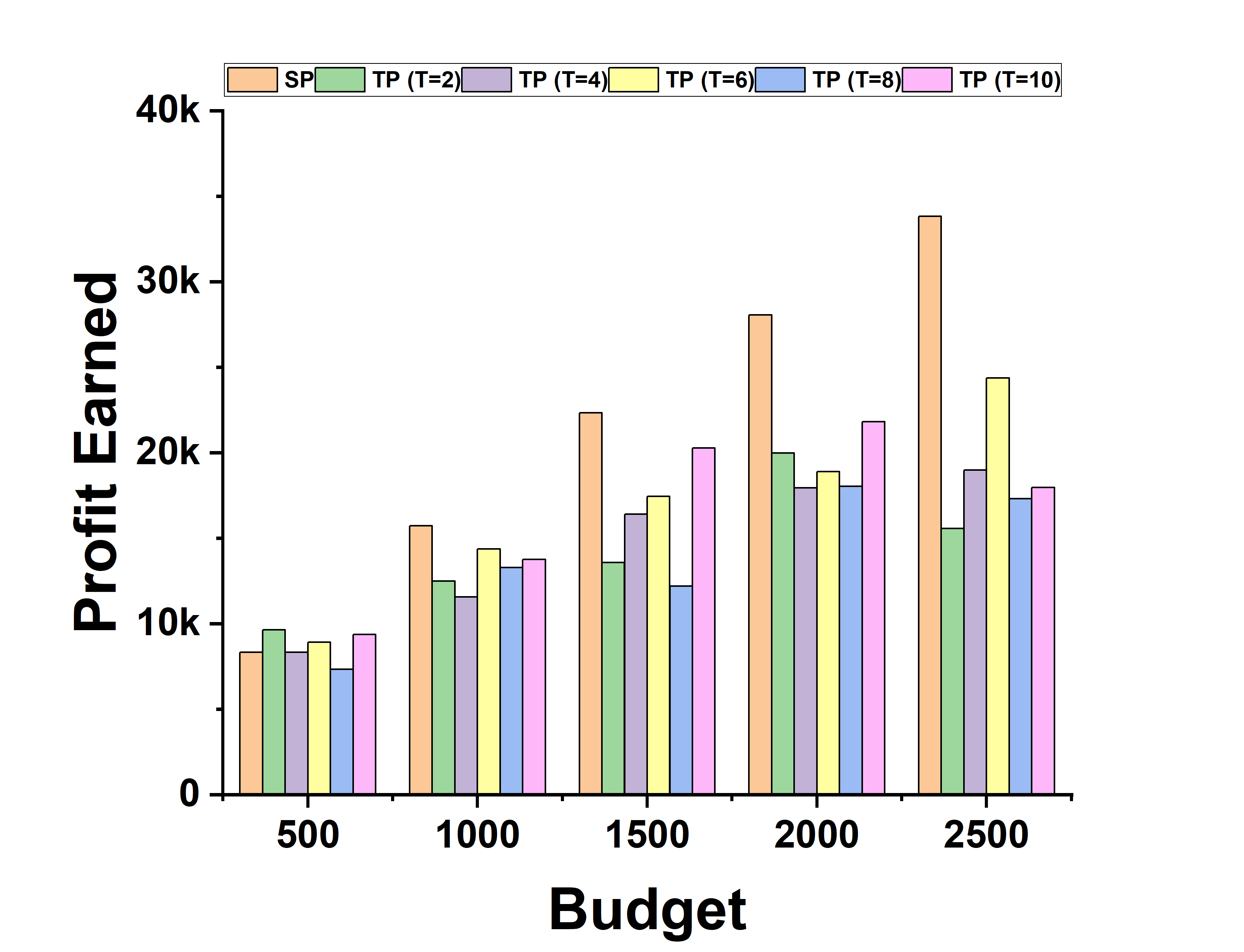}
        \caption{Stochastic Greedy}
    \end{subfigure}
\end{tabular}
\caption{Profit Earned in Single Phase Vs. Two Phase setting (split ratio 30\%, Probability Setting - Trivalency, \textit{LM} Dataset)}
\label{Fig:RQ1LM_T2}
\end{figure}

\begin{figure}[htbp]
\centering
\captionsetup[sub]{font=footnotesize}
\begin{tabular}{cccc}
    \begin{subfigure}[t]{0.22\textwidth}
        \includegraphics[width=\linewidth]{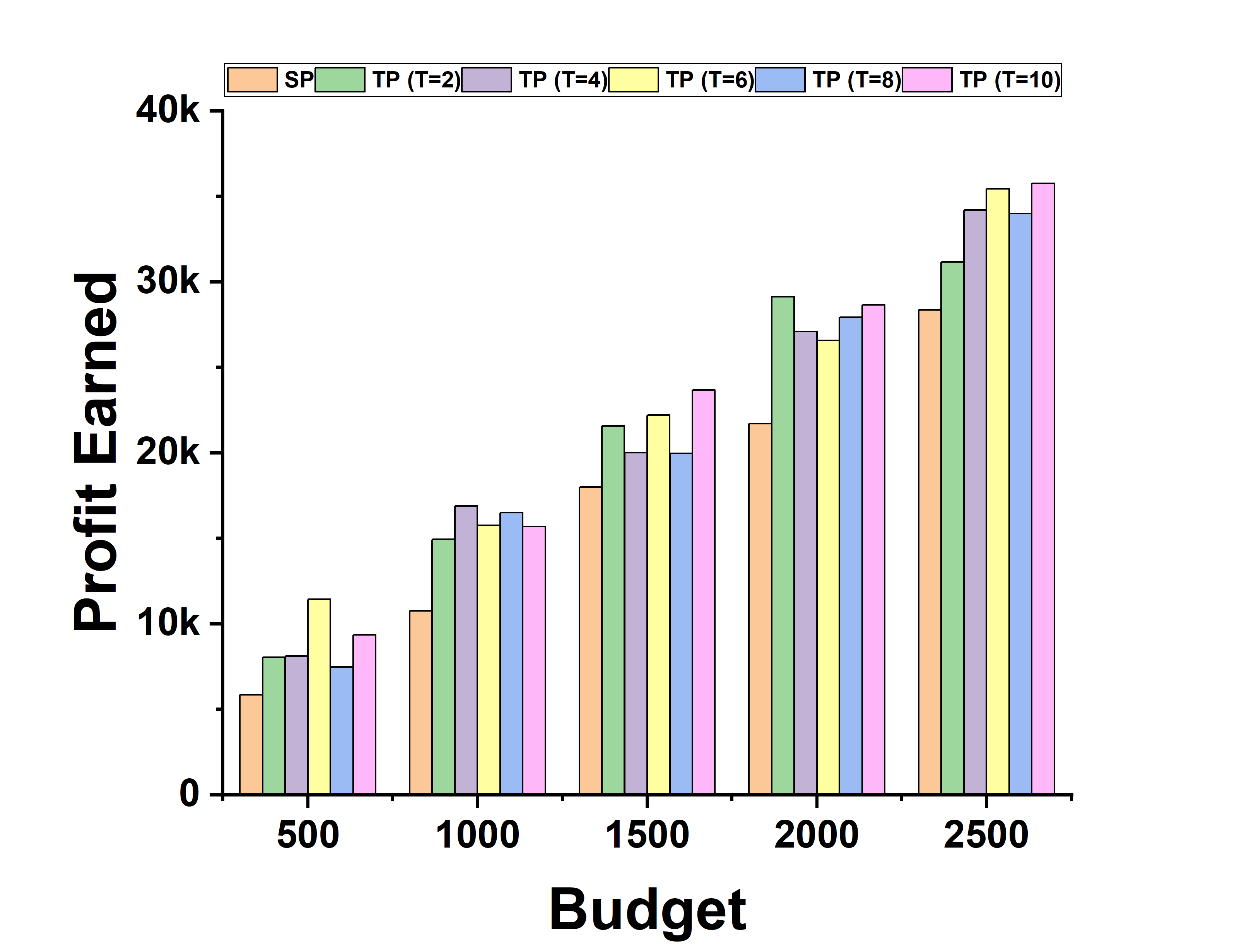}
        \caption{Random}
    \end{subfigure} &
    \begin{subfigure}[t]{0.22\textwidth}
        \includegraphics[width=\linewidth]{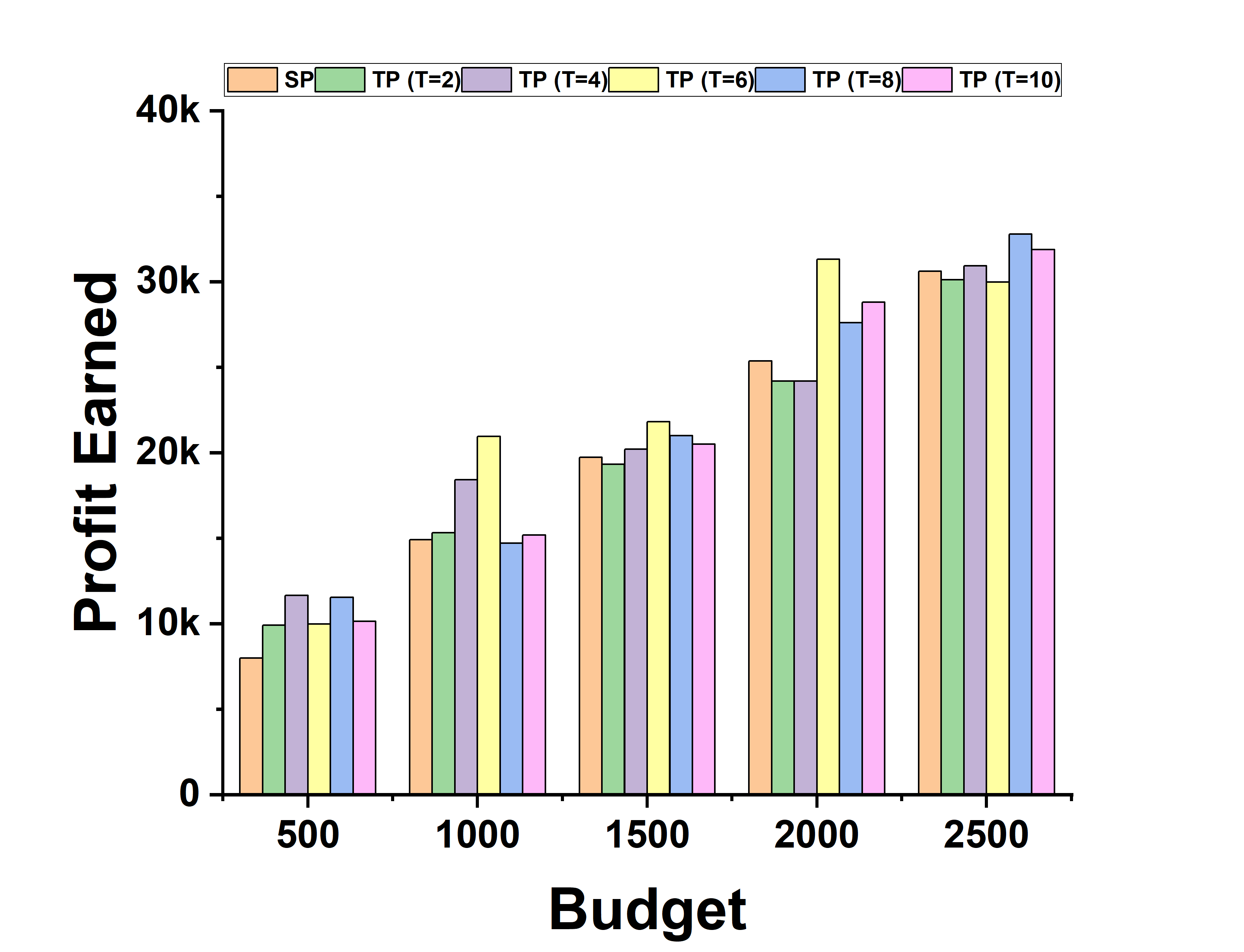}
        \caption{High Degree}
    \end{subfigure} &
    \begin{subfigure}[t]{0.22\textwidth}
        \includegraphics[width=\linewidth]{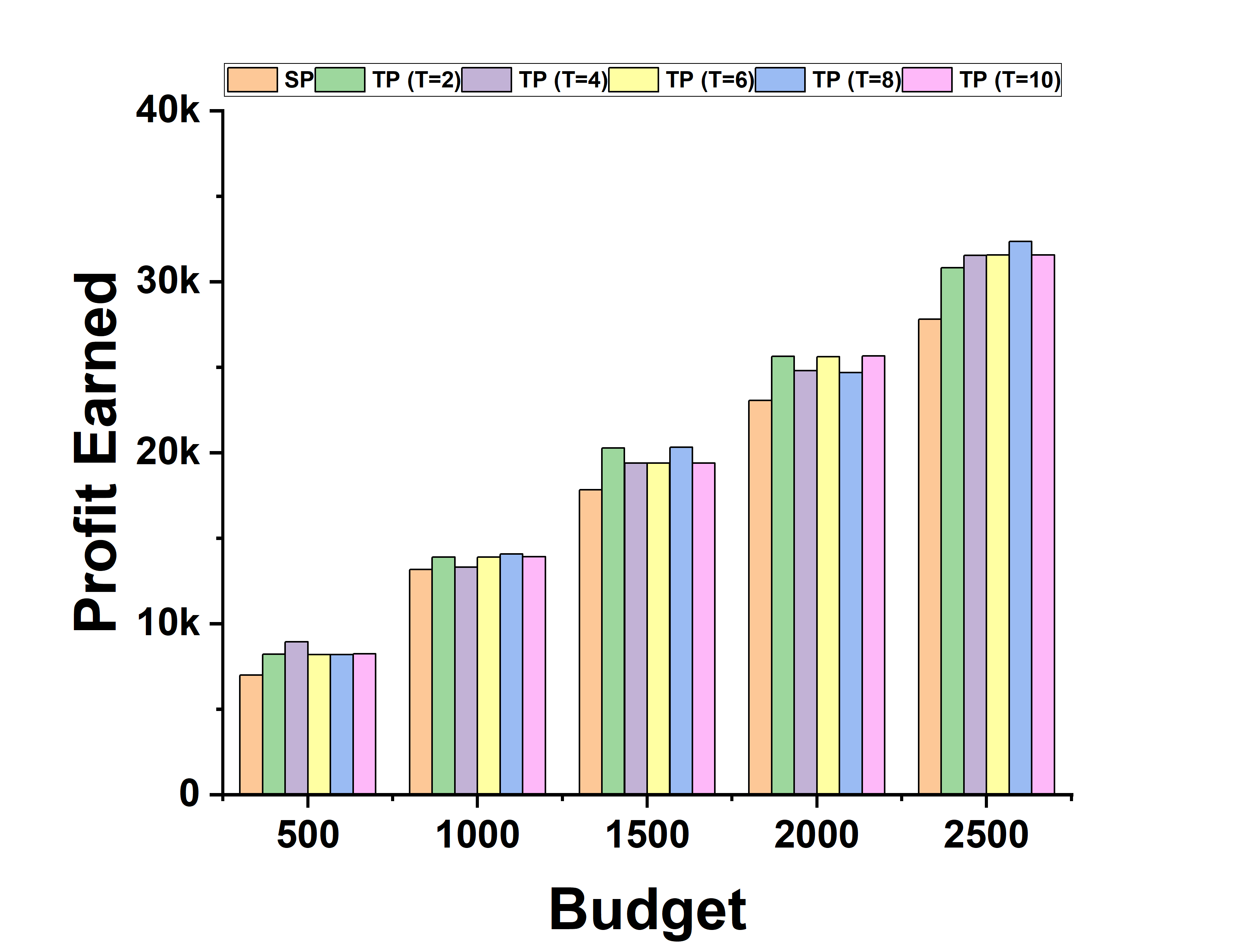}
        \caption{Clustering\\Coefficient}
    \end{subfigure} &
    \begin{subfigure}[t]{0.22\textwidth}
        \includegraphics[width=\linewidth]{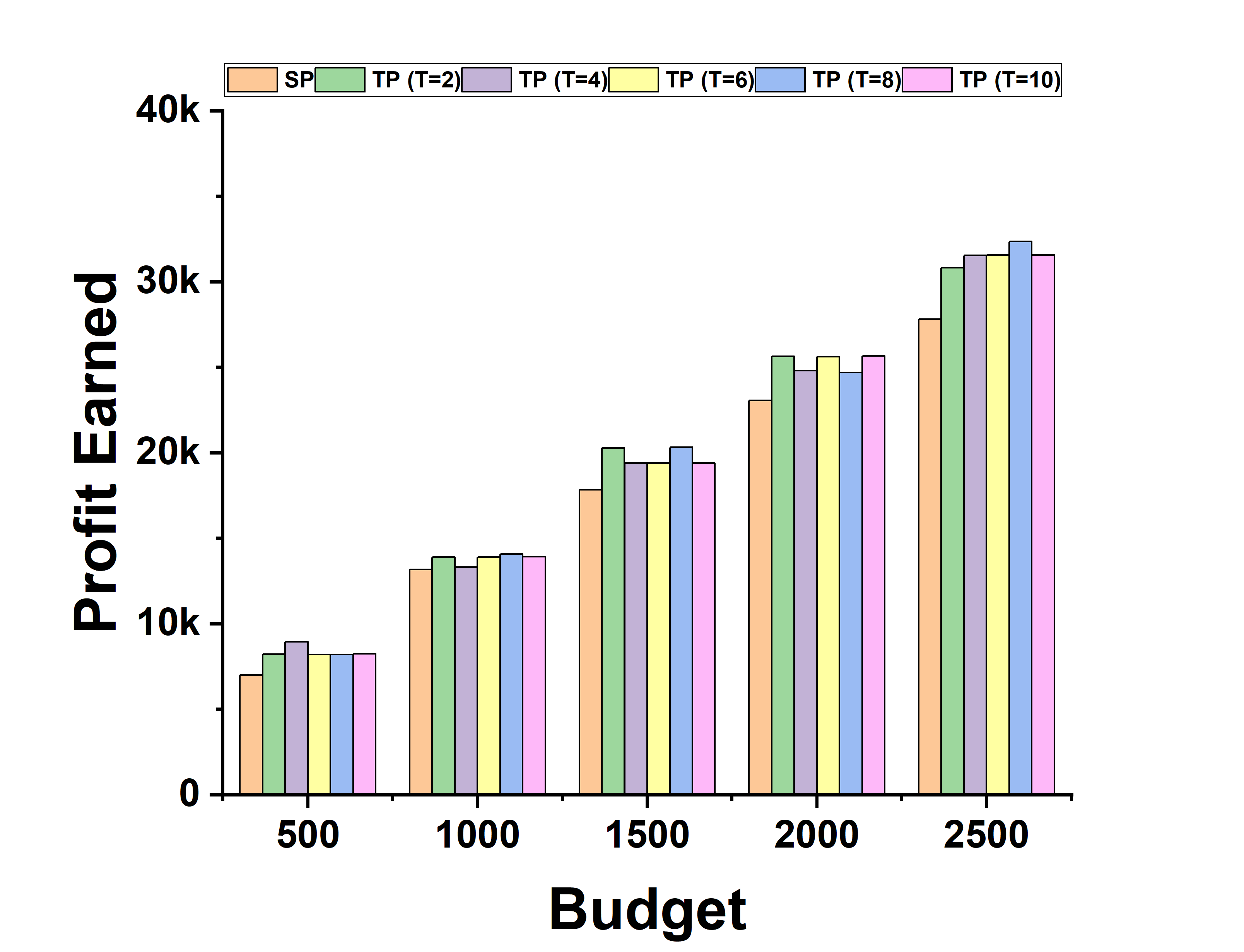}
        \caption{Degree Discount}
    \end{subfigure} \\[6pt]

    \begin{subfigure}[t]{0.22\textwidth}
        \includegraphics[width=\linewidth]{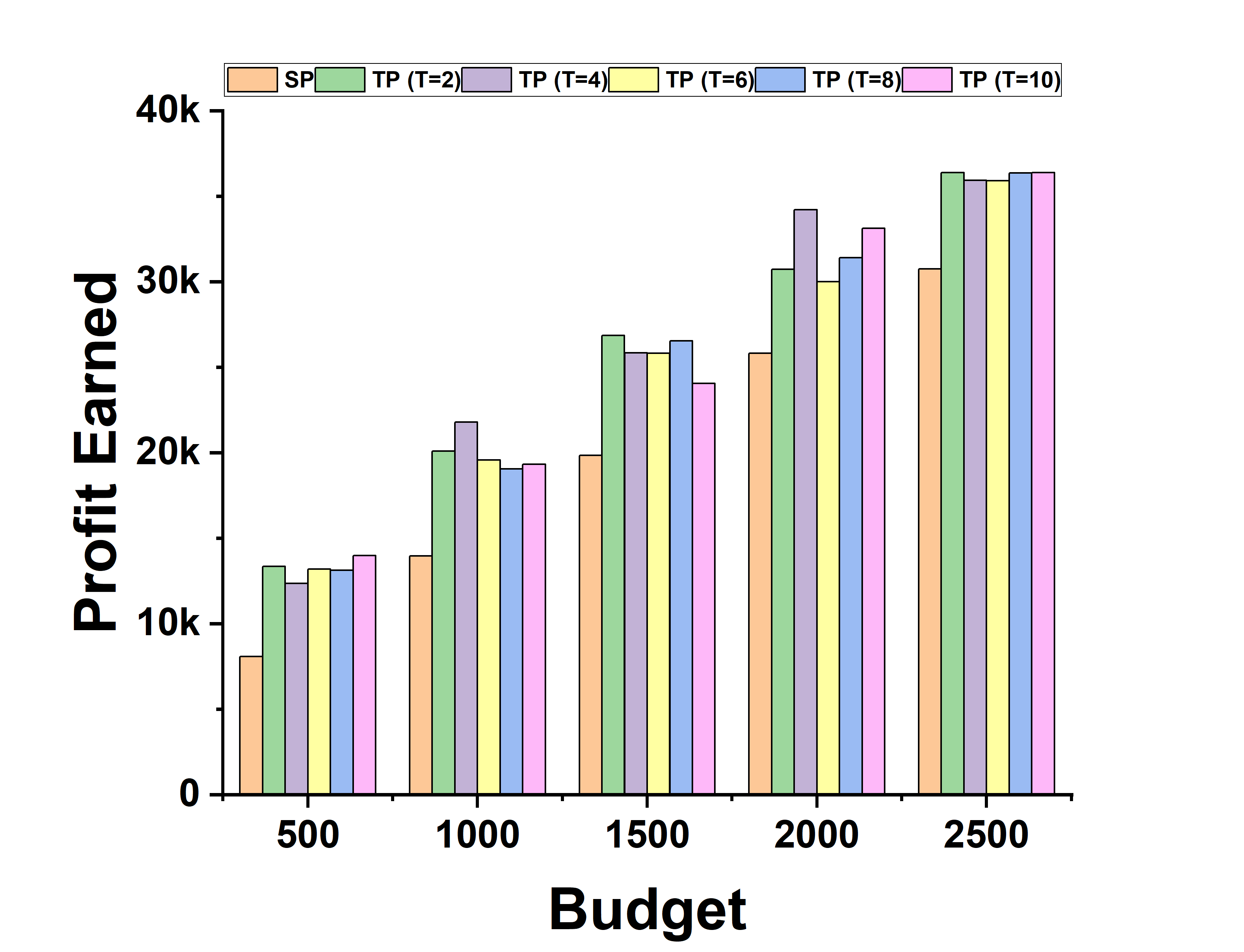}
        \caption{Single Discount}
    \end{subfigure} &
    \begin{subfigure}[t]{0.22\textwidth}
        \includegraphics[width=\linewidth]{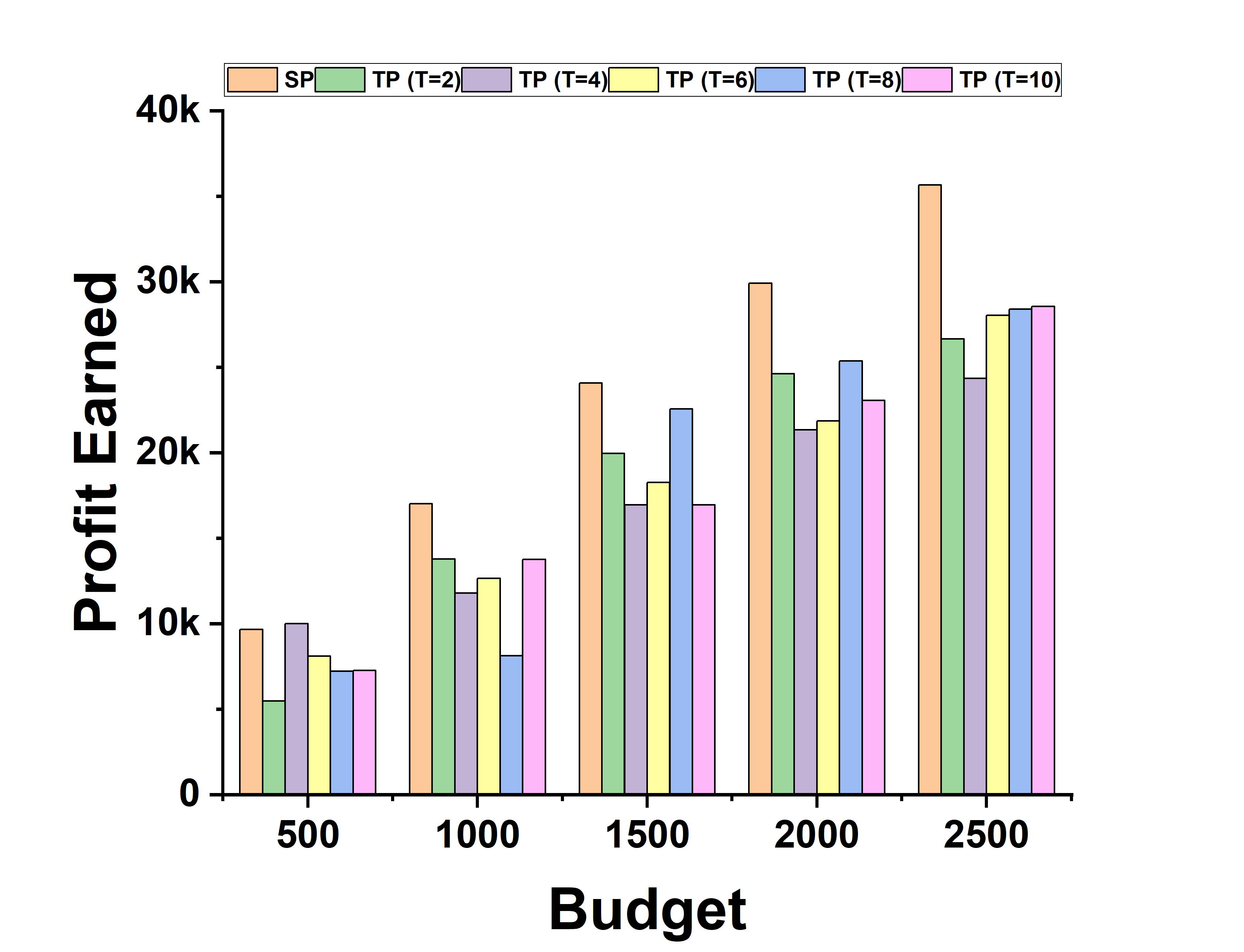}
        \caption{Simple Greedy}
    \end{subfigure} &
    \begin{subfigure}[t]{0.22\textwidth}
        \includegraphics[width=\linewidth]{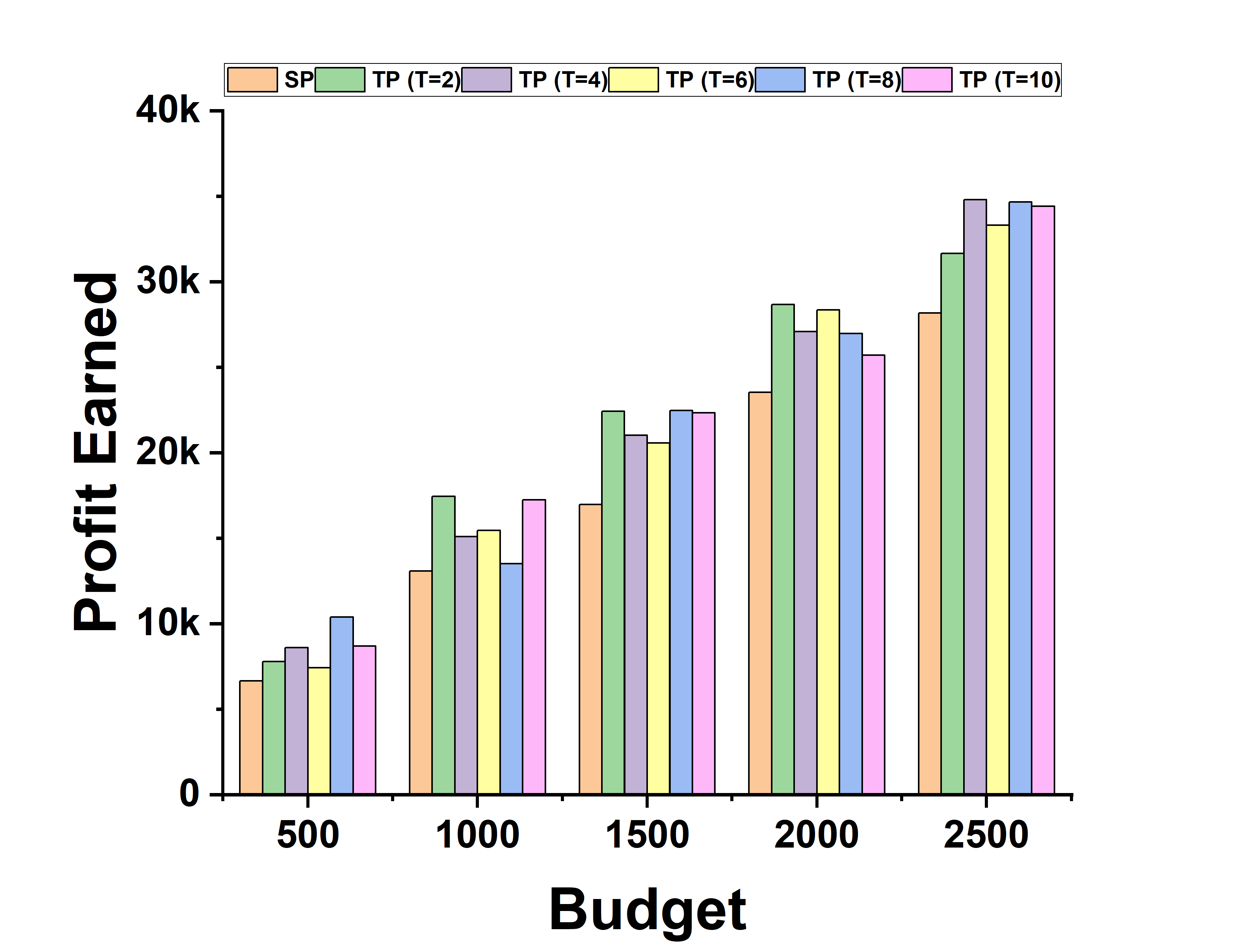}
        \caption{Double Greedy}
    \end{subfigure} &
    \begin{subfigure}[t]{0.22\textwidth}
        \includegraphics[width=\linewidth]{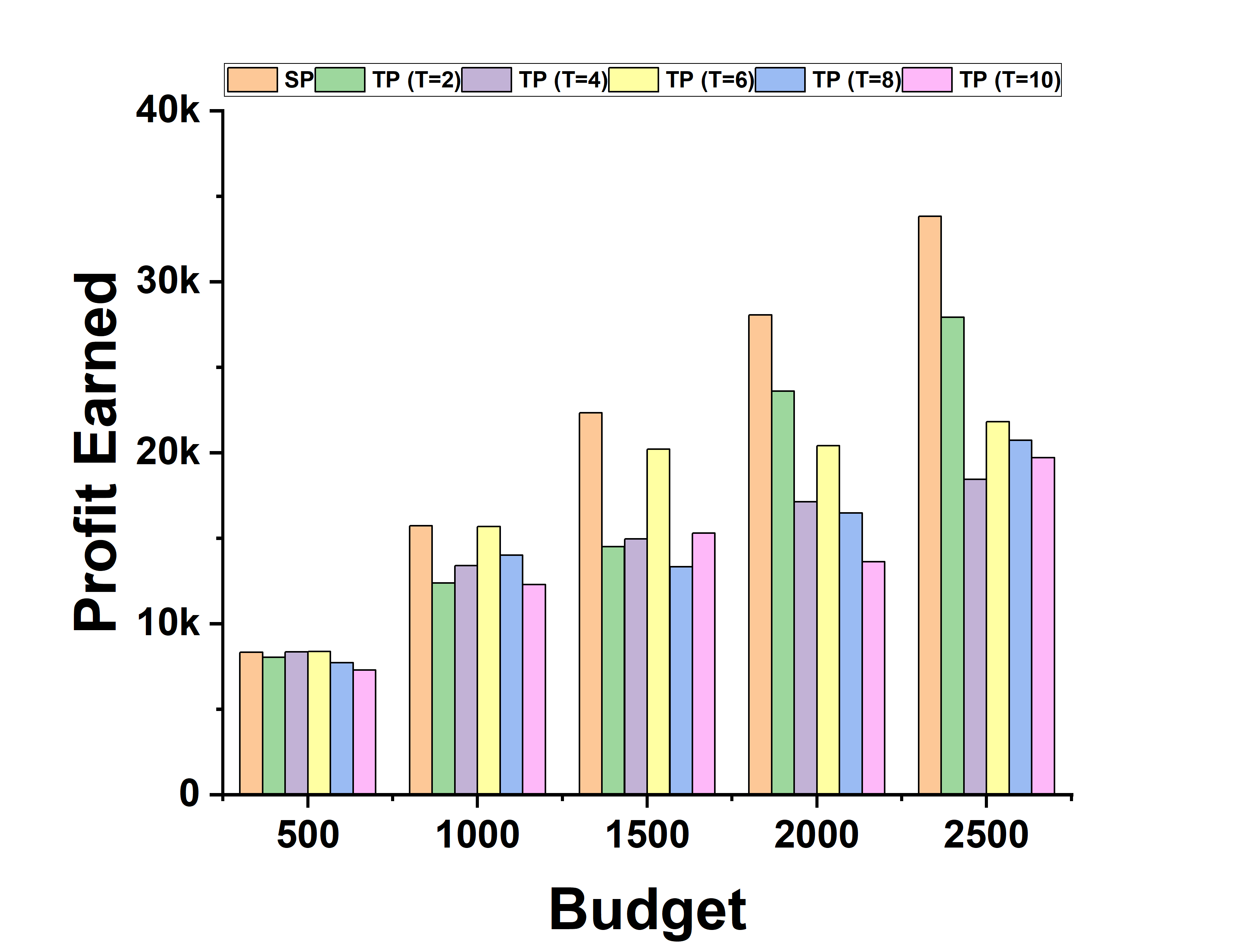}
        \caption{Stochastic Greedy}
    \end{subfigure}
\end{tabular}
\caption{Profit Earned in Single Phase Vs. Two Phase setting (split ratio 50\%, Probability Setting - Trivalency, \textit{LM} Dataset)}
\label{Fig:RQ1LM_T3}
\end{figure}

\begin{figure}[htbp]
\centering
\captionsetup[sub]{font=footnotesize}
\begin{tabular}{cccc}
    \begin{subfigure}[t]{0.22\textwidth}
        \includegraphics[width=\linewidth]{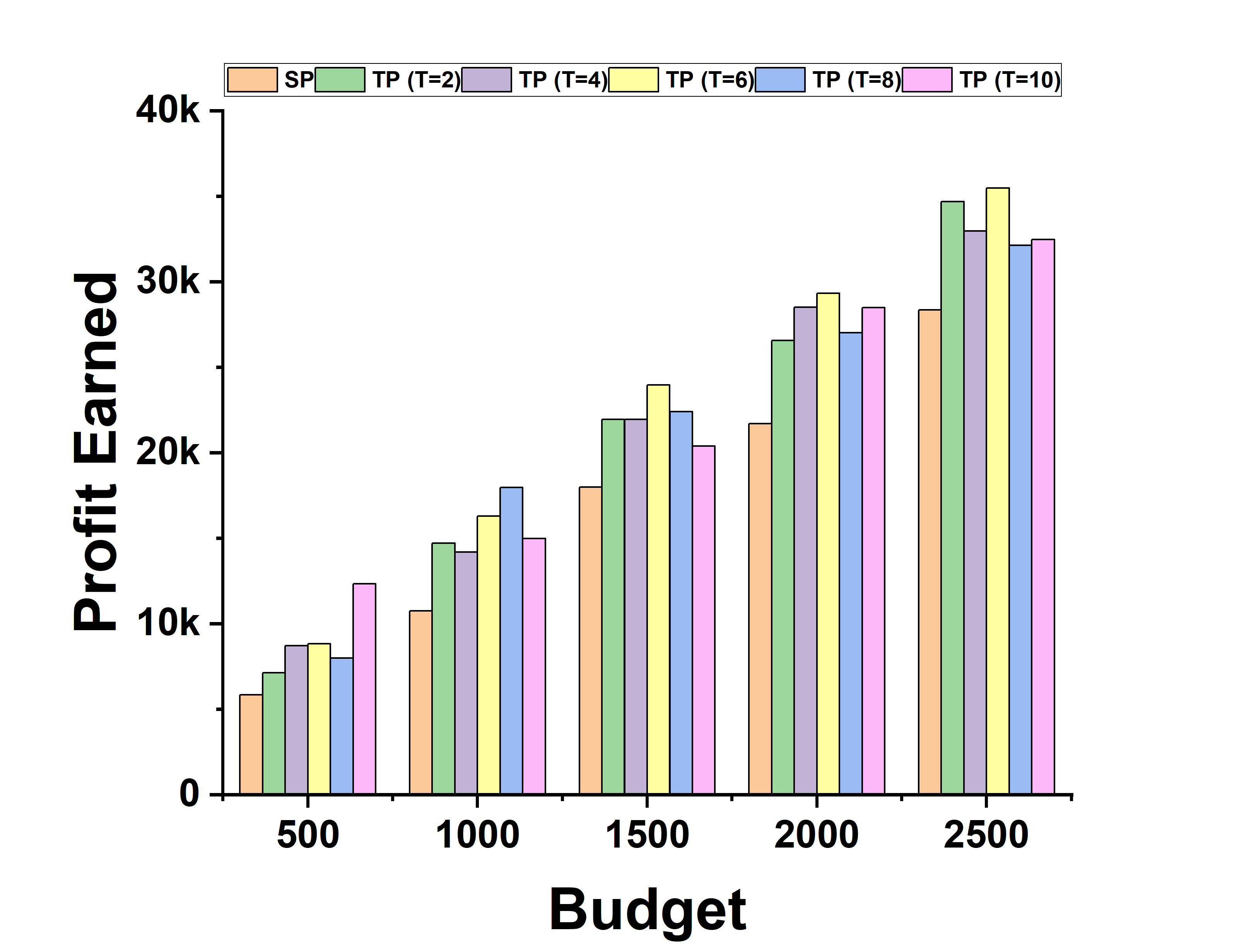}
        \caption{Random}
    \end{subfigure} &
    \begin{subfigure}[t]{0.22\textwidth}
        \includegraphics[width=\linewidth]{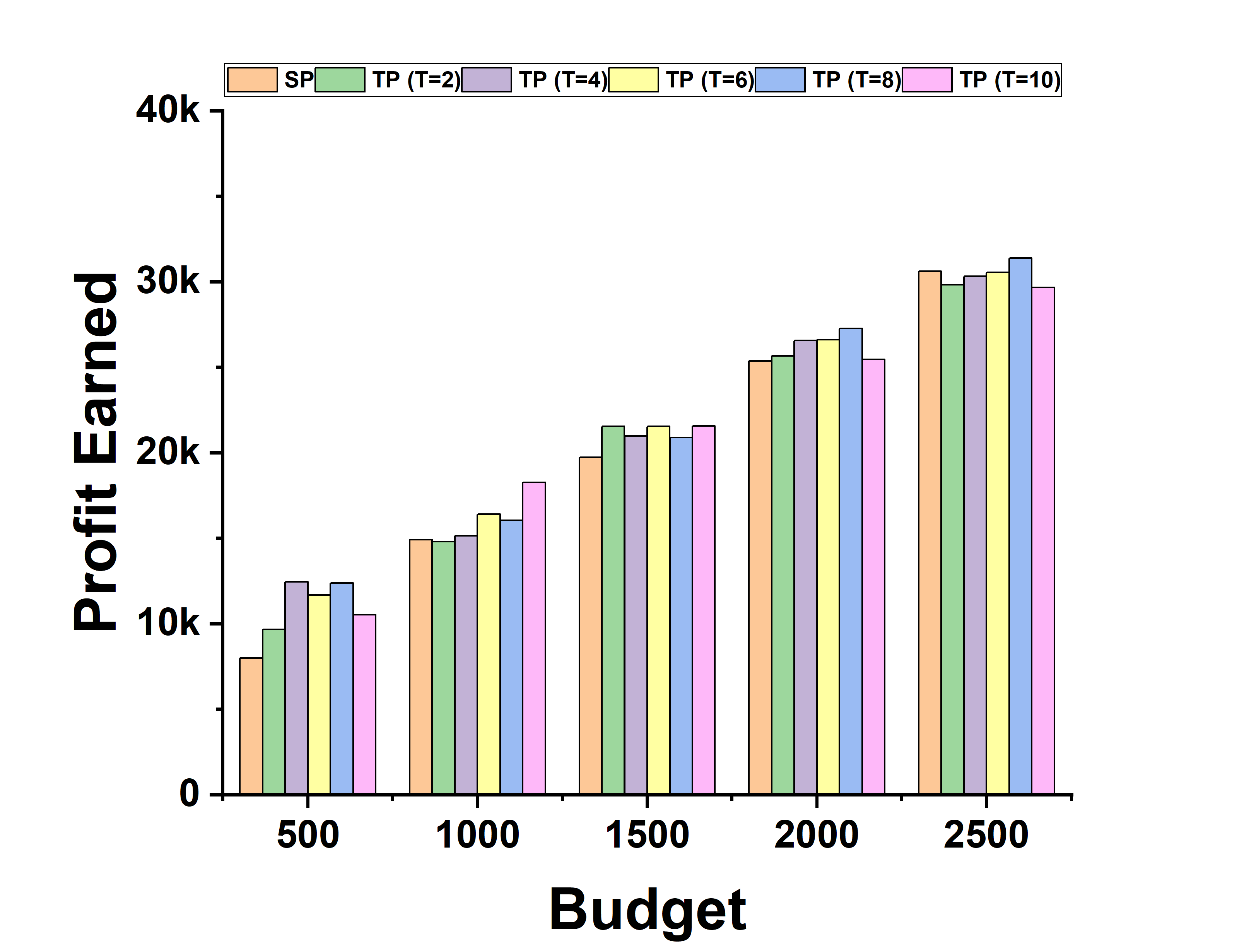}
        \caption{High Degree}
    \end{subfigure} &
    \begin{subfigure}[t]{0.22\textwidth}
        \includegraphics[width=\linewidth]{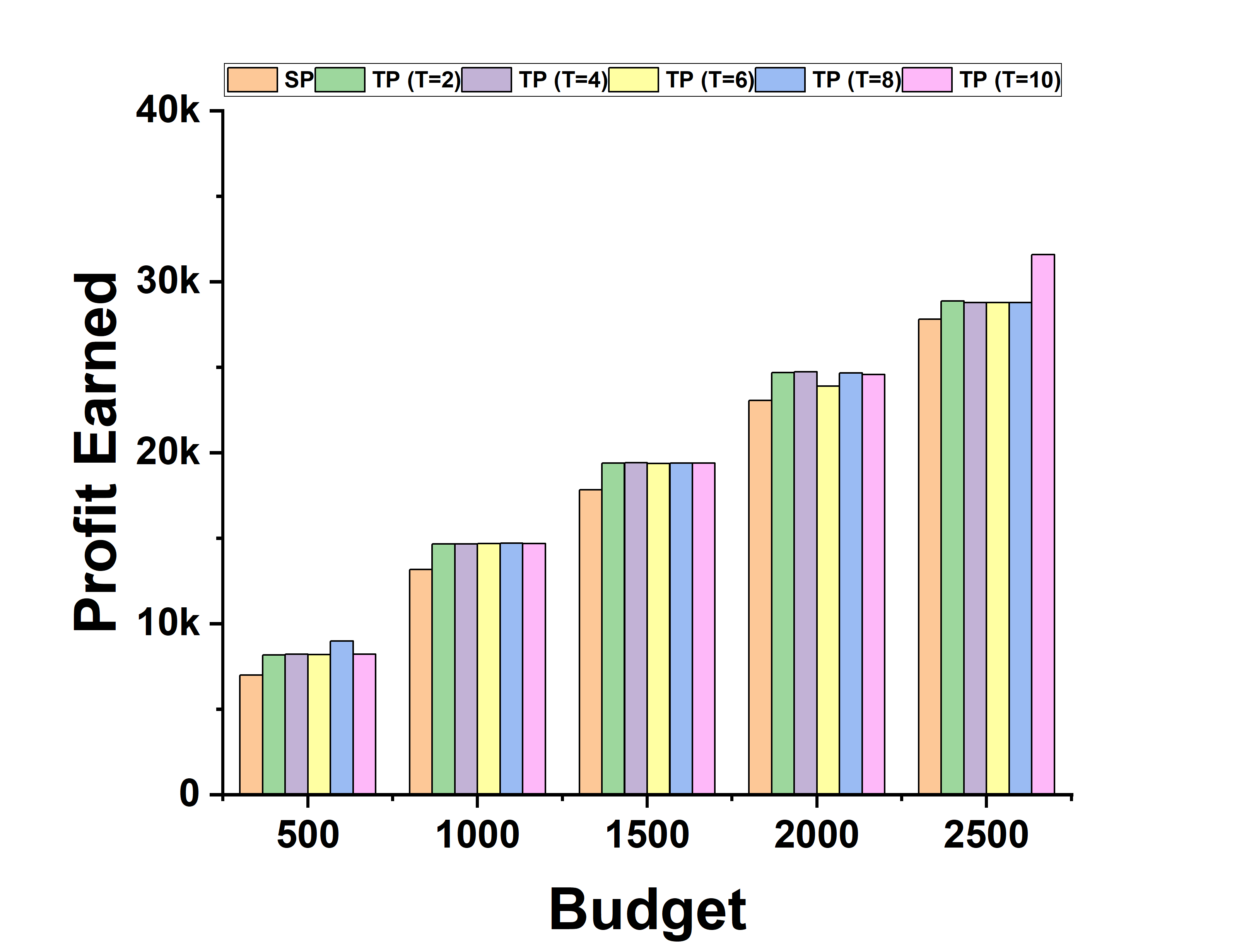}
        \caption{Clustering\\Coefficient}
    \end{subfigure} &
    \begin{subfigure}[t]{0.22\textwidth}
        \includegraphics[width=\linewidth]{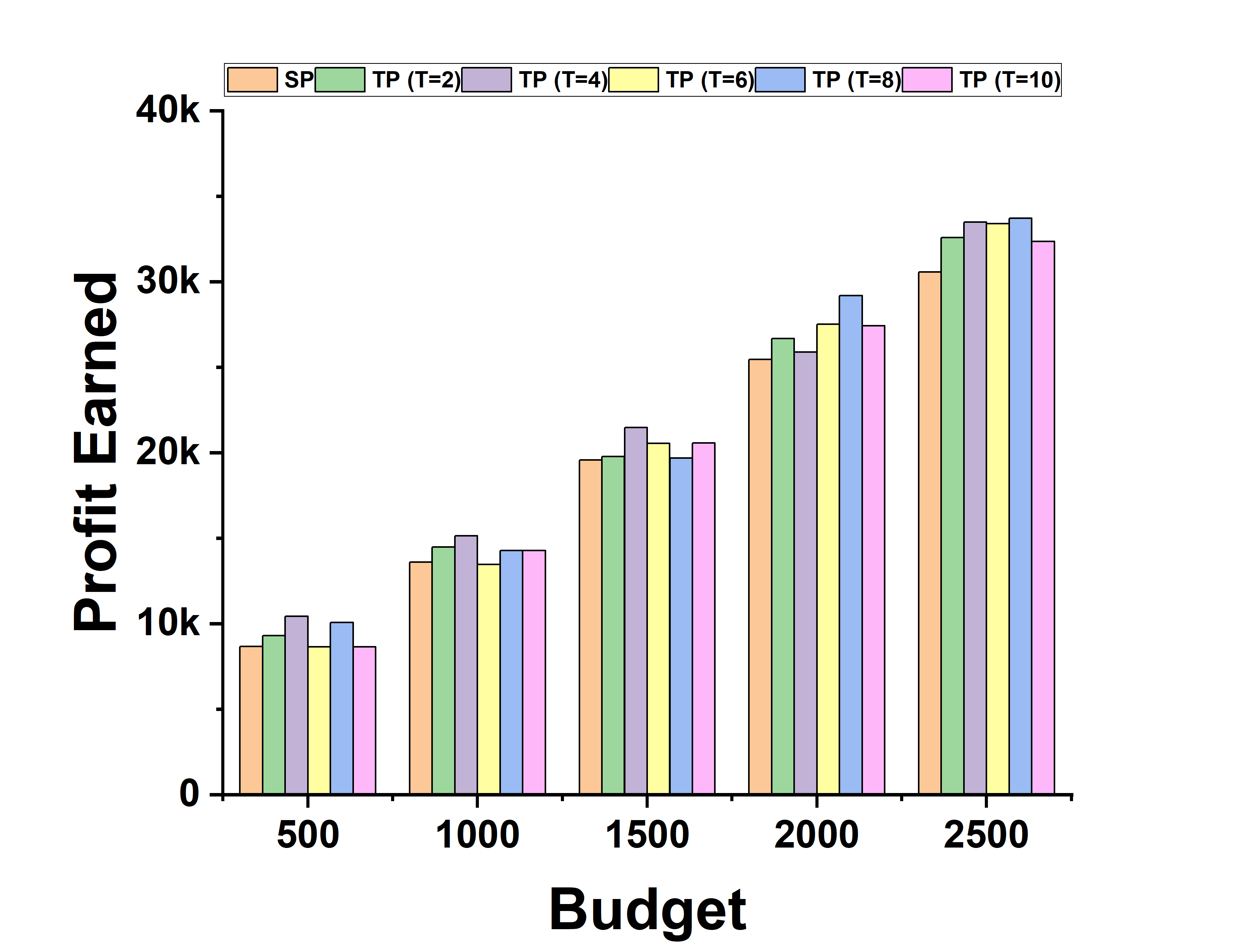}
        \caption{Degree Discount}
    \end{subfigure} \\[6pt]

    \begin{subfigure}[t]{0.22\textwidth}
        \includegraphics[width=\linewidth]{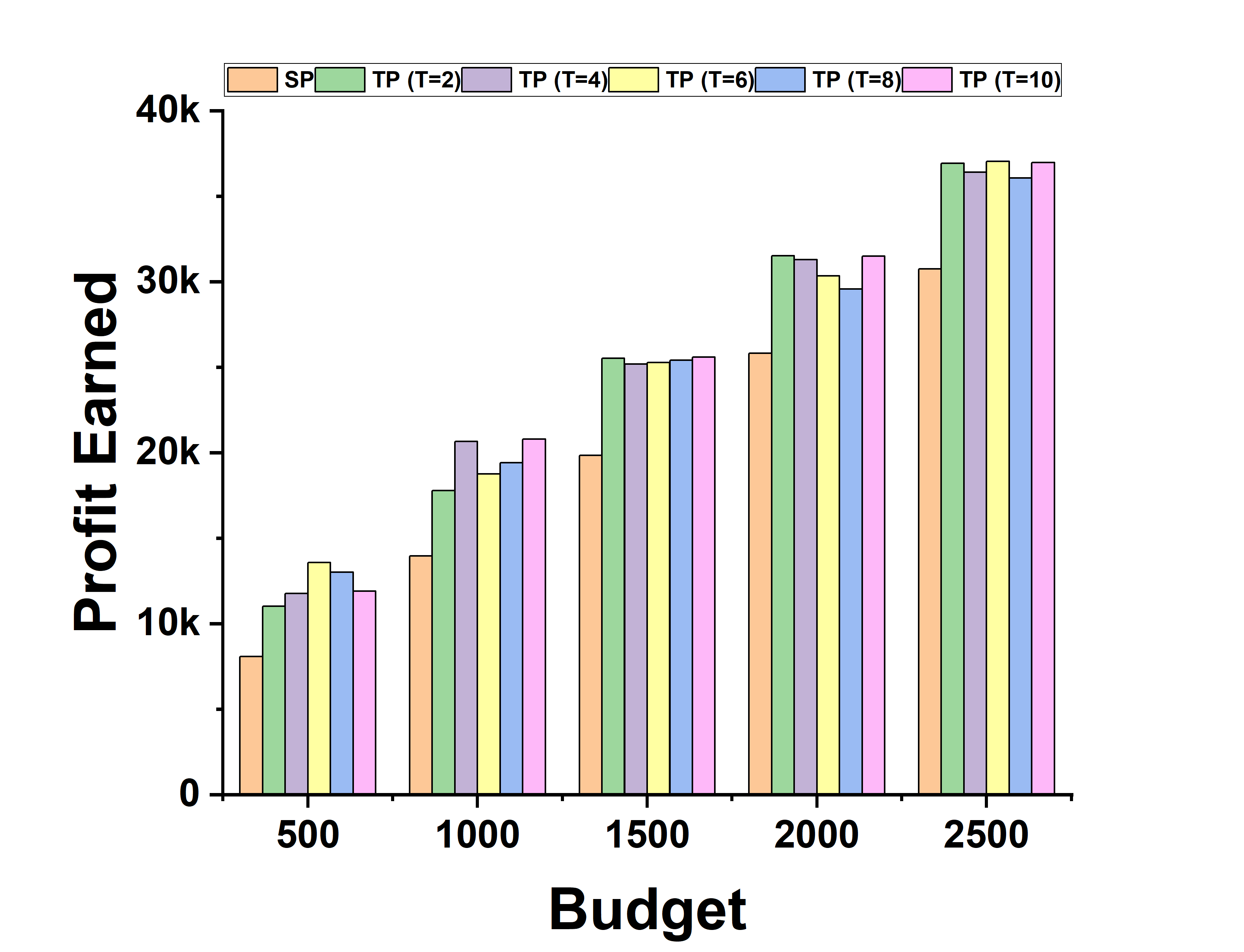}
        \caption{Single Discount}
    \end{subfigure} &
    \begin{subfigure}[t]{0.22\textwidth}
        \includegraphics[width=\linewidth]{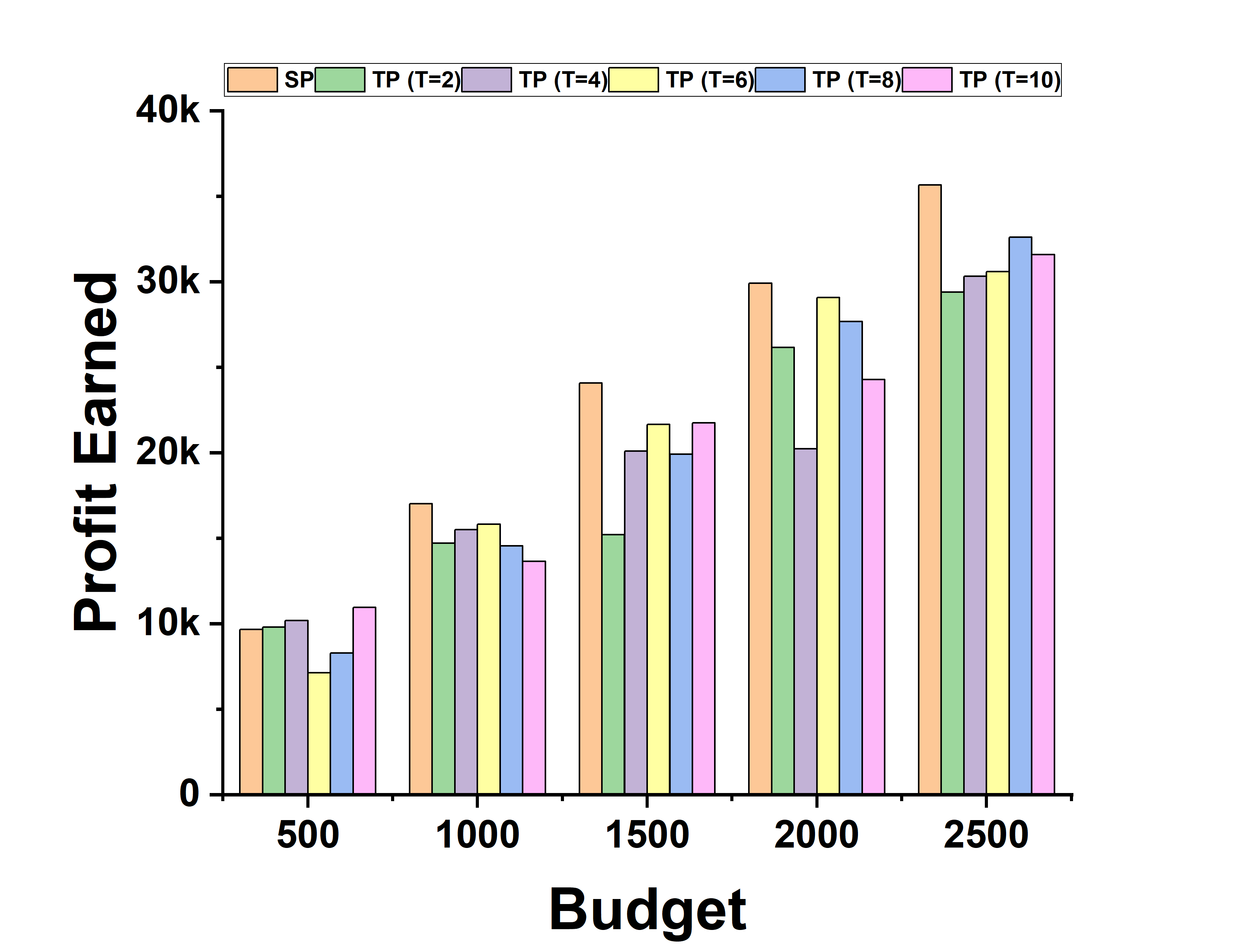}
        \caption{Simple Greedy}
    \end{subfigure} &
    \begin{subfigure}[t]{0.22\textwidth}
        \includegraphics[width=\linewidth]{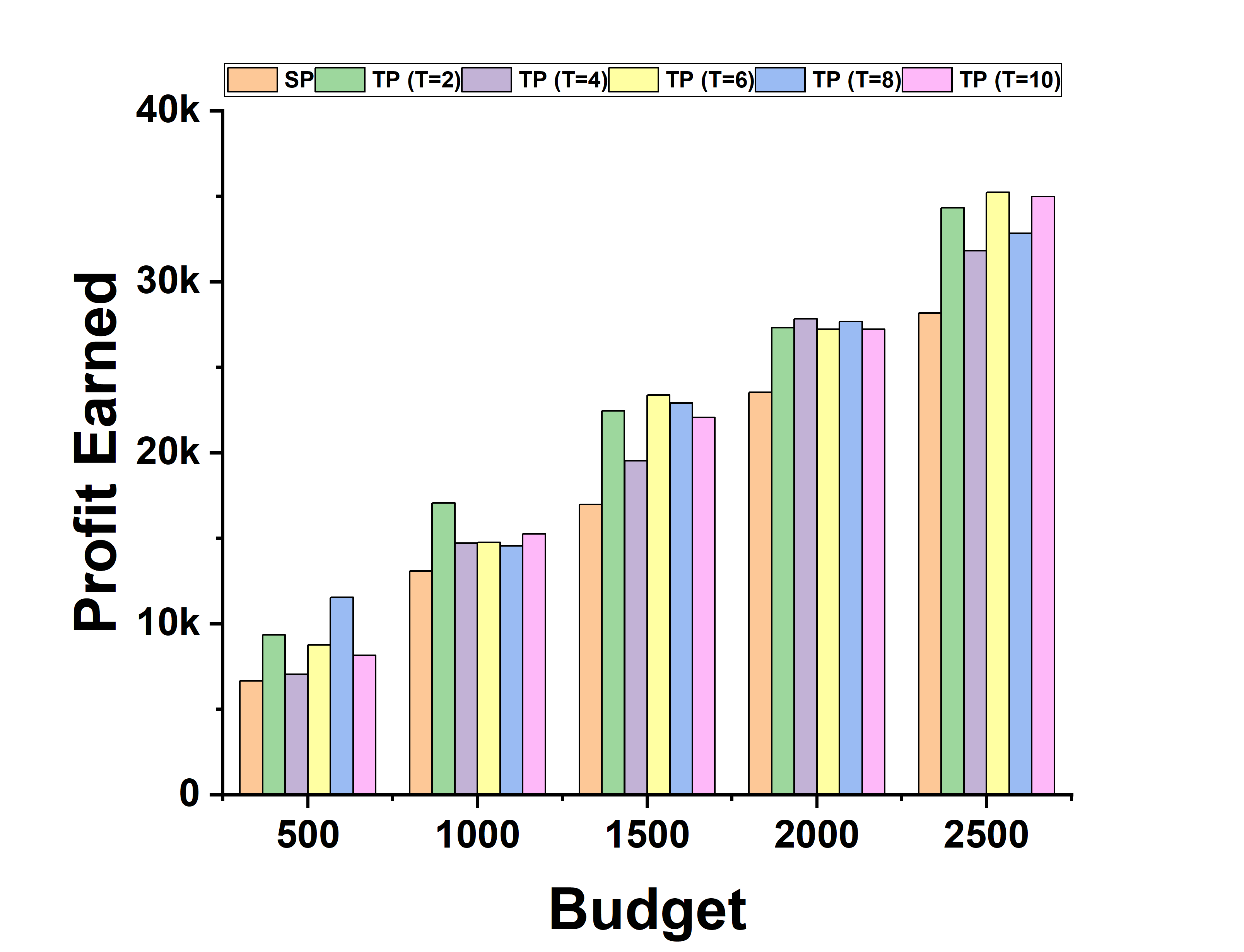}
        \caption{Double Greedy}
    \end{subfigure} &
    \begin{subfigure}[t]{0.22\textwidth}
        \includegraphics[width=\linewidth]{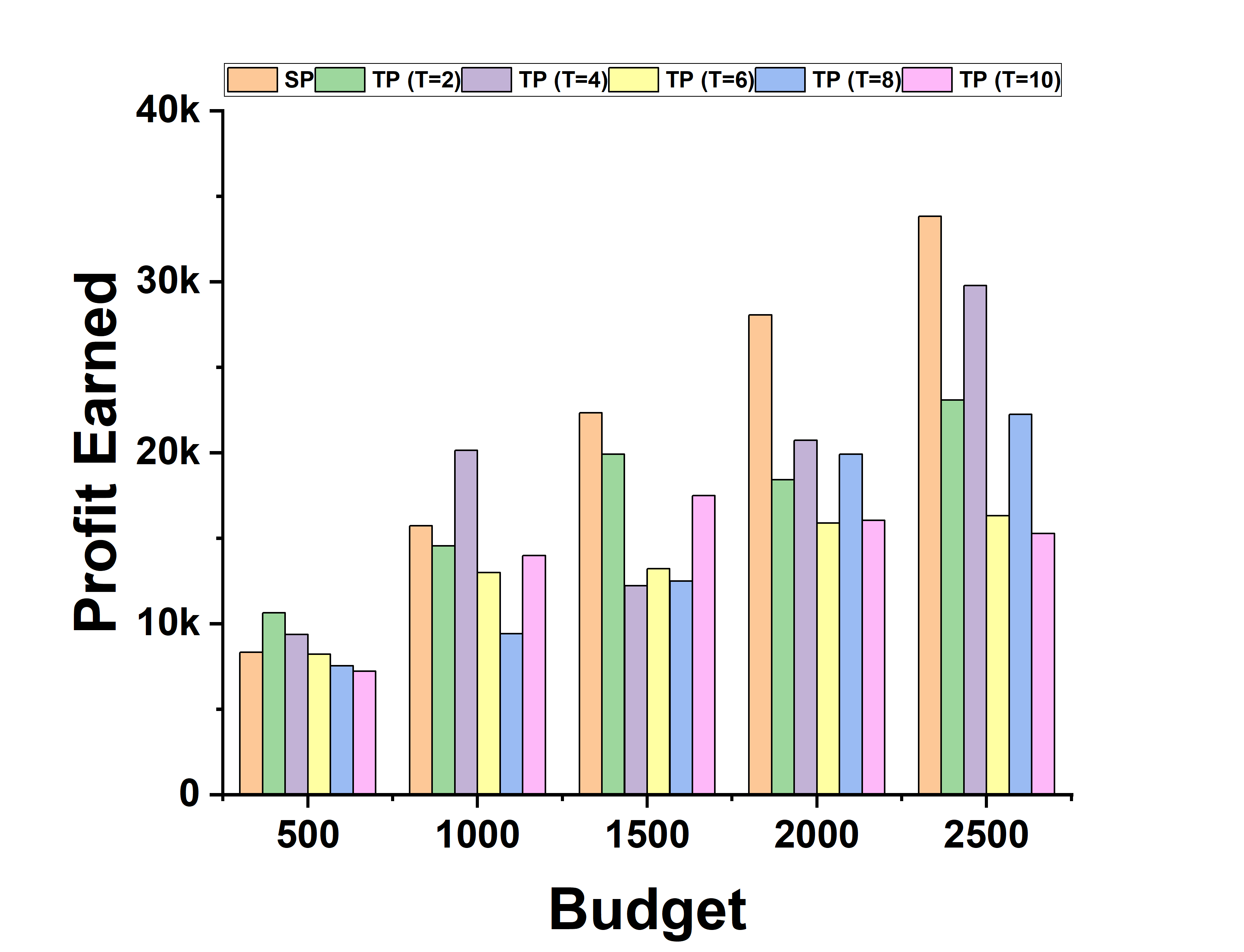}
        \caption{Stochastic Greedy}
    \end{subfigure}
\end{tabular}
\caption{Profit Earned in Single Phase Vs. Two Phase setting (split ratio 70\%, Probability Setting - Trivalency, \textit{LM} Dataset)}
\label{Fig:RQ1LM_T4}
\end{figure}

\begin{figure}[htbp]
\centering
\captionsetup[sub]{font=footnotesize}
\begin{tabular}{cccc}
    \begin{subfigure}[t]{0.22\textwidth}
        \includegraphics[width=\linewidth]{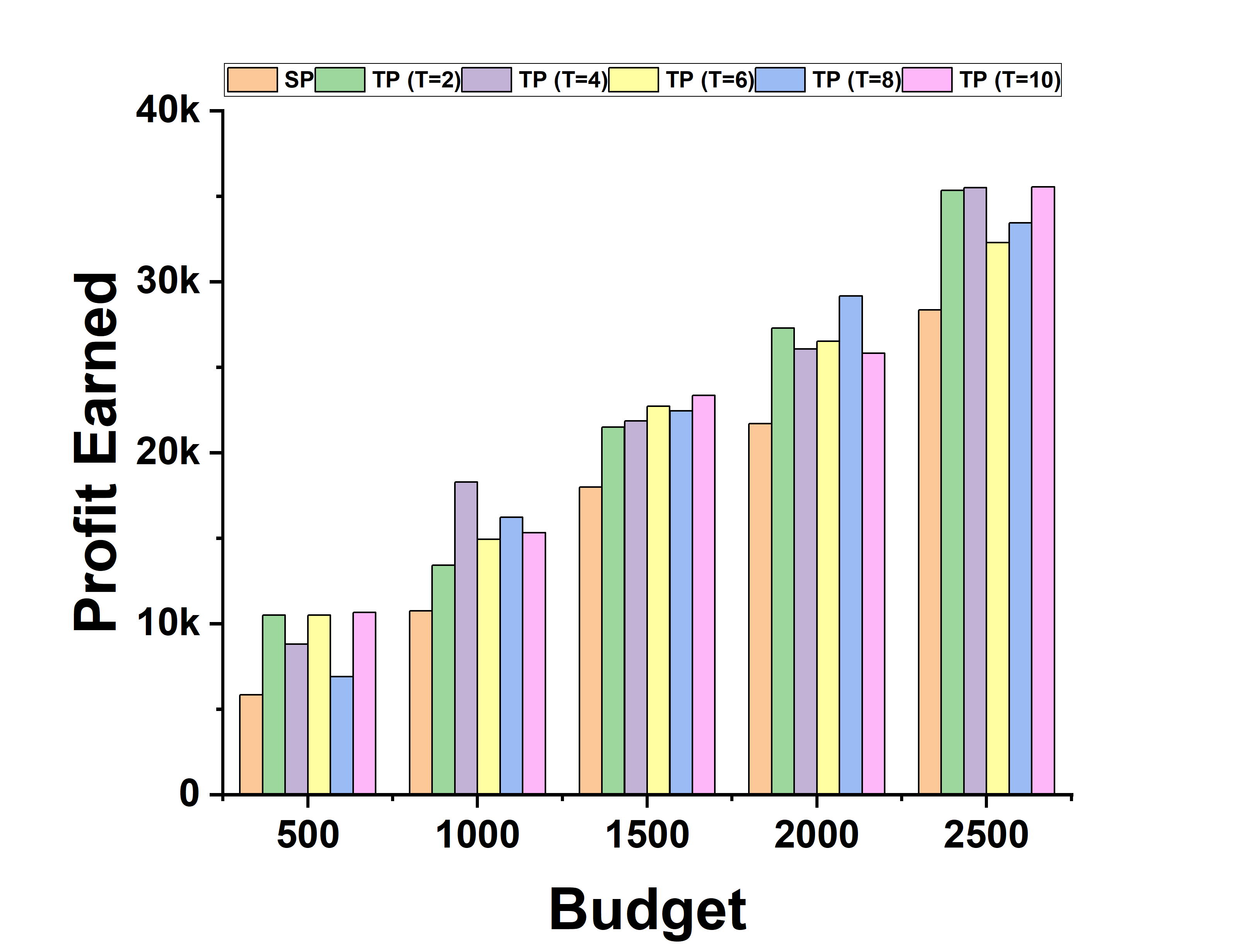}
        \caption{Random}
    \end{subfigure} &
    \begin{subfigure}[t]{0.22\textwidth}
        \includegraphics[width=\linewidth]{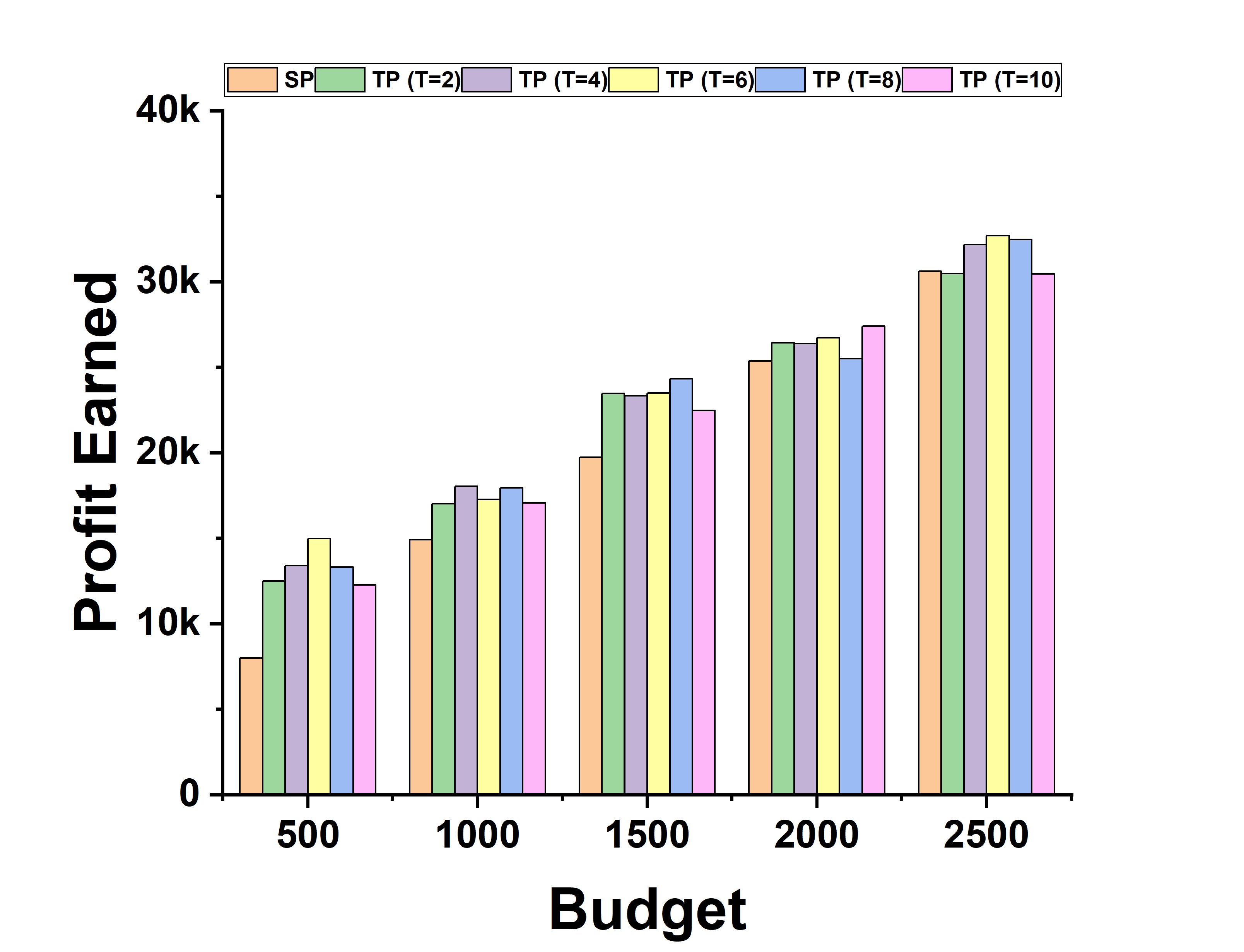}
        \caption{High Degree}
    \end{subfigure} &
    \begin{subfigure}[t]{0.22\textwidth}
        \includegraphics[width=\linewidth]{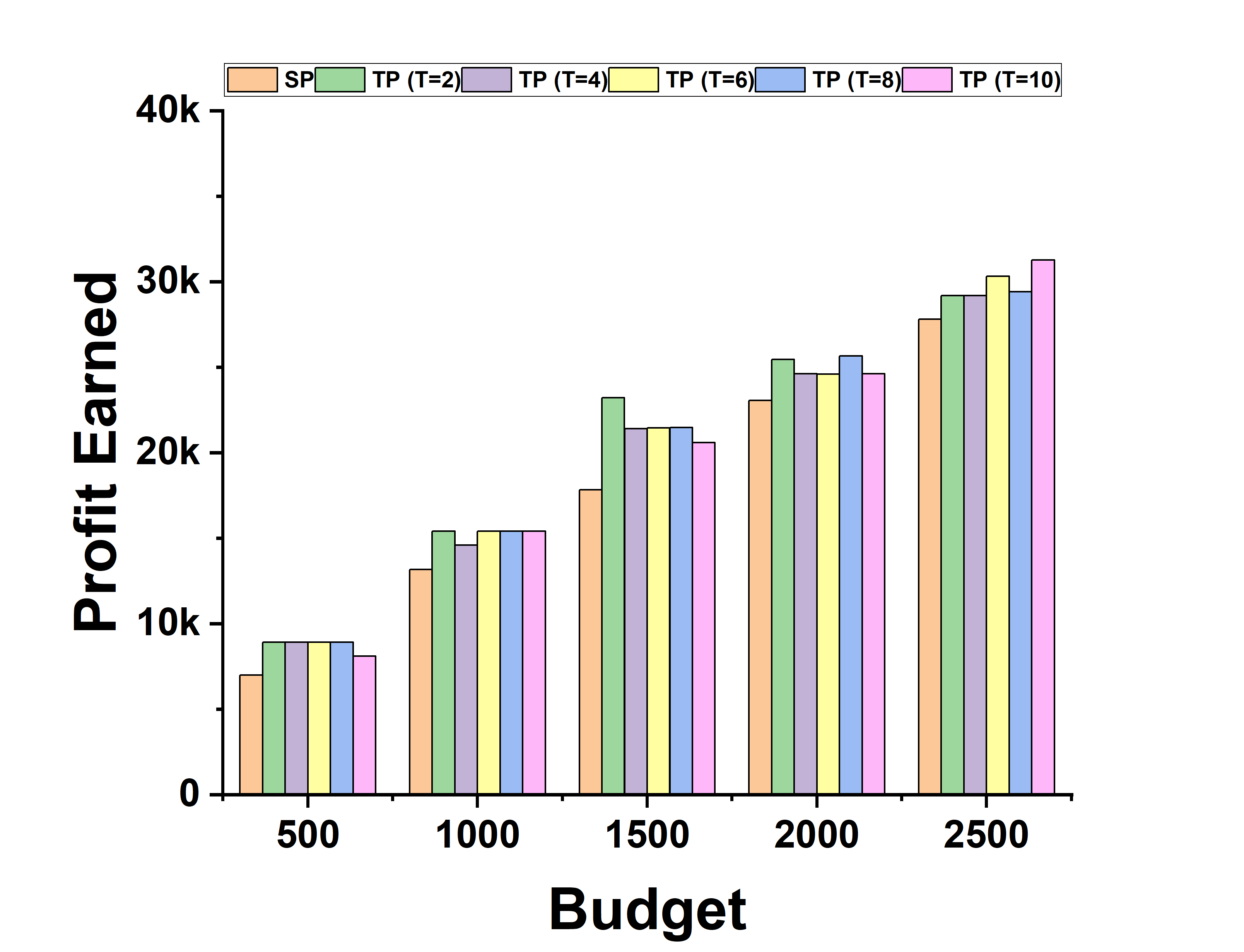}
        \caption{Clustering\\Coefficient}
    \end{subfigure} &
    \begin{subfigure}[t]{0.22\textwidth}
        \includegraphics[width=\linewidth]{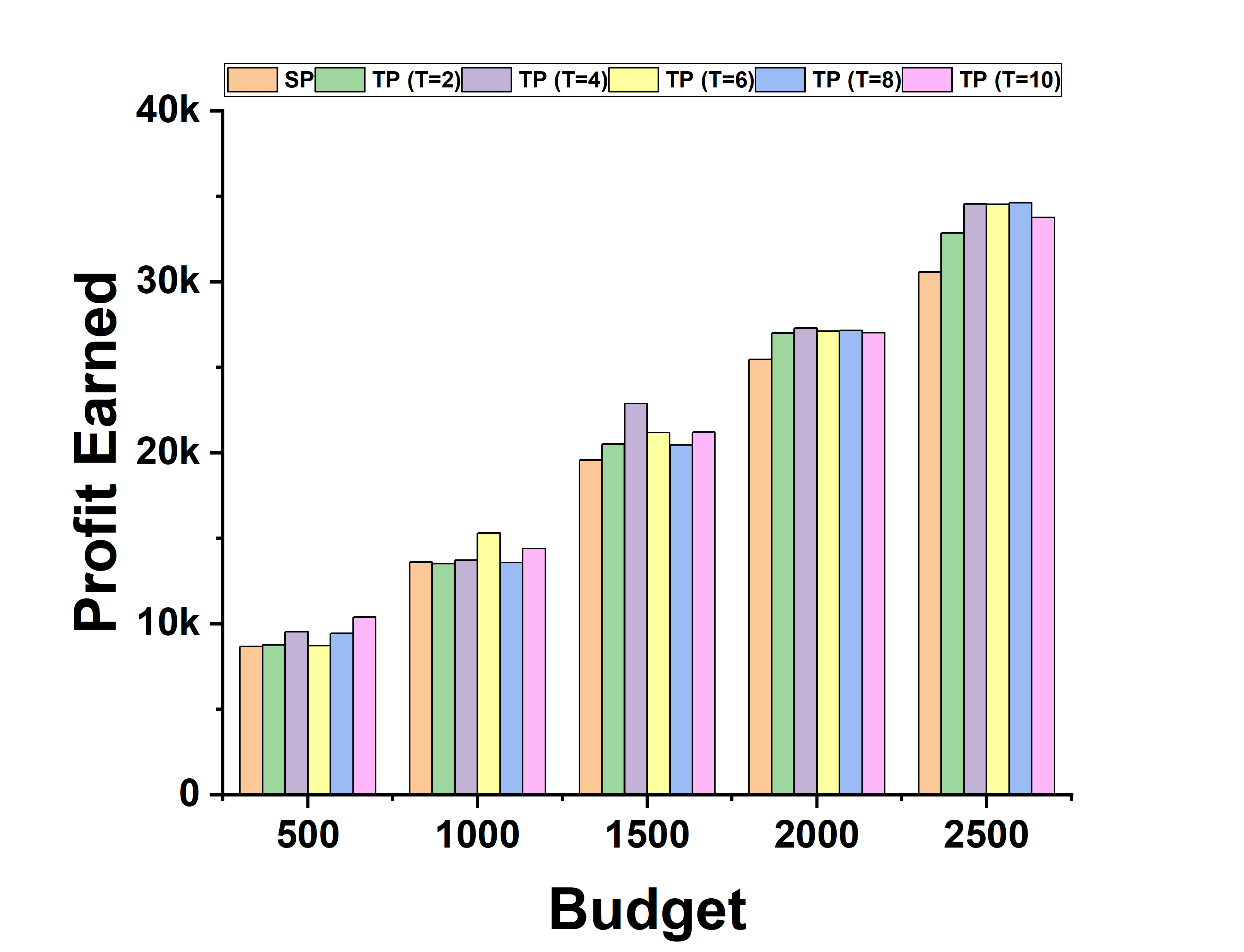}
        \caption{Degree Discount}
    \end{subfigure} \\[6pt]

    \begin{subfigure}[t]{0.22\textwidth}
        \includegraphics[width=\linewidth]{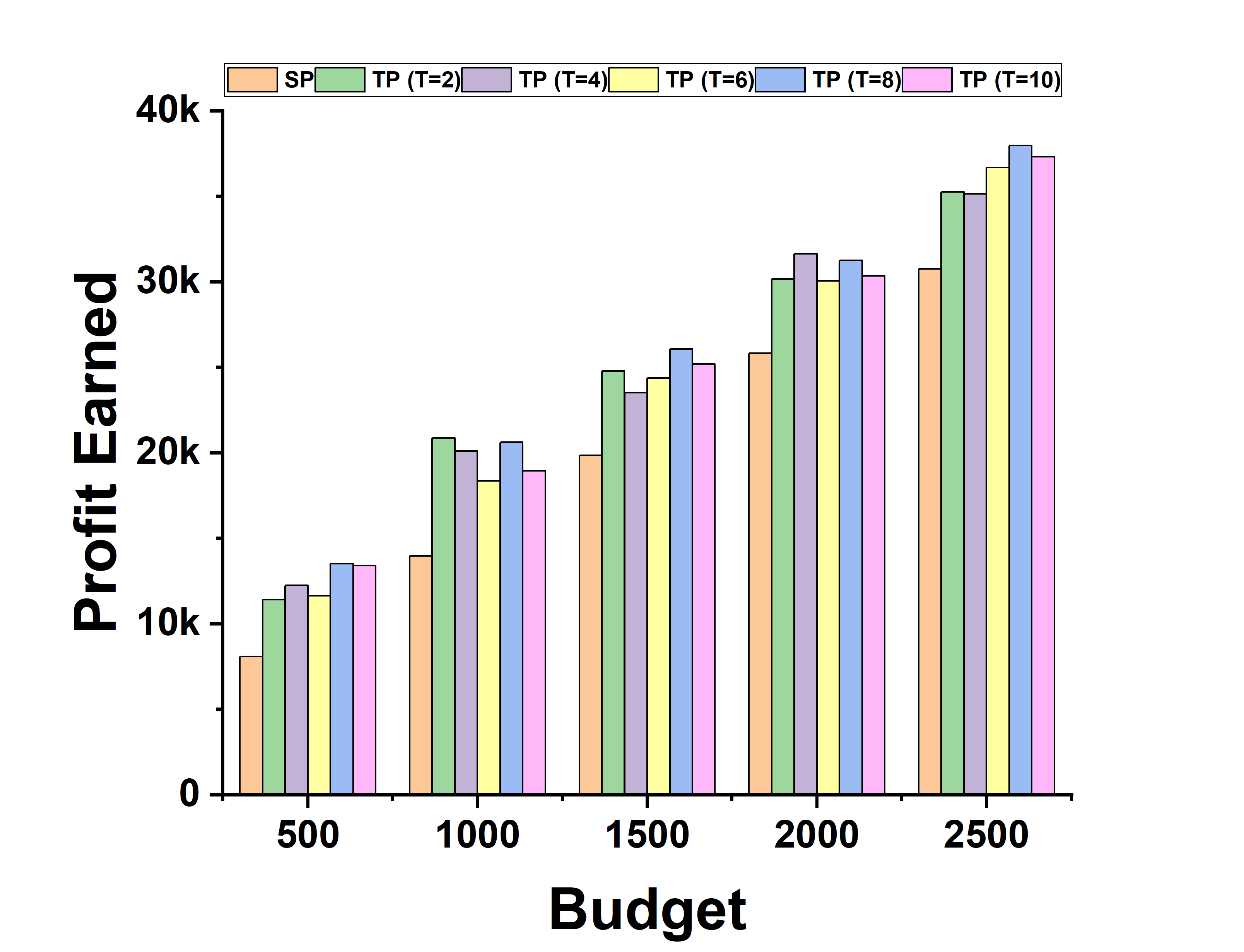}
        \caption{Single Discount}
    \end{subfigure} &
    \begin{subfigure}[t]{0.22\textwidth}
        \includegraphics[width=\linewidth]{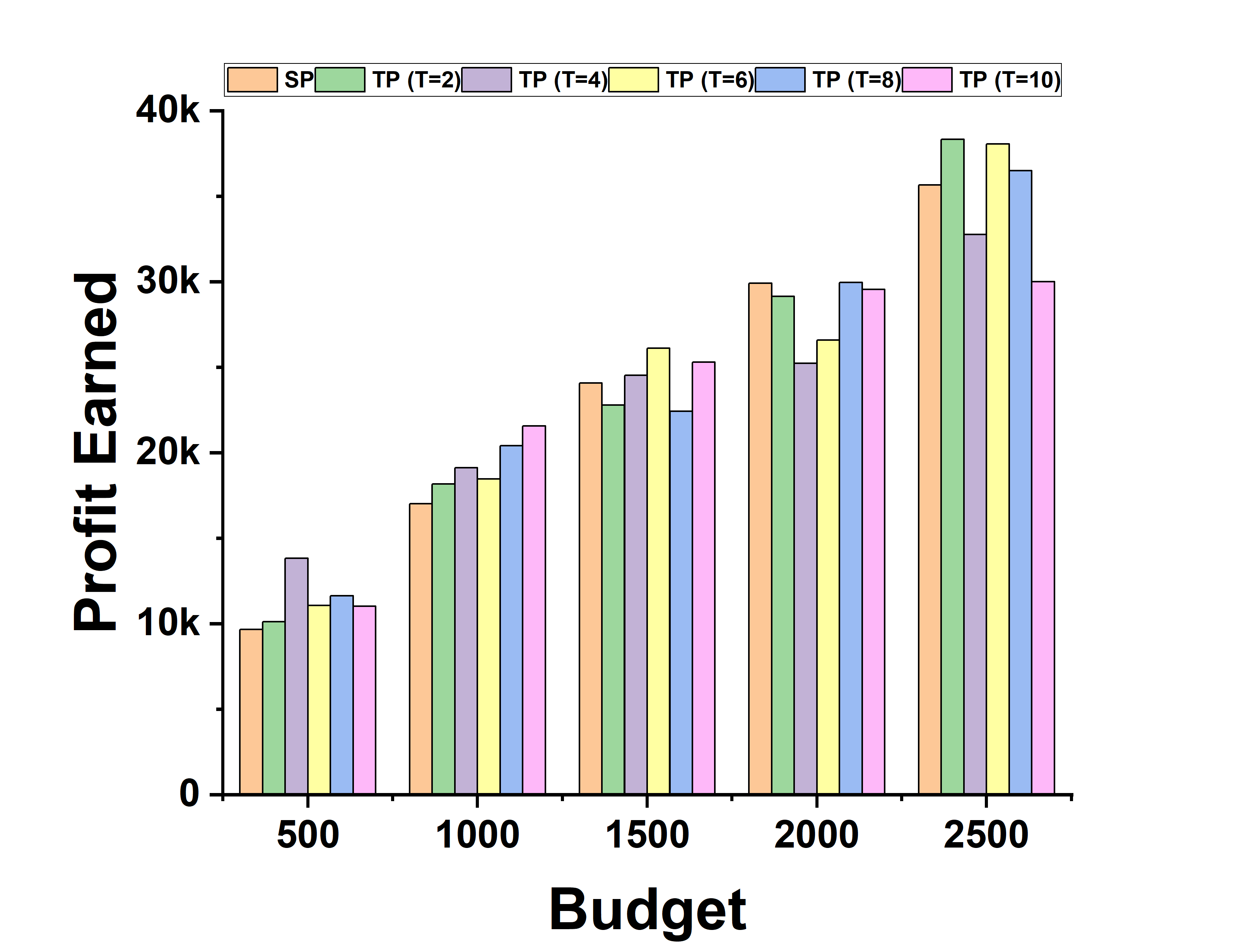}
        \caption{Simple Greedy}
    \end{subfigure} &
    \begin{subfigure}[t]{0.22\textwidth}
        \includegraphics[width=\linewidth]{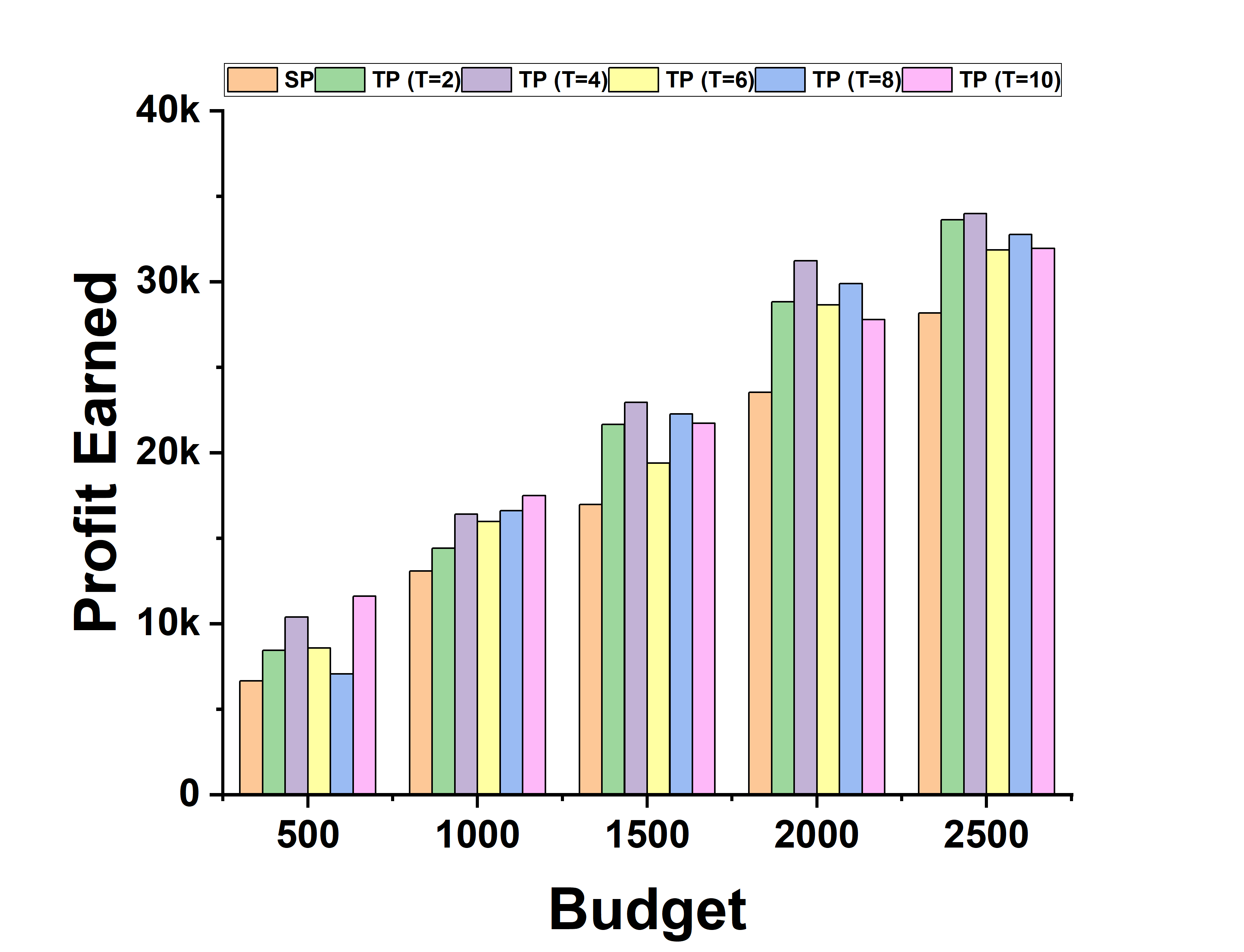}
        \caption{Double Greedy}
    \end{subfigure} &
    \begin{subfigure}[t]{0.22\textwidth}
        \includegraphics[width=\linewidth]{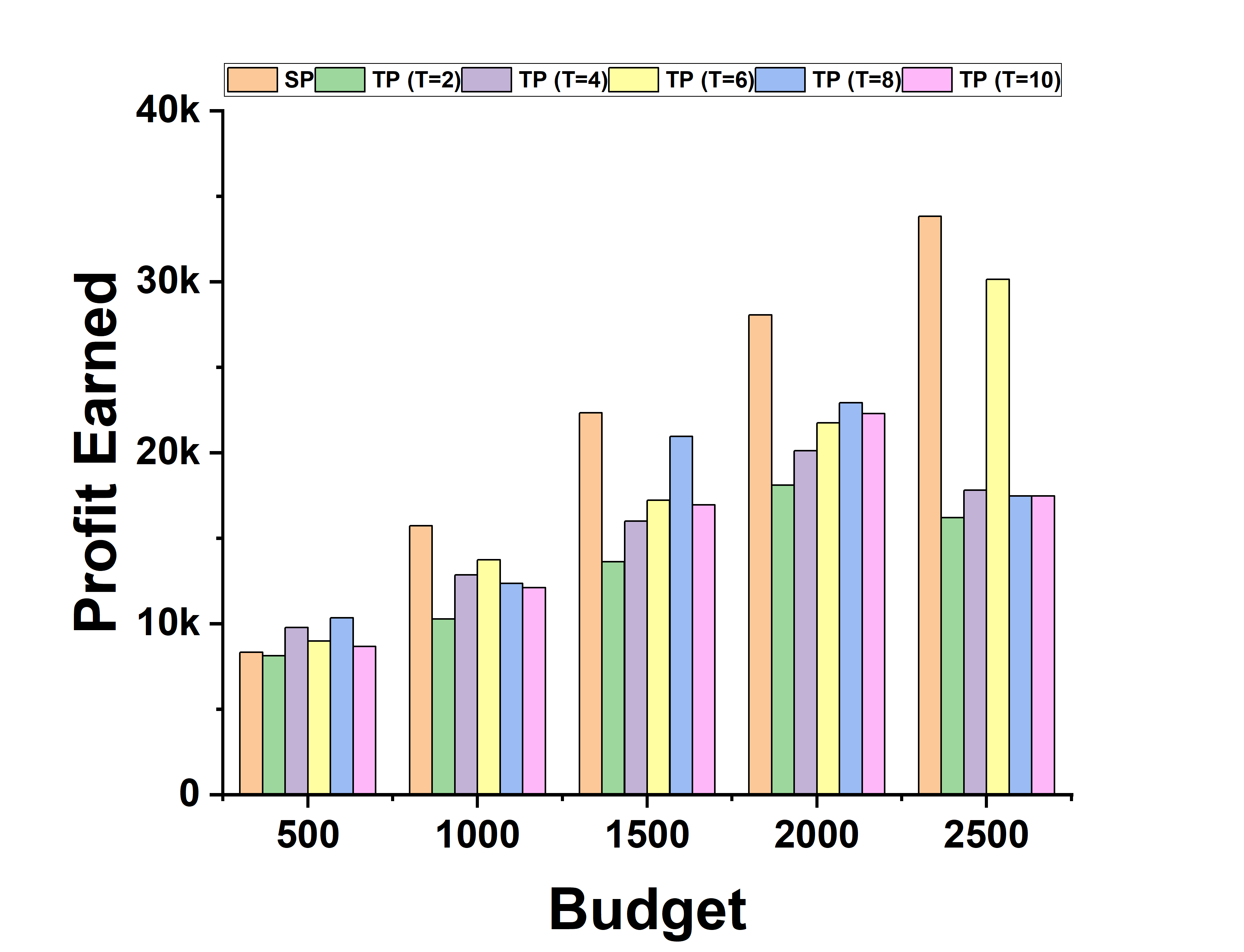}
        \caption{Stochastic Greedy}
    \end{subfigure}
\end{tabular}
\caption{Profit Earned in Single Phase Vs. Two Phase setting (split ratio 90\%, Probability Setting - Trivalency, \textit{LM} Dataset)}
\label{Fig:RQ1LM_T5}
\end{figure}


\begin{figure}[htbp]
\centering
\captionsetup[sub]{font=footnotesize}  
\begin{tabular}{cccc}
    \begin{subfigure}[t]{0.22\textwidth}
        \includegraphics[width=\linewidth]{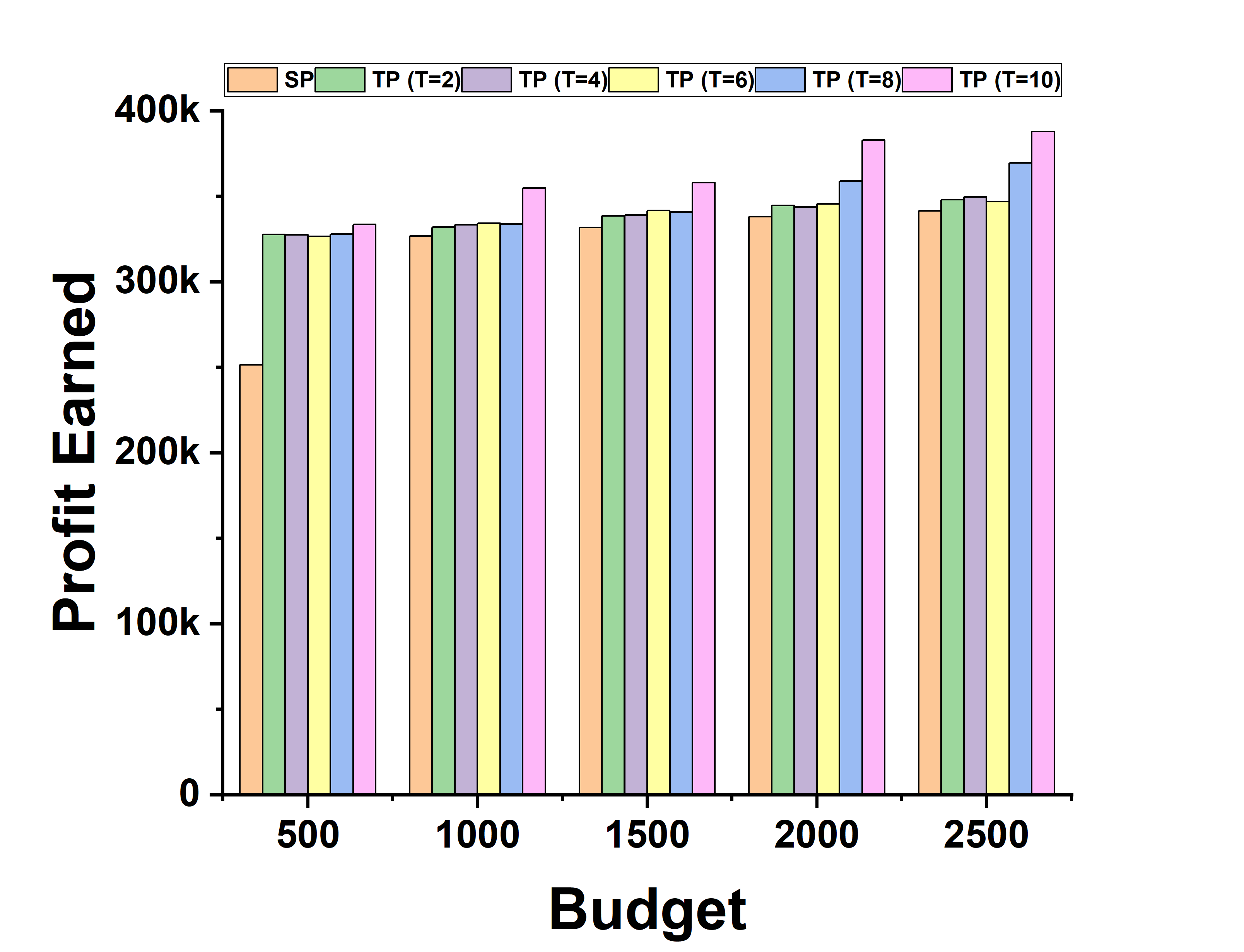}
        \caption{Random}
    \end{subfigure} &
    \begin{subfigure}[t]{0.22\textwidth}
        \includegraphics[width=\linewidth]{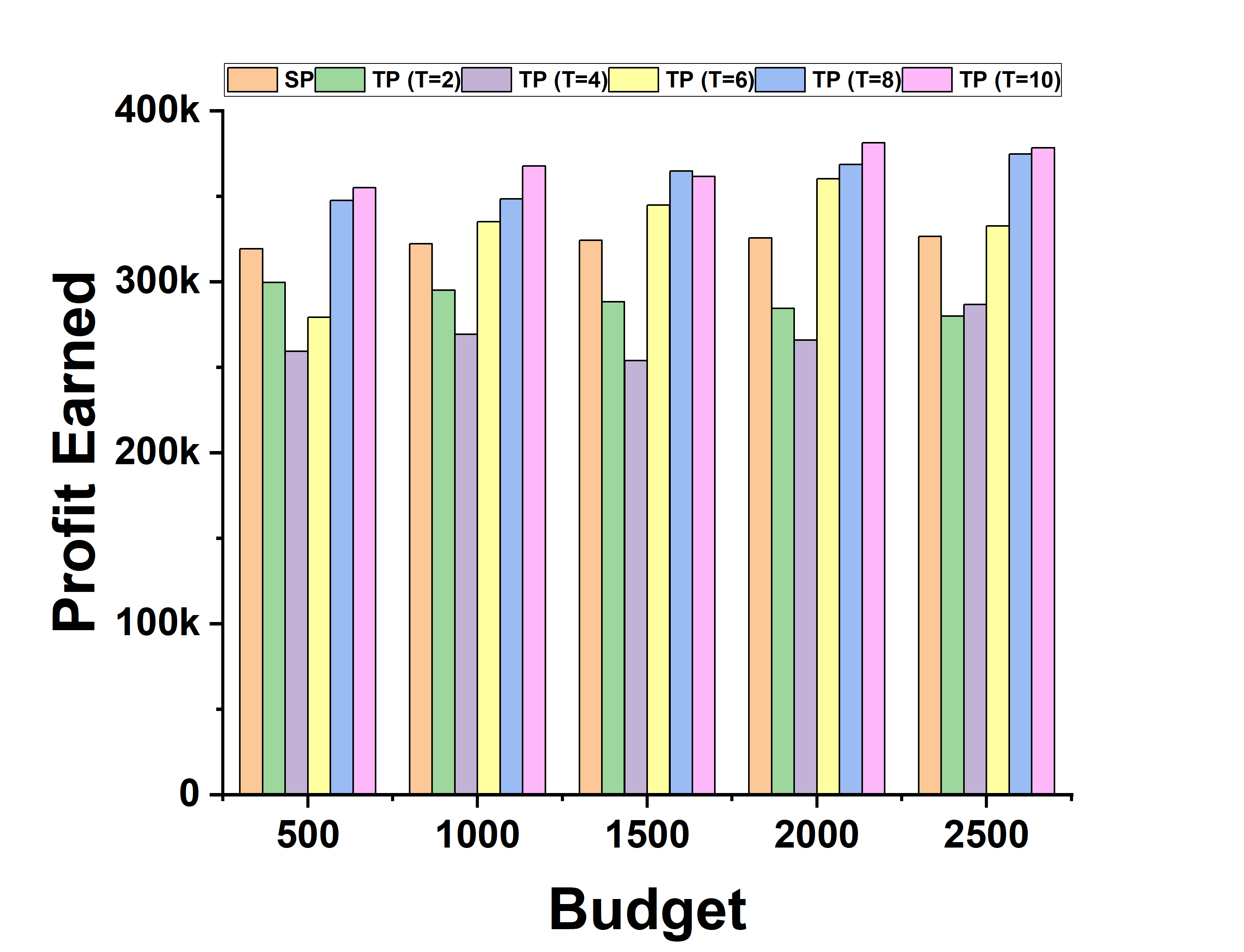}
        \caption{High Degree}
    \end{subfigure} &
    \begin{subfigure}[t]{0.22\textwidth}
        \includegraphics[width=\linewidth]{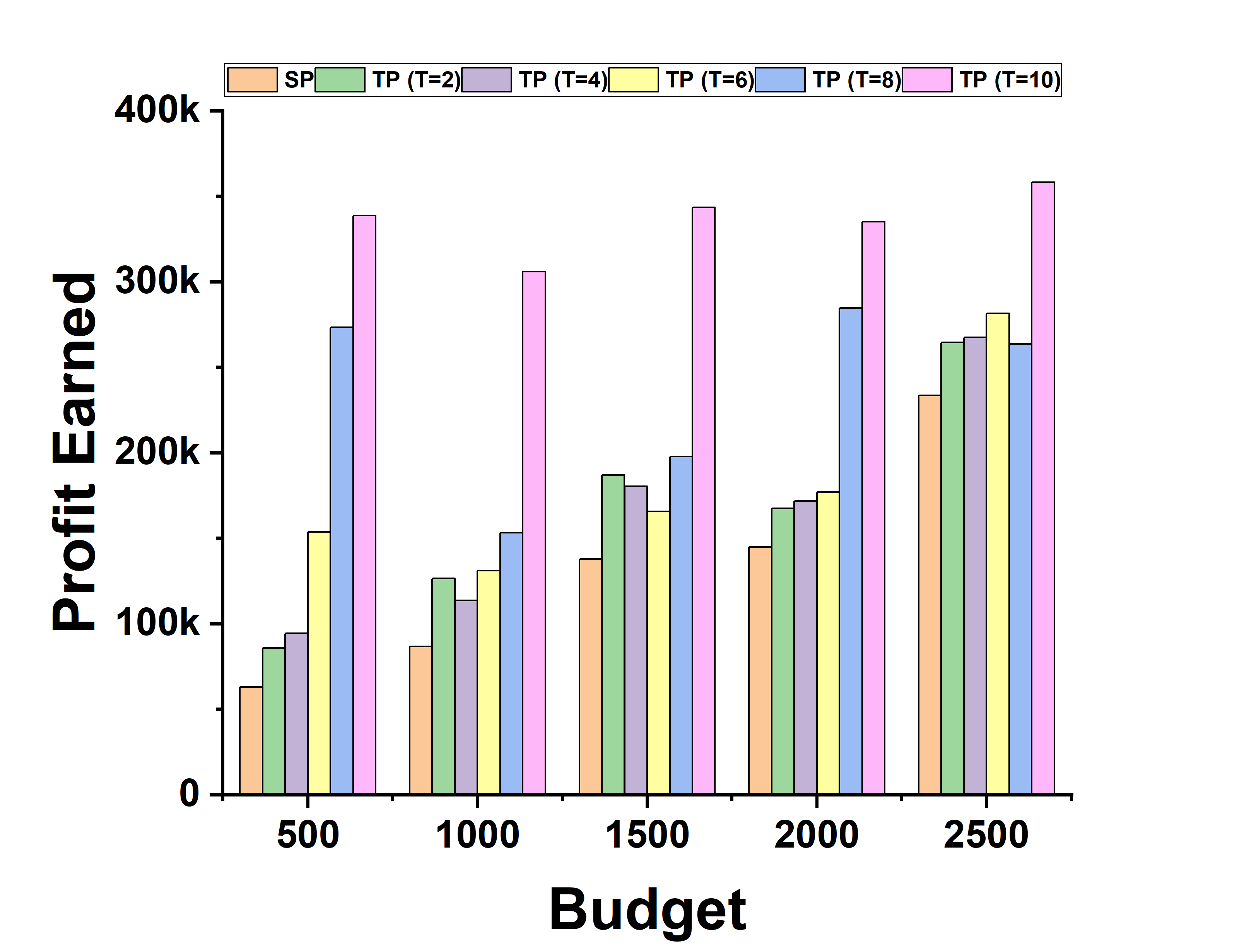}
        \caption{Clustering\\Coefficient}
    \end{subfigure} &
    \begin{subfigure}[t]{0.22\textwidth}
        \includegraphics[width=\linewidth]{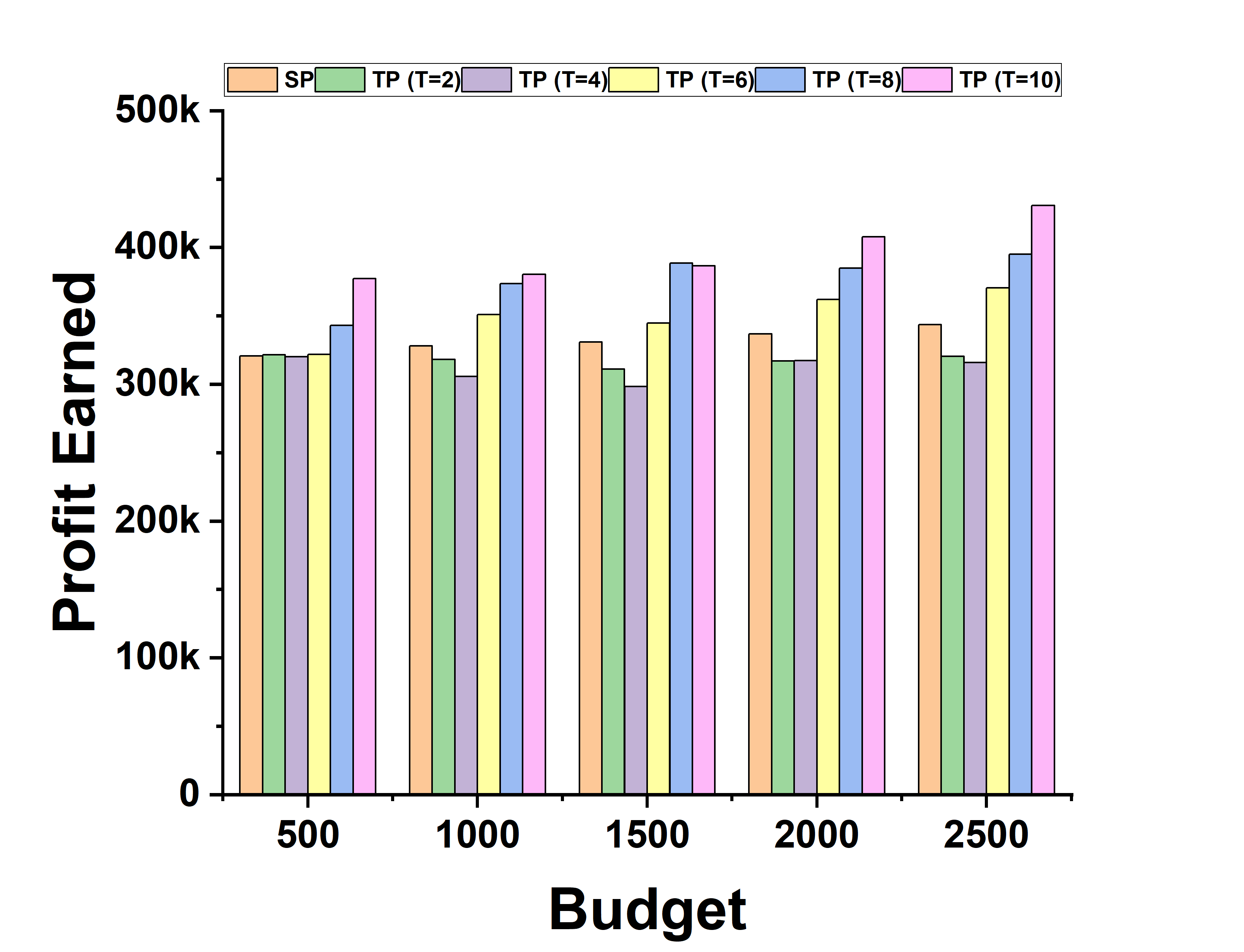}
        \caption{Degree Discount}
    \end{subfigure} \\[6pt]

    \begin{subfigure}[t]{0.22\textwidth}
        \includegraphics[width=\linewidth]{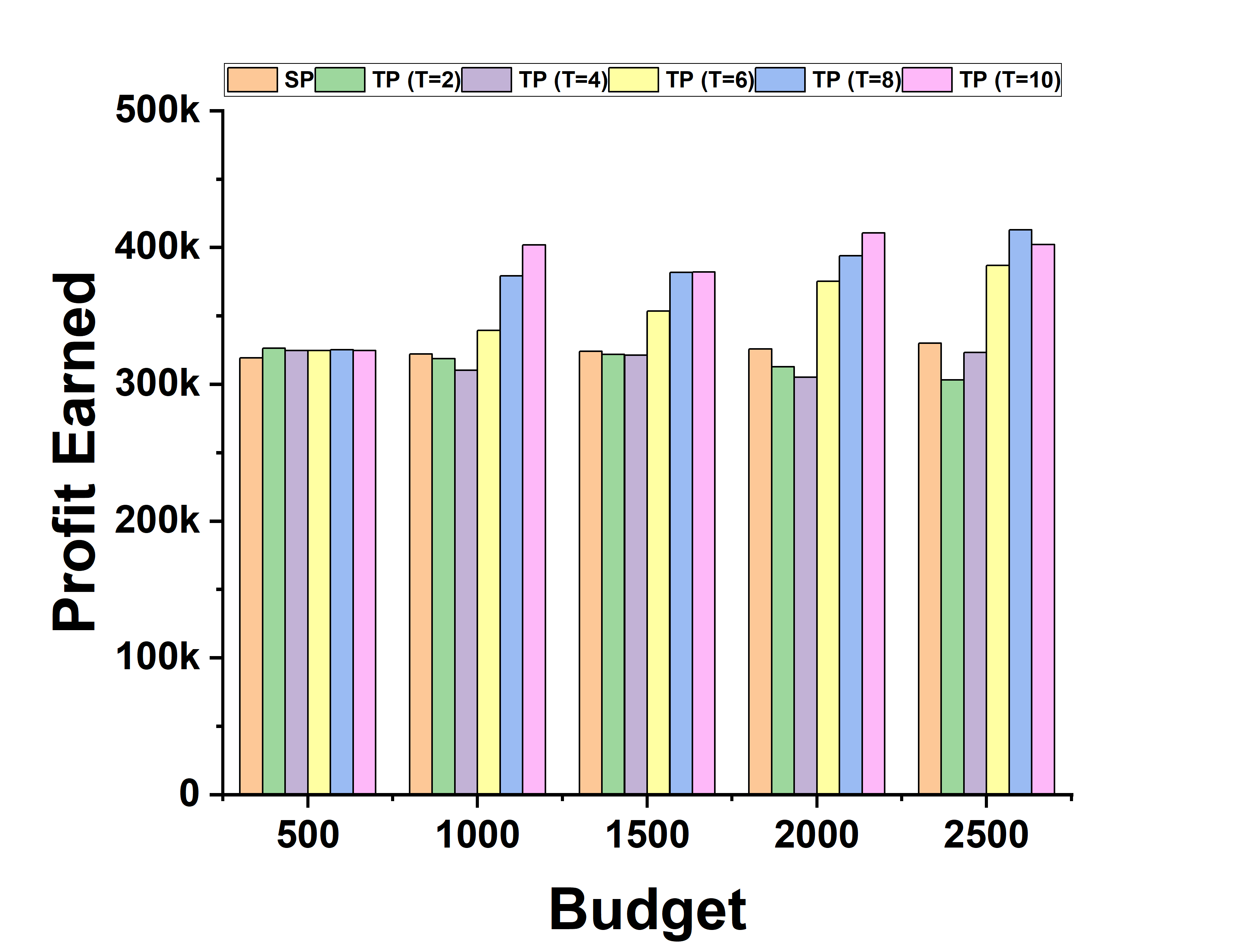}
        \caption{Single Discount}
    \end{subfigure} &
    \begin{subfigure}[t]{0.22\textwidth}
        \includegraphics[width=\linewidth]{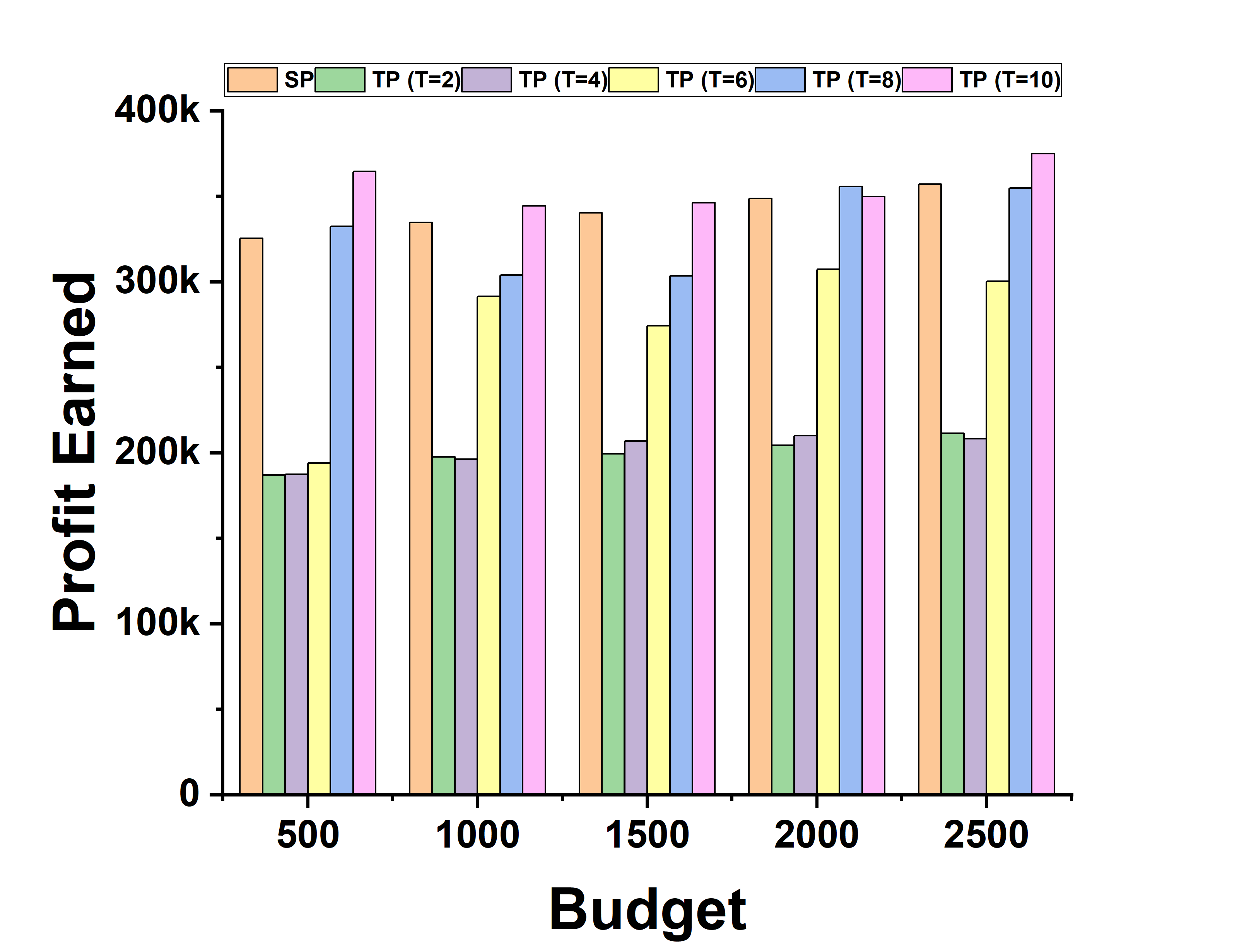}
        \caption{Simple Greedy}
    \end{subfigure} &
    \begin{subfigure}[t]{0.22\textwidth}
        \includegraphics[width=\linewidth]{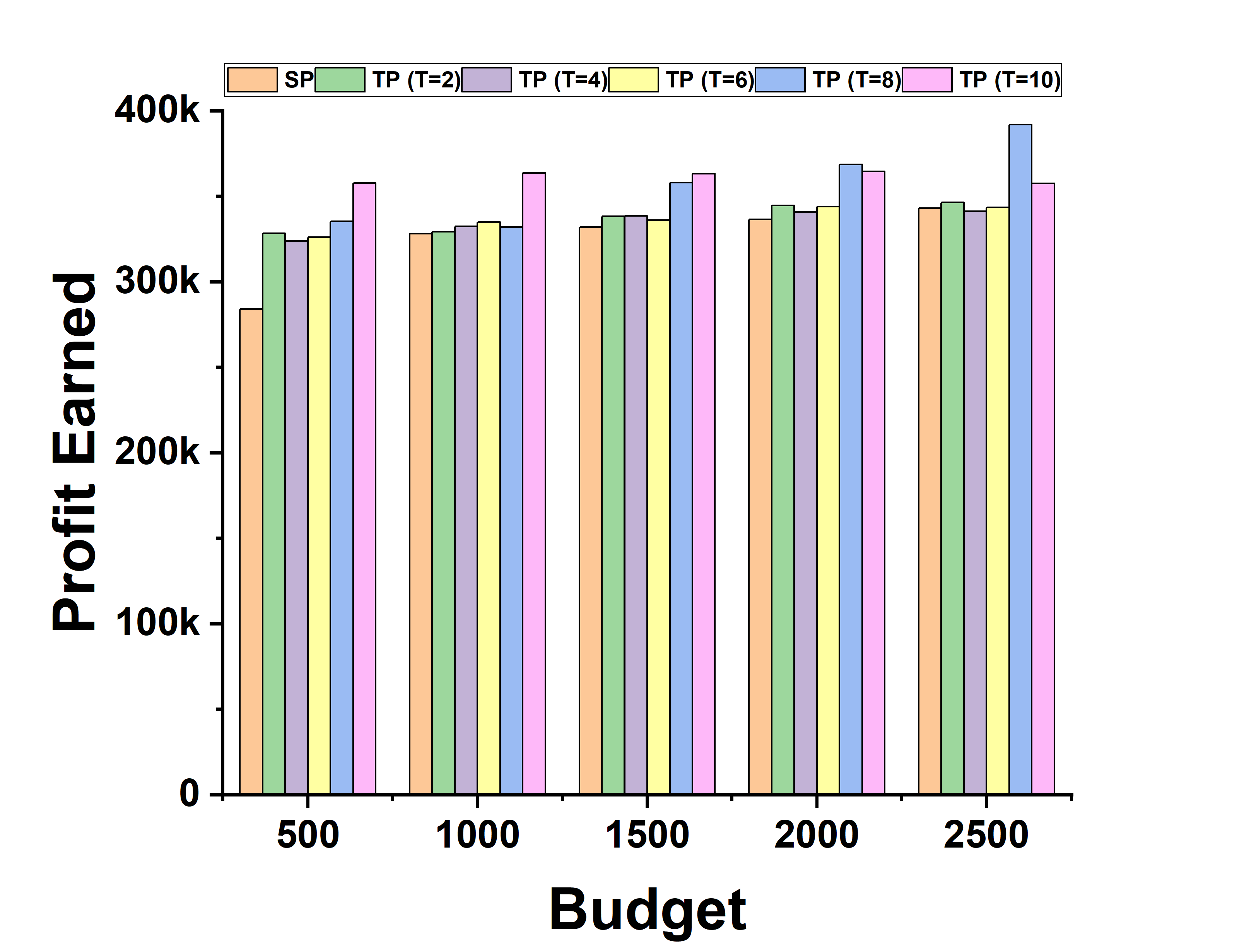}
        \caption{Double Greedy}
    \end{subfigure} &
    \begin{subfigure}[t]{0.22\textwidth}
        \includegraphics[width=\linewidth]{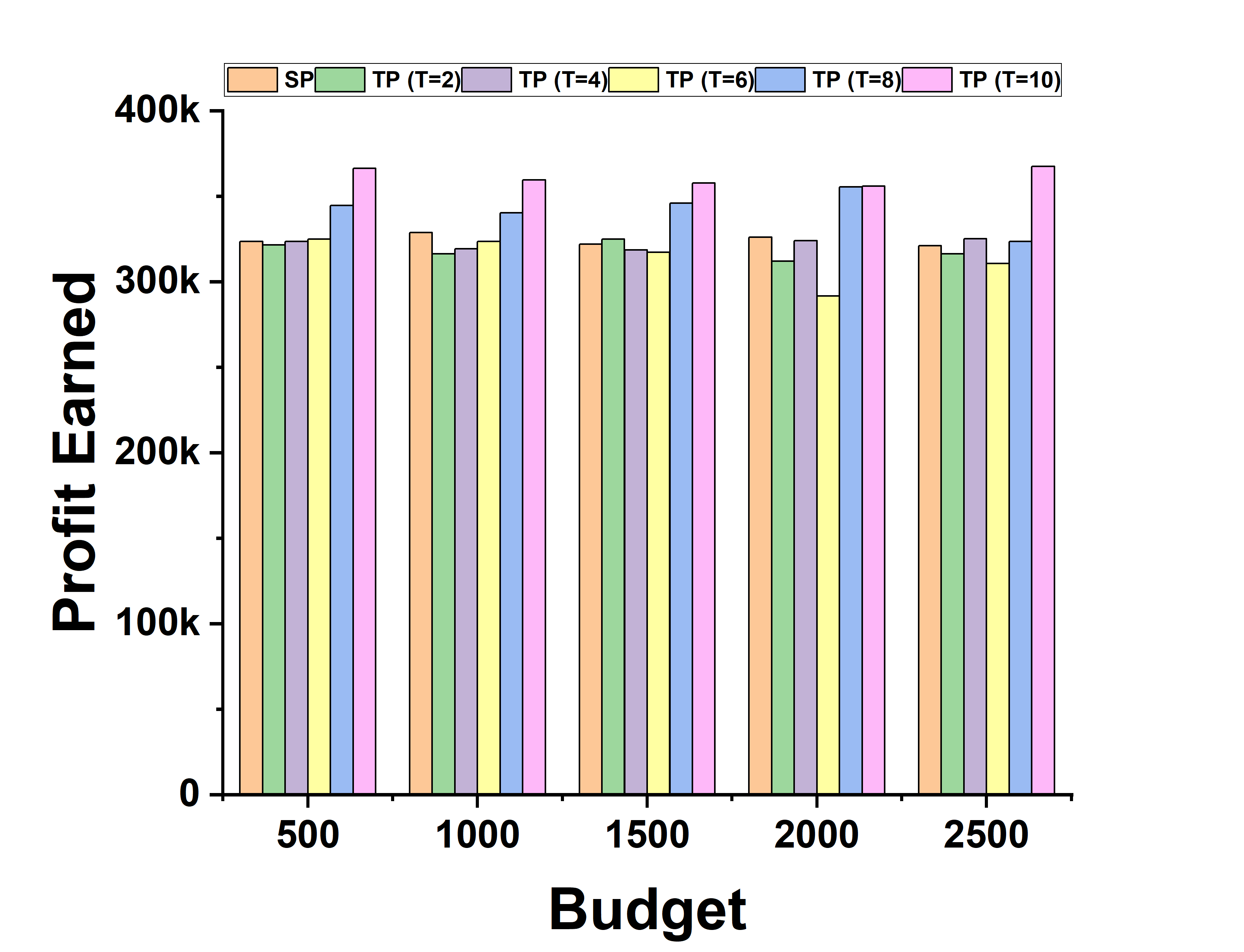}
        \caption{Stochastic Greedy}
    \end{subfigure}
\end{tabular}
\caption{Profit Earned in Single Phase Vs. Two Phase setting (split ratio 10\%, Probability Setting - Trivalency, \textit{Email-Eu-Core} Dataset)}
\label{Fig:RQ1_T1}
\end{figure}

\begin{figure}[htbp]
\centering
\captionsetup[sub]{font=footnotesize}  
\begin{tabular}{cccc}
    \begin{subfigure}[t]{0.22\textwidth}
        \includegraphics[width=\linewidth]{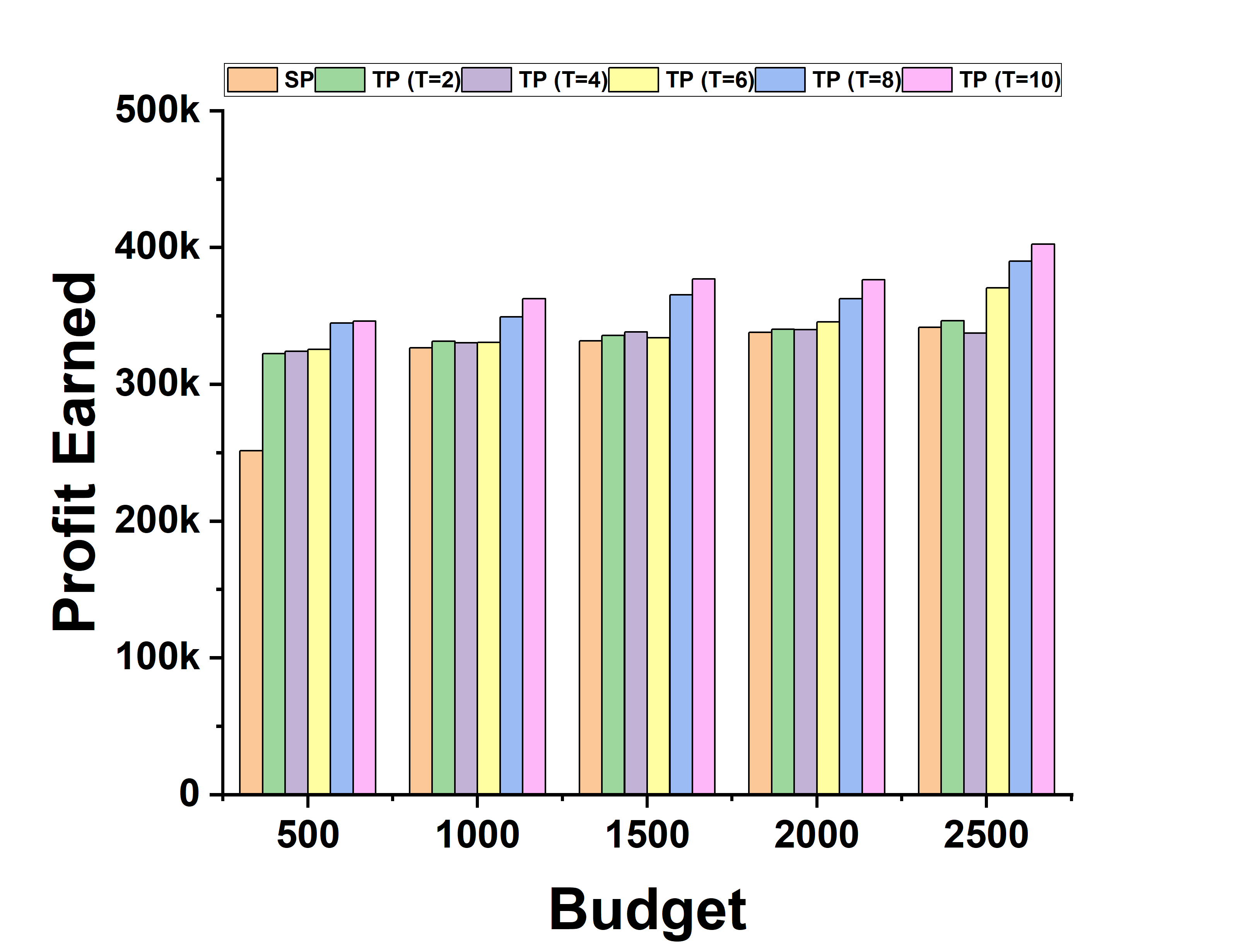}
        \caption{Random}
    \end{subfigure} &
    \begin{subfigure}[t]{0.22\textwidth}
        \includegraphics[width=\linewidth]{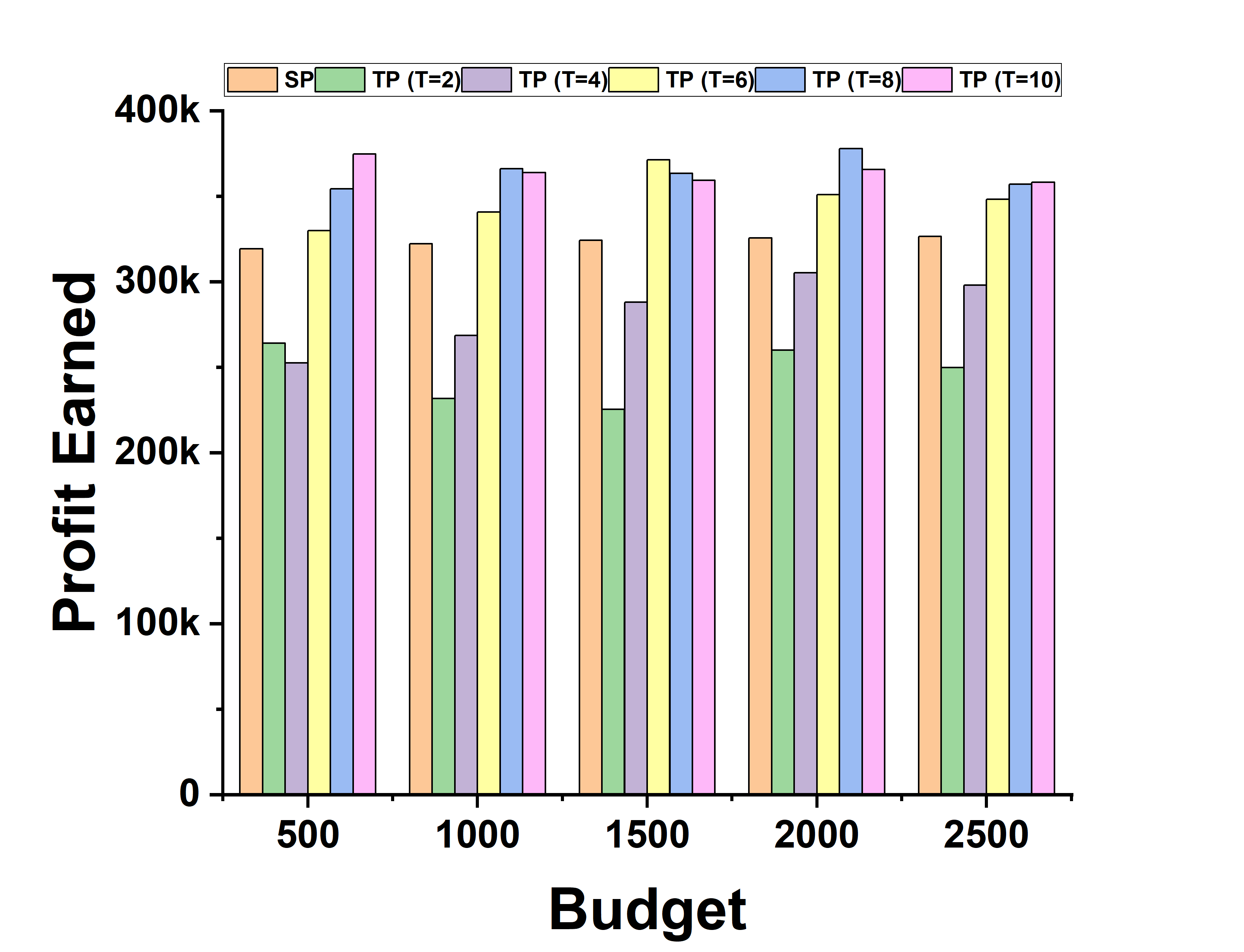}
        \caption{High Degree}
    \end{subfigure} &
    \begin{subfigure}[t]{0.22\textwidth}
        \includegraphics[width=\linewidth]{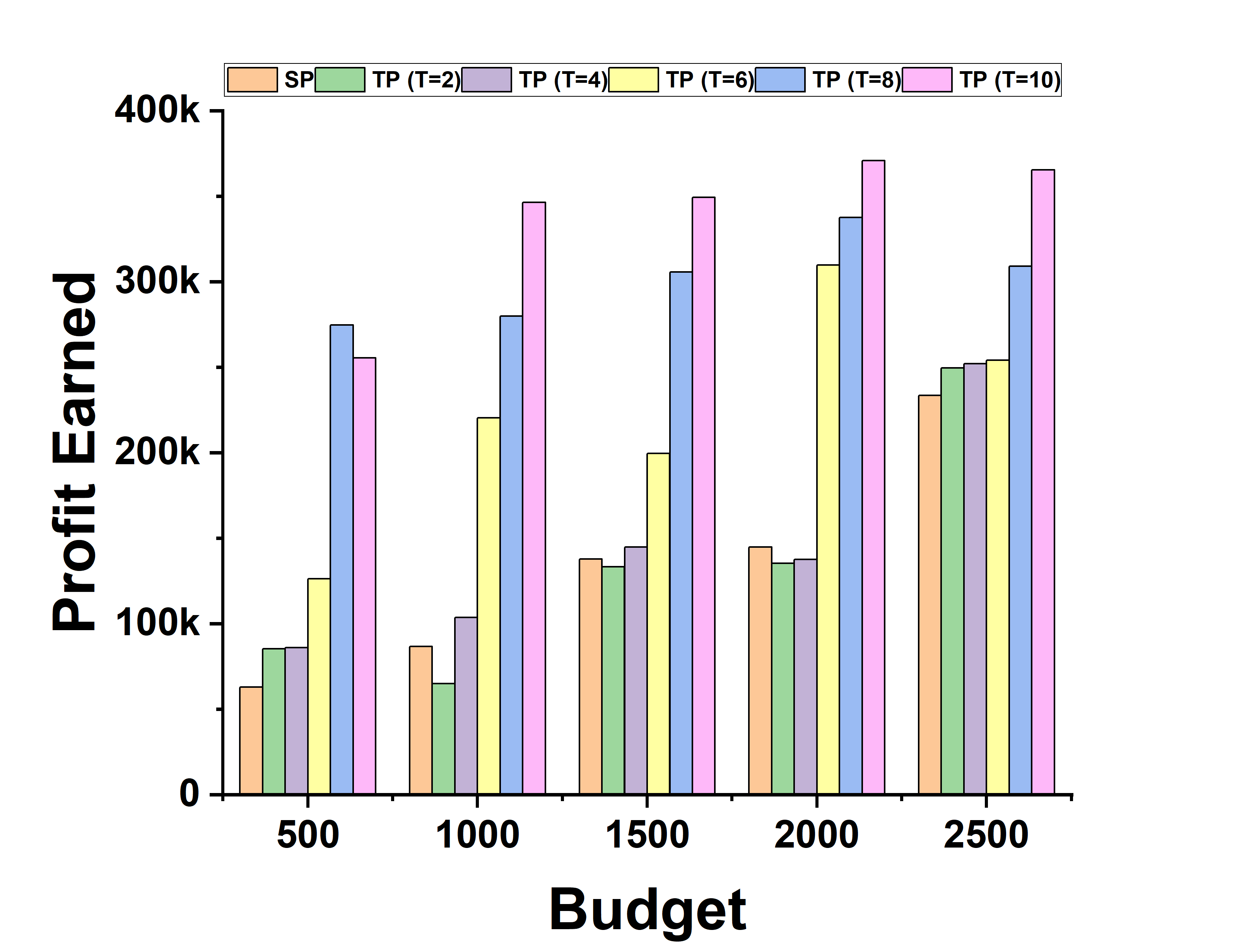}
        \caption{Clustering\\Coefficient}
    \end{subfigure} &
    \begin{subfigure}[t]{0.22\textwidth}
        \includegraphics[width=\linewidth]{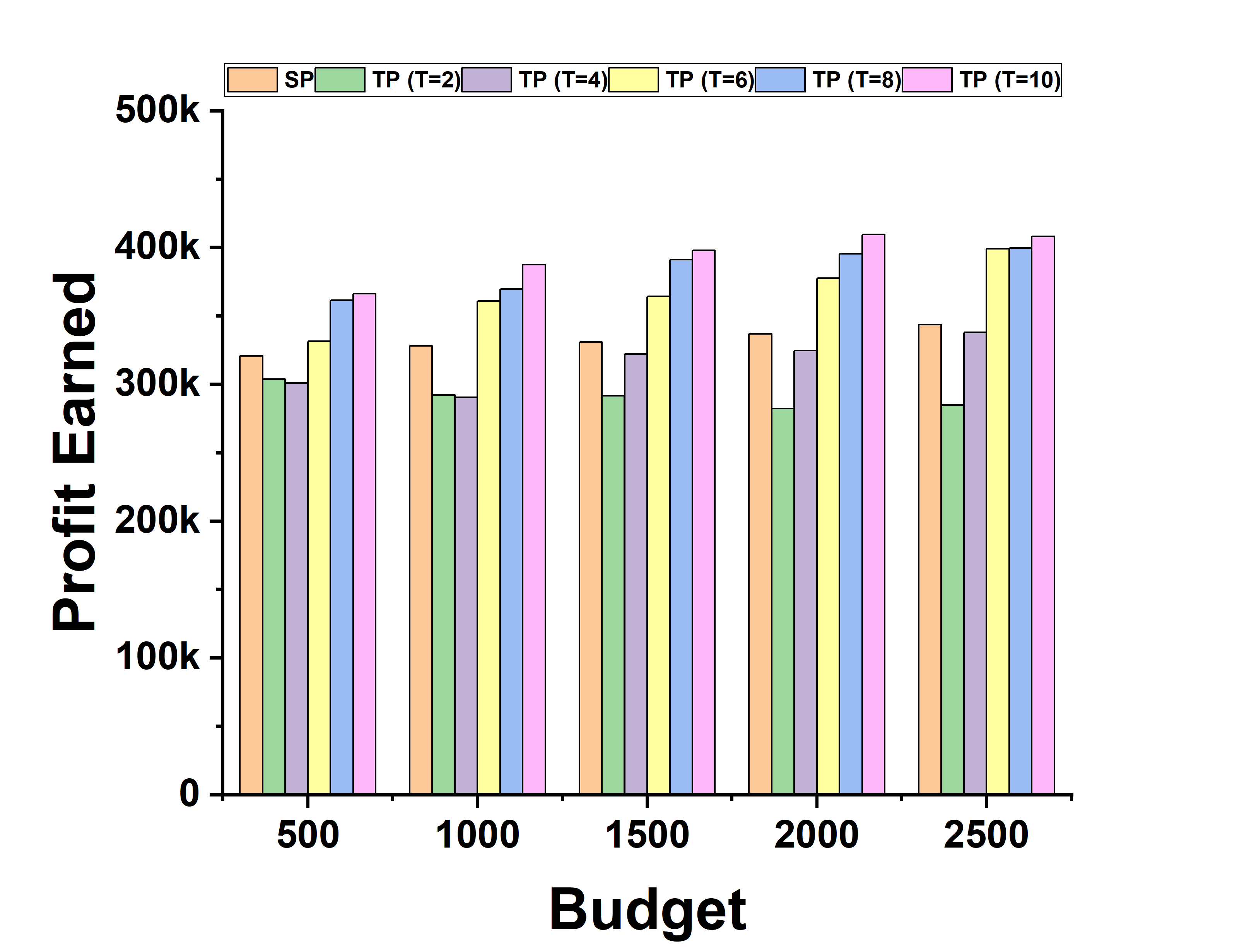}
        \caption{Degree Discount}
    \end{subfigure} \\[6pt]

    \begin{subfigure}[t]{0.22\textwidth}
        \includegraphics[width=\linewidth]{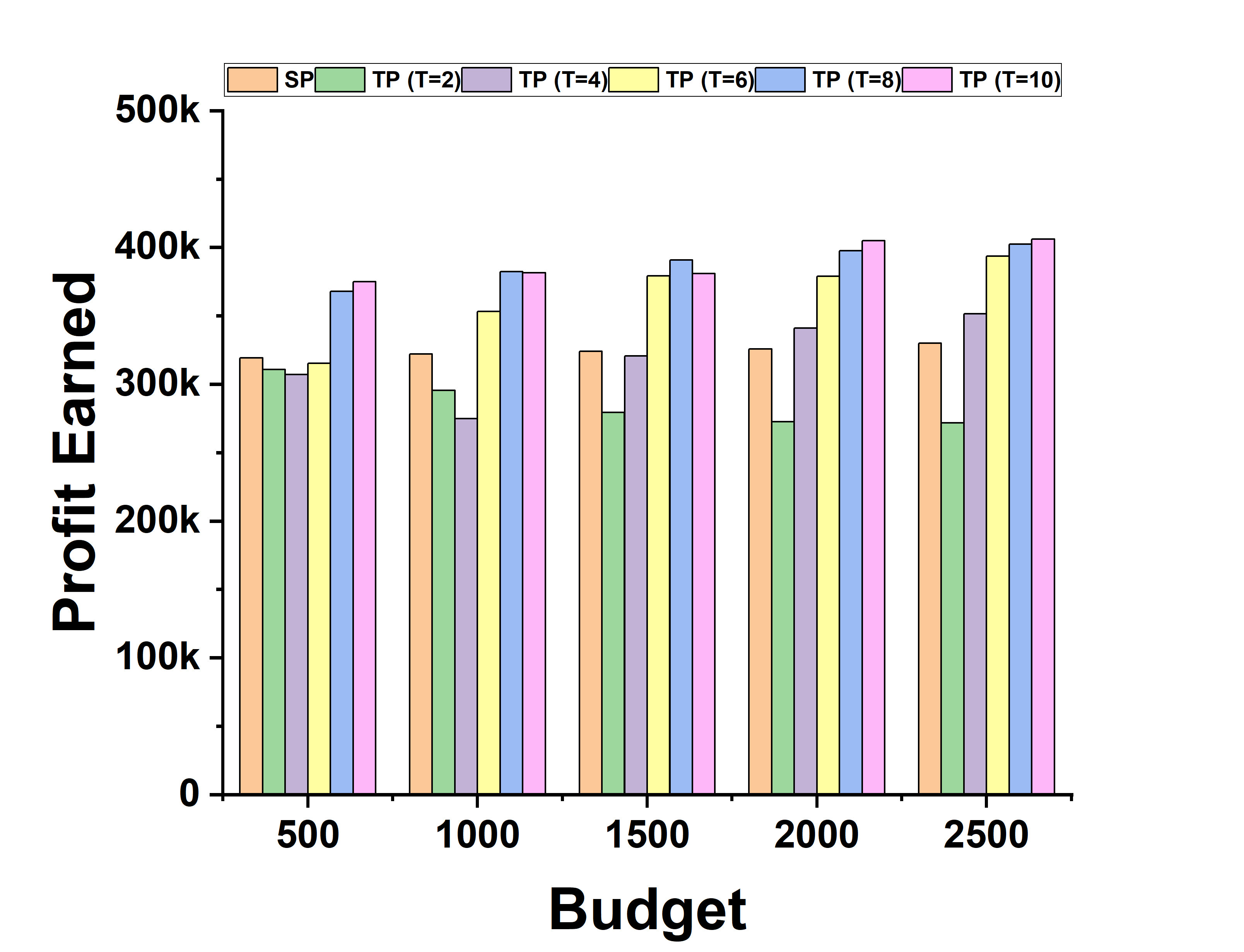}
        \caption{Single Discount}
    \end{subfigure} &
    \begin{subfigure}[t]{0.22\textwidth}
        \includegraphics[width=\linewidth]{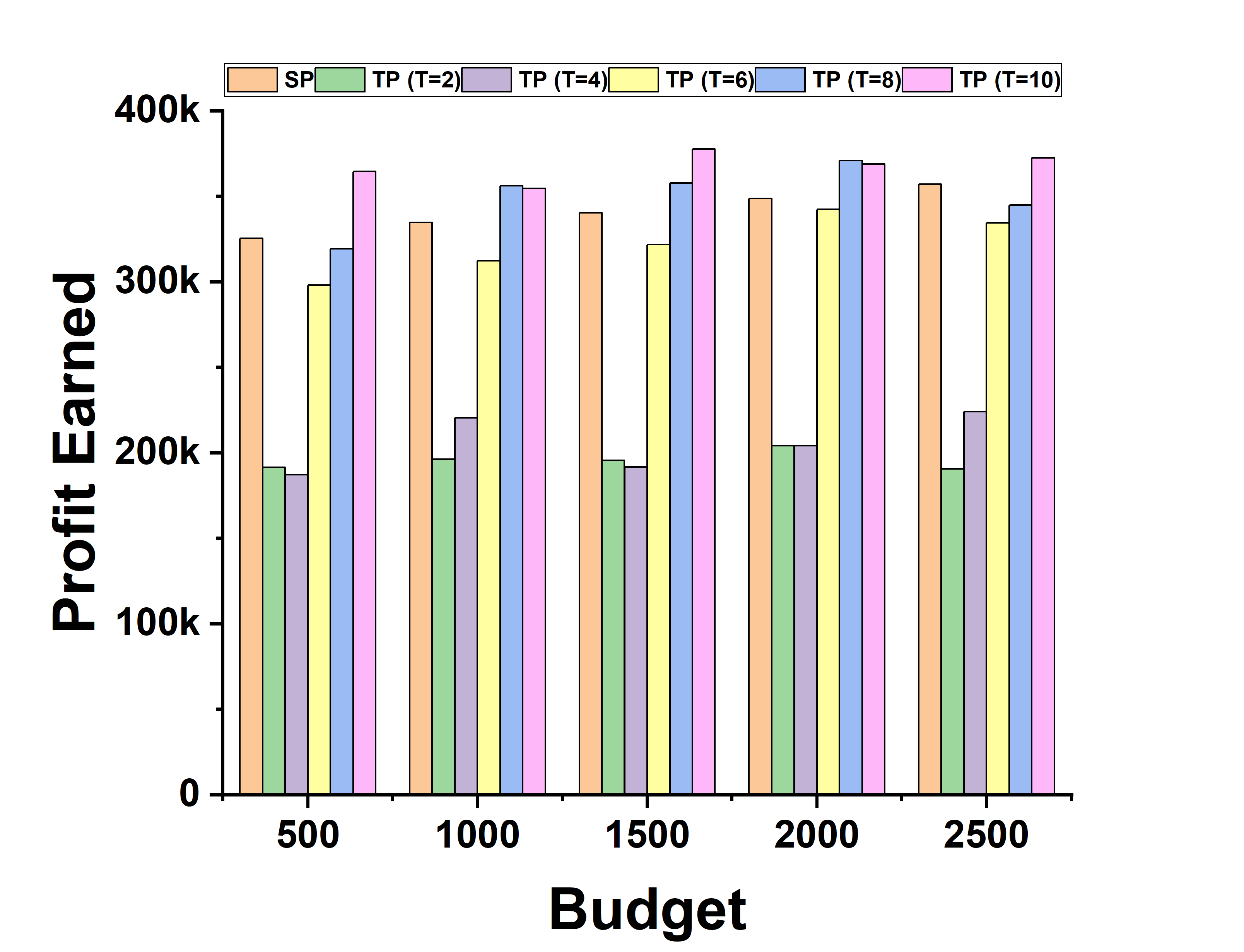}
        \caption{Simple Greedy}
    \end{subfigure} &
    \begin{subfigure}[t]{0.22\textwidth}
        \includegraphics[width=\linewidth]{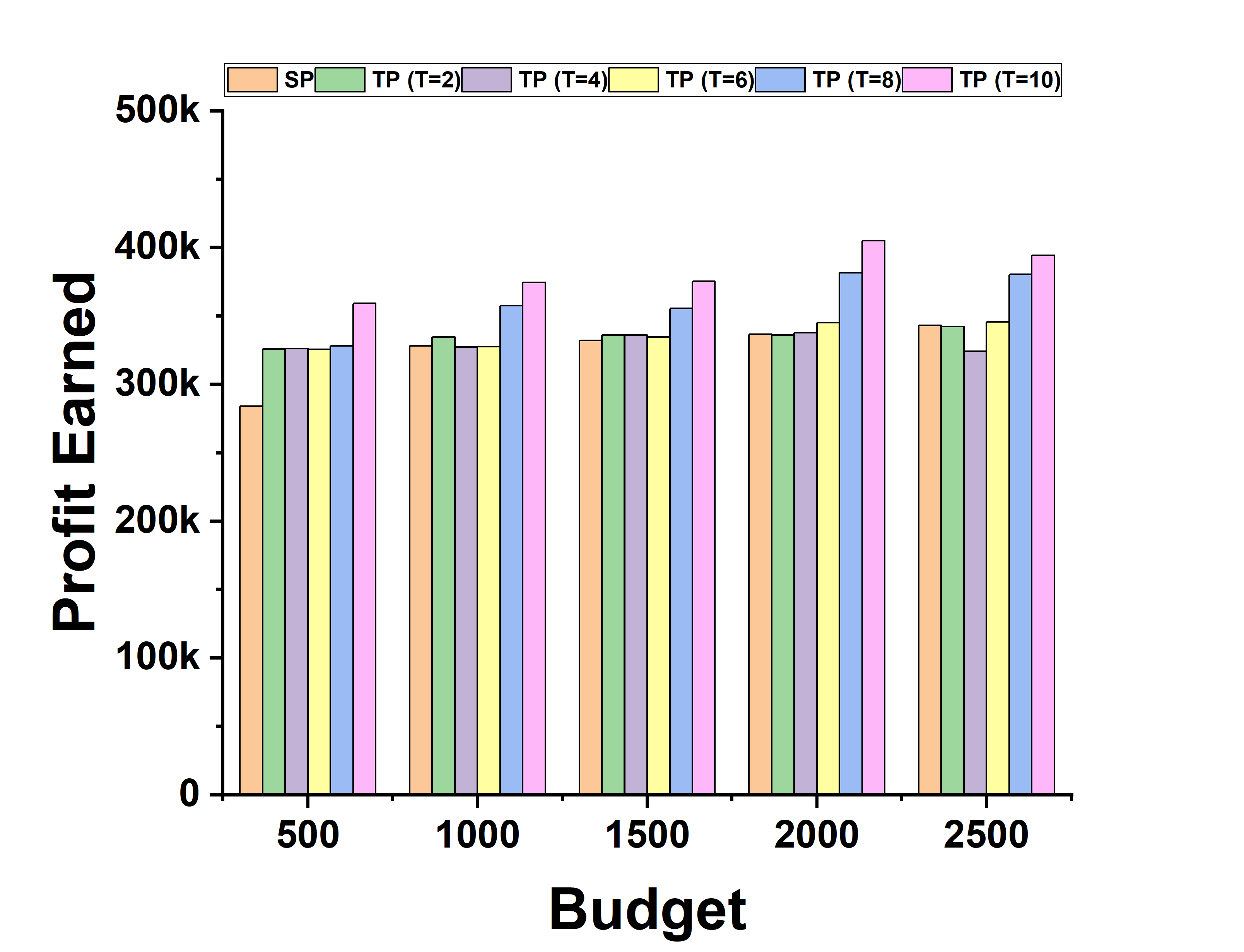}
        \caption{Double Greedy}
    \end{subfigure} &
    \begin{subfigure}[t]{0.22\textwidth}
        \includegraphics[width=\linewidth]{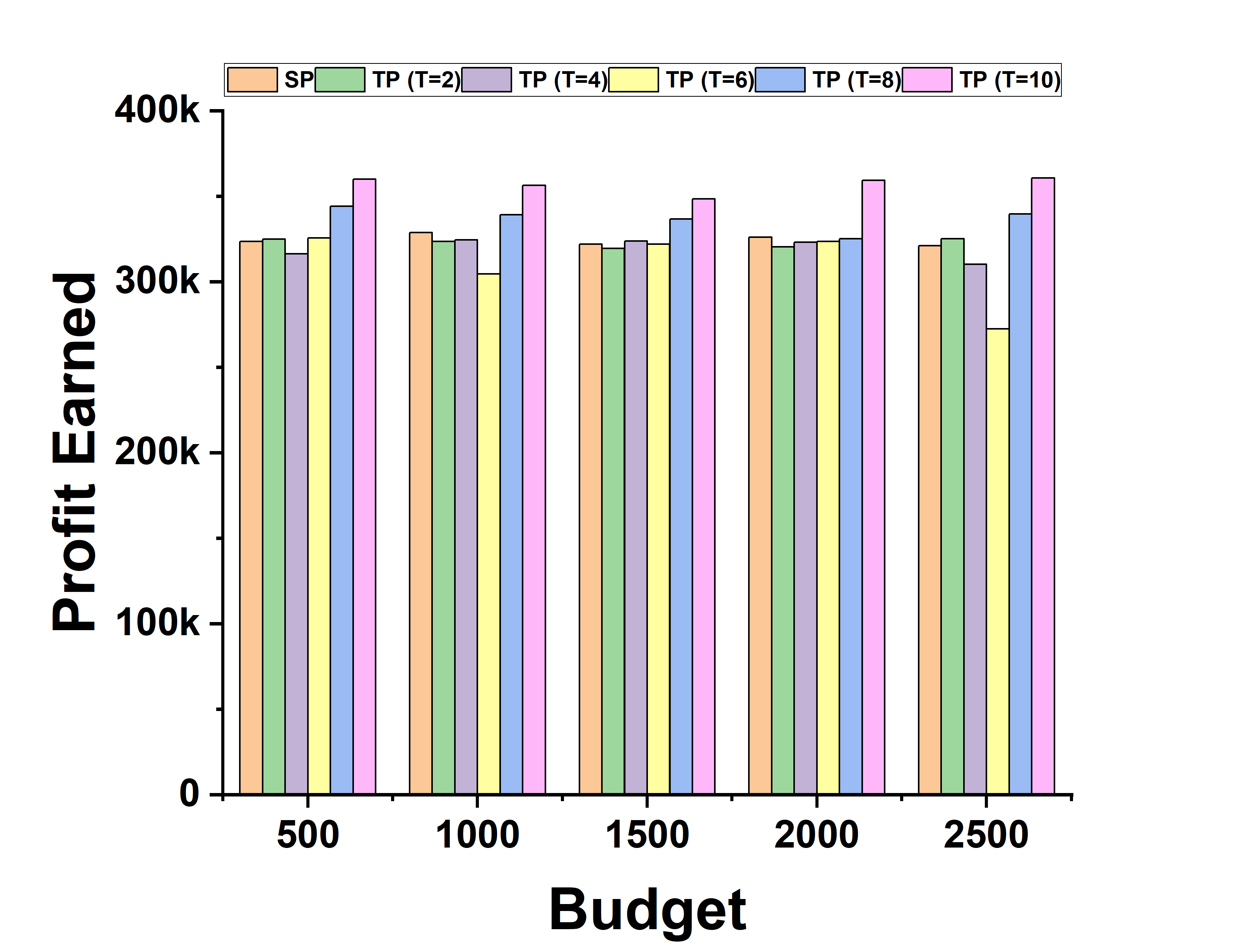}
        \caption{Stochastic Greedy}
    \end{subfigure}
\end{tabular}
\caption{Profit Earned in Single Phase Vs. Two Phase setting (split ratio 30\%, Probability Setting - Trivalency, \textit{Email-Eu-Core} Dataset)}
\label{Fig:RQ1_T2}
\end{figure}

\begin{figure}[htbp]
\centering
\captionsetup[sub]{font=footnotesize}  
\begin{tabular}{cccc}
    \begin{subfigure}[t]{0.22\textwidth}
        \includegraphics[width=\linewidth]{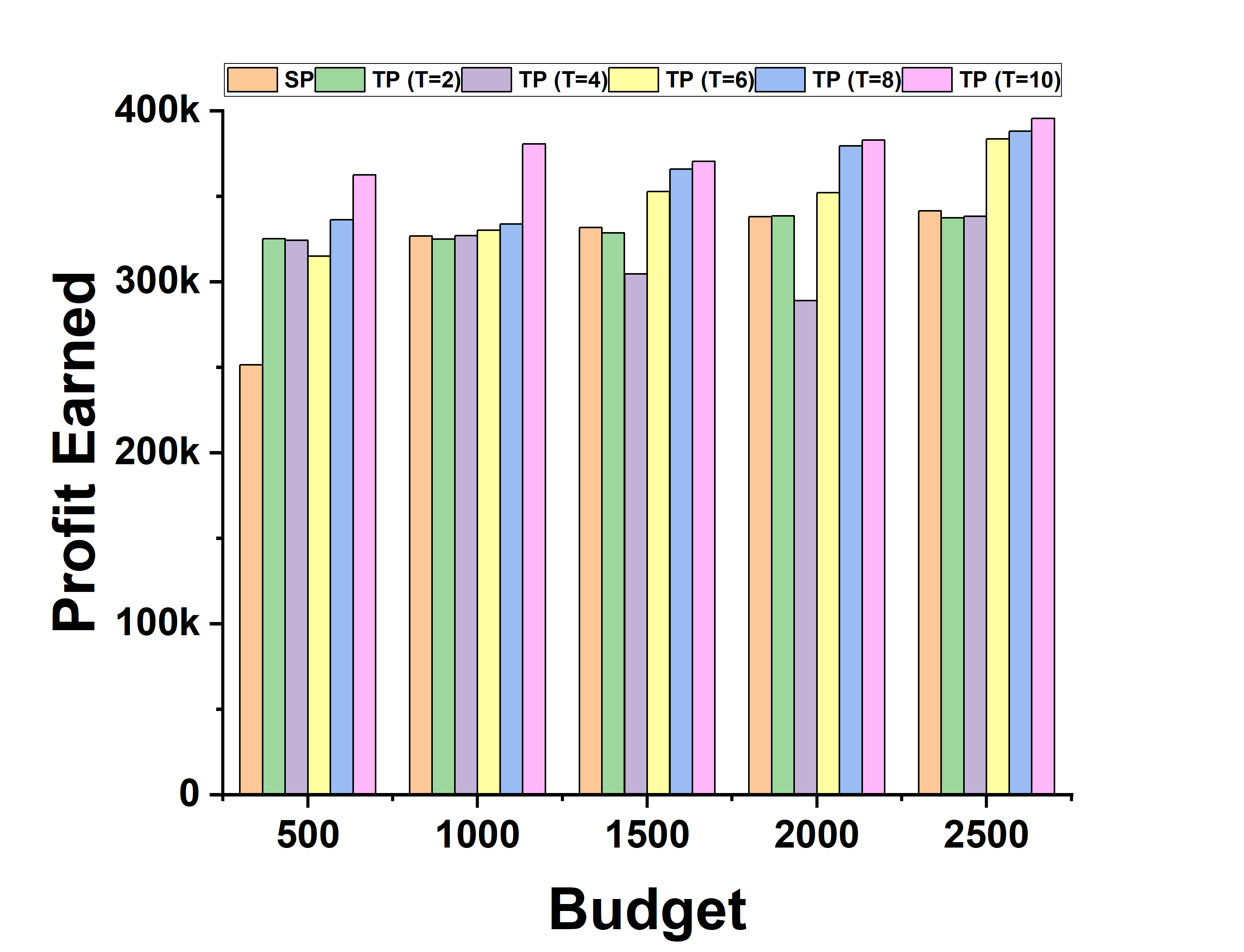}
        \caption{Random}
    \end{subfigure} &
    \begin{subfigure}[t]{0.22\textwidth}
        \includegraphics[width=\linewidth]{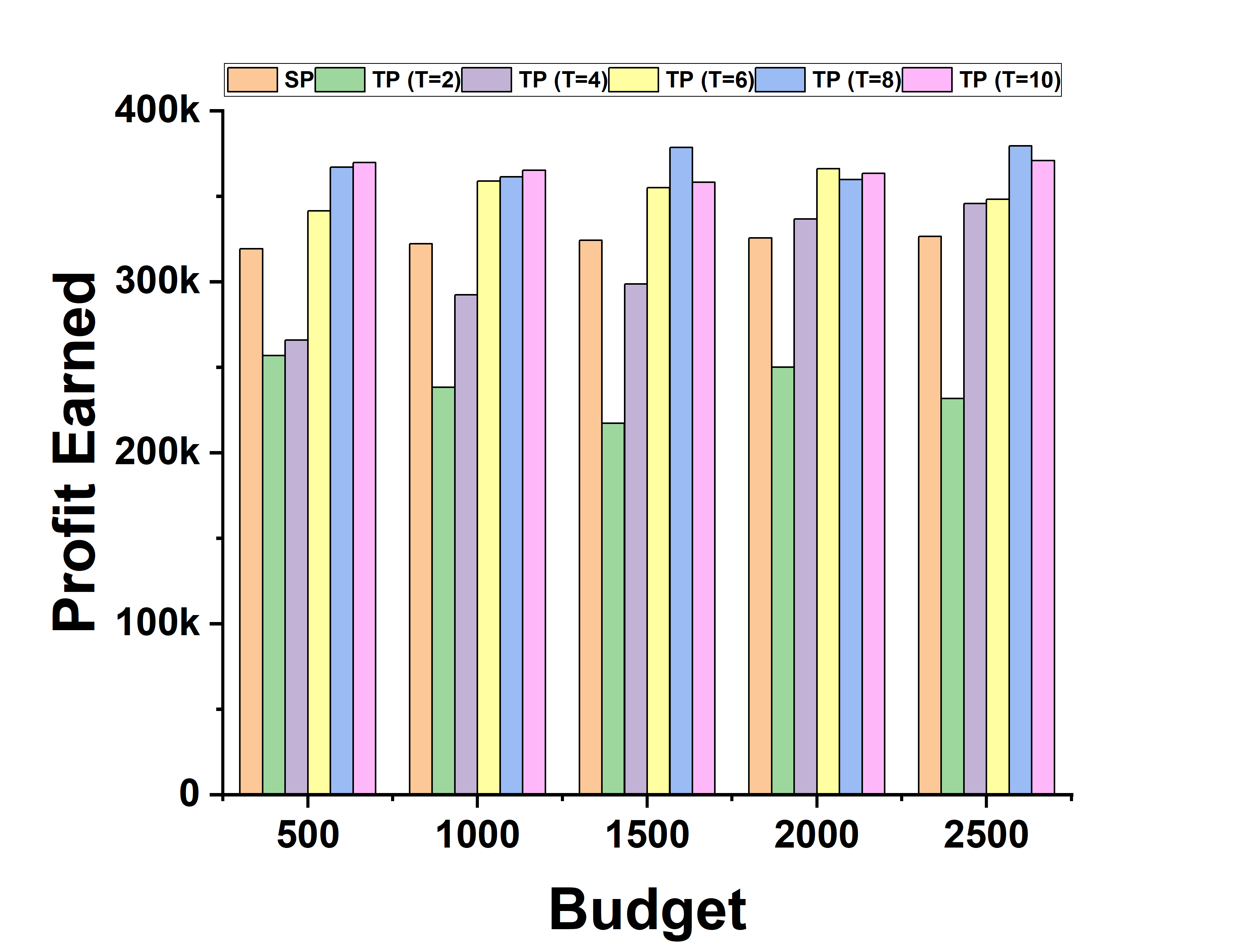}
        \caption{High Degree}
    \end{subfigure} &
    \begin{subfigure}[t]{0.22\textwidth}
        \includegraphics[width=\linewidth]{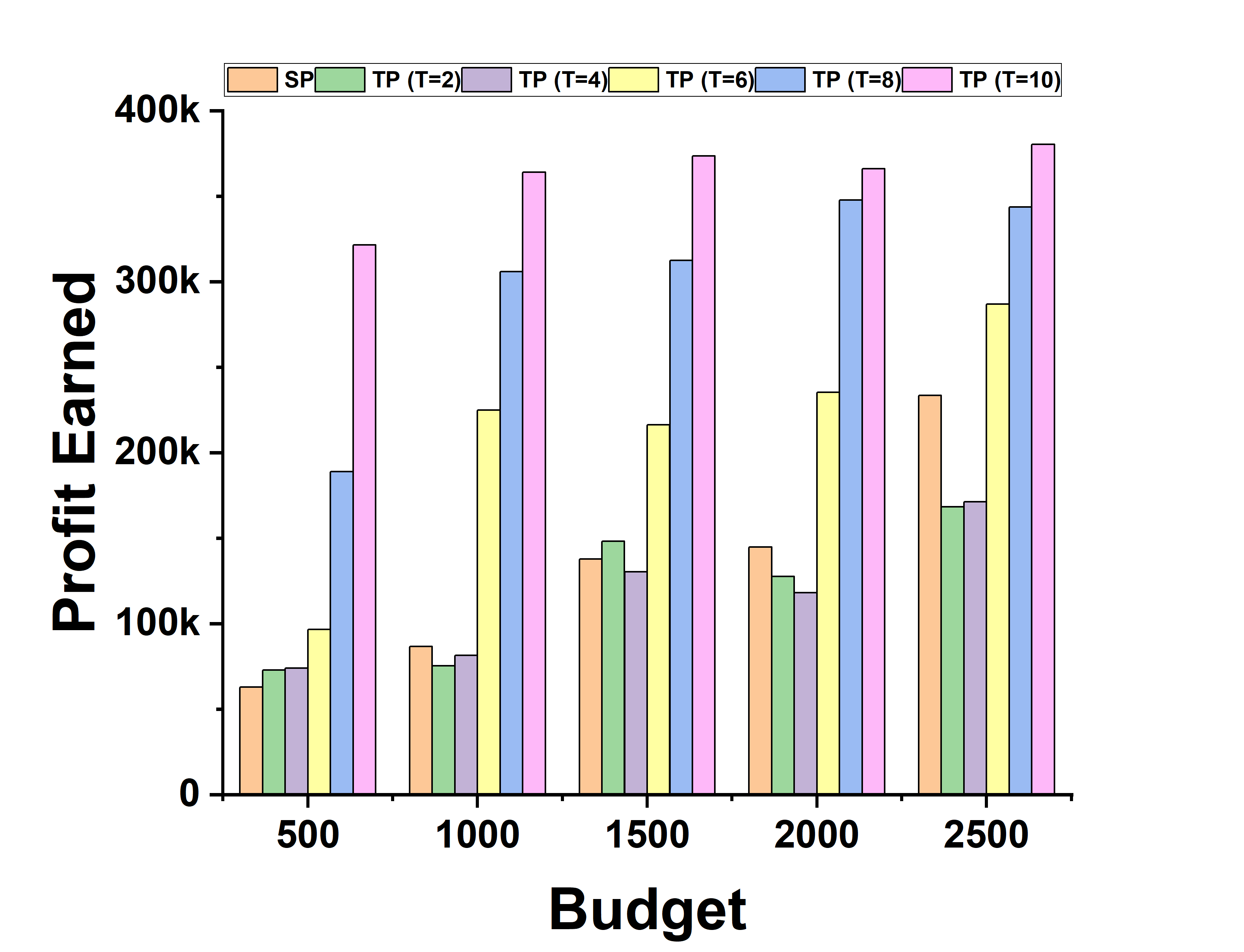}
        \caption{Clustering\\Coefficient}
    \end{subfigure} &
    \begin{subfigure}[t]{0.22\textwidth}
        \includegraphics[width=\linewidth]{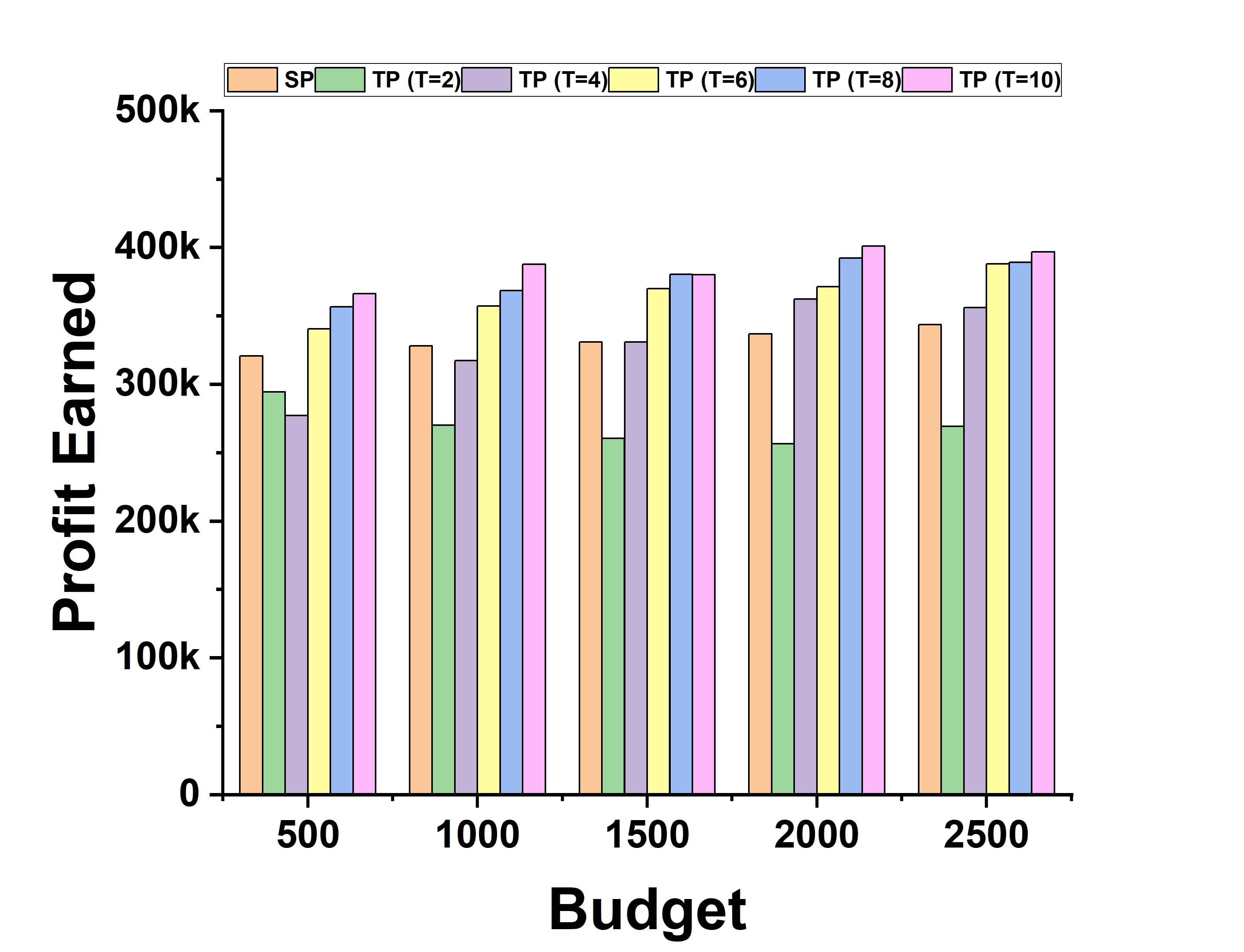}
        \caption{Degree Discount}
    \end{subfigure} \\[6pt]

    \begin{subfigure}[t]{0.22\textwidth}
        \includegraphics[width=\linewidth]{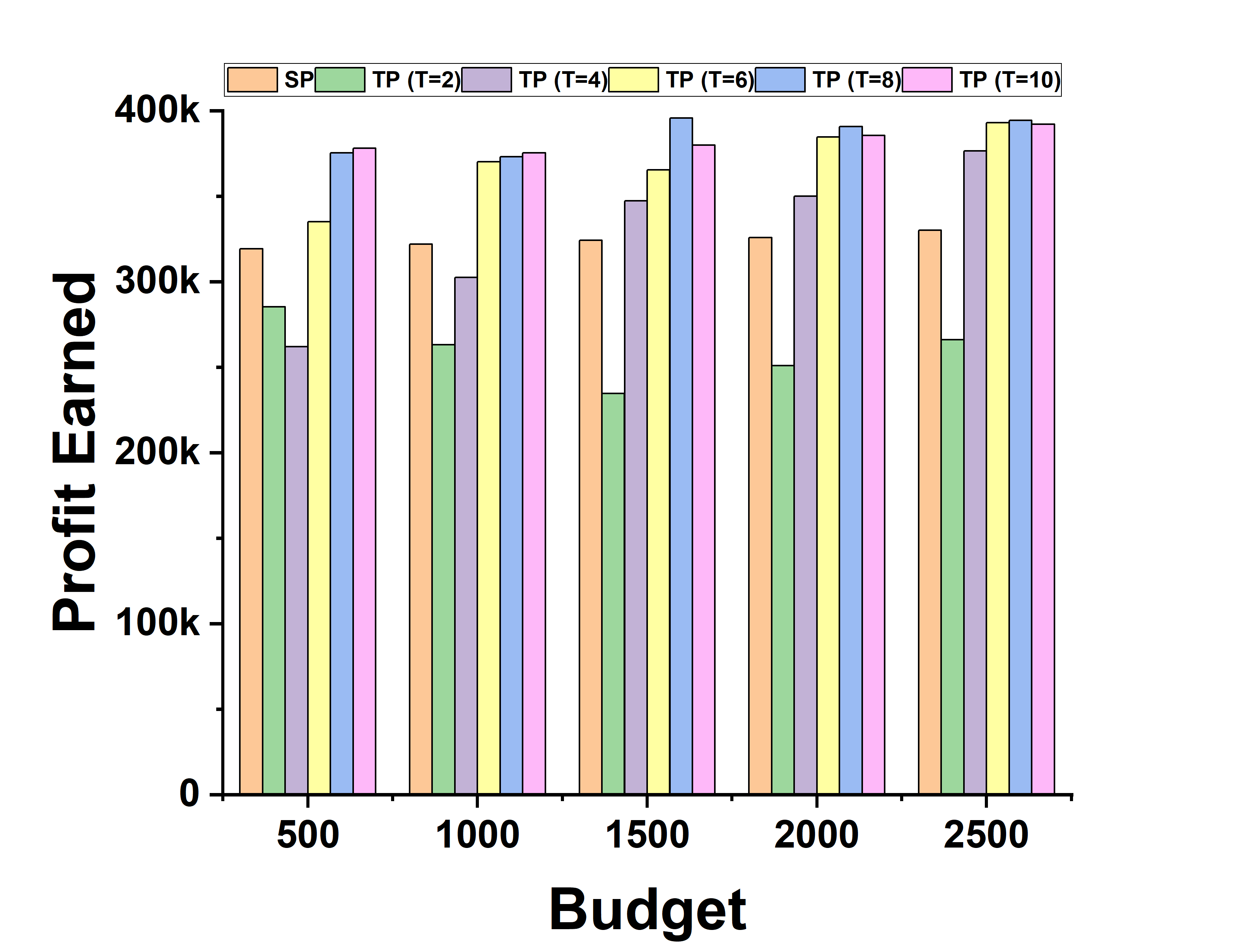}
        \caption{Single Discount}
    \end{subfigure} &
    \begin{subfigure}[t]{0.22\textwidth}
        \includegraphics[width=\linewidth]{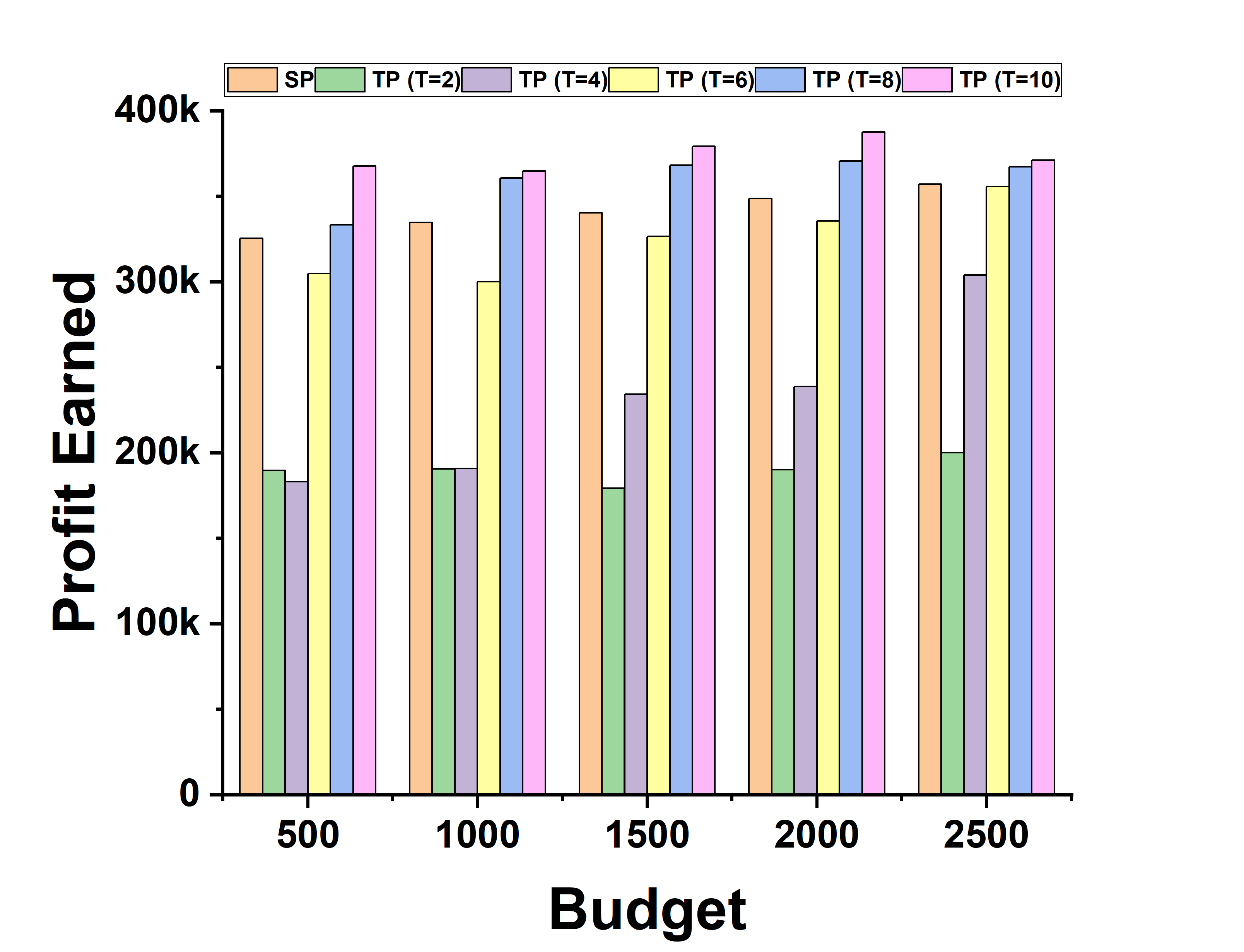}
        \caption{Simple Greedy}
    \end{subfigure} &
    \begin{subfigure}[t]{0.22\textwidth}
        \includegraphics[width=\linewidth]{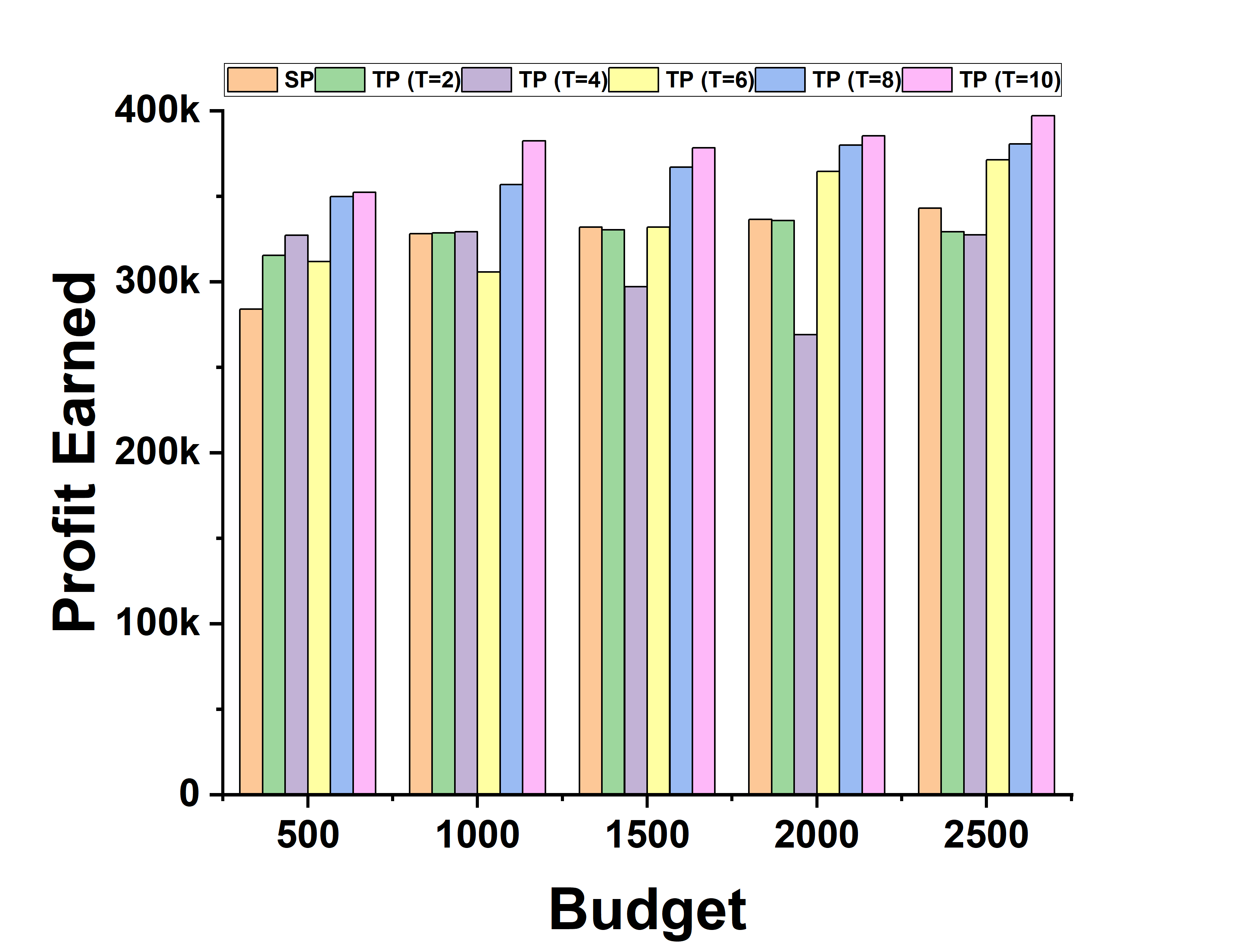}
        \caption{Double Greedy}
    \end{subfigure} &
    \begin{subfigure}[t]{0.22\textwidth}
        \includegraphics[width=\linewidth]{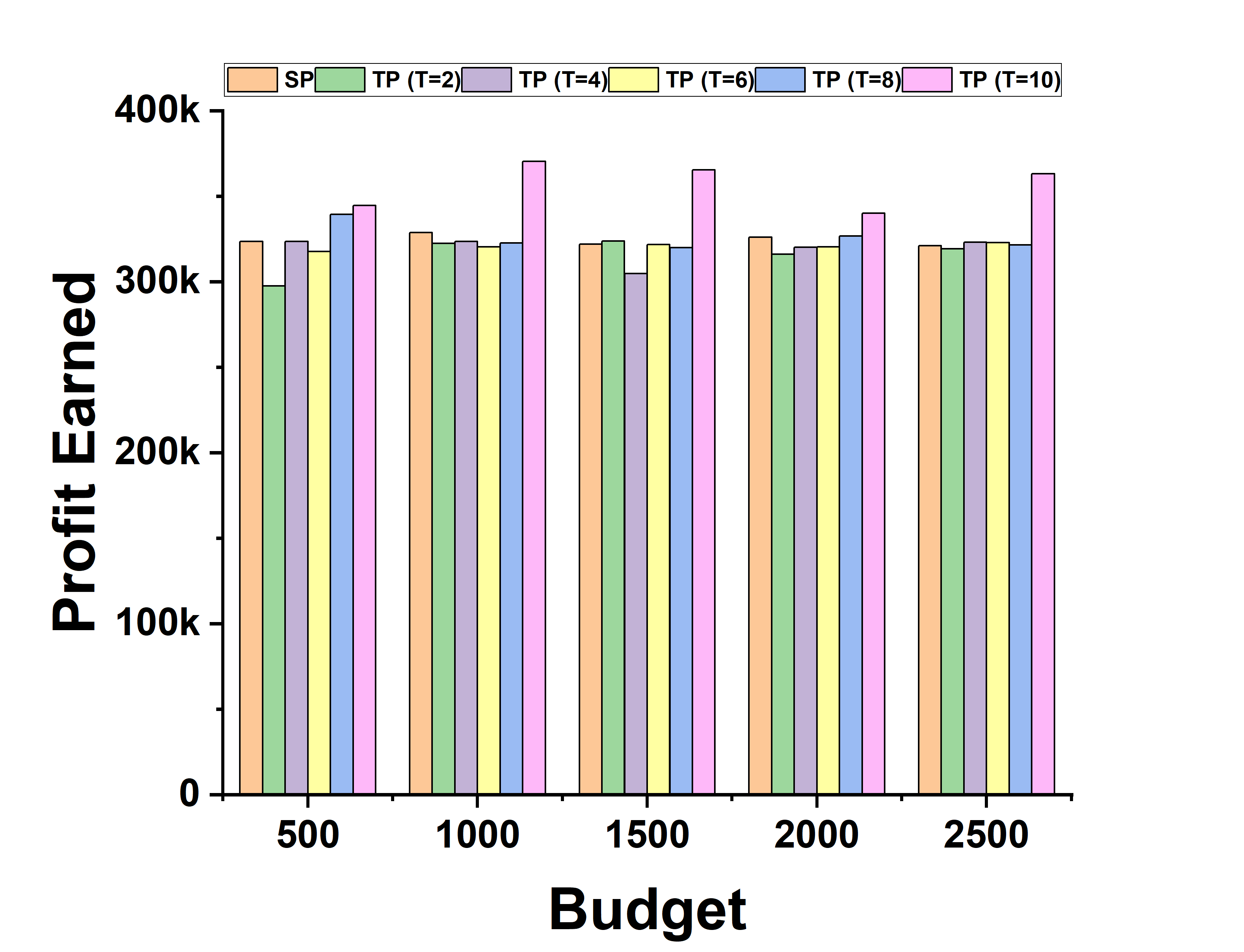}
        \caption{Stochastic Greedy}
    \end{subfigure}
\end{tabular}
\caption{Profit Earned in Single Phase Vs. Two Phase setting (split ratio 50\%, Probability Setting - Trivalency, \textit{Email-Eu-Core} Dataset)}
\label{Fig:RQ1_T3}
\end{figure}

\begin{figure}[htbp]
\centering
\captionsetup[sub]{font=footnotesize}  
\begin{tabular}{cccc}
    \begin{subfigure}[t]{0.22\textwidth}
        \includegraphics[width=\linewidth]{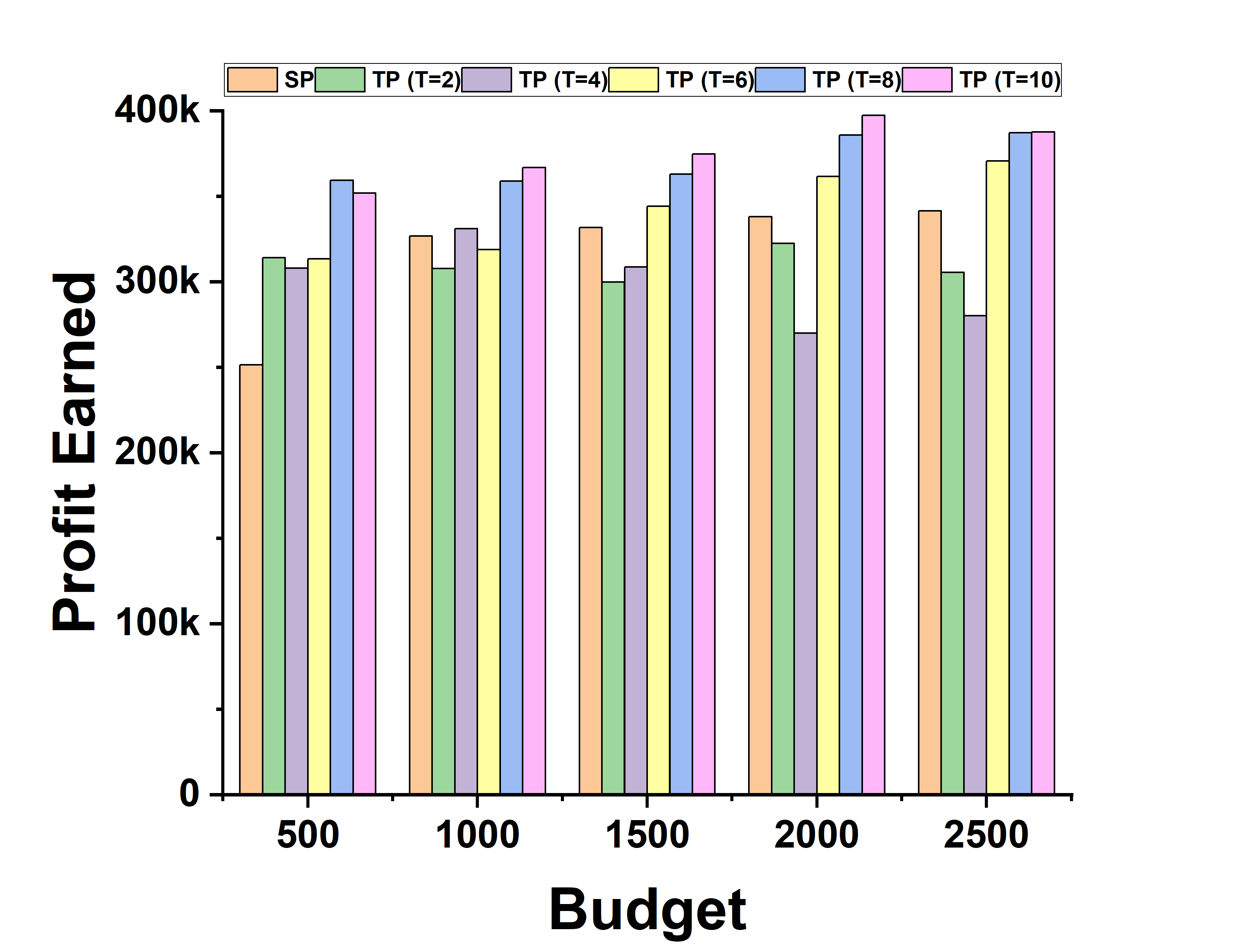}
        \caption{Random}
    \end{subfigure} &
    \begin{subfigure}[t]{0.22\textwidth}
        \includegraphics[width=\linewidth]{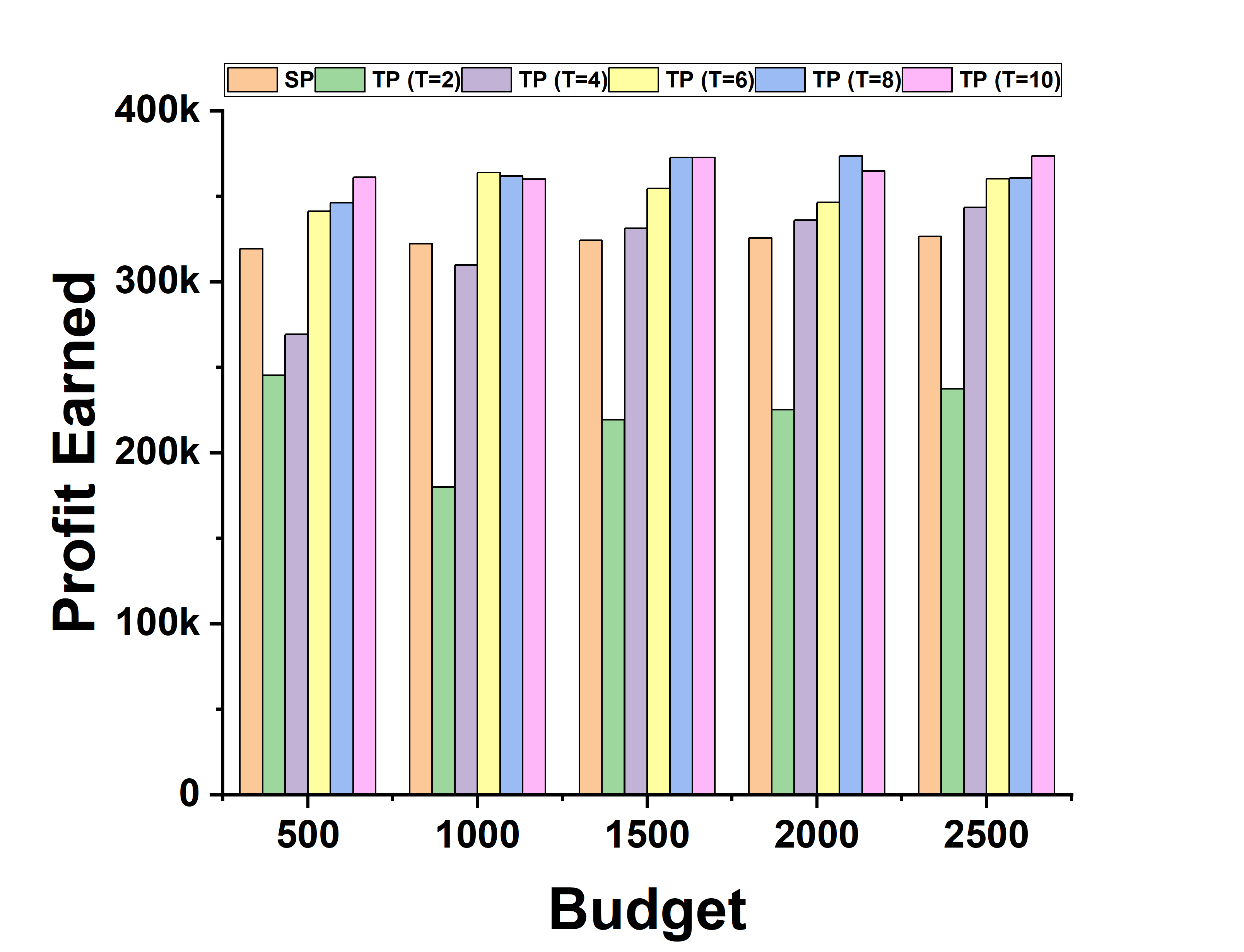}
        \caption{High Degree}
    \end{subfigure} &
    \begin{subfigure}[t]{0.22\textwidth}
        \includegraphics[width=\linewidth]{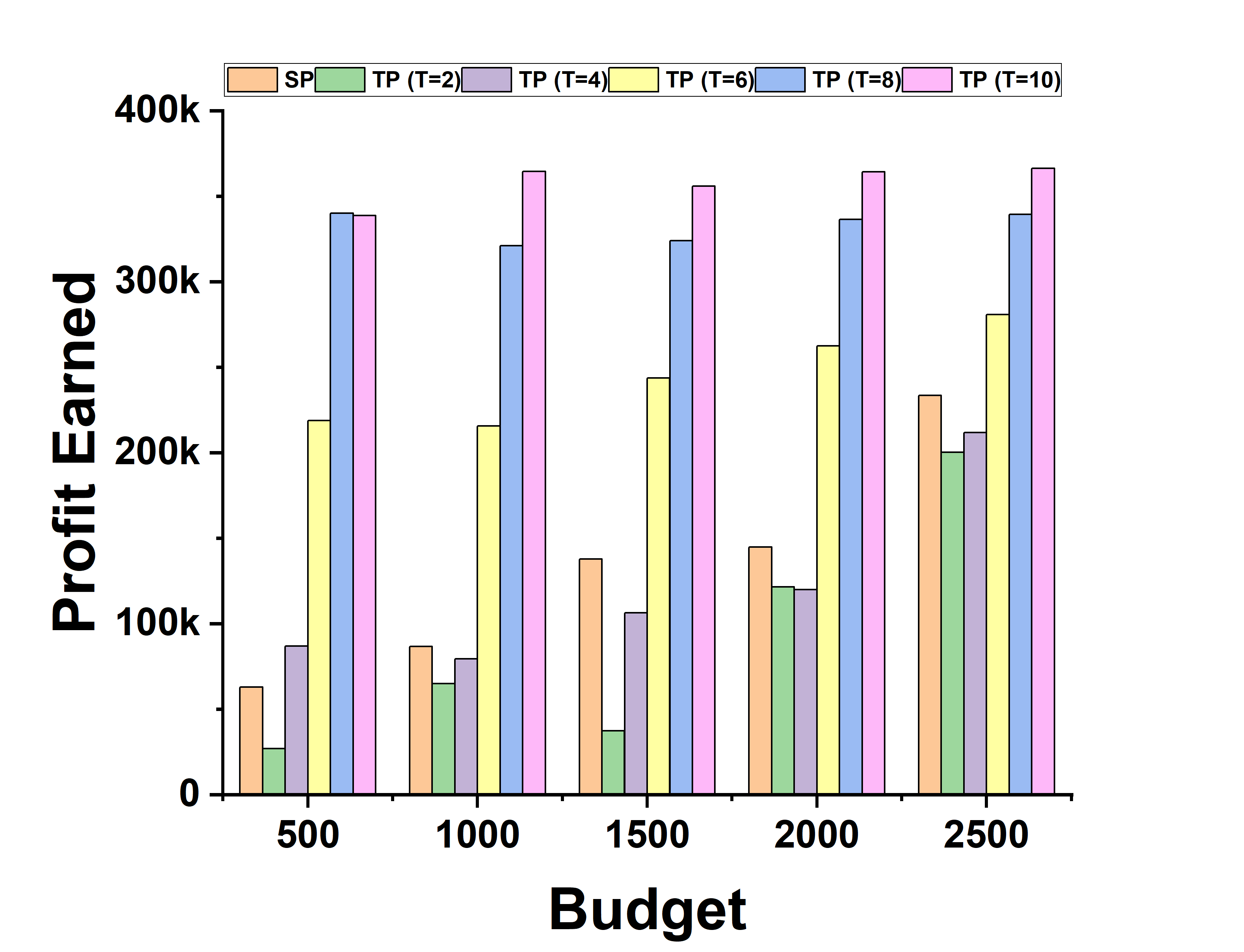}
        \caption{Clustering\\Coefficient}
    \end{subfigure} &
    \begin{subfigure}[t]{0.22\textwidth}
        \includegraphics[width=\linewidth]{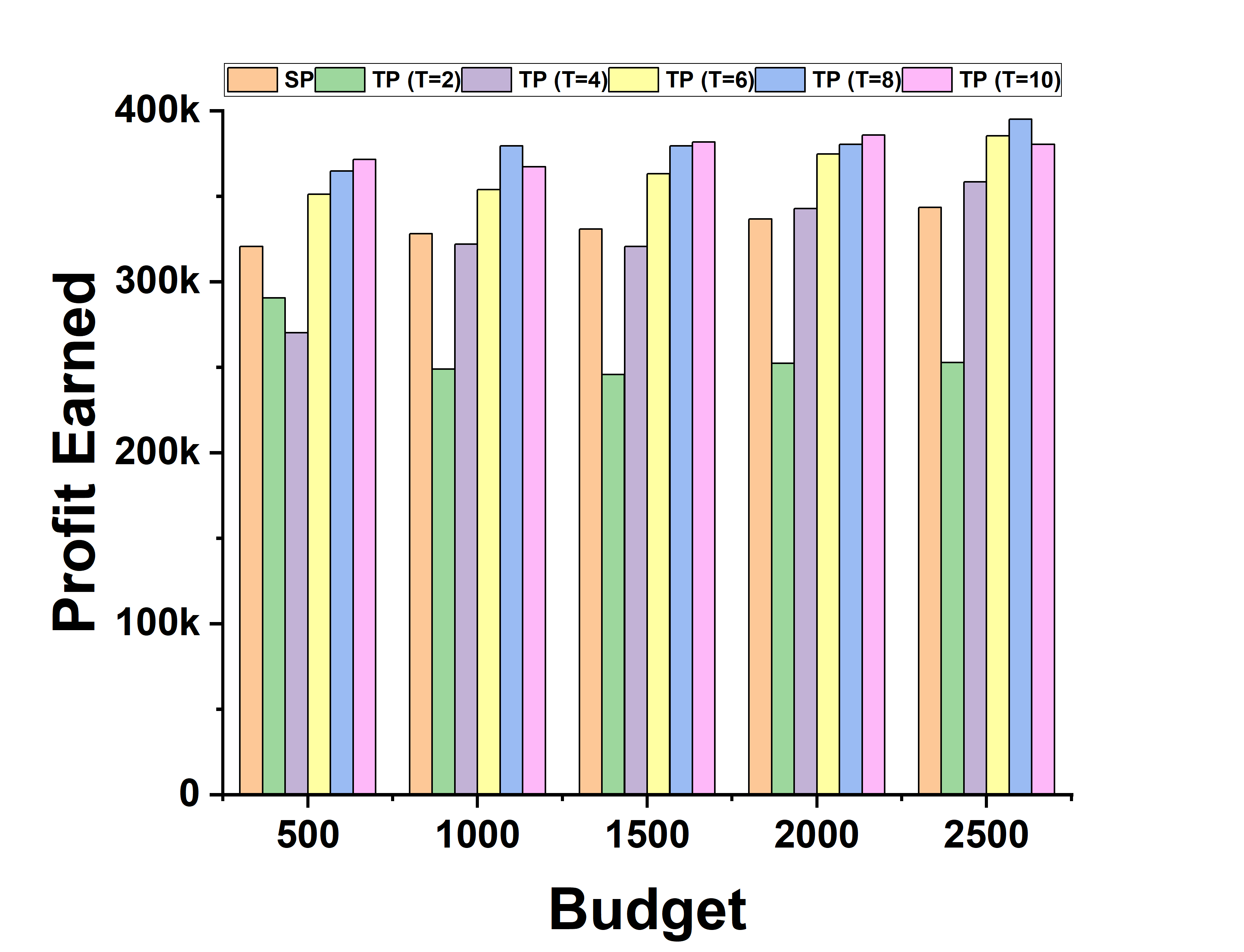}
        \caption{Degree Discount}
    \end{subfigure} \\[6pt]

    \begin{subfigure}[t]{0.22\textwidth}
        \includegraphics[width=\linewidth]{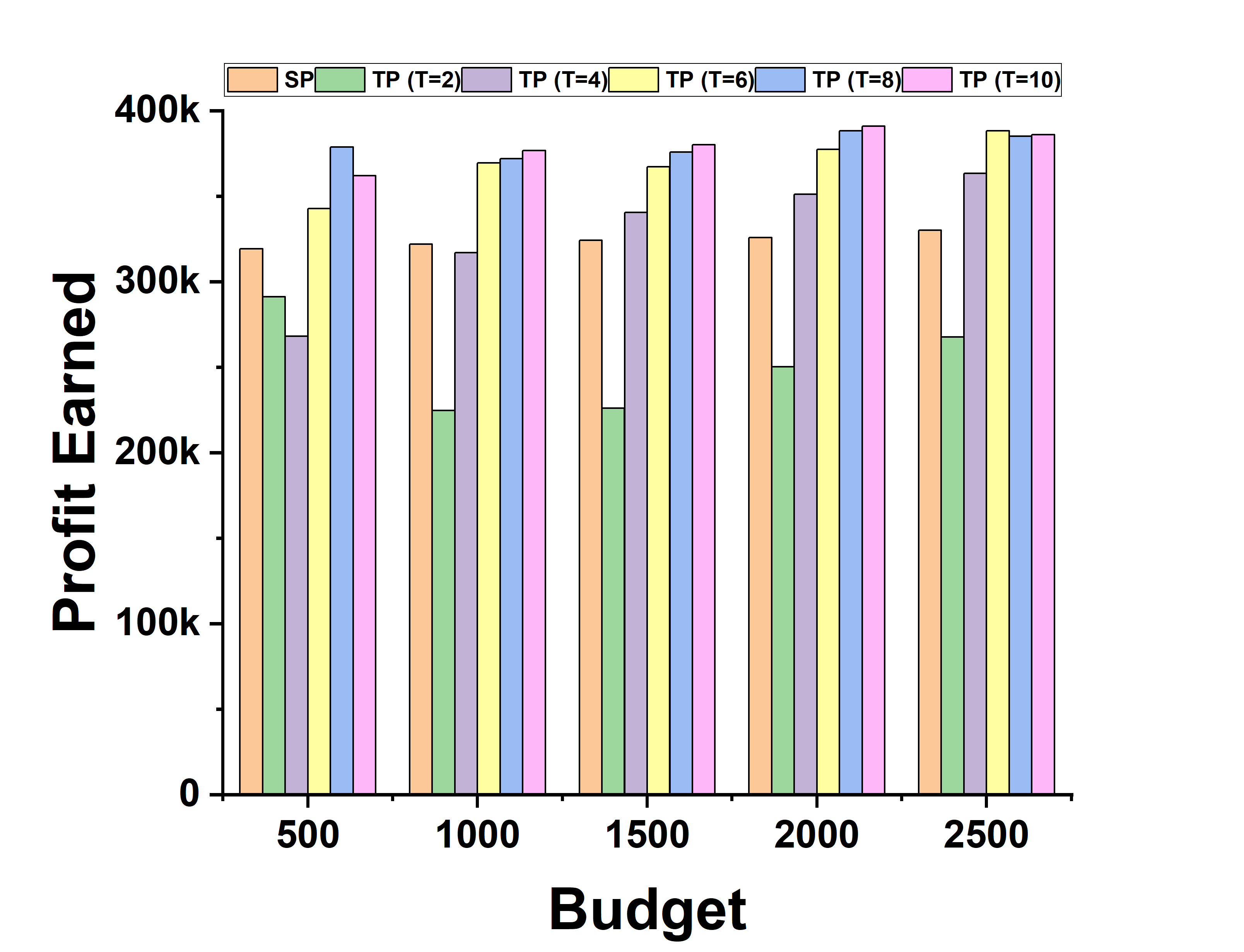}
        \caption{Single Discount}
    \end{subfigure} &
    \begin{subfigure}[t]{0.22\textwidth}
        \includegraphics[width=\linewidth]{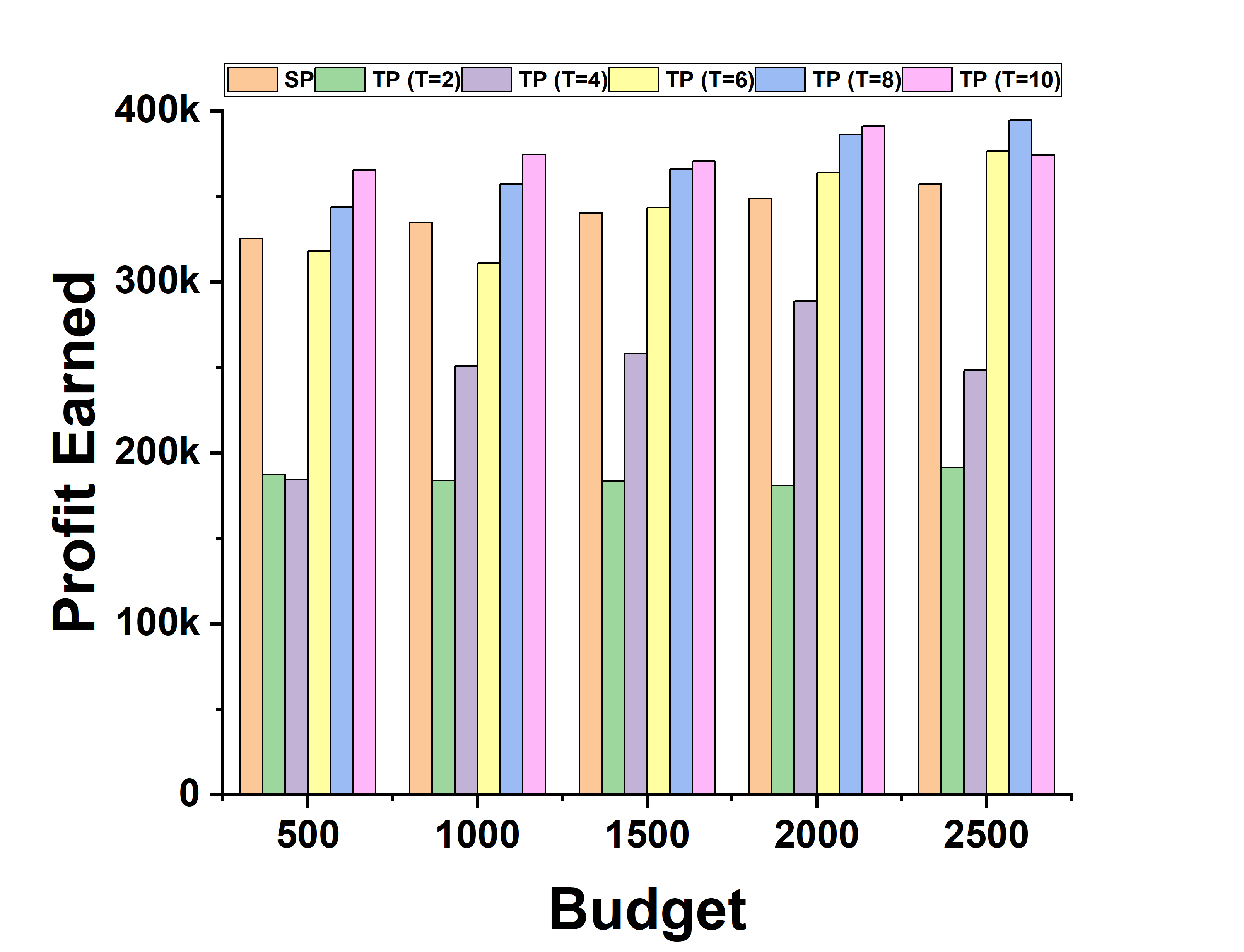}
        \caption{Simple Greedy}
    \end{subfigure} &
    \begin{subfigure}[t]{0.22\textwidth}
        \includegraphics[width=\linewidth]{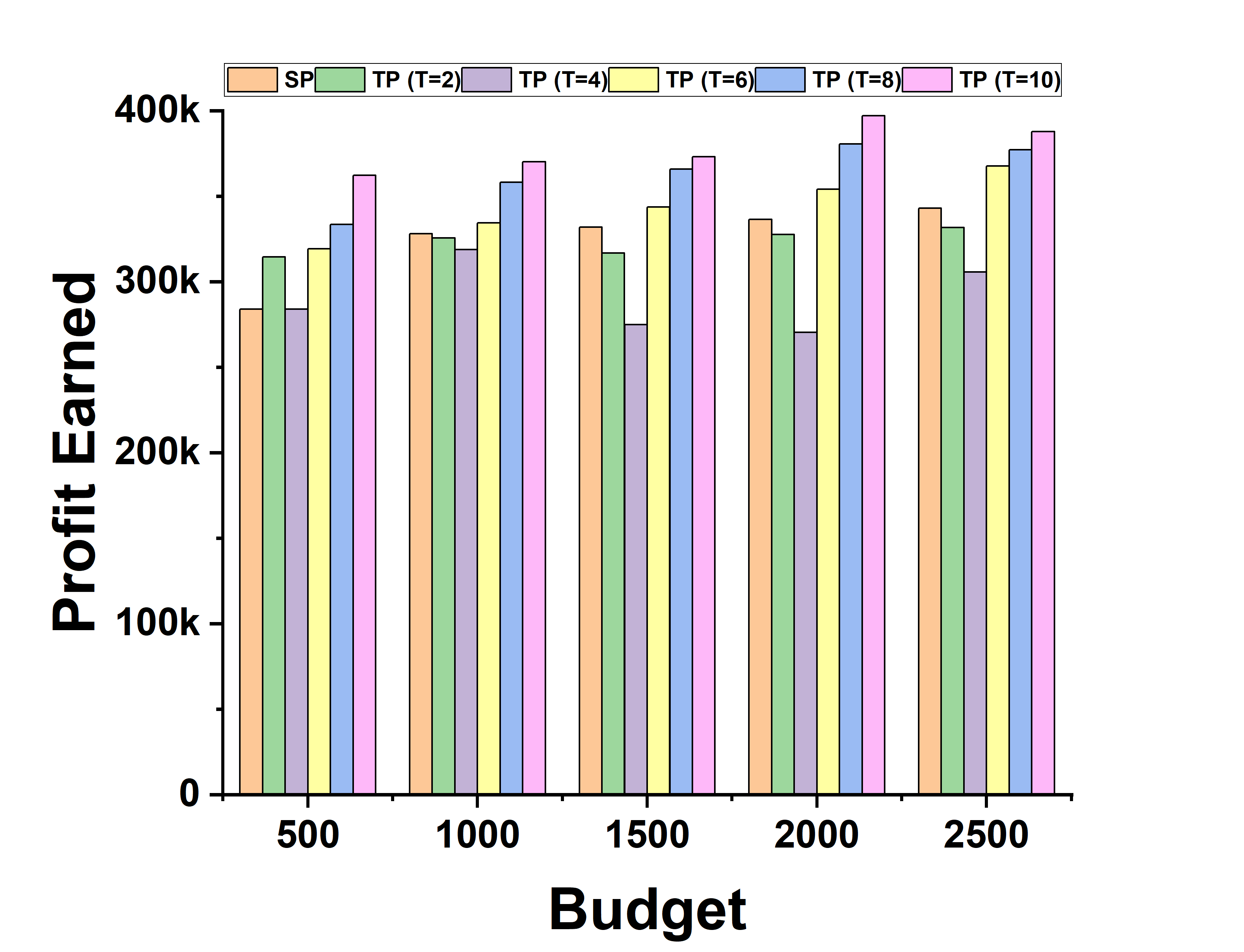}
        \caption{Double Greedy}
    \end{subfigure} &
    \begin{subfigure}[t]{0.22\textwidth}
        \includegraphics[width=\linewidth]{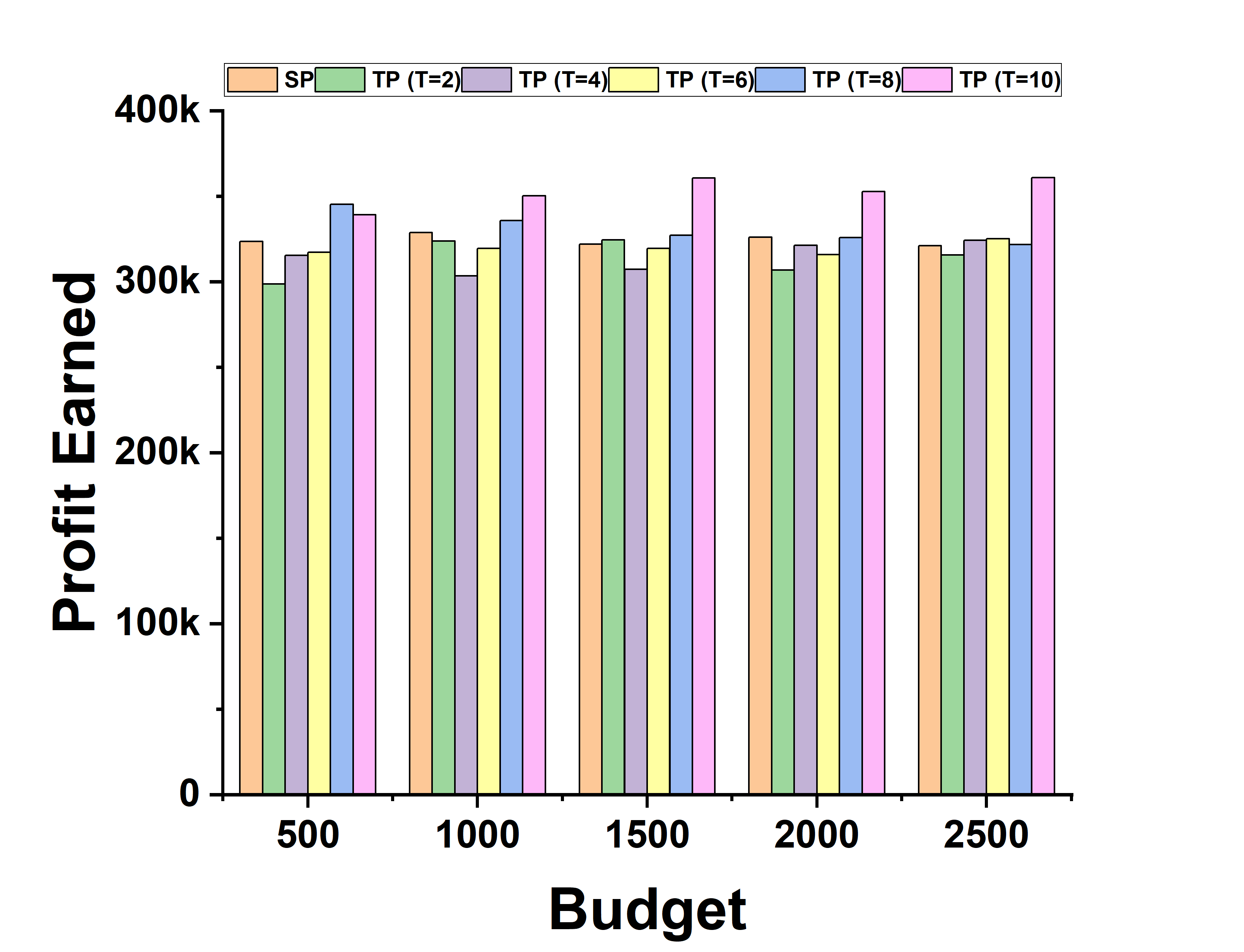}
        \caption{Stochastic Greedy}
    \end{subfigure}
\end{tabular}
\caption{Profit Earned in Single Phase Vs. Two Phase setting (split ratio 70\%, Probability Setting - Trivalency, \textit{Email-Eu-Core} Dataset)}
\label{Fig:RQ1_T4}
\end{figure}

\begin{figure}[htbp]
\centering
\captionsetup[sub]{font=footnotesize}  
\begin{tabular}{cccc}
    \begin{subfigure}[t]{0.22\textwidth}
        \includegraphics[width=\linewidth]{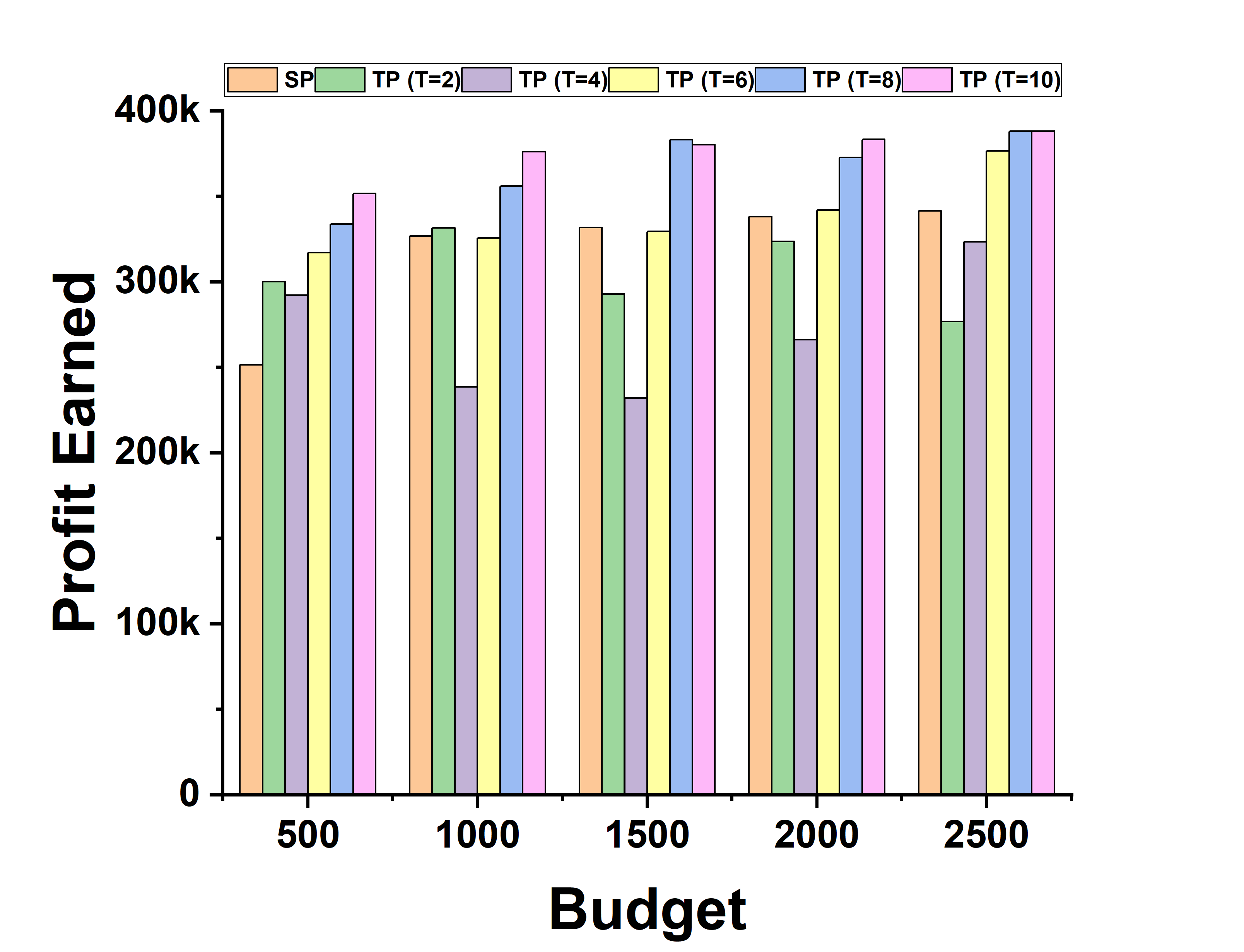}
        \caption{Random}
    \end{subfigure} &
    \begin{subfigure}[t]{0.22\textwidth}
        \includegraphics[width=\linewidth]{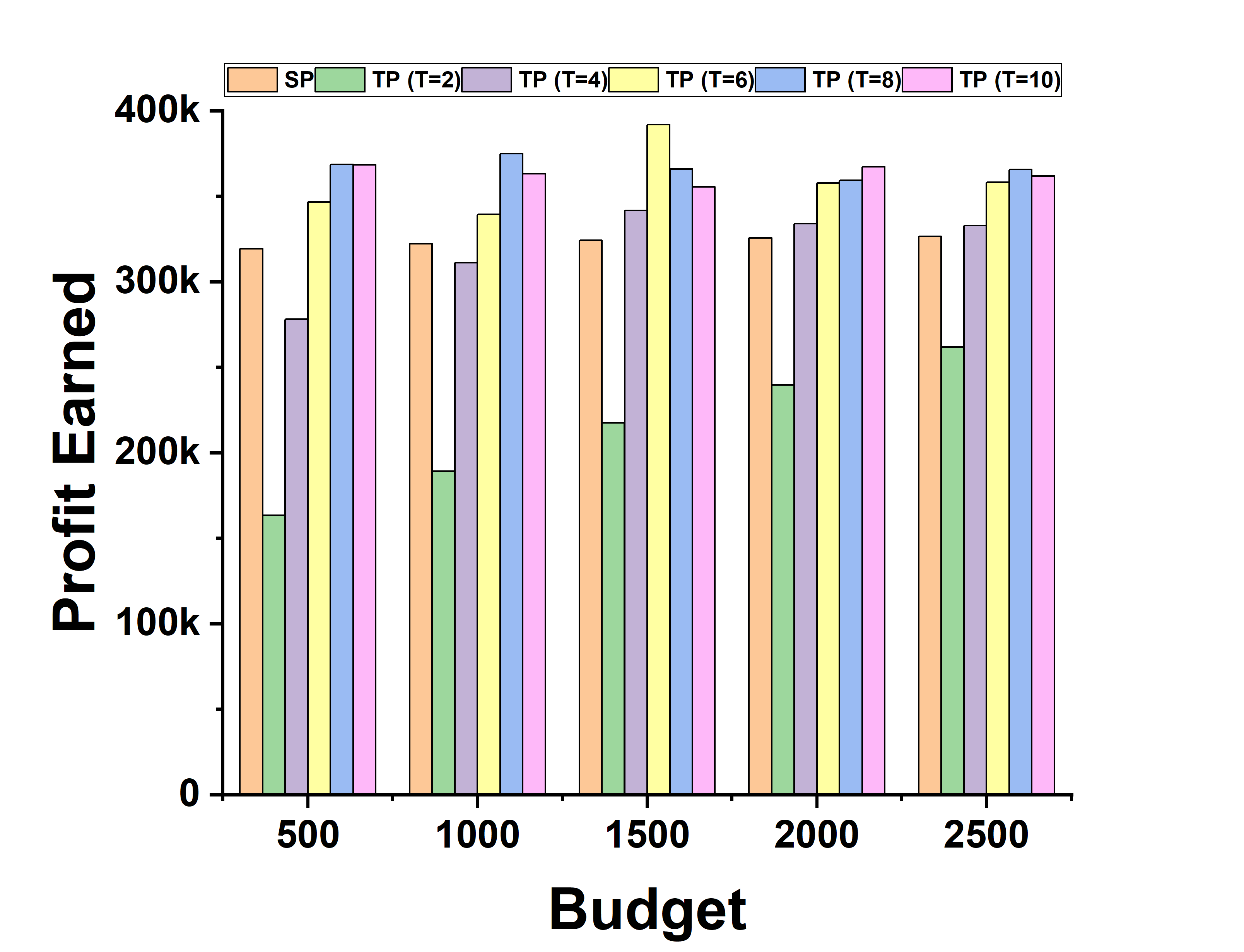}
        \caption{High Degree}
    \end{subfigure} &
    \begin{subfigure}[t]{0.22\textwidth}
        \includegraphics[width=\linewidth]{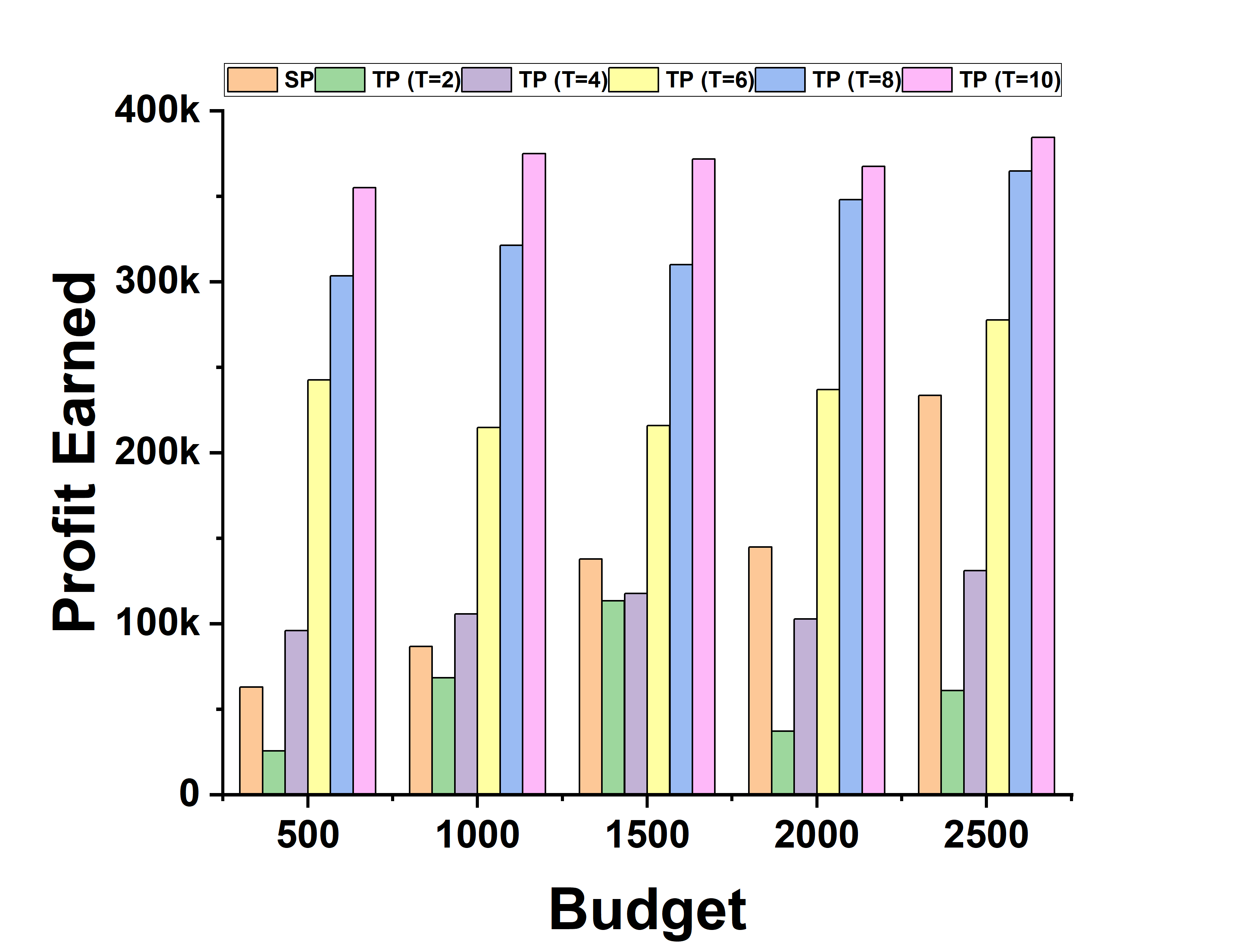}
        \caption{Clustering\\Coefficient}
    \end{subfigure} &
    \begin{subfigure}[t]{0.22\textwidth}
        \includegraphics[width=\linewidth]{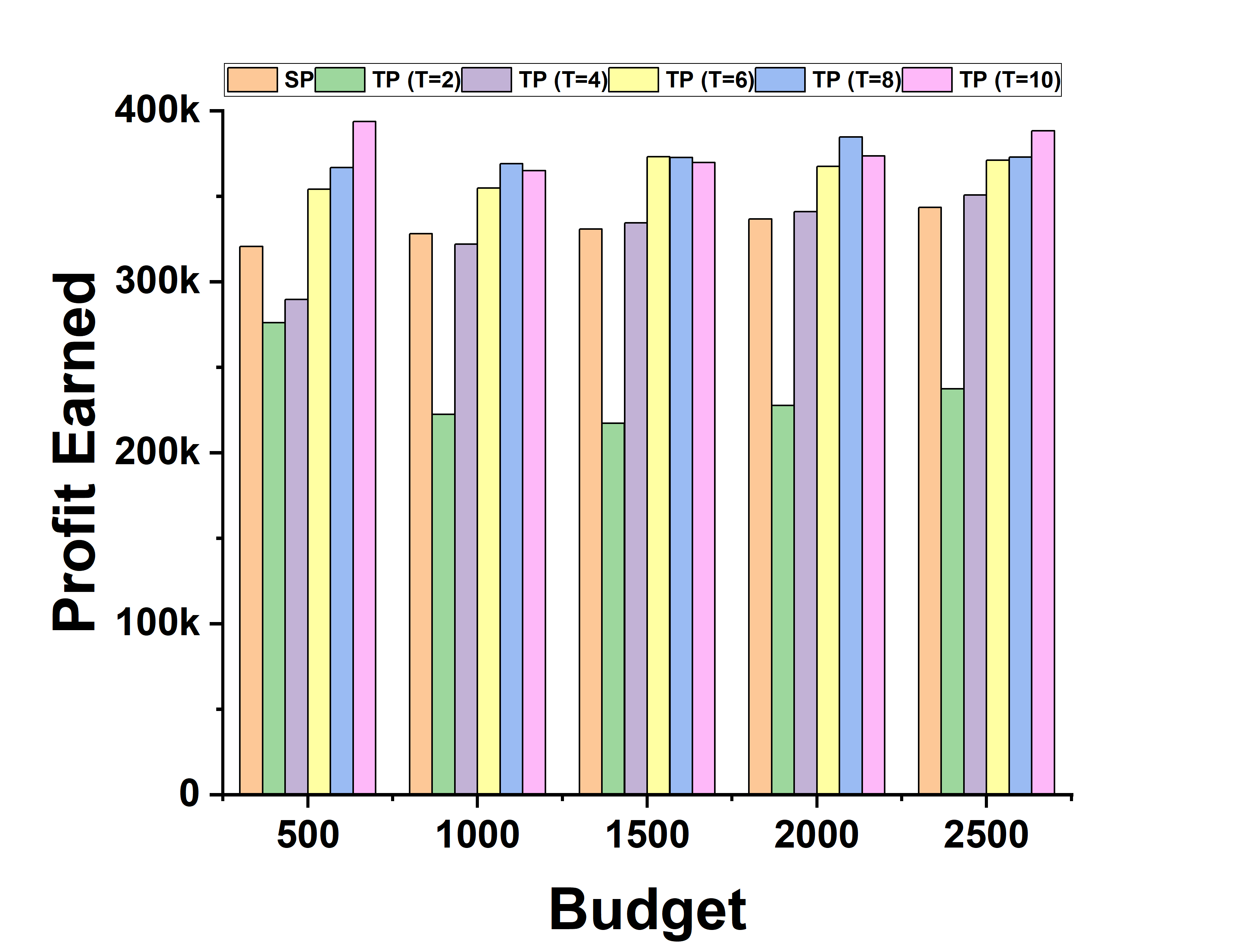}
        \caption{Degree Discount}
    \end{subfigure} \\[6pt]

    \begin{subfigure}[t]{0.22\textwidth}
        \includegraphics[width=\linewidth]{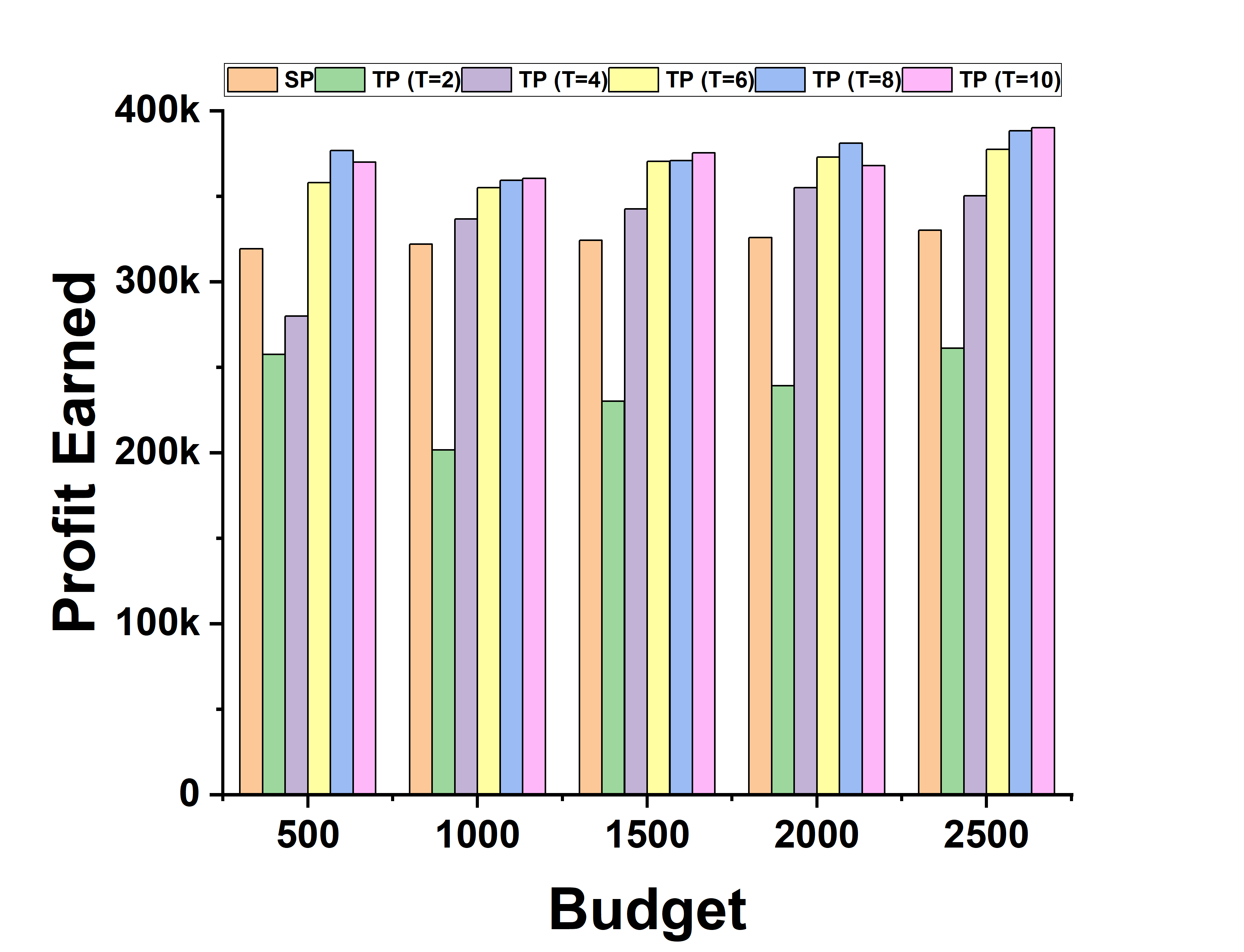}
        \caption{Single Discount}
    \end{subfigure} &
    \begin{subfigure}[t]{0.22\textwidth}
        \includegraphics[width=\linewidth]{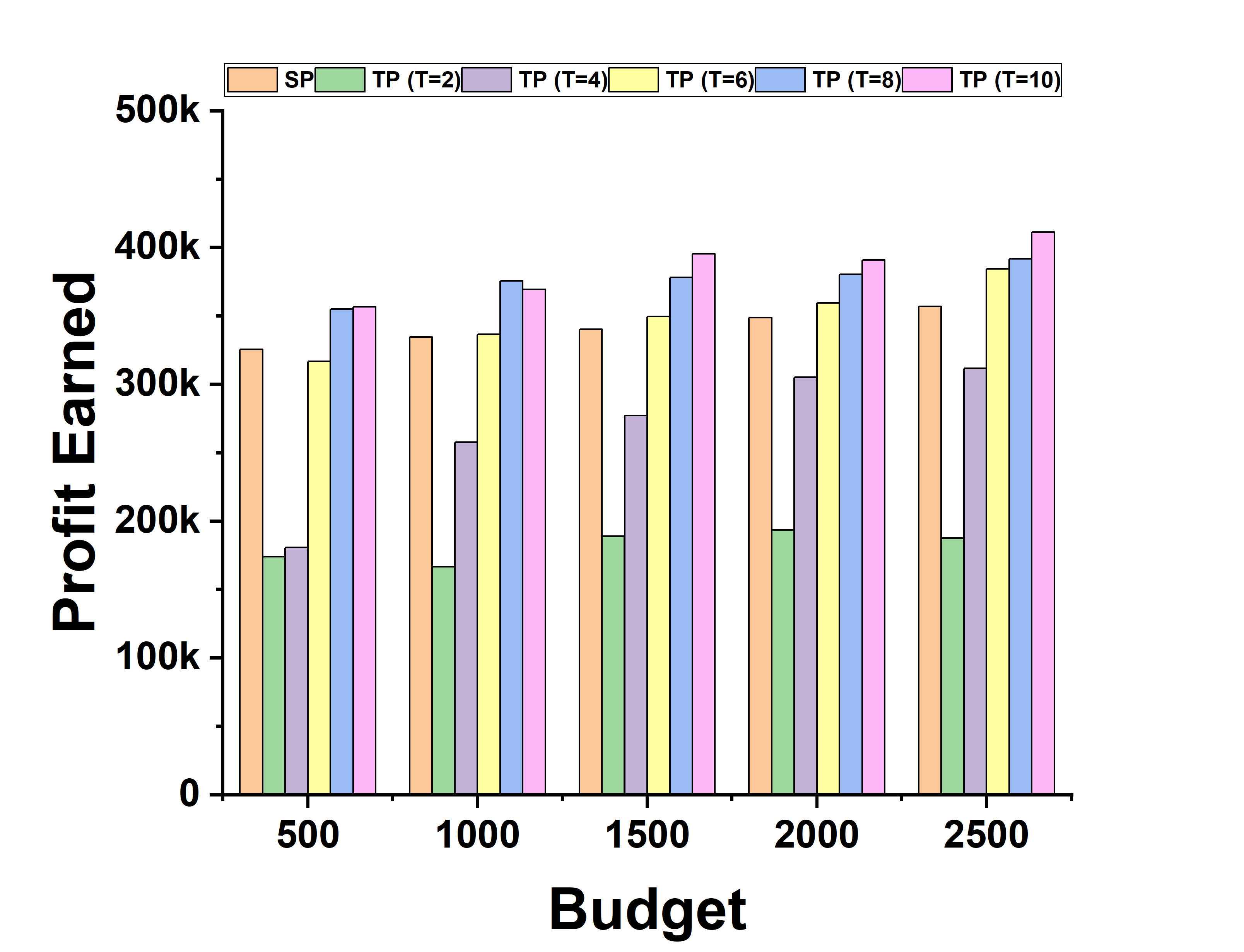}
        \caption{Simple Greedy}
    \end{subfigure} &
    \begin{subfigure}[t]{0.22\textwidth}
        \includegraphics[width=\linewidth]{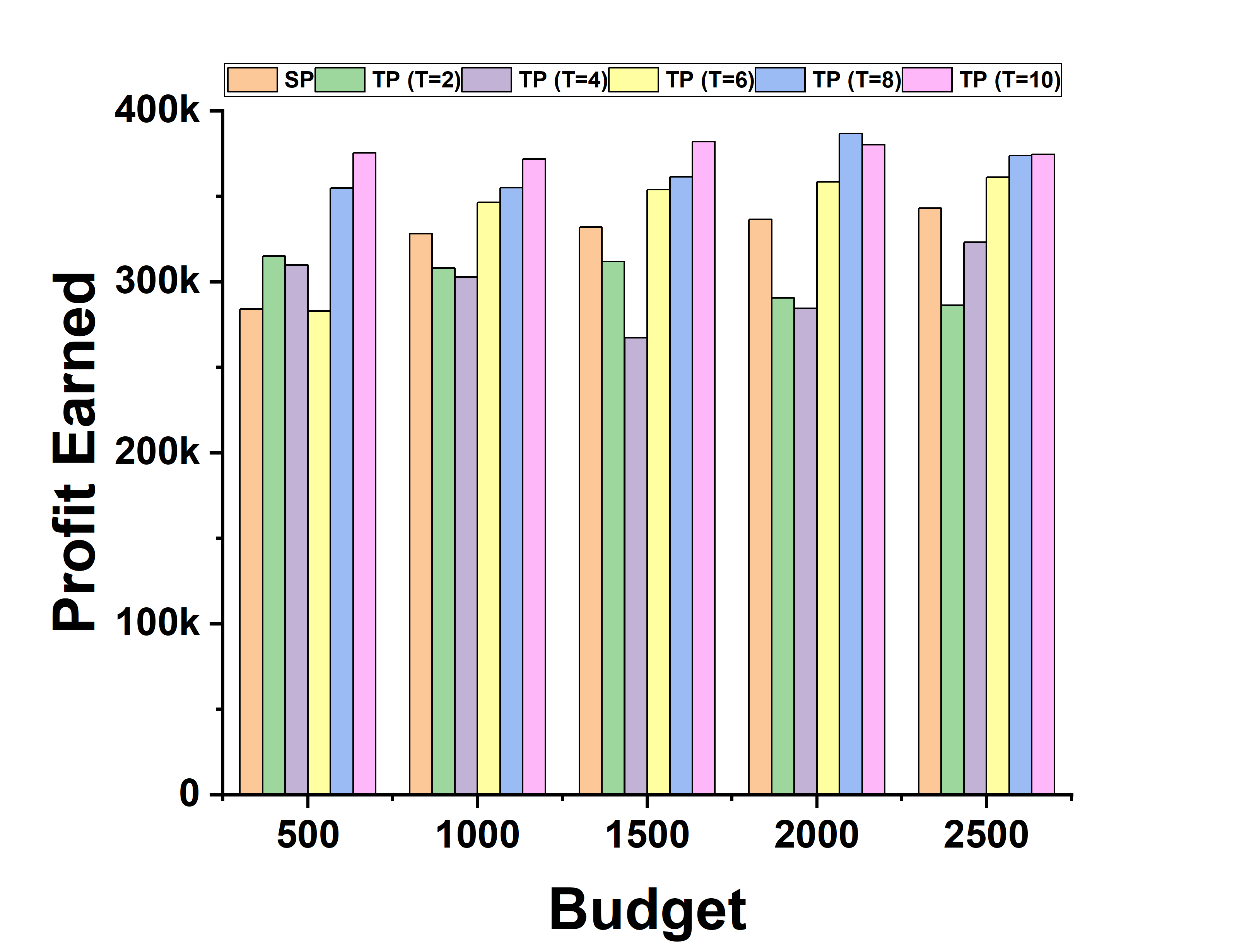}
        \caption{Double Greedy}
    \end{subfigure} &
    \begin{subfigure}[t]{0.22\textwidth}
        \includegraphics[width=\linewidth]{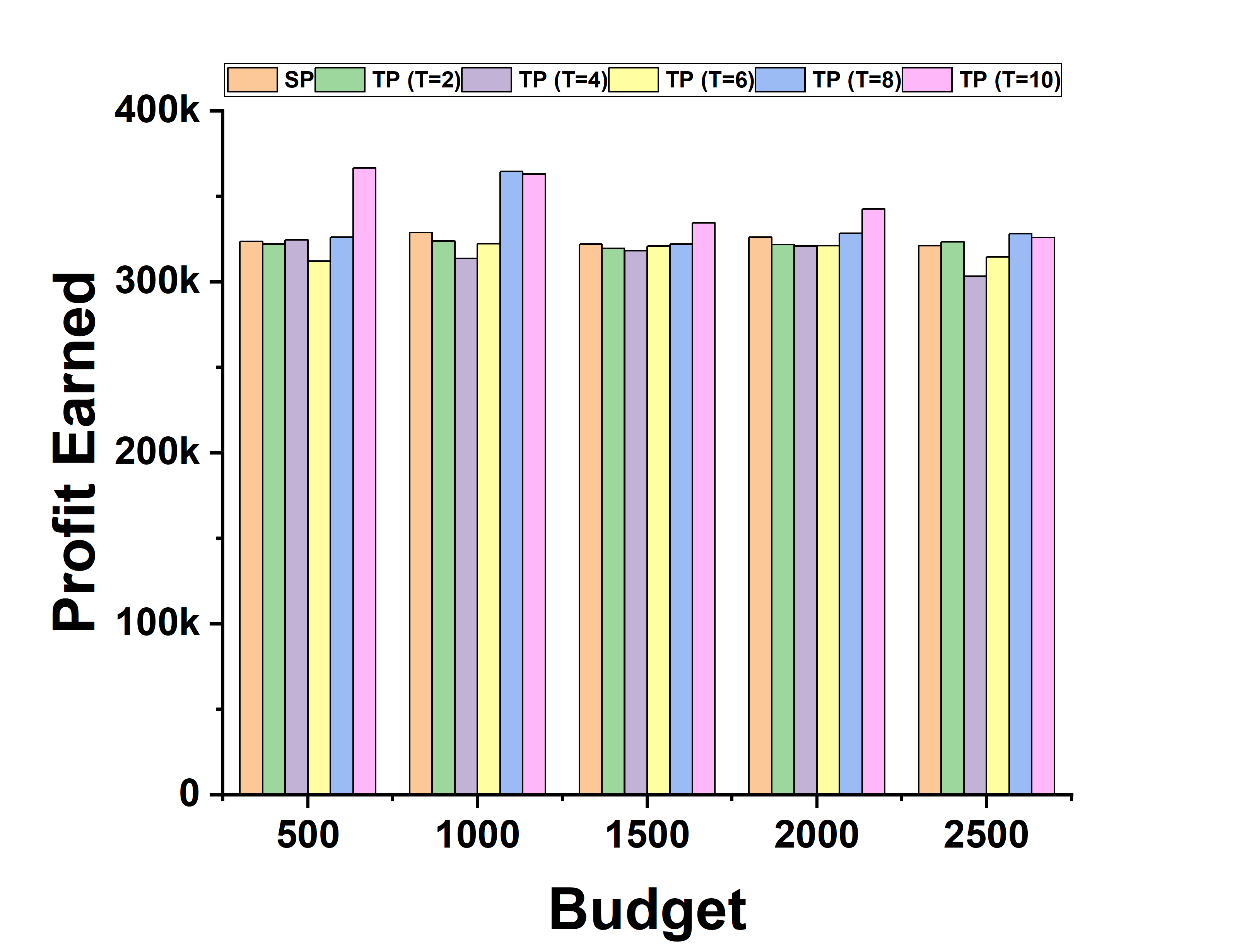}
        \caption{Stochastic Greedy}
    \end{subfigure}
\end{tabular}
\caption{Profit Earned in Single Phase Vs. Two Phase setting (split ratio 90\%, Probability Setting - Trivalency, \textit{Email-Eu-Core} Dataset)}
\label{Fig:RQ1_T5}
\end{figure}

\subsubsection{Effect of Budget Splitting Ratio on Profit Earned in Two Phase Settings}
For the \textit{LM} dataset under the trivalency probability setting, the split ratio plays an important role in shaping the performance of algorithms in the two phase framework. The results vary across algorithms, showing that the best split ratio depends on the method used. Several baseline algorithms such as \textbf{Random}, \textbf{HD}, \textbf{HighCC}, and \textbf{DD} achieved their highest profits at a low split ratio of $0.1$. For example, \textbf{Random} reached its peak profit of $35747.20$ at budget $2500$, split ratio $0.1$, and timestep $10$ (Figure~\ref{Fig:RQ2LM_T5}(a)).
While \textbf{HD} and \textbf{DD} showed strong results at all split ratios, they performed better at higher split ratios like $0.7$ and $0.9$ in terms of minimum profits, suggesting they offer more consistent results and lower risk at those settings. In contrast, the \textbf{SD} algorithm performed best when more budget was used in the first phase. It achieved its top profit of $37960.78$ at budget $2500$, split ratio $0.9$, and timestep $8$ (Figure~\ref{Fig:RQ2LM_T5}(e)). This represents a $23.5\%$ gain compared to the single-phase profit, showing the importance of choosing the right split ratio for \textbf{SD}. Among the proposed methods, \textbf{DG} showed stable and high performance at multiple split ratios. For example, at budget $500$, split ratio $0.1$, and timestep $2$, \textbf{DG} earned $7255.72$ in two phase, compared to $6645.55$ in single-phase (Figure~\ref{Fig:RQ2LM_T1}(g)). The \textbf{SG} algorithm needed a higher split ratio to perform well. It showed significant improvements when the split ratio was between $0.5$ and $0.9$. In one case (budget $500$, split ratio $0.9$, timestep $8$), \textbf{SG} earned $11620$ compared to $9653.32$ in the single-phase version (Figure~\ref{Fig:RQ2LM_T4}(f)). Our \textbf{StG0.1} algorithm also performed better with the right configuration. At budget $500$, split ratio $0.7$, and timestep $2$, it earned $10631.02$, which is higher than its single-phase profit of $9526.36$ (Figure~\ref{Fig:RQ2LM_T1}(h)). To sum up, there is no single best split ratio for all algorithms. A lower split ratio like $0.1$ works well for \textbf{Random, HD, and DD,} while higher values ($0.5$–$0.9$) are necessary for the best results from \textbf{SD} and \textbf{SG}. The \textbf{DG} algorithm is more flexible, offering strong performance across all split ratios.

Figures~\ref{Fig:RQ2_T1} to \ref{Fig:RQ2_T5} show how the profit earned in the two phase setting changes with different budget split ratios for \textit{Email-Eu-Core} dataset. The split ratio refers to how much of the total budget is used in the first phase, with the rest used in the second phase. The \textbf{Random} algorithm tends to perform better at lower to mid-range split ratios. For example, at a split ratio of $0.3$, it earned an average profit of $349061.65$, which is higher than the $331087.34$ average at a split ratio of $0.9$. The \textbf{HD} algorithm shows some changes in profit across split ratios, but there is no consistent increasing or decreasing trend. For instance, at a split ratio of $0.1$, budget $500$, and timestep $10$, \textbf{HD} earned $354999$ (Figure~\ref{Fig:RQ2_T5}(b)). However, at split ratio $0.9$, budget $500$, and timestep $2$, its profit dropped to $163208$ (Figure~\ref{Fig:RQ2_T1}(b)). For the \textbf{HighCC} algorithm, the average profit increases up to a split ratio of $0.7$ ($229018.79$), but slightly drops at $0.9$ ($221796.87$), suggesting that a balanced budget allocation may work best. The \textbf{DD} algorithm generally earns less profit as the split ratio increases. At a split ratio of $0.1$, budget $500$, and timestep $10$, \textbf{DD} earned $377199.75$. But at $0.9$, budget $500$, and timestep $4$, the profit dropped to $289481.35$, indicating better performance when more budget is saved for the second phase. Similarly, the \textbf{SD} algorithm also tends to perform better with lower split ratios. For example, at split ratio $0.1$, budget $500$, and timestep $10$, \textbf{SD} earned $324635.26$. Still, in one instance with a higher split ratio of $0.9$, budget $500$, and timestep $6$, \textbf{SD} reached a profit of $357756.21$, showing that some exceptions exist. The \textbf{SG} algorithm shows the opposite trend: it performs better when more budget is allocated in the first phase. At split ratio $0.1$, budget $500$, and timestep $2$, \textbf{SG} earned $186836.42$, while at split ratio $0.9$, budget $500$, and timestep $10$, the profit rose to $356498$. The \textbf{DG} algorithm has a more complex pattern. Its best performance is seen at split ratio $0.3$, where the average profit is $348508.80$, while at $0.9$ it decreases slightly to $336590.67$. Finally, the \textbf{StG0.1} algorithm remains very stable across all split ratios. For example, at a split ratio of $0.1$, budget $500$, and timestep $10$, it earned $366271.11$, and at $0.9$ with the same configuration, it earned $366407.54$. In summary, algorithms that rely on broader early diffusion (like \textbf{Random, DD, SD, DG}) often benefit from lower split ratios where more budget is kept for the second phase. On the other hand, algorithms such as \textbf{SG} work better when more budget is spent upfront, with profits increasing as the split ratio rises.


\begin{figure}[htbp]
\centering
\captionsetup[sub]{font=footnotesize}
\begin{tabular}{cccc}
    \begin{subfigure}[t]{0.22\textwidth}
        \includegraphics[width=\linewidth]{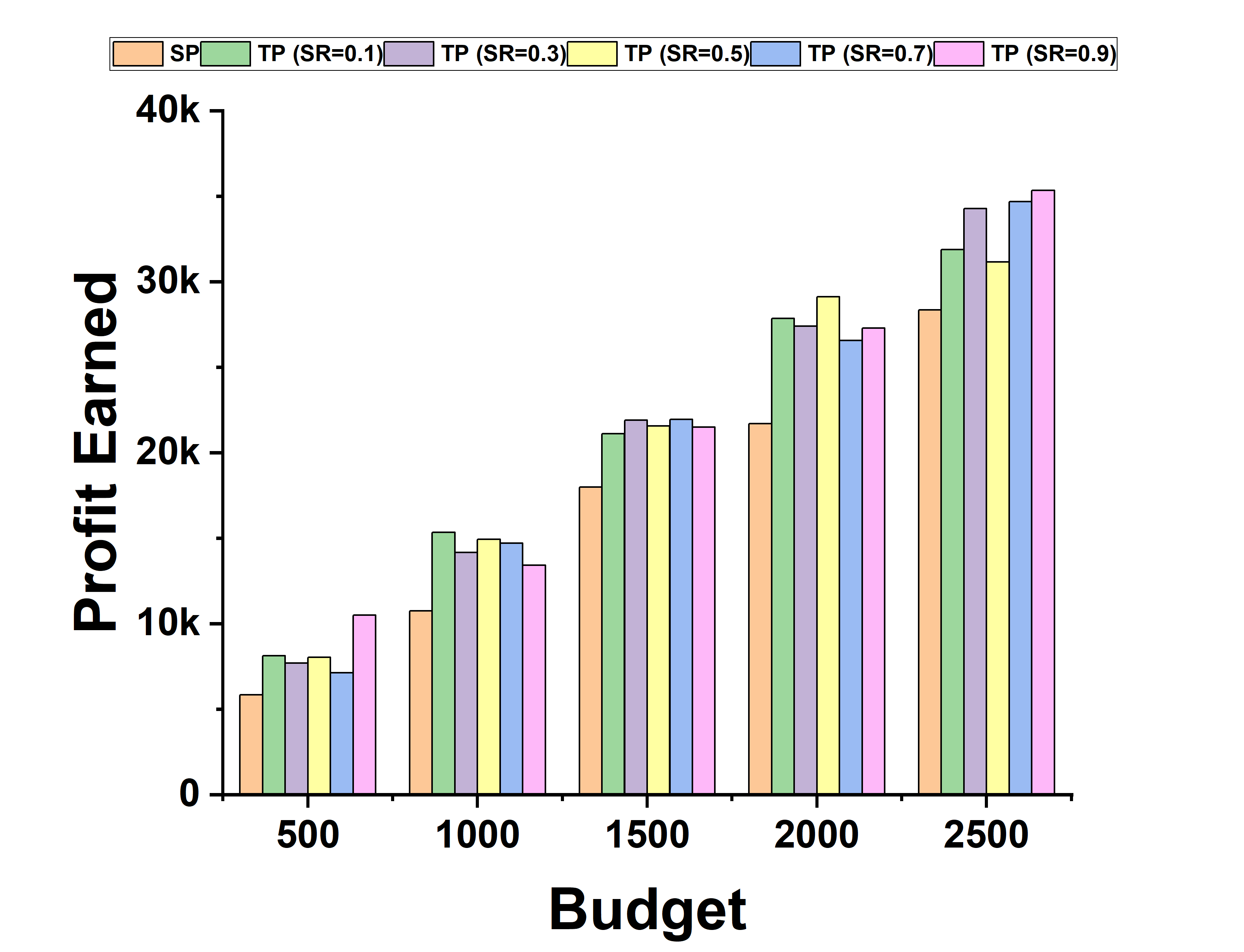}
        \caption{Random}
    \end{subfigure} &
    \begin{subfigure}[t]{0.22\textwidth}
        \includegraphics[width=\linewidth]{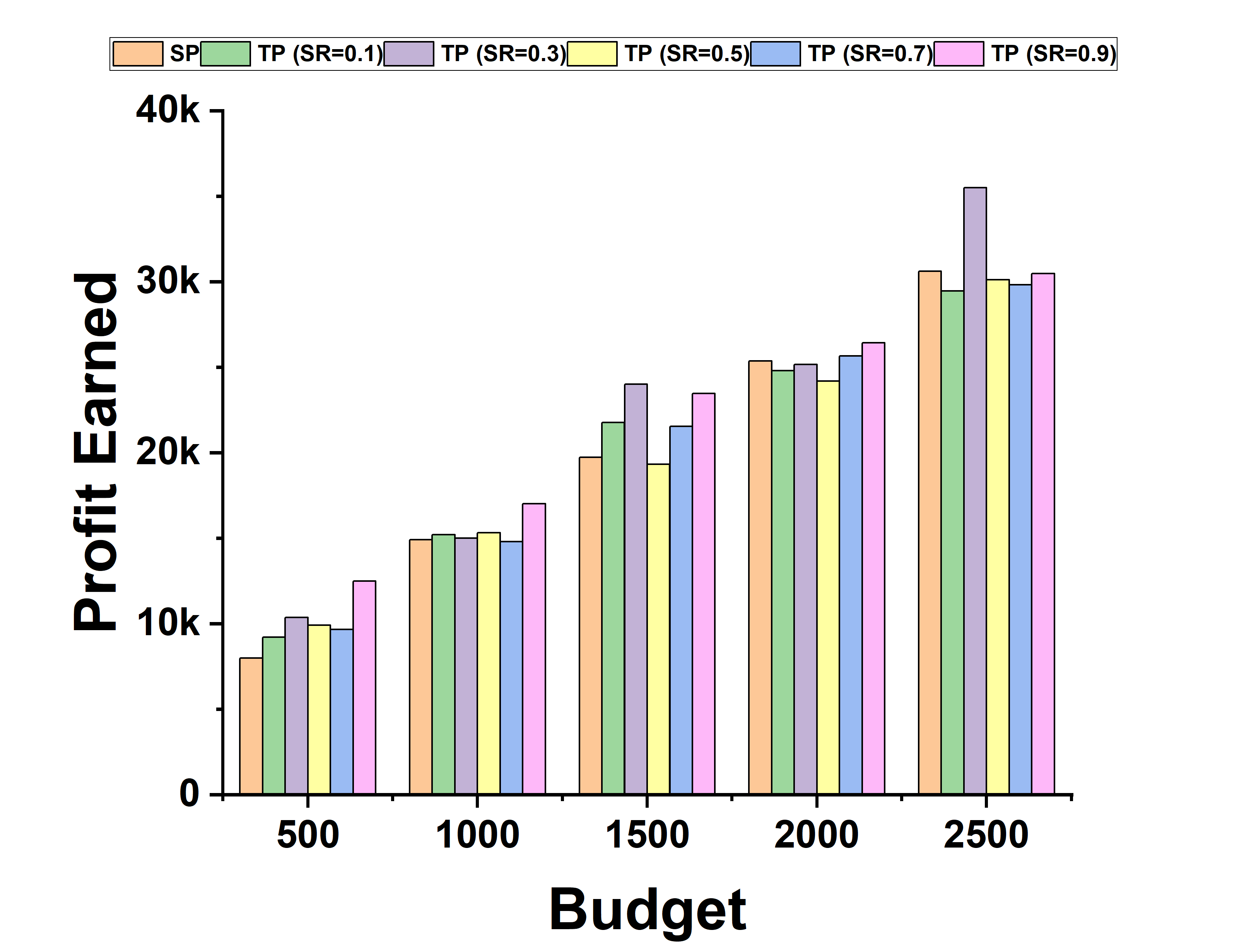}
        \caption{High Degree}
    \end{subfigure} &
    \begin{subfigure}[t]{0.22\textwidth}
        \includegraphics[width=\linewidth]{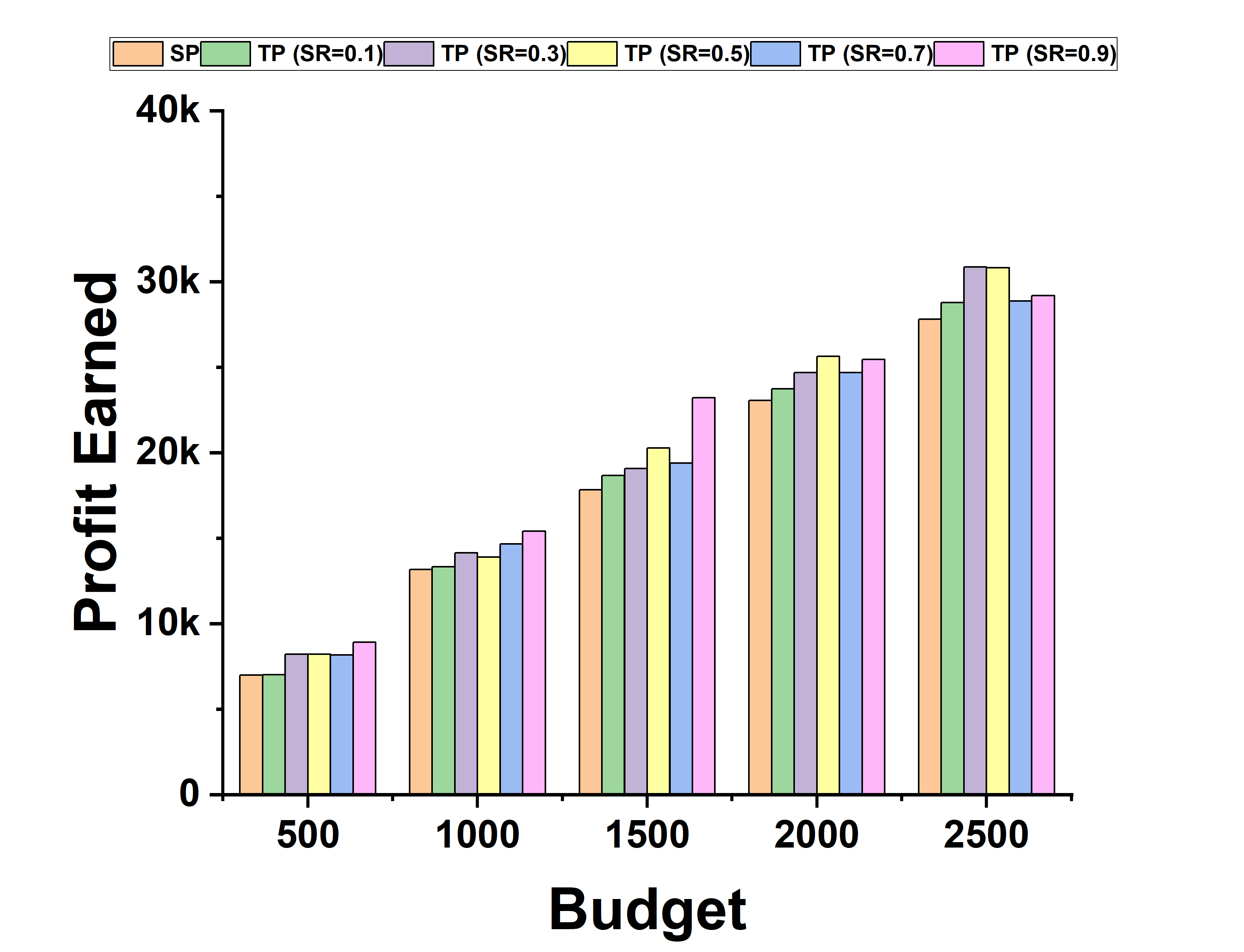}
        \caption{Clustering\\Coefficient}
    \end{subfigure} &
    \begin{subfigure}[t]{0.22\textwidth}
        \includegraphics[width=\linewidth]{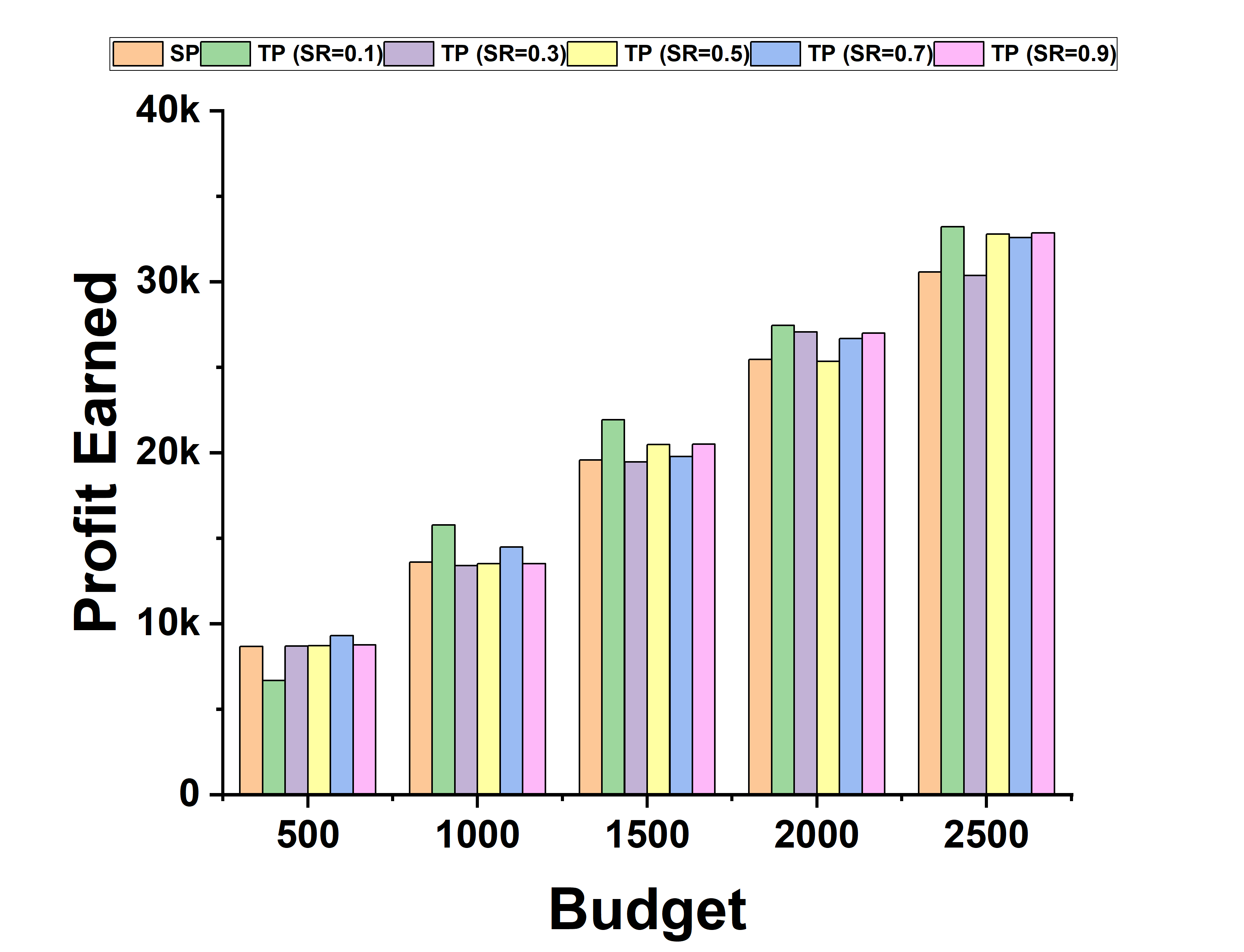}
        \caption{Degree Discount}
    \end{subfigure} \\[6pt]

    \begin{subfigure}[t]{0.22\textwidth}
        \includegraphics[width=\linewidth]{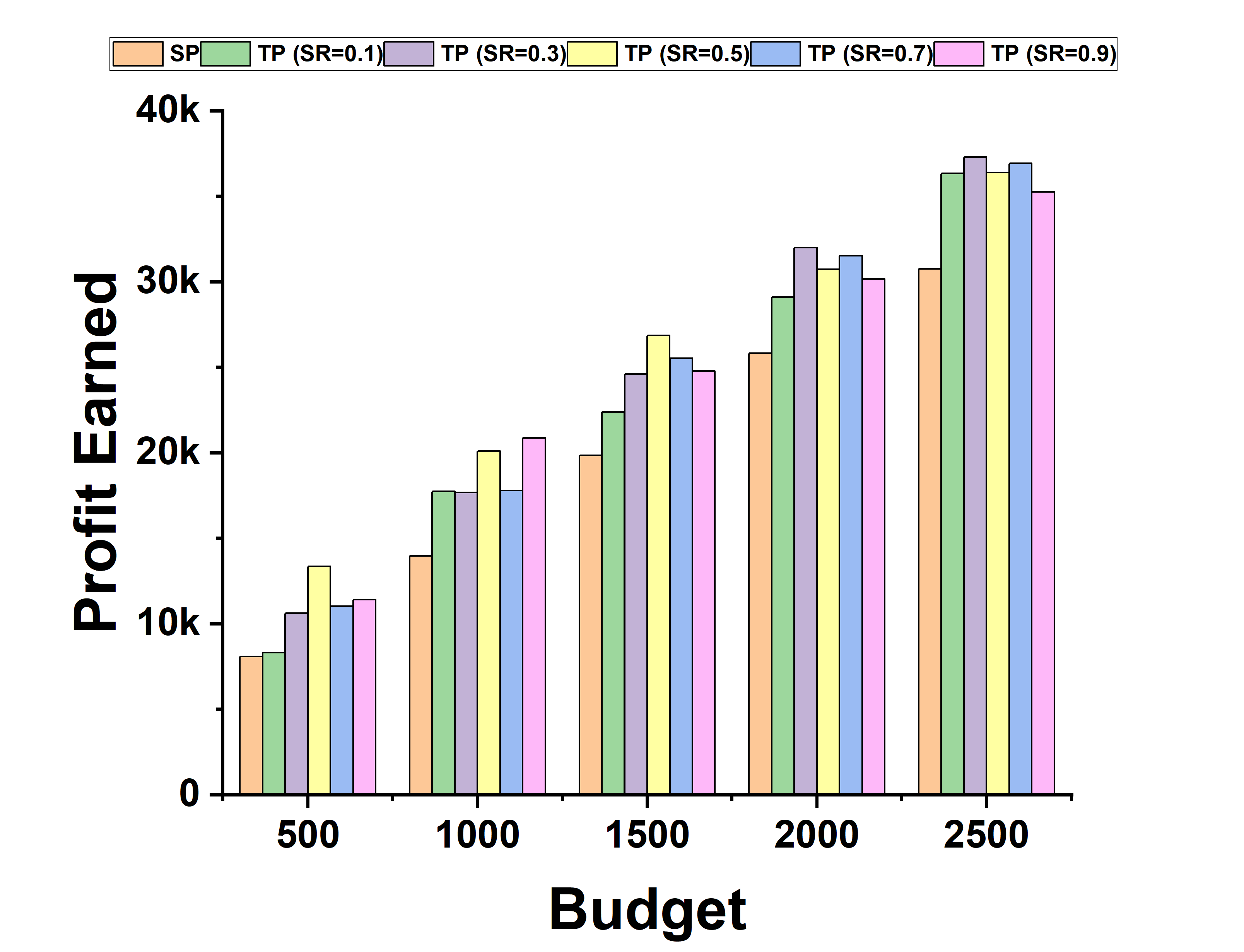}
        \caption{Single Discount}
    \end{subfigure} &
    \begin{subfigure}[t]{0.22\textwidth}
        \includegraphics[width=\linewidth]{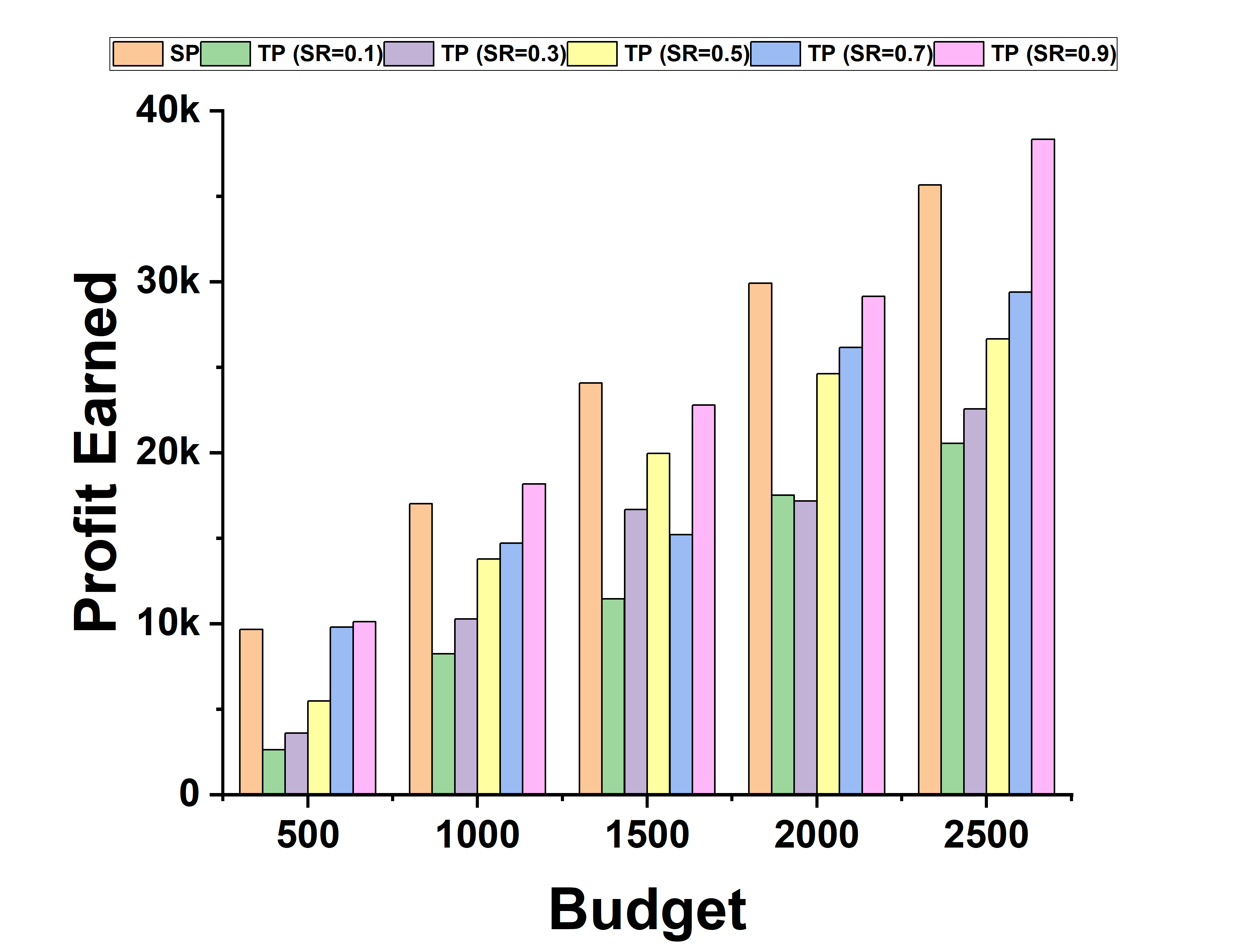}
        \caption{Simple Greedy}
    \end{subfigure} &
    \begin{subfigure}[t]{0.22\textwidth}
        \includegraphics[width=\linewidth]{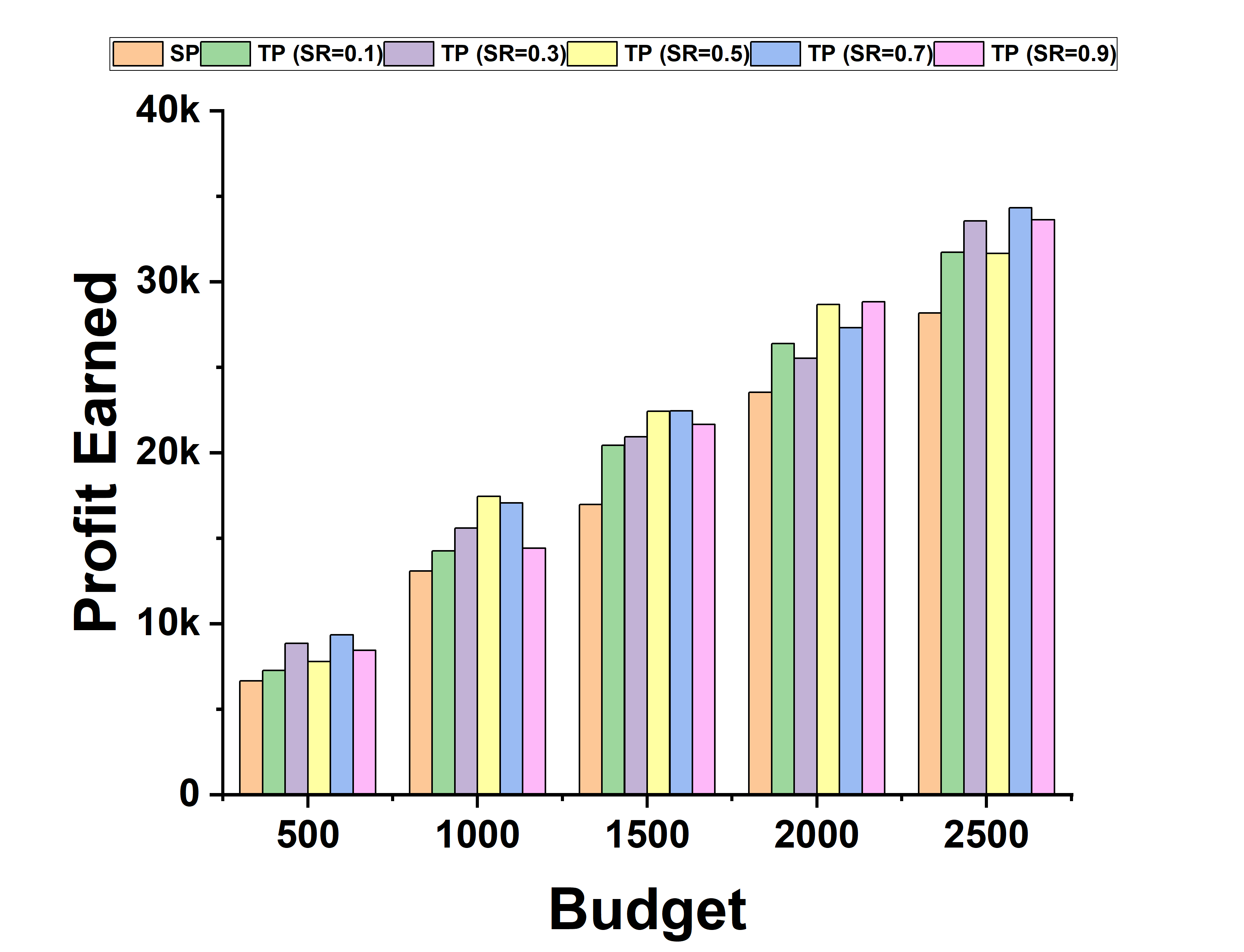}
        \caption{Double Greedy}
    \end{subfigure} &
    \begin{subfigure}[t]{0.22\textwidth}
        \includegraphics[width=\linewidth]{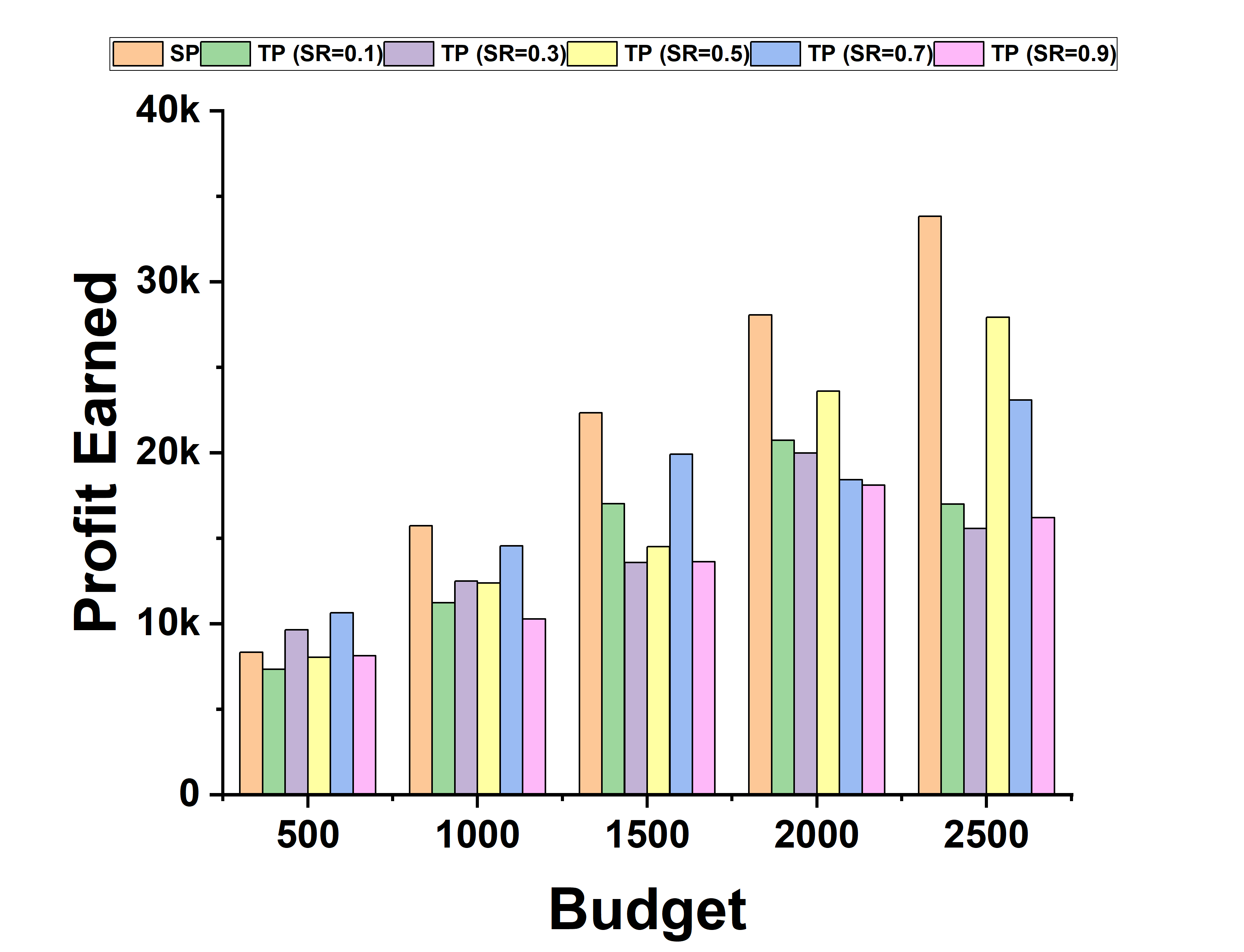}
        \caption{Stochastic Greedy}
    \end{subfigure}
\end{tabular}
\caption{Profit Earned in Single Phase Vs. Two Phase setting  (Timestep 2, Probability Setting - Trivalency, \textit{LM} Dataset)}
\label{Fig:RQ2LM_T1}
\end{figure}

\begin{figure}[htbp]
\centering
\captionsetup[sub]{font=footnotesize}
\begin{tabular}{cccc}
    \begin{subfigure}[t]{0.22\textwidth}
        \includegraphics[width=\linewidth]{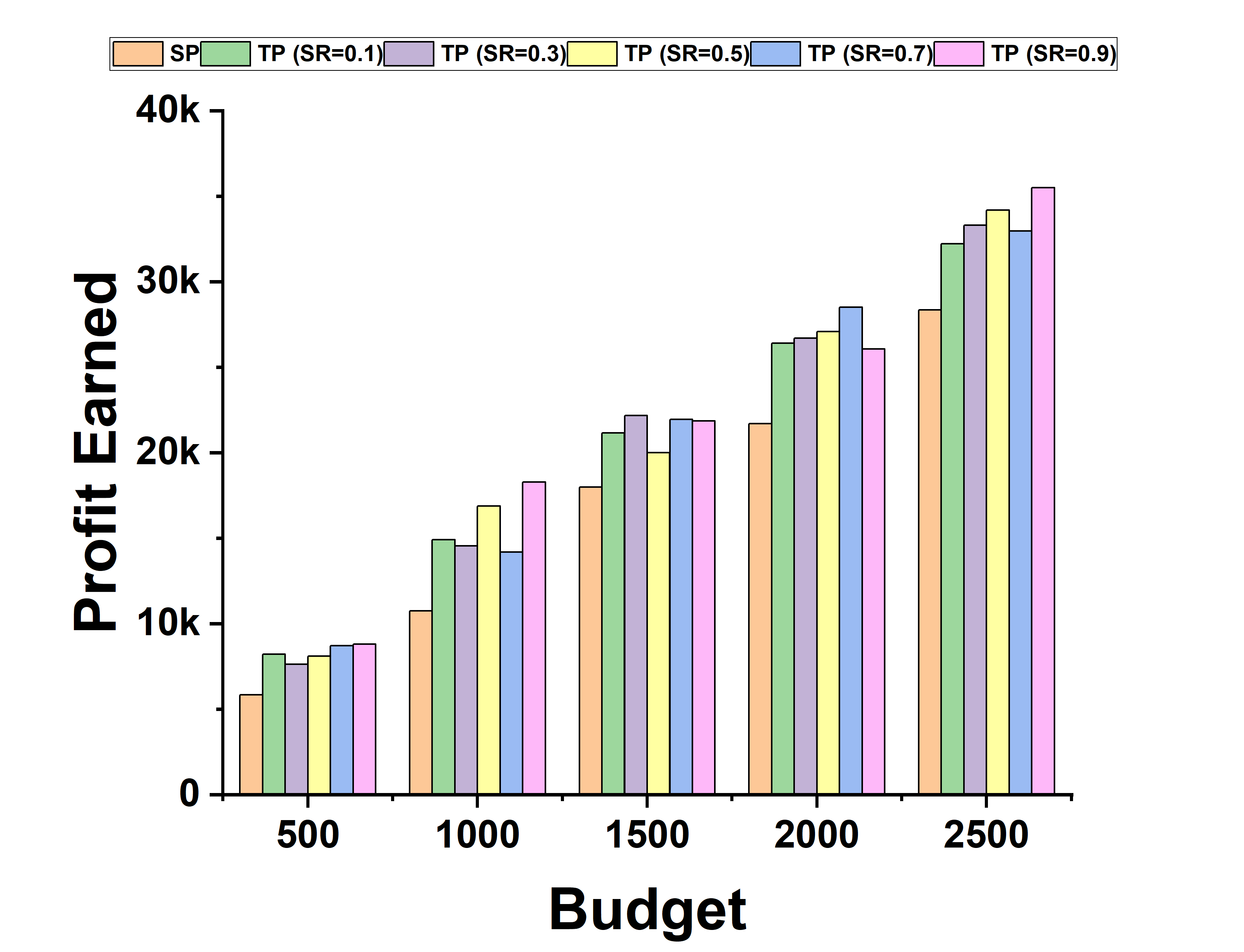}
        \caption{Random}
    \end{subfigure} &
    \begin{subfigure}[t]{0.22\textwidth}
        \includegraphics[width=\linewidth]{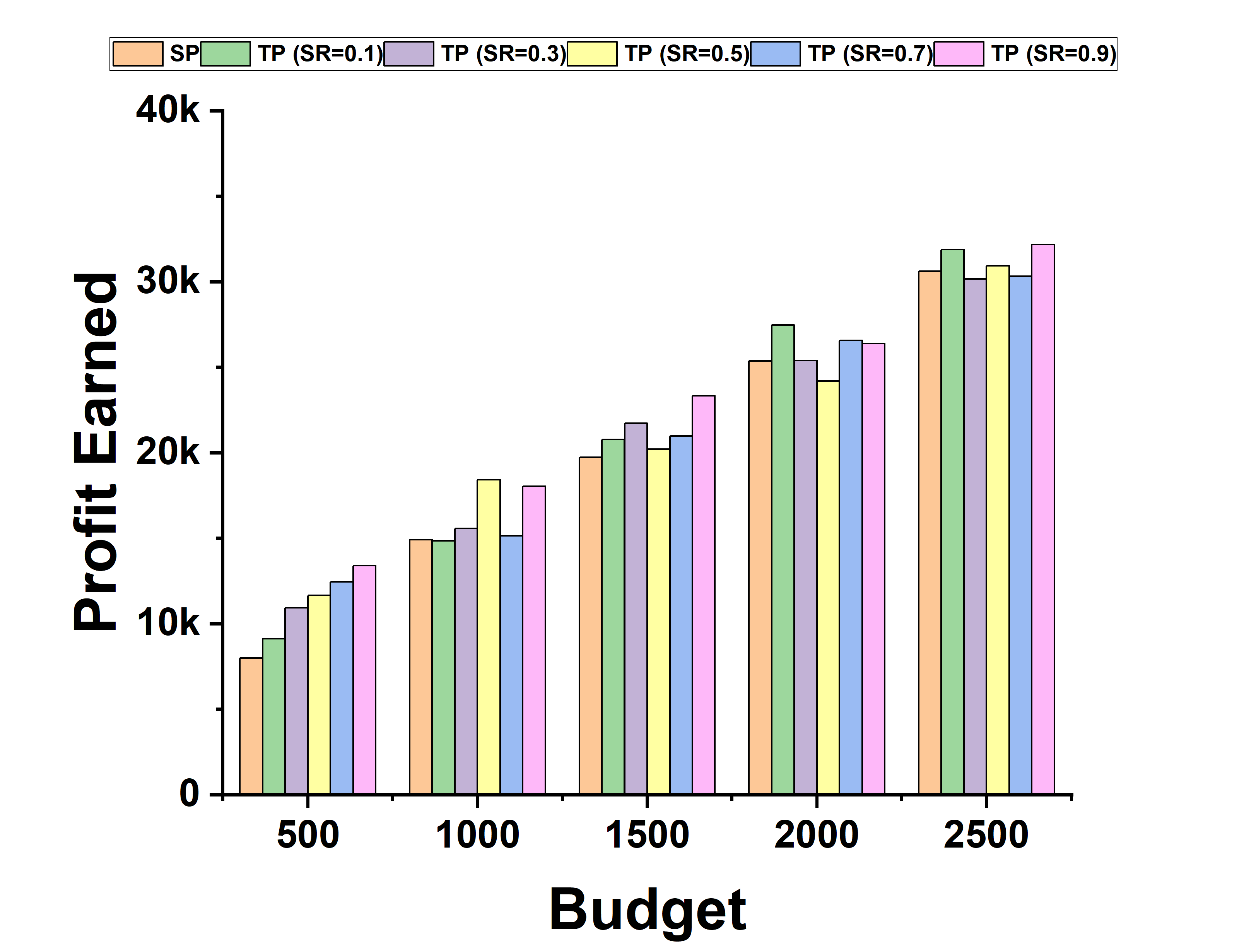}
        \caption{High Degree}
    \end{subfigure} &
    \begin{subfigure}[t]{0.22\textwidth}
        \includegraphics[width=\linewidth]{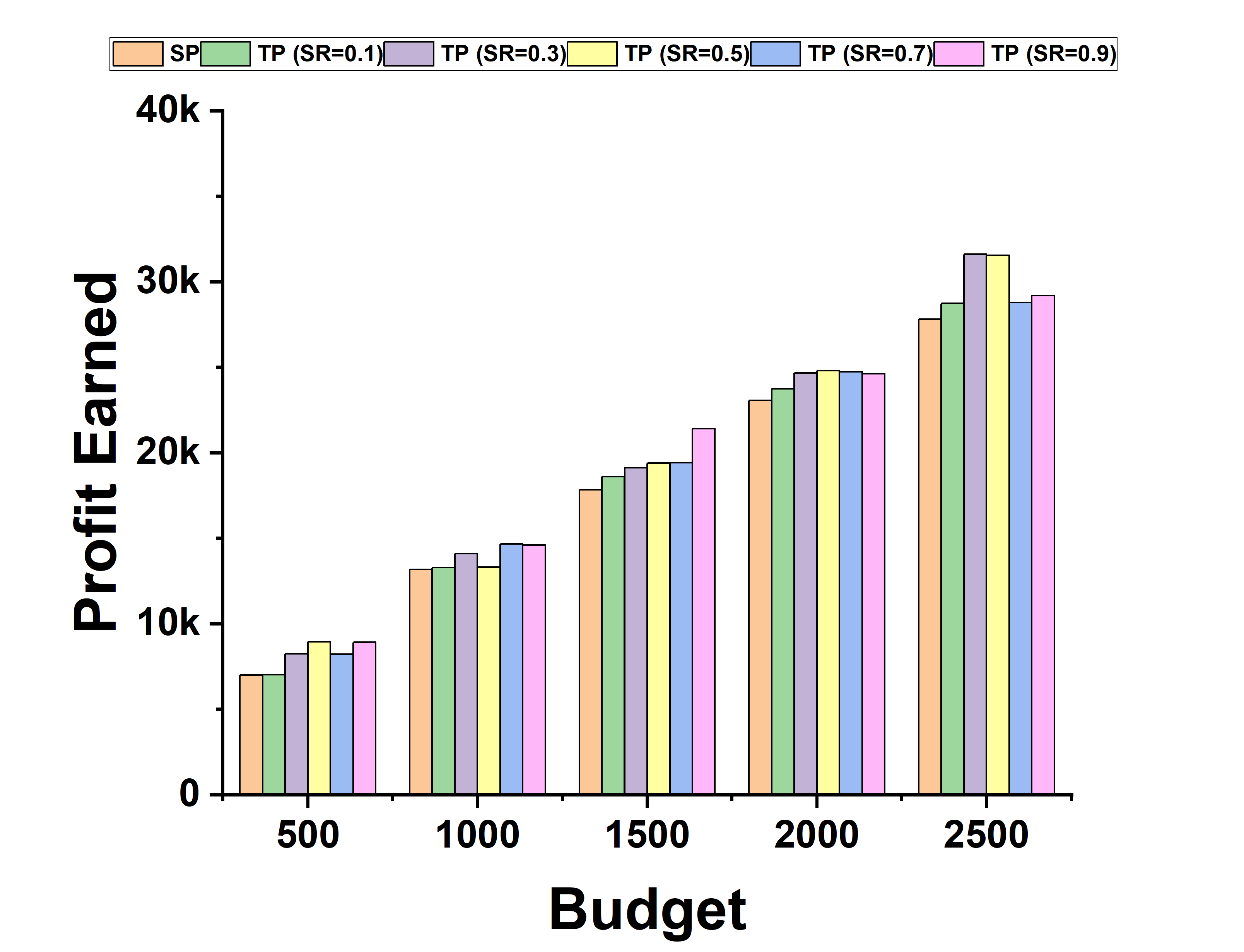}
        \caption{Clustering\\Coefficient}
    \end{subfigure} &
    \begin{subfigure}[t]{0.22\textwidth}
        \includegraphics[width=\linewidth]{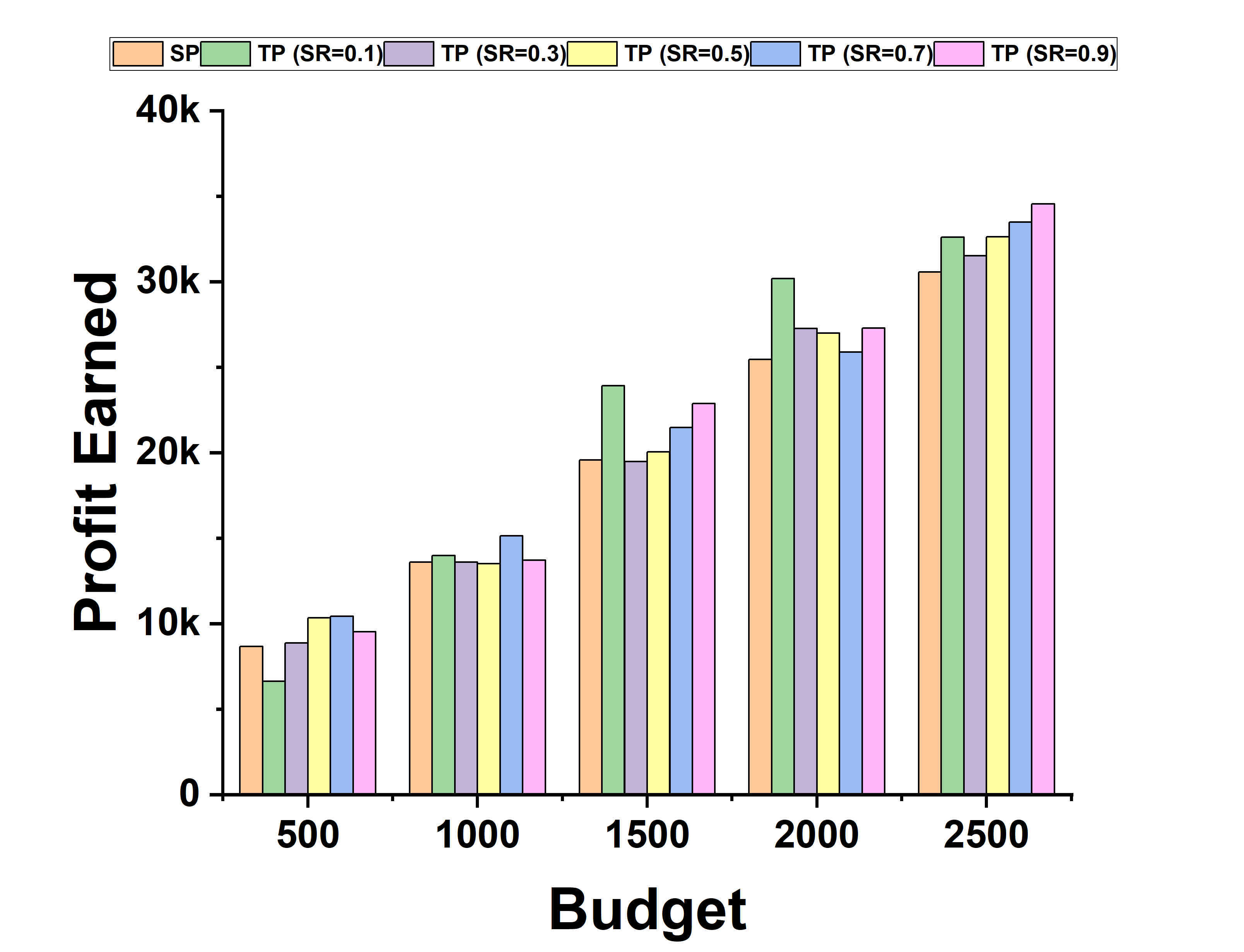}
        \caption{Degree Discount}
    \end{subfigure} \\[6pt]

    \begin{subfigure}[t]{0.22\textwidth}
        \includegraphics[width=\linewidth]{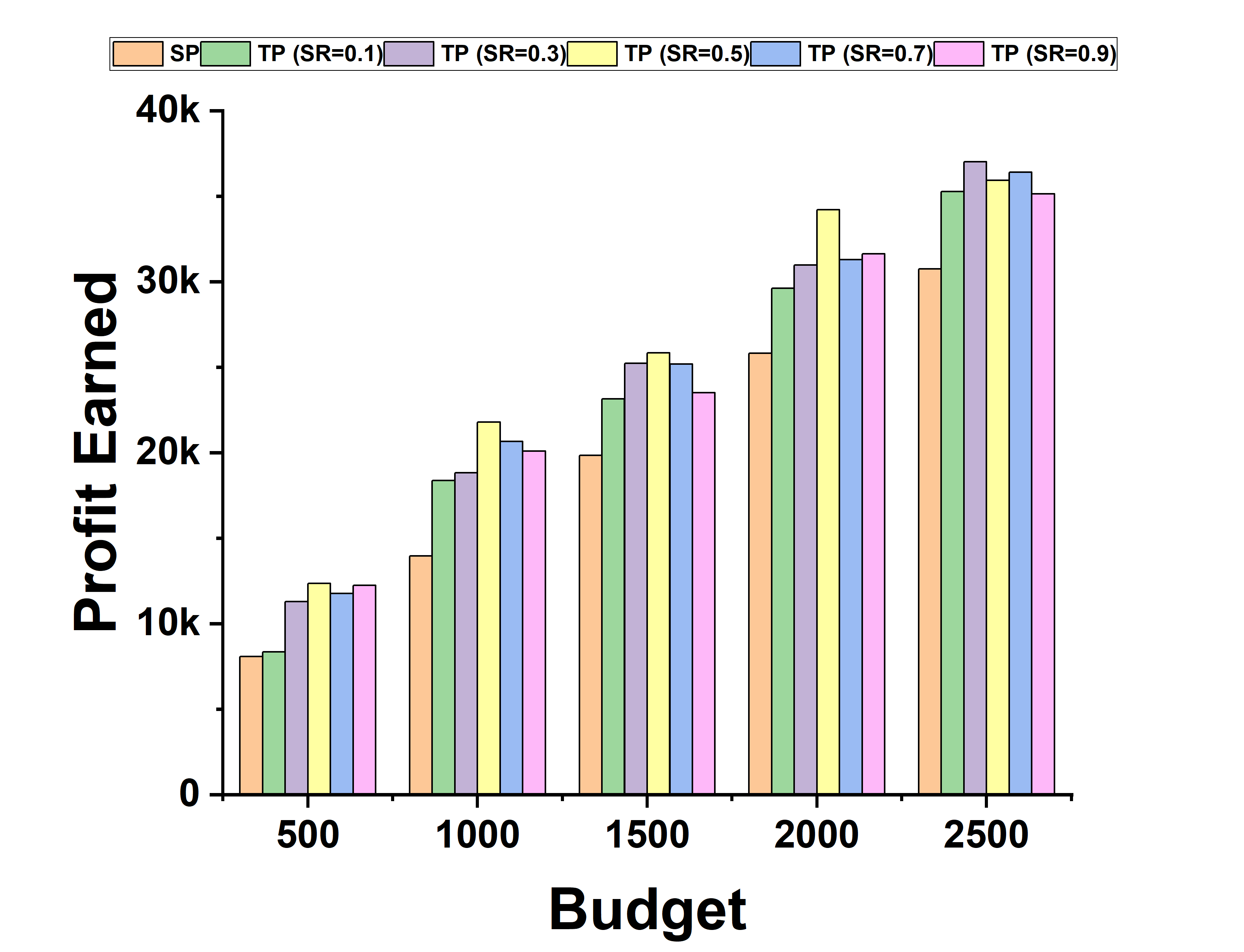}
        \caption{Single Discount}
    \end{subfigure} &
    \begin{subfigure}[t]{0.22\textwidth}
        \includegraphics[width=\linewidth]{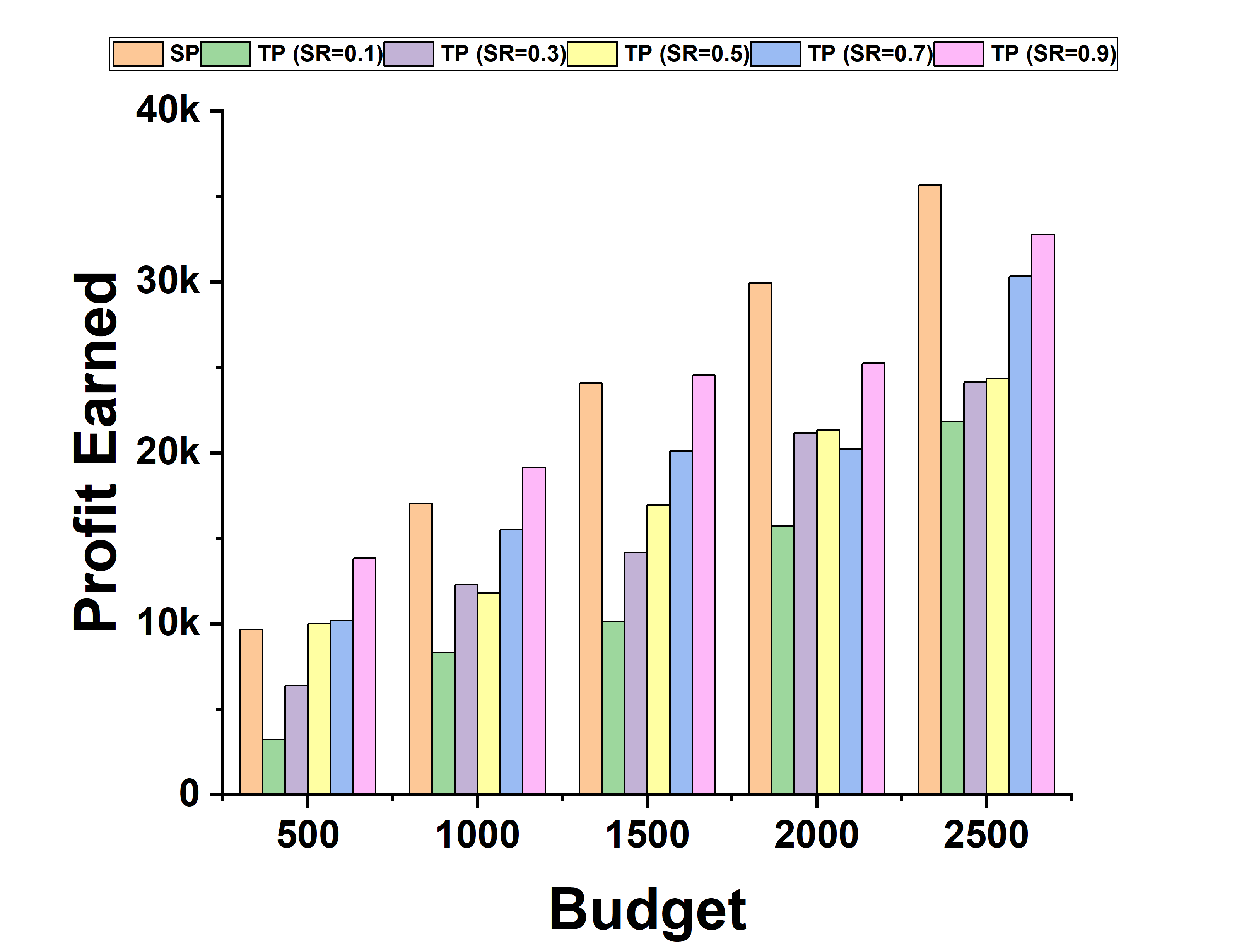}
        \caption{Simple Greedy}
    \end{subfigure} &
    \begin{subfigure}[t]{0.22\textwidth}
        \includegraphics[width=\linewidth]{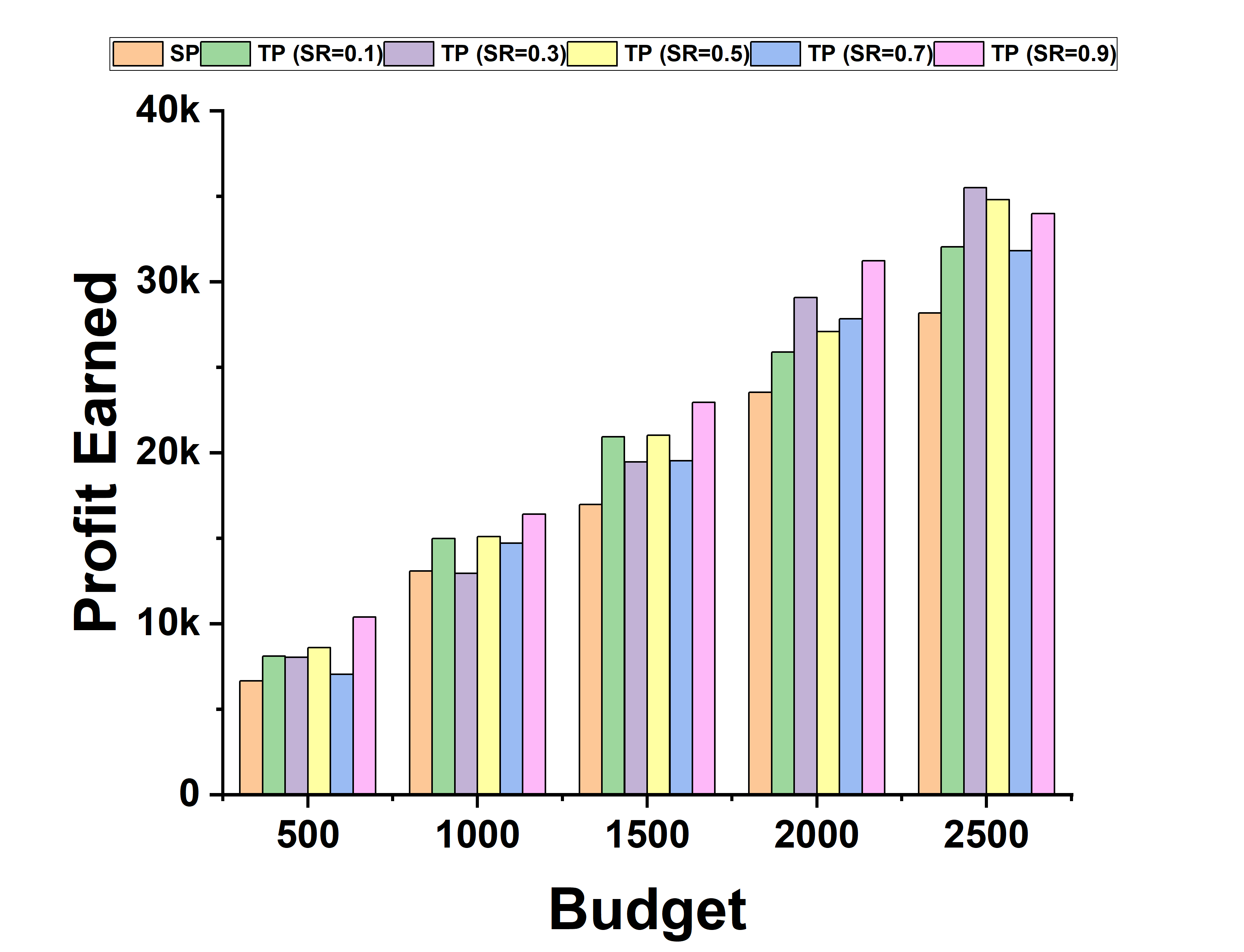}
        \caption{Double Greedy}
    \end{subfigure} &
    \begin{subfigure}[t]{0.22\textwidth}
        \includegraphics[width=\linewidth]{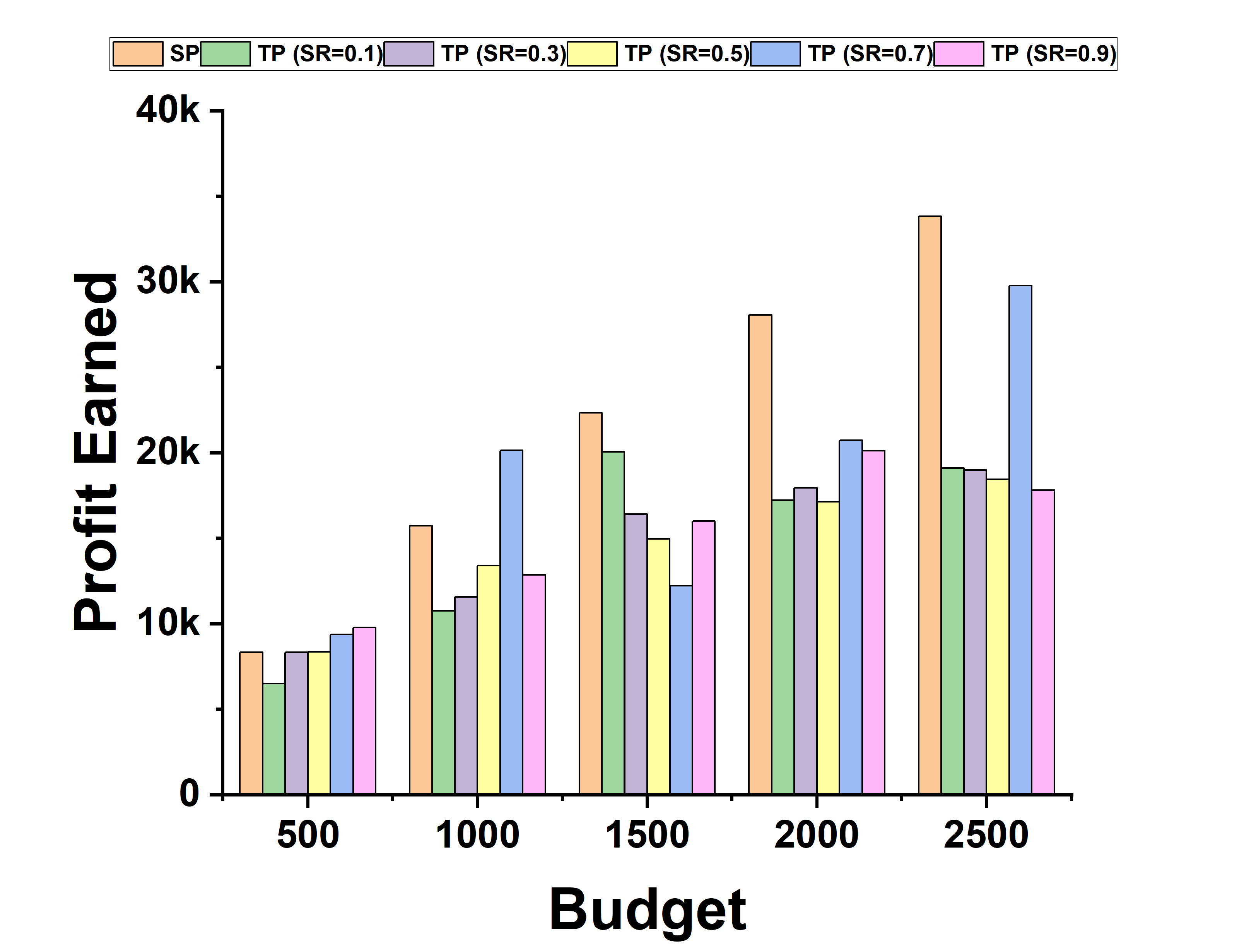}
        \caption{Stochastic Greedy}
    \end{subfigure}
\end{tabular}
\caption{Profit Earned in Single Phase Vs. Two Phase setting (Timestep 4, Probability Setting - Trivalency, \textit{LM} Dataset)}
\label{Fig:RQ2LM_T2}
\end{figure}

\begin{figure}[htbp]
\centering
\captionsetup[sub]{font=footnotesize}
\begin{tabular}{cccc}
    \begin{subfigure}[t]{0.22\textwidth}
        \includegraphics[width=\linewidth]{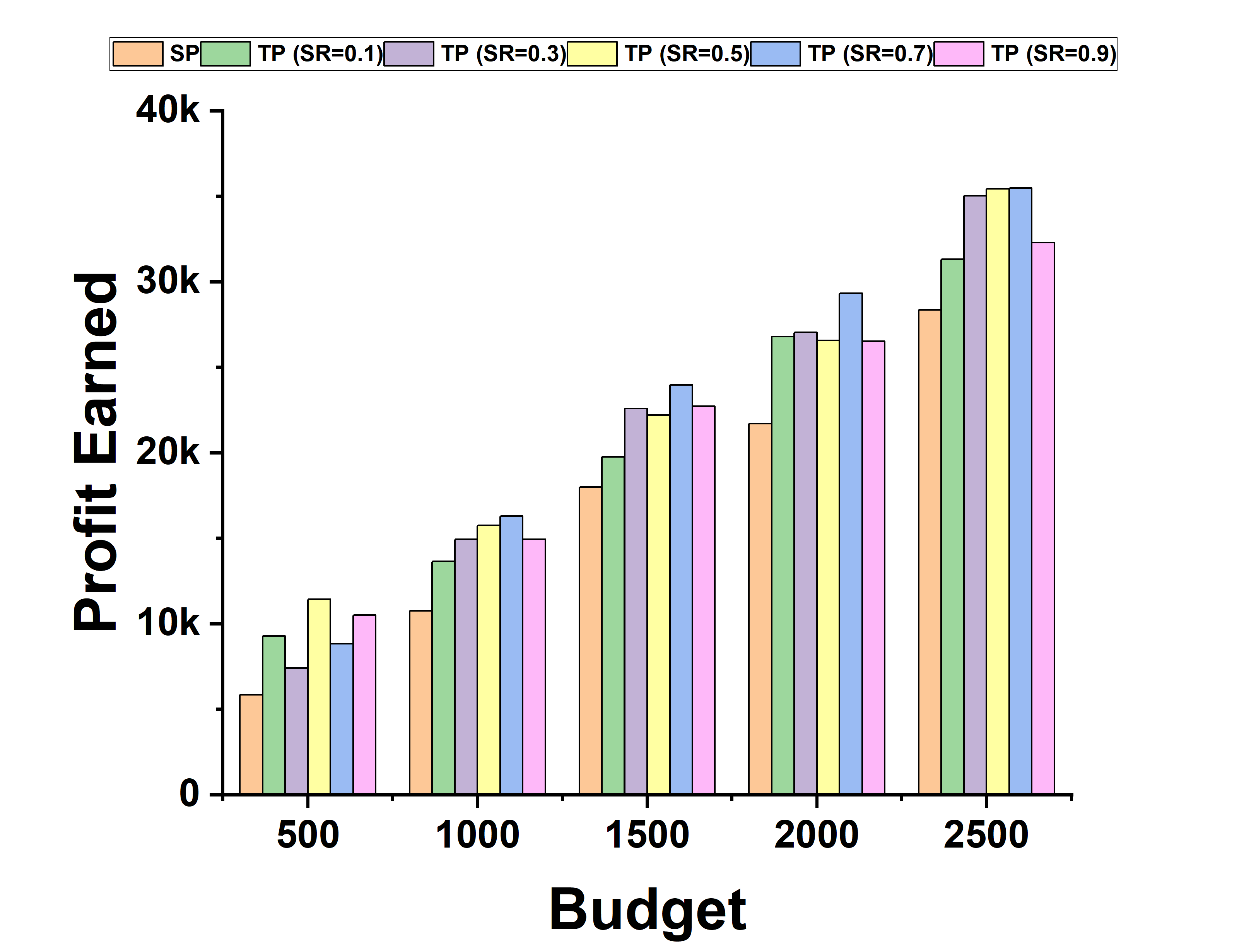}
        \caption{Random}
    \end{subfigure} &
    \begin{subfigure}[t]{0.22\textwidth}
        \includegraphics[width=\linewidth]{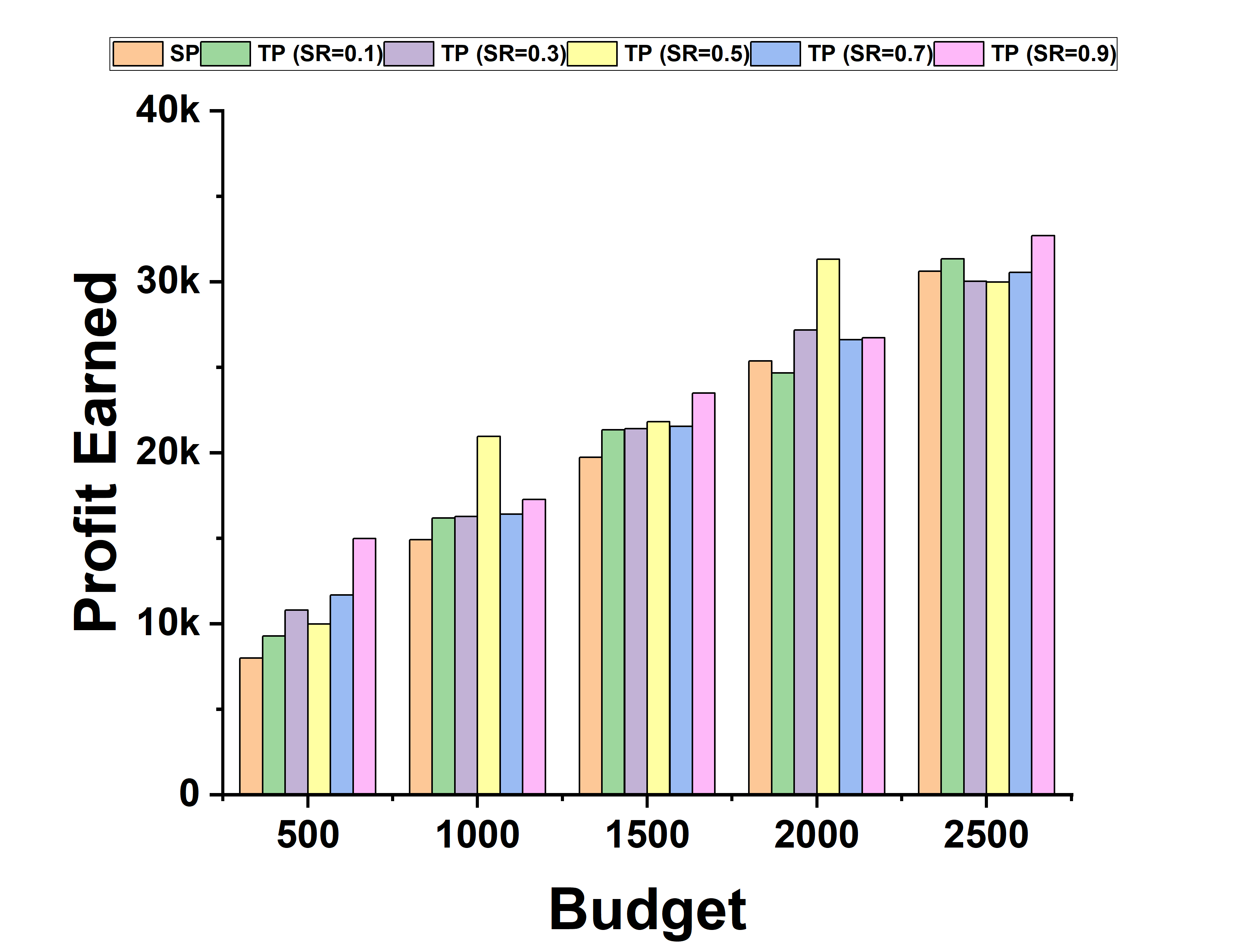}
        \caption{High Degree}
    \end{subfigure} &
    \begin{subfigure}[t]{0.22\textwidth}
        \includegraphics[width=\linewidth]{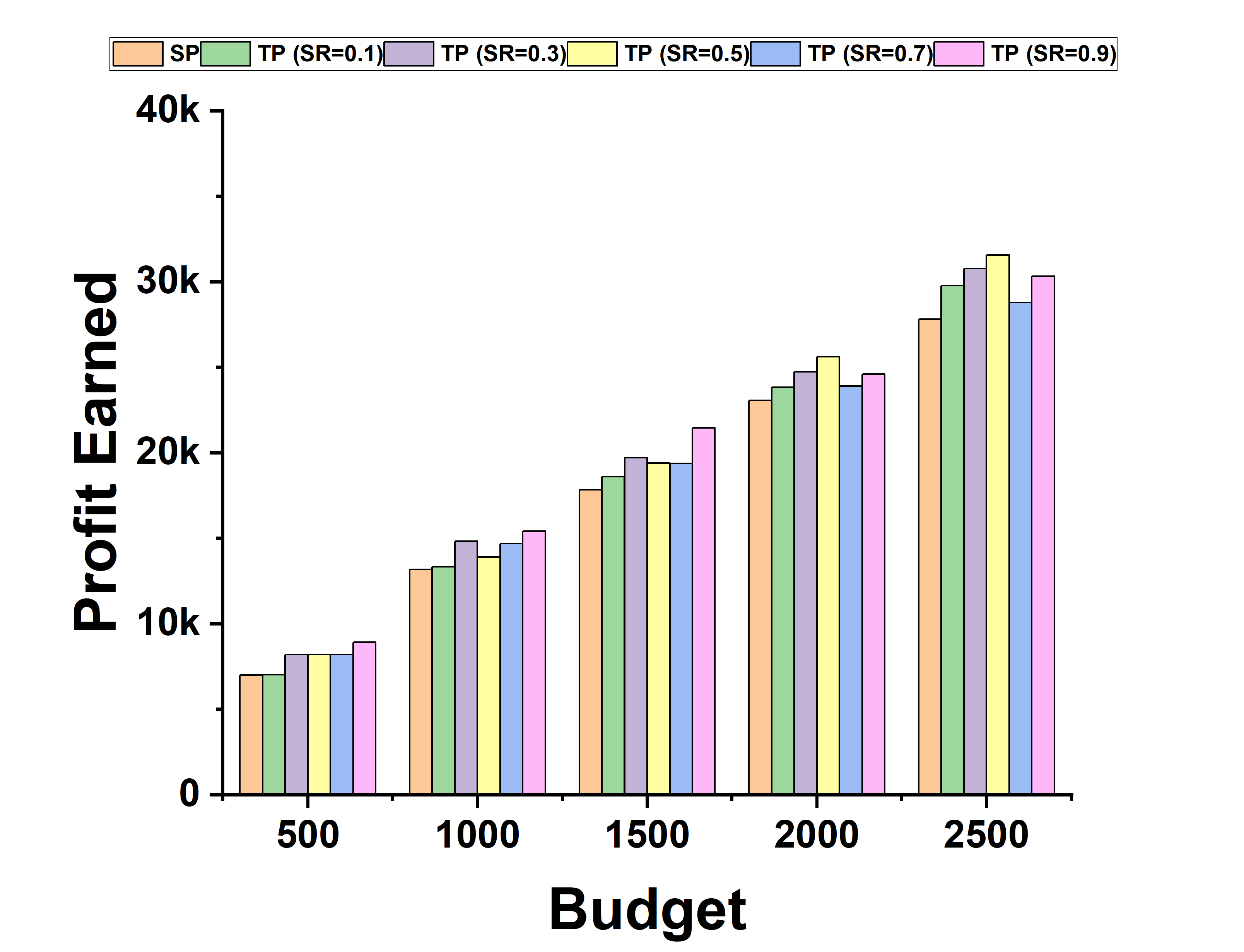}
        \caption{Clustering\\Coefficient}
    \end{subfigure} &
    \begin{subfigure}[t]{0.22\textwidth}
        \includegraphics[width=\linewidth]{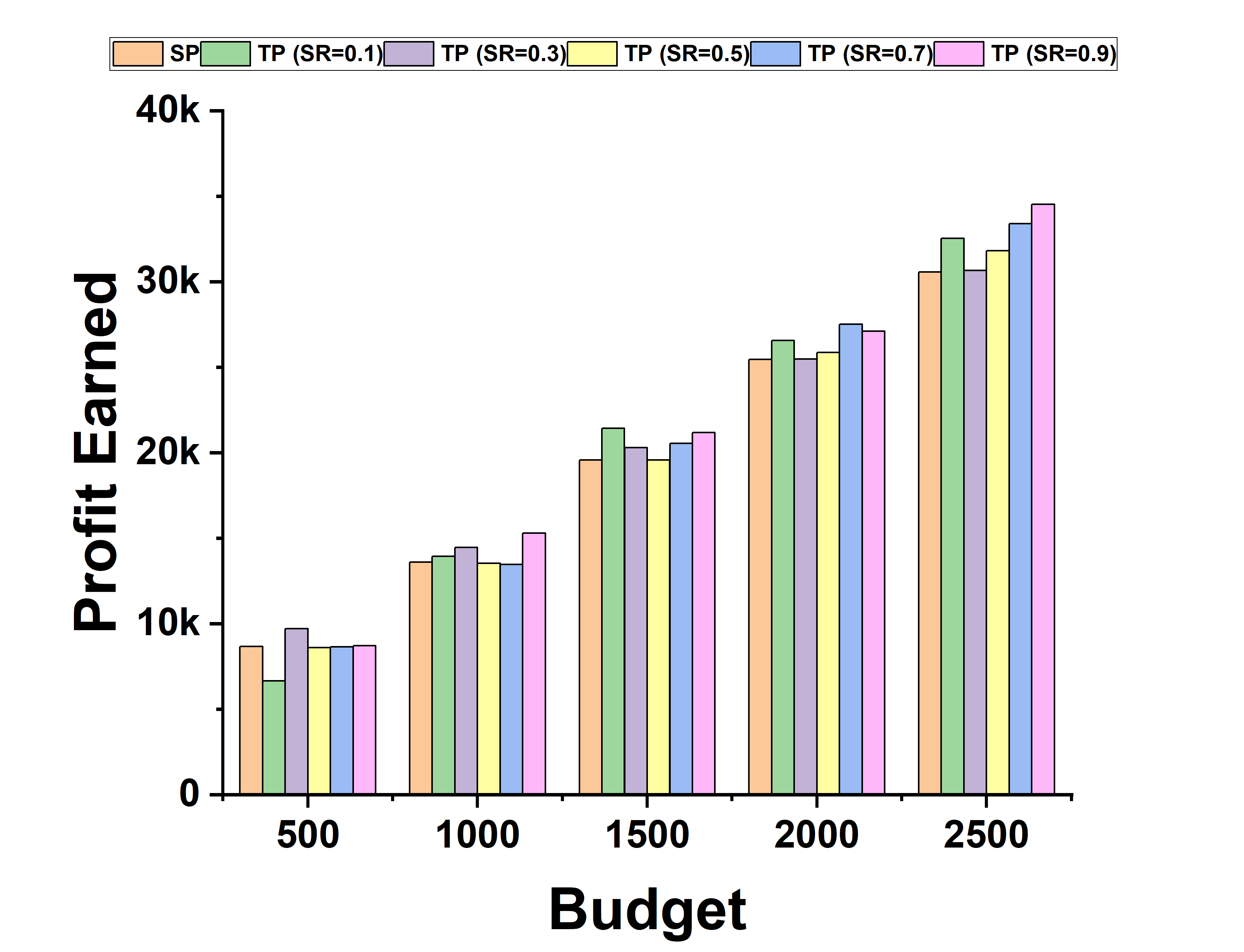}
        \caption{Degree Discount}
    \end{subfigure} \\[6pt]

    \begin{subfigure}[t]{0.22\textwidth}
        \includegraphics[width=\linewidth]{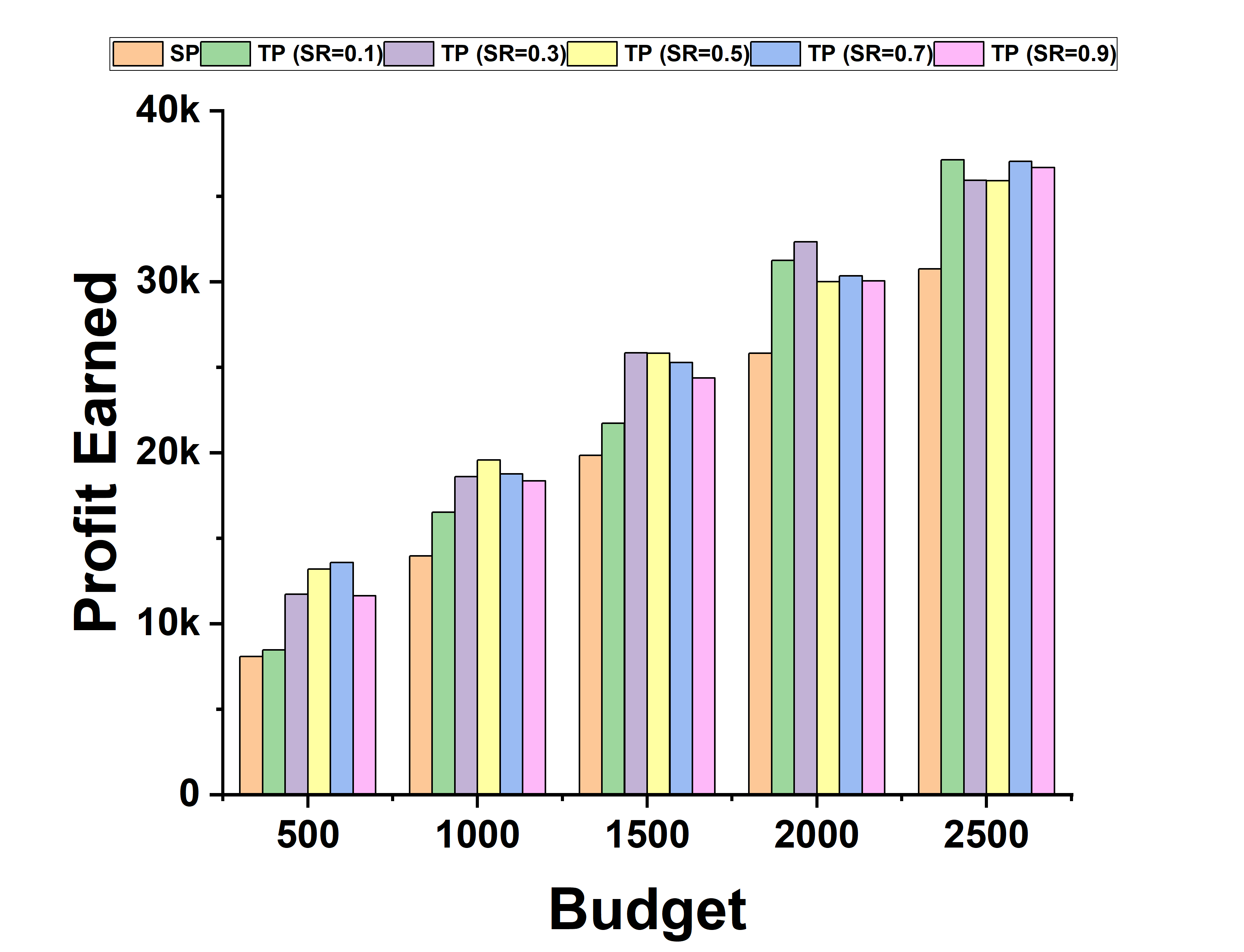}
        \caption{Single Discount}
    \end{subfigure} &
    \begin{subfigure}[t]{0.22\textwidth}
        \includegraphics[width=\linewidth]{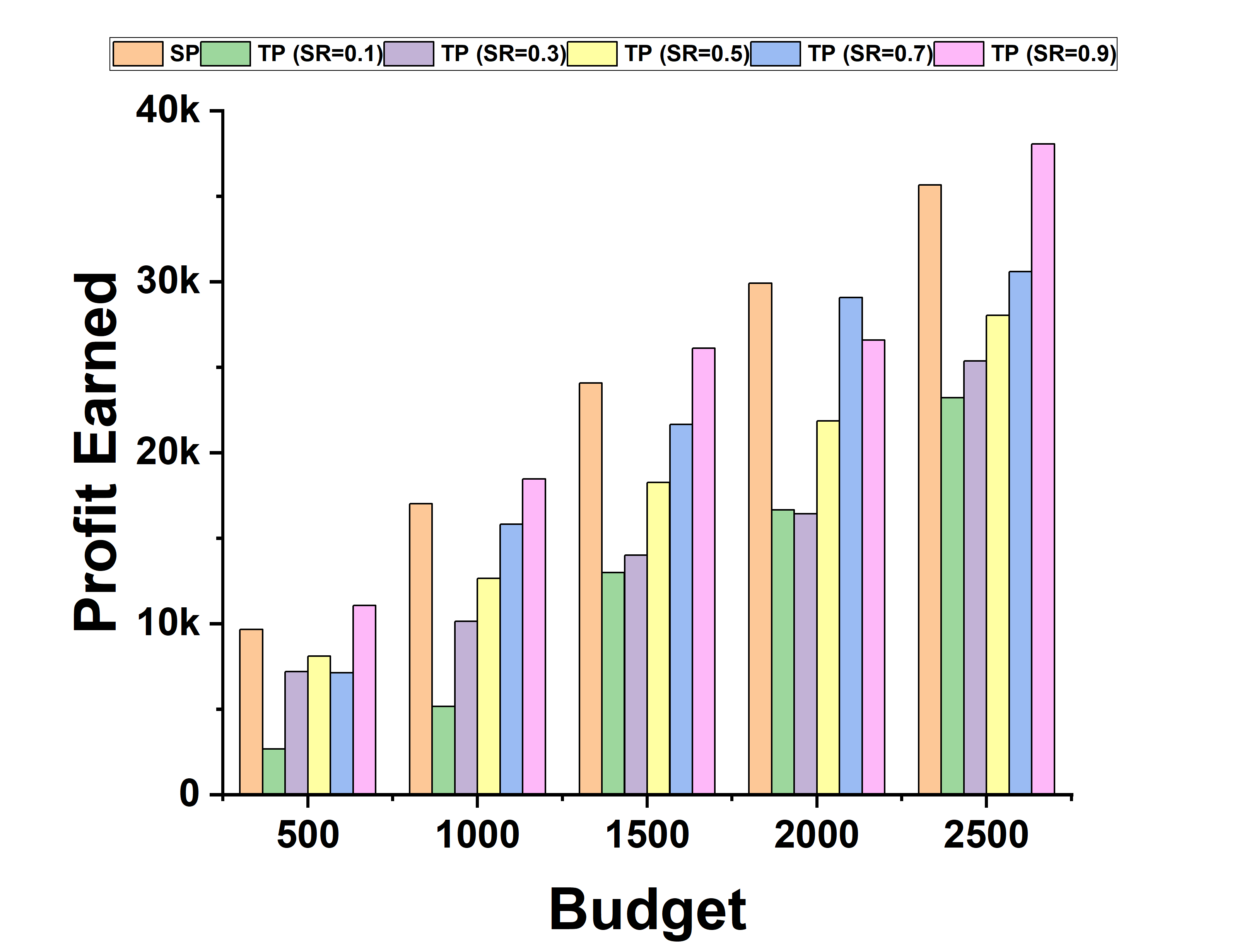}
        \caption{Simple Greedy}
    \end{subfigure} &
    \begin{subfigure}[t]{0.22\textwidth}
        \includegraphics[width=\linewidth]{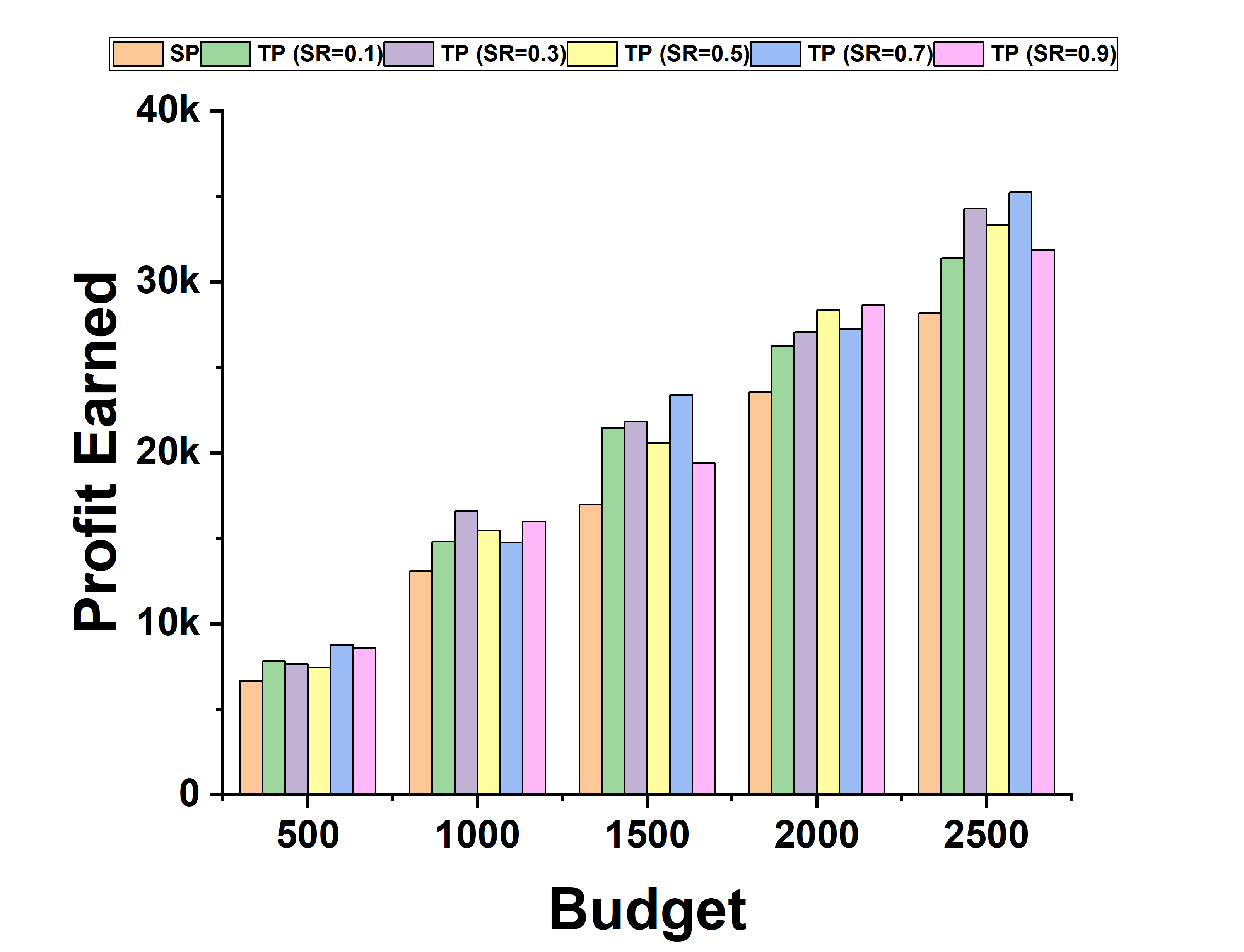}
        \caption{Double Greedy}
    \end{subfigure} &
    \begin{subfigure}[t]{0.22\textwidth}
        \includegraphics[width=\linewidth]{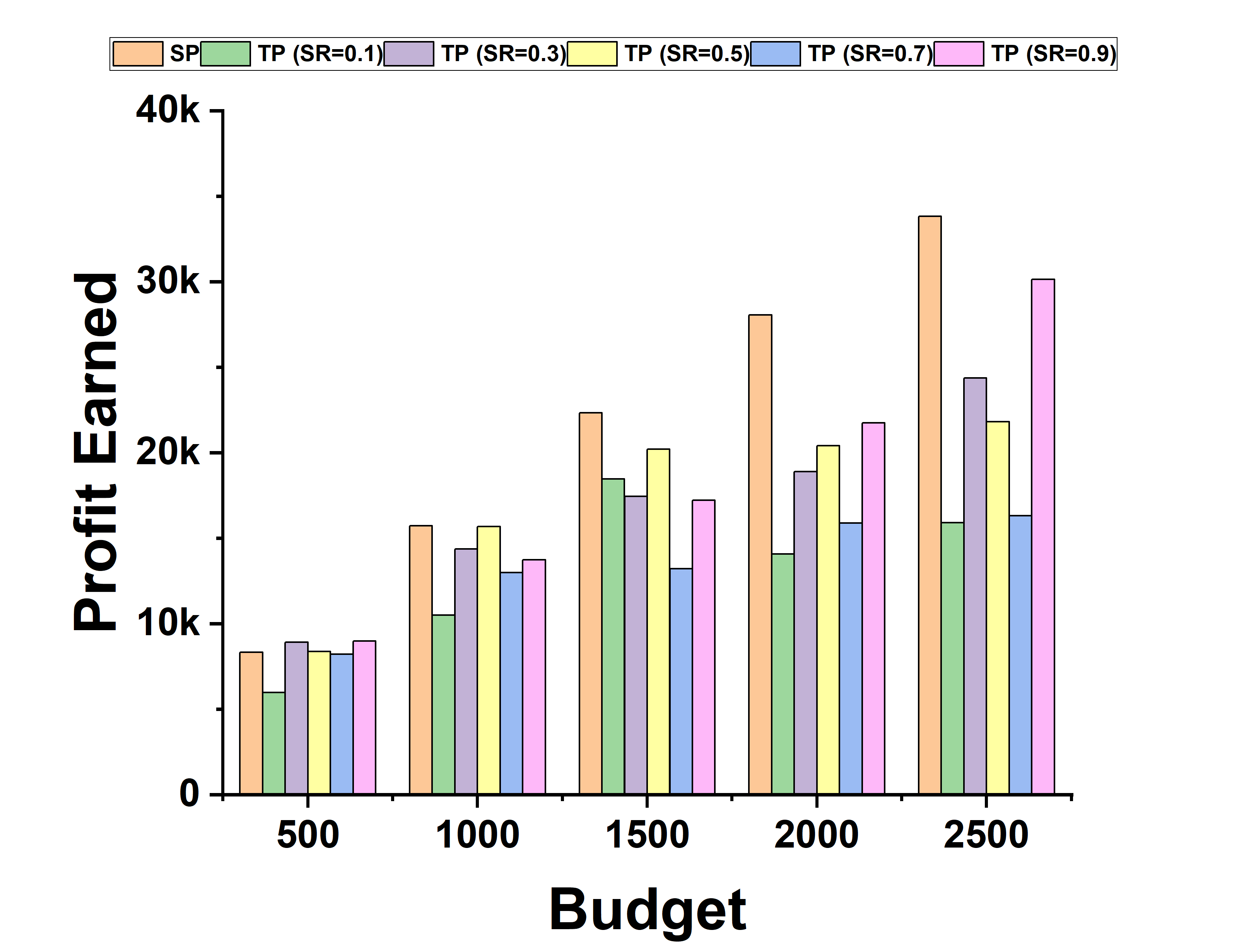}
        \caption{Stochastic Greedy}
    \end{subfigure}
\end{tabular}
\caption{Profit Earned in Single Phase Vs. Two Phase setting (Timestep 6, Probability Setting - Trivalency, \textit{LM} Dataset)}
\label{Fig:RQ2LM_T3}
\end{figure}

\begin{figure}[htbp]
\centering
\captionsetup[sub]{font=footnotesize}
\begin{tabular}{cccc}
    \begin{subfigure}[t]{0.22\textwidth}
        \includegraphics[width=\linewidth]{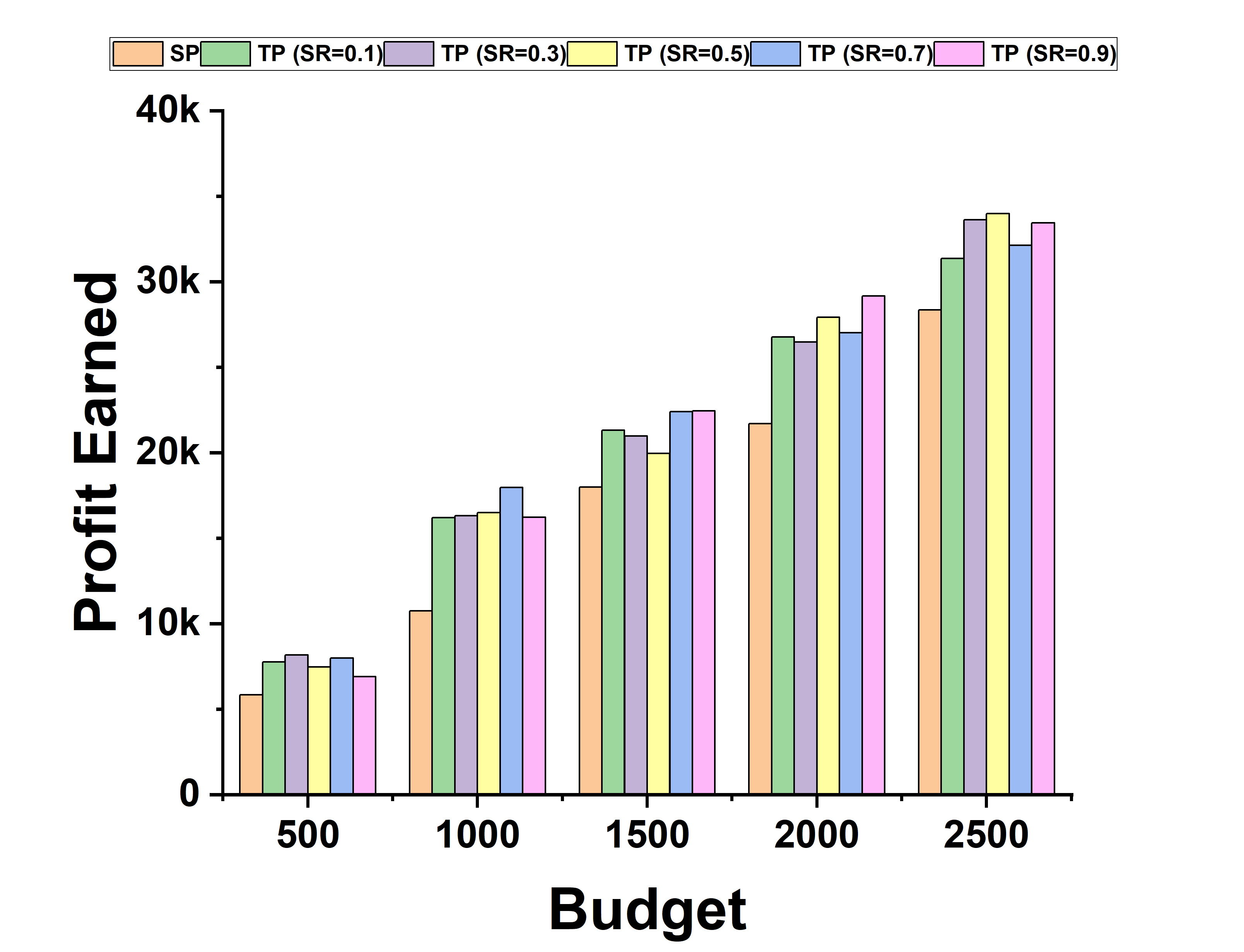}
        \caption{Random}
    \end{subfigure} &
    \begin{subfigure}[t]{0.22\textwidth}
        \includegraphics[width=\linewidth]{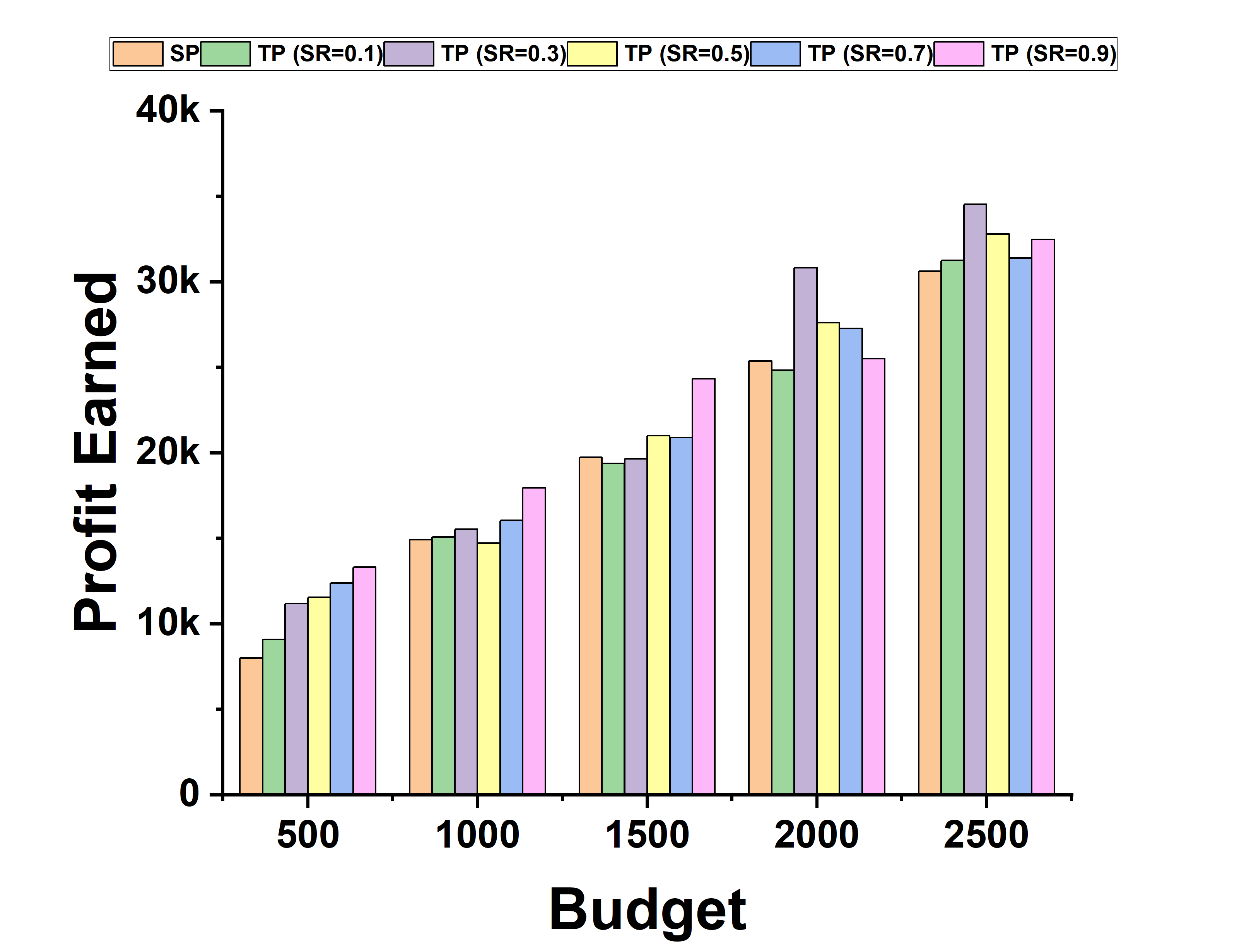}
        \caption{High Degree}
    \end{subfigure} &
    \begin{subfigure}[t]{0.22\textwidth}
        \includegraphics[width=\linewidth]{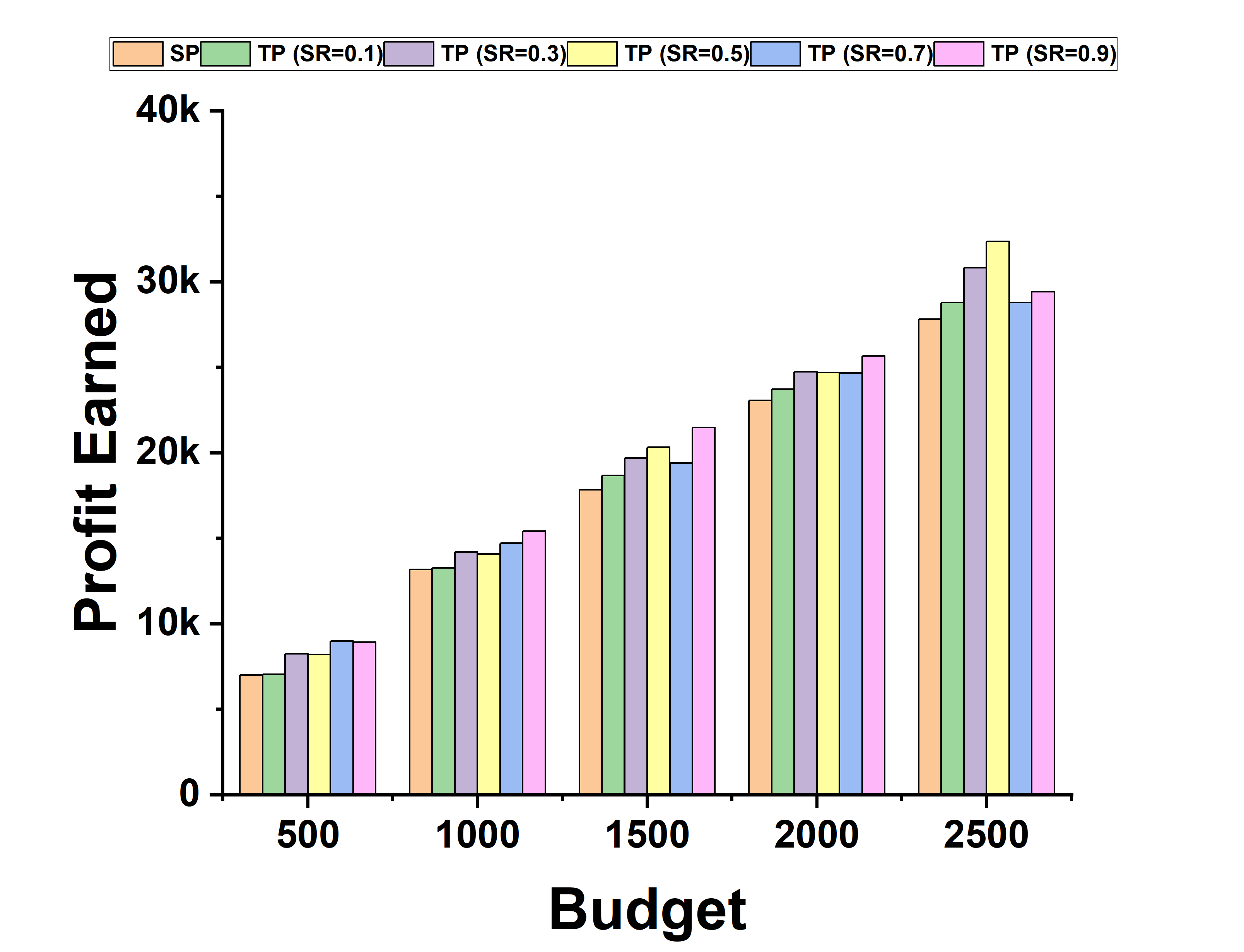}
        \caption{Clustering\\Coefficient}
    \end{subfigure} &
    \begin{subfigure}[t]{0.22\textwidth}
        \includegraphics[width=\linewidth]{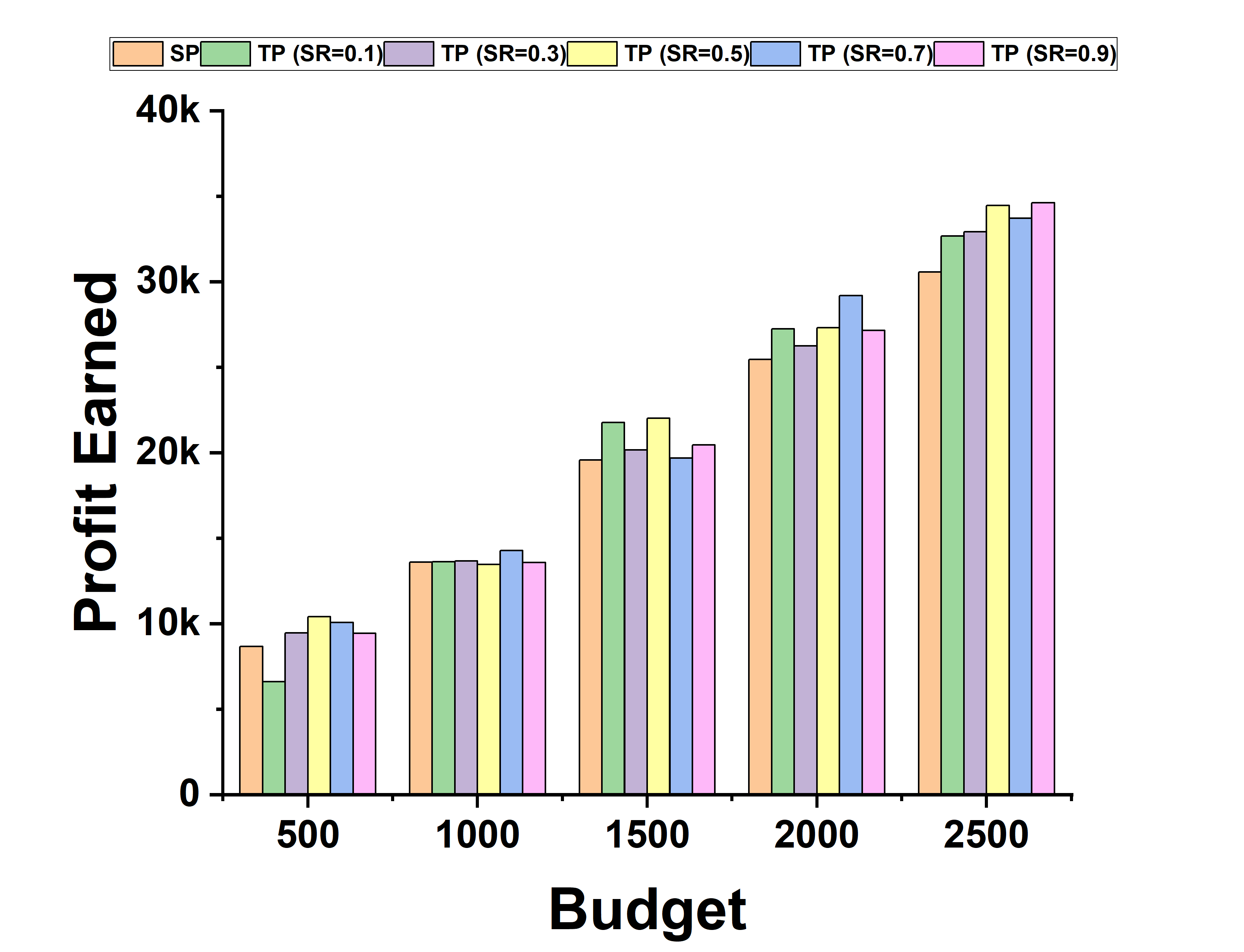}
        \caption{Degree Discount}
    \end{subfigure} \\[6pt]

    \begin{subfigure}[t]{0.22\textwidth}
        \includegraphics[width=\linewidth]{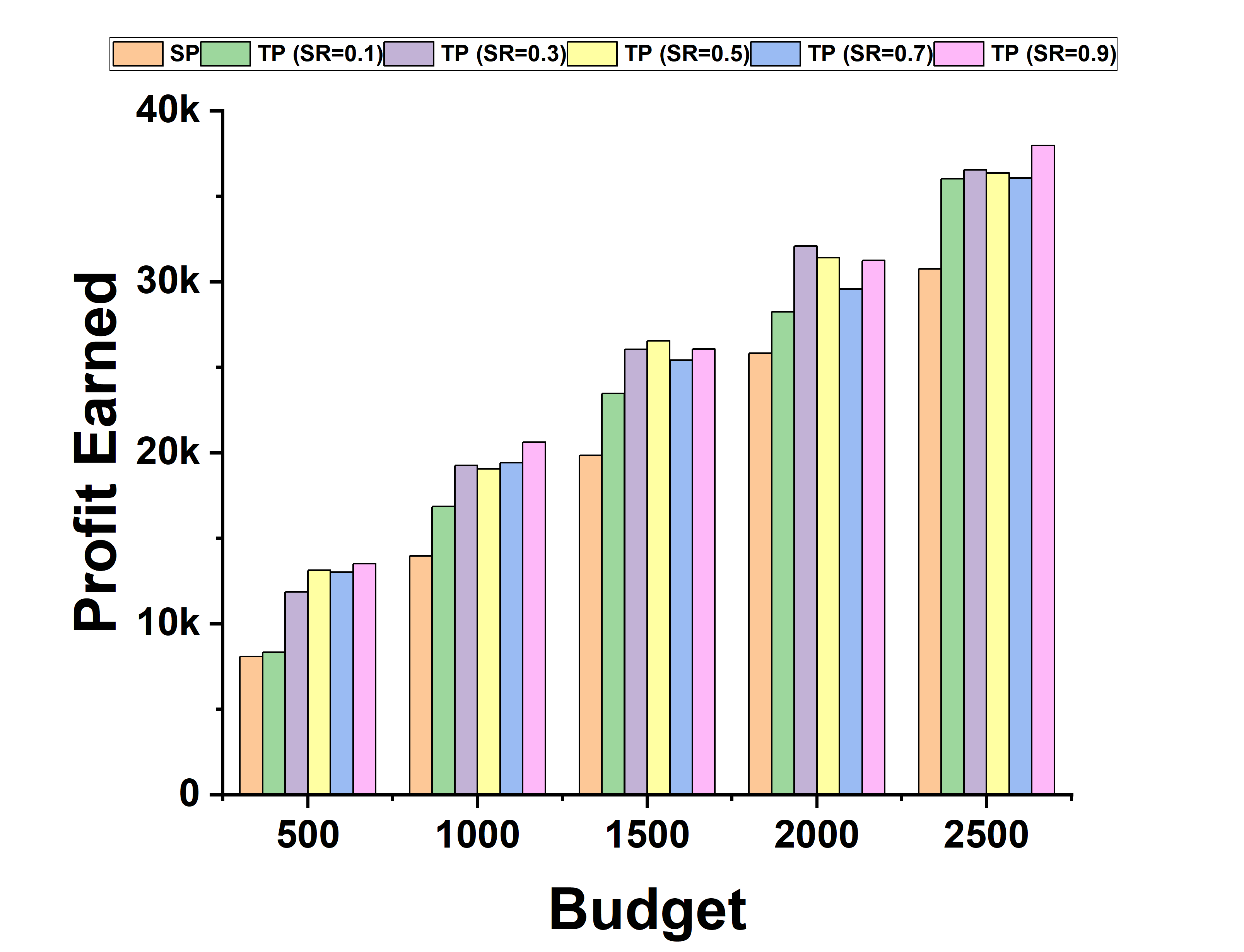}
        \caption{Single Discount}
    \end{subfigure} &
    \begin{subfigure}[t]{0.22\textwidth}
        \includegraphics[width=\linewidth]{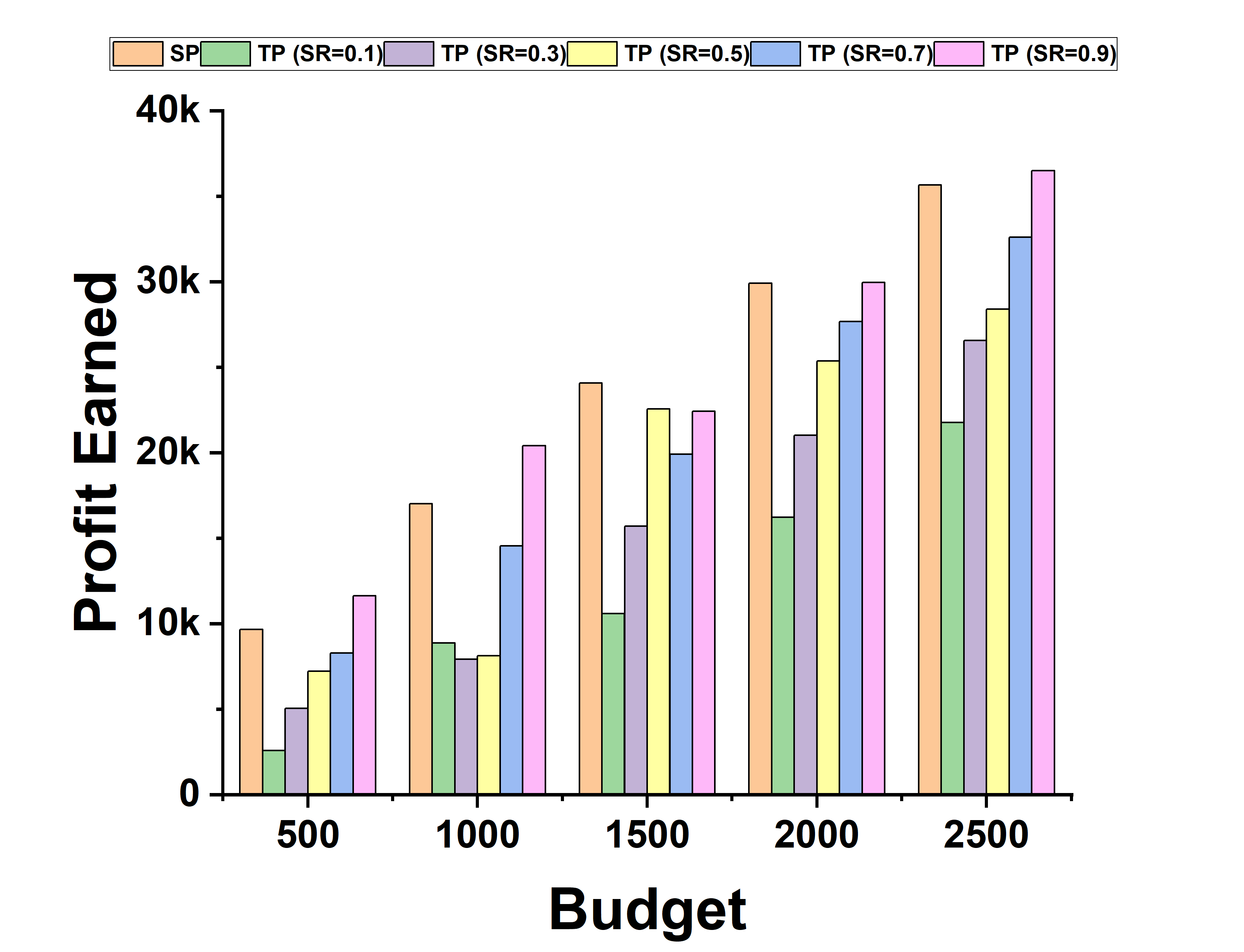}
        \caption{Simple Greedy}
    \end{subfigure} &
    \begin{subfigure}[t]{0.22\textwidth}
        \includegraphics[width=\linewidth]{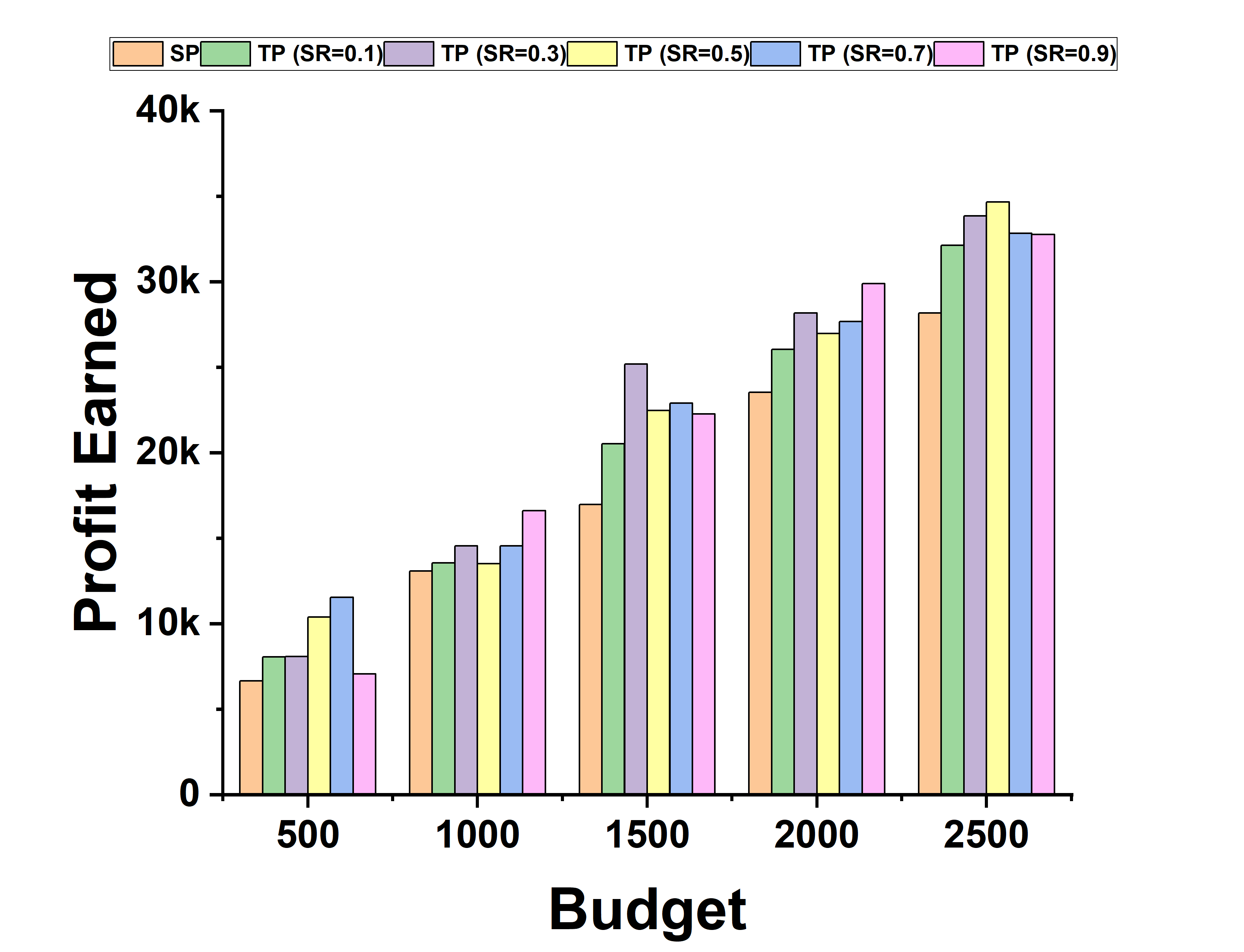}
        \caption{Double Greedy}
    \end{subfigure} &
    \begin{subfigure}[t]{0.22\textwidth}
        \includegraphics[width=\linewidth]{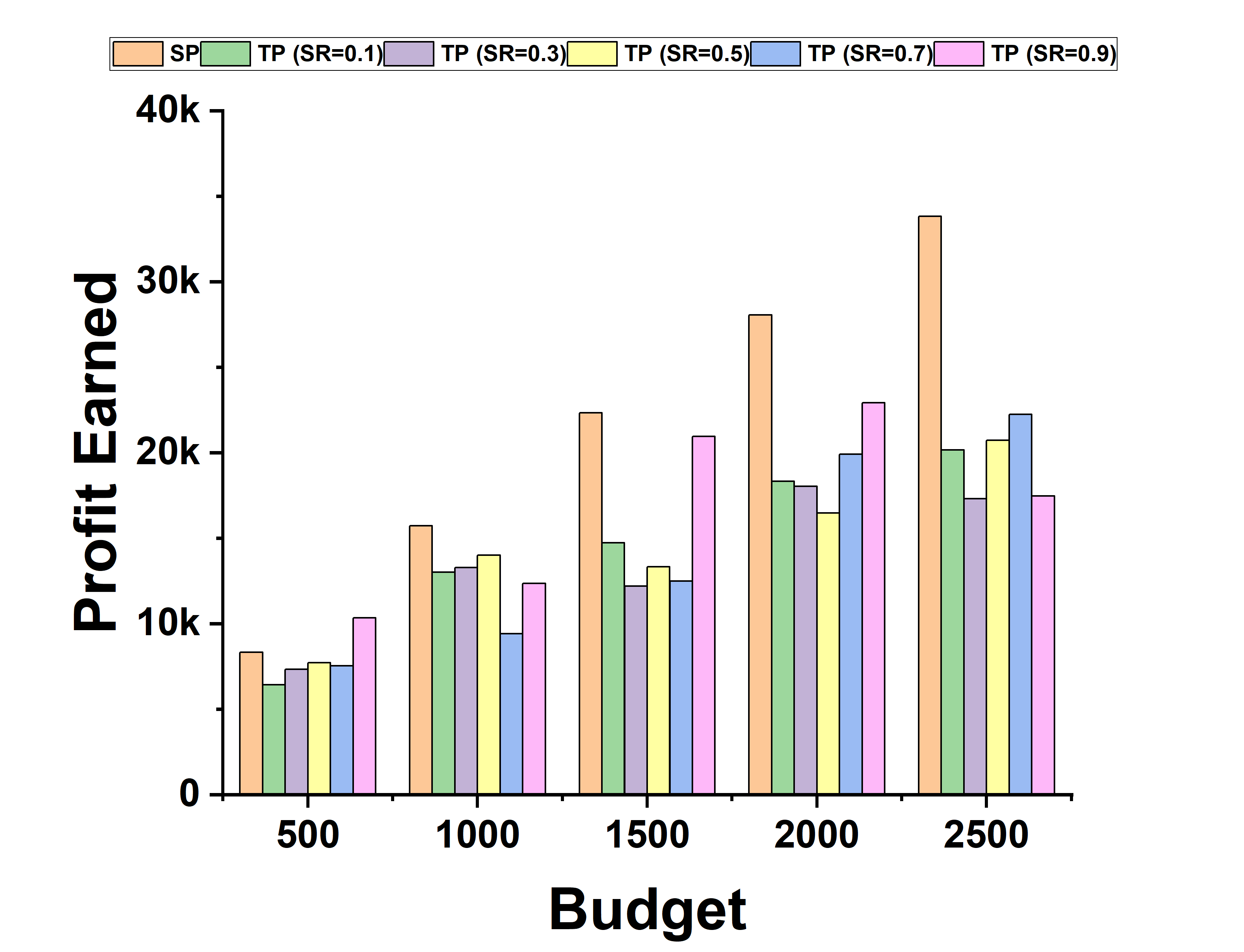}
        \caption{Stochastic Greedy}
    \end{subfigure}
\end{tabular}
\caption{Profit Earned in Single Phase Vs. Two Phase setting (Timestep 8, Probability Setting - Trivalency, \textit{LM} Dataset)}
\label{Fig:RQ2LM_T4}
\end{figure}

\begin{figure}[htbp]
\centering
\captionsetup[sub]{font=footnotesize}
\begin{tabular}{cccc}
    \begin{subfigure}[t]{0.22\textwidth}
        \includegraphics[width=\linewidth]{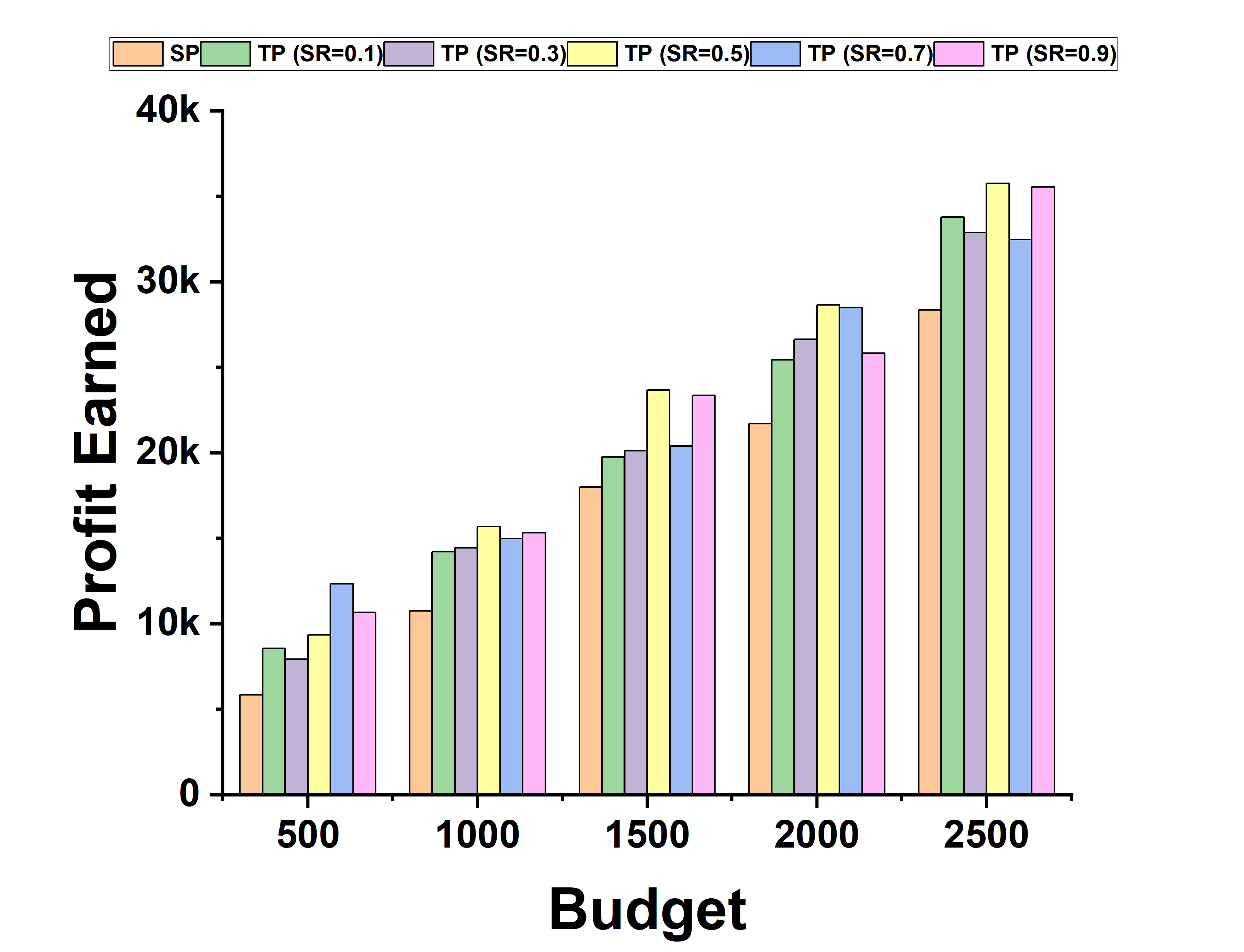}
        \caption{Random}
    \end{subfigure} &
    \begin{subfigure}[t]{0.22\textwidth}
        \includegraphics[width=\linewidth]{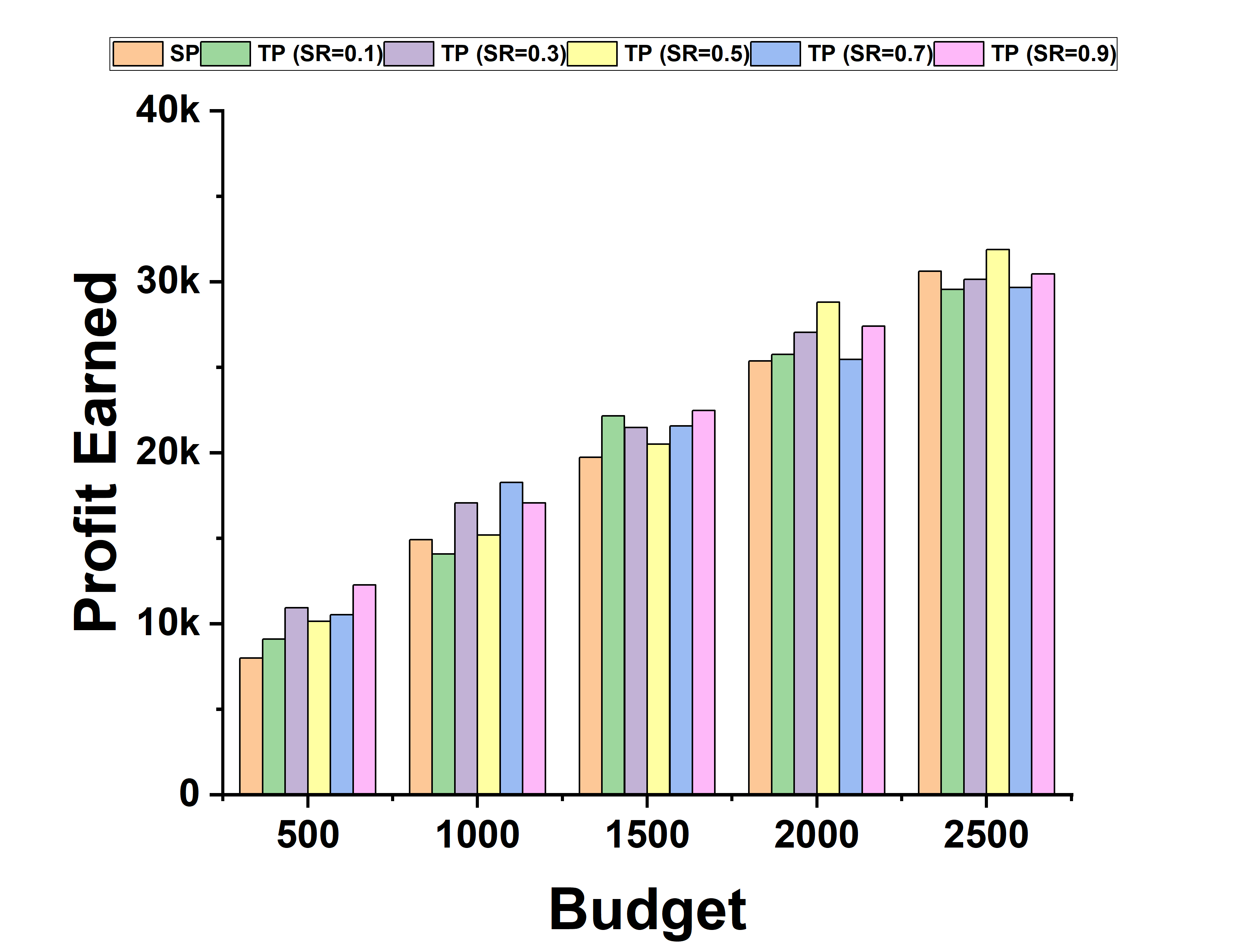}
        \caption{High Degree}
    \end{subfigure} &
    \begin{subfigure}[t]{0.22\textwidth}
        \includegraphics[width=\linewidth]{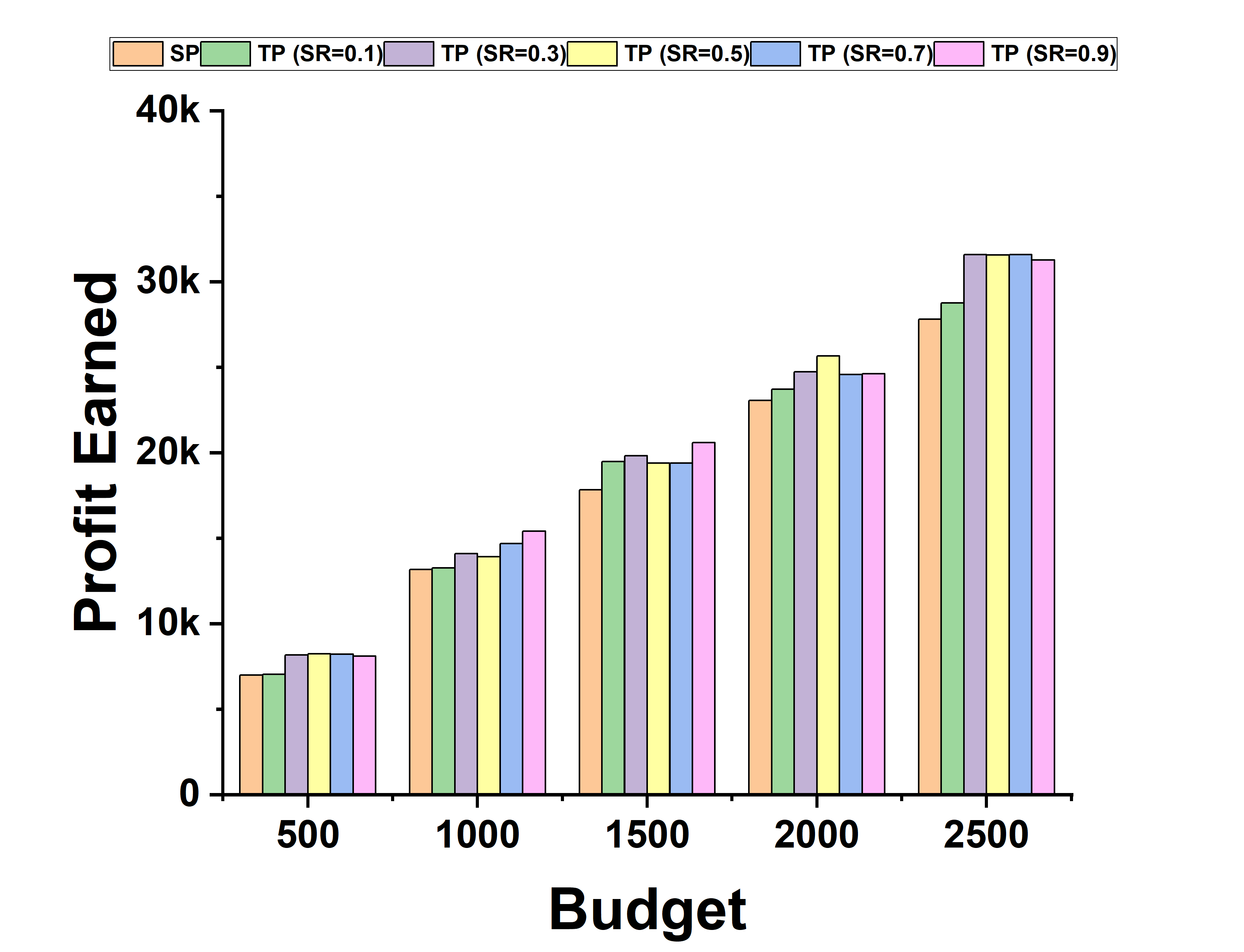}
        \caption{Clustering\\Coefficient}
    \end{subfigure} &
    \begin{subfigure}[t]{0.22\textwidth}
        \includegraphics[width=\linewidth]{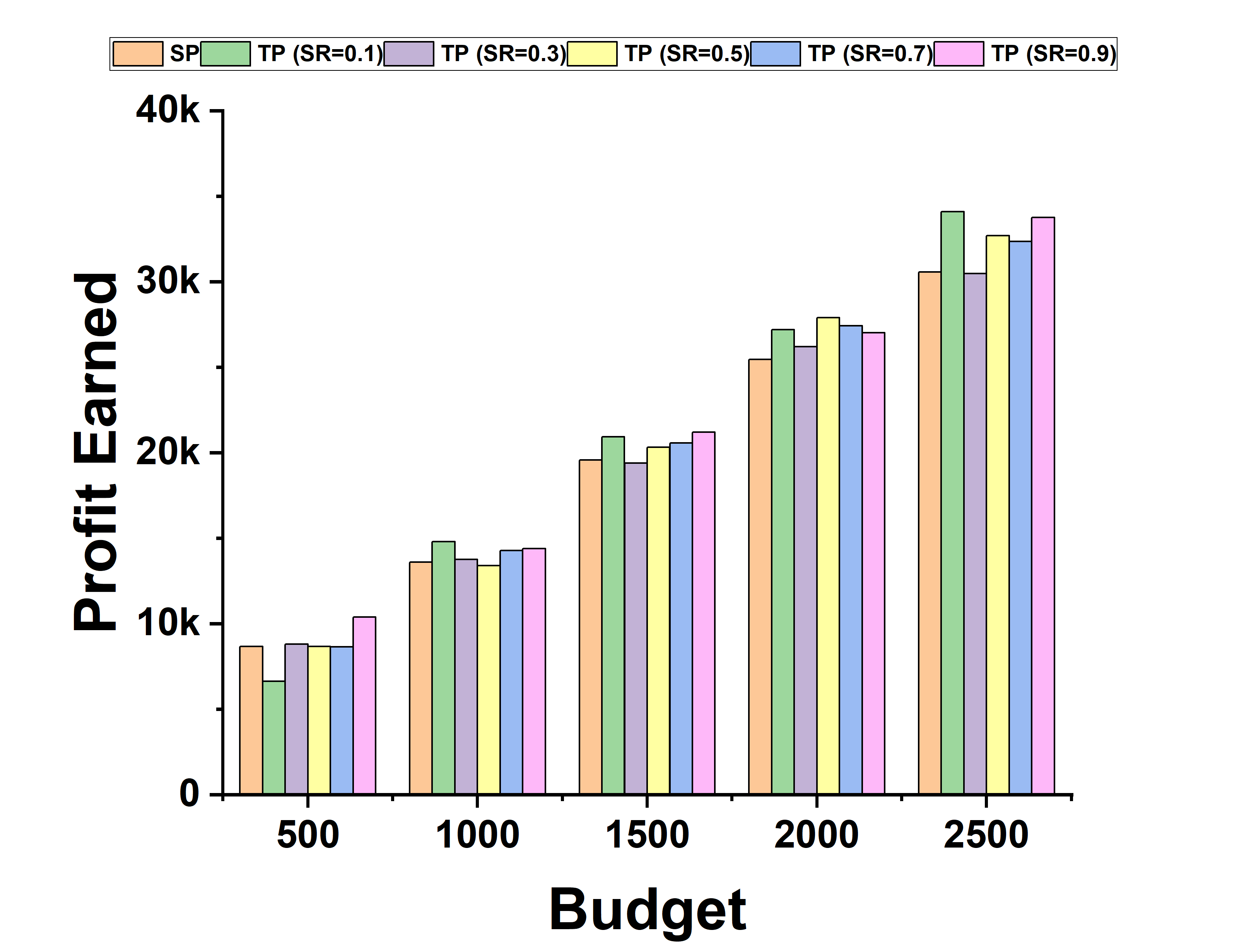}
        \caption{Degree Discount}
    \end{subfigure} \\[6pt]

    \begin{subfigure}[t]{0.22\textwidth}
        \includegraphics[width=\linewidth]{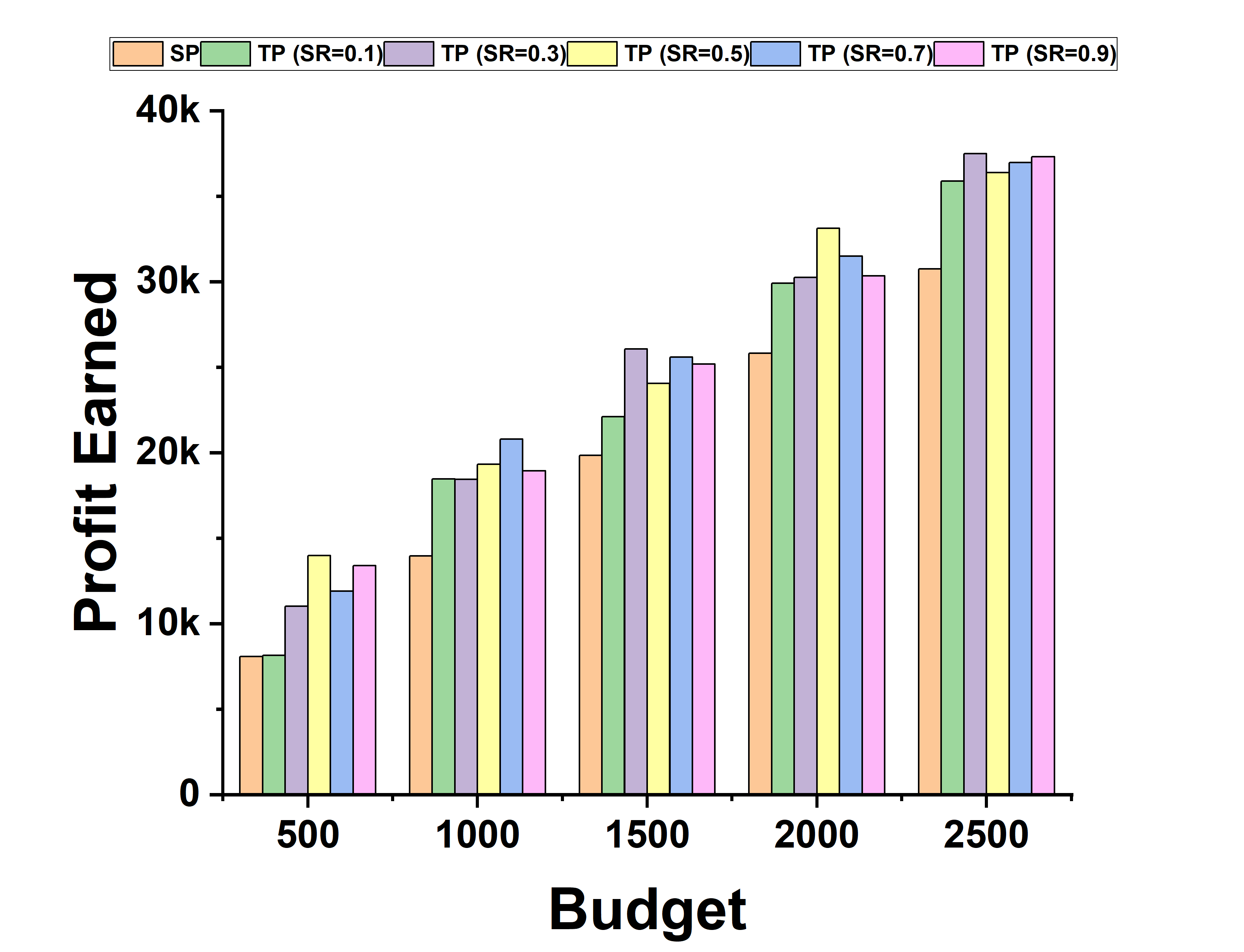}
        \caption{Single Discount}
    \end{subfigure} &
    \begin{subfigure}[t]{0.22\textwidth}
        \includegraphics[width=\linewidth]{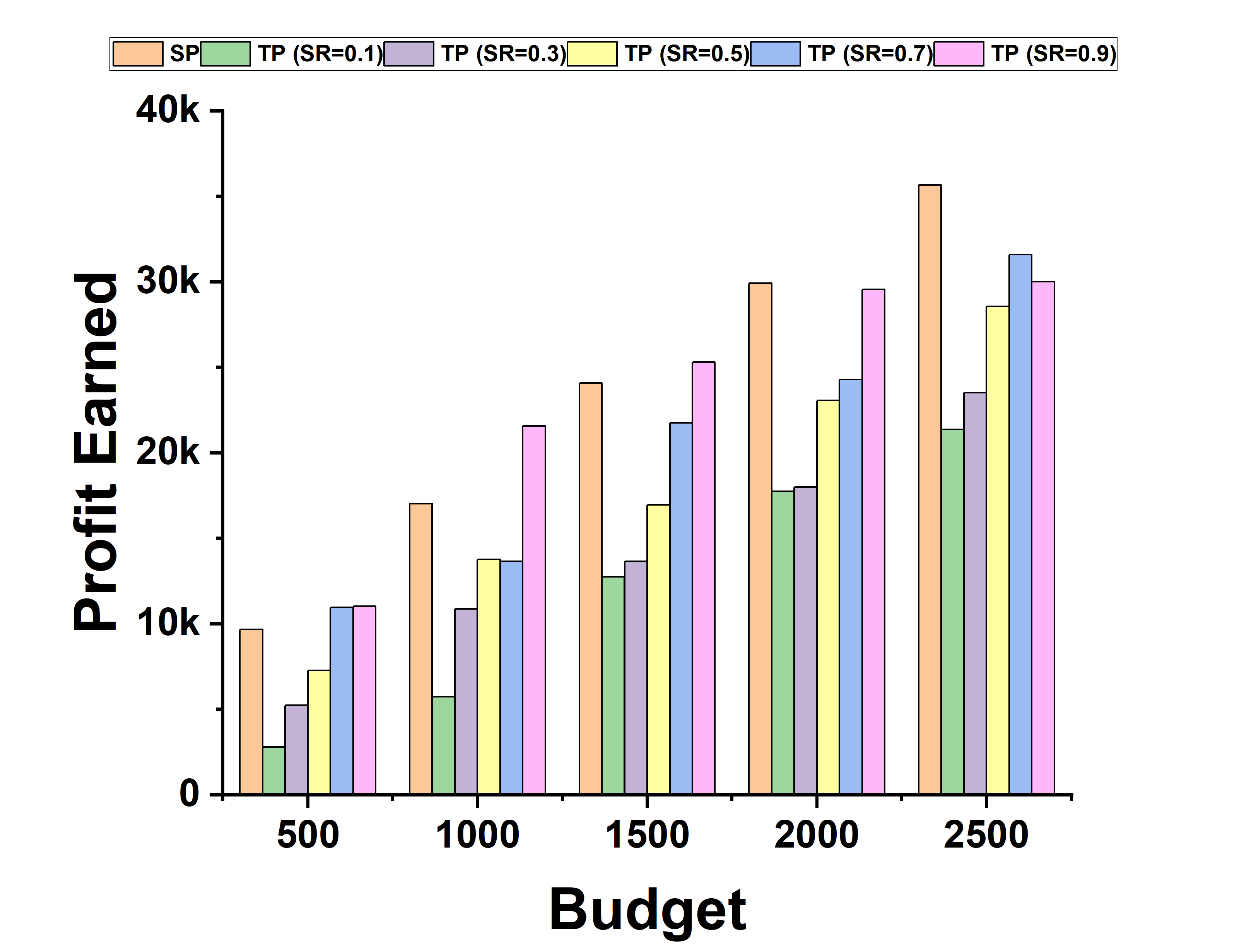}
        \caption{Simple Greedy}
    \end{subfigure} &
    \begin{subfigure}[t]{0.22\textwidth}
        \includegraphics[width=\linewidth]{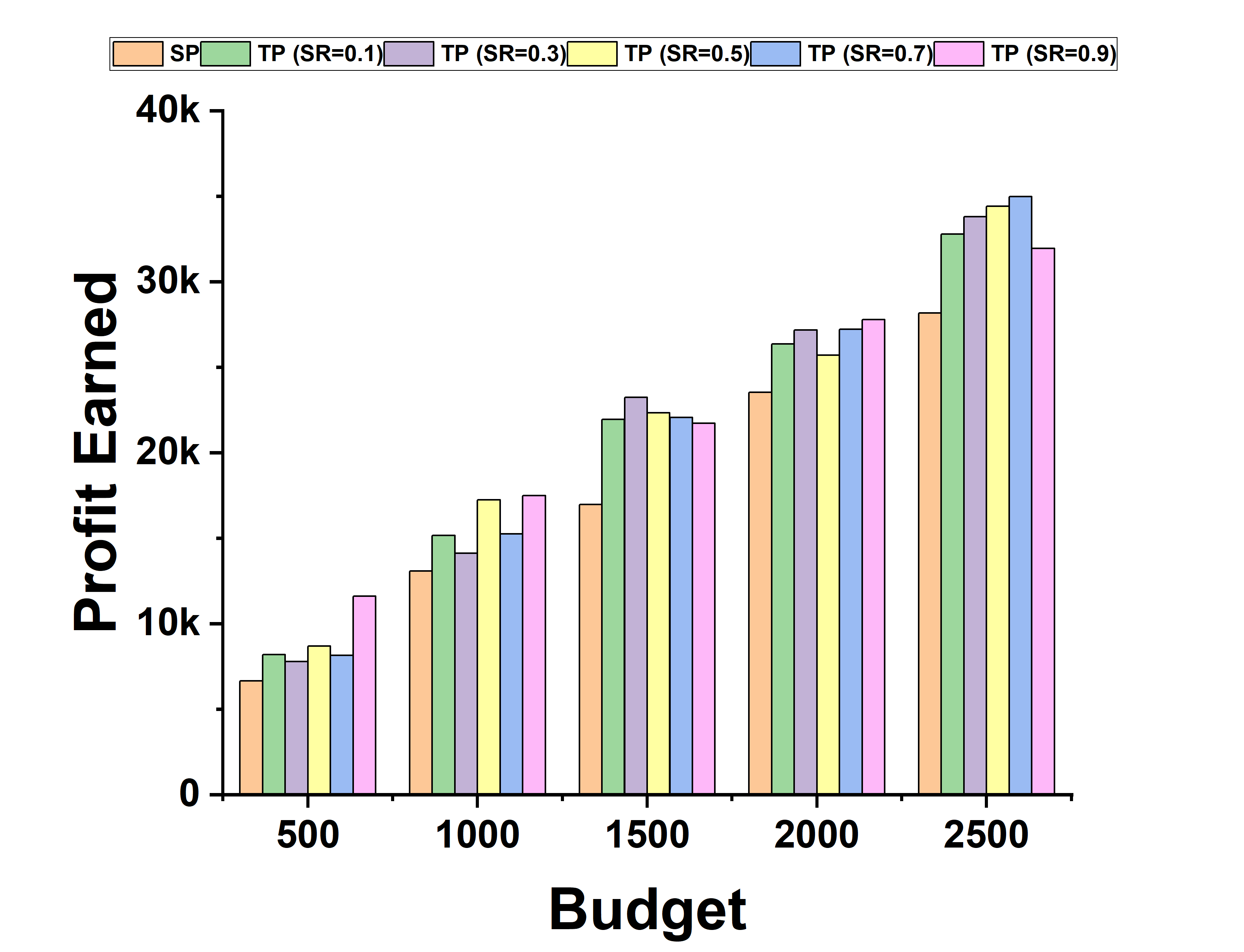}
        \caption{Double Greedy}
    \end{subfigure} &
    \begin{subfigure}[t]{0.22\textwidth}
        \includegraphics[width=\linewidth]{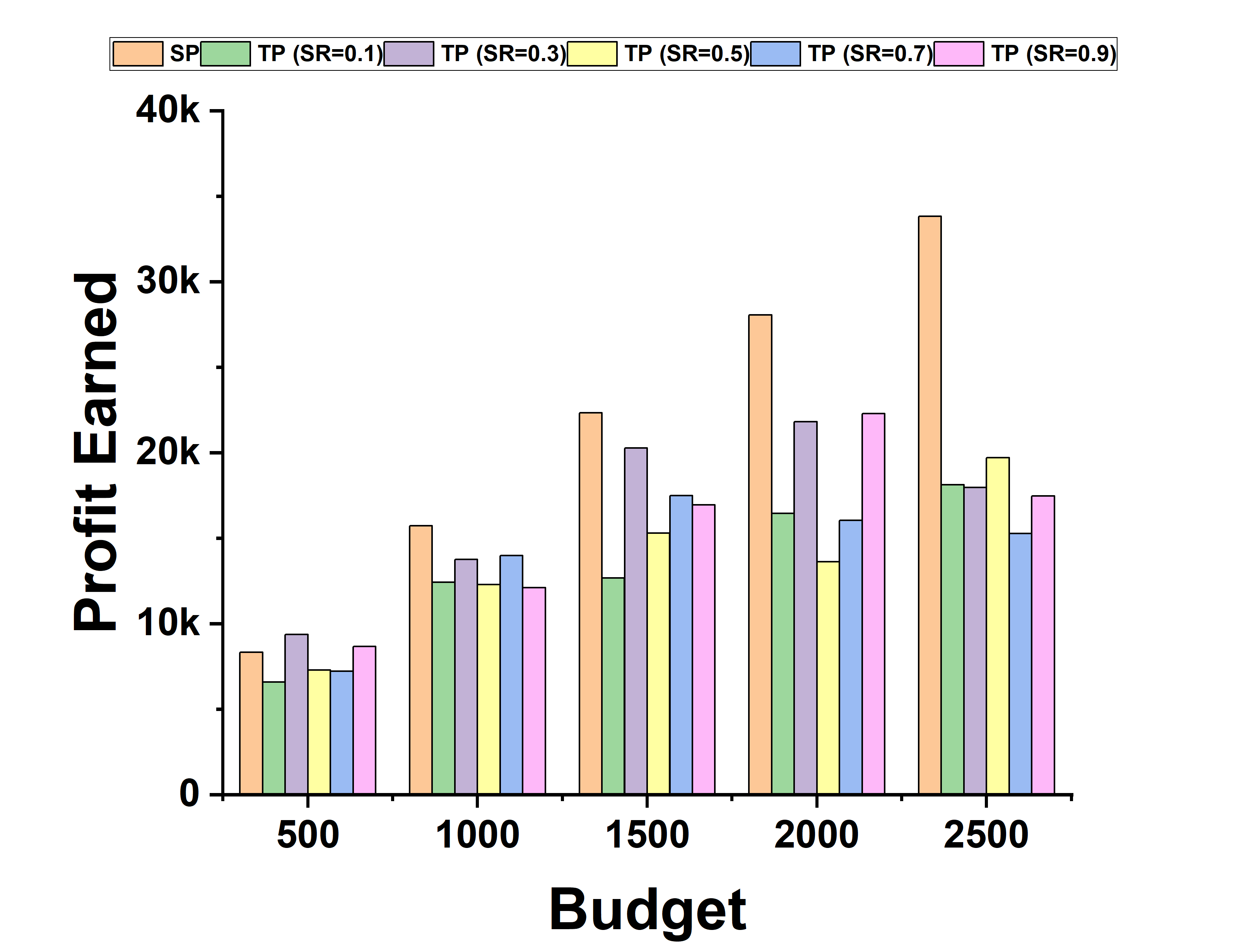}
        \caption{Stochastic Greedy}
    \end{subfigure}
\end{tabular}
\caption{Profit Earned in Single Phase Vs. Two Phase setting (Timestep 10, Probability Setting - Trivalency, \textit{LM} Dataset)}
\label{Fig:RQ2LM_T5}
\end{figure}


\begin{figure}[htbp]
\centering
\captionsetup[sub]{font=footnotesize}  
\begin{tabular}{cccc}
    \begin{subfigure}[t]{0.22\textwidth}
        \includegraphics[width=\linewidth]{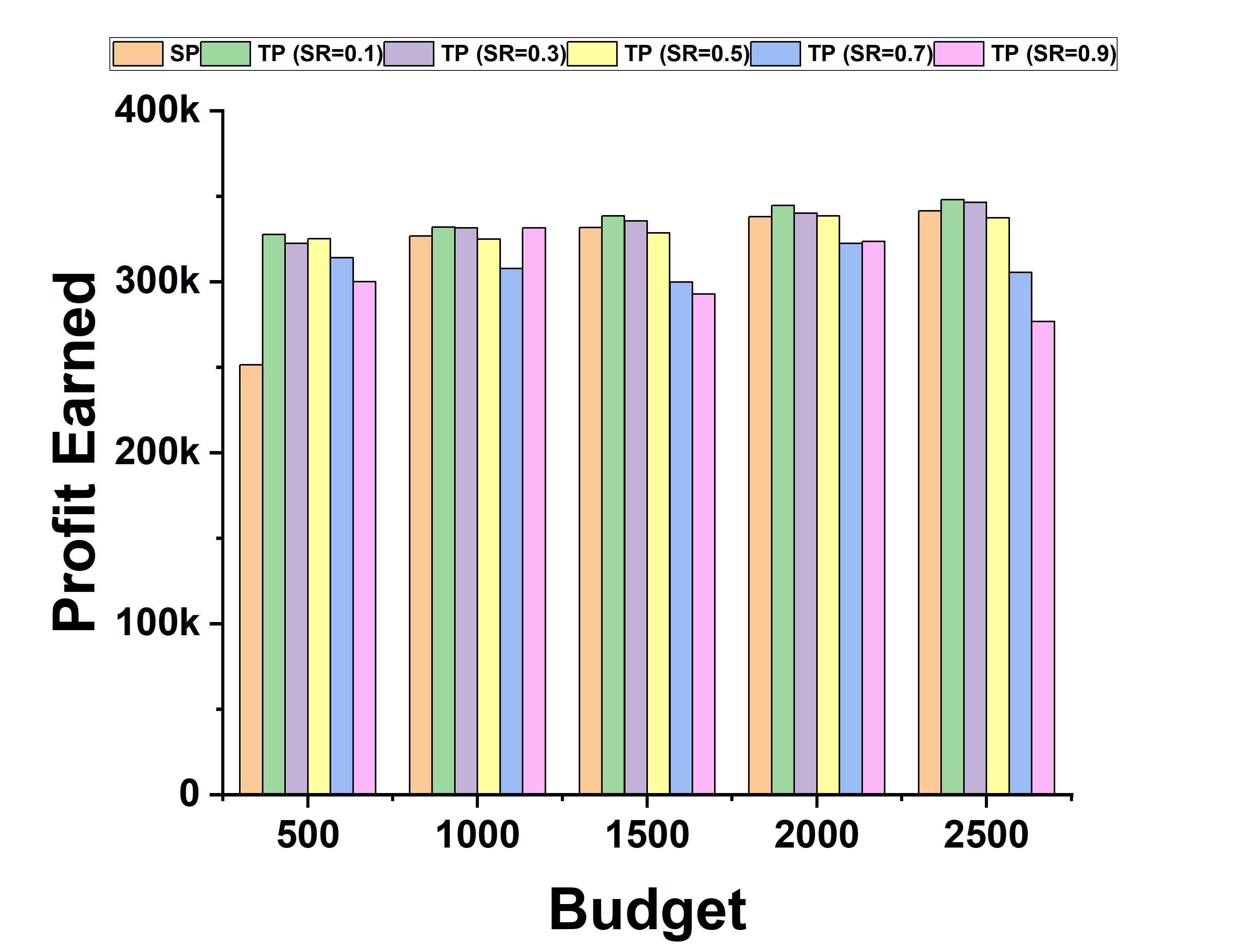}
        \caption{Random}
    \end{subfigure} &
    \begin{subfigure}[t]{0.22\textwidth}
        \includegraphics[width=\linewidth]{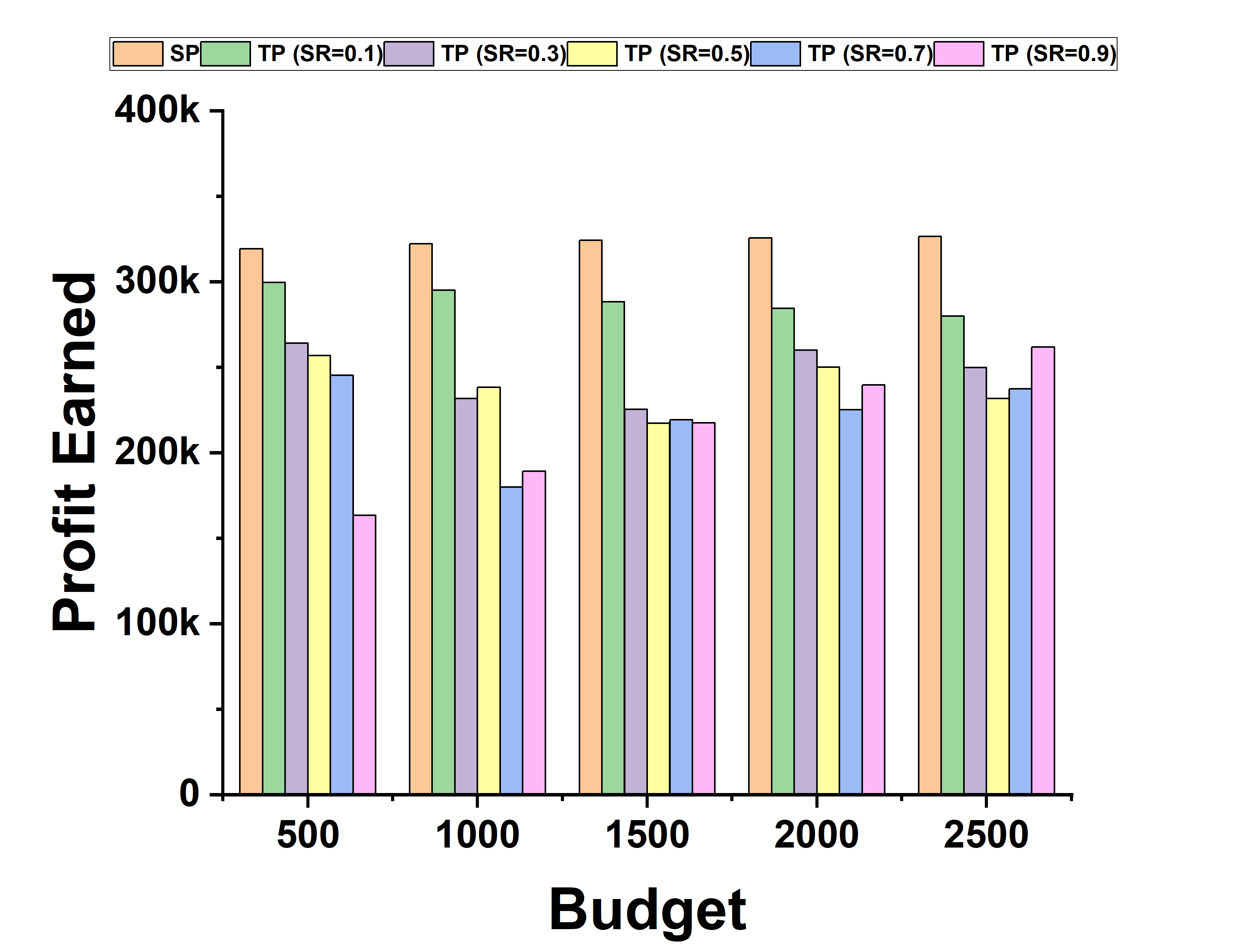}
        \caption{High Degree}
    \end{subfigure} &
    \begin{subfigure}[t]{0.22\textwidth}
        \includegraphics[width=\linewidth]{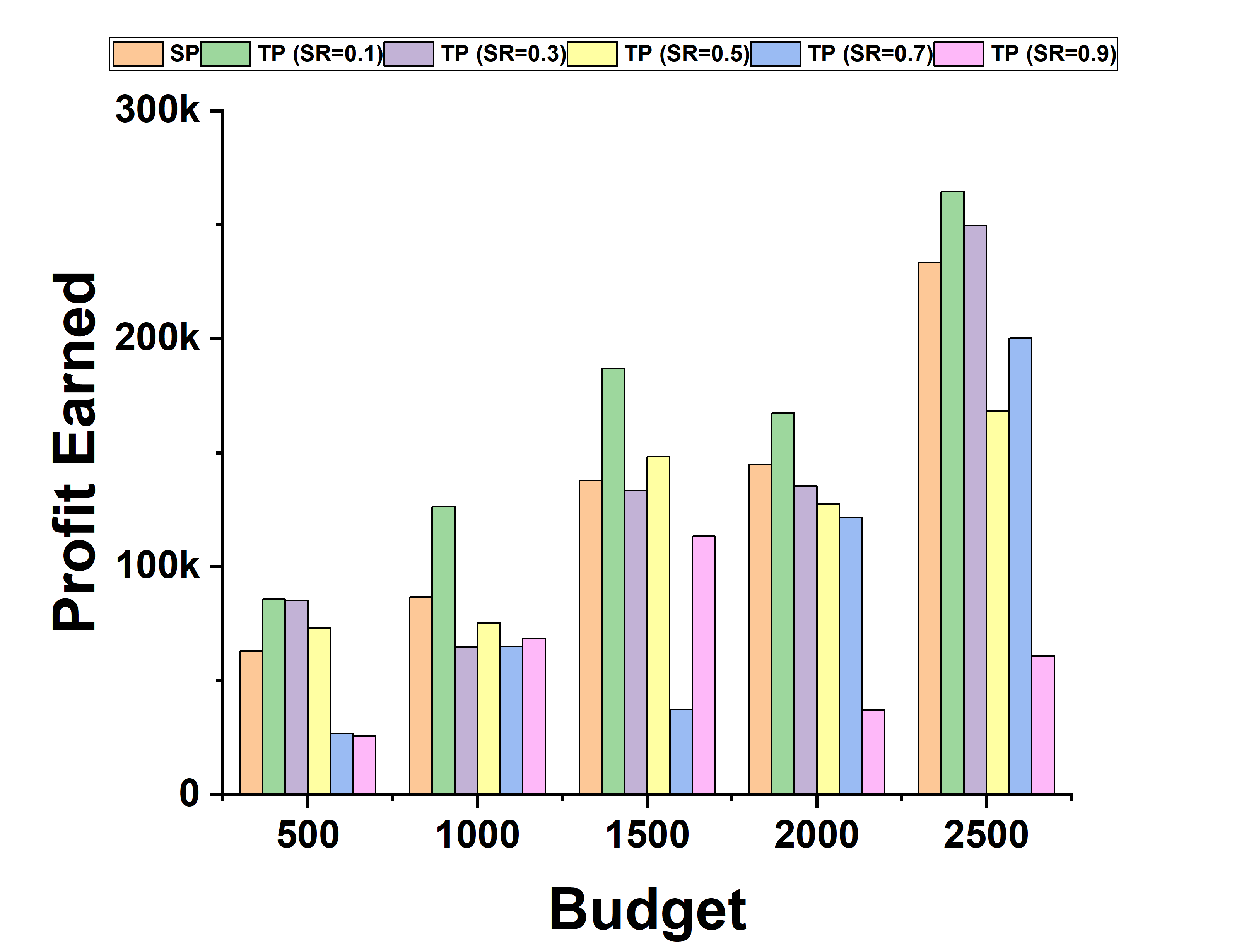}
        \caption{Clustering\\Coefficient}
    \end{subfigure} &
    \begin{subfigure}[t]{0.22\textwidth}
        \includegraphics[width=\linewidth]{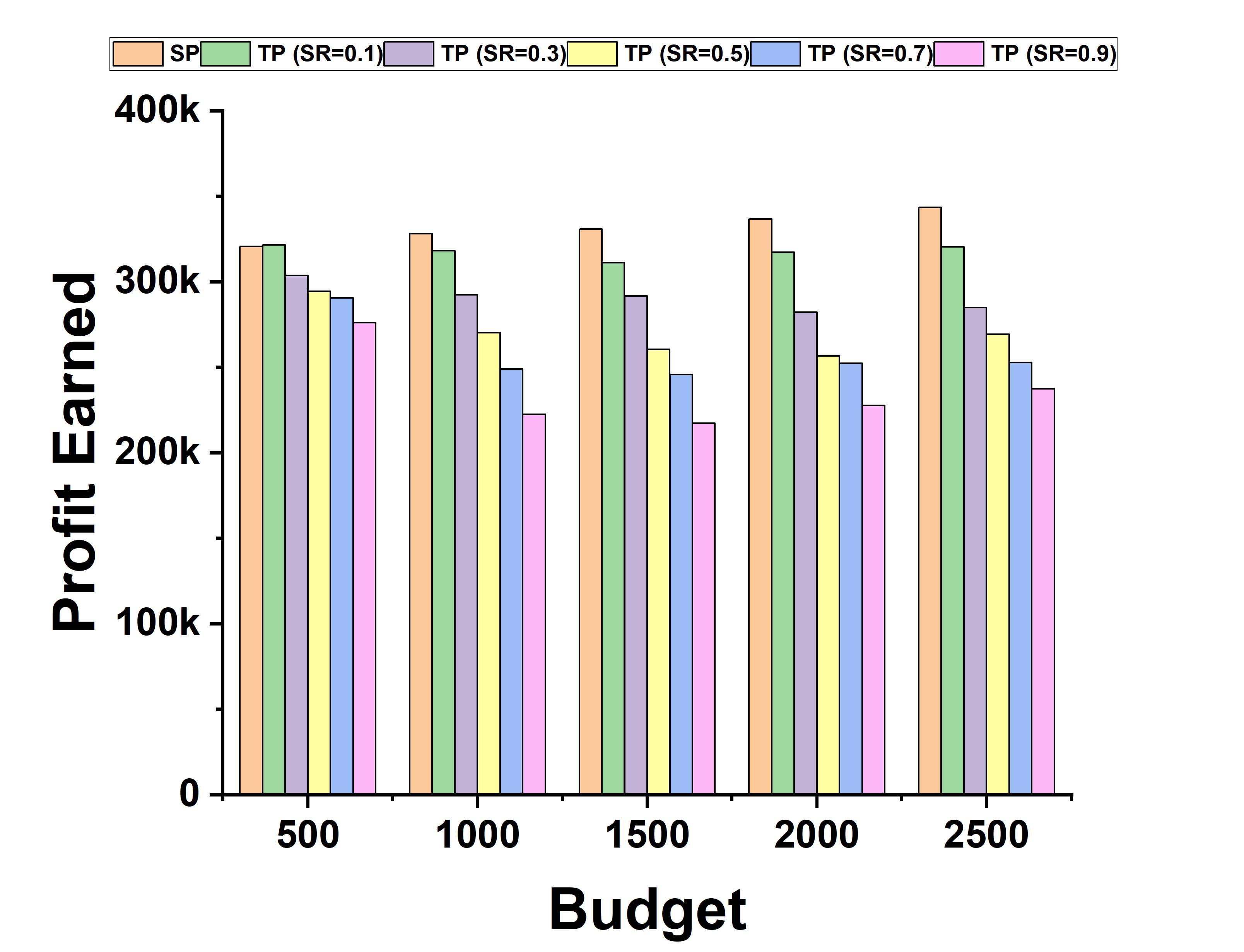}
        \caption{Degree Discount}
    \end{subfigure} \\[6pt]

    \begin{subfigure}[t]{0.22\textwidth}
        \includegraphics[width=\linewidth]{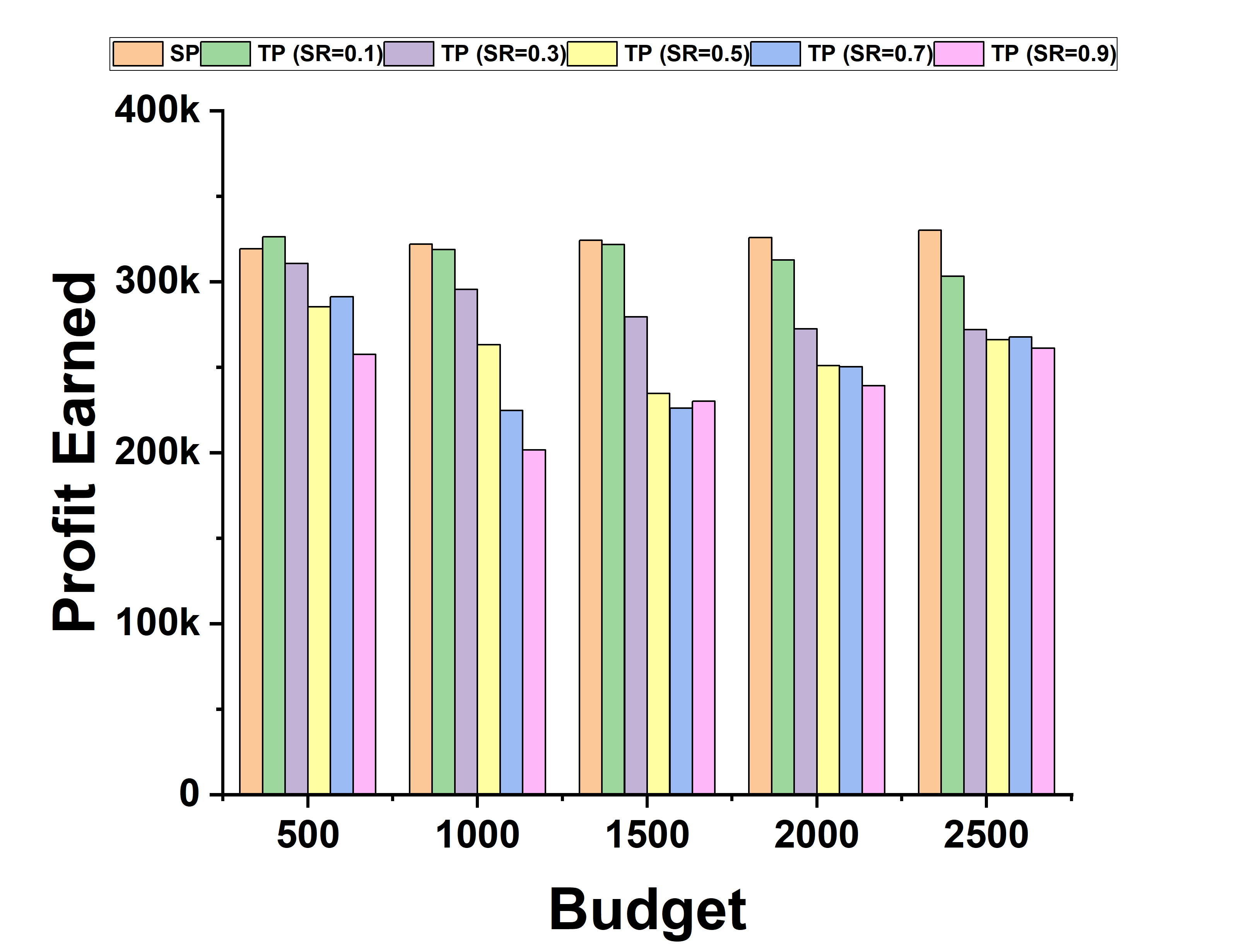}
        \caption{Single Discount}
    \end{subfigure} &
    \begin{subfigure}[t]{0.22\textwidth}
        \includegraphics[width=\linewidth]{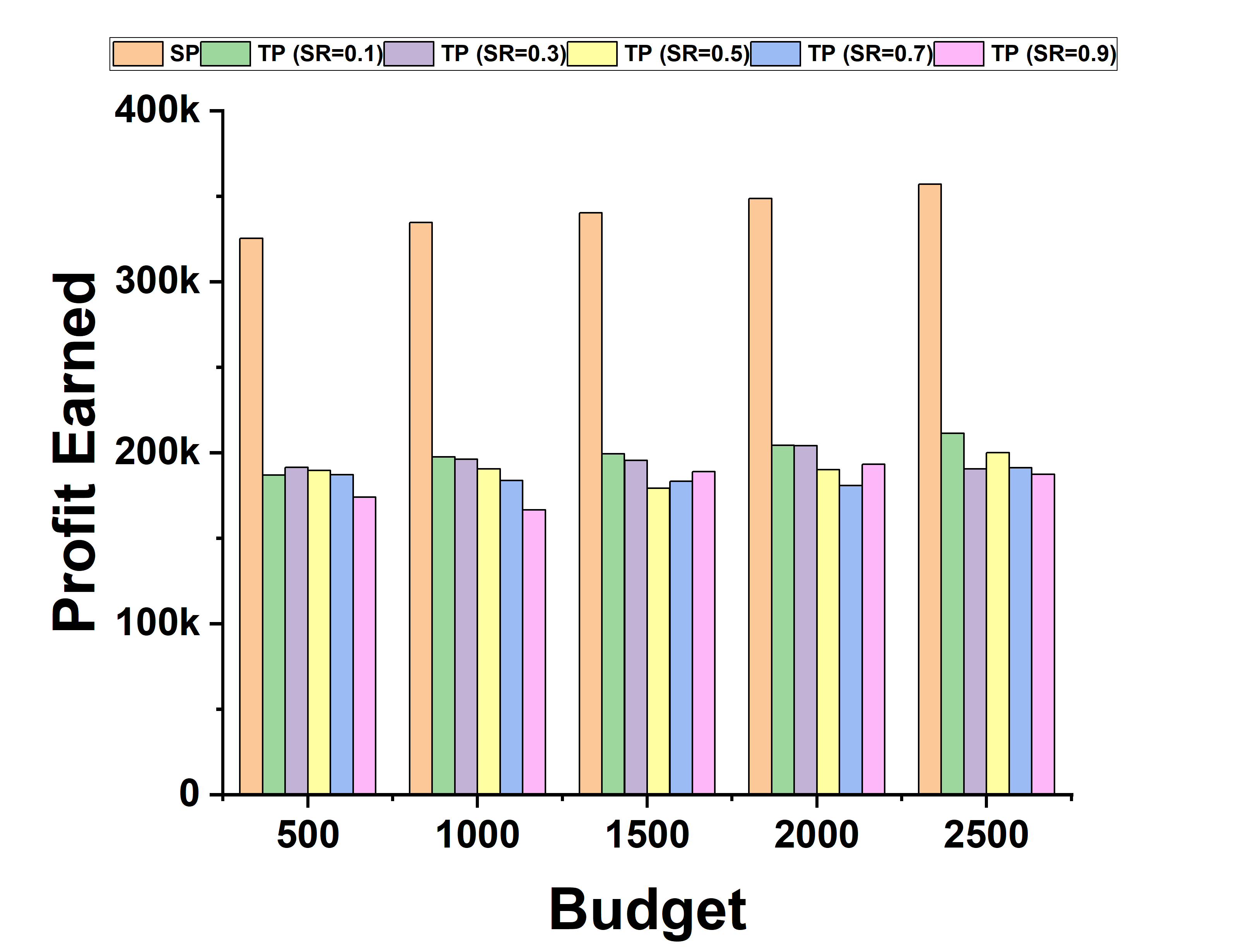}
        \caption{Simple Greedy}
    \end{subfigure} &
    \begin{subfigure}[t]{0.22\textwidth}
        \includegraphics[width=\linewidth]{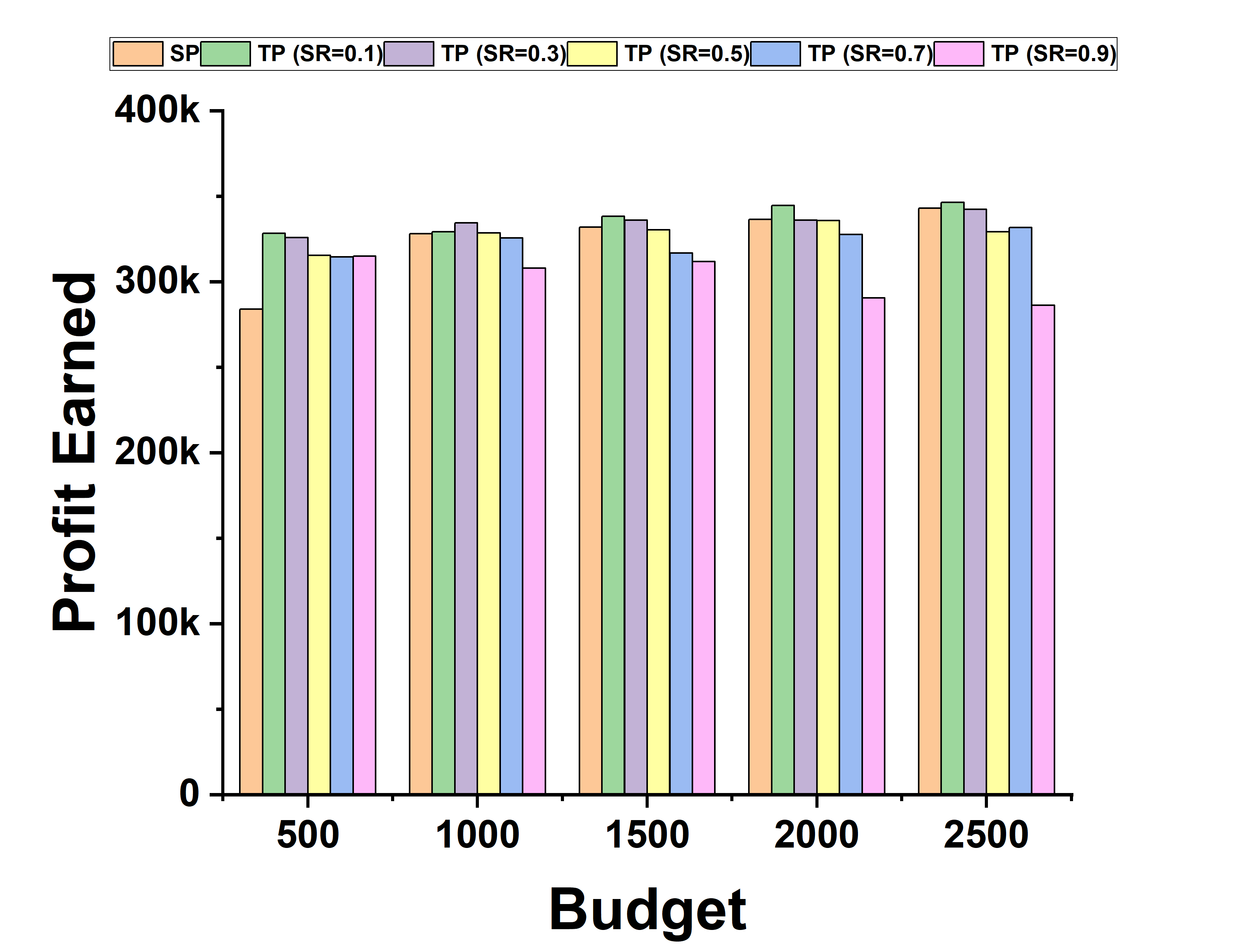}
        \caption{Double Greedy}
    \end{subfigure} &
    \begin{subfigure}[t]{0.22\textwidth}
        \includegraphics[width=\linewidth]{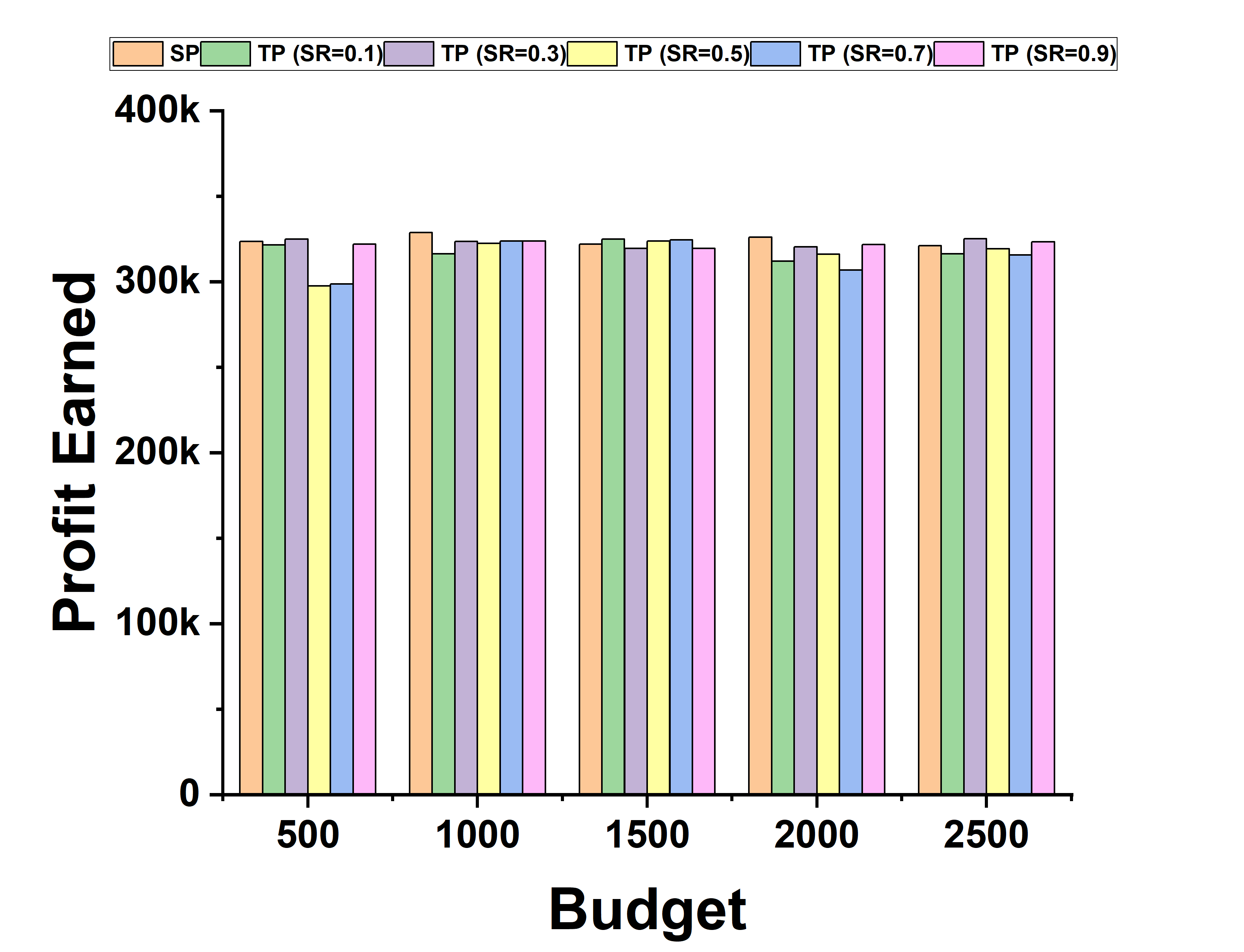}
        \caption{Stochastic Greedy}
    \end{subfigure}
\end{tabular}
\caption{Profit Earned in Single Phase Vs. Two Phase setting (Timestep 2, Probability Setting - Trivalency, \textit{Email-Eu-Core} Dataset)}
\label{Fig:RQ2_T1}
\end{figure}

\begin{figure}[htbp]
\centering
\captionsetup[sub]{font=footnotesize}  
\begin{tabular}{cccc}
    \begin{subfigure}[t]{0.22\textwidth}
        \includegraphics[width=\linewidth]{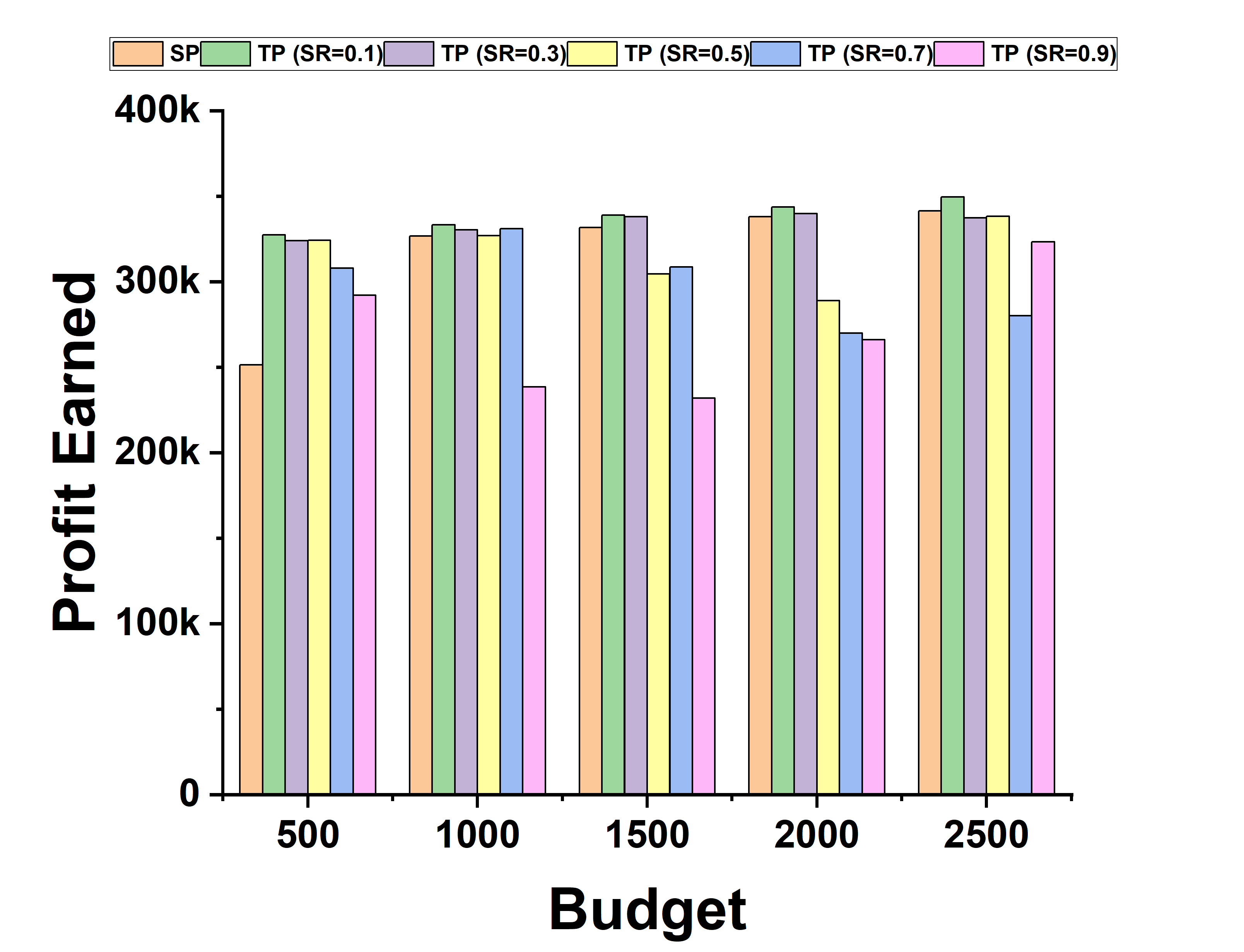}
        \caption{Random}
    \end{subfigure} &
    \begin{subfigure}[t]{0.22\textwidth}
        \includegraphics[width=\linewidth]{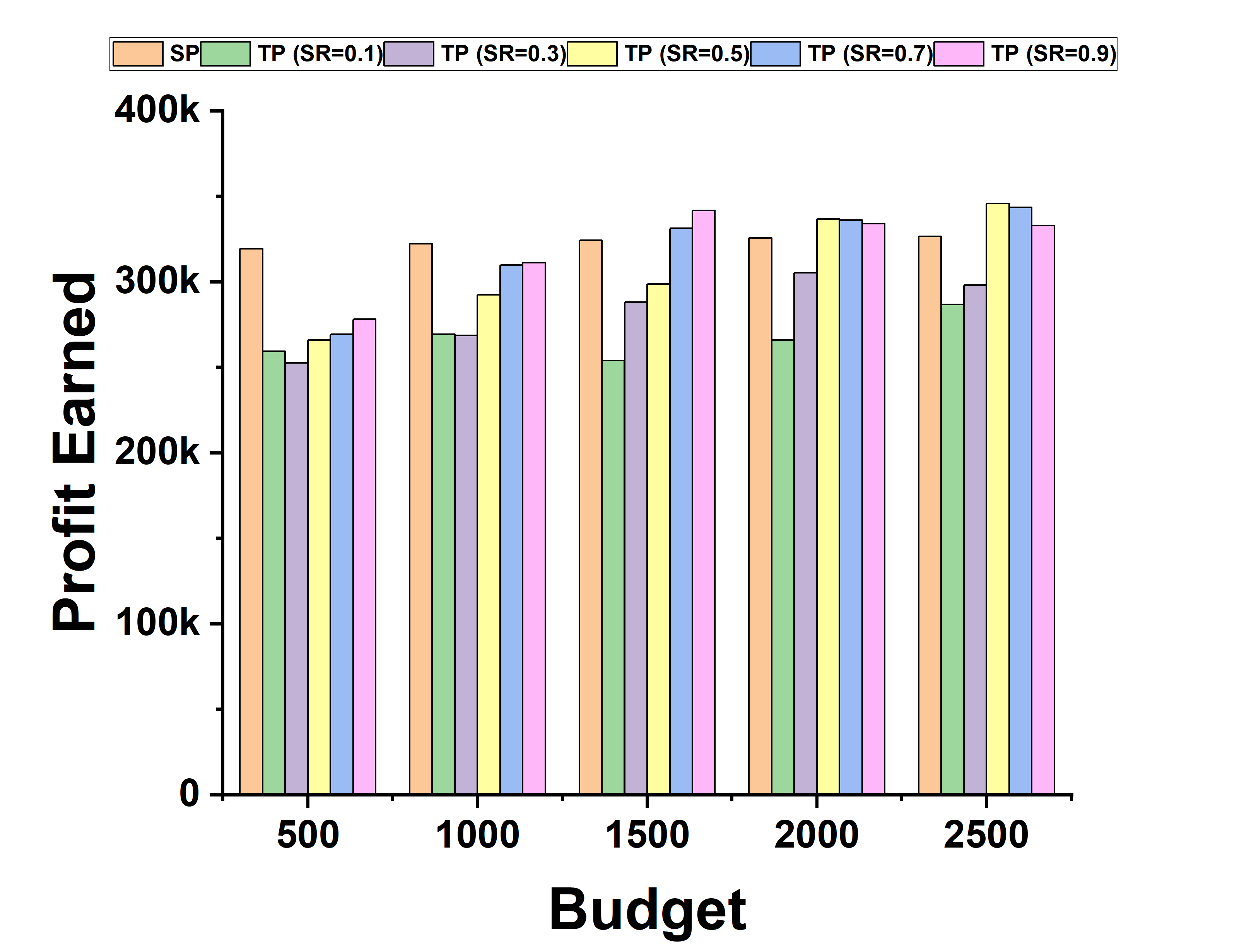}
        \caption{High Degree}
    \end{subfigure} &
    \begin{subfigure}[t]{0.22\textwidth}
        \includegraphics[width=\linewidth]{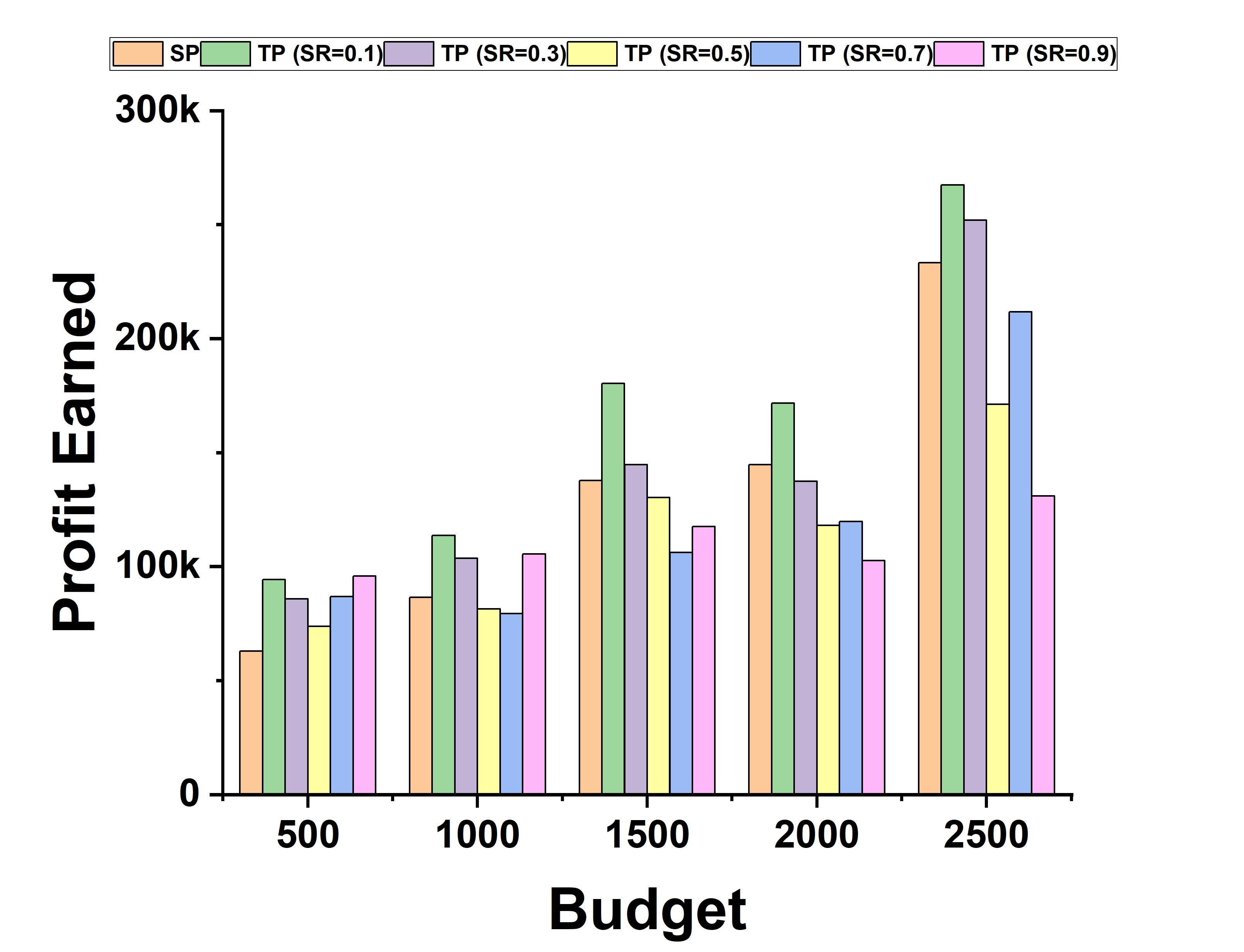}
        \caption{Clustering\\Coefficient}
    \end{subfigure} &
    \begin{subfigure}[t]{0.22\textwidth}
        \includegraphics[width=\linewidth]{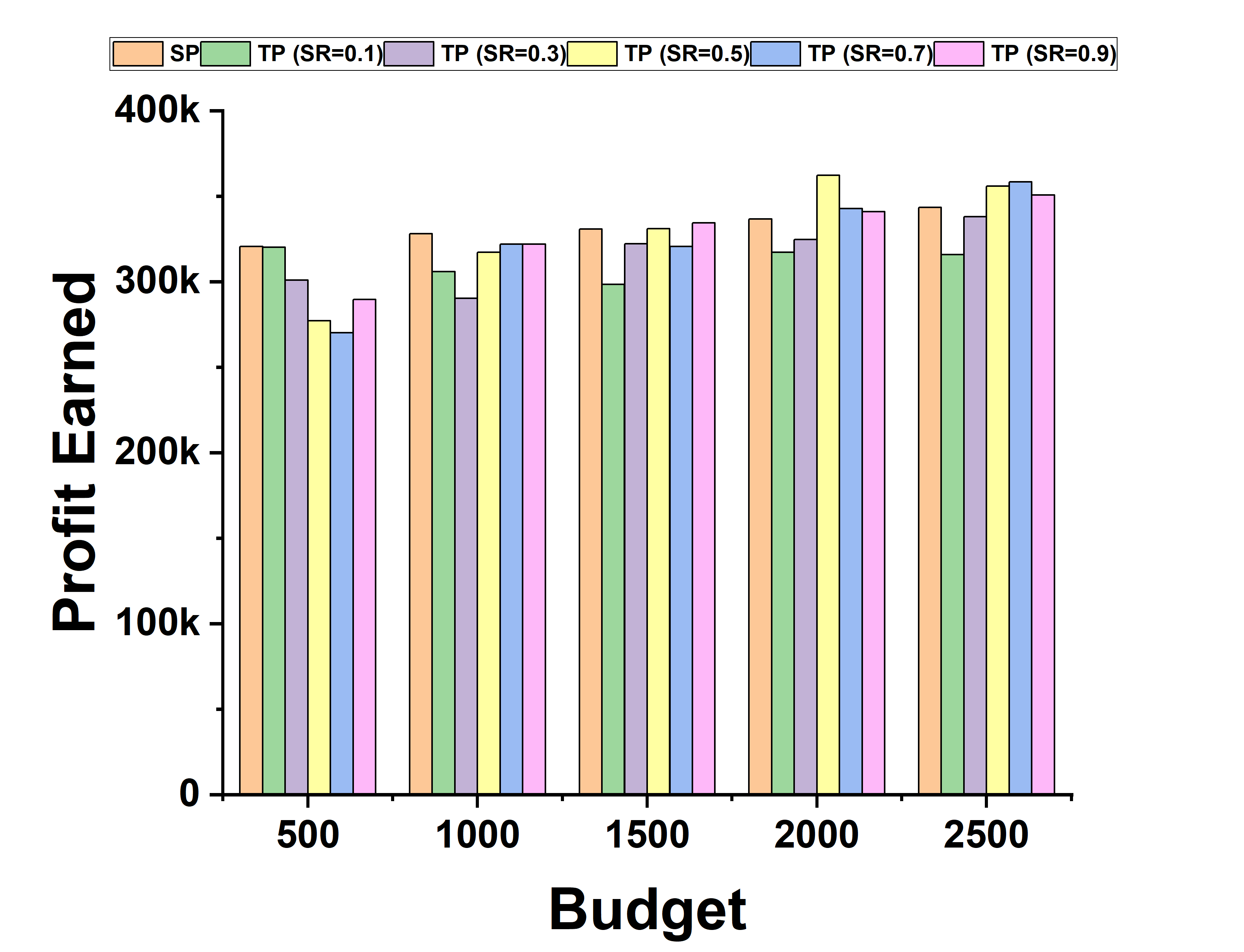}
        \caption{Degree Discount}
    \end{subfigure} \\[6pt]

    \begin{subfigure}[t]{0.22\textwidth}
        \includegraphics[width=\linewidth]{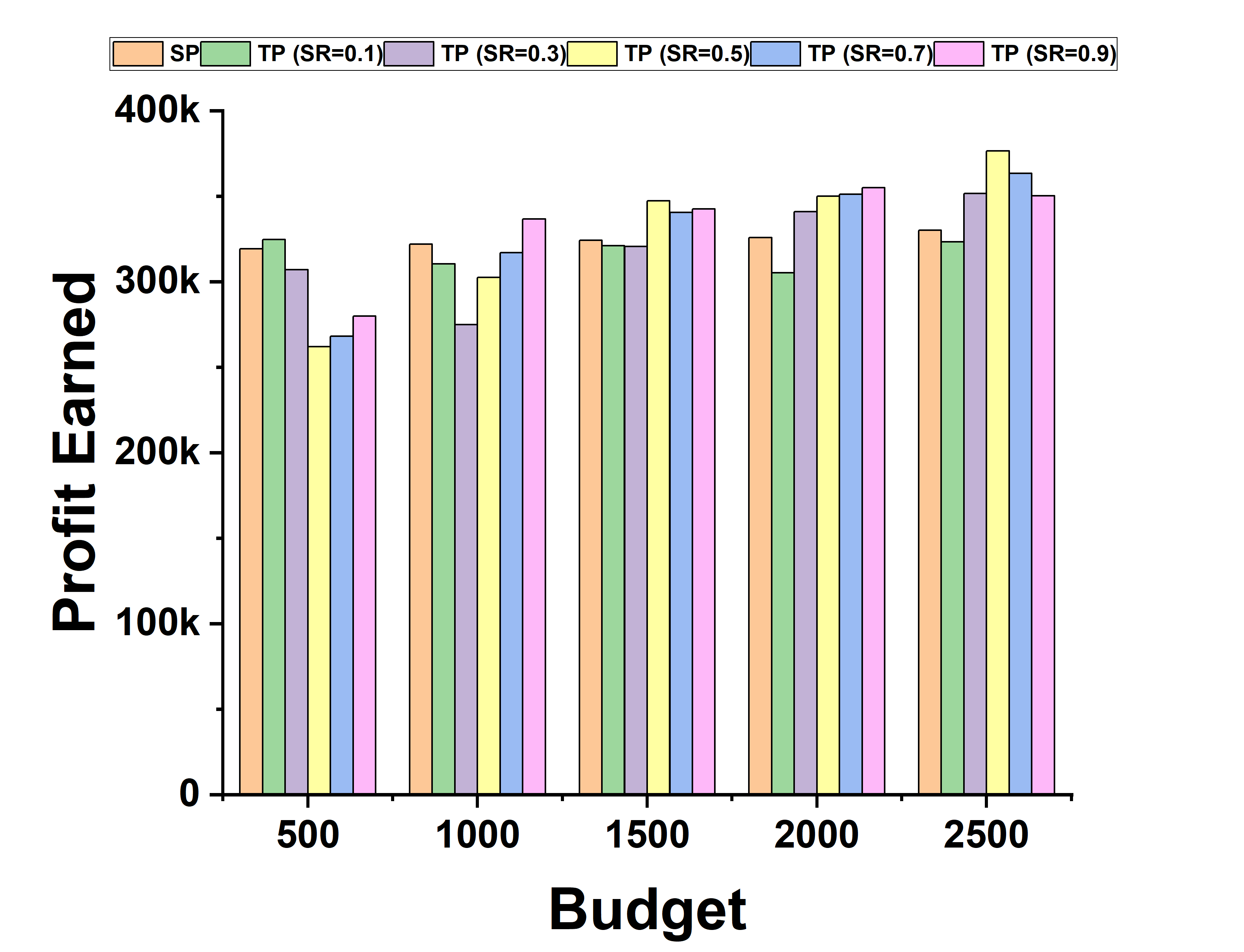}
        \caption{Single Discount}
    \end{subfigure} &
    \begin{subfigure}[t]{0.22\textwidth}
        \includegraphics[width=\linewidth]{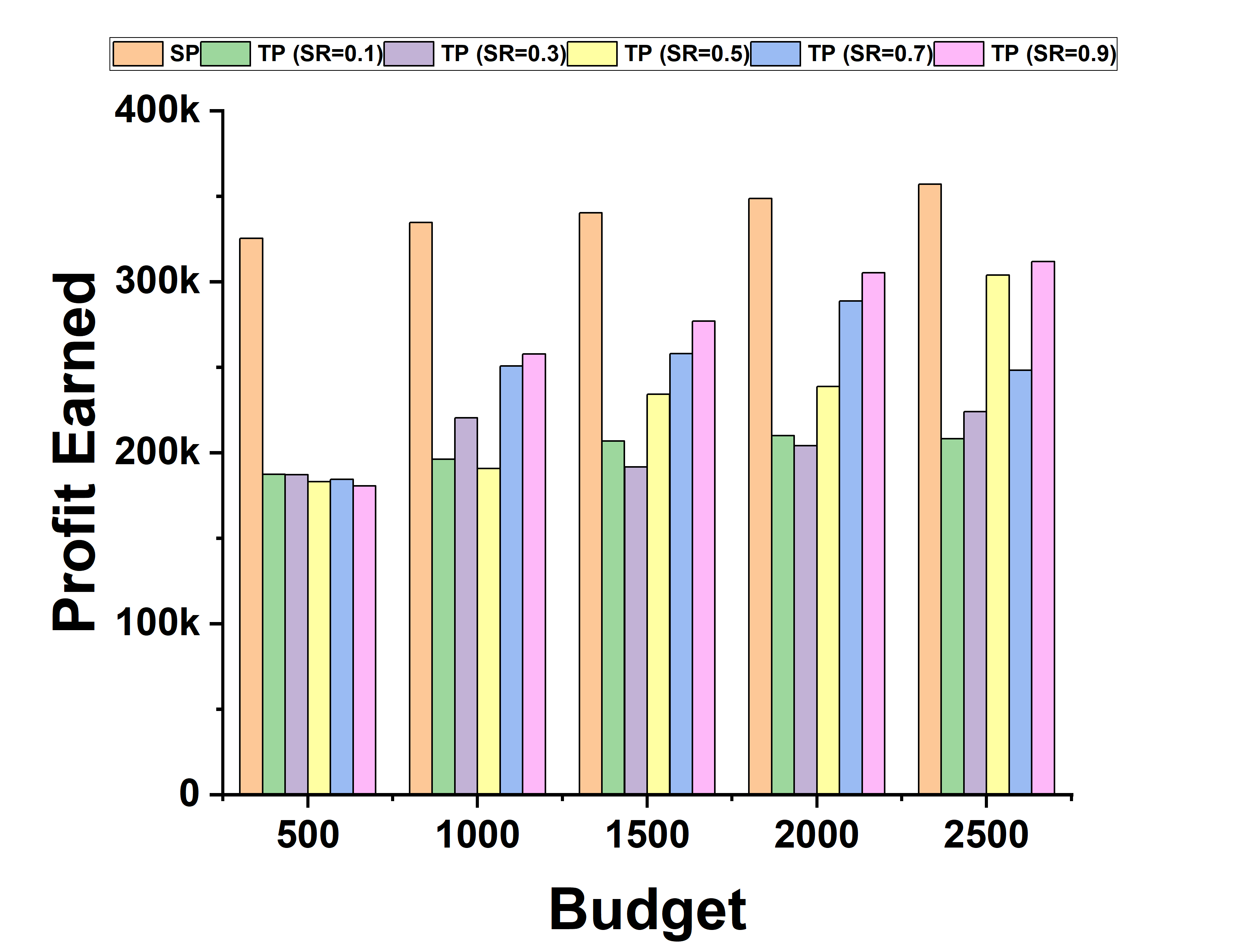}
        \caption{Simple Greedy}
    \end{subfigure} &
    \begin{subfigure}[t]{0.22\textwidth}
        \includegraphics[width=\linewidth]{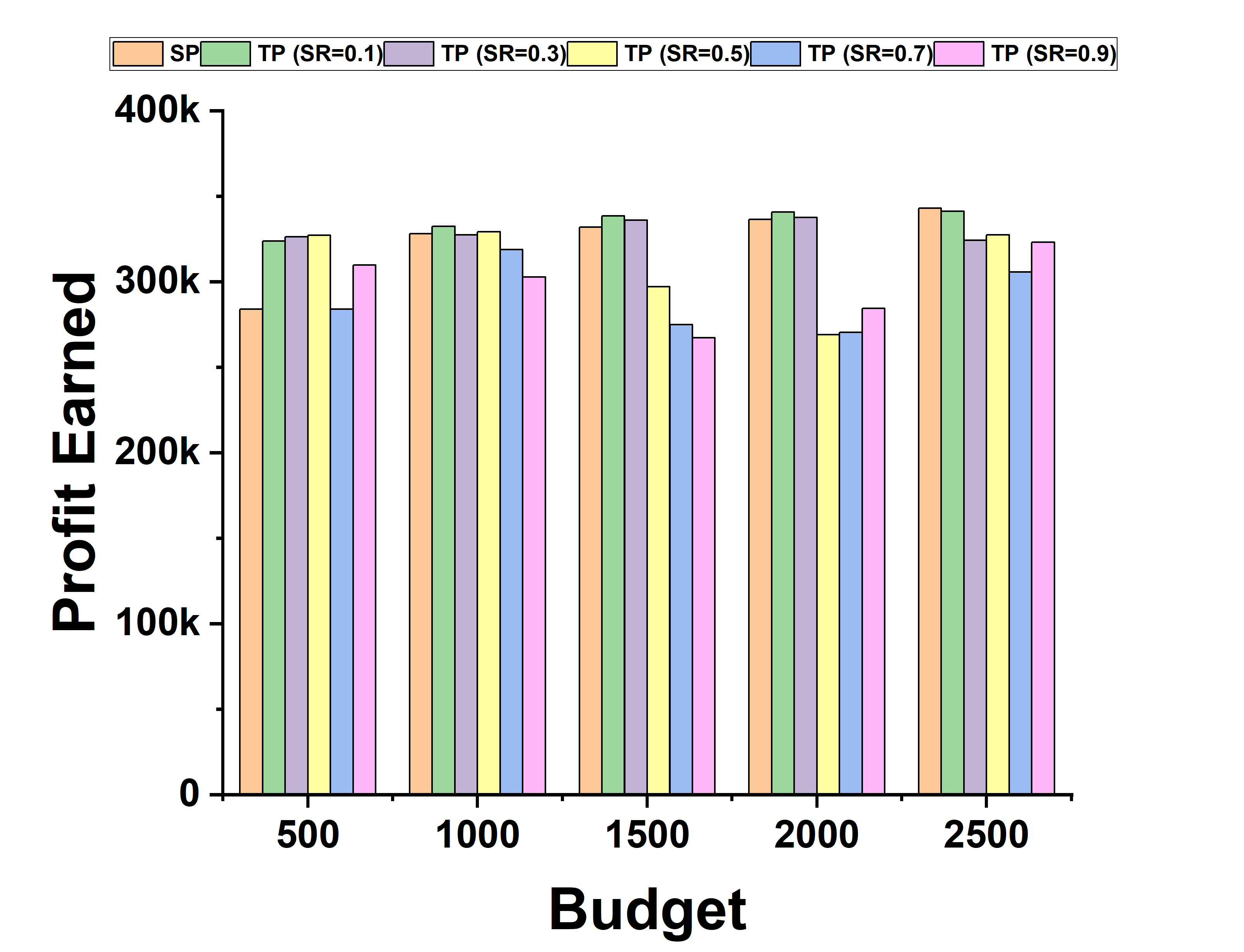}
        \caption{Double Greedy}
    \end{subfigure} &
    \begin{subfigure}[t]{0.22\textwidth}
        \includegraphics[width=\linewidth]{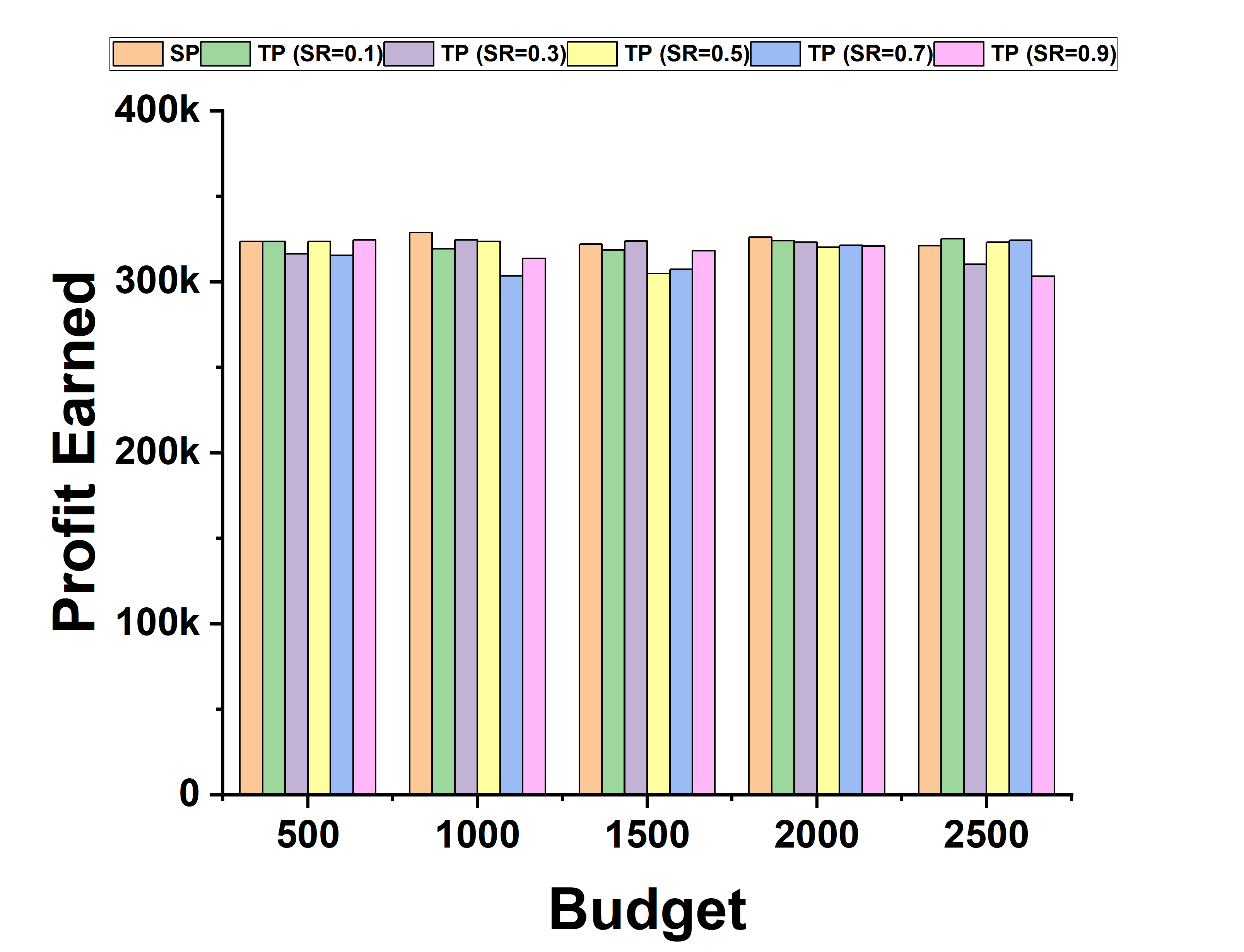}
        \caption{Stochastic Greedy}
    \end{subfigure}
\end{tabular}
\caption{Profit Earned in Single Phase Vs. Two Phase setting (Timestep 4, Probability Setting - Trivalency, \textit{Email-Eu-Core} Dataset)}
\label{Fig:RQ2_T2}
\end{figure}

\begin{figure}[htbp]
\centering
\captionsetup[sub]{font=footnotesize}  
\begin{tabular}{cccc}
    \begin{subfigure}[t]{0.22\textwidth}
        \includegraphics[width=\linewidth]{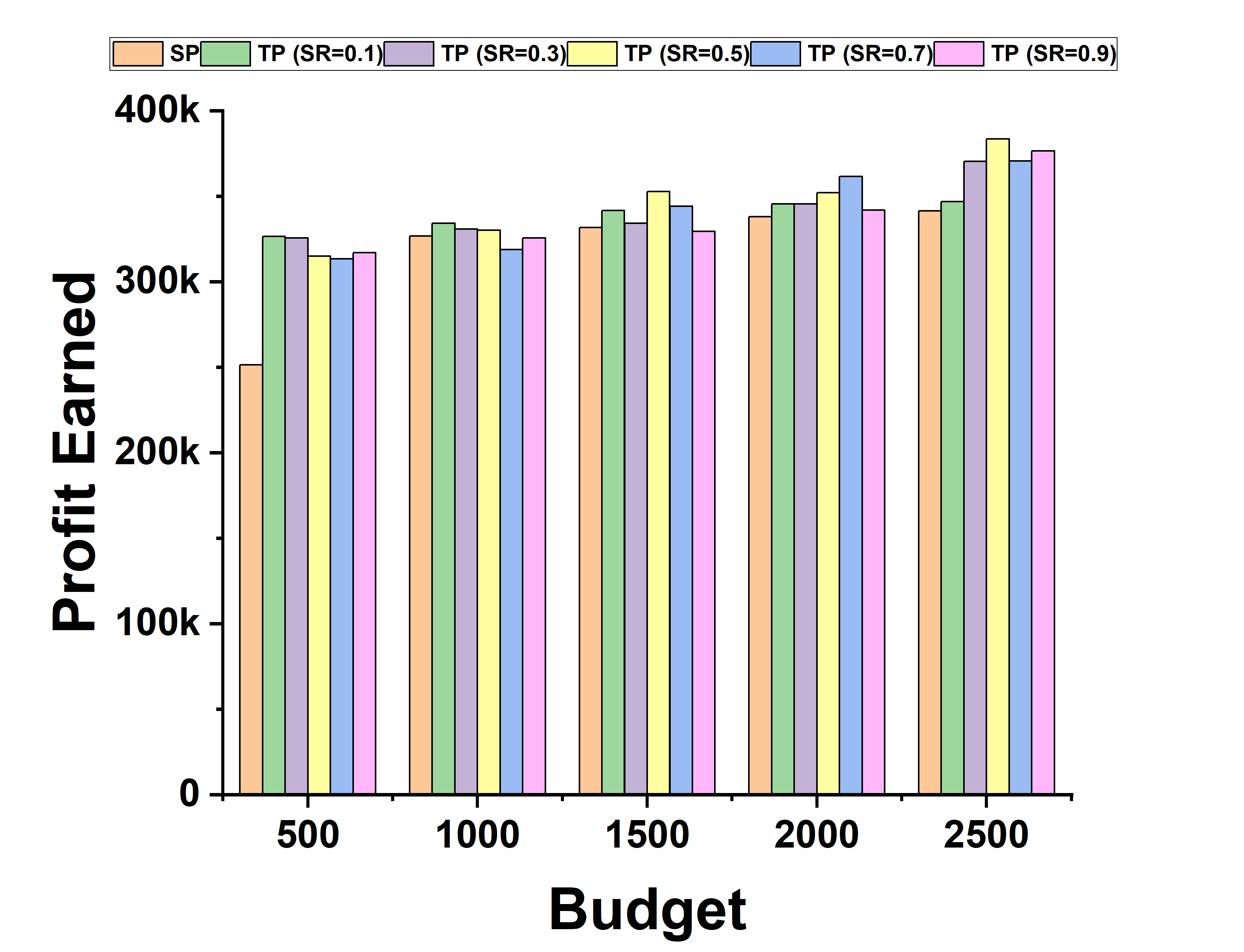}
        \caption{Random}
    \end{subfigure} &
    \begin{subfigure}[t]{0.22\textwidth}
        \includegraphics[width=\linewidth]{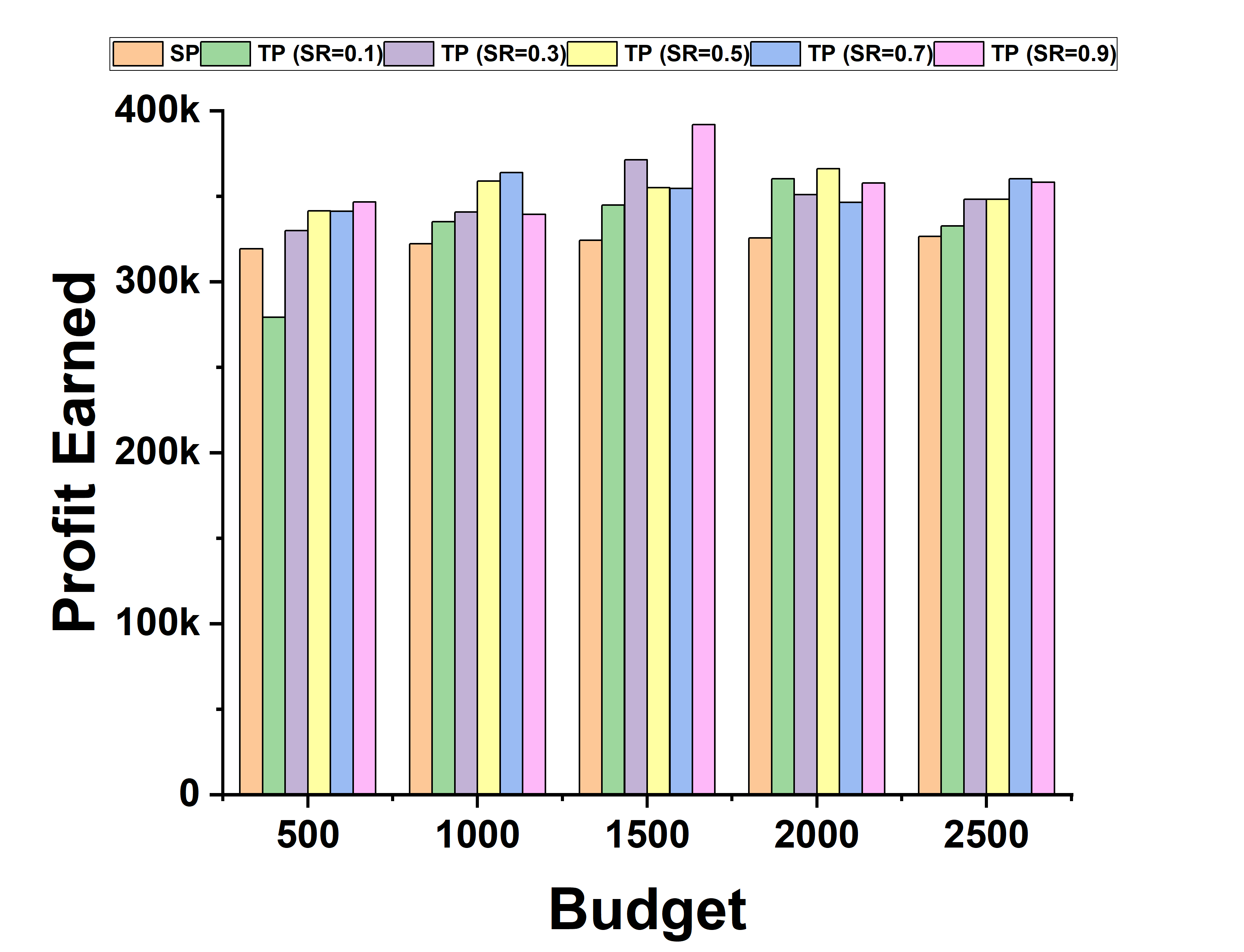}
        \caption{High Degree}
    \end{subfigure} &
    \begin{subfigure}[t]{0.22\textwidth}
        \includegraphics[width=\linewidth]{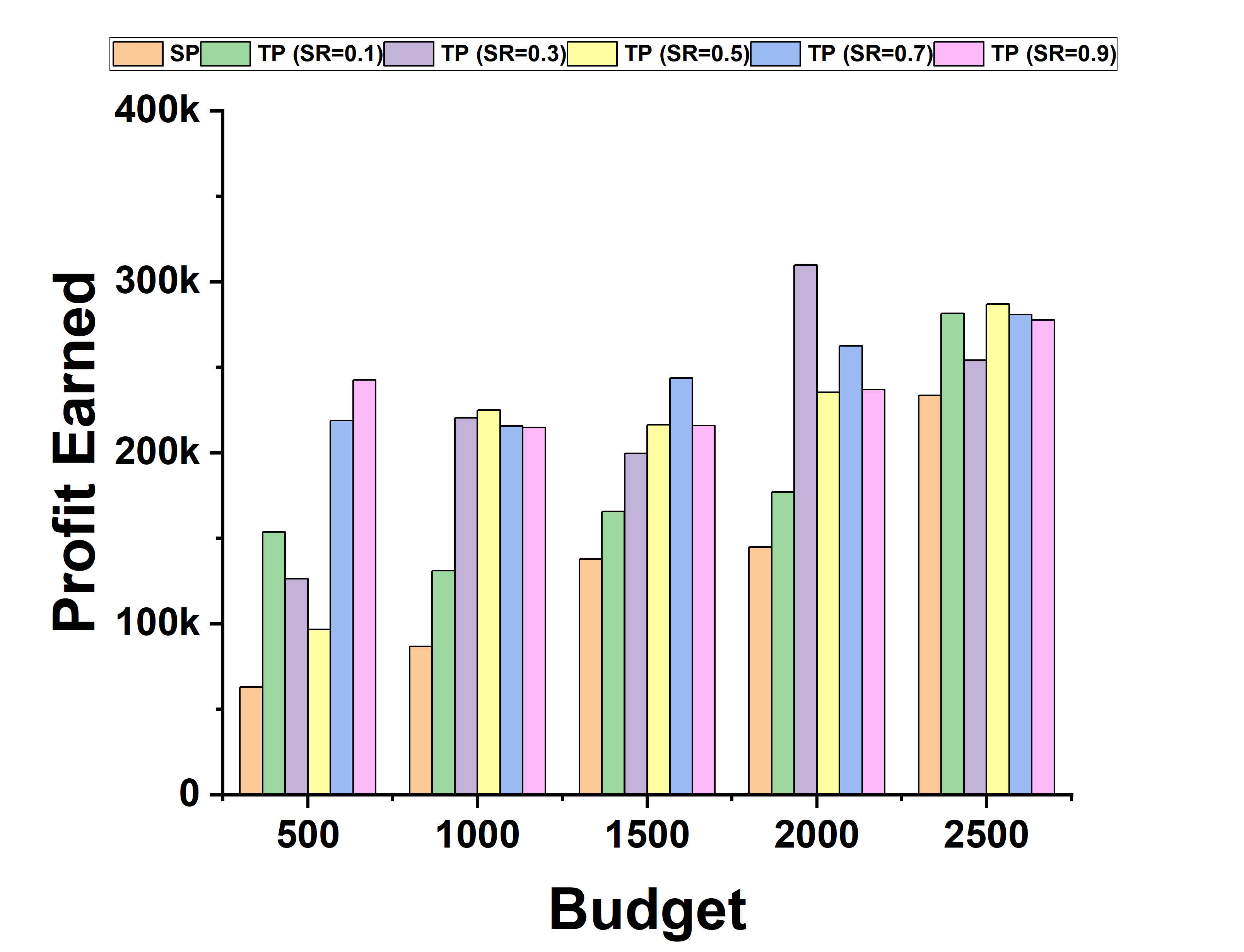}
        \caption{Clustering\\Coefficient}
    \end{subfigure} &
    \begin{subfigure}[t]{0.22\textwidth}
        \includegraphics[width=\linewidth]{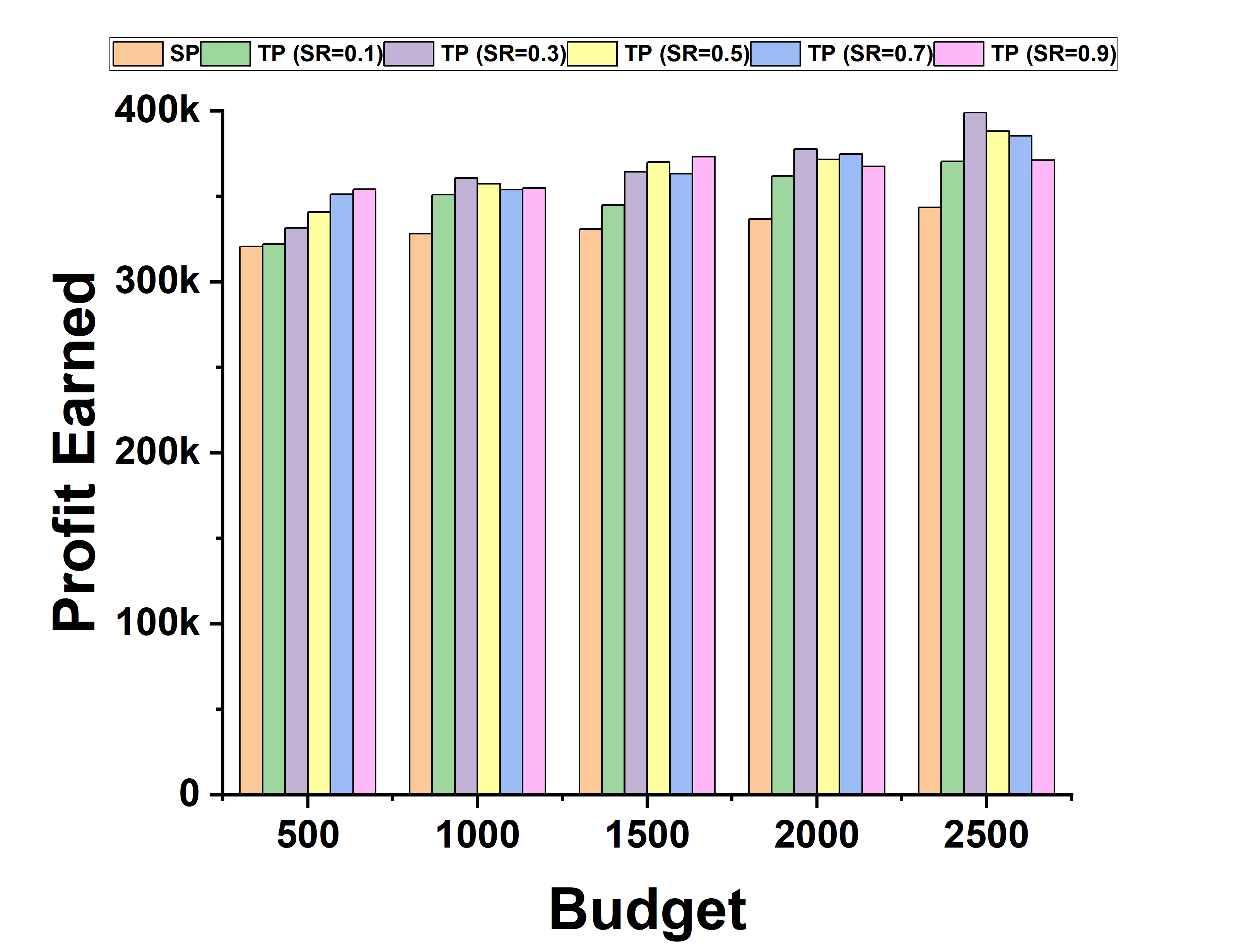}
        \caption{Degree Discount}
    \end{subfigure} \\[6pt]

    \begin{subfigure}[t]{0.22\textwidth}
        \includegraphics[width=\linewidth]{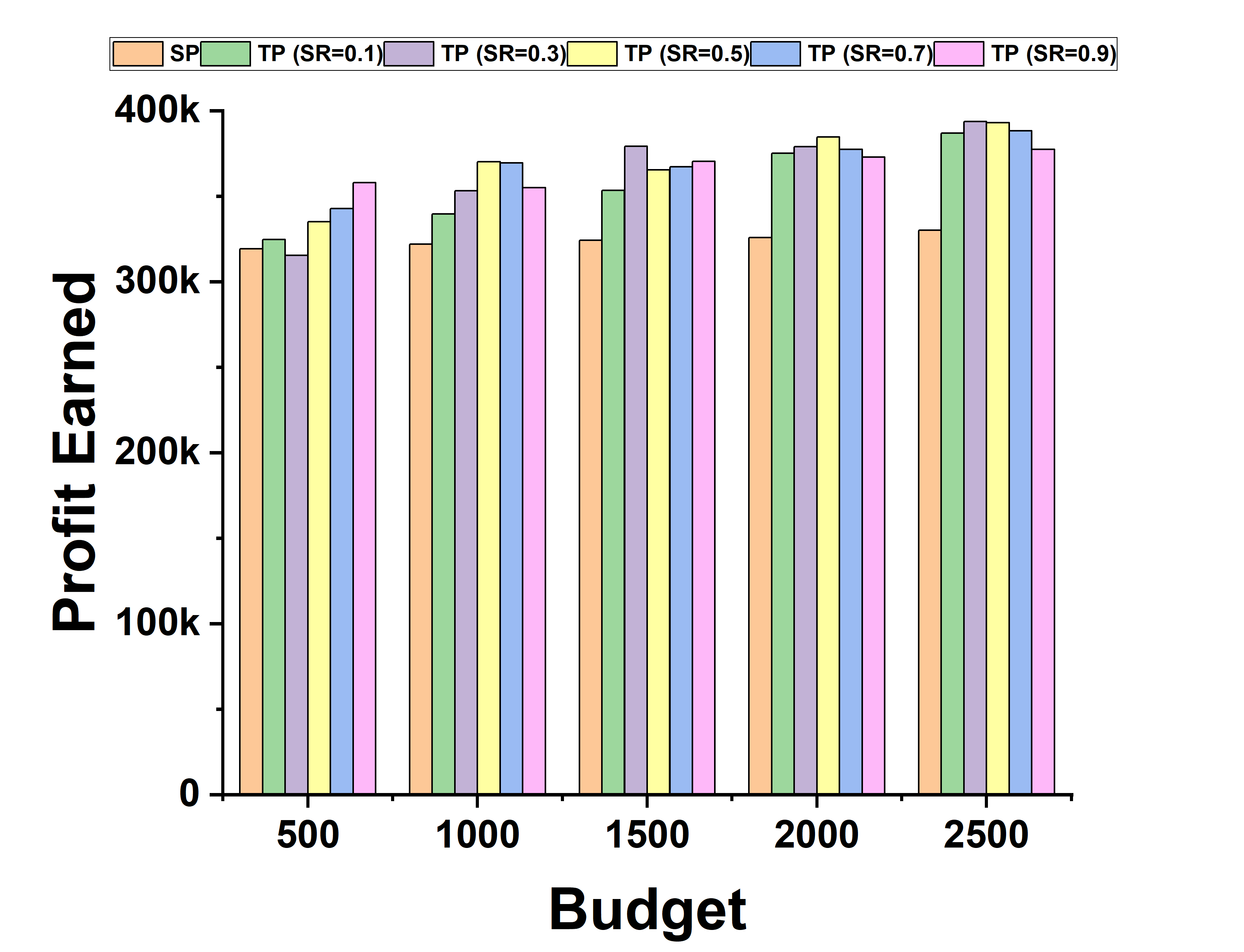}
        \caption{Single Discount}
    \end{subfigure} &
    \begin{subfigure}[t]{0.22\textwidth}
        \includegraphics[width=\linewidth]{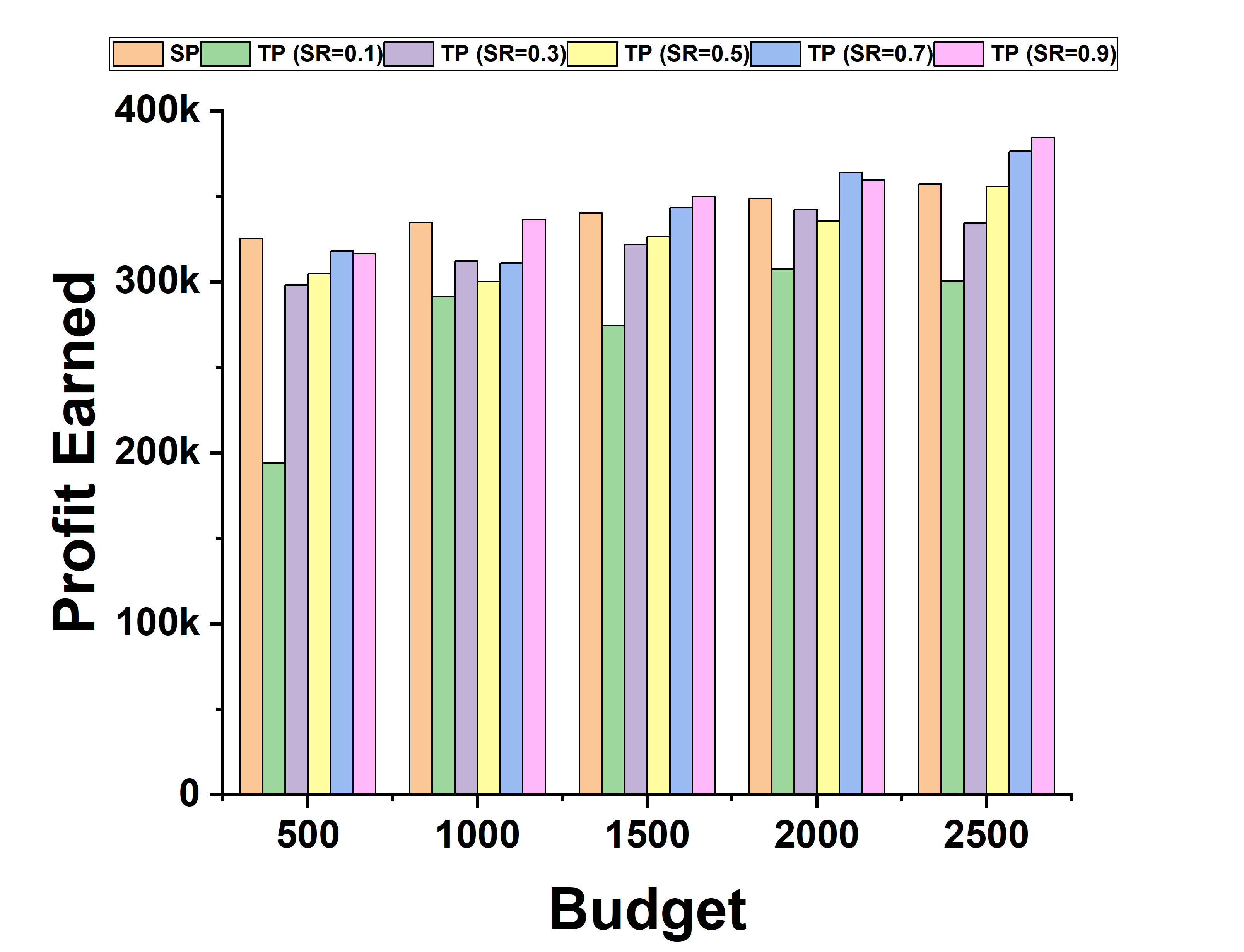}
        \caption{Simple Greedy}
    \end{subfigure} &
    \begin{subfigure}[t]{0.22\textwidth}
        \includegraphics[width=\linewidth]{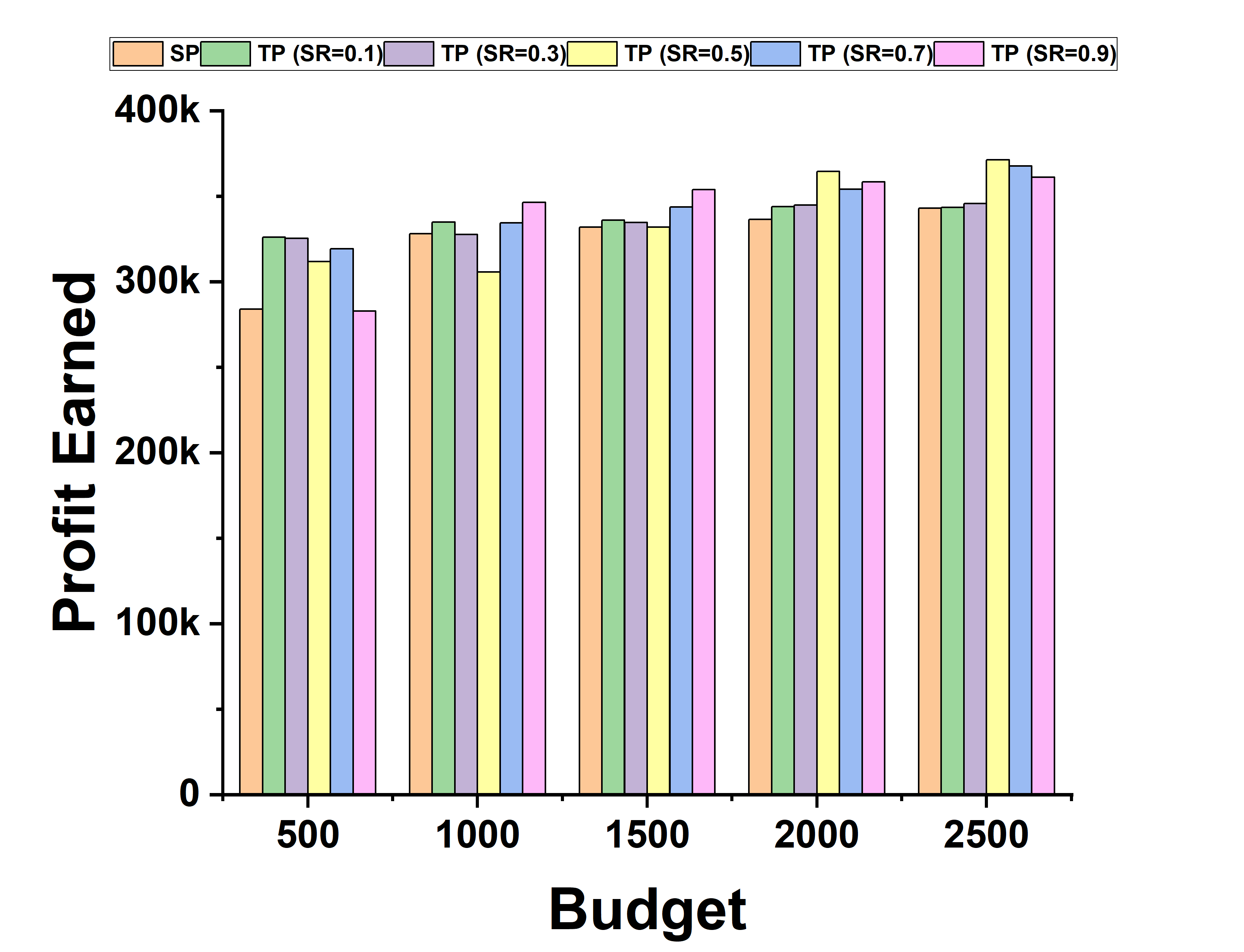}
        \caption{Double Greedy}
    \end{subfigure} &
    \begin{subfigure}[t]{0.22\textwidth}
        \includegraphics[width=\linewidth]{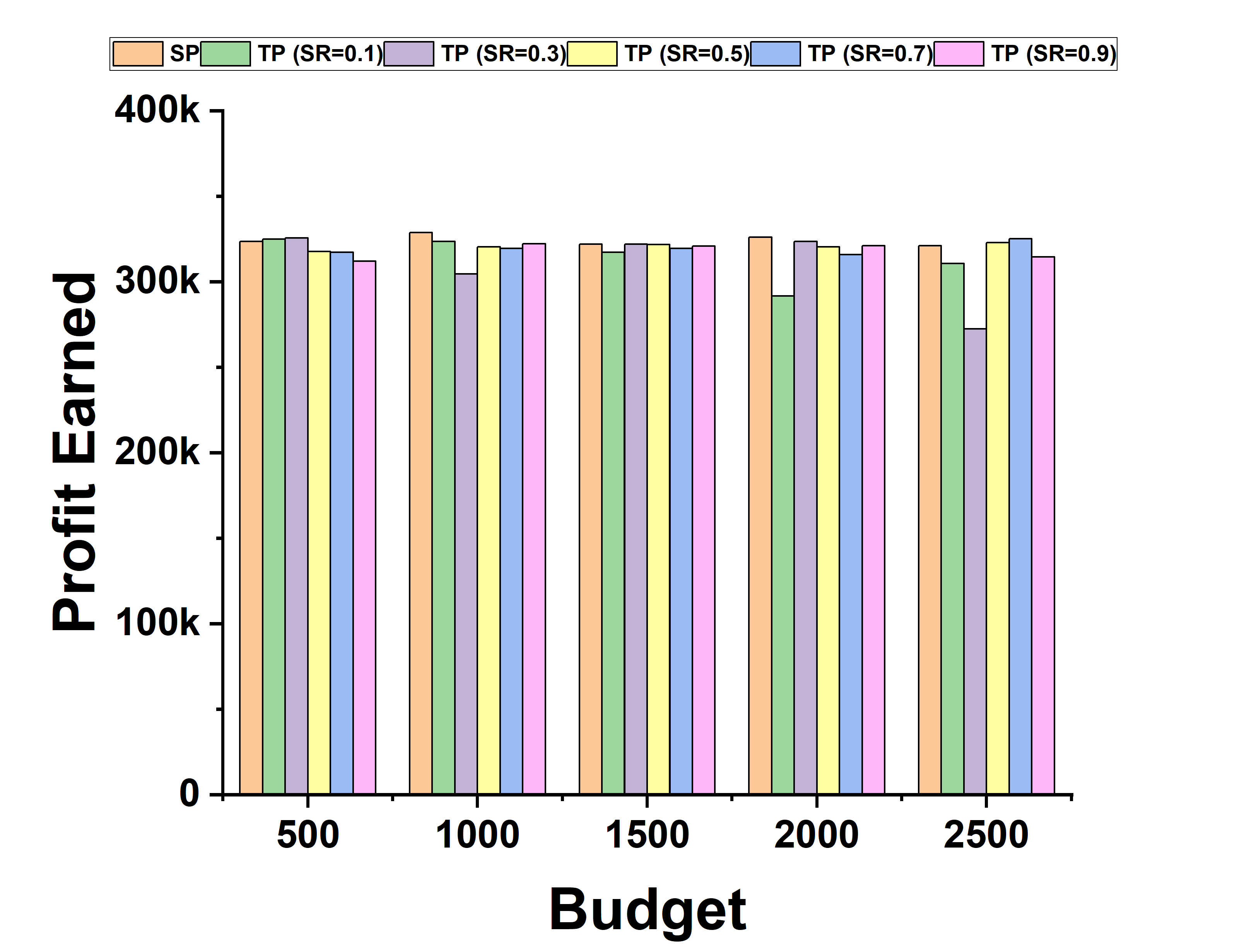}
        \caption{Stochastic Greedy}
    \end{subfigure}
\end{tabular}
\caption{Profit Earned in Single Phase Vs. Two Phase setting (Timestep 6, Probability Setting - Trivalency, \textit{Email-Eu-Core} Dataset)}
\label{Fig:RQ2_T3}
\end{figure}

\begin{figure}[htbp]
\centering
\captionsetup[sub]{font=footnotesize}  
\begin{tabular}{cccc}
    \begin{subfigure}[t]{0.22\textwidth}
        \includegraphics[width=\linewidth]{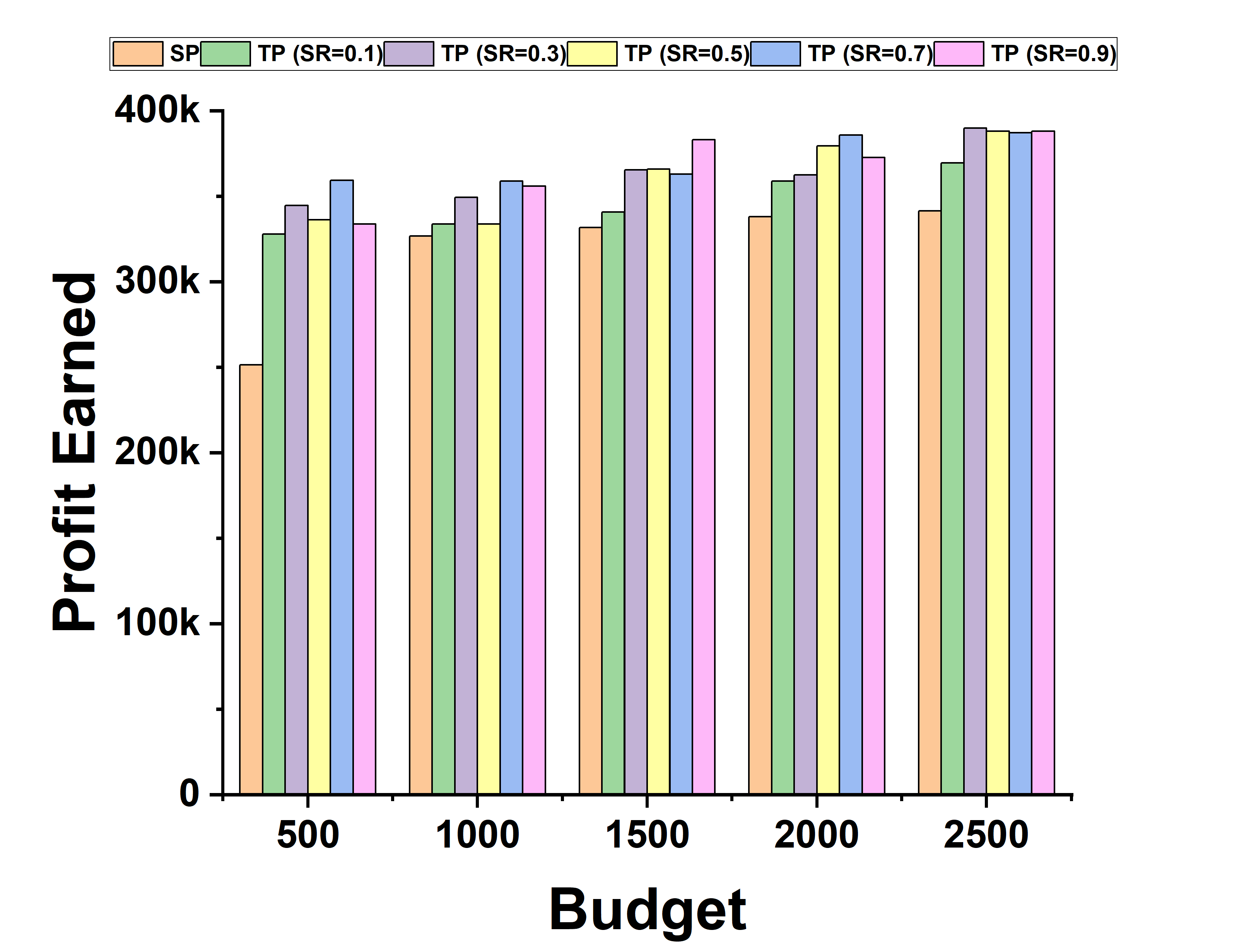}
        \caption{Random}
    \end{subfigure} &
    \begin{subfigure}[t]{0.22\textwidth}
        \includegraphics[width=\linewidth]{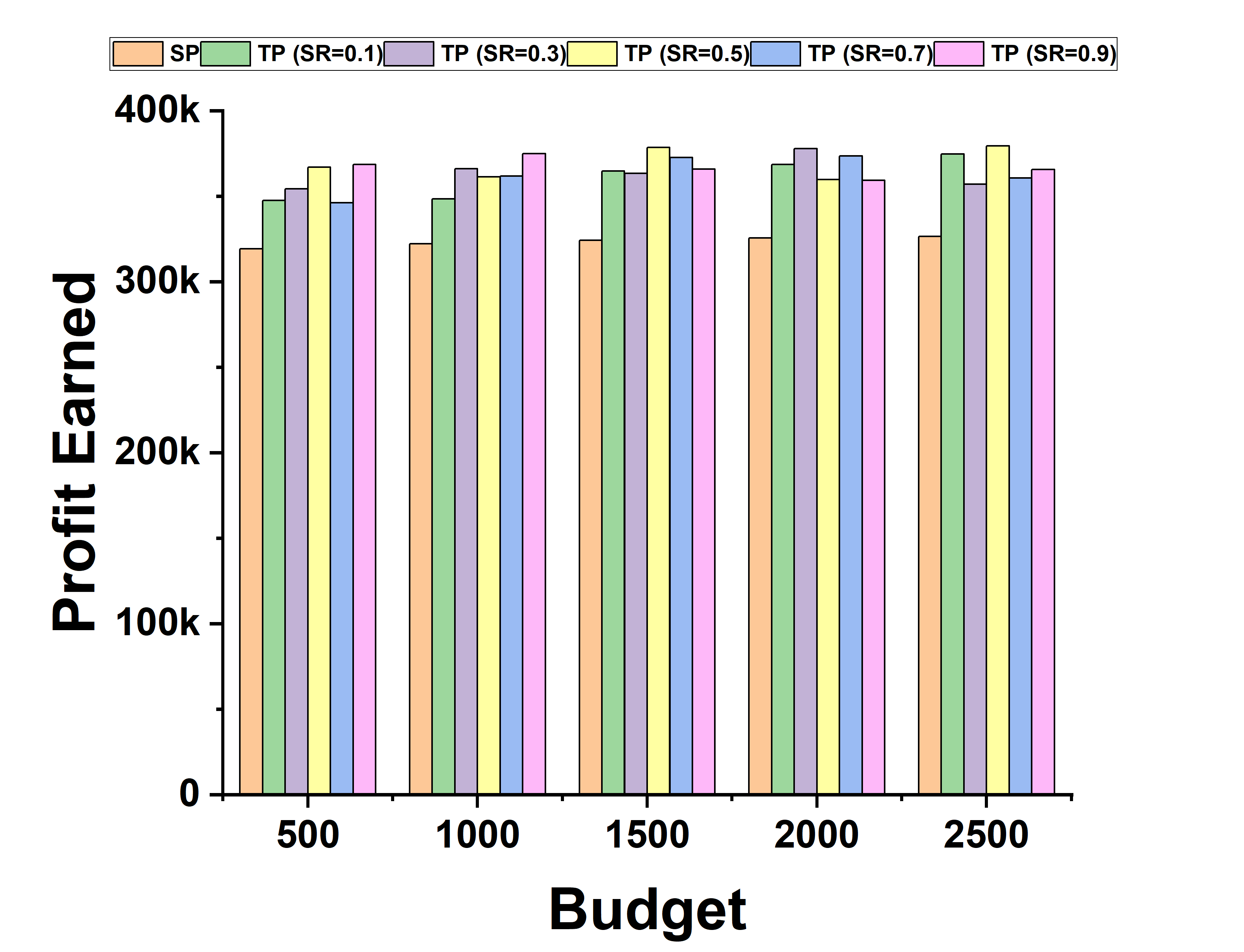}
        \caption{High Degree}
    \end{subfigure} &
    \begin{subfigure}[t]{0.22\textwidth}
        \includegraphics[width=\linewidth]{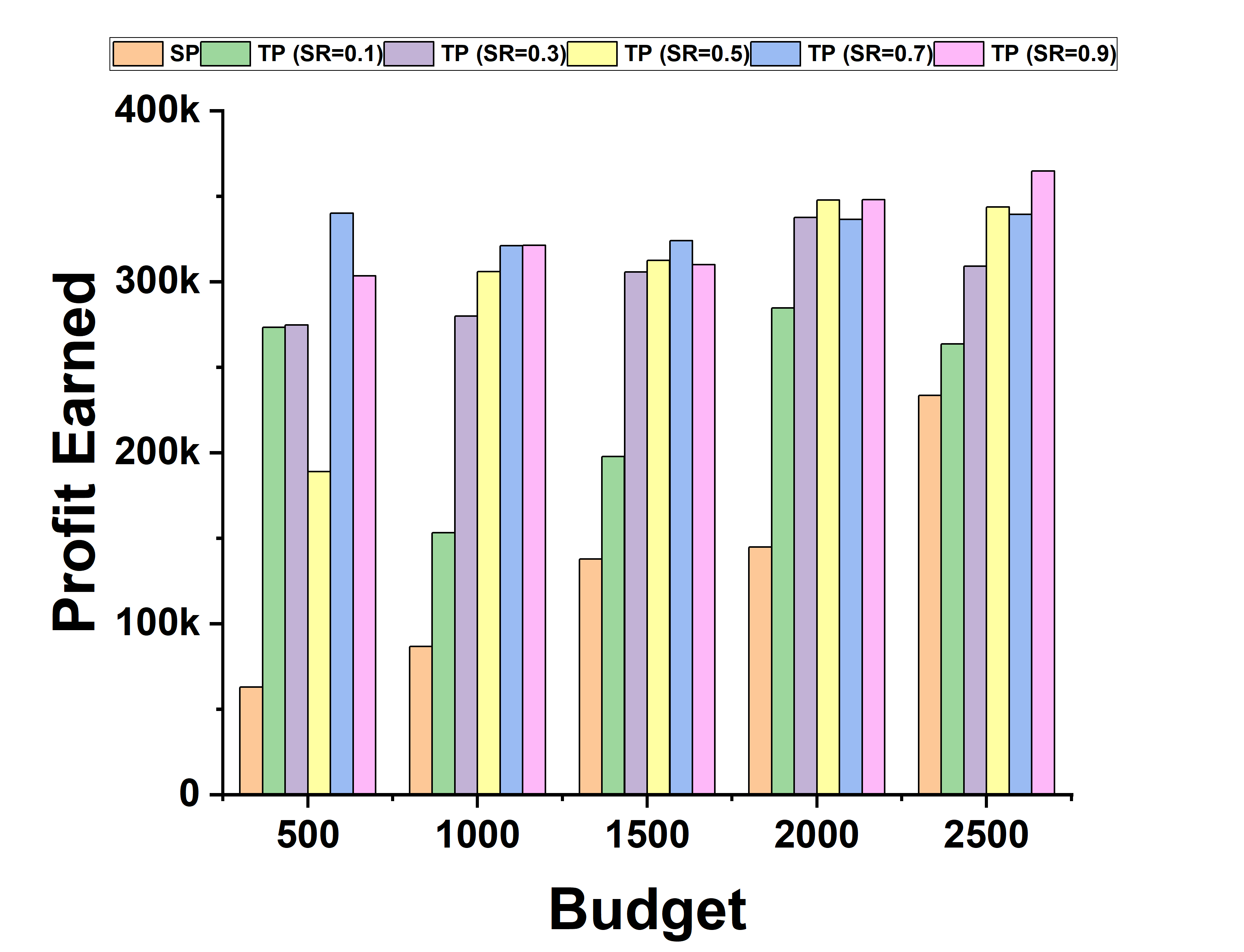}
        \caption{Clustering\\Coefficient}
    \end{subfigure} &
    \begin{subfigure}[t]{0.22\textwidth}
        \includegraphics[width=\linewidth]{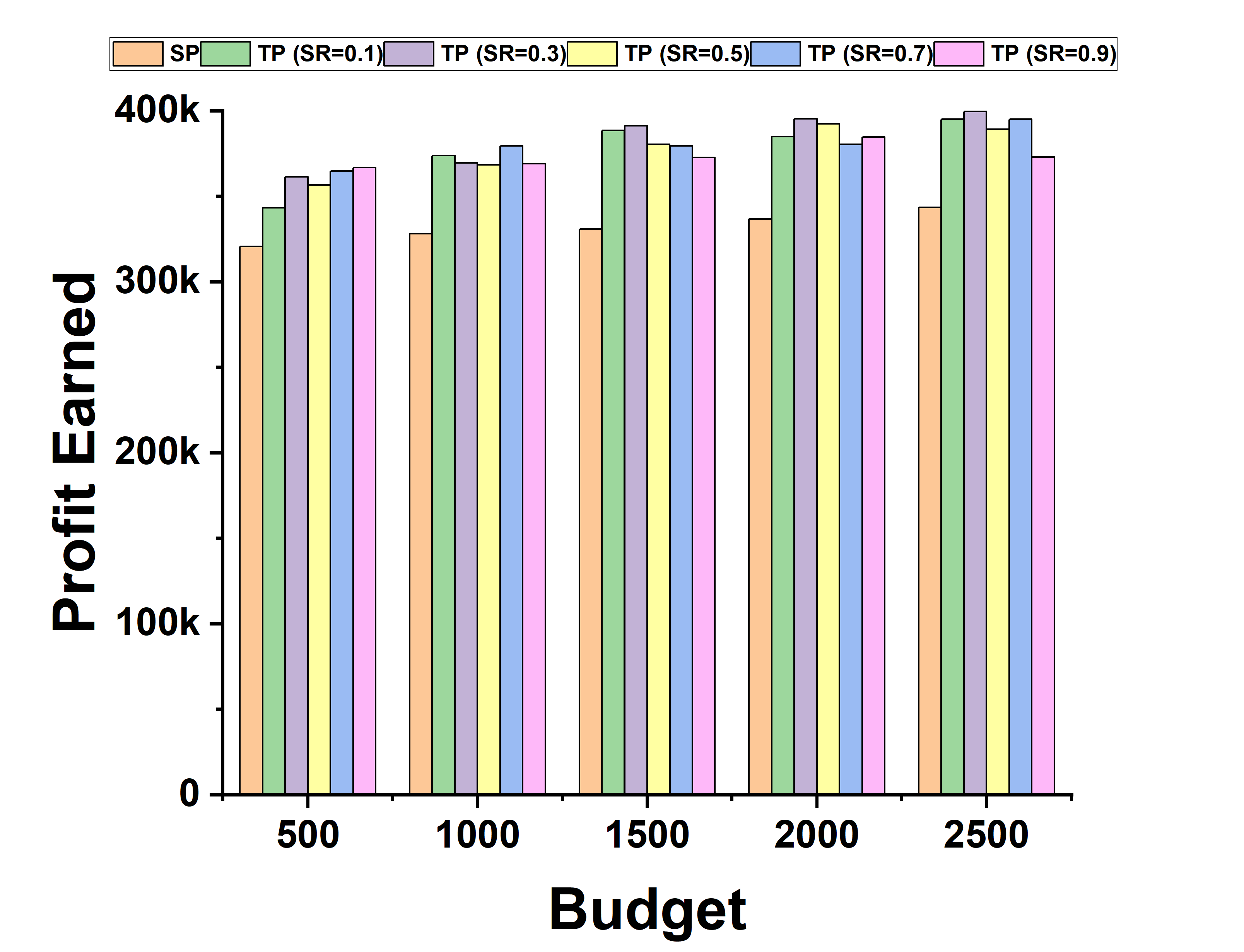}
        \caption{Degree Discount}
    \end{subfigure} \\[6pt]

    \begin{subfigure}[t]{0.22\textwidth}
        \includegraphics[width=\linewidth]{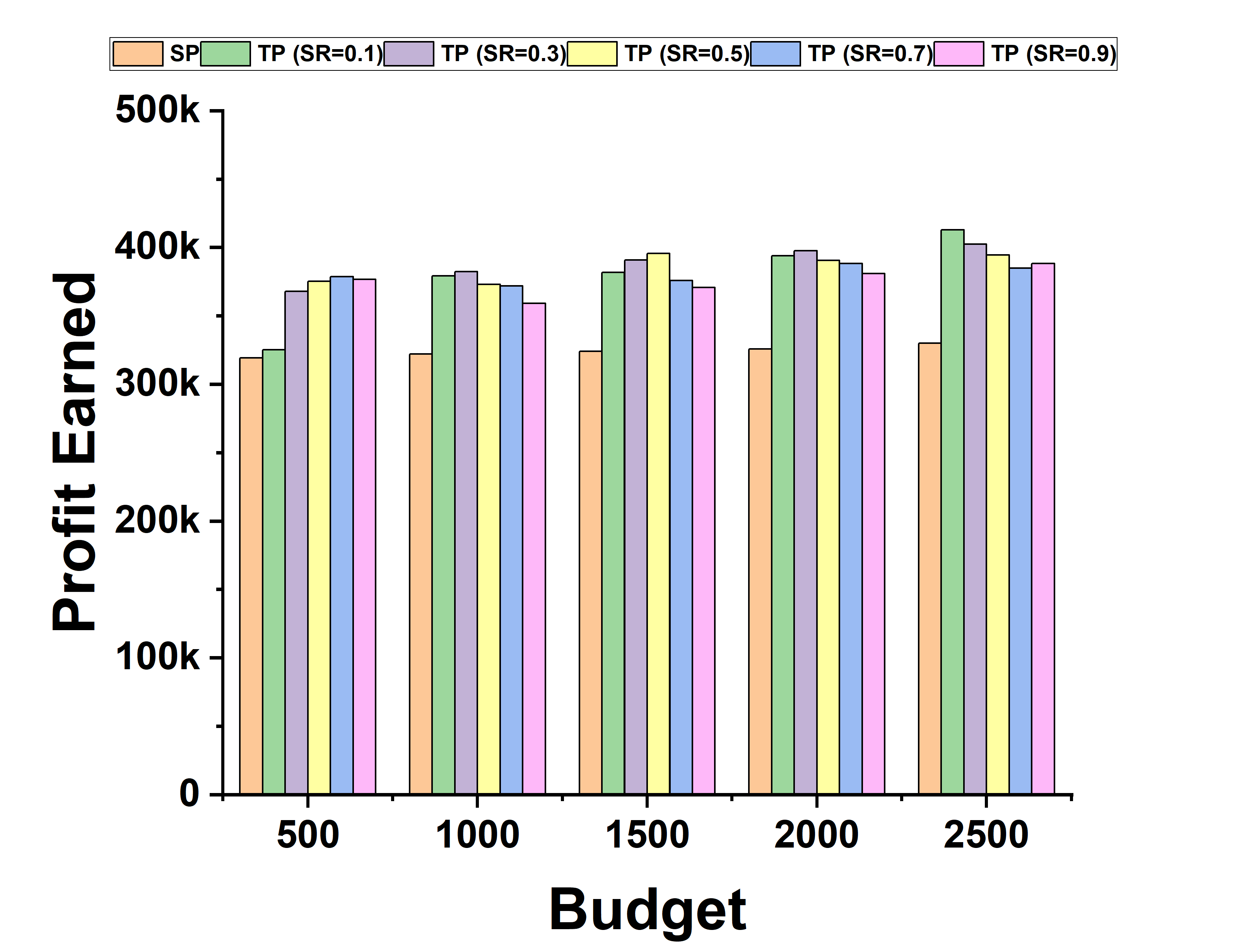}
        \caption{Single Discount}
    \end{subfigure} &
    \begin{subfigure}[t]{0.22\textwidth}
        \includegraphics[width=\linewidth]{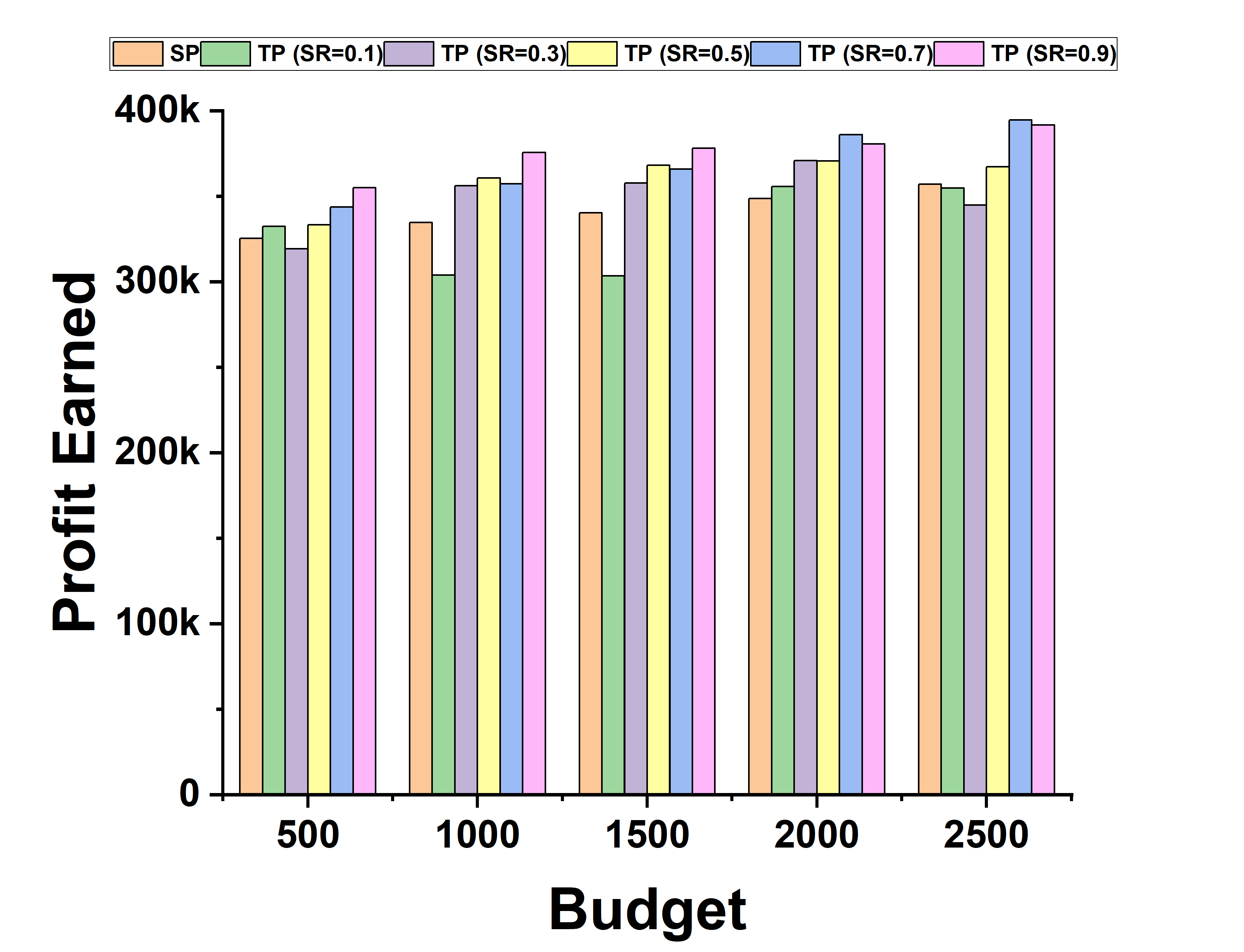}
        \caption{Simple Greedy}
    \end{subfigure} &
    \begin{subfigure}[t]{0.22\textwidth}
        \includegraphics[width=\linewidth]{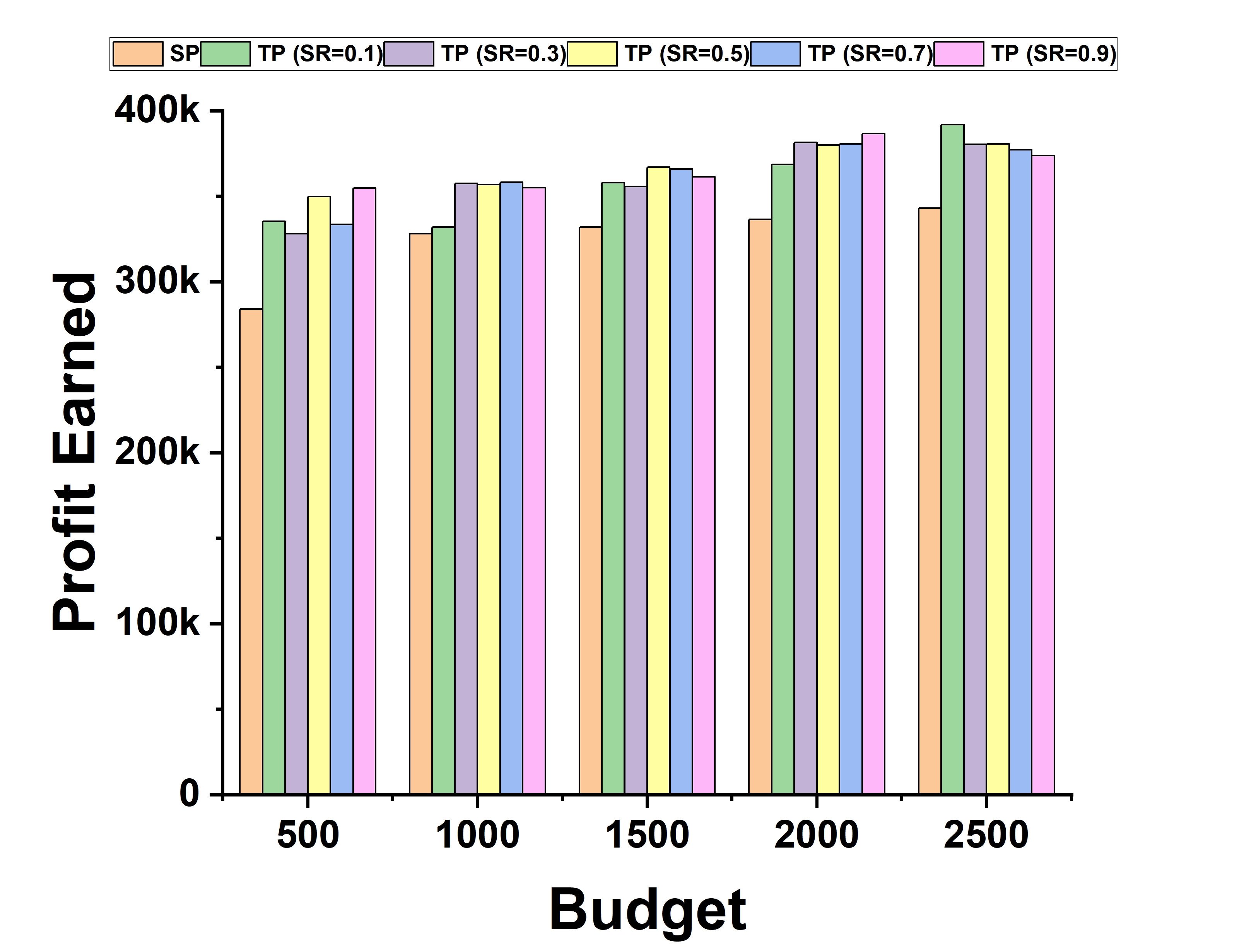}
        \caption{Double Greedy}
    \end{subfigure} &
    \begin{subfigure}[t]{0.22\textwidth}
        \includegraphics[width=\linewidth]{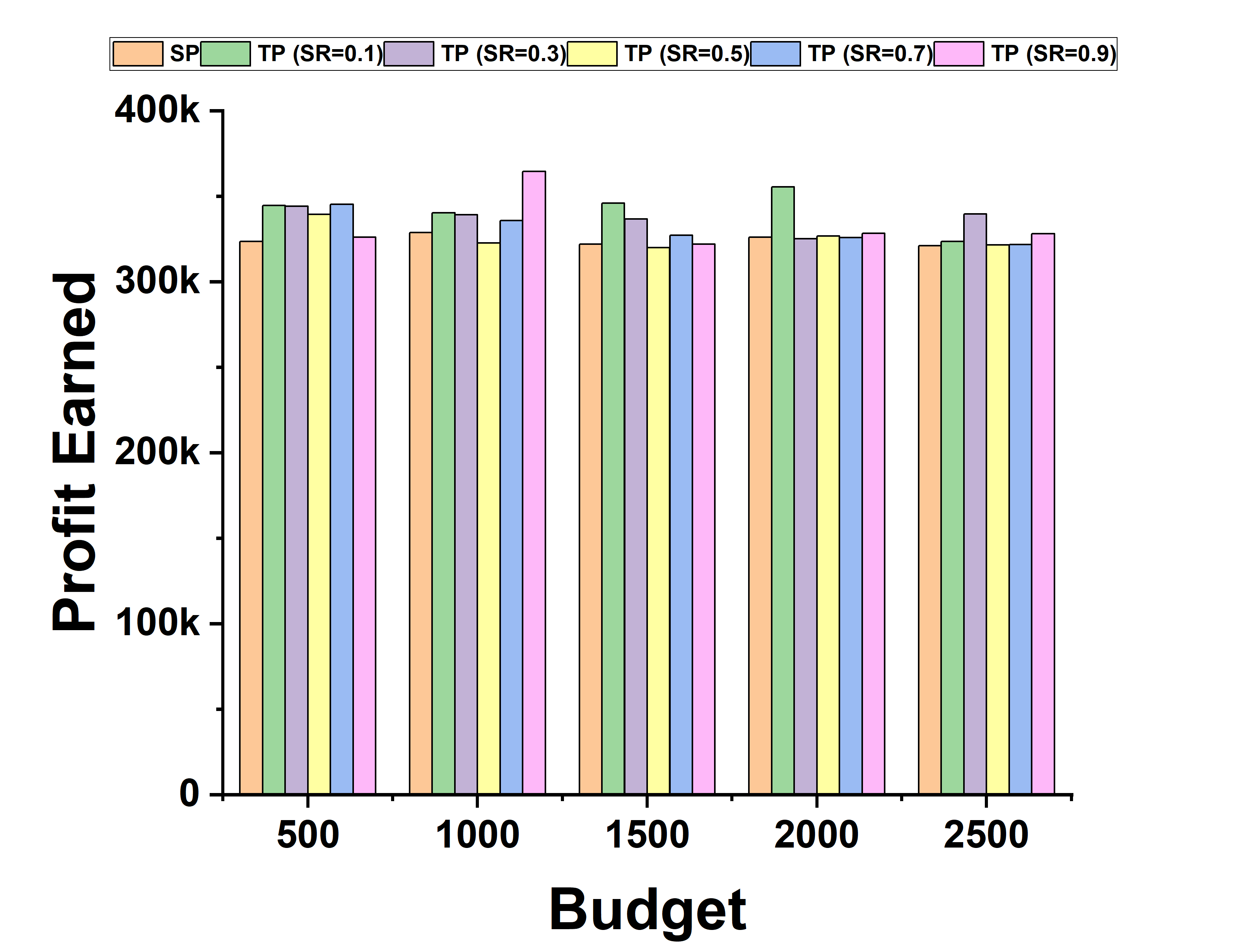}
        \caption{Stochastic Greedy}
    \end{subfigure}
\end{tabular}
\caption{Profit Earned in Single Phase Vs. Two Phase setting (Timestep 8, Probability Setting - Trivalency, \textit{Email-Eu-Core} Dataset)}
\label{Fig:RQ2_T4}
\end{figure}

\begin{figure}[htbp]
\centering
\captionsetup[sub]{font=footnotesize}  
\begin{tabular}{cccc}
    \begin{subfigure}[t]{0.22\textwidth}
        \includegraphics[width=\linewidth]{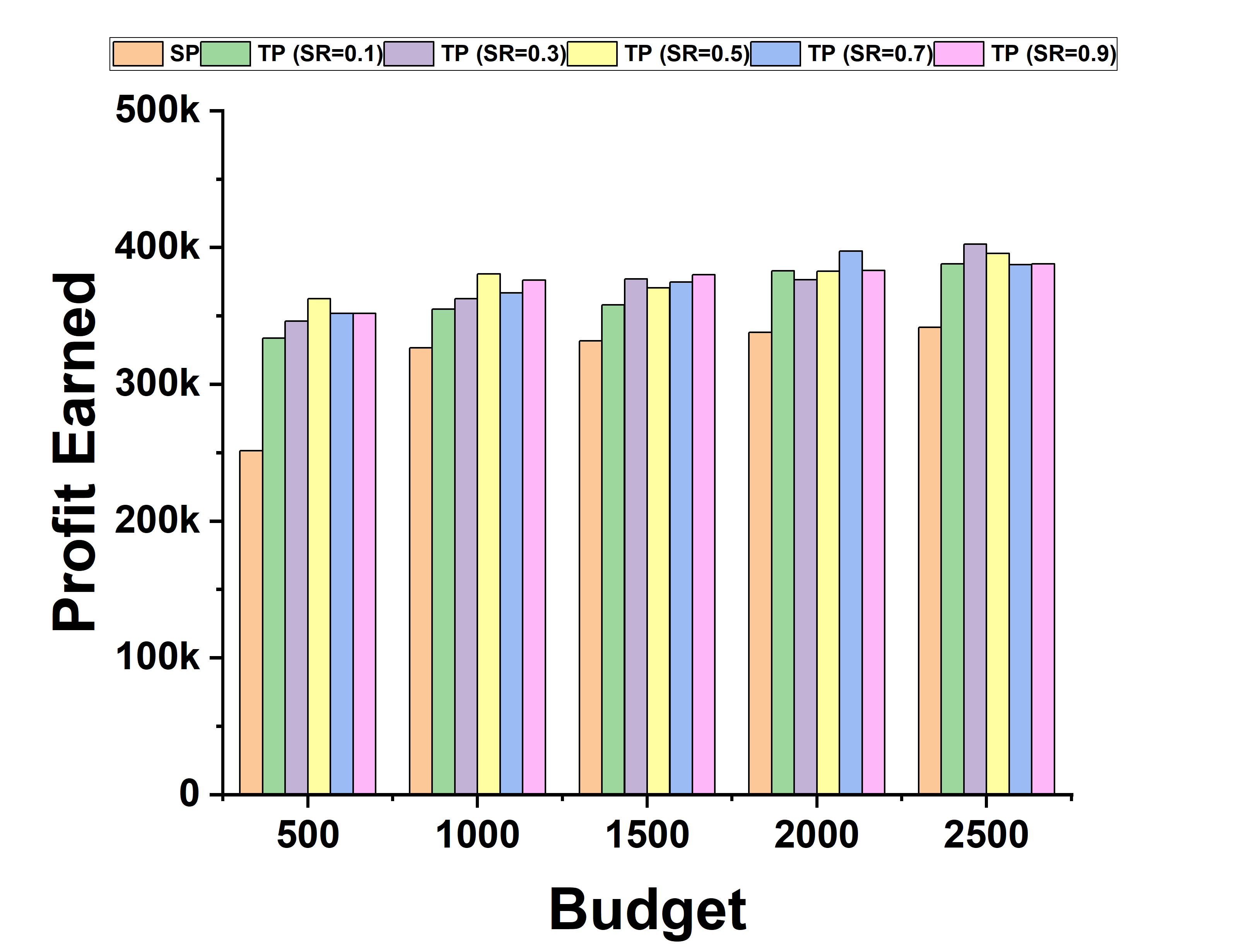}
        \caption{Random}
    \end{subfigure} &
    \begin{subfigure}[t]{0.22\textwidth}
        \includegraphics[width=\linewidth]{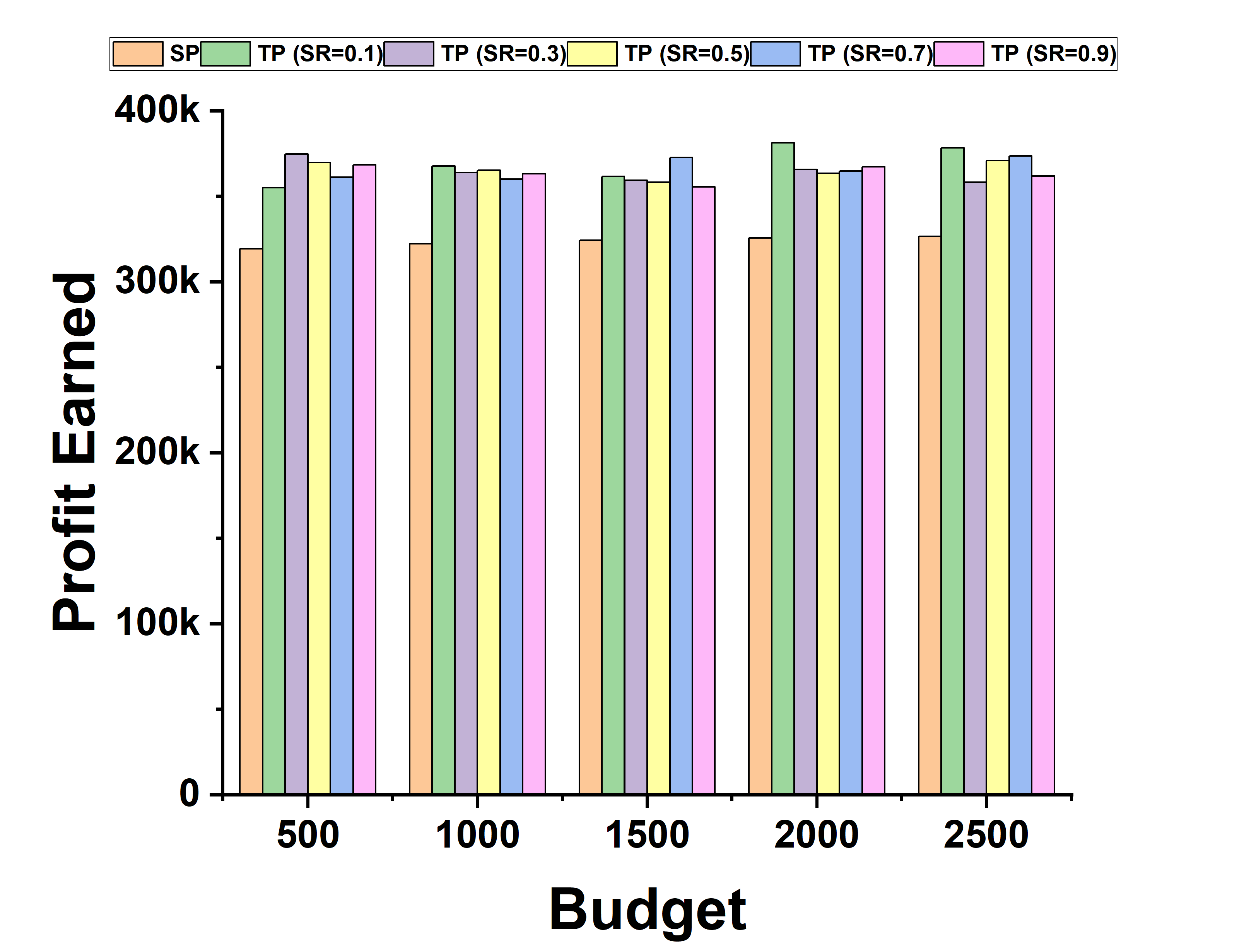}
        \caption{High Degree}
    \end{subfigure} &
    \begin{subfigure}[t]{0.22\textwidth}
        \includegraphics[width=\linewidth]{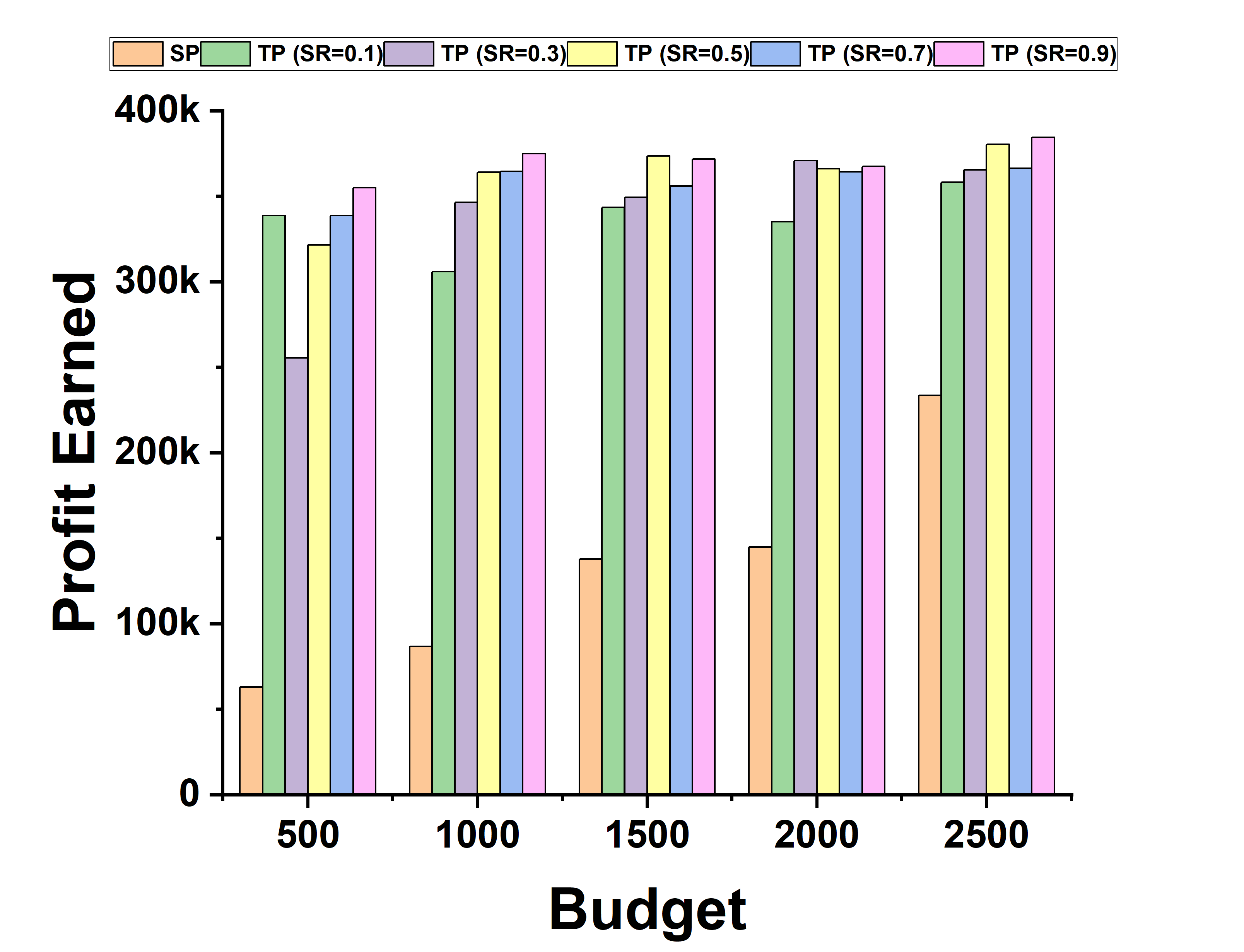}
        \caption{Clustering\\Coefficient}
    \end{subfigure} &
    \begin{subfigure}[t]{0.22\textwidth}
        \includegraphics[width=\linewidth]{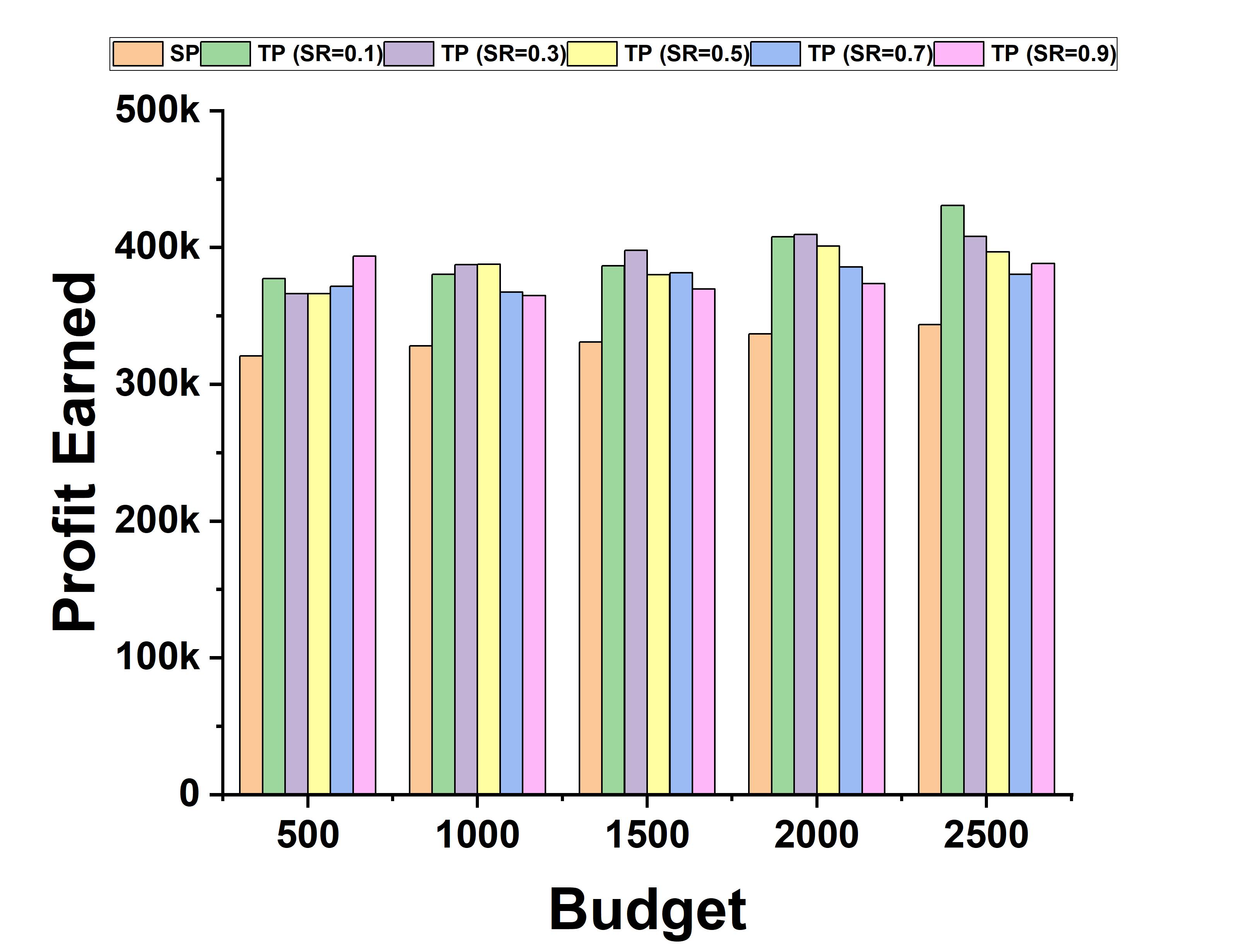}
        \caption{Degree Discount}
    \end{subfigure} \\[6pt]

    \begin{subfigure}[t]{0.22\textwidth}
        \includegraphics[width=\linewidth]{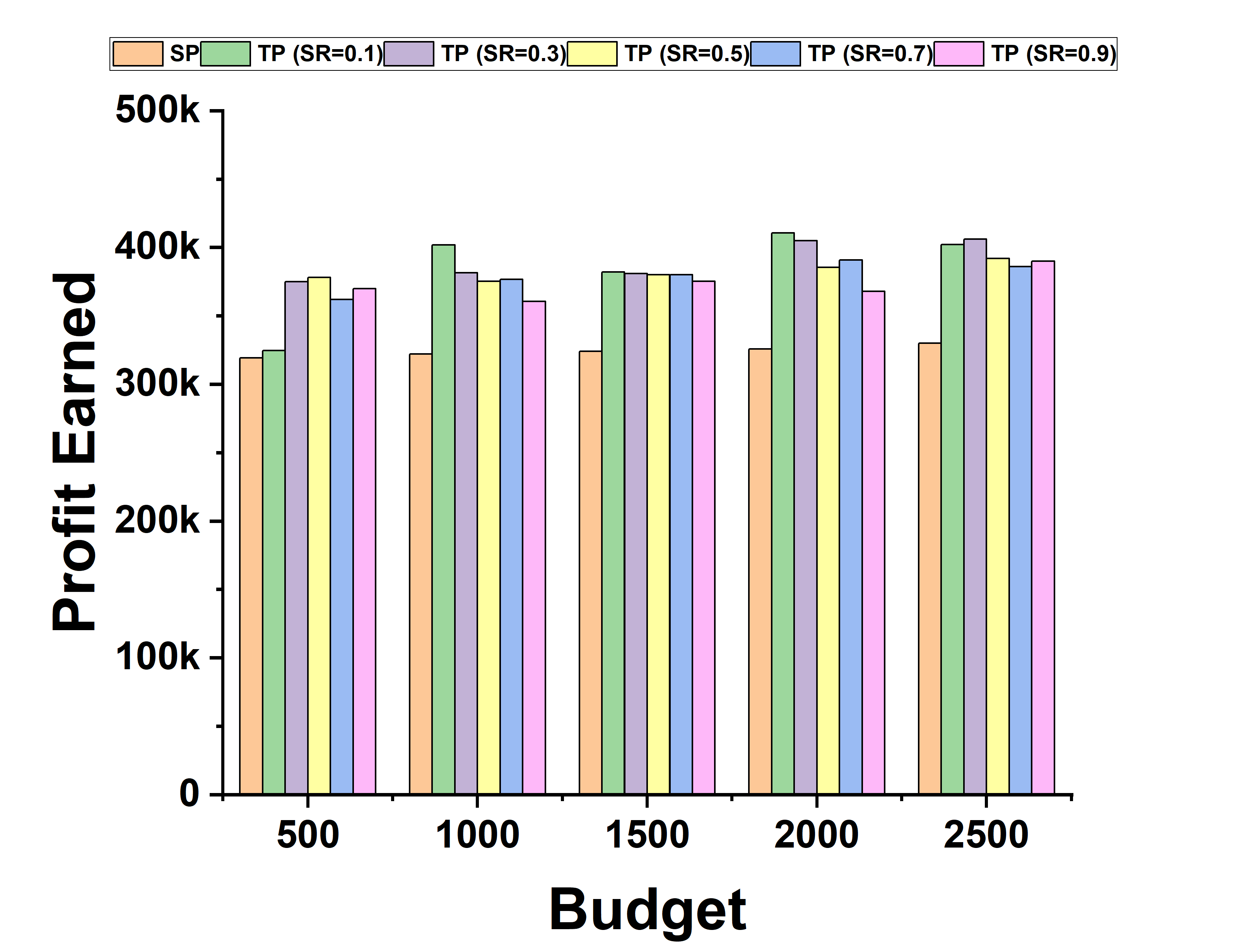}
        \caption{Single Discount}
    \end{subfigure} &
    \begin{subfigure}[t]{0.22\textwidth}
        \includegraphics[width=\linewidth]{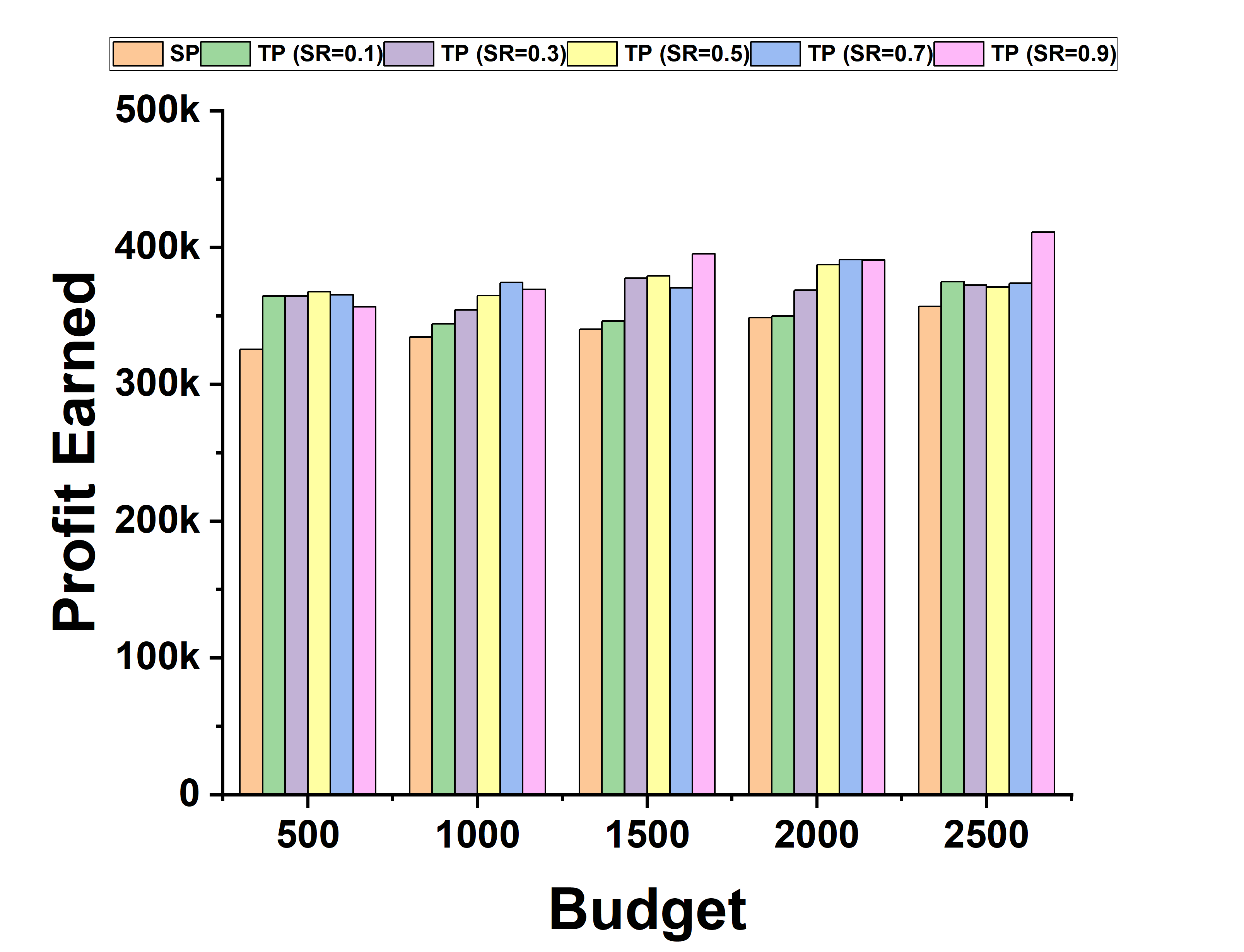}
        \caption{Simple Greedy}
    \end{subfigure} &
    \begin{subfigure}[t]{0.22\textwidth}
        \includegraphics[width=\linewidth]{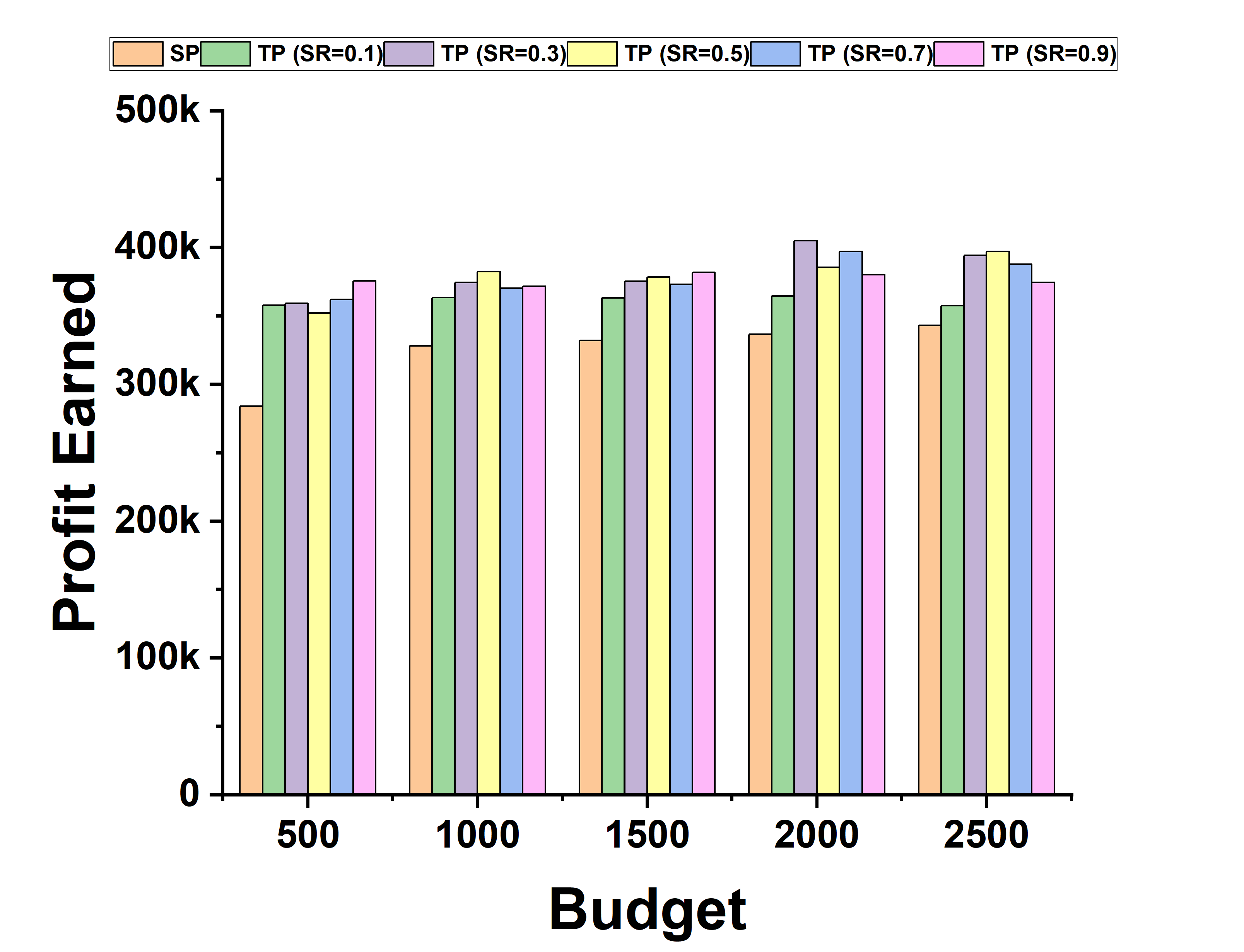}
        \caption{Double Greedy}
    \end{subfigure} &
    \begin{subfigure}[t]{0.22\textwidth}
        \includegraphics[width=\linewidth]{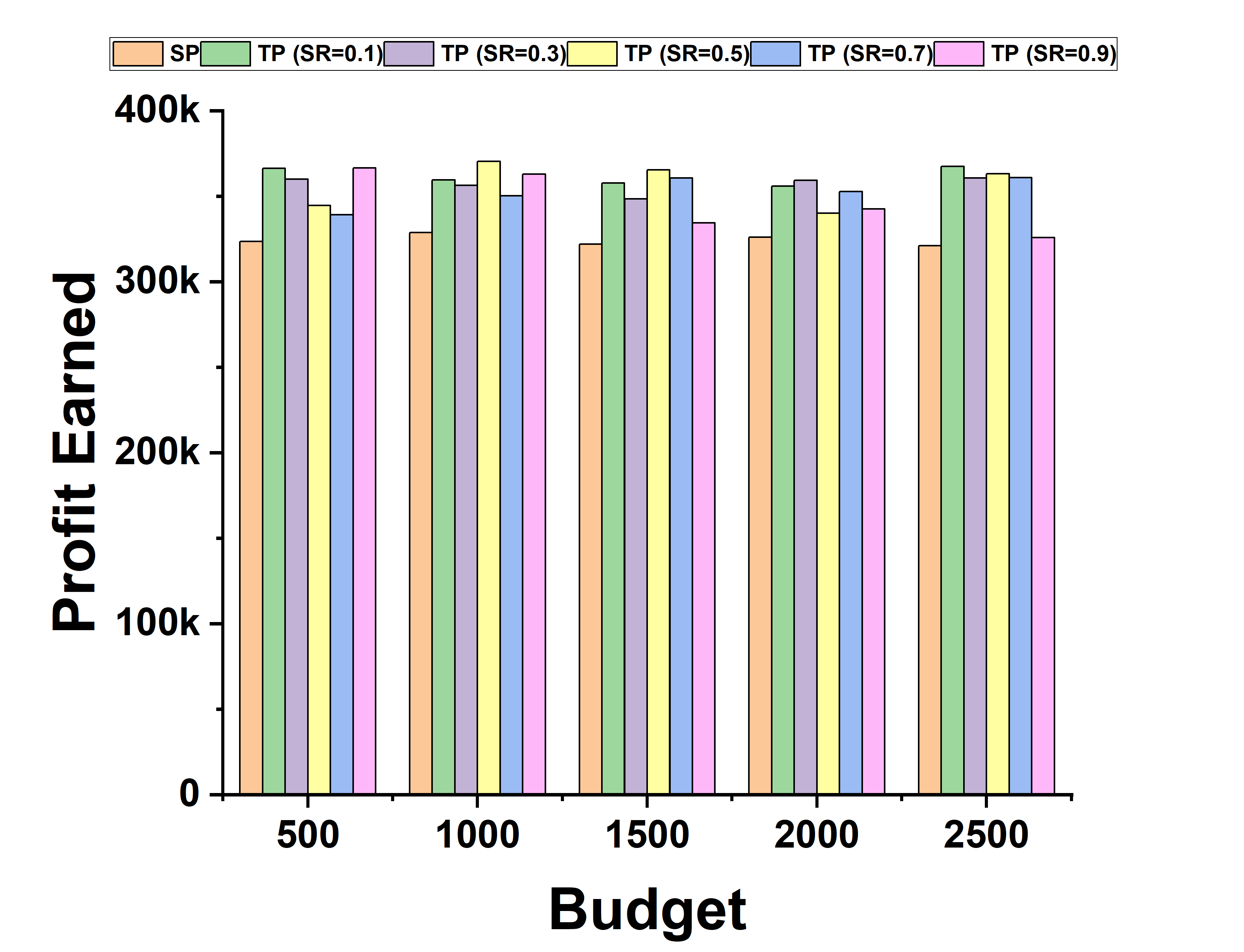}
        \caption{Stochastic Greedy}
    \end{subfigure}
\end{tabular}
\caption{Profit Earned in Single Phase Vs. Two Phase setting (Timestep 10, Probability Setting - Trivalency, \textit{Email-Eu-Core} Dataset)}
\label{Fig:RQ2_T5}
\end{figure}

\subsubsection{Impact of Timestep on Profit Earned in Two Phase Settings}
The results from the \textit{LM} dataset using the trivalency probability setting also show how the timestep—the delay before starting the second phase—impacts the success of the two phase approach. All algorithms follow a consistent pattern: starting the second phase later leads to better profit. Among baseline methods, \textbf{Random} clearly benefits from longer delays. For instance, with budget $500$, split ratio $0.7$, and timestep $10$, Random earned $12328.1$ (Figure~\ref{Fig:RQ1LM_T4}(a)). In comparison, at budget $2500$, split ratio $0.5$, and the same timestep, it reached $35747.20$ (Figure~\ref{Fig:RQ1LM_T3}(a)). \textbf{HD} also improved as timestep increased. It earned $9055.71$ at timestep $8$, budget $500$, and split ratio $0.1$ (Figure~\ref{Fig:RQ1LM_T1}(b)), and reached $35499.96$ at timestep $10$, budget $2500$, and split ratio $0.1$ (Figure~\ref{Fig:RQ1LM_T1}(b)). \textbf{HighCC} showed even bigger improvements. It earned $7016.9$ at timestep $2$, budget $500$, and split ratio $0.1$ (Figure~\ref{Fig:RQ1_T1}(c)). At timestep $8$, budget $2500$, and split ratio $0.5$, CC earned $32353.18$ (Figure~\ref{Fig:RQ1LM_T3}(c)). \textbf{DD} also gained more with longer delays. At timestep $10$, budget $500$, and split ratio $0.5$, \textbf{DD} earned $8861.85$ (Figure~\ref{Fig:RQ1LM_T3}(d)). At timestep $8$, budget $2500$, and split ratio $0.9$, \textbf{DD}’s profit increased to $34613$ (Figure~\ref{Fig:RQ1LM_T5}(d)). \textbf{SD} reached its best profit of $37960.78$ at timestep $8$ (Figure~\ref{Fig:RQ1LM_T5}(e)), confirming that allowing more time for the first phase leads to stronger overall results. Our proposed algorithms follow the same trend. \textbf{DG} showed consistent growth with timestep. For example, it reached its highest profit of $35491.86$ at timestep $4$, budget $2500$, and split ratio $0.3$ (Figure~\ref{Fig:RQ1LM_T2}(g)), which is higher than its single-phase value of $28168.89$. The \textbf{SG} algorithm, which can be more variable, earned $29140$ at timestep $2$. But with a longer delay—at timestep $8$, budget $2500$, and split ratio $0.9$ it improved to $29949$ (Figure~\ref{Fig:RQ1LM_T5}(f)). \textbf{StG0.1} followed the same pattern. At timestep $2$, with budget $500$ and split ratio $0.9$, it earned $8118$. At timestep $8$, using the same budget and split ratio, it reached $10343.97$ (Figure~\ref{Fig:RQ1LM_T5}(h)). To summarize, all results support the idea that delaying the second phase by increasing the timestep leads to better profit. A longer first phase allows influence to spread further before adding new seeds, which improves the overall outcome.

Figures~\ref{Fig:RQ1_T1} to \ref{Fig:RQ1_T5} show how profit changes at different timesteps for the \textit{Email-Eu-Core} dataset under trivalency probability setting. The results clearly show that most algorithms earn more profit when the timestep is increased. For the \textbf{Random} algorithm, profit generally increases as the timestep becomes larger. At timestep $2$, the profit ranges from $276552.58$ to $347820.62$. For example, with budget $500$ and split ratio $0.9$, \textbf{Random} earned $299968.36$ at timestep $2$ (Figure~\ref{Fig:RQ1_T5}(a)). When the same budget and split ratio were used at timestep $10$, the profit increased to $351635$. The \textbf{HD} algorithm also shows a strong increase in profit with higher timesteps. At timestep $2$, profit values are between $163208.12$ and $299453.12$, but this rises to between $354999$ and $381091.27$ at timestep $10$ (Figure~\ref{Fig:RQ1_T1}(b)). The \textbf{HighCC} algorithm also benefits from longer timesteps. For example, at timestep $2$, budget $500$, and split ratio $0.9$, \textbf{HighCC} earned only $25554.23$ (Figure~\ref{Fig:RQ1_T5}(c)). But at timestep $10$ with the same settings, profit increased to $354839$. \textbf{DD} also performs better with more time. At timestep $2$, budget $500$, and split ratio $0.9$, \textbf{DD} earned $275871.65$ (Figure~\ref{Fig:RQ1_T5}(d)). At timestep $10$, profit increased to $393677$. The \textbf{DG} algorithm shows consistent profit growth with timestep. For example, at timestep $2$, with budget $500$ and split ratio $0.1$, \textbf{DG} earned $328196.79$ (Figure~\ref{Fig:RQ1_T1}(g)). At timestep $10$, using budget $500$ and split ratio $0.9$, it earned $375395.27$. \textbf{StG0.1} also shows a clear increase in profit with higher timesteps. At timestep $2$, with budget $500$ and split ratio $0.9$, the profit was $321930.27$ (Figure~\ref{Fig:RQ1_T5}(f)), and it rose to $366407.54$ at timestep $10$ with the same settings. Overall, increasing the timestep in the two phase setting has a strong and consistent positive effect on profit for nearly all algorithms. More time for diffusion in the first phase allows greater spread, leading to better final results. \textbf{StG0.1} is especially notable for showing both steady gains and reliable outcomes as timestep increases.

\subsubsection{Comparison of Seed Set Sizes Between Single and Two phase Strategies}
For the \textit{LM} dataset under the trivalency probability setting, the impact of two phase settings on seed set size varied across algorithms. Baseline methods such as \textbf{Random} and \textbf{HD} showed small increases in seed count. For example, with budget $500$, split ratio $0.1$, and timestep $2$, \textbf{Random} increased its seed size from $6$ to $8$ (Figure~\ref{RQ4LM_T1}(a)). Similarly, \textbf{HD}’s size increased from $6$ to $7$ under all tested timesteps for the same budget, as seen in Figure~\ref{RQ4LM_T2}. \textbf{HighCC}, \textbf{DD}, and \textbf{SD} remained stable. \textbf{HighCC} showed almost no change except for one case where the seed size increased by $1$ node at budget $1500$, split ratio $0.3$, and timestep $2$ (Figure~\ref{RQ4LM_T3}(a)). \textbf{DD} and \textbf{SD} followed the same trend (Figures~\ref{RQ4LM_T4} and \ref{RQ4LM_T5}), with changes limited to $1$ or $2$ additional nodes. In contrast, our proposed algorithms—especially \textbf{SG} and \textbf{StG0.1}—frequently selected smaller but more effective seed sets. For instance, \textbf{SG} reduced the seed set by $3$ nodes at budget $2500$, split ratio $0.9$, and timestep $2$, while still earning higher profit (Figure~\ref{RQ4LM_T6}(a)). Similarly, \textbf{StG0.1} selected $1$ fewer node at budget $500$, split ratio $0.9$, and timestep $8$ (Figure~\ref{RQ4LM_T8}(d)) and still performed better. \textbf{DG} showed mixed behavior. At lower budgets, it sometimes selected fewer or equal nodes. For example, at budget $500$, split ratio $0.5$, and timestep $2$, \textbf{DG} reduced the seed set by $1$ (Figure~\ref{RQ4LM_T7}(a)). However, at budget $1500$, split ratio $0.1$, and timestep $2$, it increased the seed count by $3$ nodes (Figure~\ref{RQ4LM_T7}(a)). In summary, the effect of the two phase setting on seed set size varies by algorithm. Baseline methods tend to keep seed size the same or slightly increase it, while proposed methods like \textbf{SG} and \textbf{StG0.1} often identify smaller, more efficient seed sets that still result in higher profit.

For the \textit{Email-Eu-Core} dataset, Figures~\ref{RQ4_T1} to \ref{RQ4_T8} compare the number of seed nodes selected in the single-phase and two phase settings under trivalency probability setting. For the \textbf{Random} algorithm, the seed set size in the two phase setting ranged from $6$ to $9$ at a budget of $500$. For example, it selected $7$ nodes with a split ratio of $0.1$ and timestep $2$ (Figure~\ref{RQ4_T1}(a)), compared to $6$ in the single-phase. At a higher budget of $2500$, the two phase seed set ranged from $32$ to $36$ nodes. One instance shows a size of $36$ at split ratio $0.5$ and timestep $4$ (Figure~\ref{RQ4_T1}(b)), slightly higher than the single-phase size of $33$. For \textbf{HD}, the two phase sizes were also slightly higher. At budget $500$, seed set sizes ranged from $6$ to $9$. One configuration with split ratio $0.1$ and timestep $2$ gave a size of $7$, compared to $6$ in the single-phase (Figure~\ref{RQ4_T2}(a)). At budget $2500$, sizes reached up to $36$, with $36$ nodes observed at split ratio $0.7$ and timestep $10$ (Figure~\ref{RQ4_T2}(e)), compared to a single-phase size of $33$. \textbf{HighCC} usually showed similar or slightly larger two phase seed sizes. At budget $500$, the size ranged between $9$ and $11$. For example, Figure~\ref{RQ4_T3}(a) shows a two phase size of $9$ at split ratio $0.1$ and timestep $2$, matching the single-phase value. Another configuration had a size of $10$ at split ratio $0.5$ and timestep $2$. At budget $2500$, the sizes ranged from $37$ to $43$. A size of $42$ was seen with split ratio $0.9$ and timestep $10$ (Figure~\ref{RQ4_T3}(e)). For \textbf{DD}, the two phase seed sizes were typically close to or slightly above single-phase values. At budget $500$, sizes ranged from $6$ to $9$. One case had $6$ nodes at split ratio $0.1$ and timestep $2$ (Figure~\ref{RQ4_T4}(a)), matching the single-phase size. In another case, $9$ nodes were selected at split ratio $0.9$ and timestep $10$ (Figure~\ref{RQ4_T4}(e)). \textbf{SD} followed a similar pattern. At budget $500$, seed set sizes in two phase ranged from $6$ to $9$. One instance (Figure~\ref{RQ4_T5}(a)) had a size of $6$ with split ratio $0.1$ and timestep $2$, same as the single-phase size. Another example showed $9$ nodes selected with split ratio $0.5$ and timestep $10$ (Figure~\ref{RQ4_T5}(e)). \textbf{SG} generally chose slightly more seed nodes in two phase settings. At a budget of $2500$, seed sets ranged from $43$ to $54$ nodes. For example, a size of $44$ was recorded at split ratio $0.1$ and timestep $2$, which matched the single-phase size, while $54$ nodes were chosen at split ratio $0.7$ and timestep $8$ (Figure~\ref{RQ4_T6}(d)). \textbf{DG} usually showed similar or slightly higher sizes in two phase. At budget $500$, sizes varied between $7$ and $10$. For instance, $7$ nodes were selected with split ratio $0.1$ and timestep $2$, same as single-phase (Figure~\ref{RQ4_T7}(a)), while another setting recorded $10$ nodes at split ratio $0.5$ and timestep $8$ (Figure~\ref{RQ4_T7}(d)). In contrast, \textbf{StG0.1} often had smaller seed sets in the two phase setup, especially at higher budgets. At budget $2500$, the size was between $3$ and $4$. In one instance with split ratio $0.1$ and timestep $2$, it selected $3$ nodes (Figure~\ref{RQ4_T8}(a)), fewer than the single-phase size of $4$. Overall, most algorithms showed similar or slightly larger seed sets in the two phase setting compared to the single-phase, with some exceptions like \textbf{StG0.1} showing reductions.


\begin{figure}[htbp]
    \centering
    \begin{subfigure}[t]{0.3\linewidth}
        \centering
        \includegraphics[width=\linewidth]{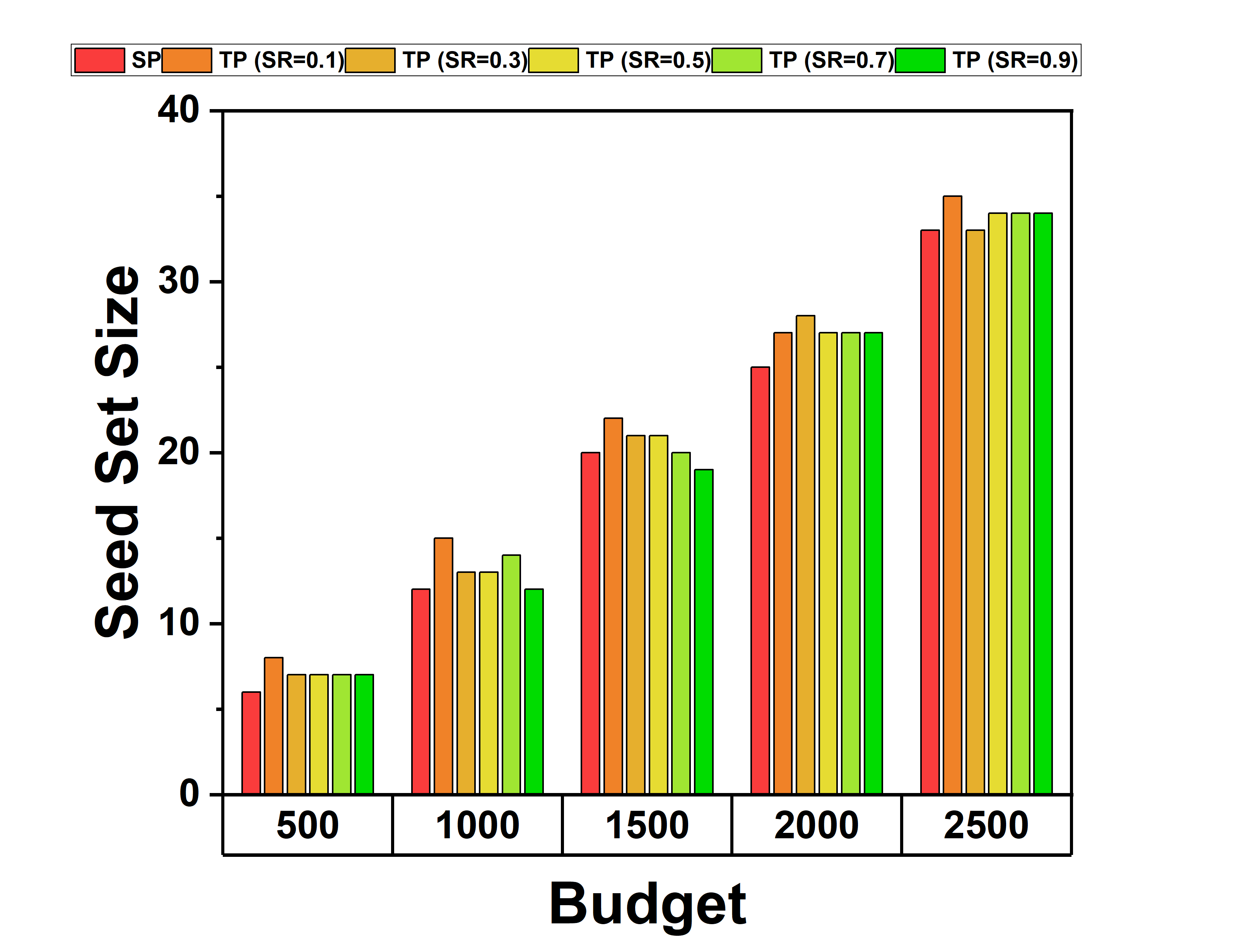}
        \caption{Timestep 2}
    \end{subfigure}
    \hspace{0.05\linewidth}
    \begin{subfigure}[t]{0.3\linewidth}
        \centering
        \includegraphics[width=\linewidth]{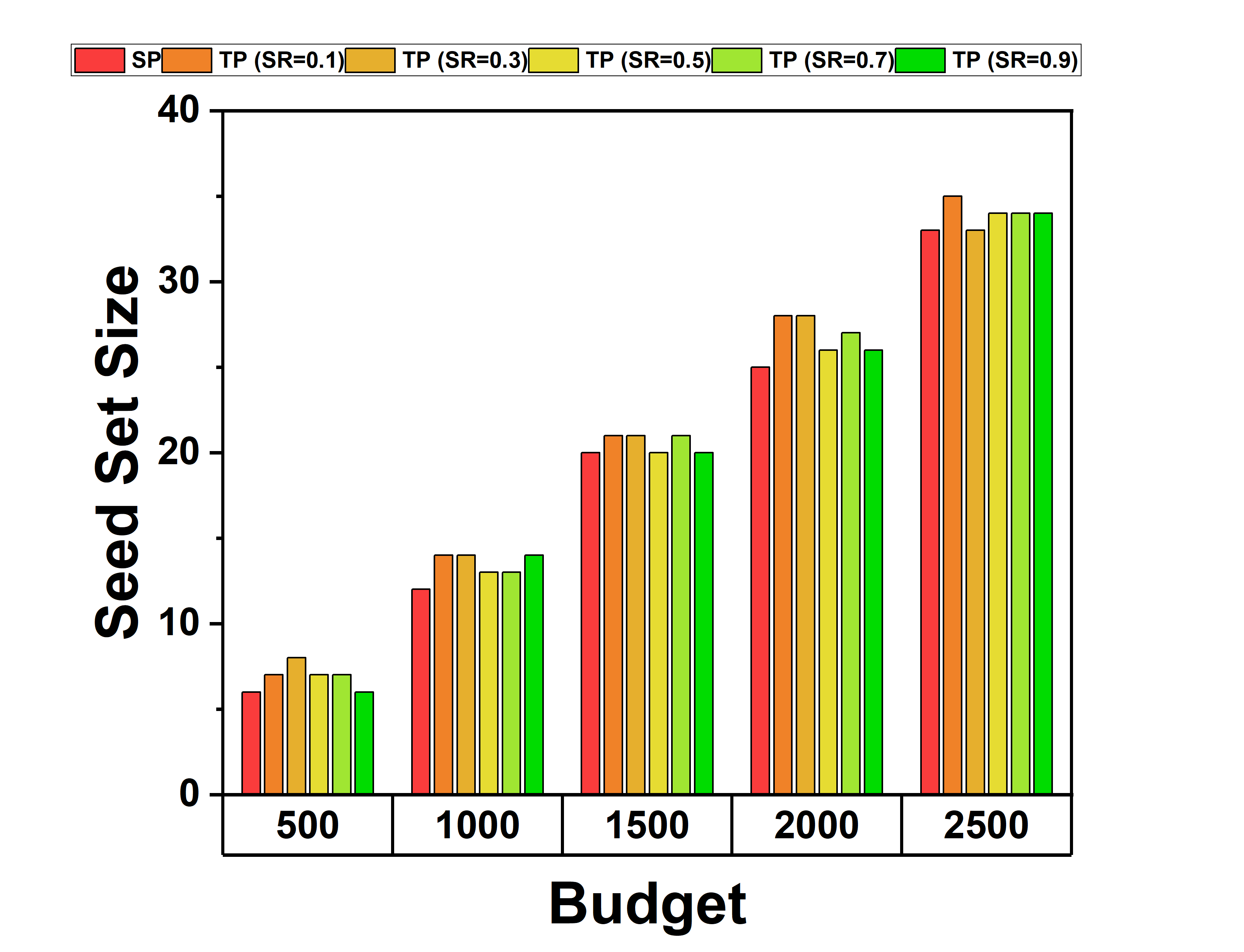}
        \caption{Timestep 4}
    \end{subfigure}

    \vspace{0.5cm}

    \begin{subfigure}[t]{0.3\linewidth}
        \centering
        \includegraphics[width=\linewidth]{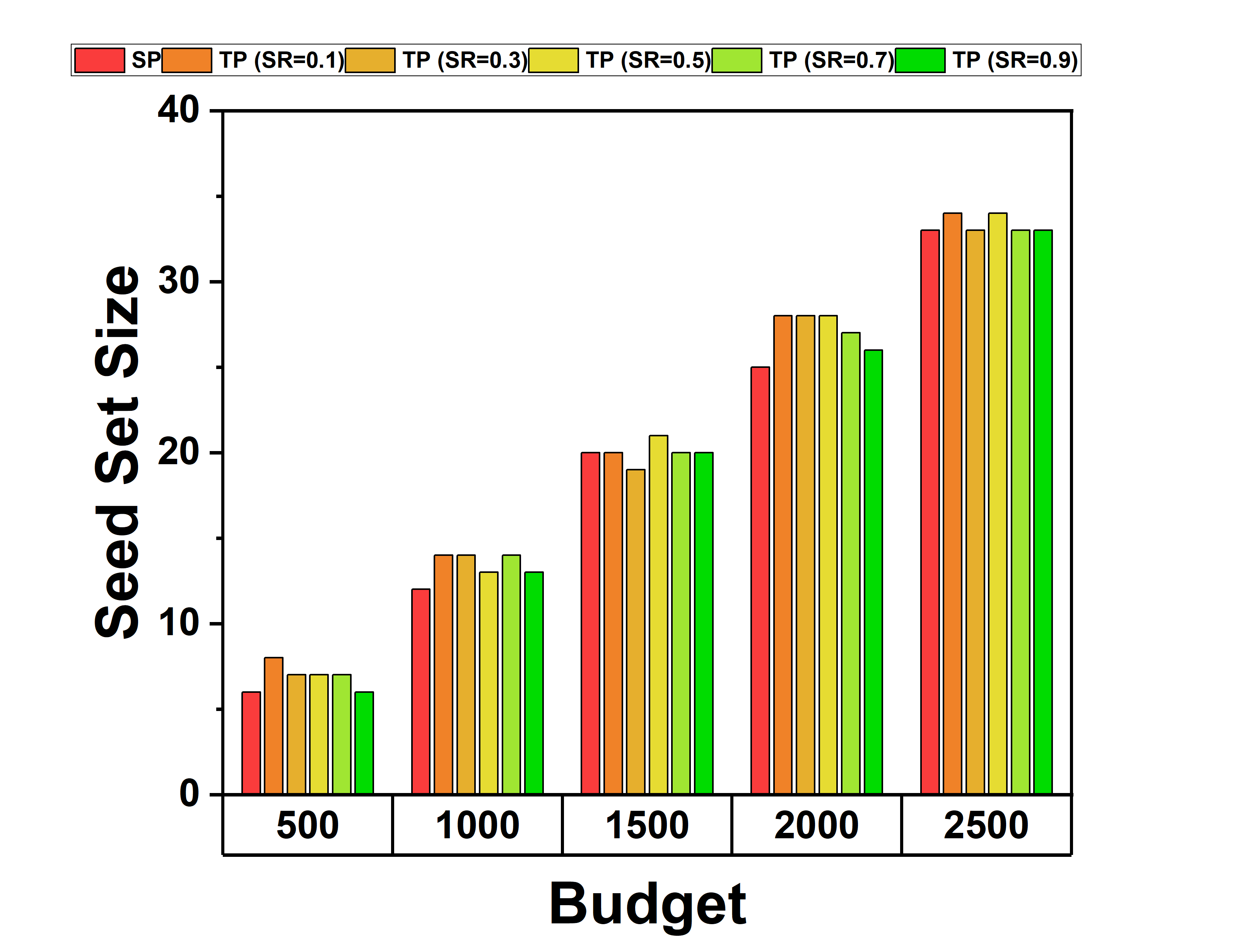}
        \caption{Timestep 6}
    \end{subfigure}
    \hfill
    \begin{subfigure}[t]{0.3\linewidth}
        \centering
        \includegraphics[width=\linewidth]{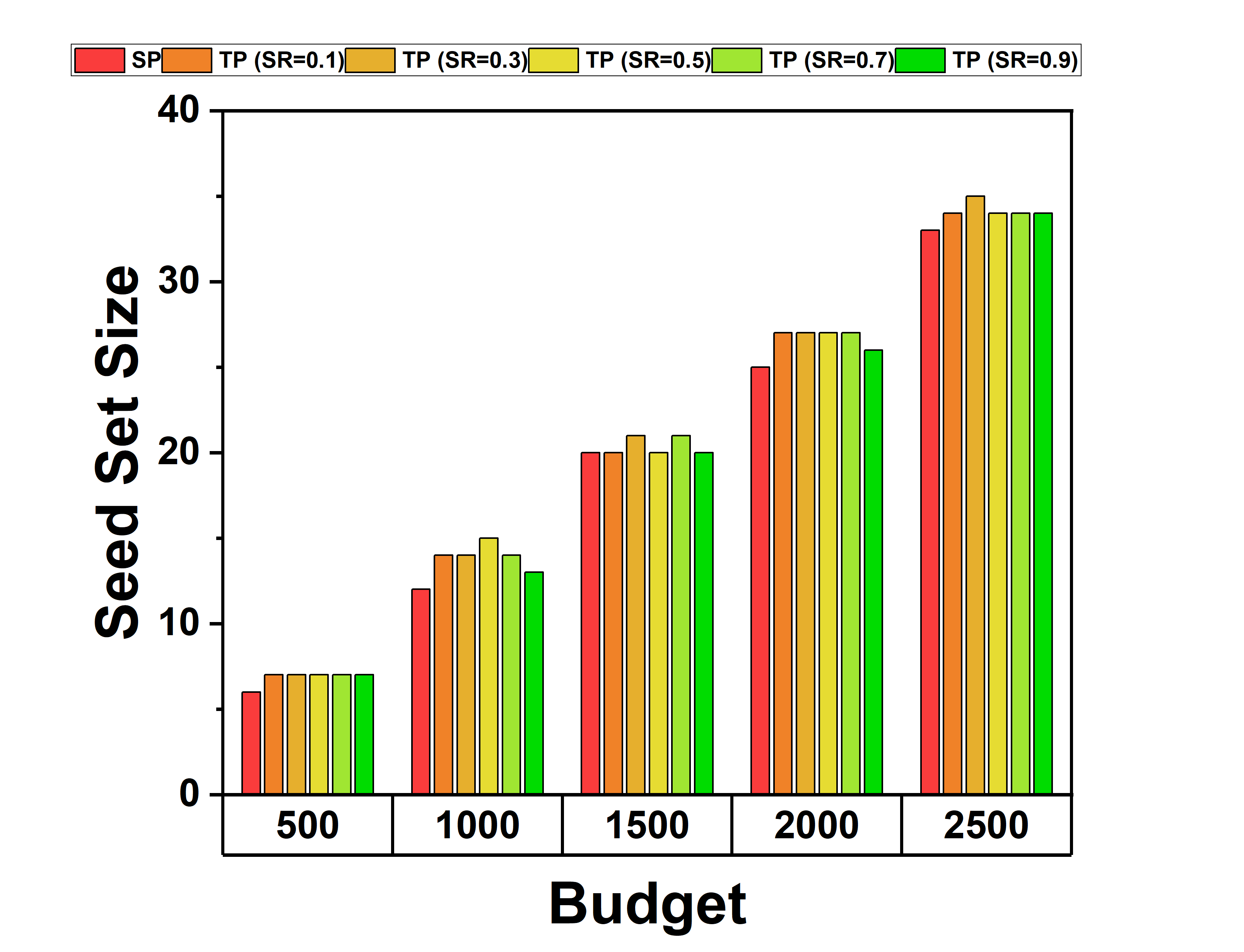}
        \caption{Timestep 8}
    \end{subfigure}
    \hfill
    \begin{subfigure}[t]{0.3\linewidth}
        \centering
        \includegraphics[width=\linewidth]{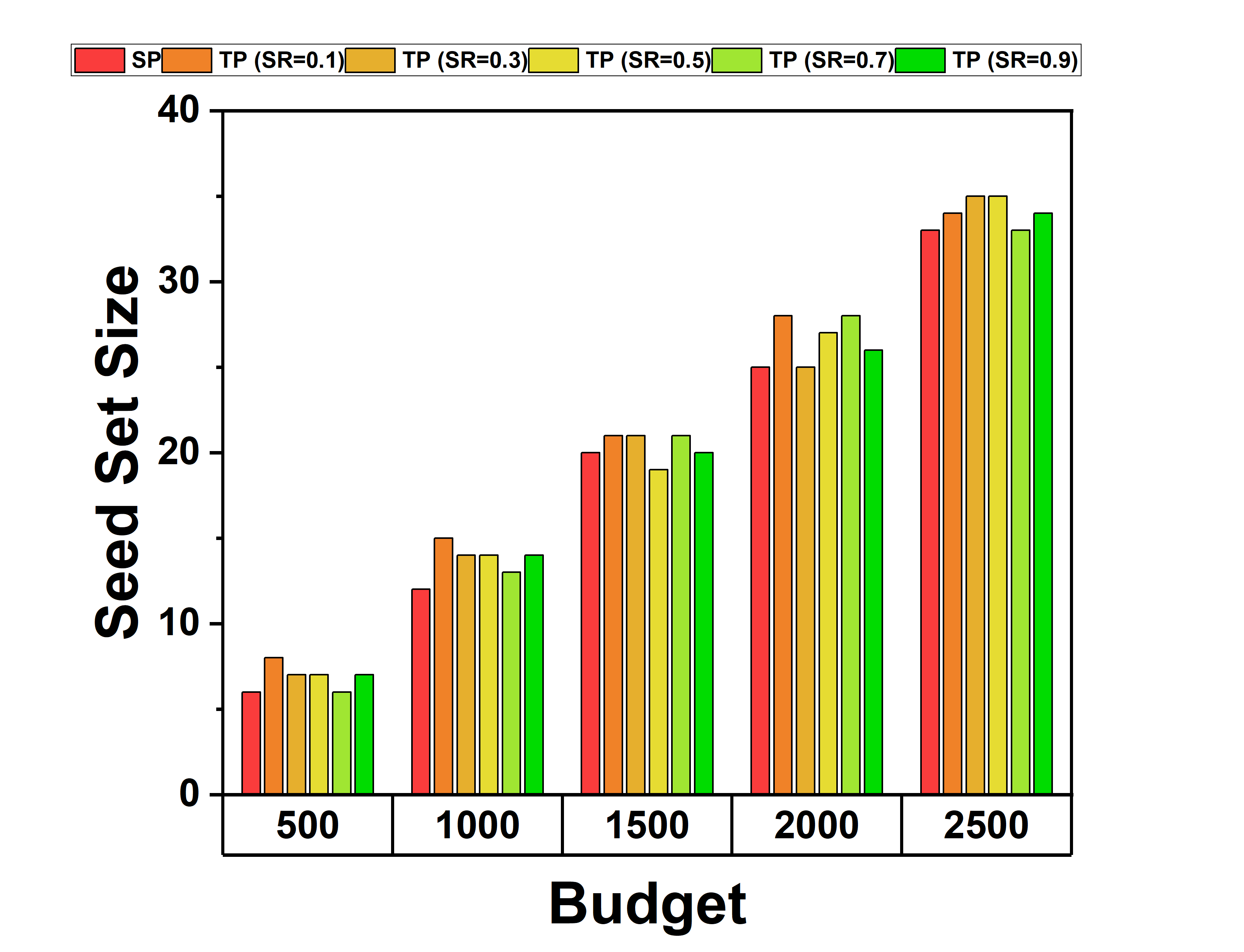}
        \caption{Timestep 10}
    \end{subfigure}

    \caption{Seed Set Size Distribution of Single Phase Vs. Two Phase (Random Algorithm, \textit{LM} Dataset, Probability Setting - Trivalency)}
    \label{RQ4LM_T1}
\end{figure}

\begin{figure}[htbp]
    \centering
    \begin{subfigure}[t]{0.3\linewidth}
        \centering
        \includegraphics[width=\linewidth]{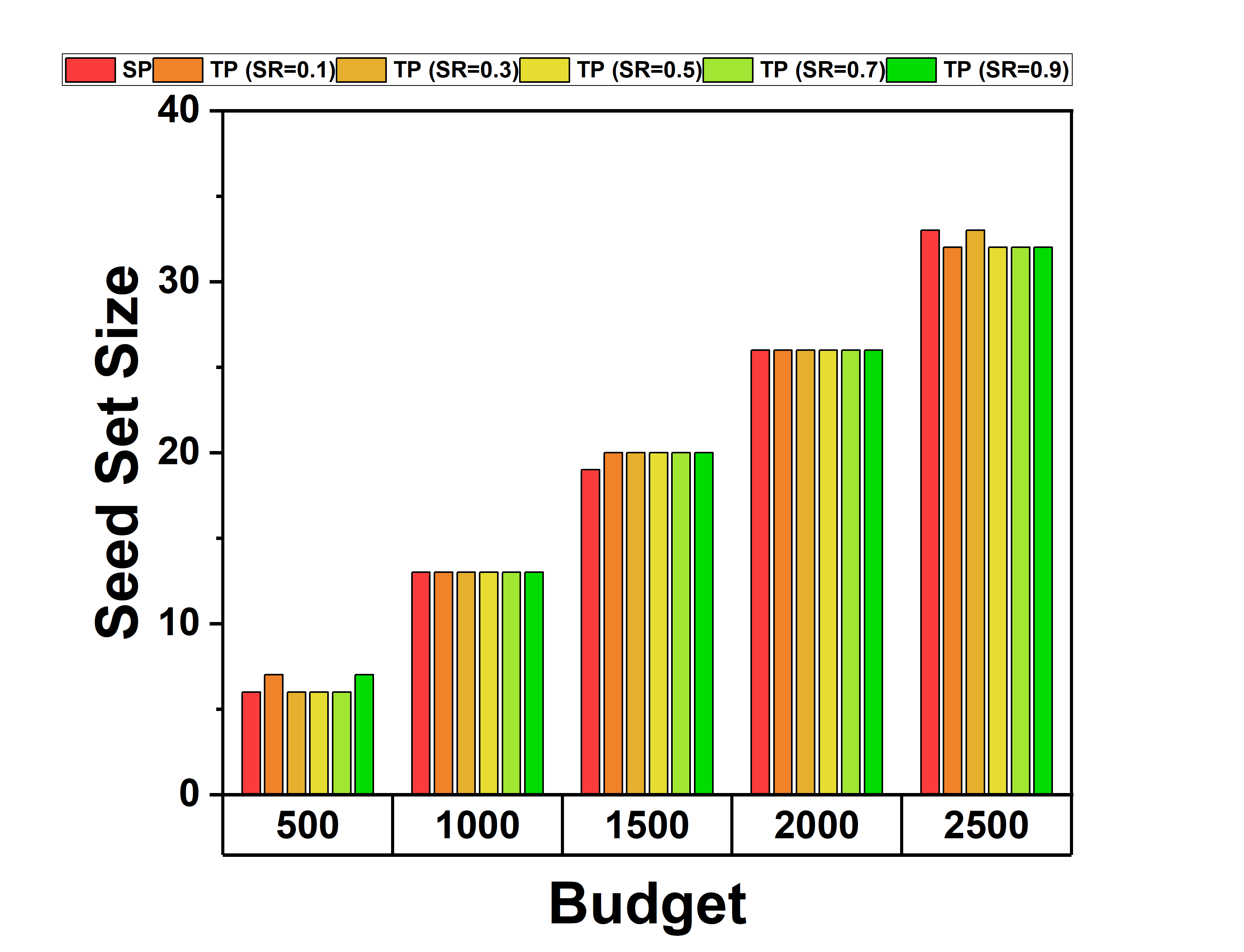}
        \caption{Timestep 2}
    \end{subfigure}
    \hspace{0.05\linewidth}
    \begin{subfigure}[t]{0.3\linewidth}
        \centering
        \includegraphics[width=\linewidth]{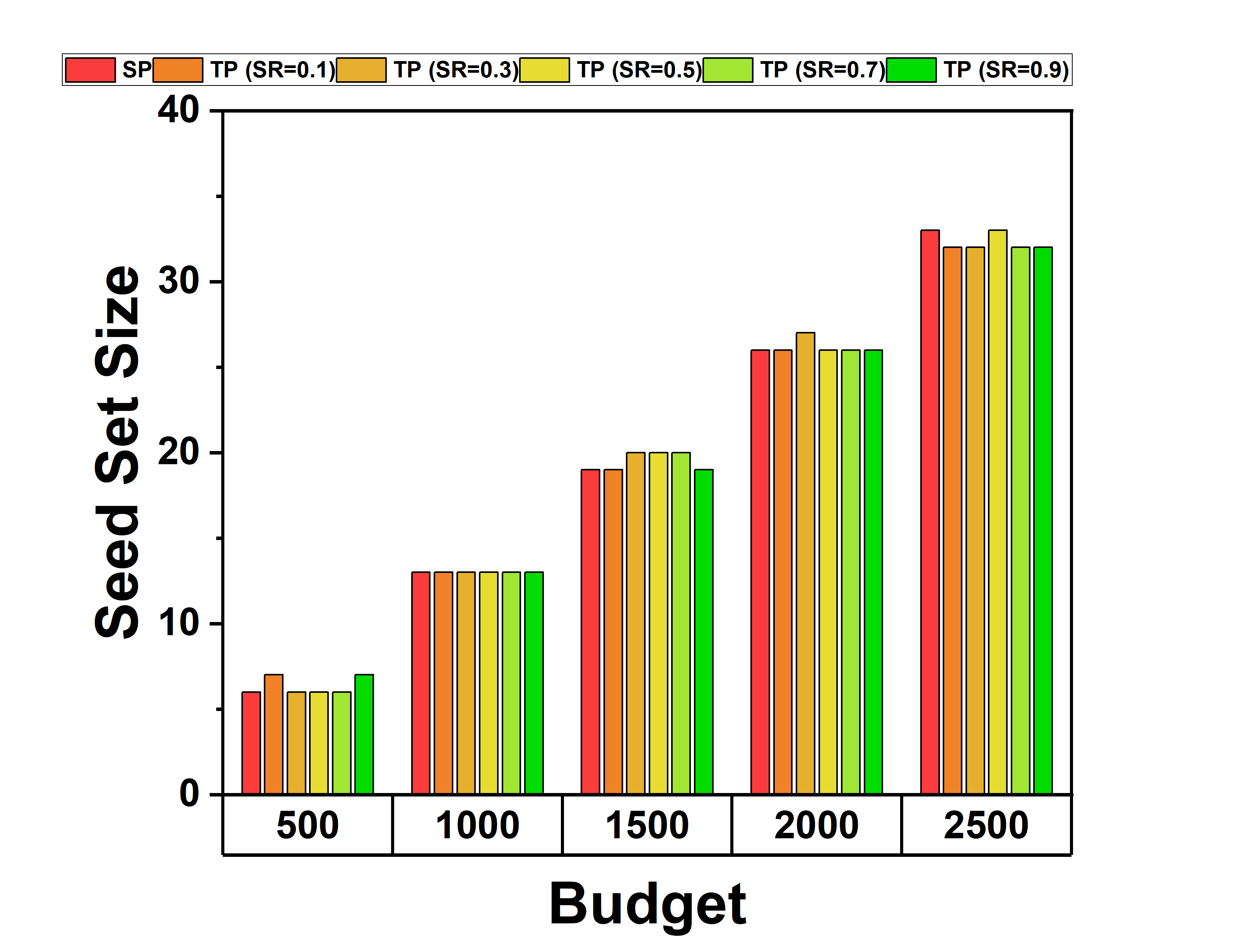}
        \caption{Timestep 4}
    \end{subfigure}

    \vspace{0.5cm}

    \begin{subfigure}[t]{0.3\linewidth}
        \centering
        \includegraphics[width=\linewidth]{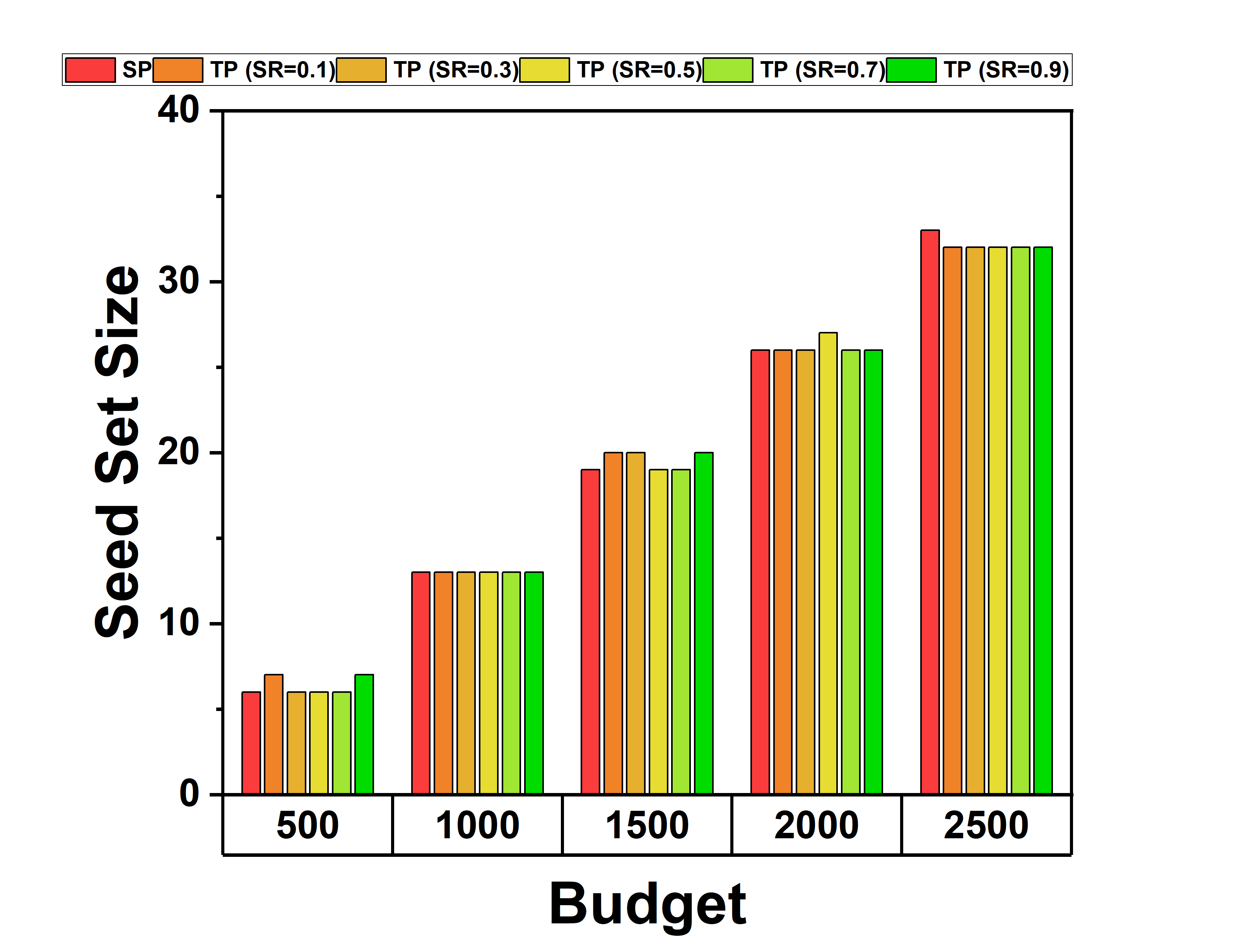}
        \caption{Timestep 6}
    \end{subfigure}
    \hfill
    \begin{subfigure}[t]{0.3\linewidth}
        \centering
        \includegraphics[width=\linewidth]{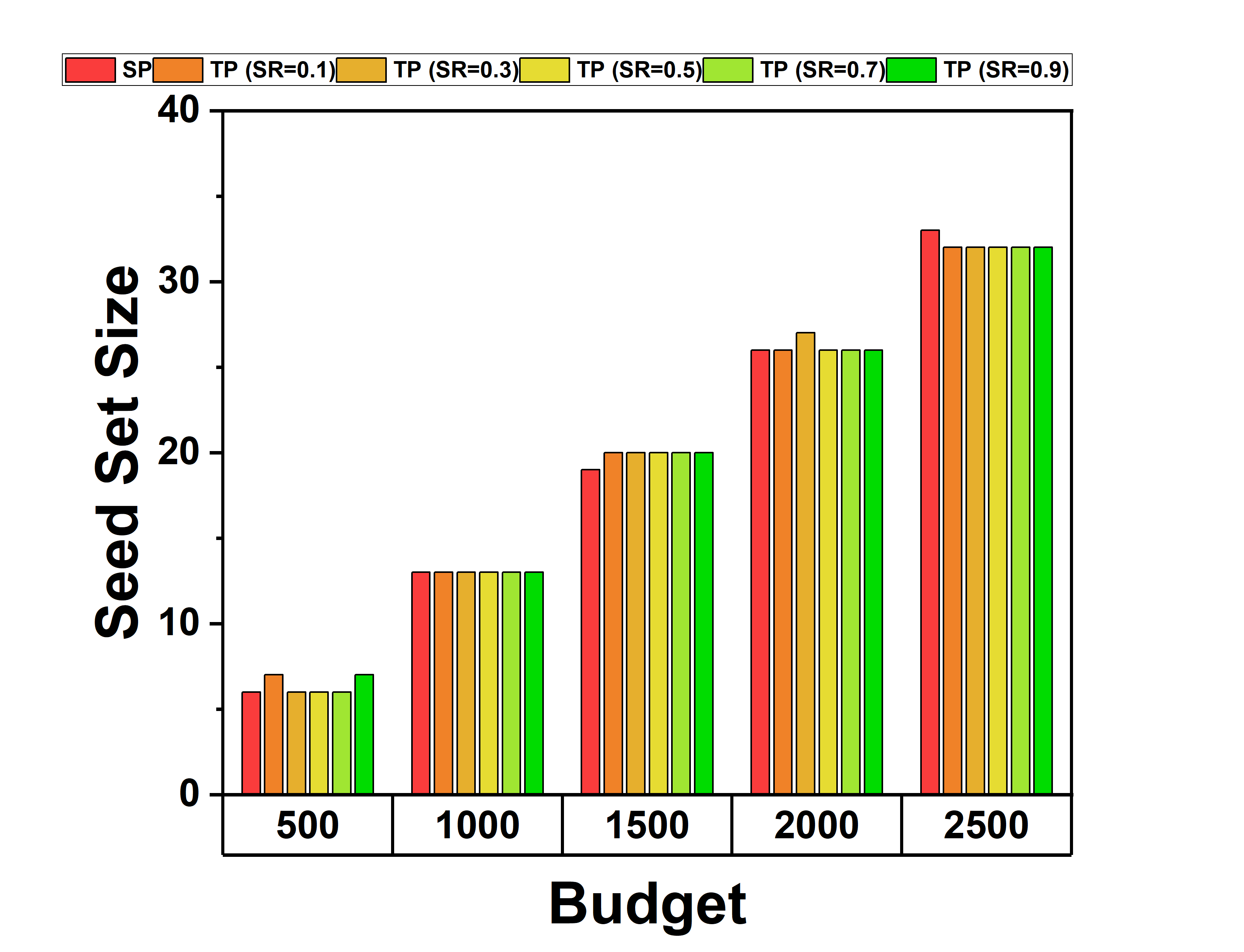}
        \caption{Timestep 8}
    \end{subfigure}
    \hfill
    \begin{subfigure}[t]{0.3\linewidth}
        \centering
        \includegraphics[width=\linewidth]{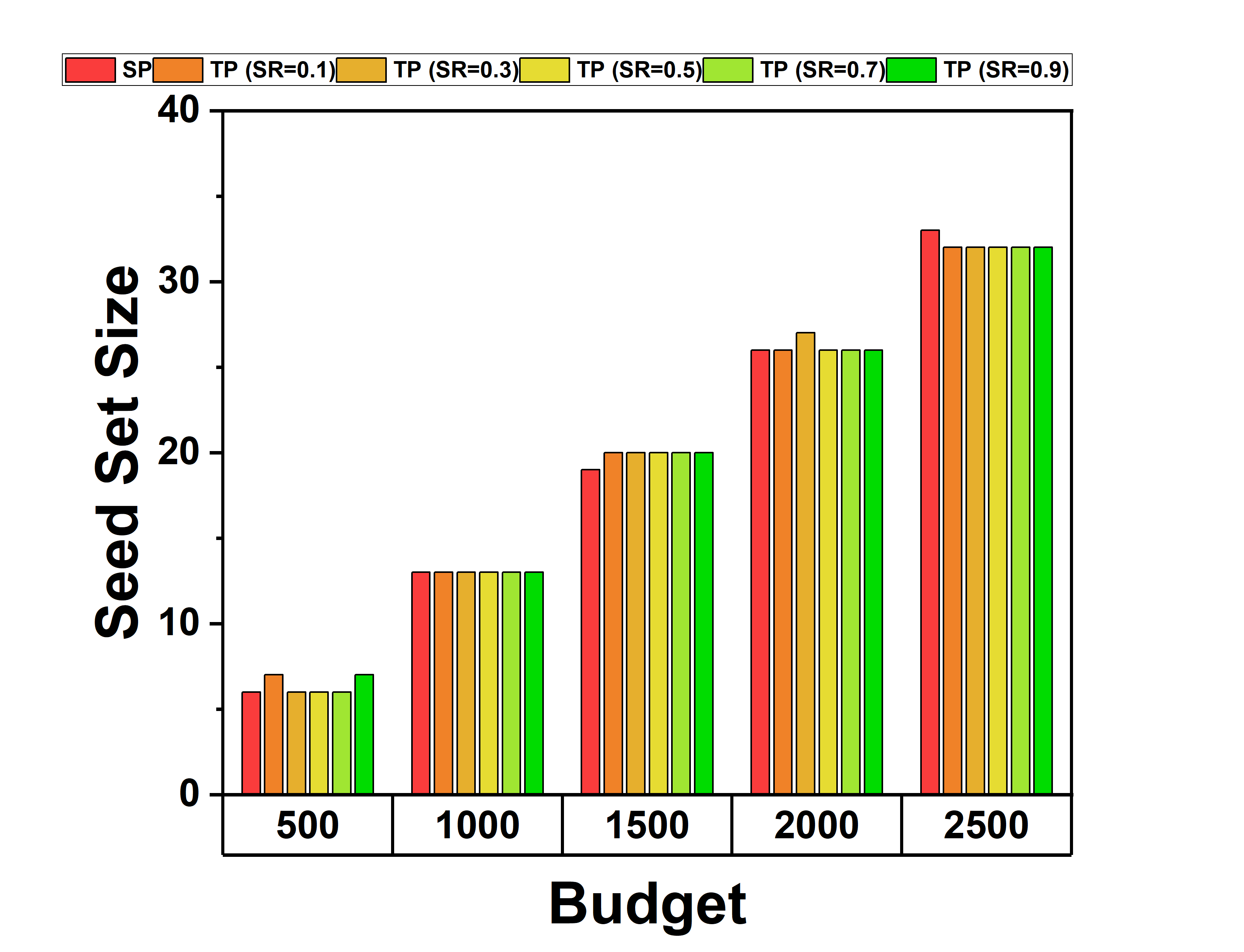}
        \caption{Timestep 10}
    \end{subfigure}

    \caption{Seed Set Size Distribution of Single Phase Vs. Two Phase (High Degree Algorithm, \textit{LM} Dataset, Probability Setting - Trivalency)}
    \label{RQ4LM_T2}
\end{figure}

\begin{figure}[htbp]
    \centering
    \begin{subfigure}[t]{0.3\linewidth}
        \centering
        \includegraphics[width=\linewidth]{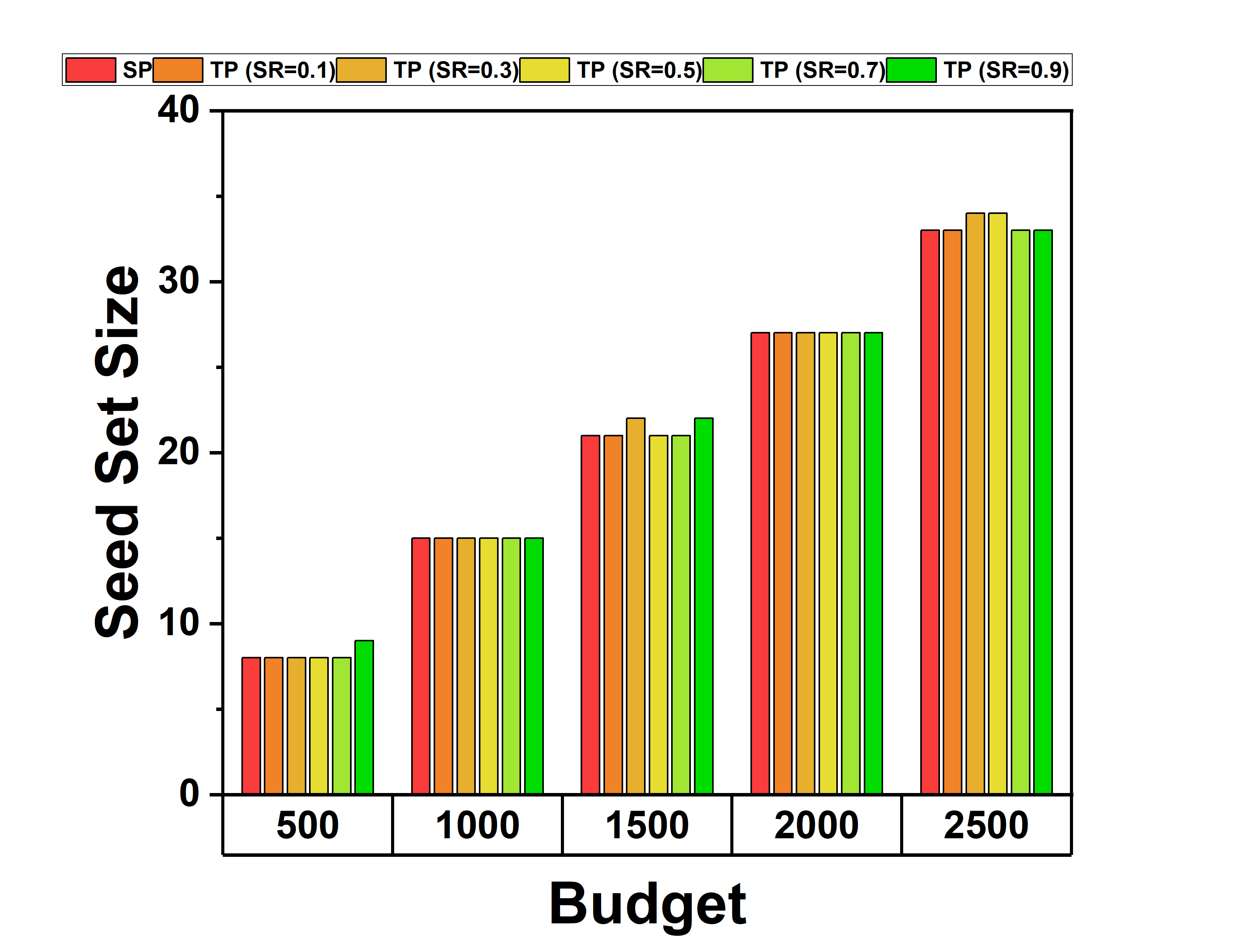}
        \caption{Timestep 2}
    \end{subfigure}
    \hspace{0.05\linewidth}
    \begin{subfigure}[t]{0.3\linewidth}
        \centering
        \includegraphics[width=\linewidth]{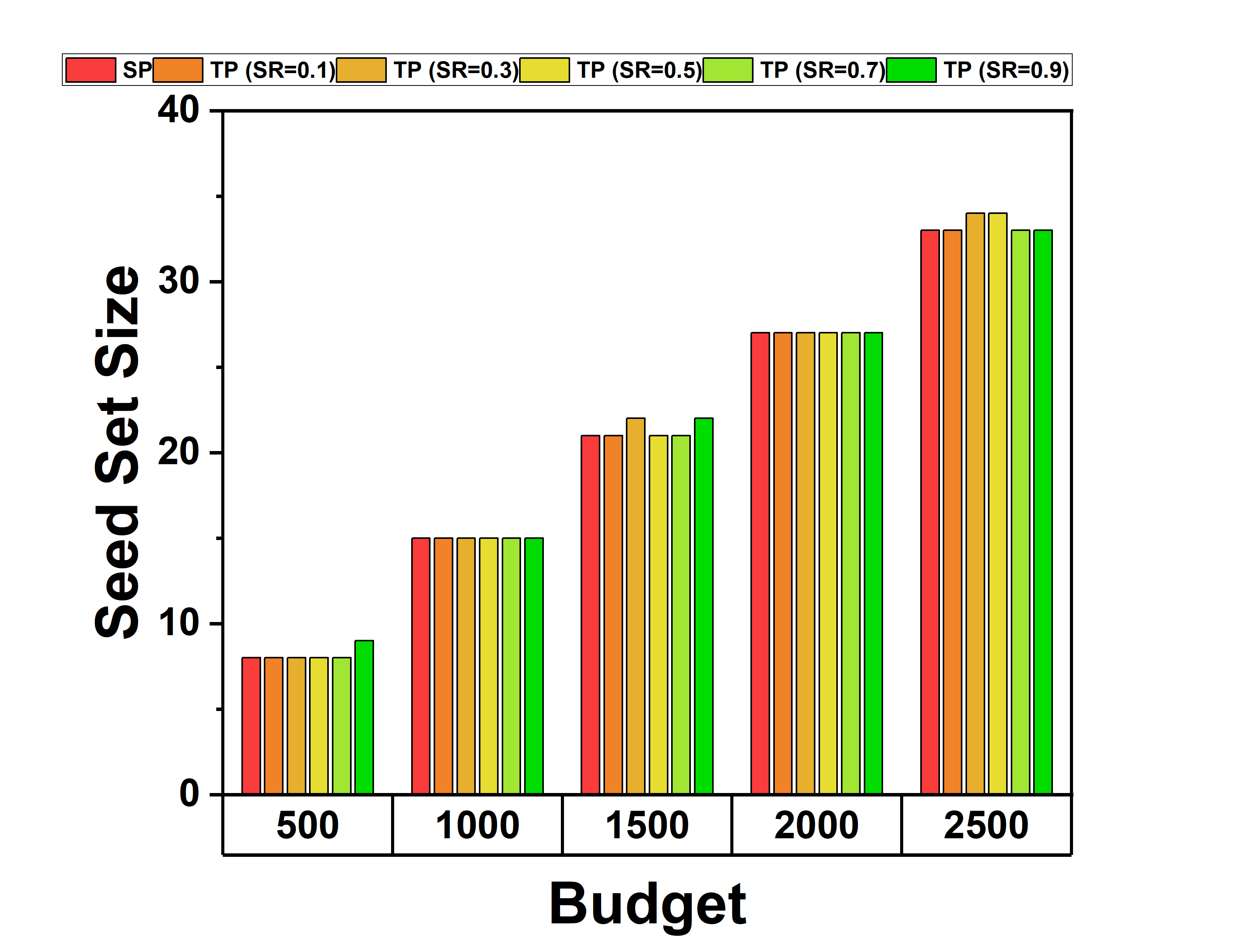}
        \caption{Timestep 4}
    \end{subfigure}

    \vspace{0.5cm}

    \begin{subfigure}[t]{0.3\linewidth}
        \centering
        \includegraphics[width=\linewidth]{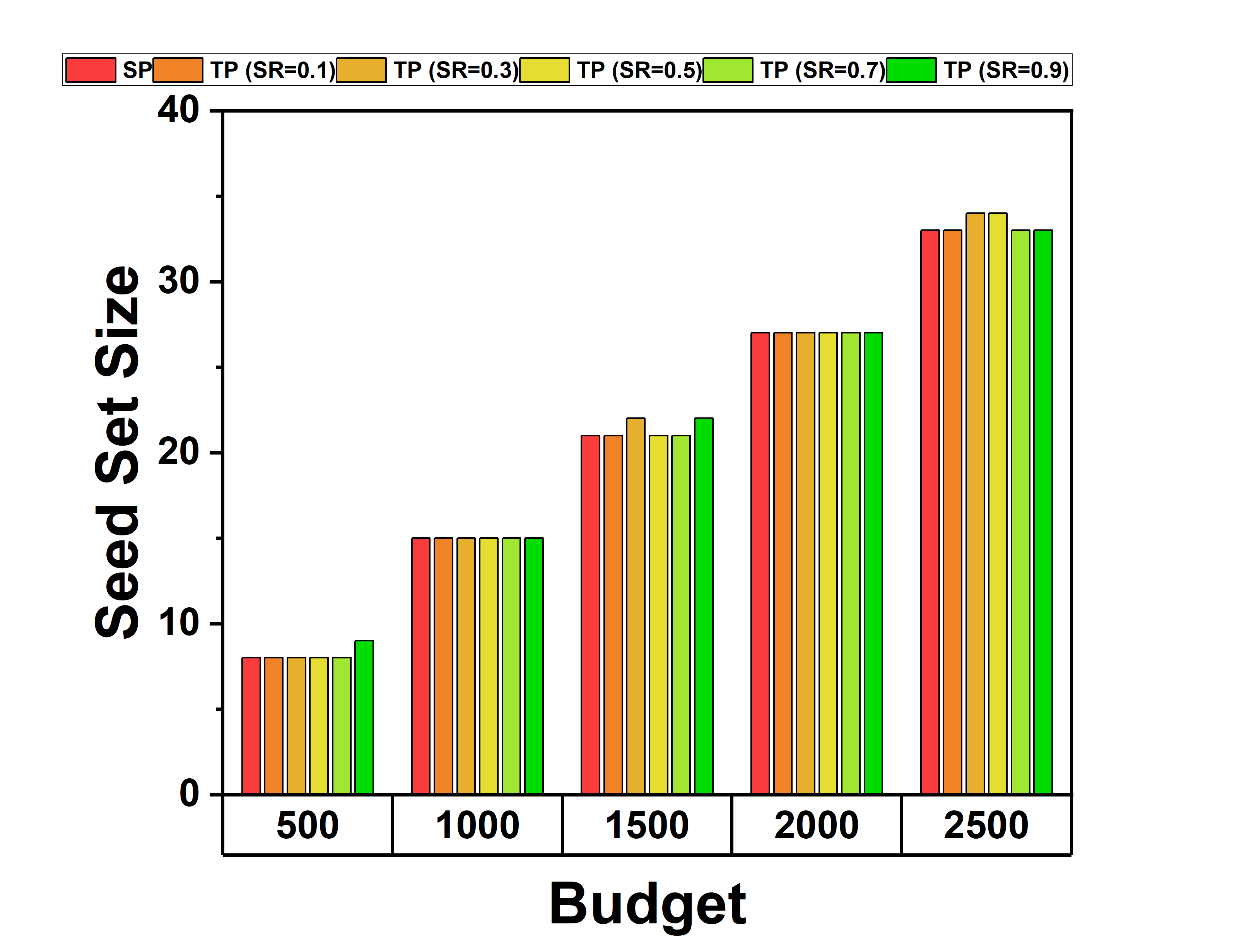}
        \caption{Timestep 6}
    \end{subfigure}
    \hfill
    \begin{subfigure}[t]{0.3\linewidth}
        \centering
        \includegraphics[width=\linewidth]{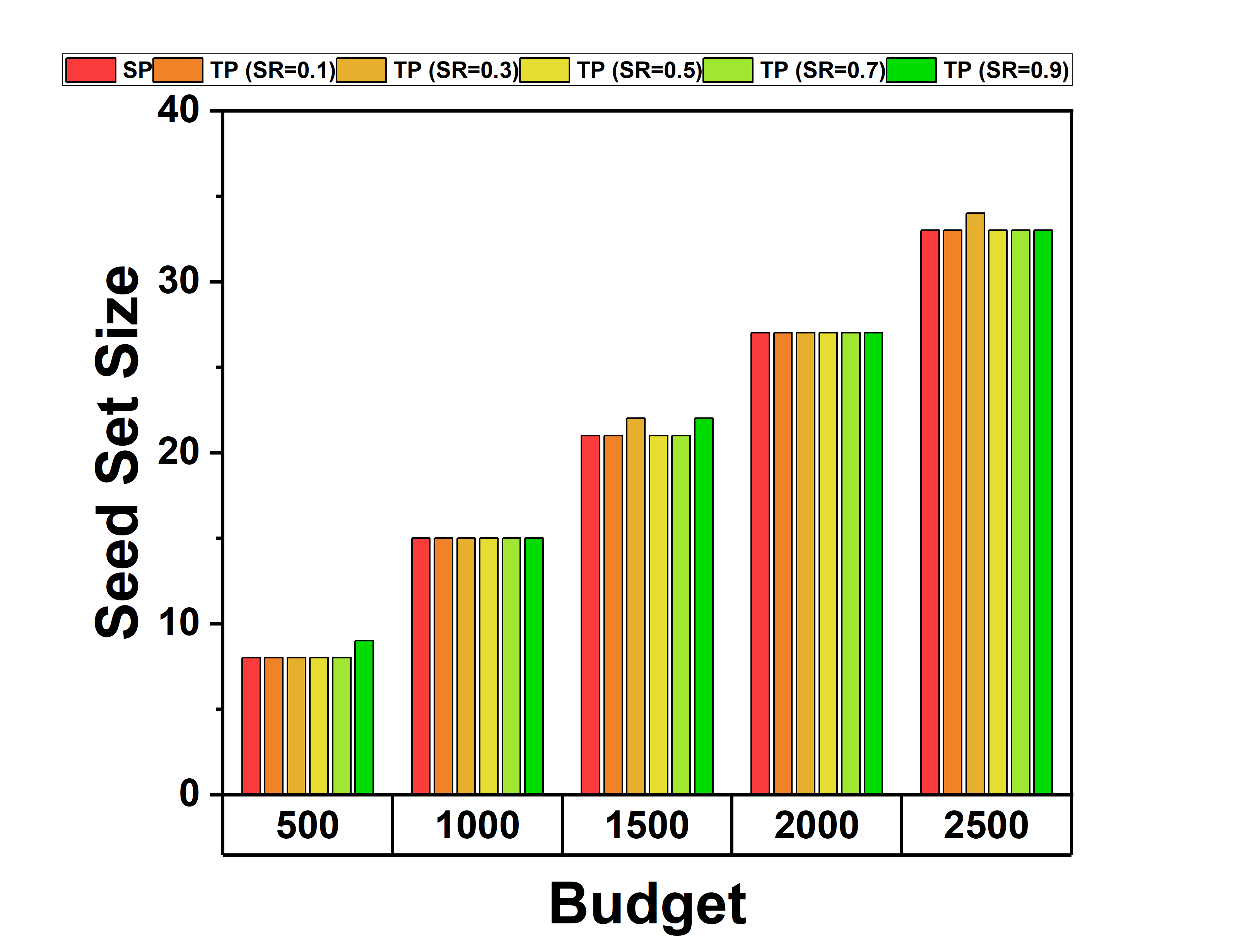}
        \caption{Timestep 8}
    \end{subfigure}
    \hfill
    \begin{subfigure}[t]{0.3\linewidth}
        \centering
        \includegraphics[width=\linewidth]{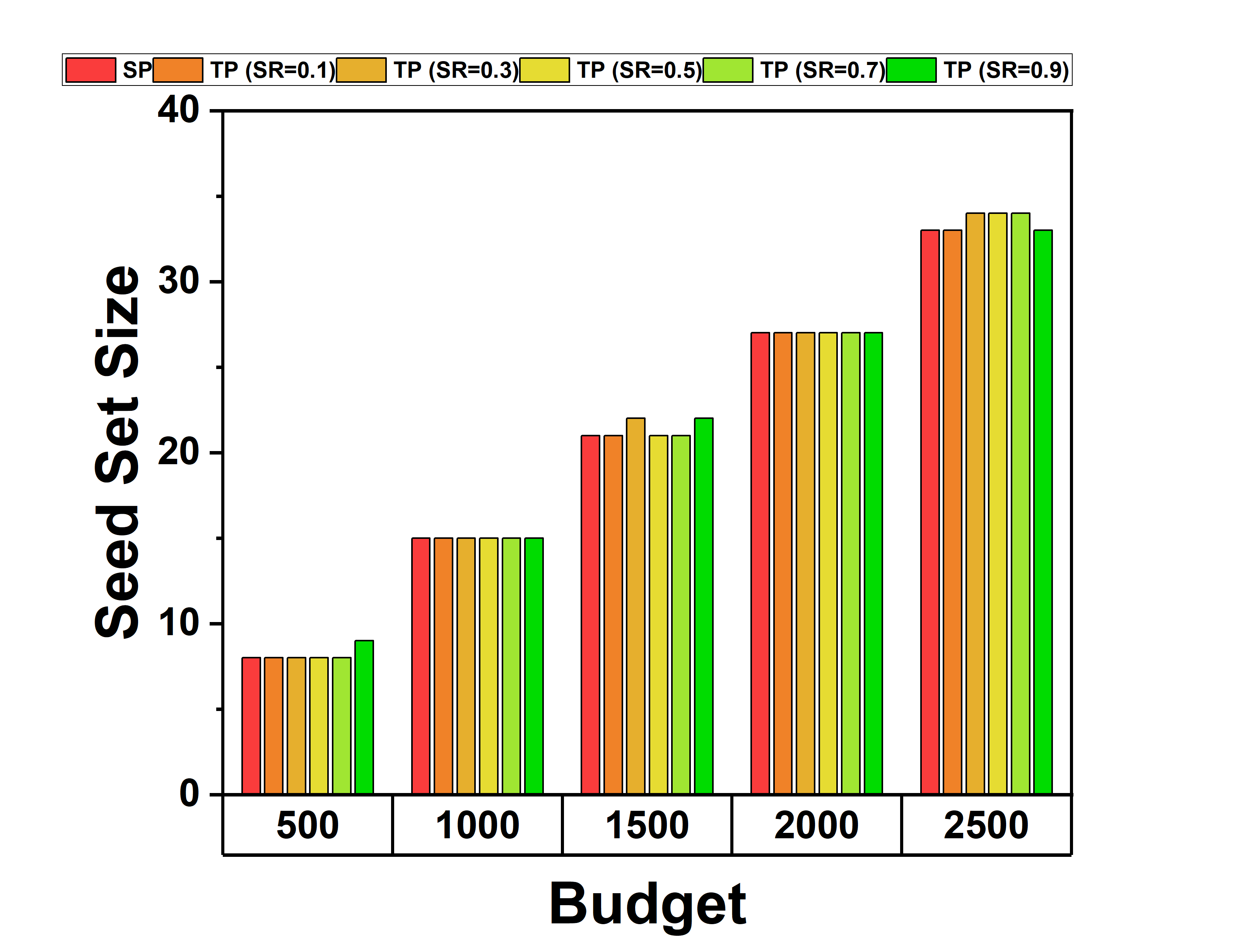}
        \caption{Timestep 10}
    \end{subfigure}

    \caption{Seed Set Size Distribution of Single Phase Vs. Two Phase (Clustering Coefficient Algorithm, \textit{LM} Dataset, Probability Setting - Trivalency)}
    \label{RQ4LM_T3}
\end{figure}

\begin{figure}[htbp]
    \centering
    \begin{subfigure}[t]{0.3\linewidth}
        \centering
        \includegraphics[width=\linewidth]{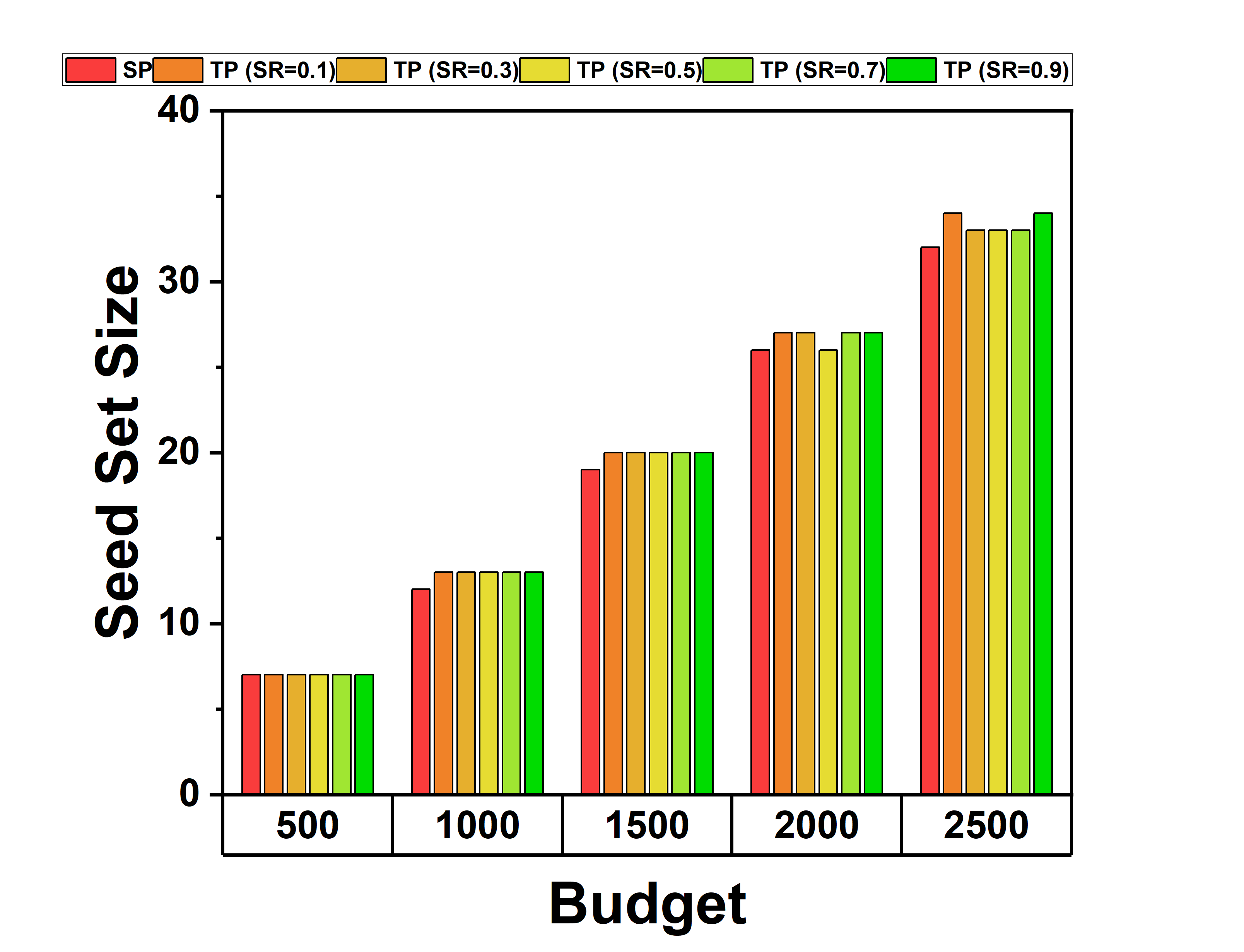}
        \caption{Timestep 2}
    \end{subfigure}
    \hspace{0.05\linewidth}
    \begin{subfigure}[t]{0.3\linewidth}
        \centering
        \includegraphics[width=\linewidth]{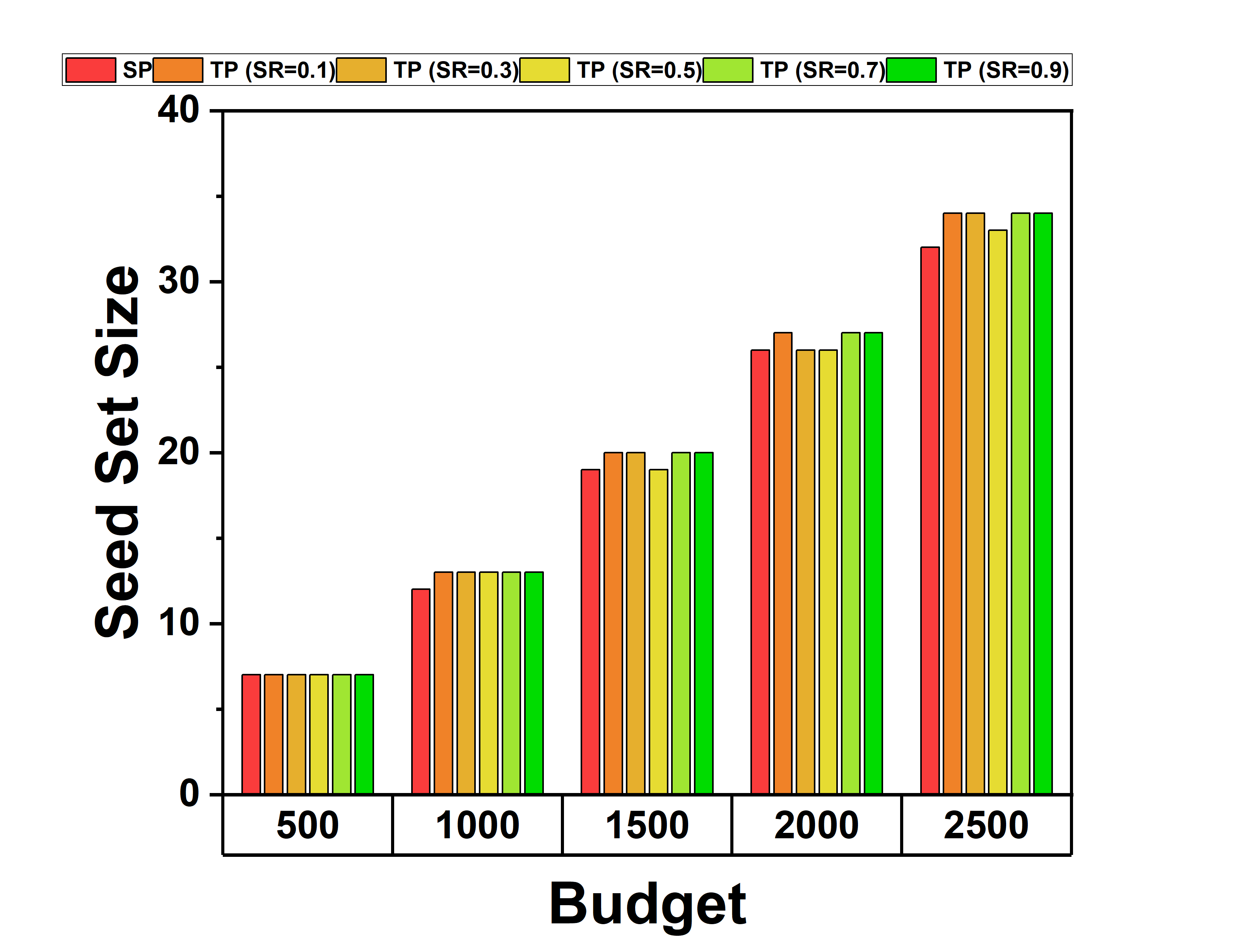}
        \caption{Timestep 4}
    \end{subfigure}

    \vspace{0.5cm}

    \begin{subfigure}[t]{0.3\linewidth}
        \centering
        \includegraphics[width=\linewidth]{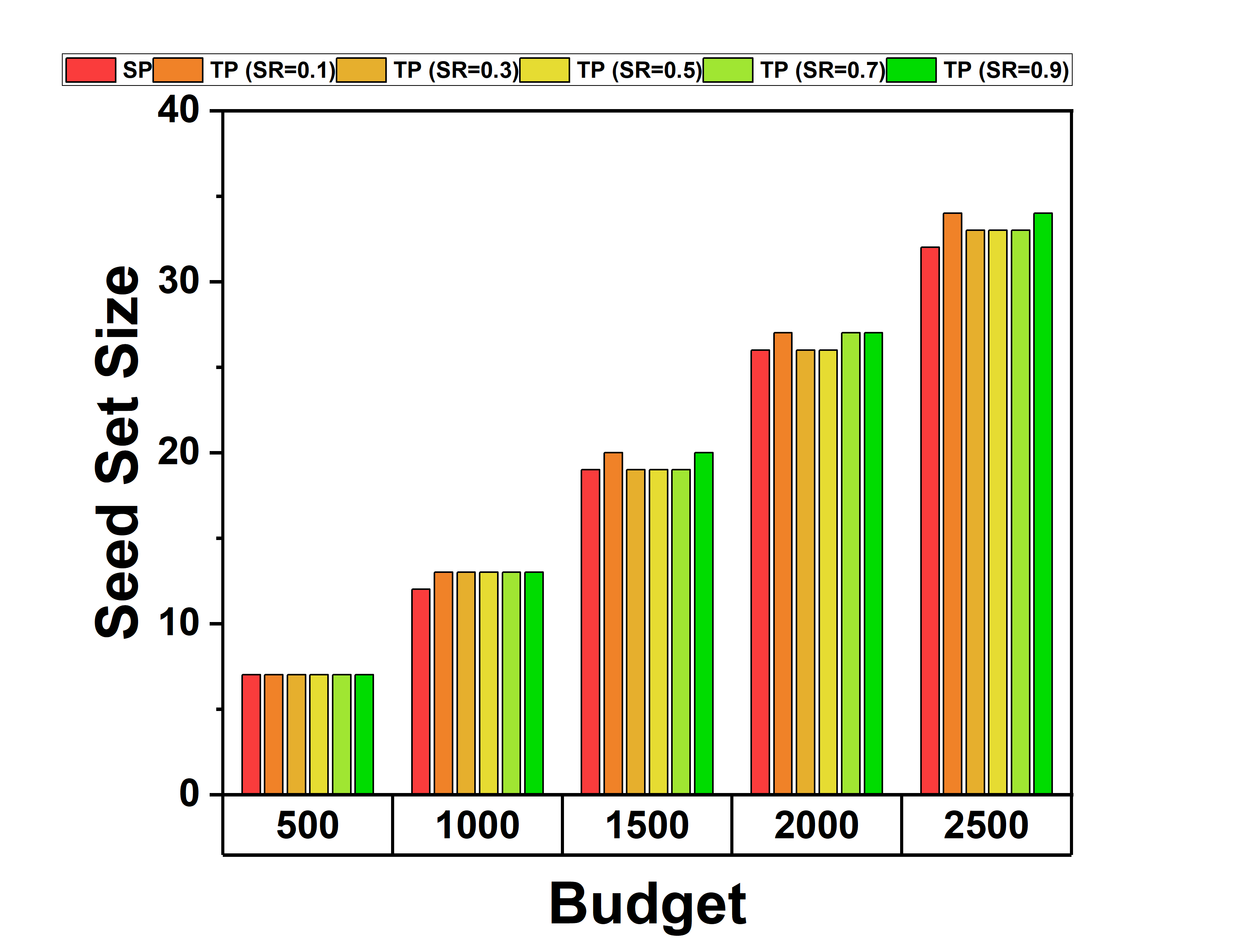}
        \caption{Timestep 6}
    \end{subfigure}
    \hfill
    \begin{subfigure}[t]{0.3\linewidth}
        \centering
        \includegraphics[width=\linewidth]{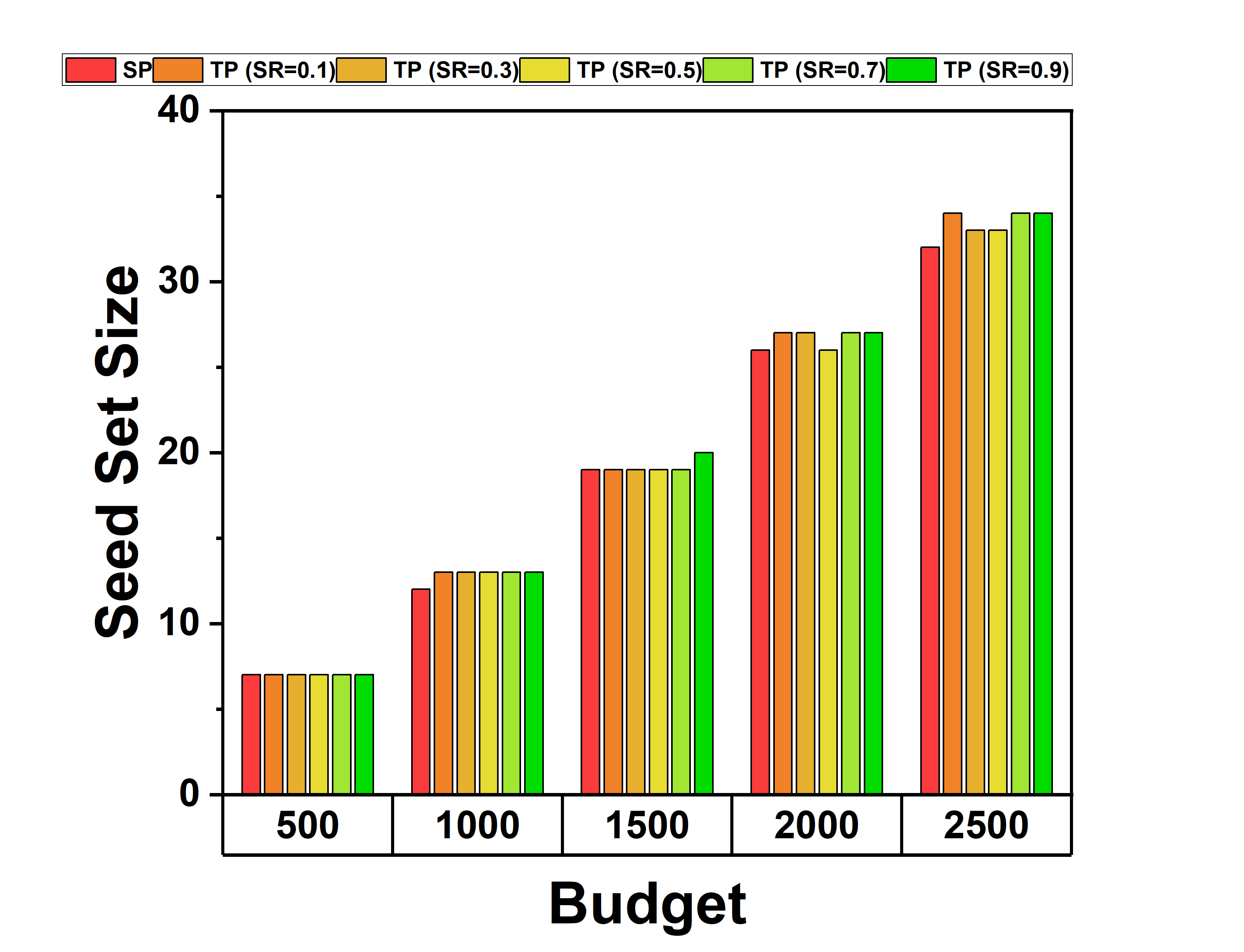}
        \caption{Timestep 8}
    \end{subfigure}
    \hfill
    \begin{subfigure}[t]{0.3\linewidth}
        \centering
        \includegraphics[width=\linewidth]{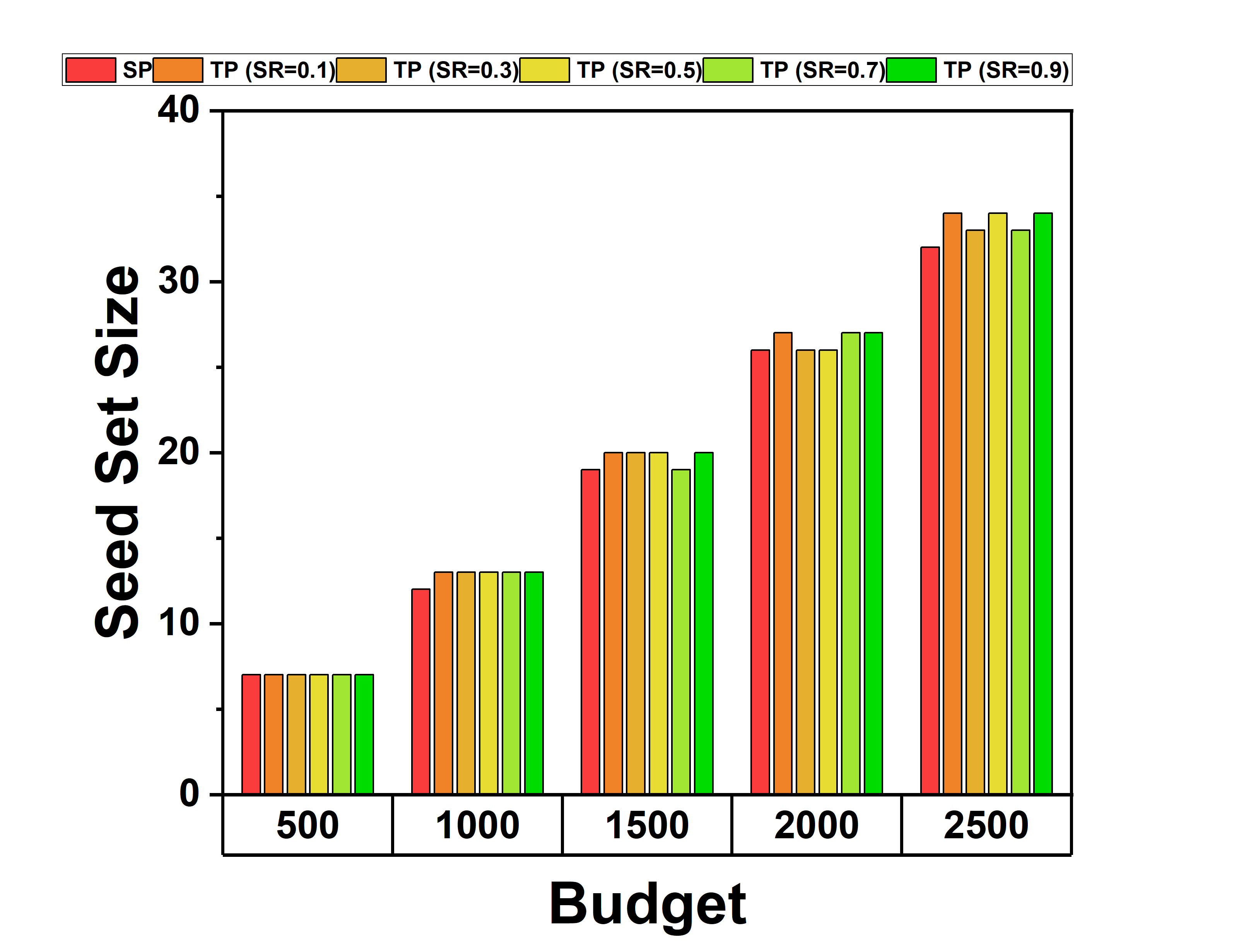}
        \caption{Timestep 10}
    \end{subfigure}

    \caption{Seed Set Size Distribution of Single Phase Vs. Two Phase (Degree Discount Algorithm, \textit{LM} Dataset, Probability Setting - Trivalency)}
    \label{RQ4LM_T4}
\end{figure}

\begin{figure}[htbp]
    \centering
    \begin{subfigure}[t]{0.3\linewidth}
        \centering
        \includegraphics[width=\linewidth]{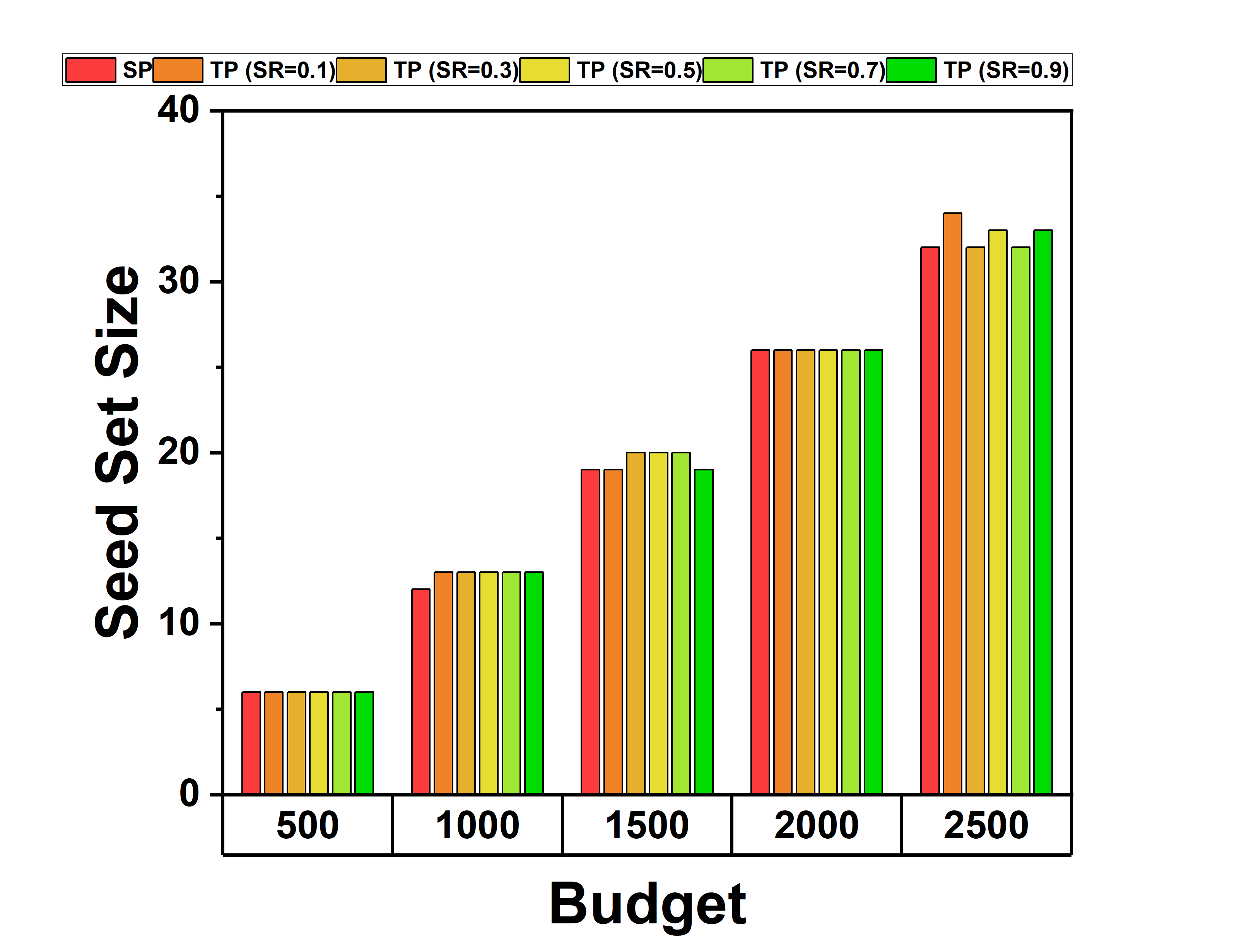}
        \caption{Timestep 2}
    \end{subfigure}
    \hspace{0.05\linewidth}
    \begin{subfigure}[t]{0.3\linewidth}
        \centering
        \includegraphics[width=\linewidth]{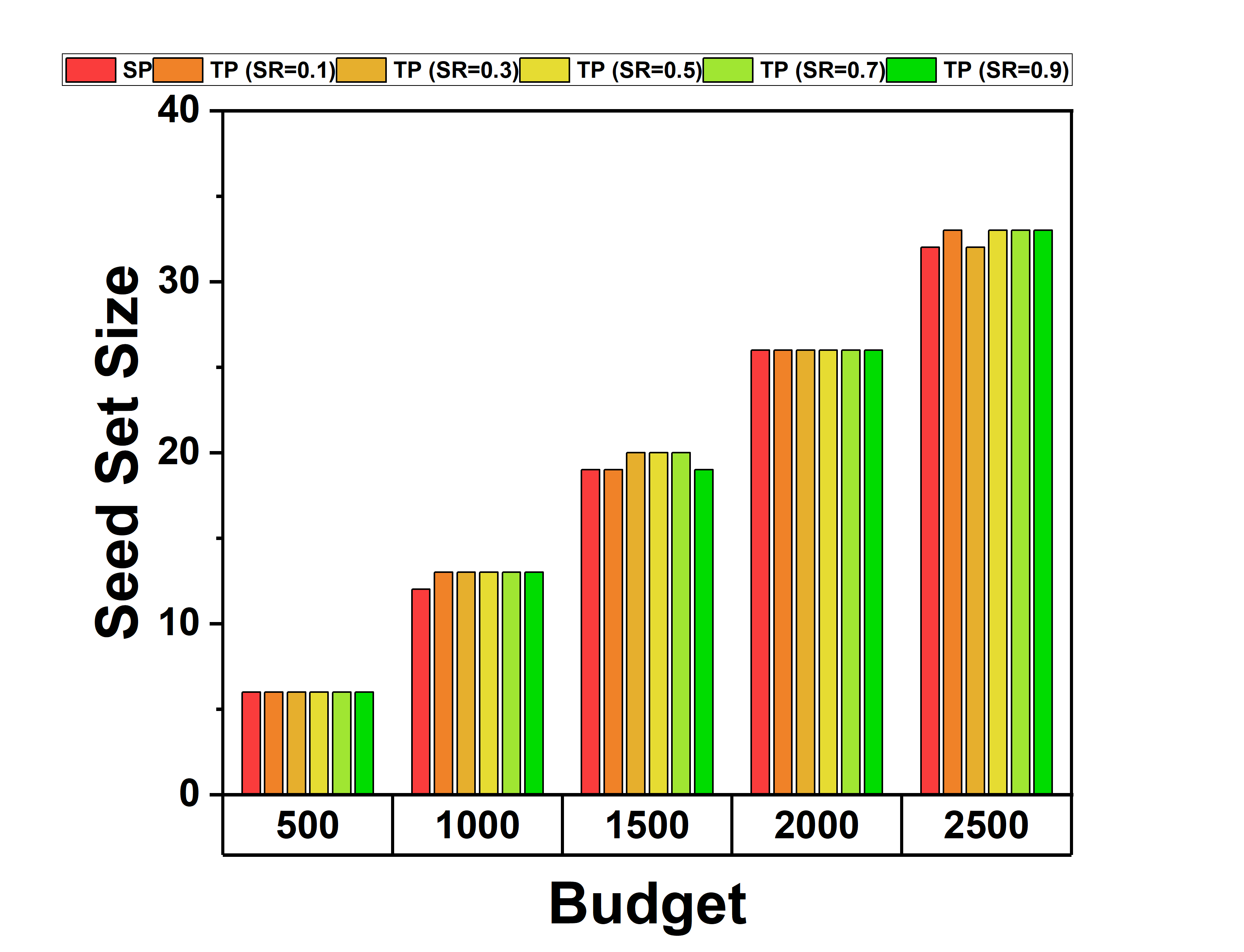}
        \caption{Timestep 4}
    \end{subfigure}

    \vspace{0.5cm}

    \begin{subfigure}[t]{0.3\linewidth}
        \centering
        \includegraphics[width=\linewidth]{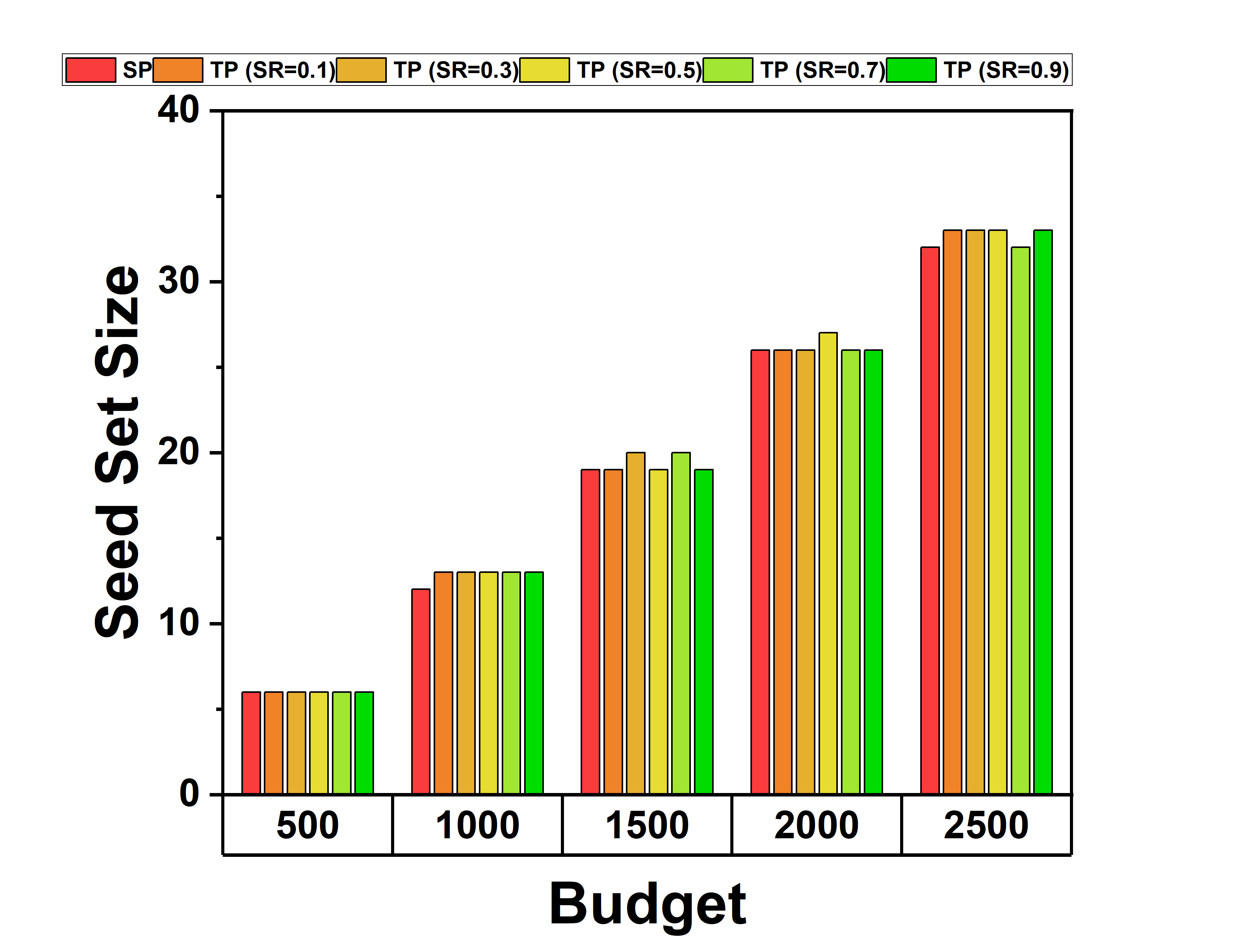}
        \caption{Timestep 6}
    \end{subfigure}
    \hfill
    \begin{subfigure}[t]{0.3\linewidth}
        \centering
        \includegraphics[width=\linewidth]{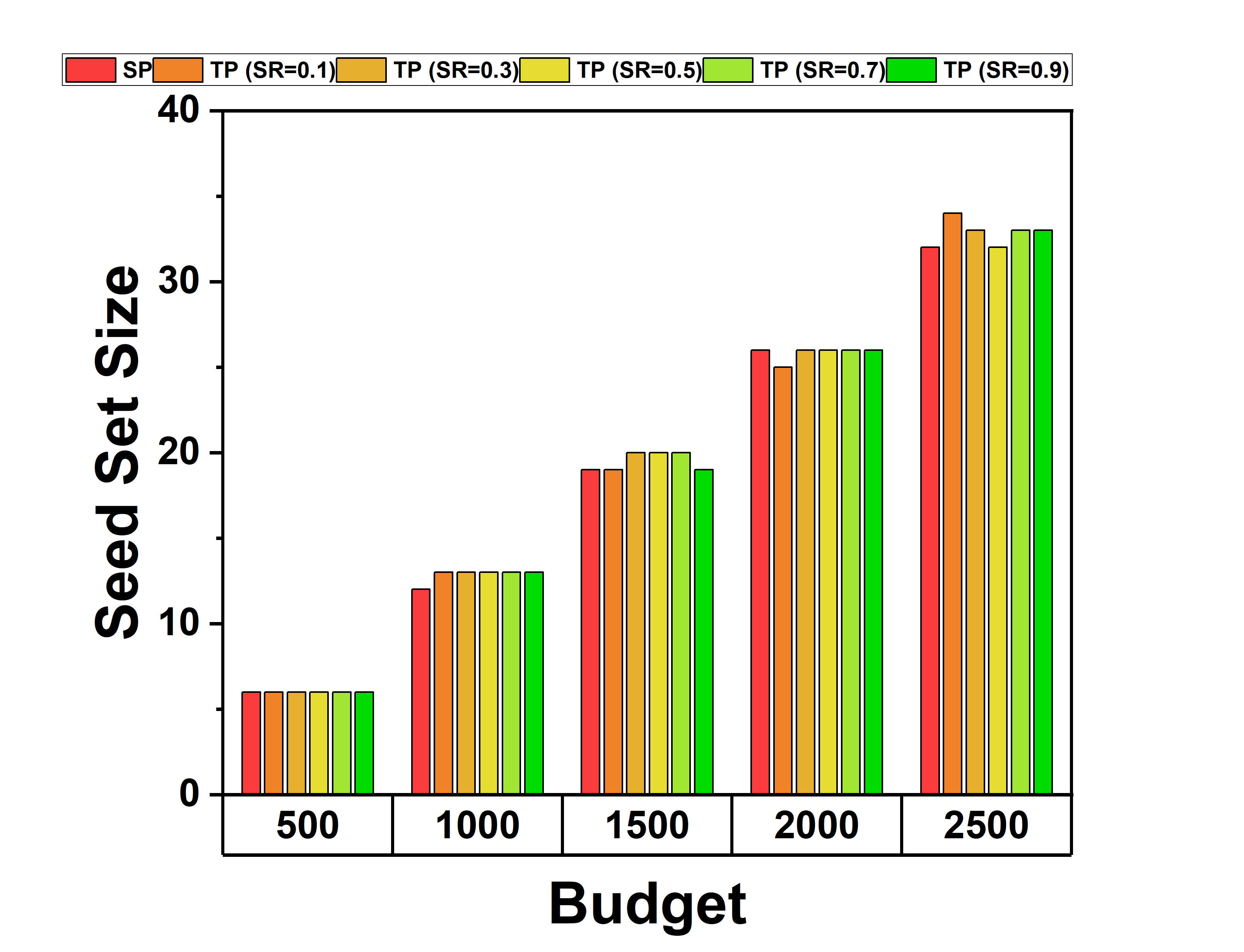}
        \caption{Timestep 8}
    \end{subfigure}
    \hfill
    \begin{subfigure}[t]{0.3\linewidth}
        \centering
        \includegraphics[width=\linewidth]{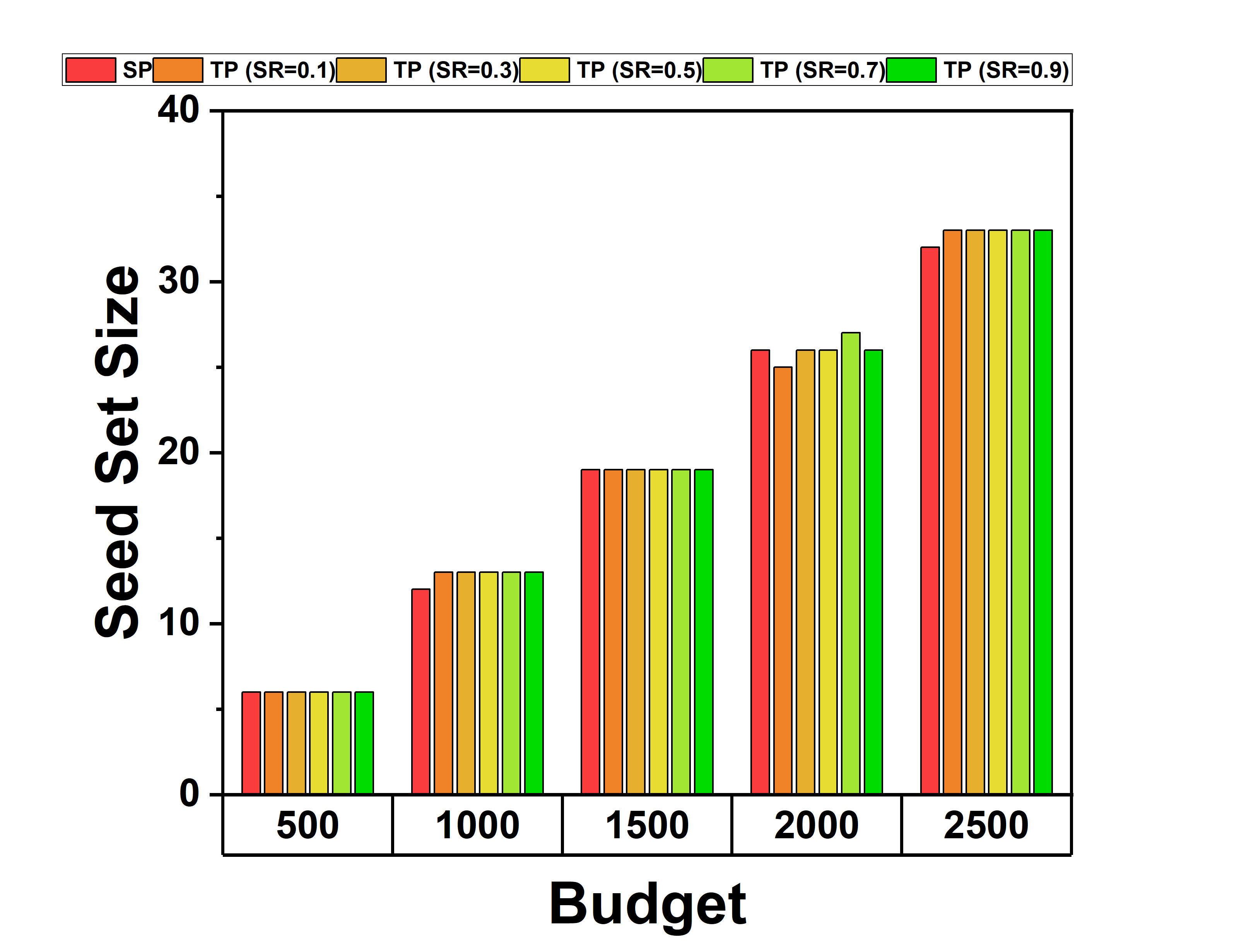}
        \caption{Timestep 10}
    \end{subfigure}

    \caption{Seed Set Size Distribution of Single Phase Vs. Two Phase (Single Discount Algorithm, \textit{LM} Dataset, Probability Setting - Trivalency)}
    \label{RQ4LM_T5}
\end{figure}

\begin{figure}[htbp]
    \centering
    \begin{subfigure}[t]{0.3\linewidth}
        \centering
        \includegraphics[width=\linewidth]{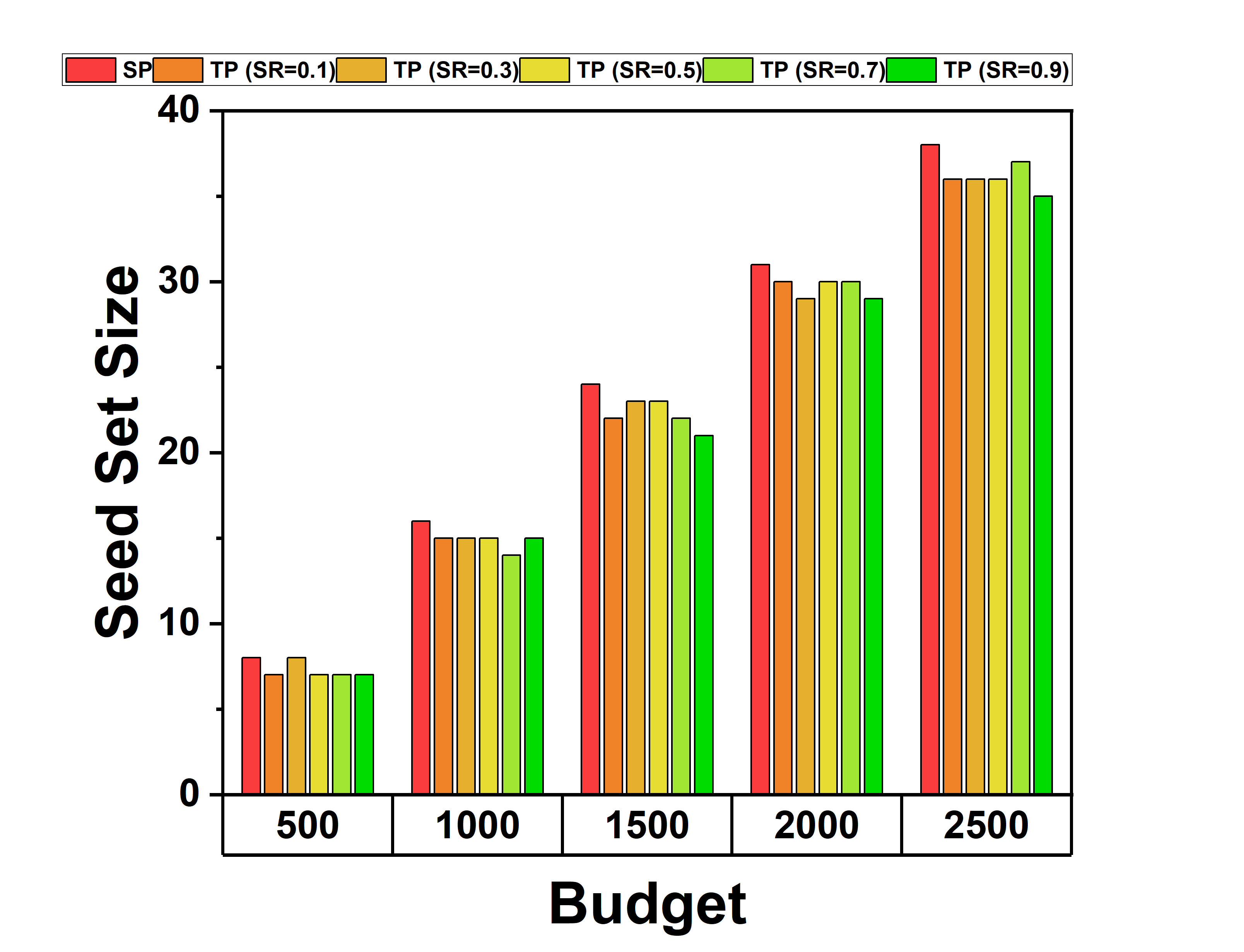}
        \caption{Timestep 2}
    \end{subfigure}
    \hspace{0.05\linewidth}
    \begin{subfigure}[t]{0.3\linewidth}
        \centering
        \includegraphics[width=\linewidth]{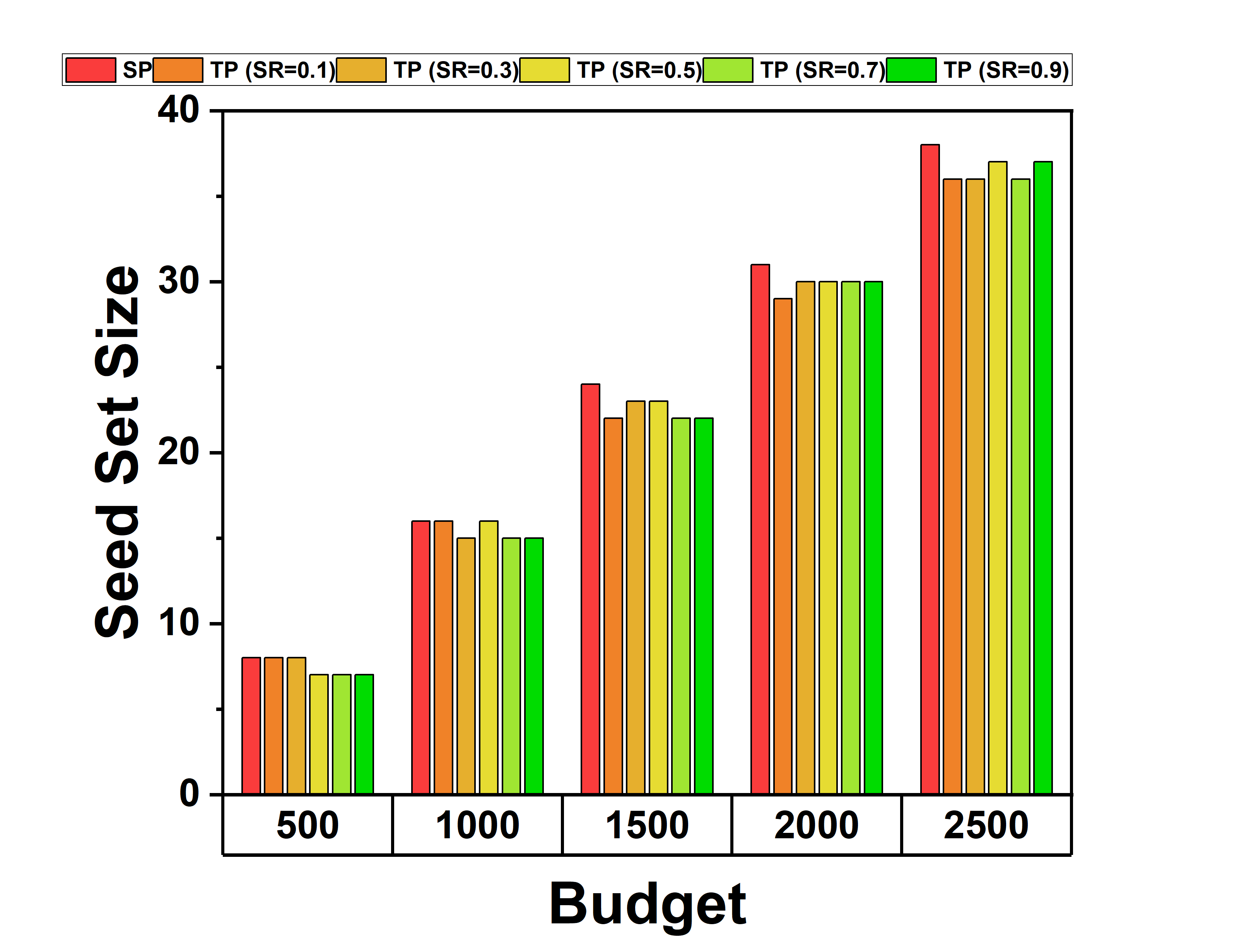}
        \caption{Timestep 4}
    \end{subfigure}

    \vspace{0.5cm}

    \begin{subfigure}[t]{0.3\linewidth}
        \centering
        \includegraphics[width=\linewidth]{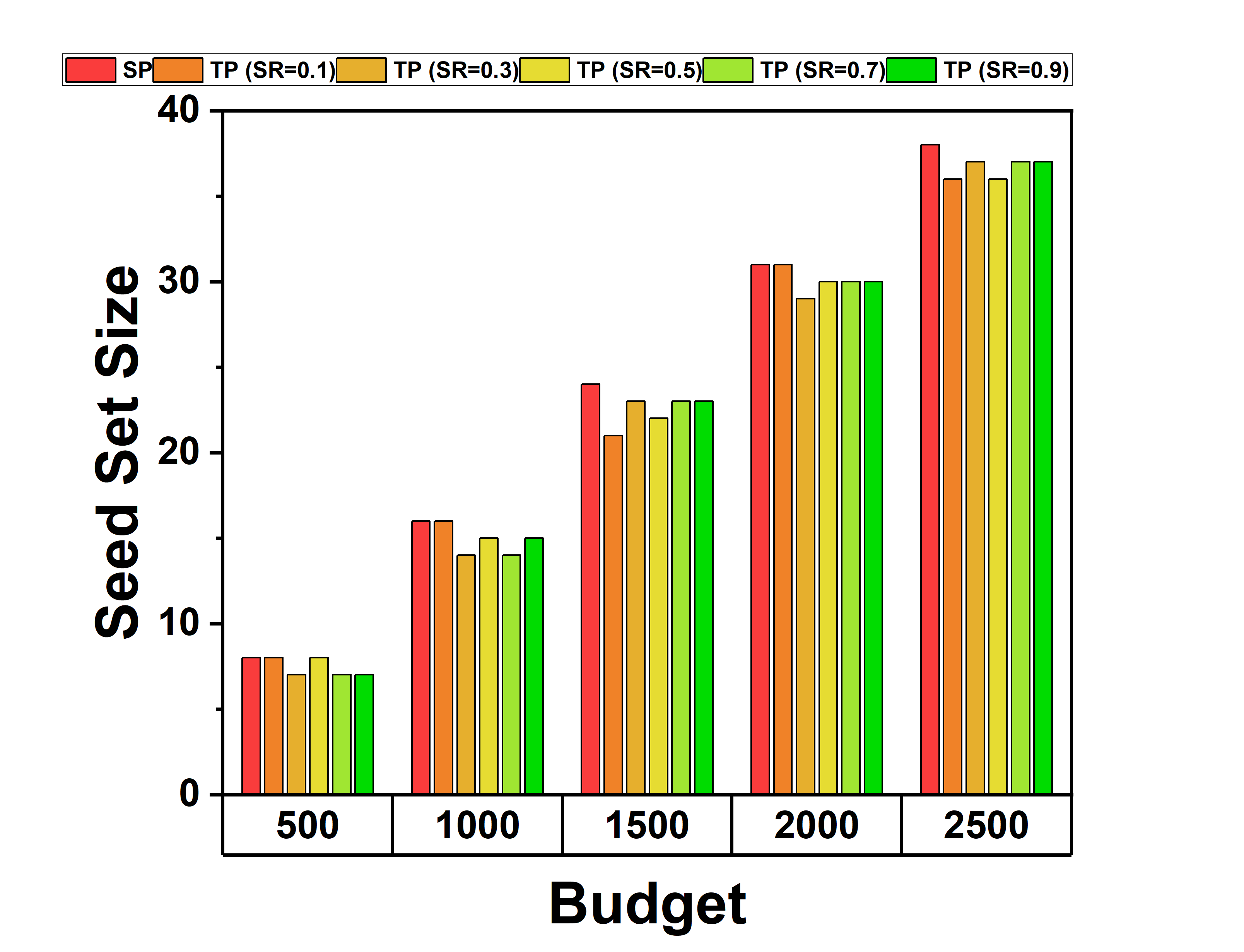}
        \caption{Timestep 6}
    \end{subfigure}
    \hfill
    \begin{subfigure}[t]{0.3\linewidth}
        \centering
        \includegraphics[width=\linewidth]{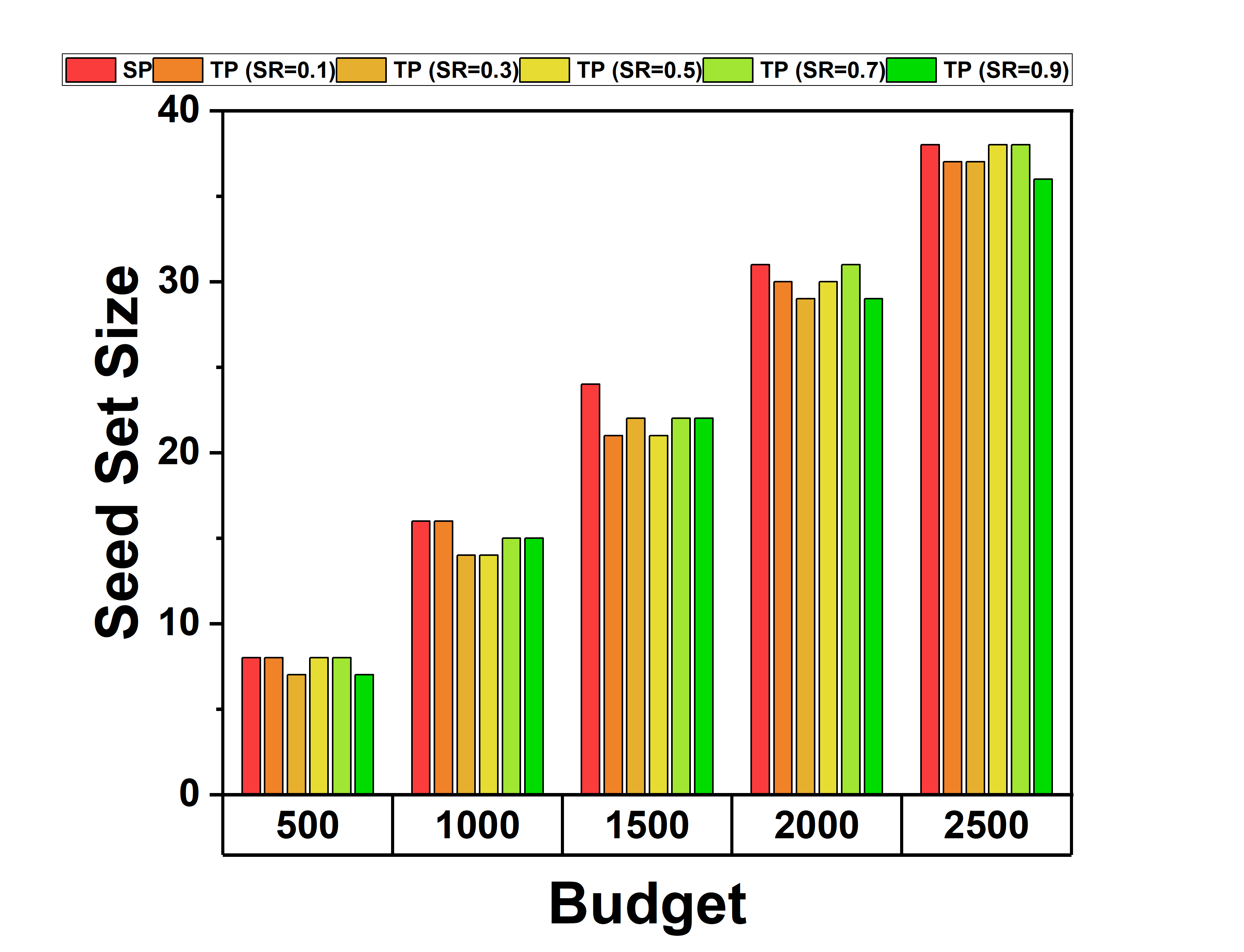}
        \caption{Timestep 8}
    \end{subfigure}
    \hfill
    \begin{subfigure}[t]{0.3\linewidth}
        \centering
        \includegraphics[width=\linewidth]{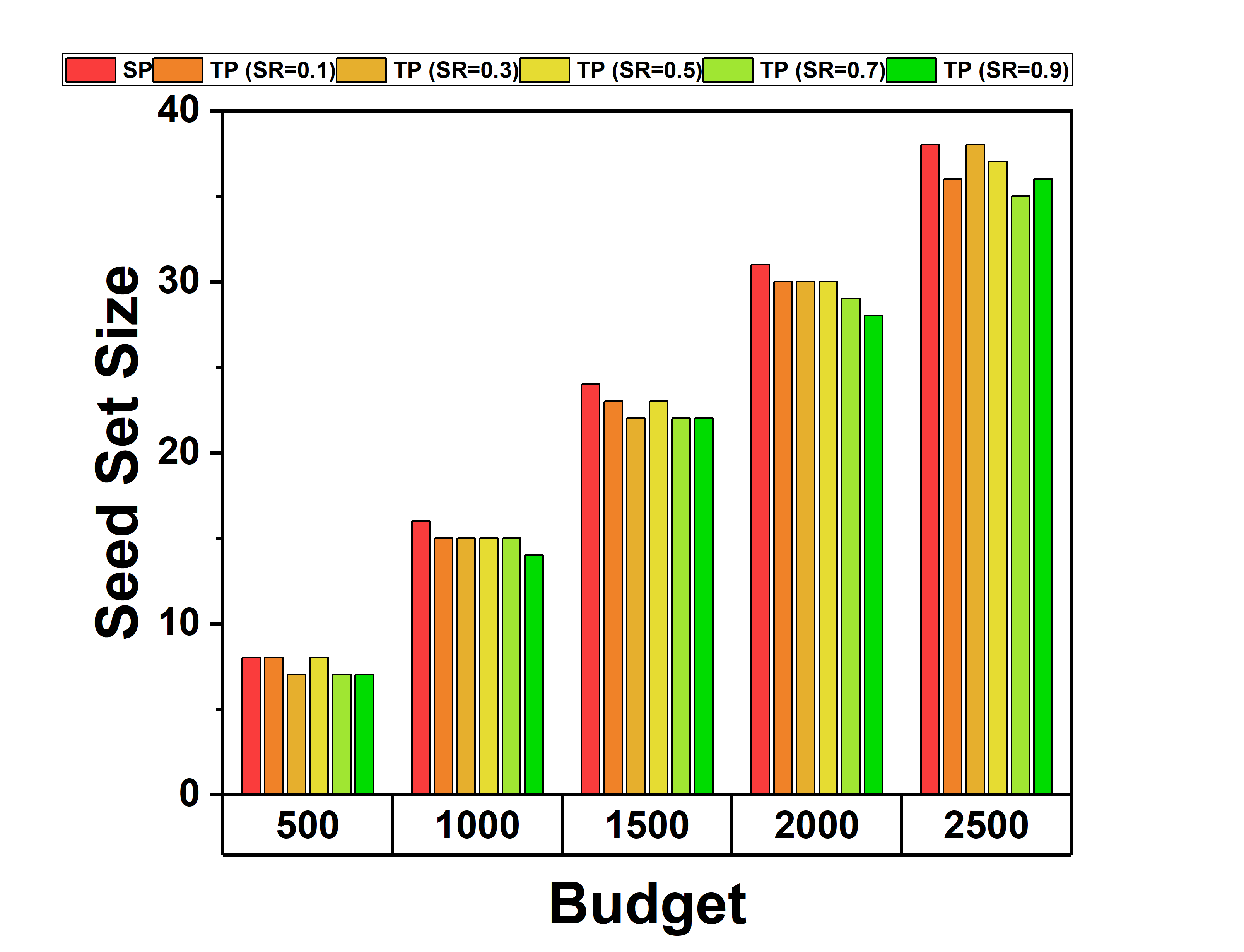}
        \caption{Timestep 10}
    \end{subfigure}

    \caption{Seed Set Size Distribution of Single Phase Vs. Two Phase (Simple Greedy Algorithm, \textit{LM} Dataset, Probability Setting - Trivalency)}
    \label{RQ4LM_T6}
\end{figure}

\begin{figure}[htbp]
    \centering
    \begin{subfigure}[t]{0.3\linewidth}
        \centering
        \includegraphics[width=\linewidth]{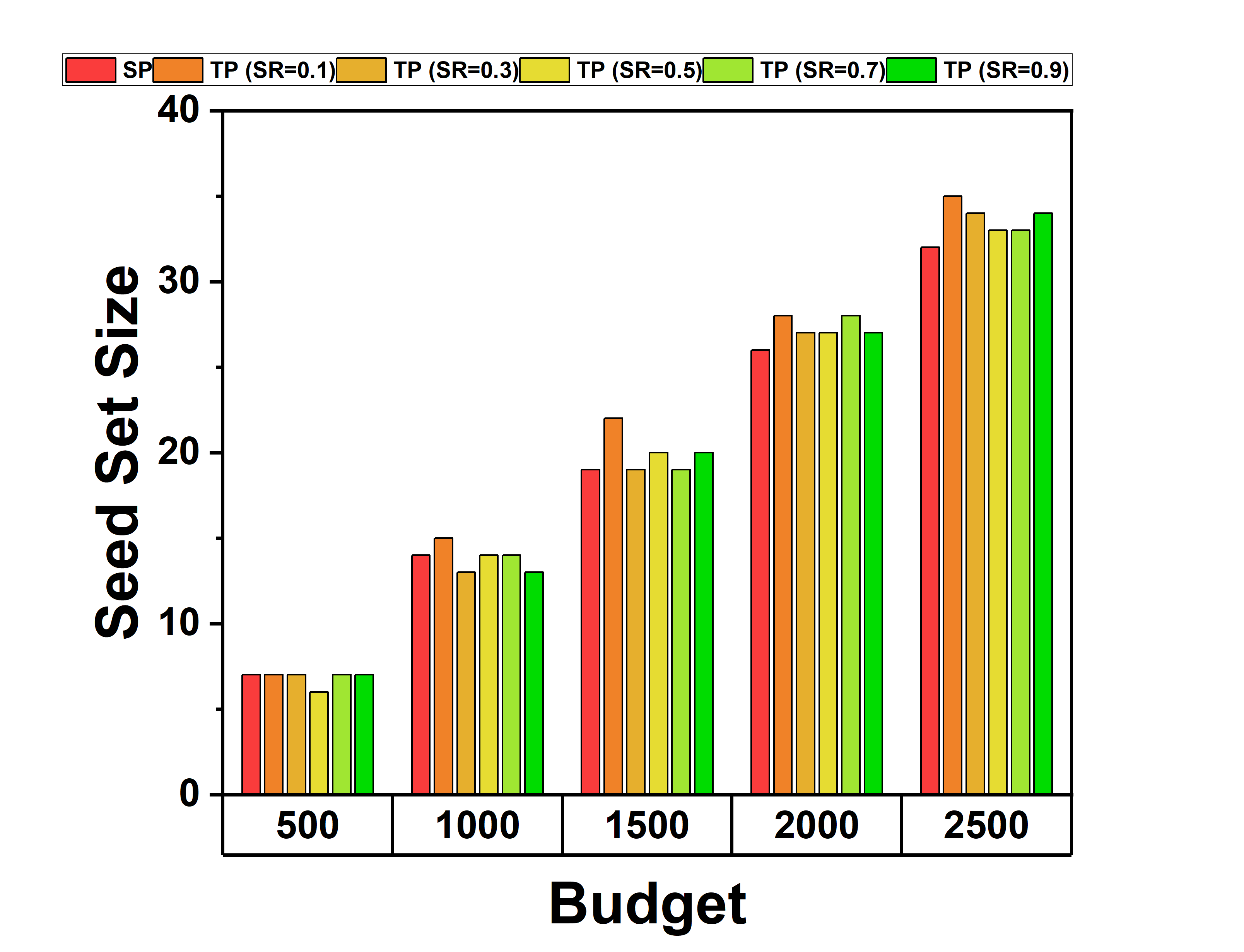}
        \caption{Timestep 2}
    \end{subfigure}
    \hspace{0.05\linewidth}
    \begin{subfigure}[t]{0.3\linewidth}
        \centering
        \includegraphics[width=\linewidth]{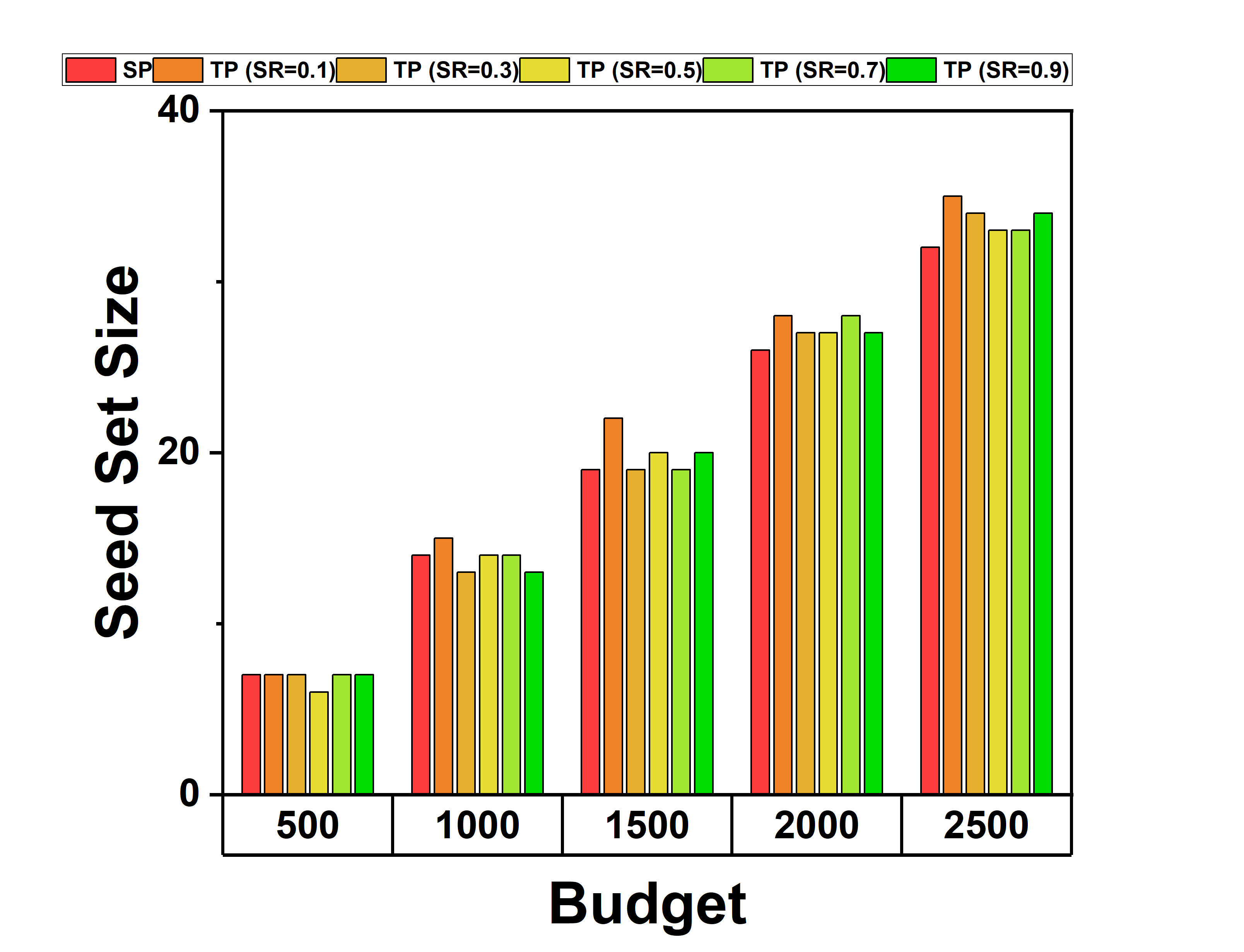}
        \caption{Timestep 4}
    \end{subfigure}

    \vspace{0.5cm}

    \begin{subfigure}[t]{0.3\linewidth}
        \centering
        \includegraphics[width=\linewidth]{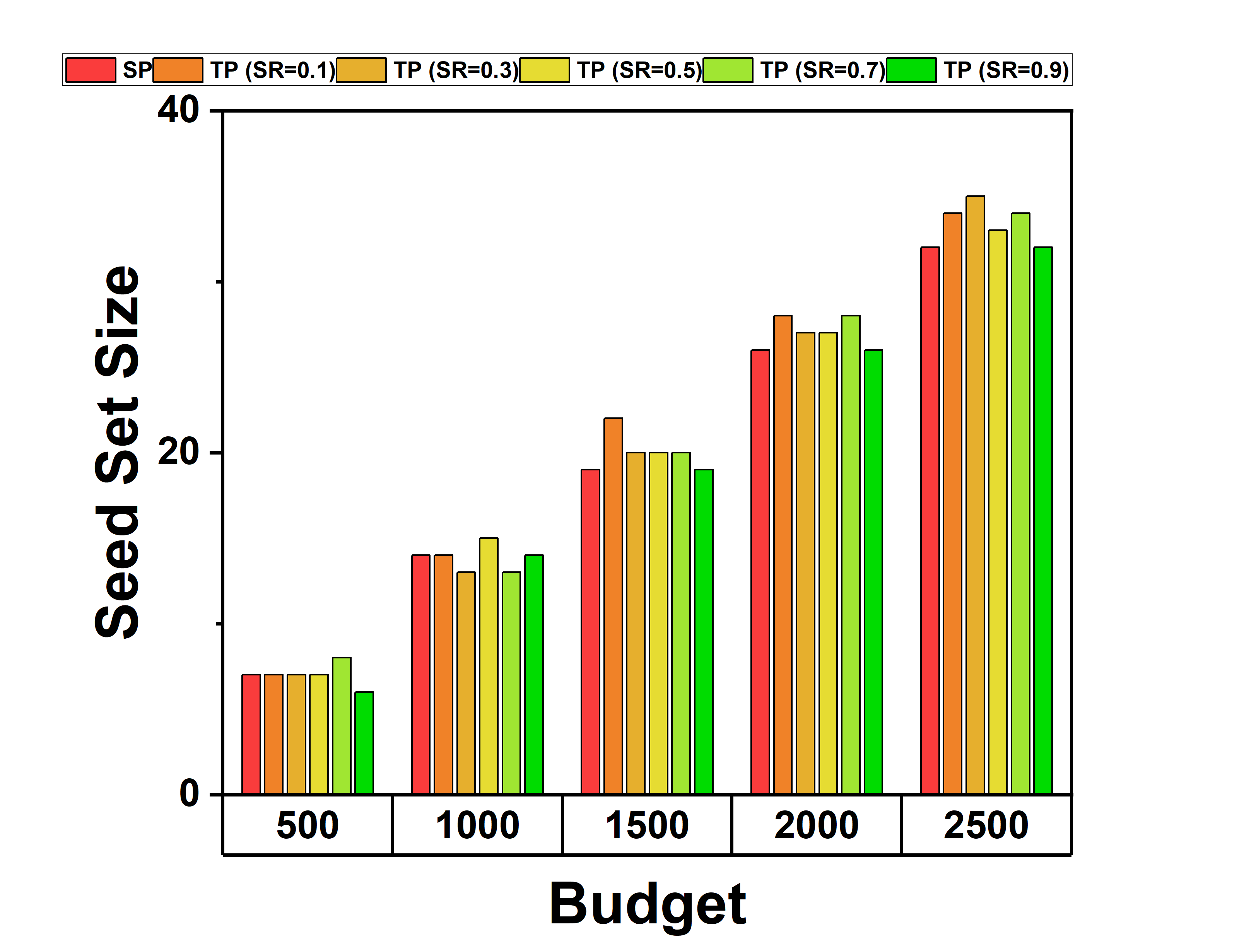}
        \caption{Timestep 6}
    \end{subfigure}
    \hfill
    \begin{subfigure}[t]{0.3\linewidth}
        \centering
        \includegraphics[width=\linewidth]{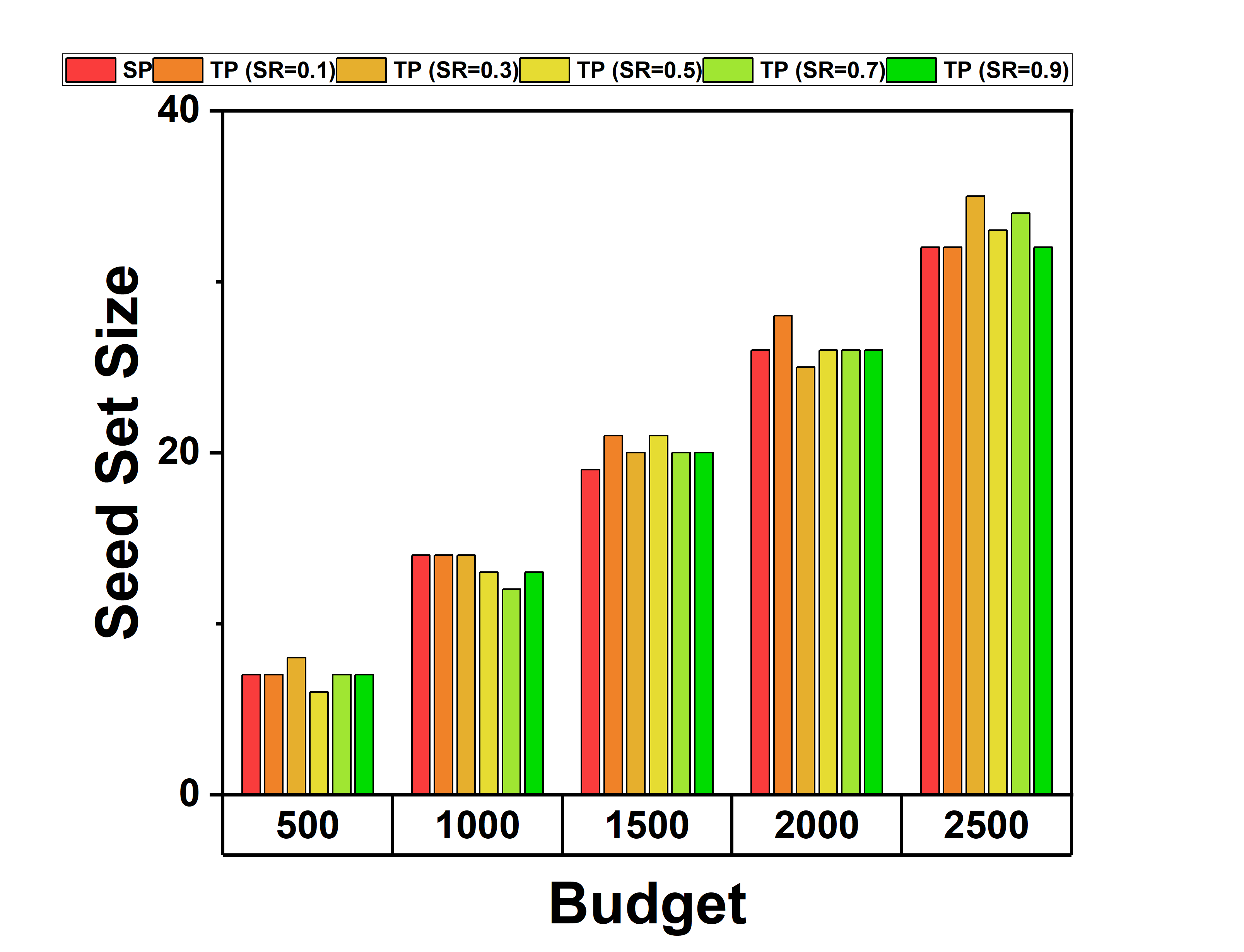}
        \caption{Timestep 8}
    \end{subfigure}
    \hfill
    \begin{subfigure}[t]{0.3\linewidth}
        \centering
        \includegraphics[width=\linewidth]{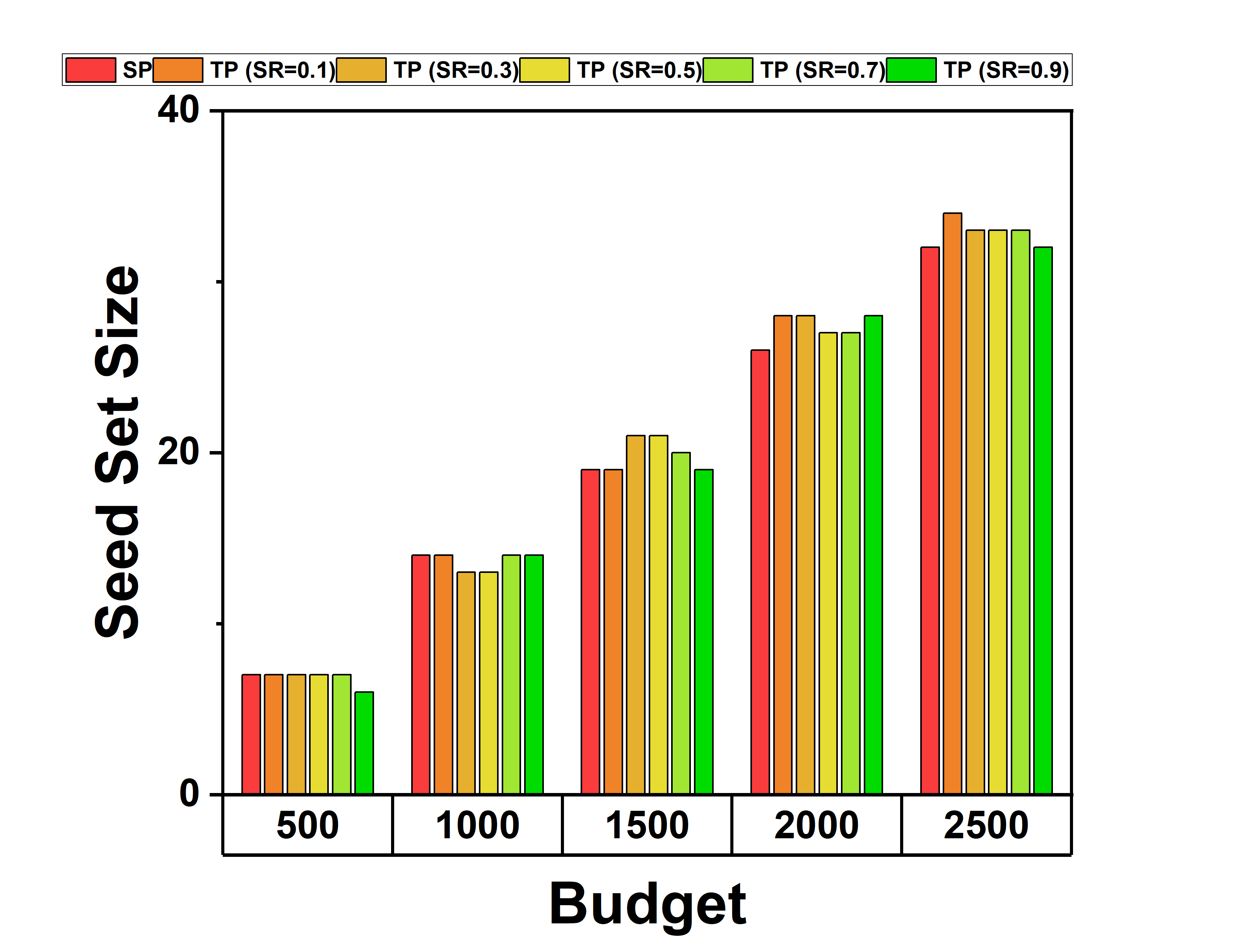}
        \caption{Timestep 10}
    \end{subfigure}

    \caption{Seed Set Size Distribution of Single Phase Vs. Two Phase (Double Greedy Algorithm, \textit{LM} Dataset, Probability Setting - Trivalency)}
    \label{RQ4LM_T7}
\end{figure}

\begin{figure}[htbp]
    \centering
    \begin{subfigure}[t]{0.3\linewidth}
        \centering
        \includegraphics[width=\linewidth]{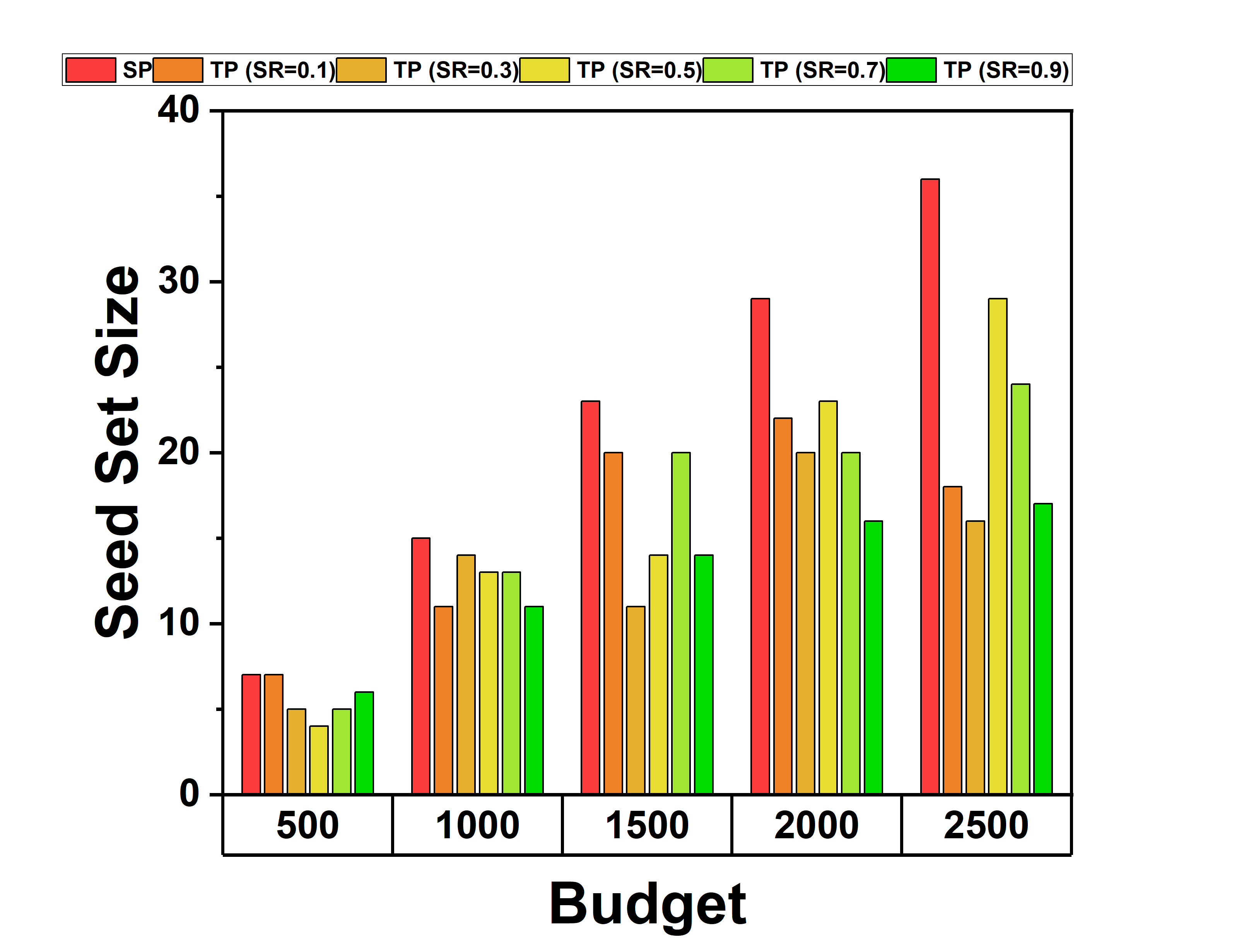}
        \caption{Timestep 2}
    \end{subfigure}
    \hspace{0.05\linewidth}
    \begin{subfigure}[t]{0.3\linewidth}
        \centering
        \includegraphics[width=\linewidth]{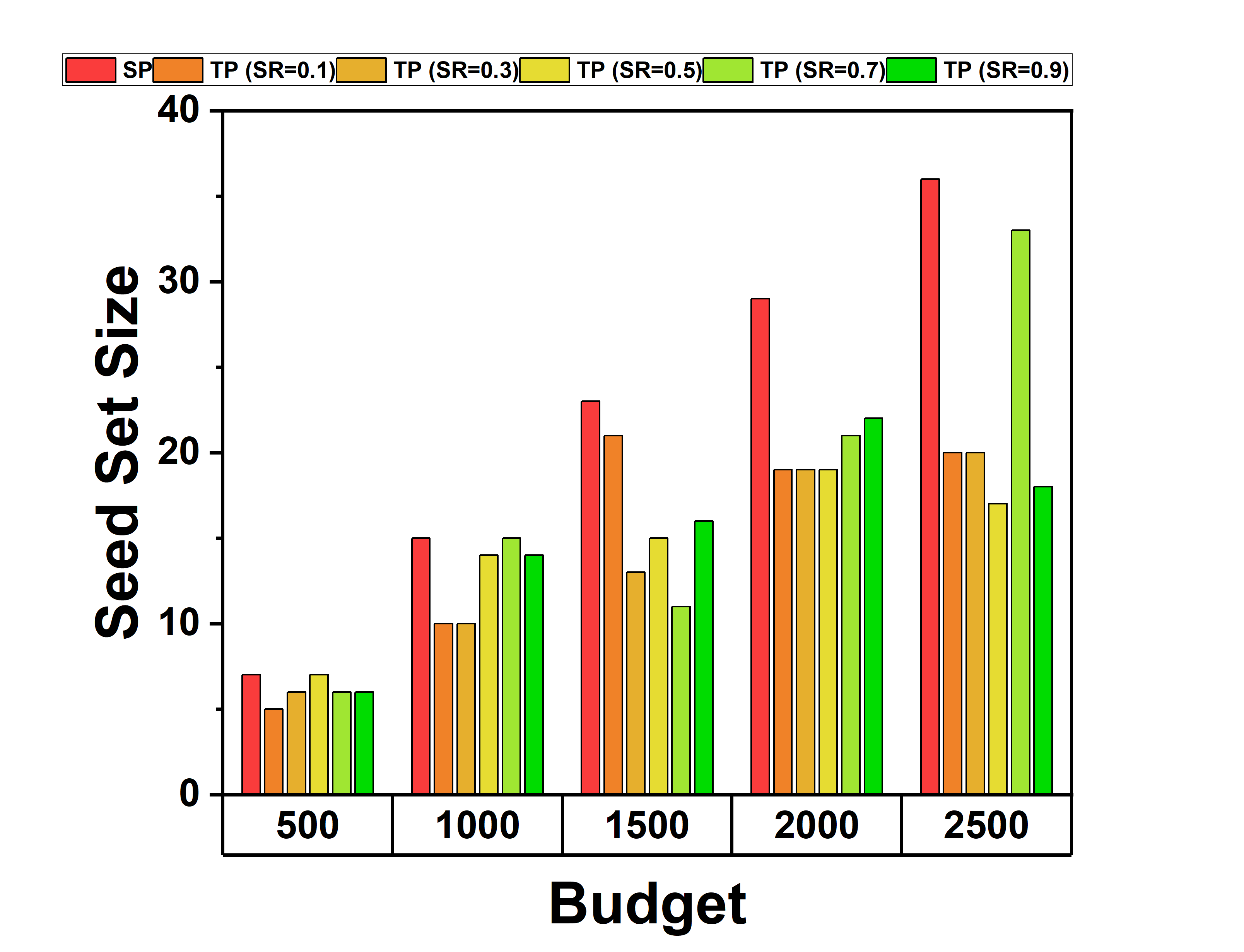}
        \caption{Timestep 4}
    \end{subfigure}

    \vspace{0.5cm}

    \begin{subfigure}[t]{0.3\linewidth}
        \centering
        \includegraphics[width=\linewidth]{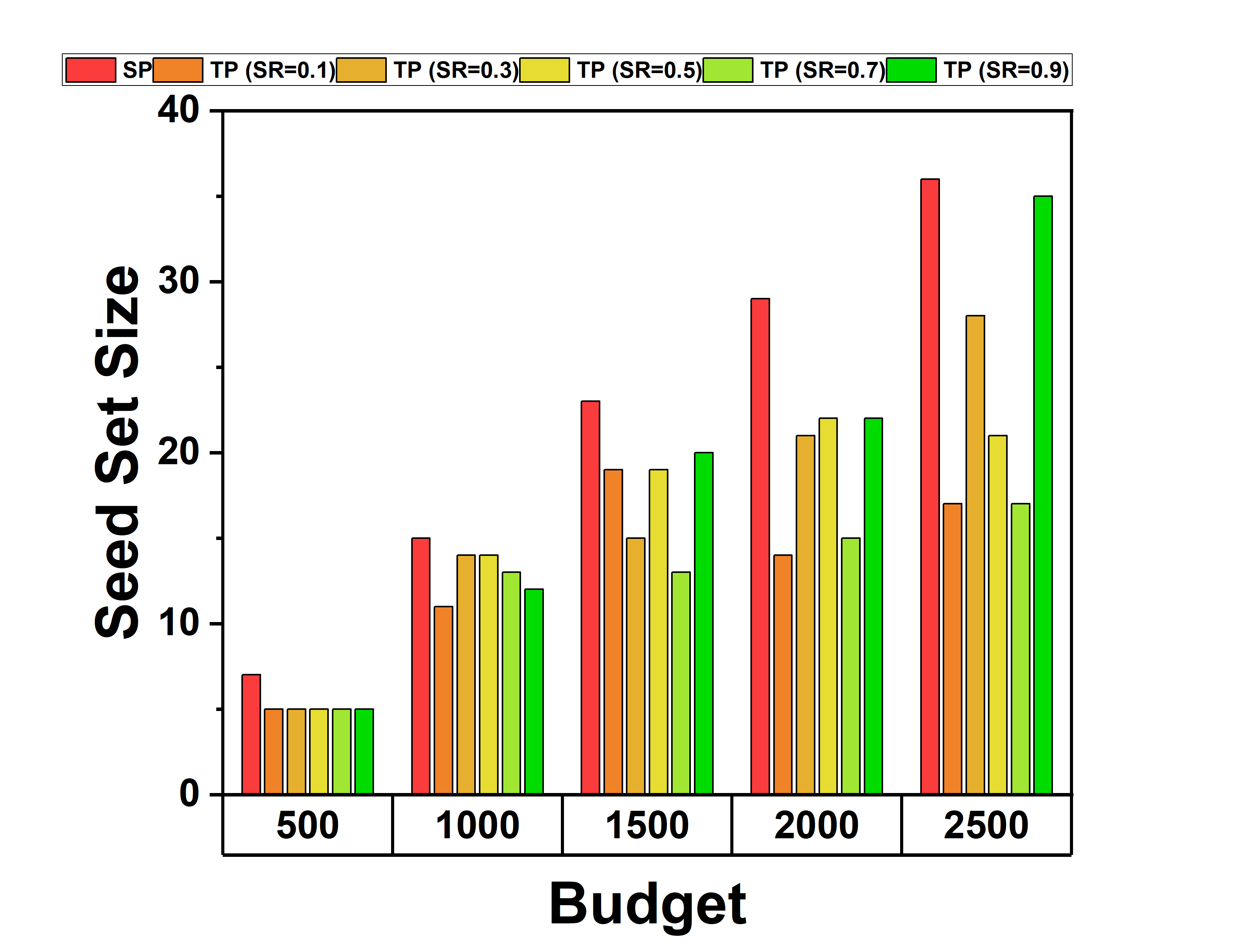}
        \caption{Timestep 6}
    \end{subfigure}
    \hfill
    \begin{subfigure}[t]{0.3\linewidth}
        \centering
        \includegraphics[width=\linewidth]{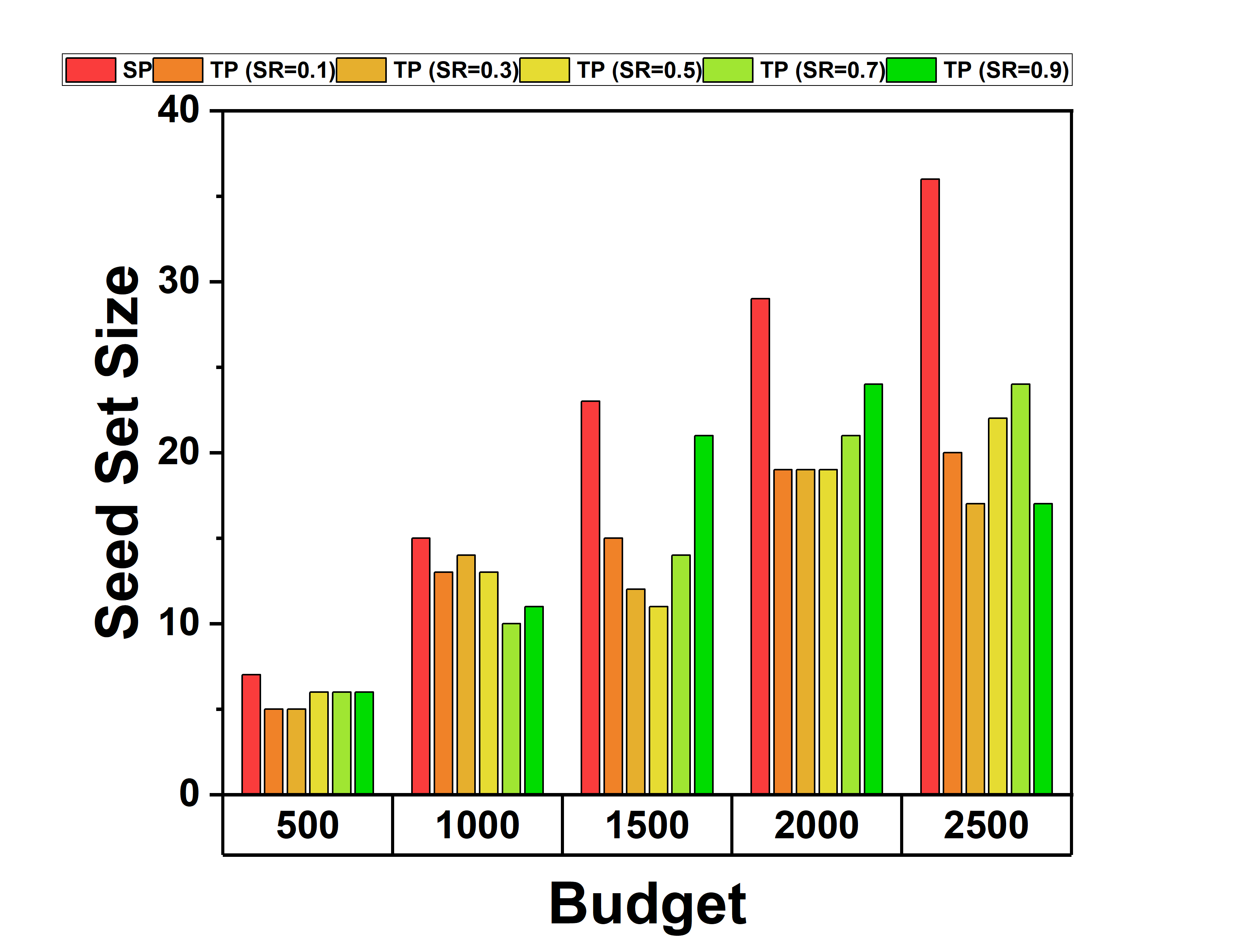}
        \caption{Timestep 8}
    \end{subfigure}
    \hfill
    \begin{subfigure}[t]{0.3\linewidth}
        \centering
        \includegraphics[width=\linewidth]{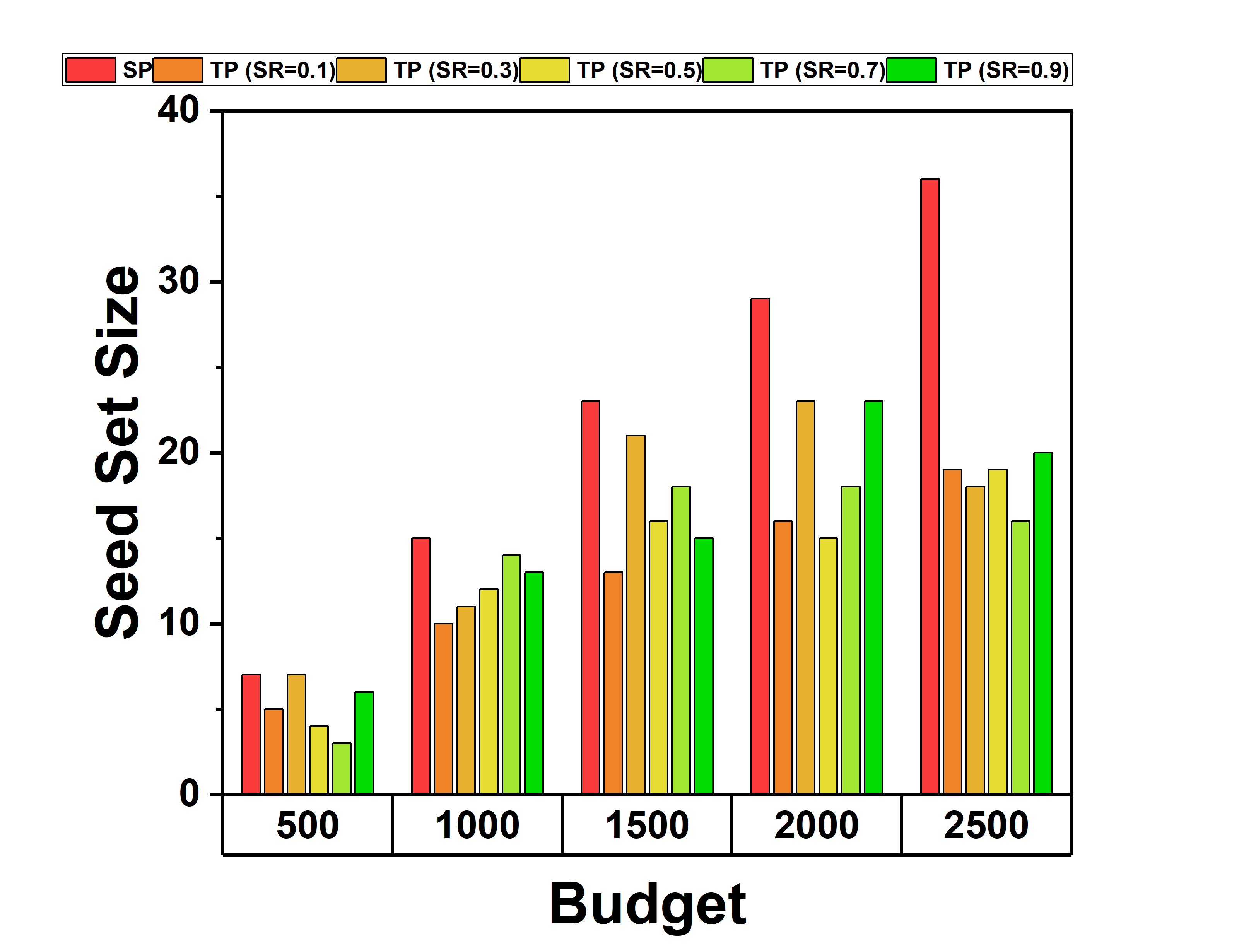}
        \caption{Timestep 10}
    \end{subfigure}

    \caption{Seed Set Size Distribution of Single Phase Vs. Two Phase (Stochastic Greedy Algorithm, \textit{LM} Dataset, Probability Setting - Trivalency)}
    \label{RQ4LM_T8}
\end{figure}


\begin{figure}[htbp]
    \centering

    \begin{subfigure}[t]{0.3\linewidth}
        \centering
        \includegraphics[width=\linewidth]{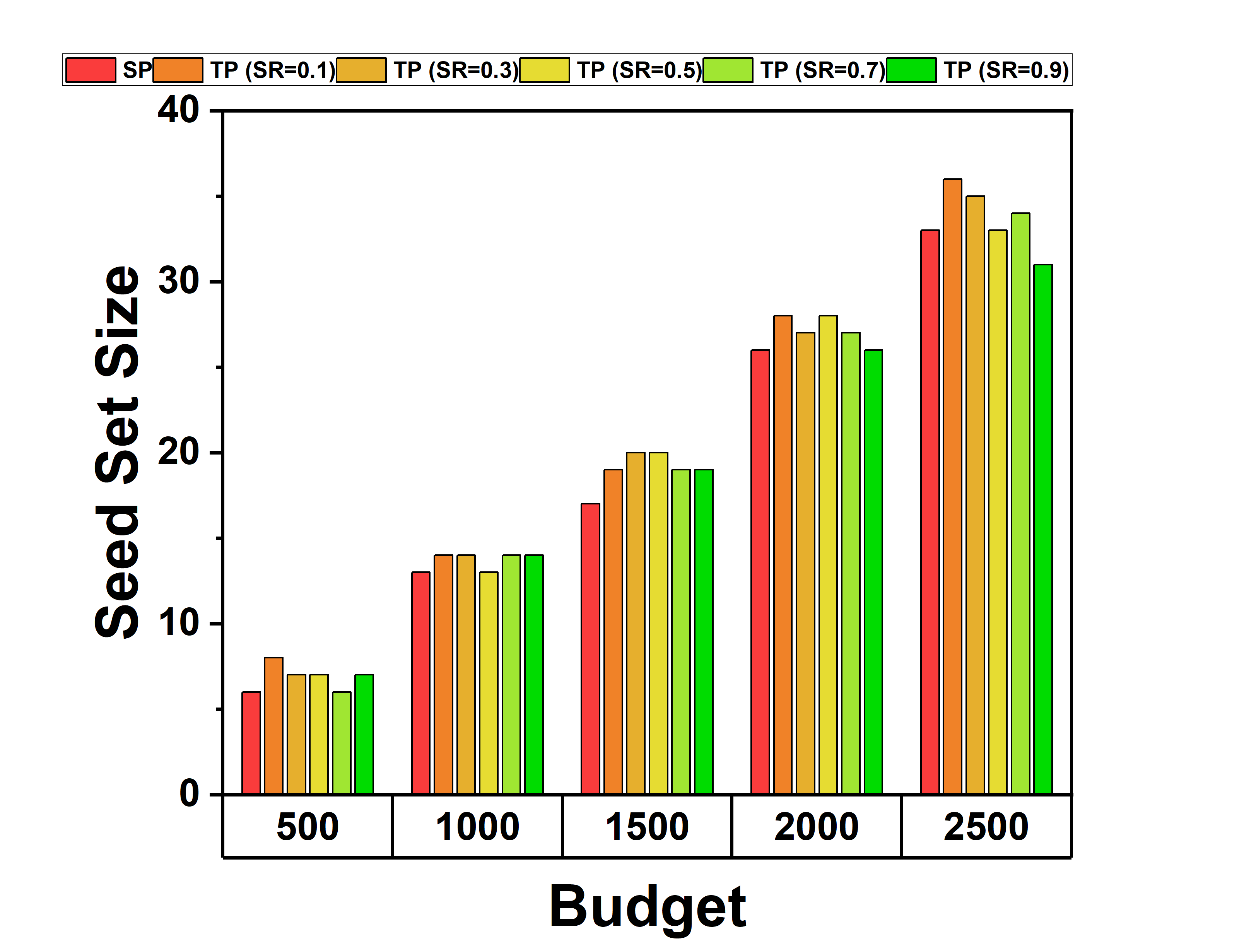}
        \caption{Timestep 2}
    \end{subfigure}
    \hspace{0.05\linewidth}
    \begin{subfigure}[t]{0.3\linewidth}
        \centering
        \includegraphics[width=\linewidth]{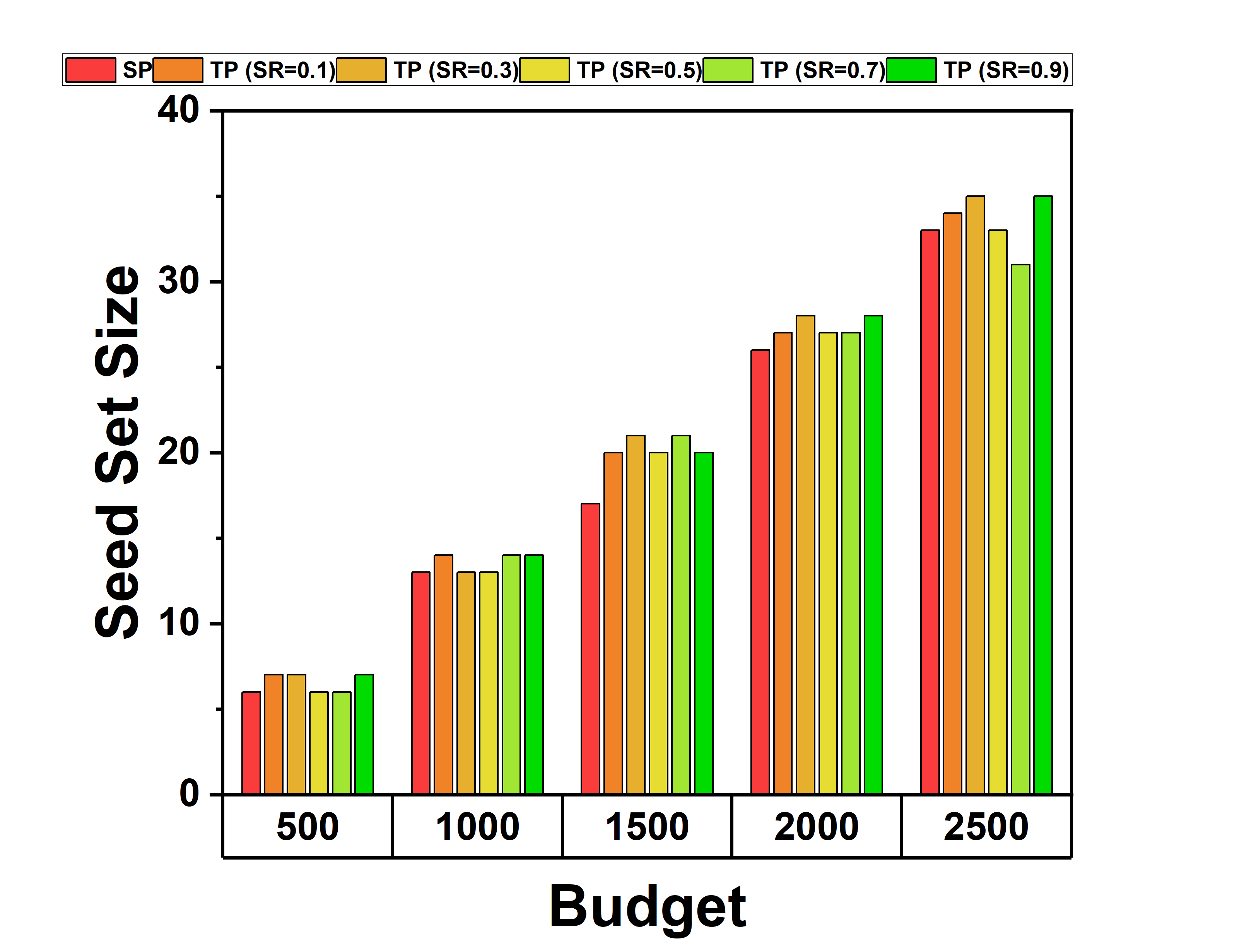}
        \caption{Timestep 4}
    \end{subfigure}

    \vspace{0.5cm}

    \begin{subfigure}[t]{0.3\linewidth}
        \centering
        \includegraphics[width=\linewidth]{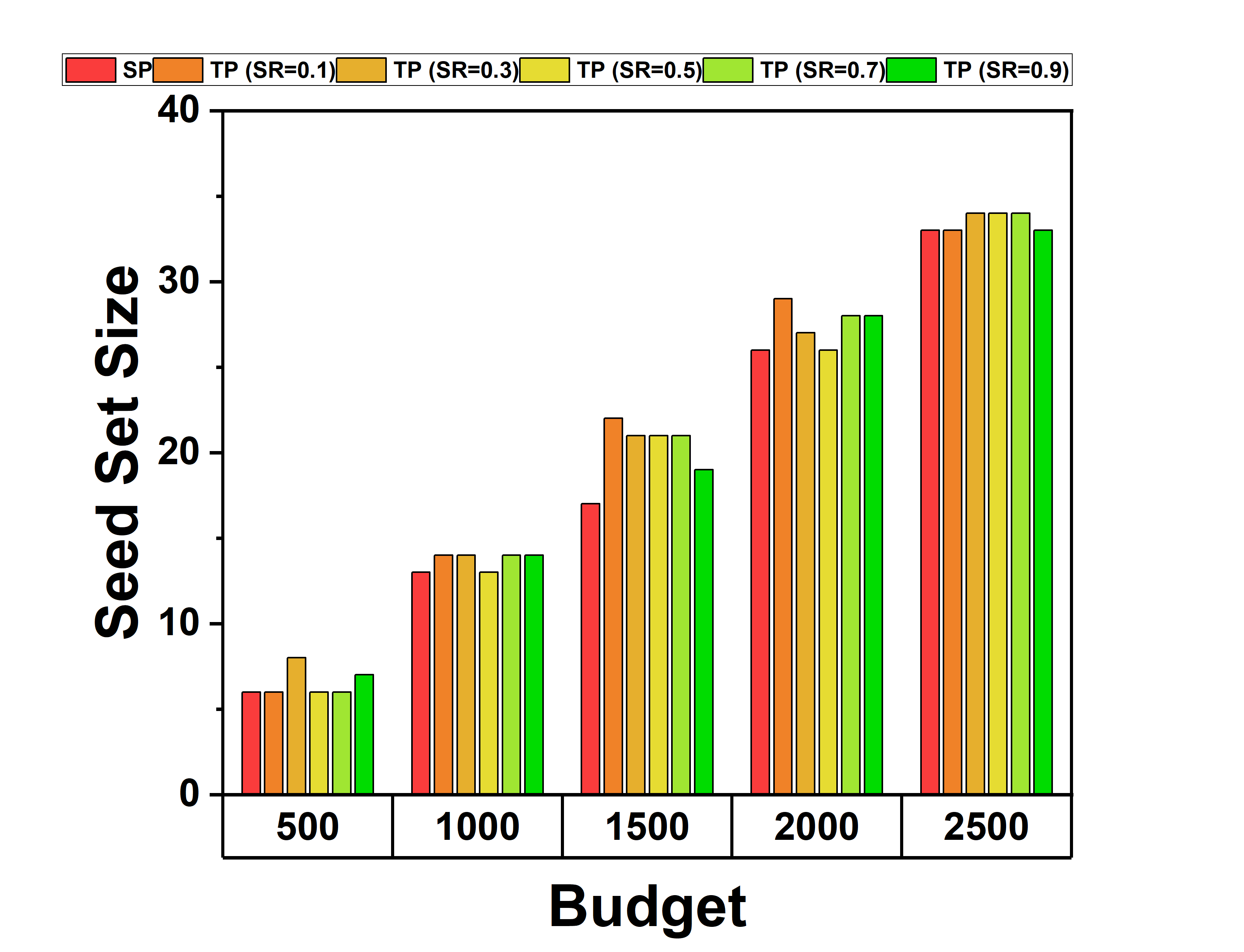}
        \caption{Timestep 6}
    \end{subfigure}
    \hfill
    \begin{subfigure}[t]{0.3\linewidth}
        \centering
        \includegraphics[width=\linewidth]{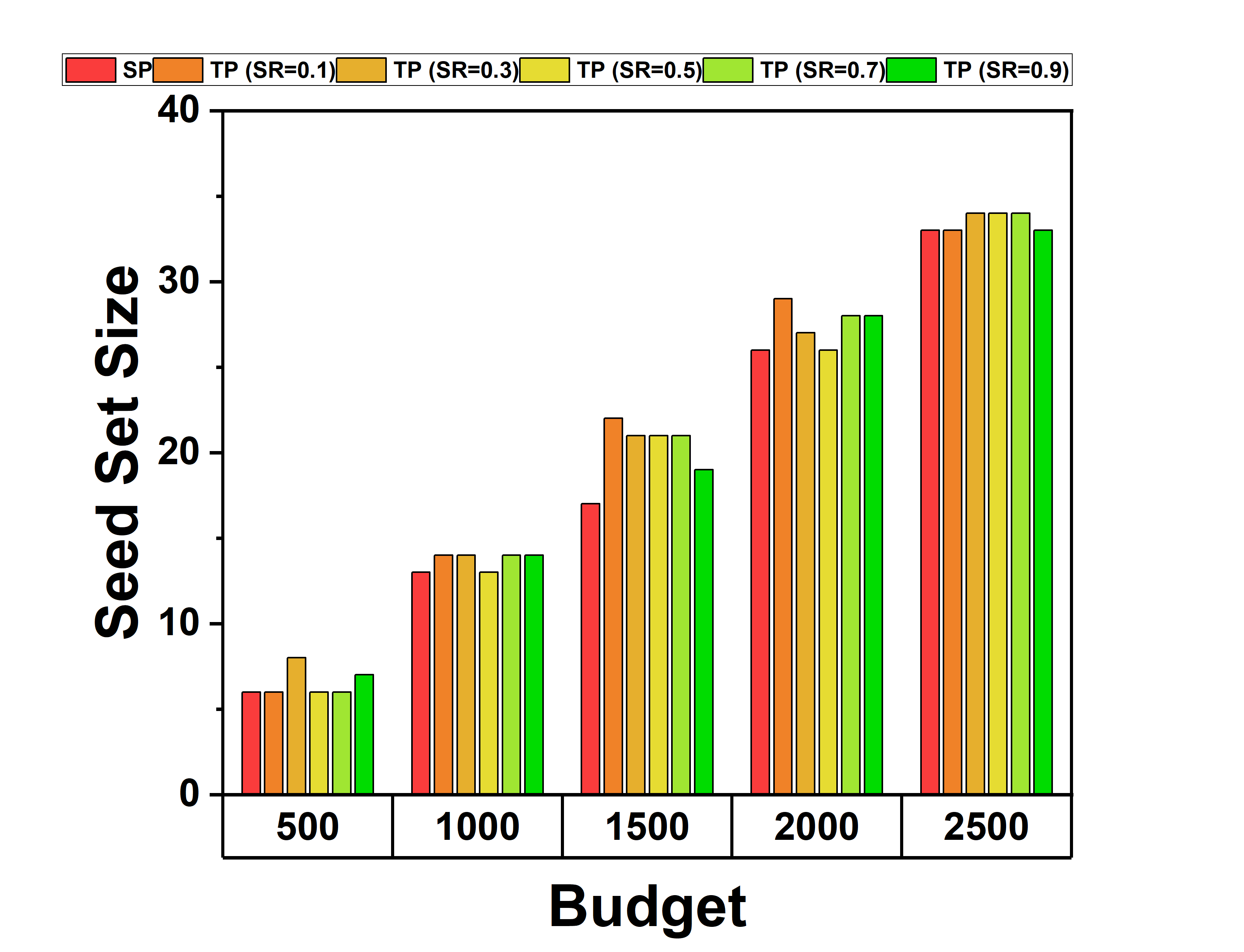}
        \caption{Timestep 8}
    \end{subfigure}
    \hfill
    \begin{subfigure}[t]{0.3\linewidth}
        \centering
        \includegraphics[width=\linewidth]{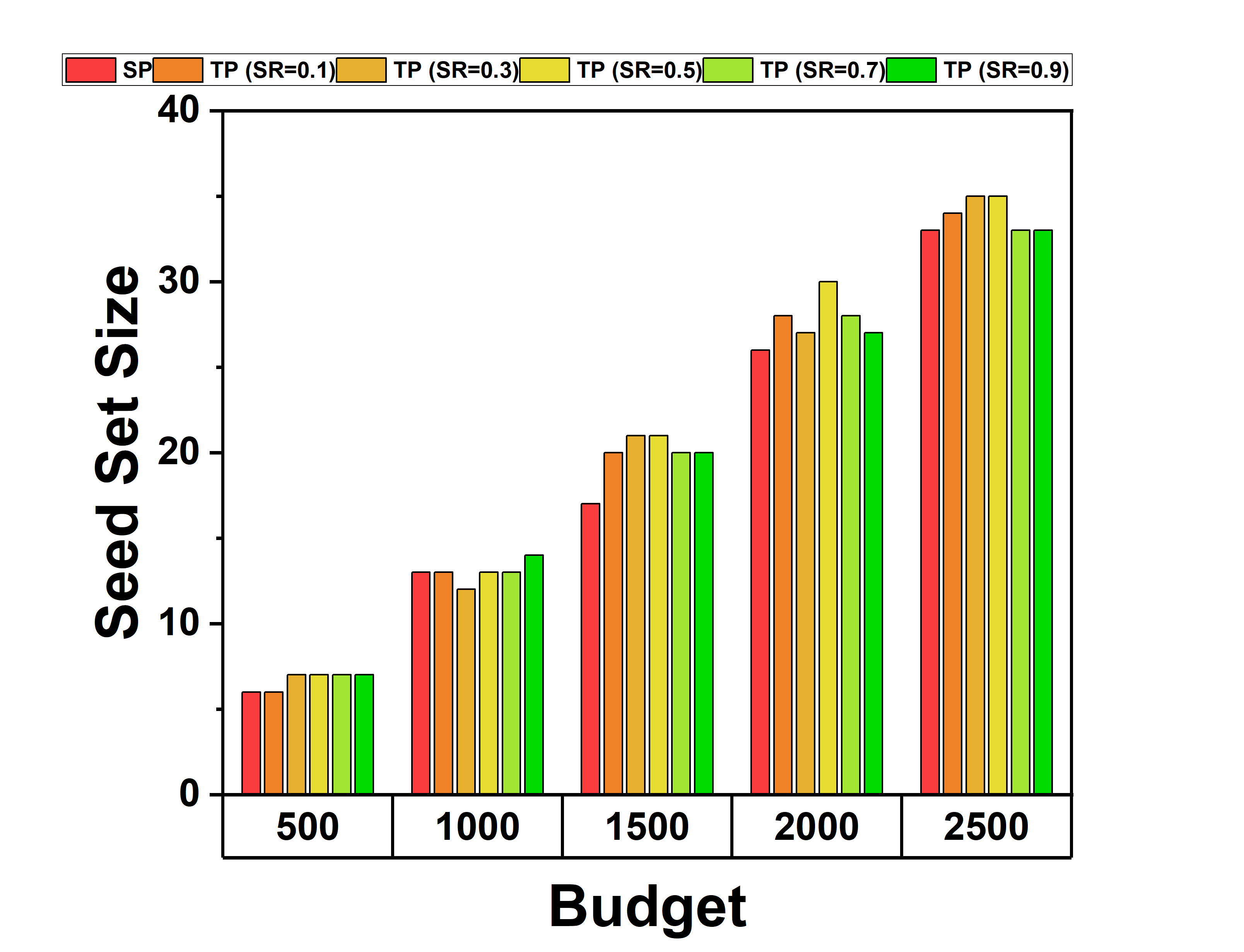}
        \caption{Timestep 10}
    \end{subfigure}

    \caption{Seed Set Size Distribution of Single Phase Vs. Two Phase (Random Algorithm, \textit{Email-Eu-Core} Dataset, Probability Setting - Trivalency)}
    \label{RQ4_T1}
\end{figure}

\begin{figure}[htbp]
    \centering

    \begin{subfigure}[t]{0.3\linewidth}
        \centering
        \includegraphics[width=\linewidth]{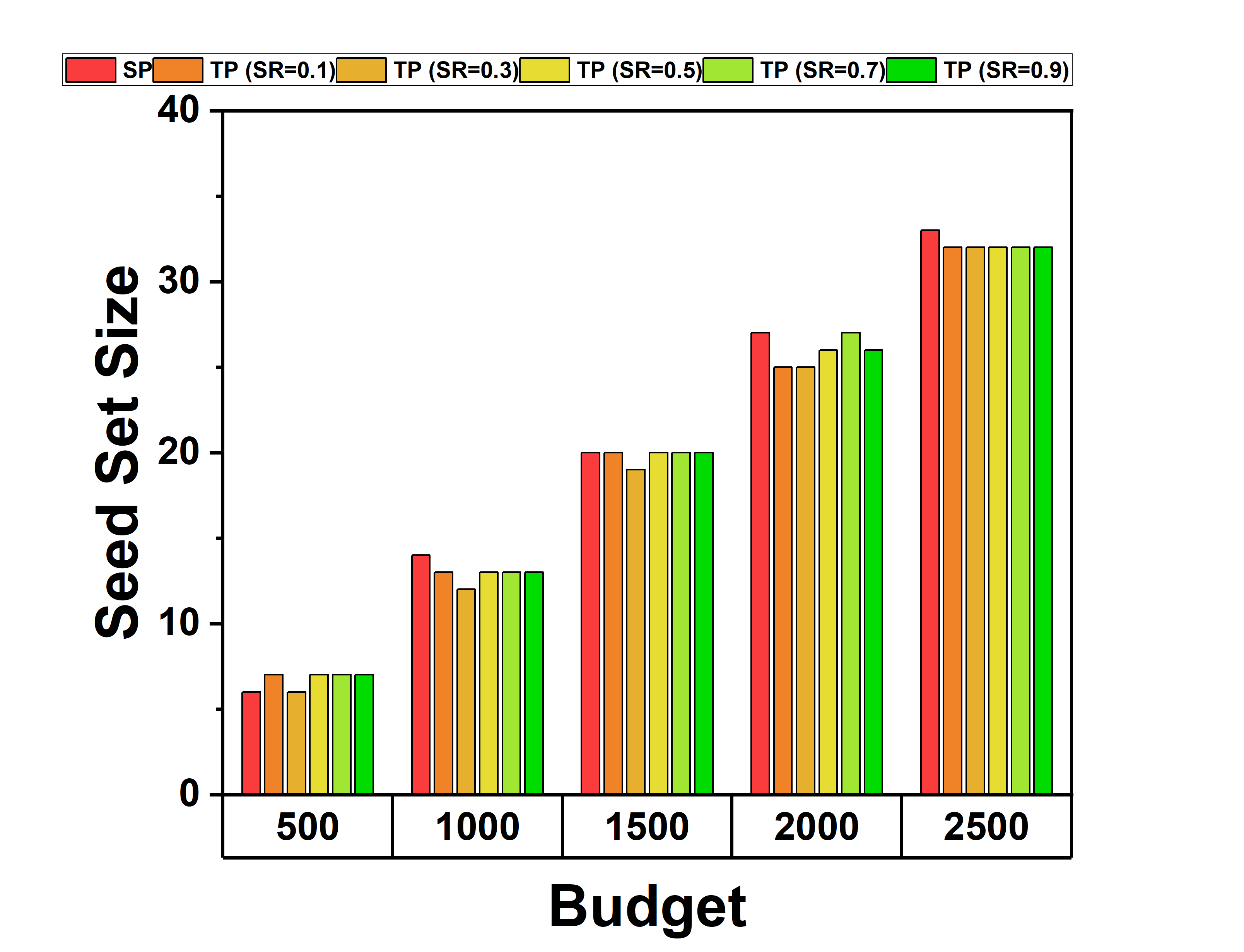}
        \caption{Timestep 2}
    \end{subfigure}
    \hspace{0.05\linewidth}
    \begin{subfigure}[t]{0.3\linewidth}
        \centering
        \includegraphics[width=\linewidth]{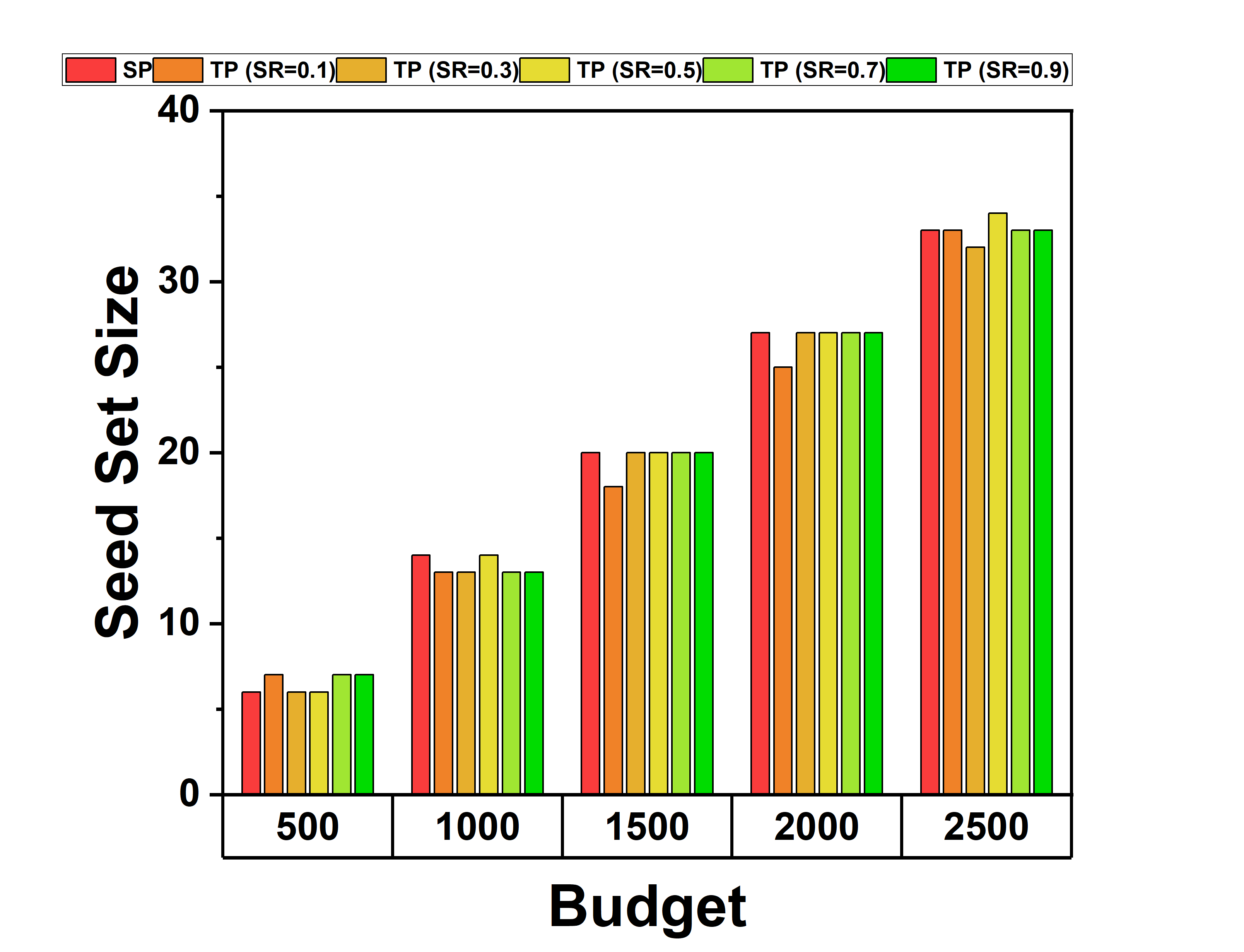}
        \caption{Timestep 4}
    \end{subfigure}

    \vspace{0.5cm}

    \begin{subfigure}[t]{0.3\linewidth}
        \centering
        \includegraphics[width=\linewidth]{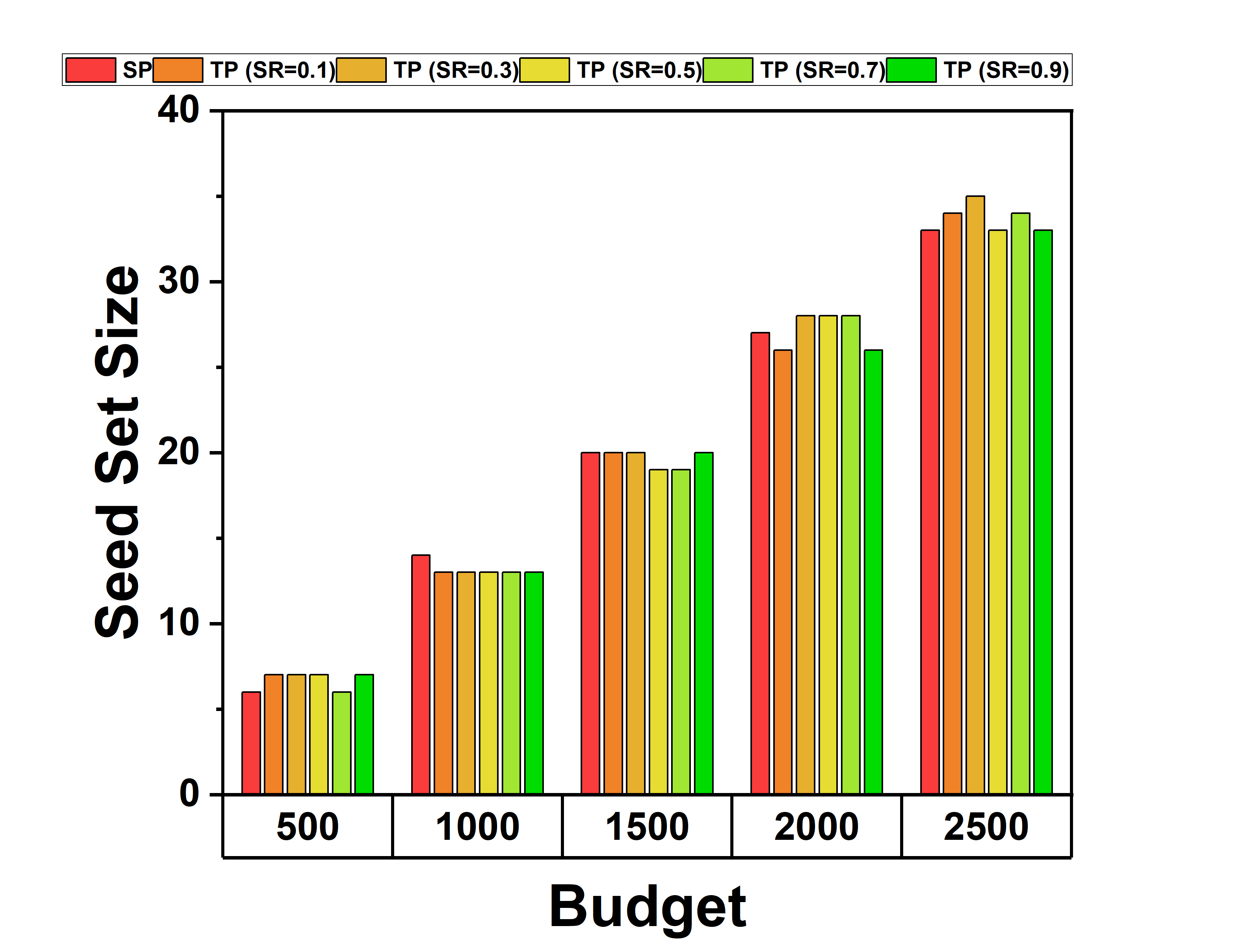}
        \caption{Timestep 6}
    \end{subfigure}
    \hfill
    \begin{subfigure}[t]{0.3\linewidth}
        \centering
        \includegraphics[width=\linewidth]{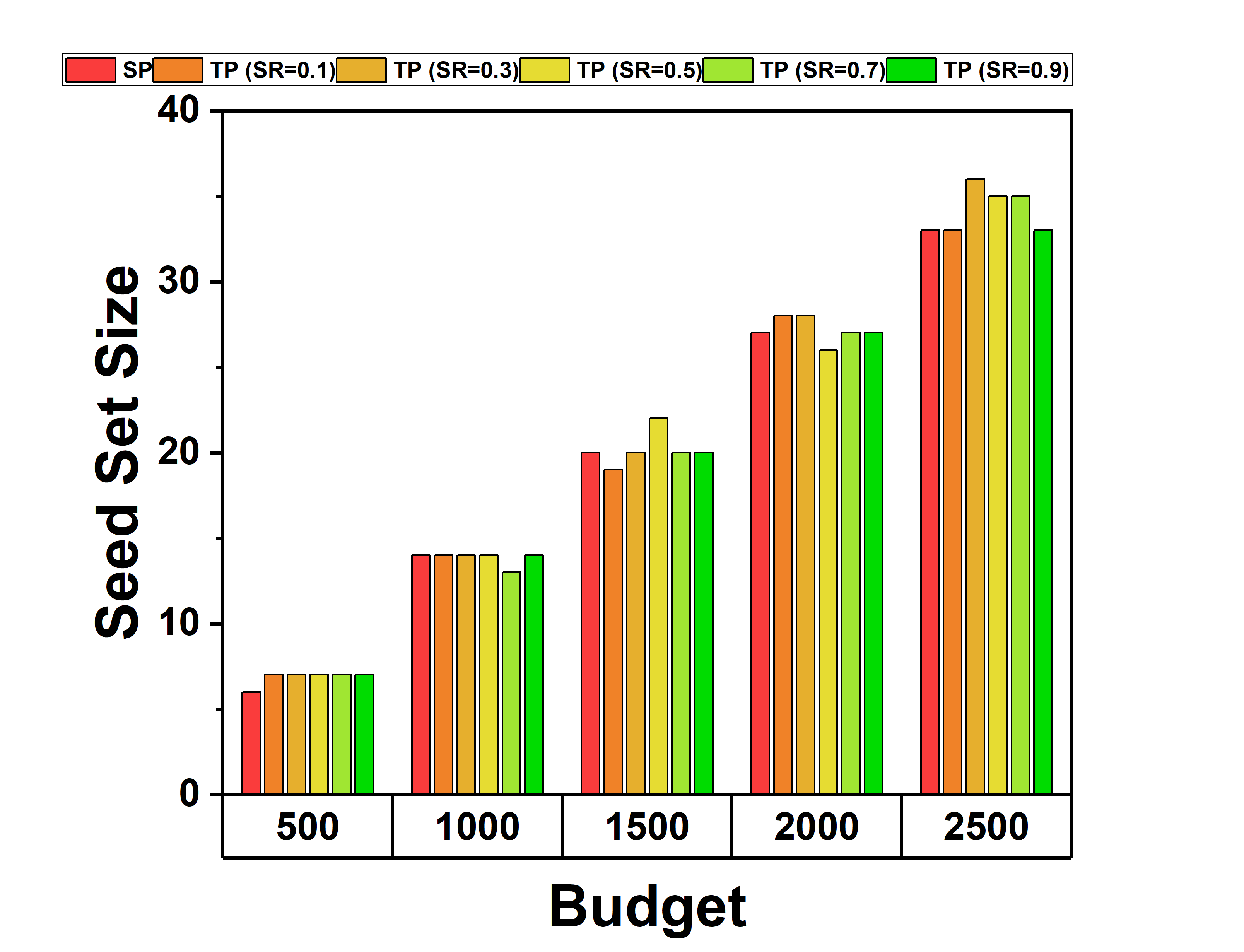}
        \caption{Timestep 8}
    \end{subfigure}
    \hfill
    \begin{subfigure}[t]{0.3\linewidth}
        \centering
        \includegraphics[width=\linewidth]{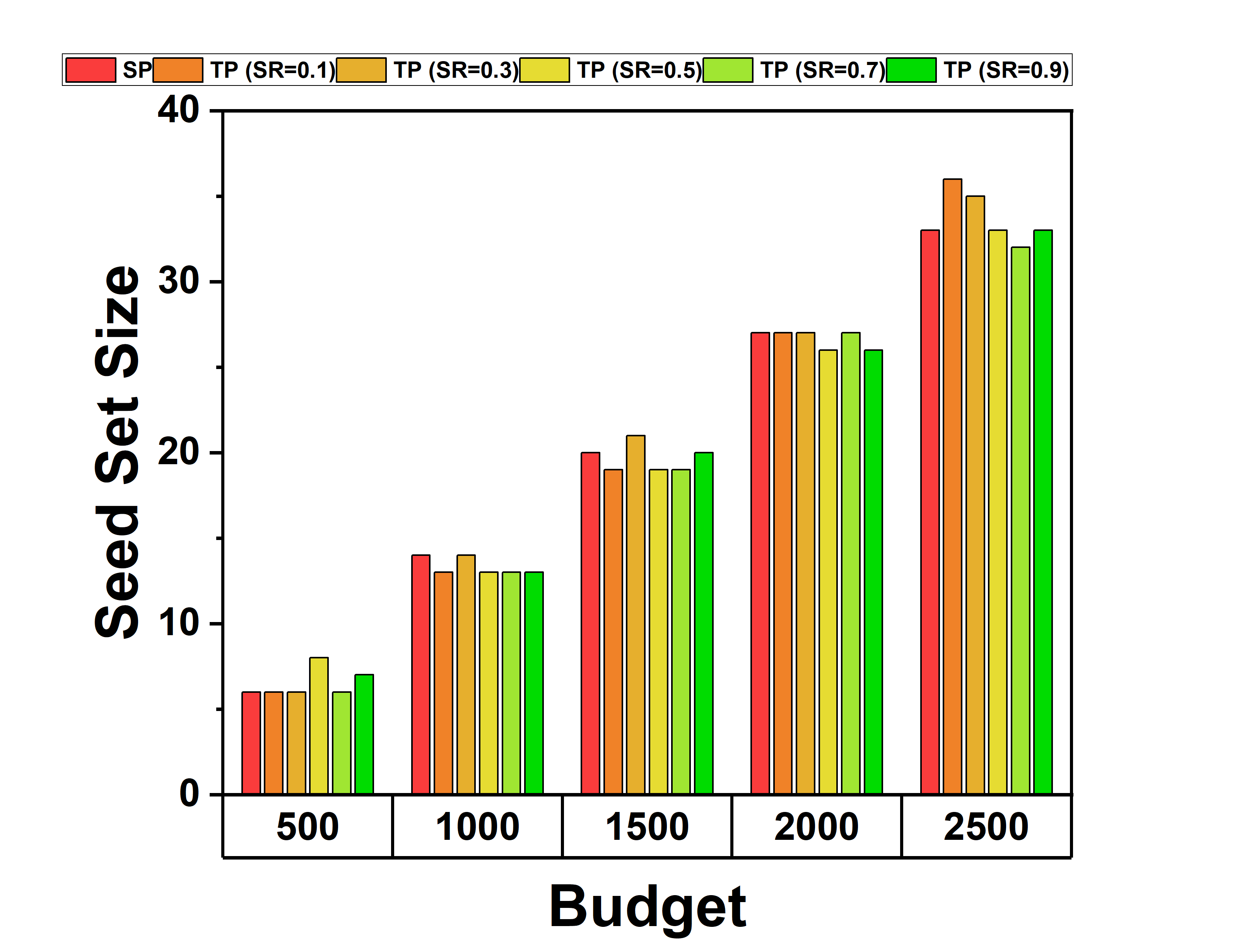}
        \caption{Timestep 10}
    \end{subfigure}

    \caption{Seed Set Size Distribution of Single Phase Vs. Two Phase (High Degree Algorithm, \textit{Email-Eu-Core} Dataset, Probability Setting - Trivalency)}
    \label{RQ4_T2}
\end{figure}

\begin{figure}[htbp]
    \centering

    \begin{subfigure}[t]{0.3\linewidth}
        \centering
        \includegraphics[width=\linewidth]{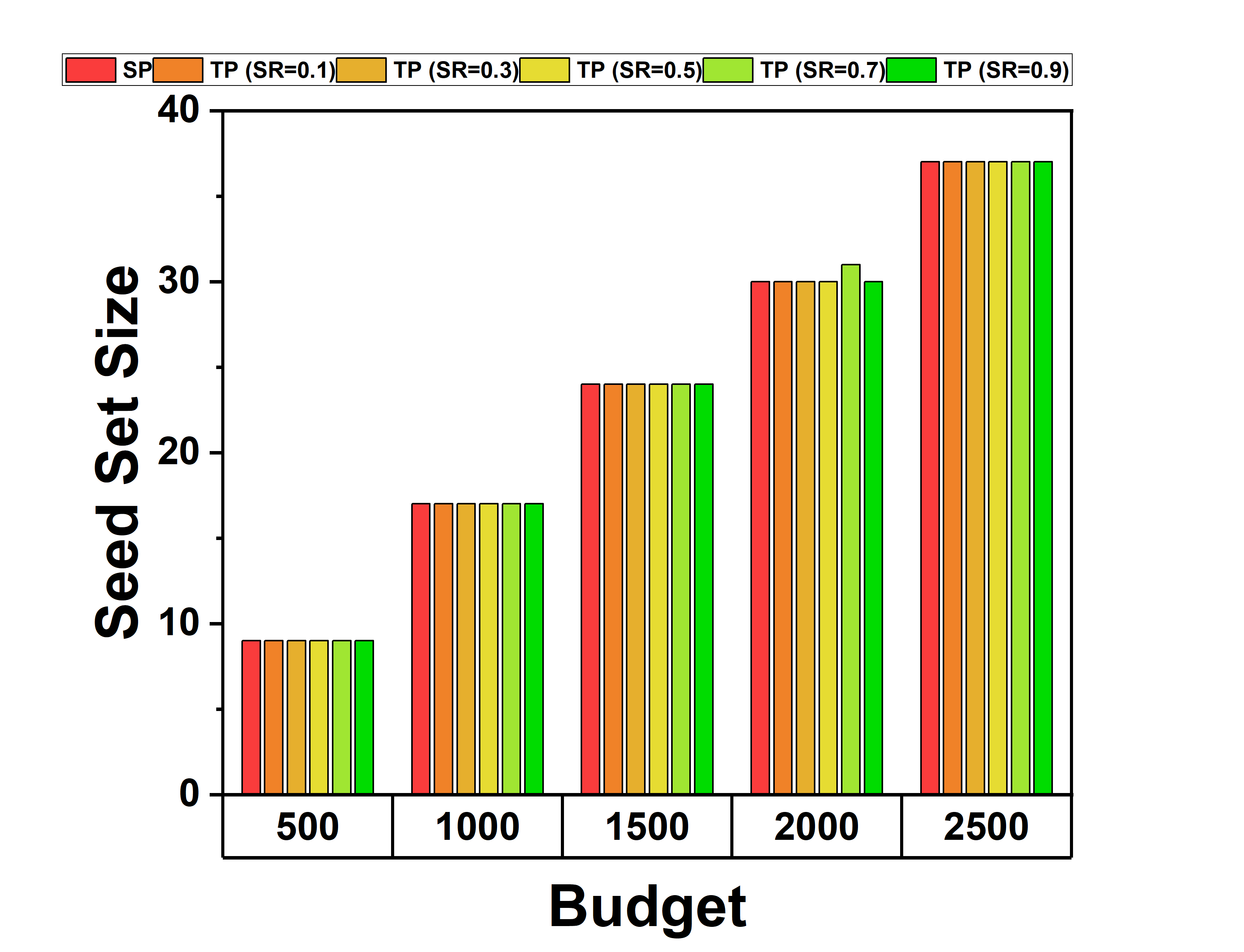}
        \caption{Timestep 2}
    \end{subfigure}
    \hspace{0.05\linewidth}
    \begin{subfigure}[t]{0.3\linewidth}
        \centering
        \includegraphics[width=\linewidth]{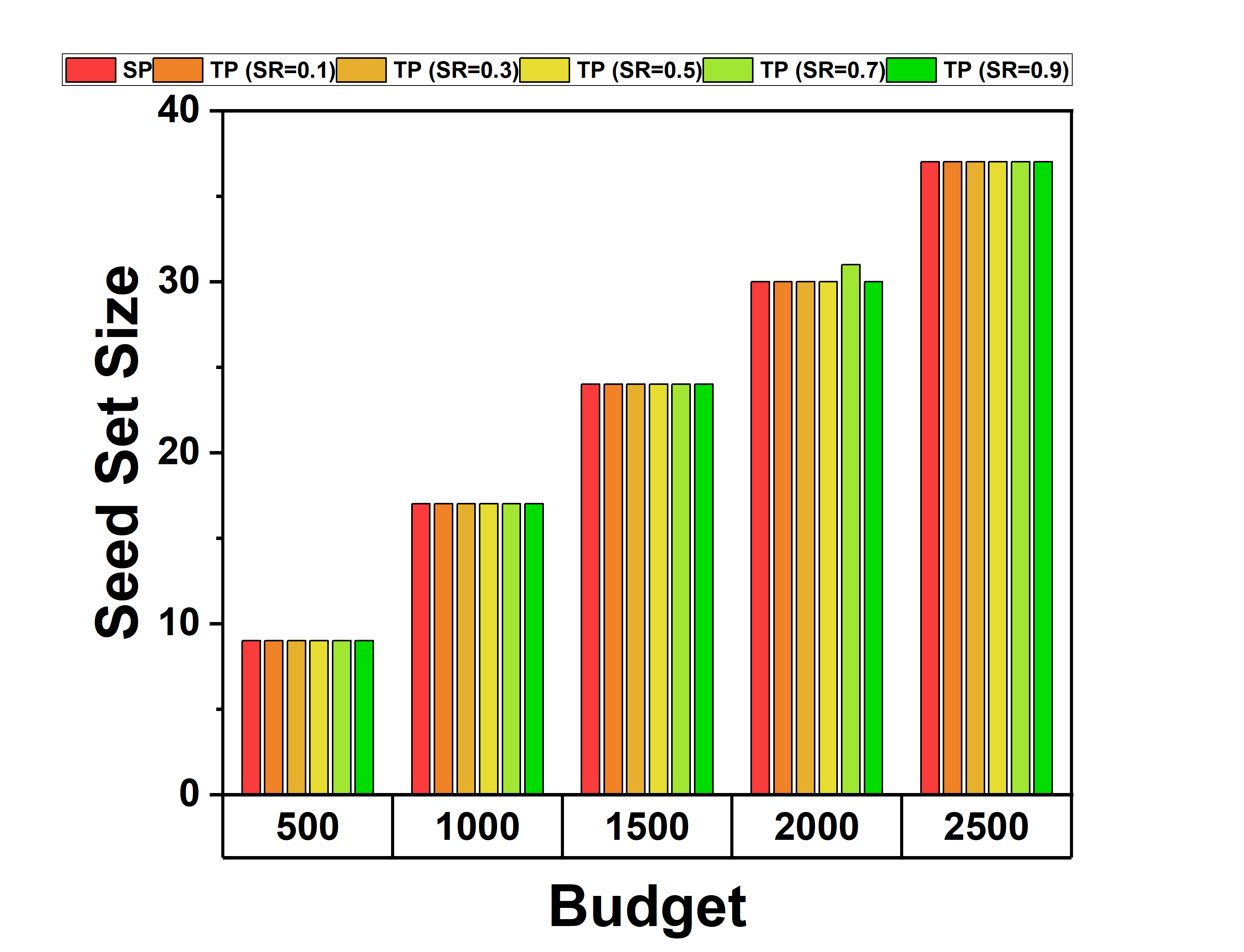}
        \caption{Timestep 4}
    \end{subfigure}

    \vspace{0.5cm}

    \begin{subfigure}[t]{0.3\linewidth}
        \centering
        \includegraphics[width=\linewidth]{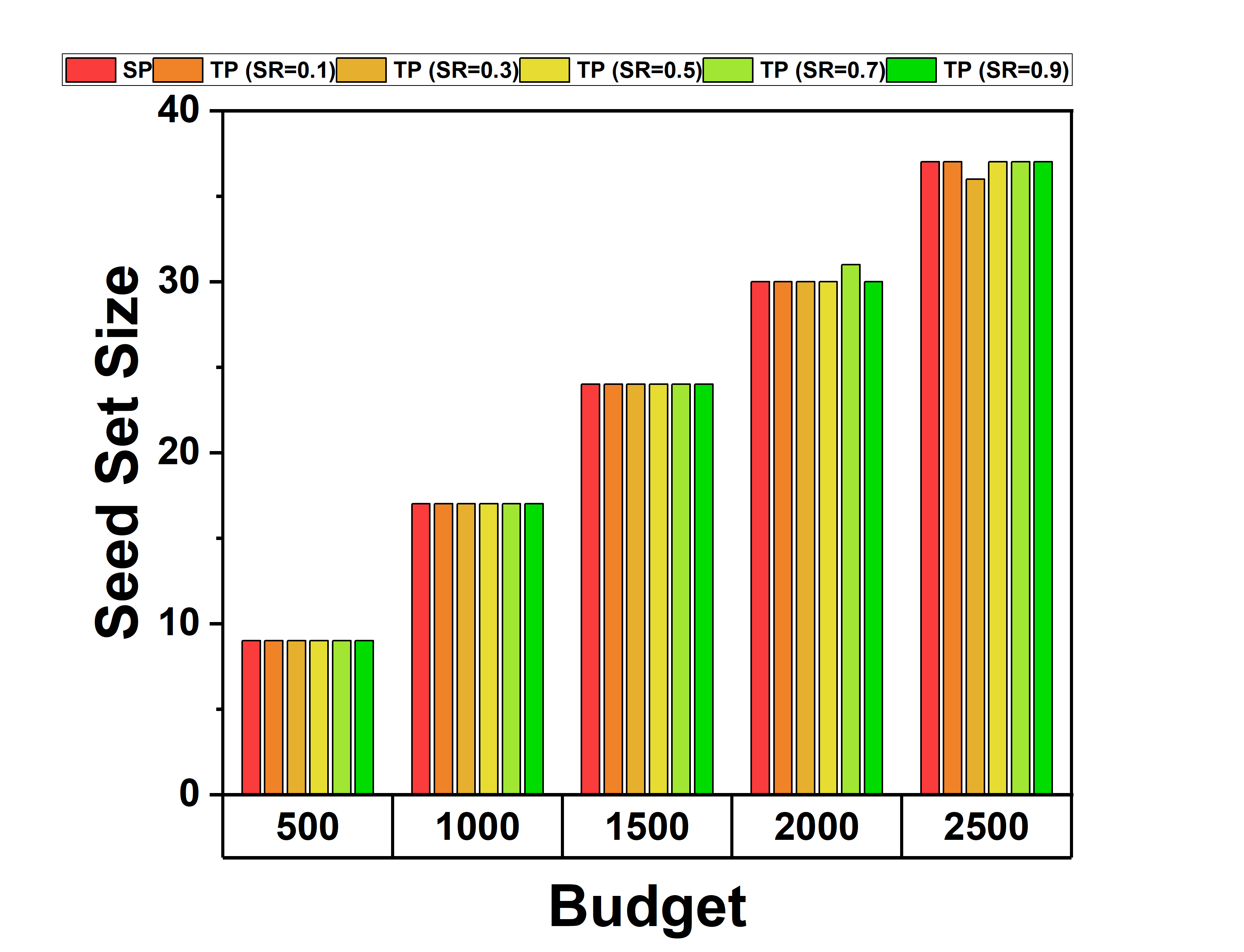}
        \caption{Timestep 6}
    \end{subfigure}
    \hfill
    \begin{subfigure}[t]{0.3\linewidth}
        \centering
        \includegraphics[width=\linewidth]{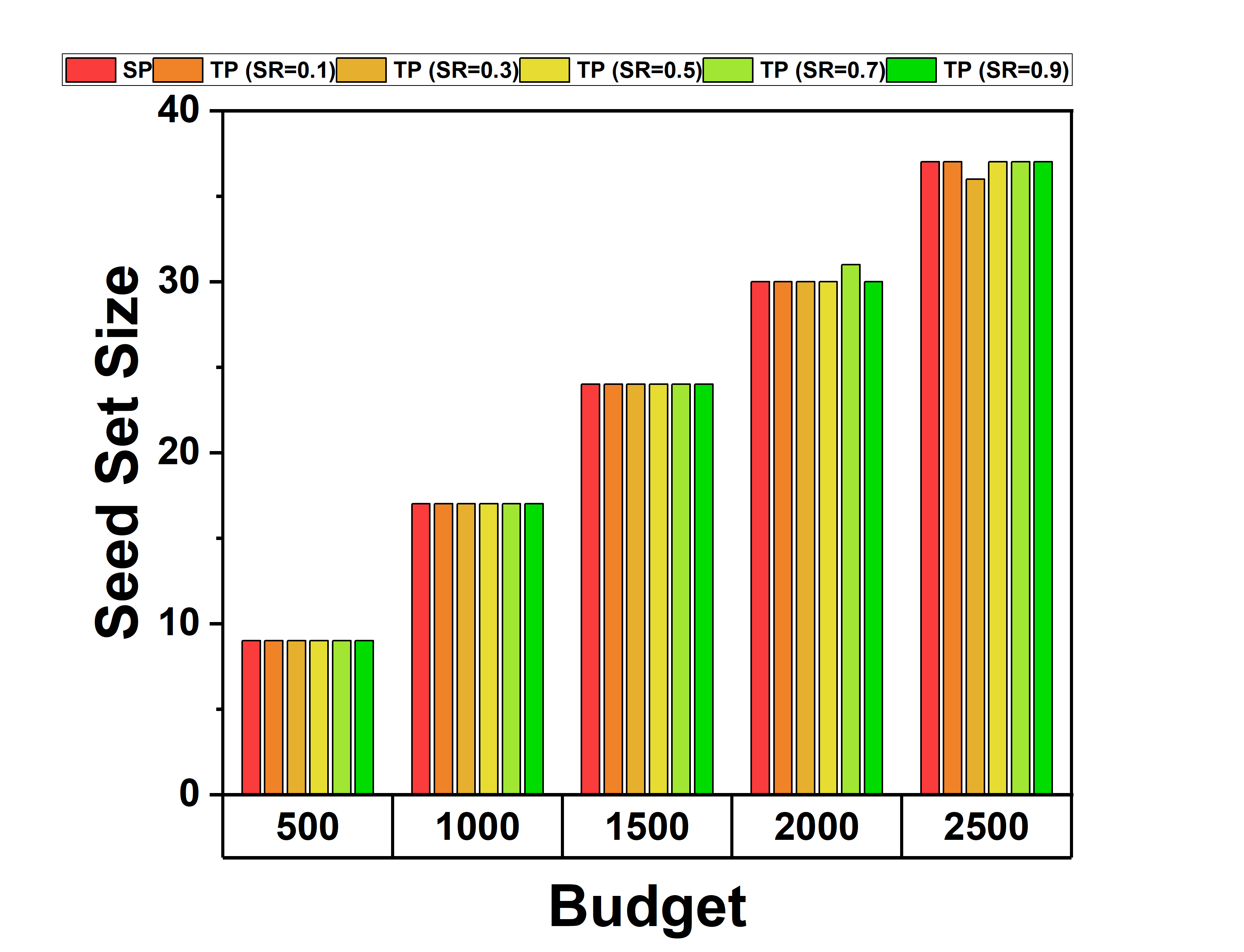}
        \caption{Timestep 8}
    \end{subfigure}
    \hfill
    \begin{subfigure}[t]{0.3\linewidth}
        \centering
        \includegraphics[width=\linewidth]{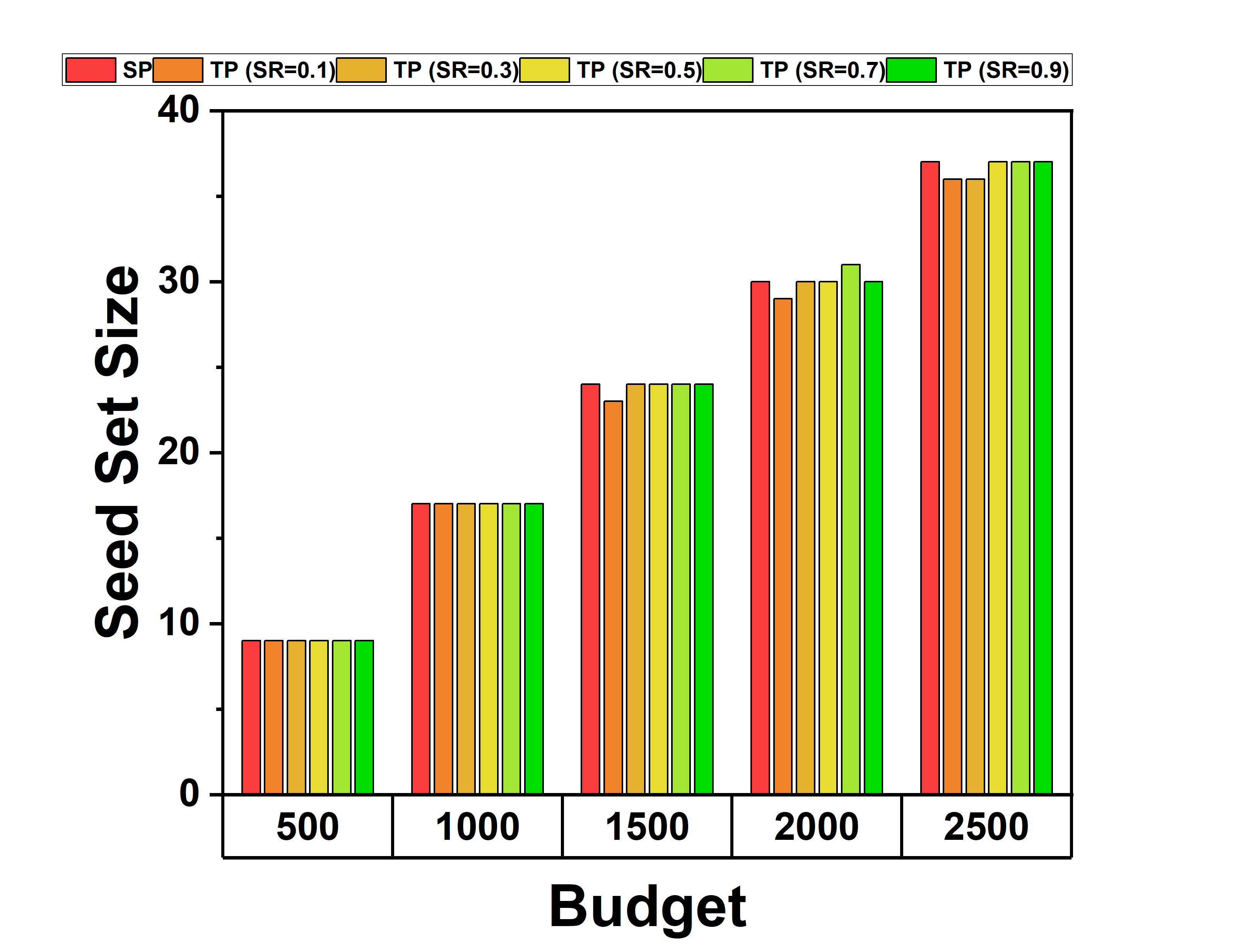}
        \caption{Timestep 10}
    \end{subfigure}

    \caption{Seed Set Size Distribution of Single Phase Vs. Two Phase (Clustering Coefficient  Algorithm, \textit{Email-Eu-Core} Dataset, Probability Setting - Trivalency)}
    \label{RQ4_T3}
\end{figure}

\begin{figure}[htbp]
    \centering

    \begin{subfigure}[t]{0.3\linewidth}
        \centering
        \includegraphics[width=\linewidth]{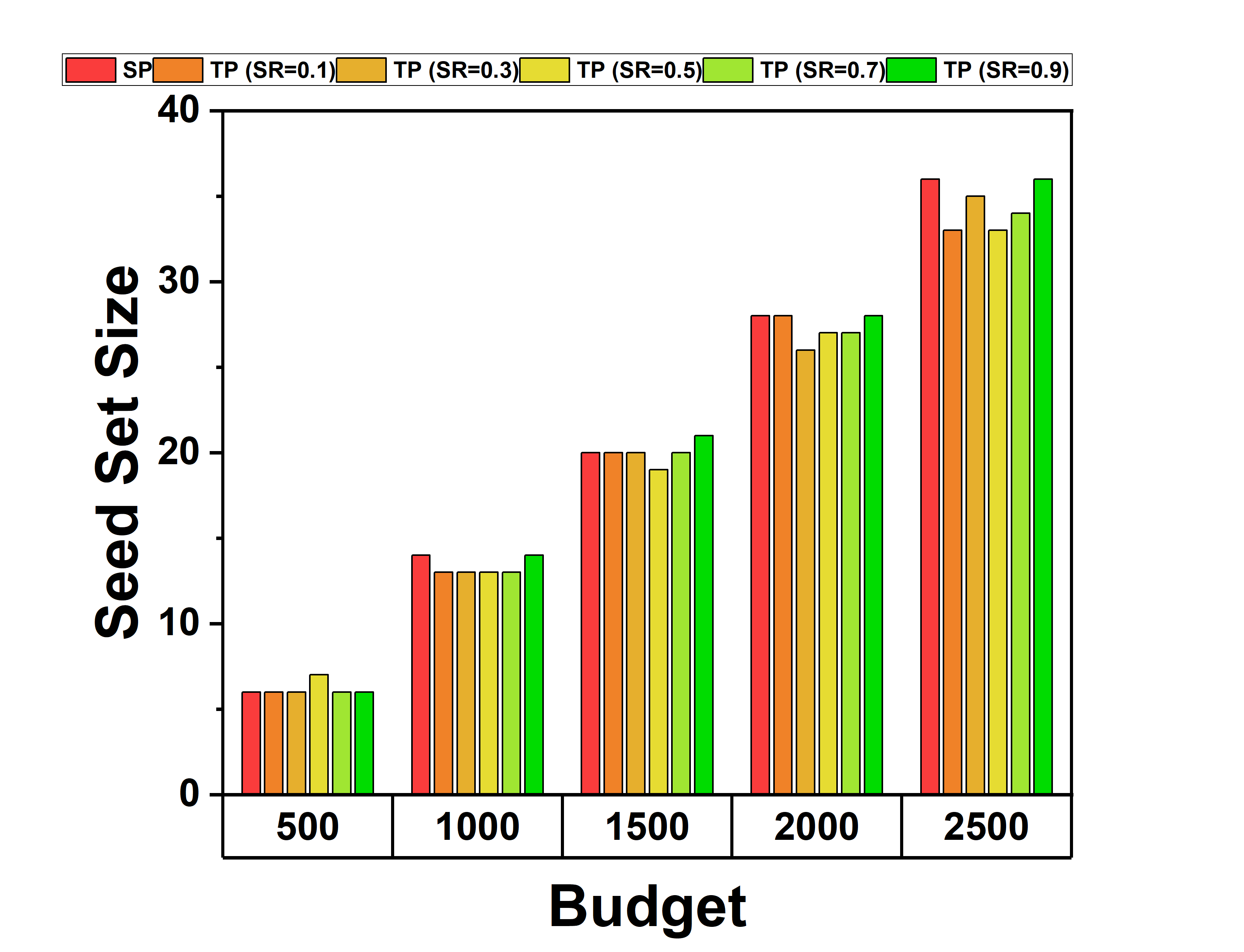}
        \caption{Timestep 2}
    \end{subfigure}
    \hspace{0.05\linewidth}
    \begin{subfigure}[t]{0.3\linewidth}
        \centering
        \includegraphics[width=\linewidth]{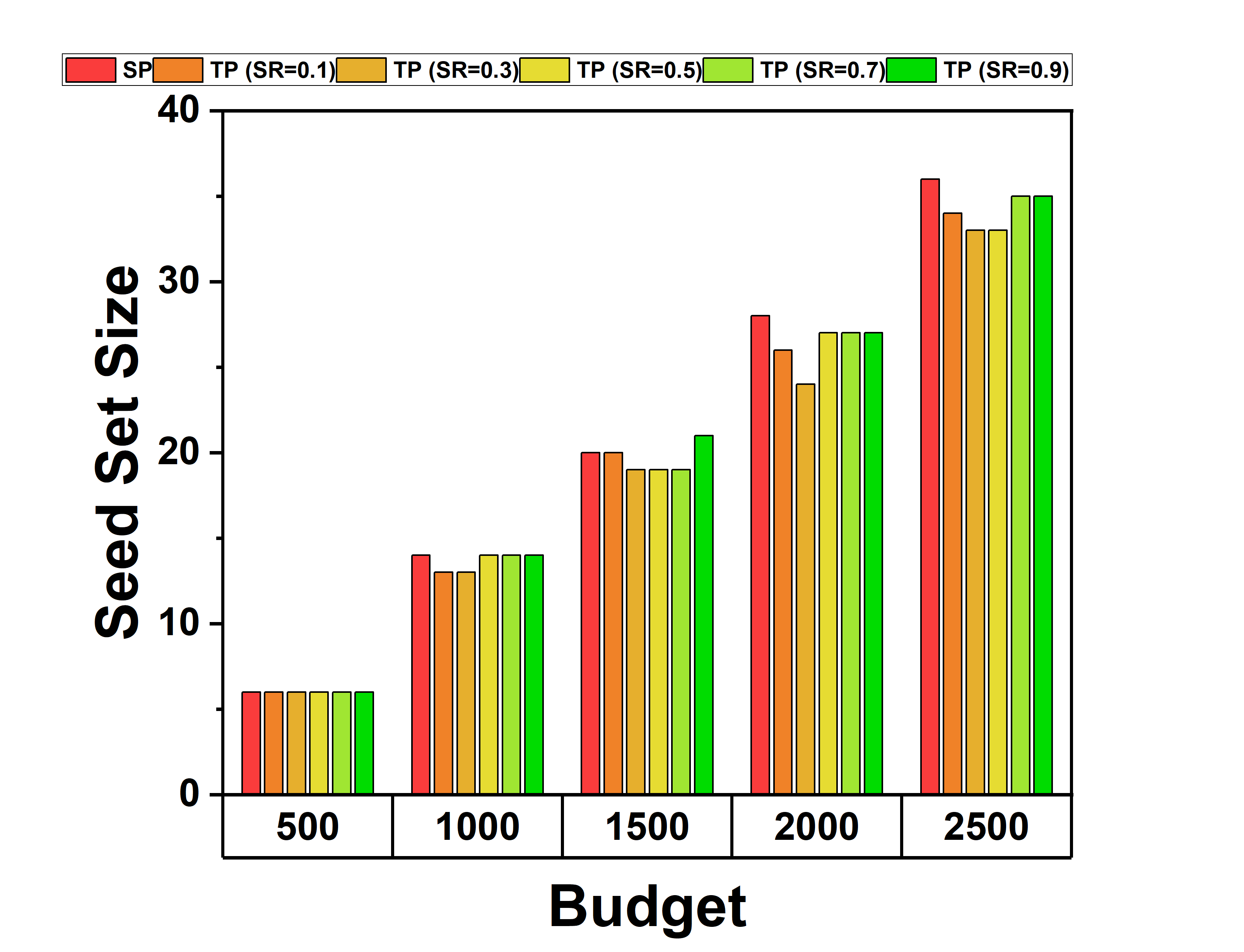}
        \caption{Timestep 4}
    \end{subfigure}

    \vspace{0.5cm}

    \begin{subfigure}[t]{0.3\linewidth}
        \centering
        \includegraphics[width=\linewidth]{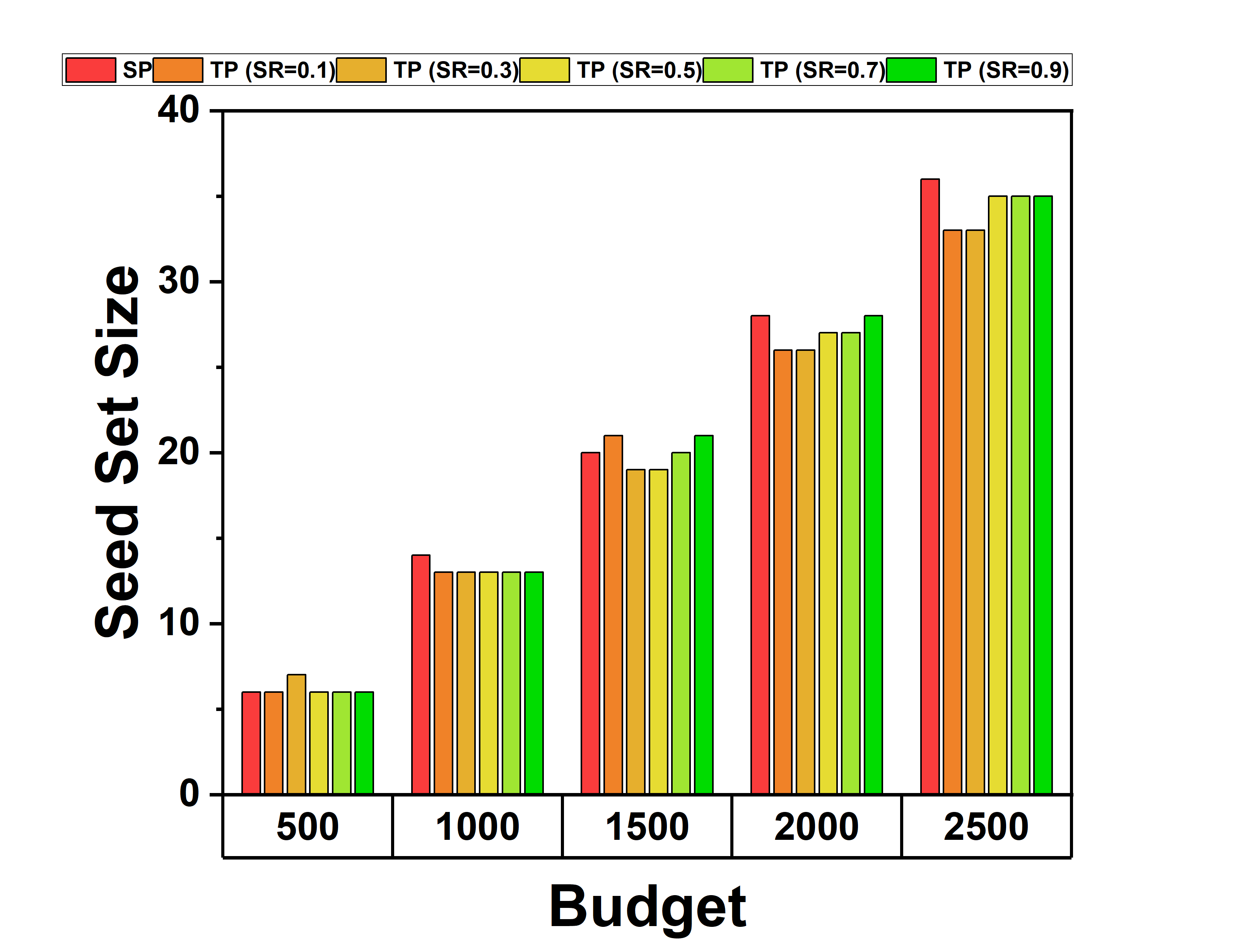}
        \caption{Timestep 6}
    \end{subfigure}
    \hfill
    \begin{subfigure}[t]{0.3\linewidth}
        \centering
        \includegraphics[width=\linewidth]{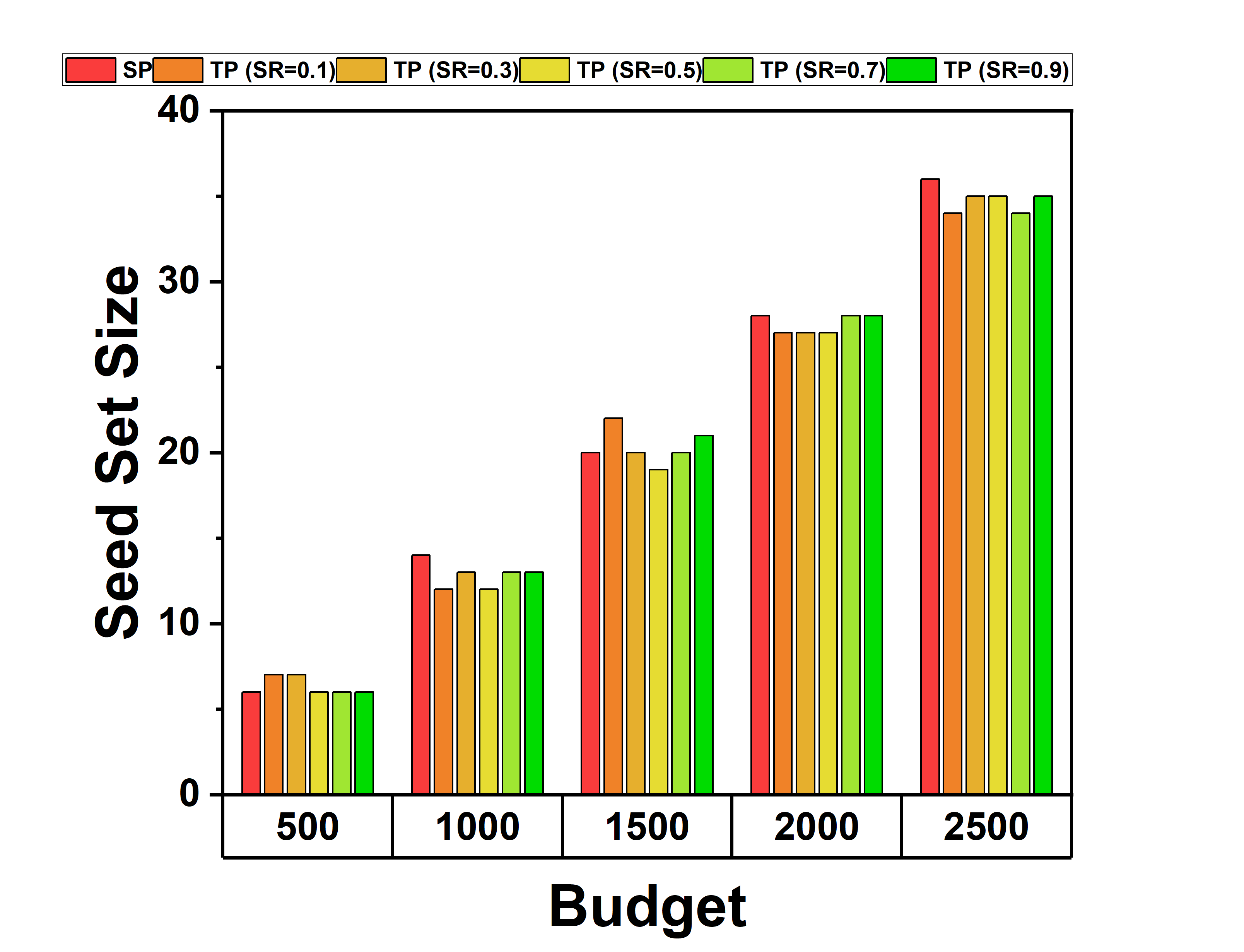}
        \caption{Timestep 8}
    \end{subfigure}
    \hfill
    \begin{subfigure}[t]{0.3\linewidth}
        \centering
        \includegraphics[width=\linewidth]{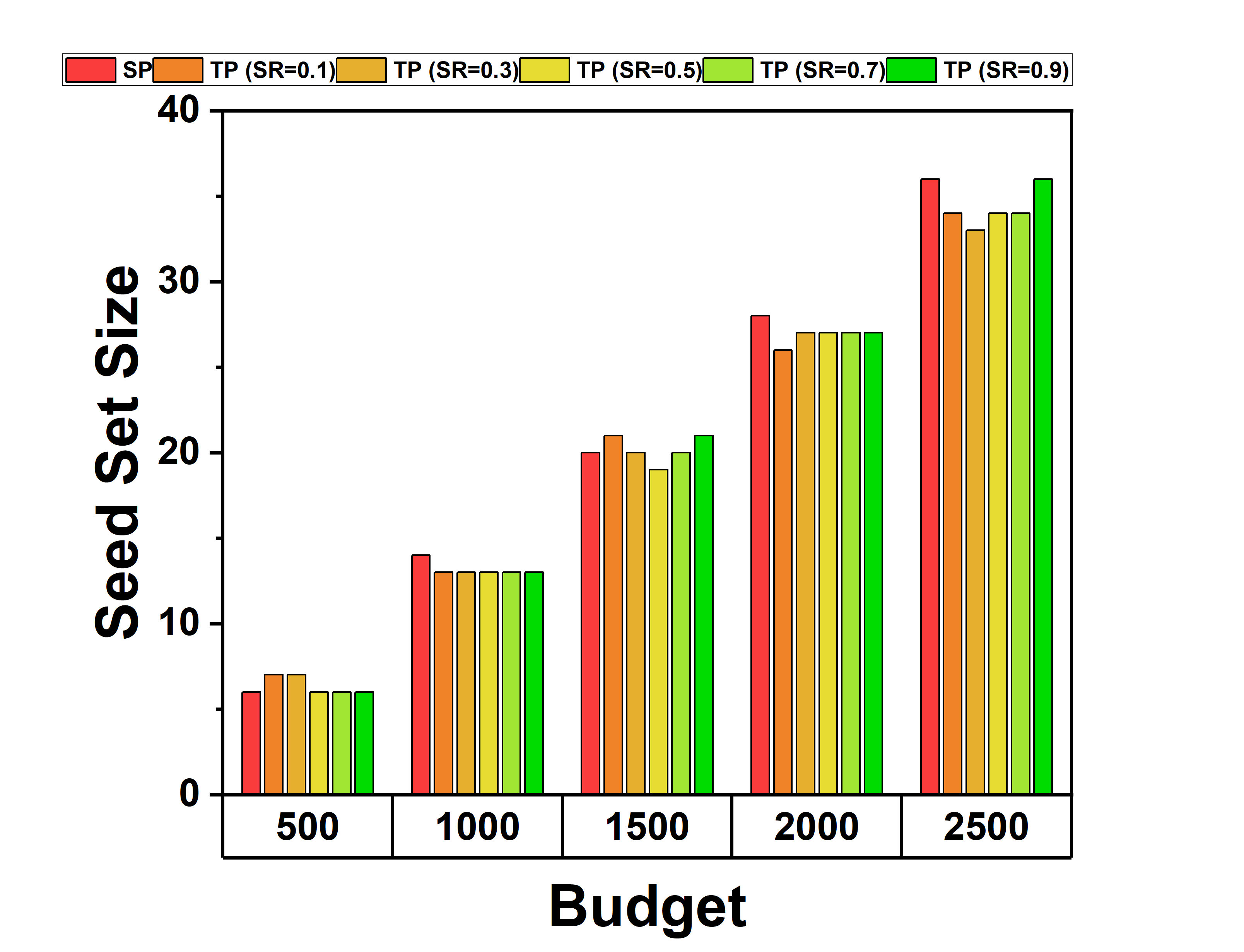}
        \caption{Timestep 10}
    \end{subfigure}

    \caption{Seed Set Size Distribution of Single Phase Vs. Two Phase (Degree Discount Algorithm, \textit{Email-Eu-Core} Dataset, Probability Setting - Trivalency)}
    \label{RQ4_T4}
\end{figure}

\begin{figure}[htbp]
    \centering

    \begin{subfigure}[t]{0.3\linewidth}
        \centering
        \includegraphics[width=\linewidth]{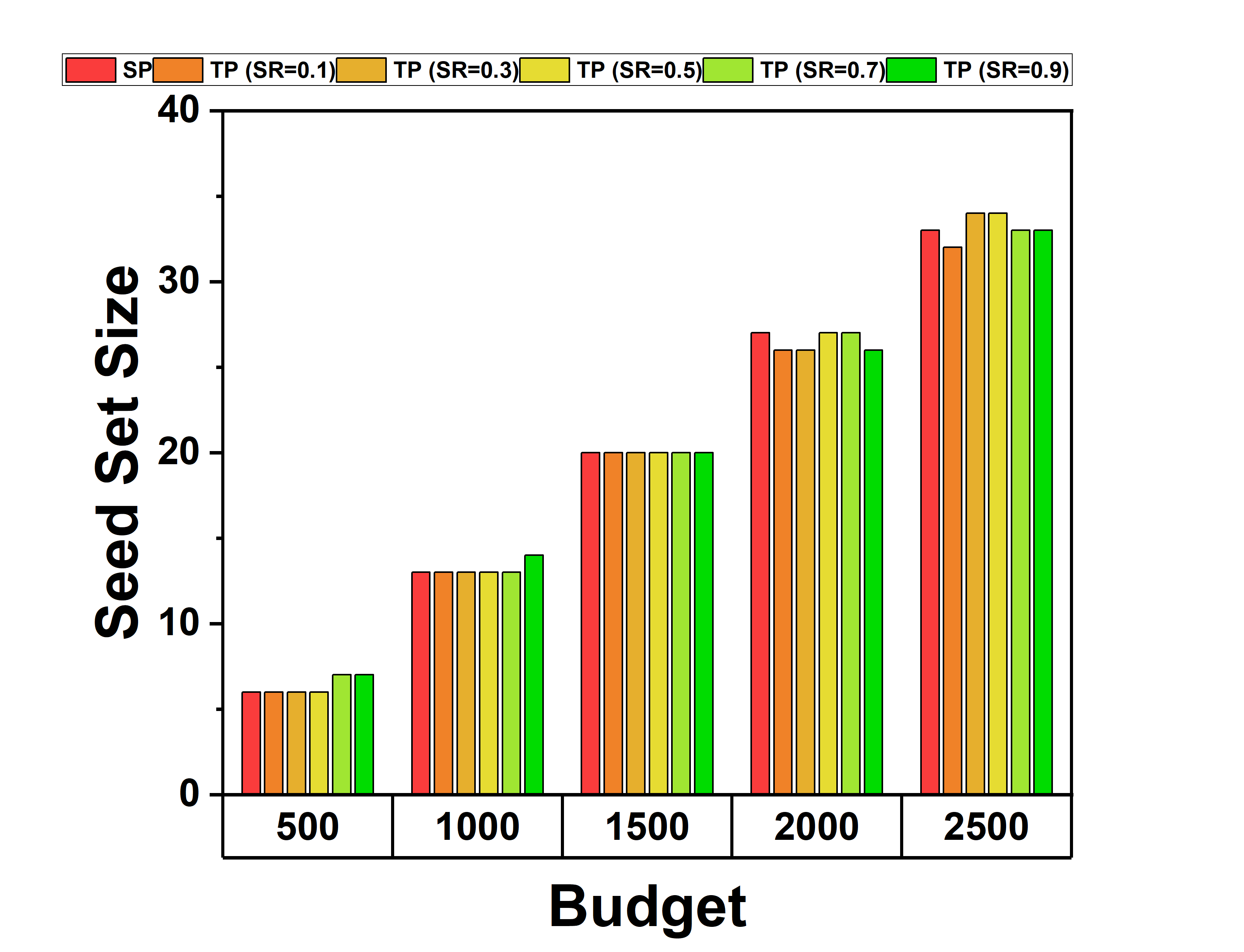}
        \caption{Timestep 2}
    \end{subfigure}
    \hspace{0.05\linewidth}
    \begin{subfigure}[t]{0.3\linewidth}
        \centering
        \includegraphics[width=\linewidth]{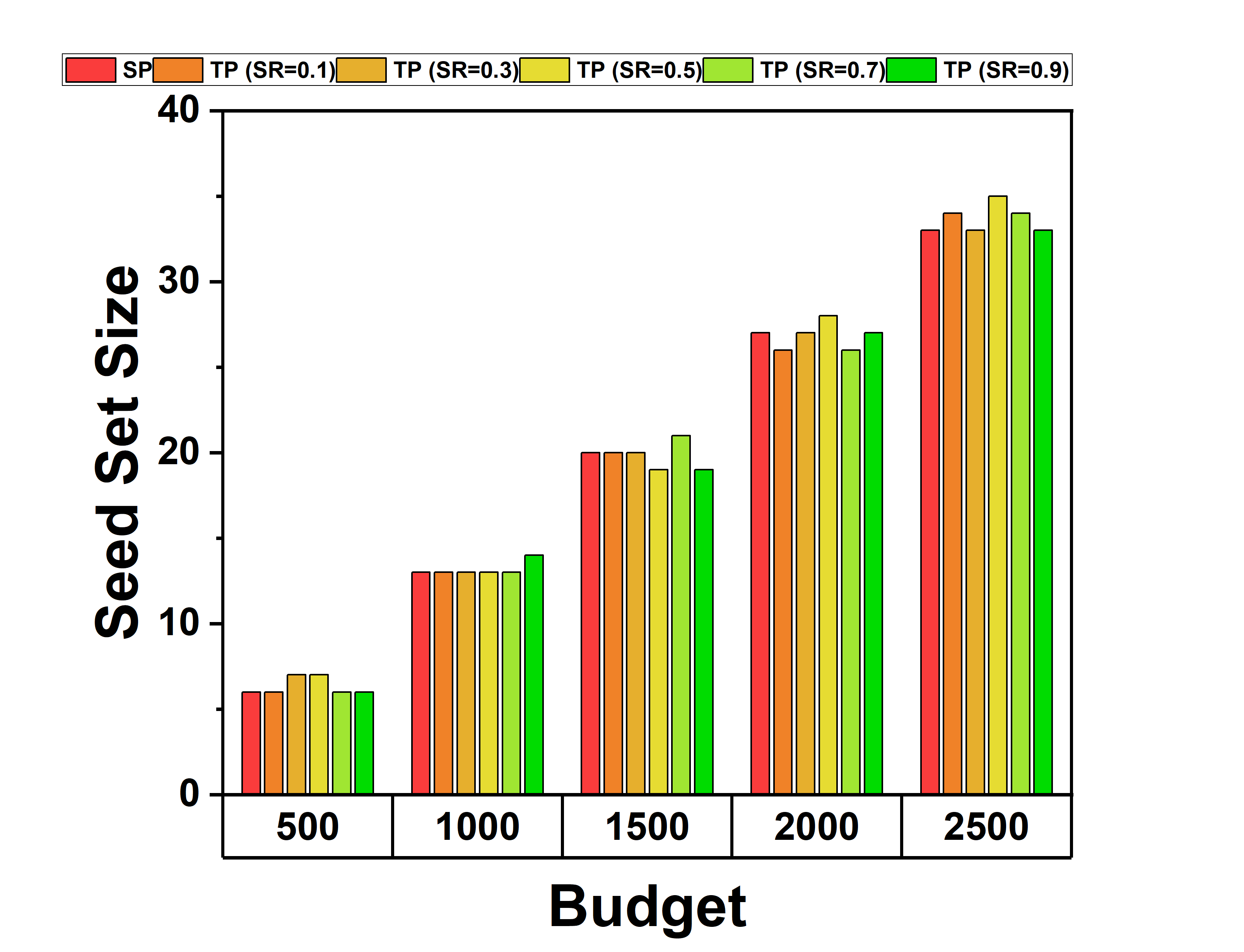}
        \caption{Timestep 4}
    \end{subfigure}

    \vspace{0.5cm}

    \begin{subfigure}[t]{0.3\linewidth}
        \centering
        \includegraphics[width=\linewidth]{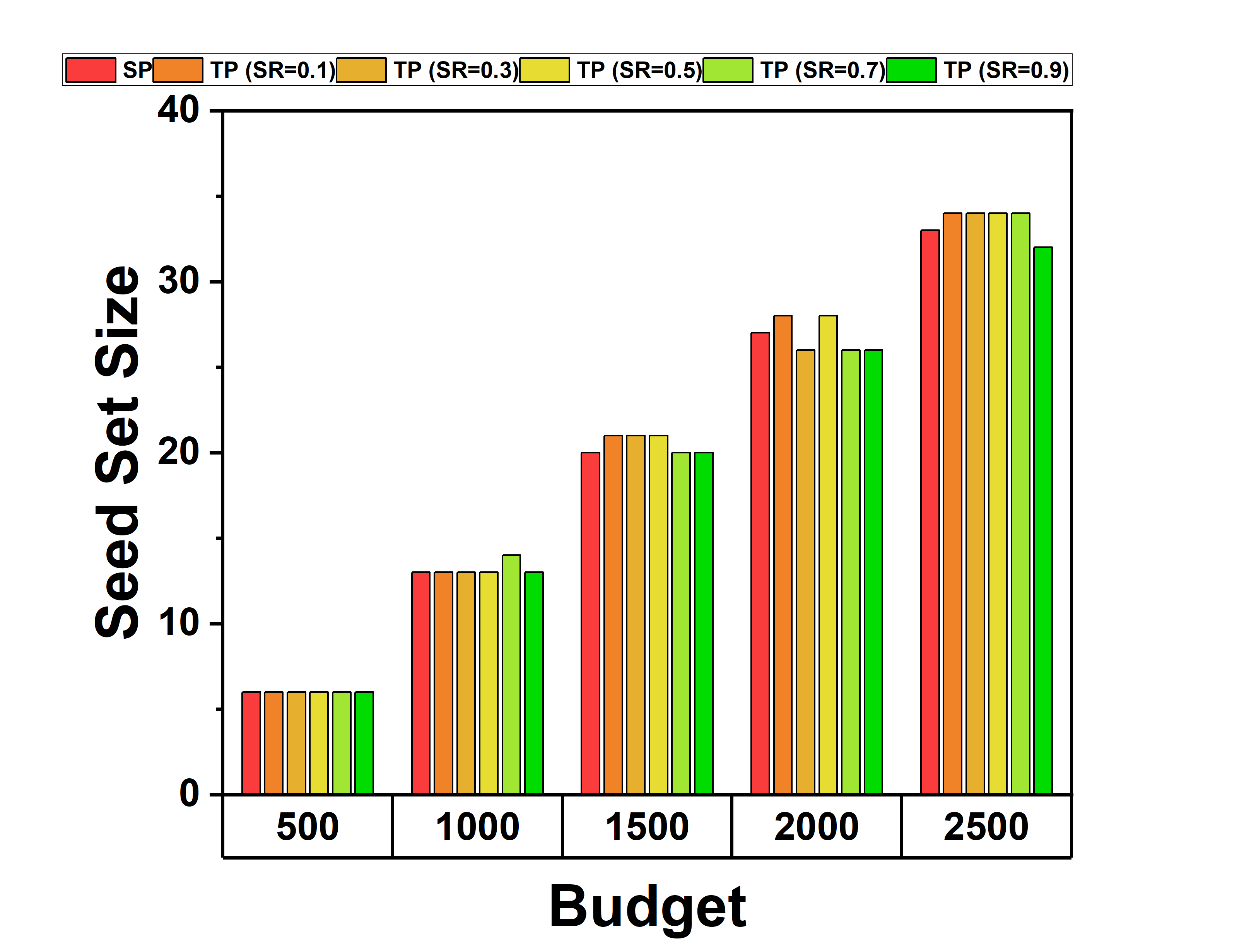}
        \caption{Timestep 6}
    \end{subfigure}
    \hfill
    \begin{subfigure}[t]{0.3\linewidth}
        \centering
        \includegraphics[width=\linewidth]{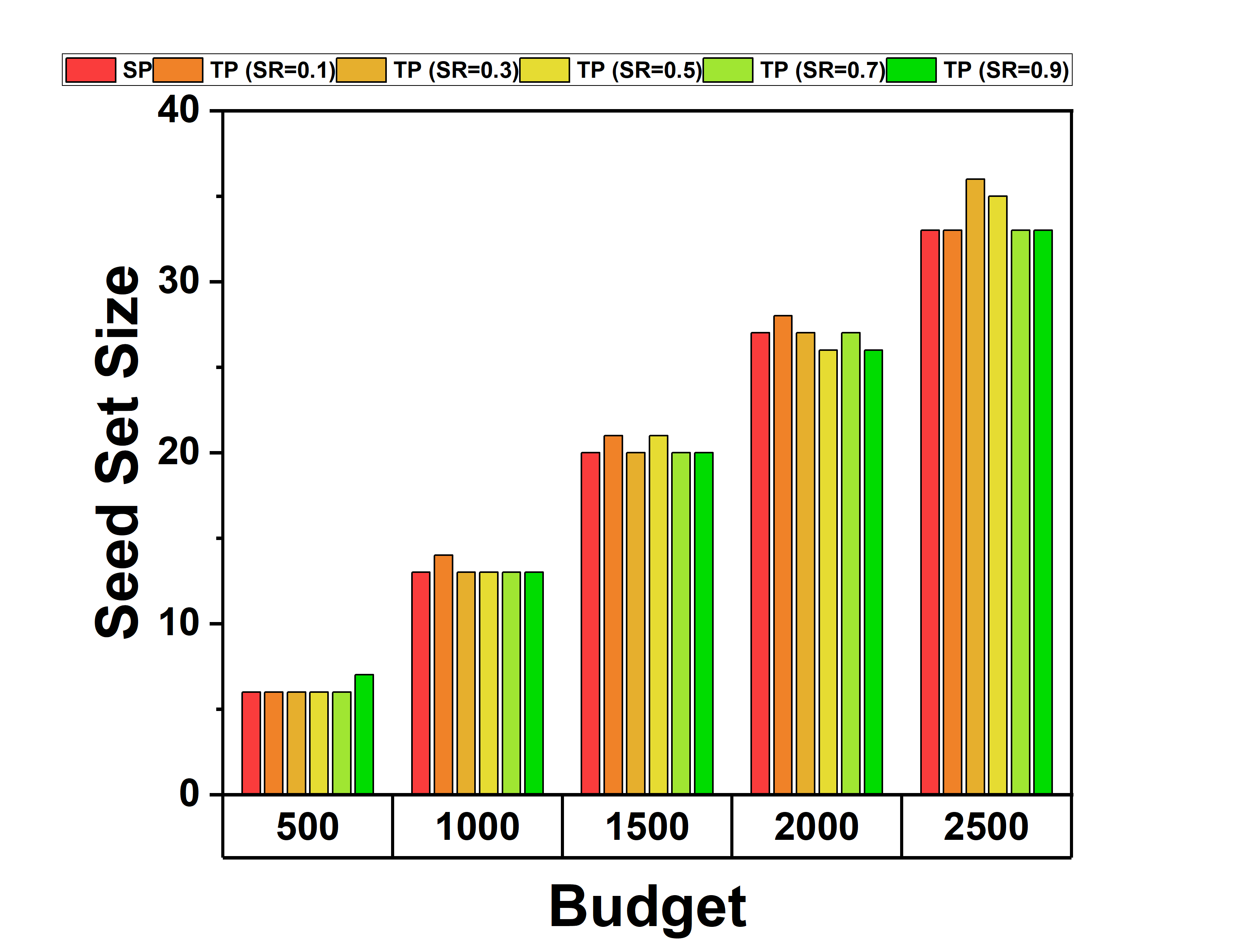}
        \caption{Timestep 8}
    \end{subfigure}
    \hfill
    \begin{subfigure}[t]{0.3\linewidth}
        \centering
        \includegraphics[width=\linewidth]{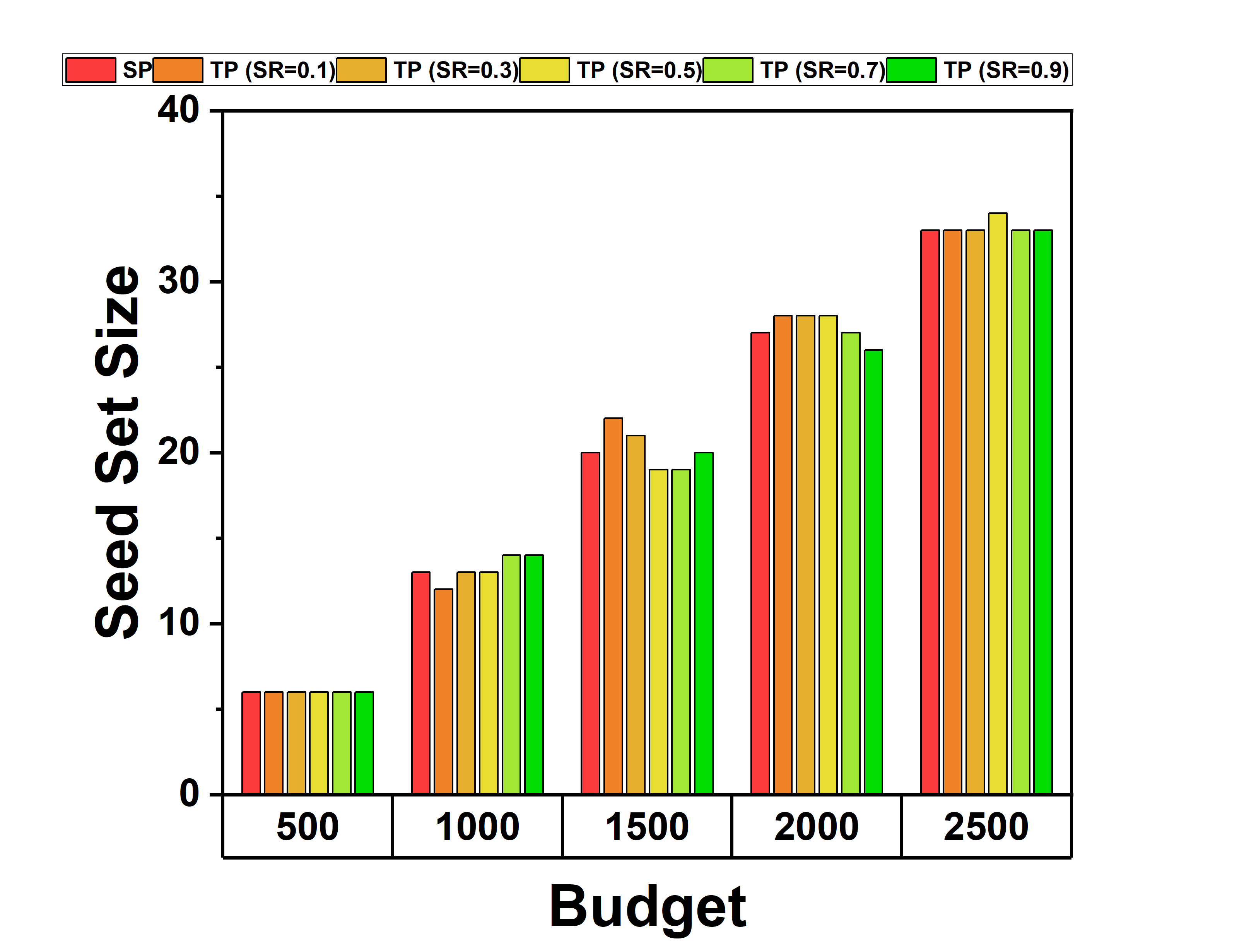}
        \caption{Timestep 10}
    \end{subfigure}

    \caption{Seed Set Size Distribution of Single Phase Vs. Two Phase (Single Discount Algorithm, \textit{Email-Eu-Core} Dataset, Probability Setting - Trivalency)}
    \label{RQ4_T5}
\end{figure}

\begin{figure}[htbp]
    \centering

    \begin{subfigure}[t]{0.3\linewidth}
        \centering
        \includegraphics[width=\linewidth]{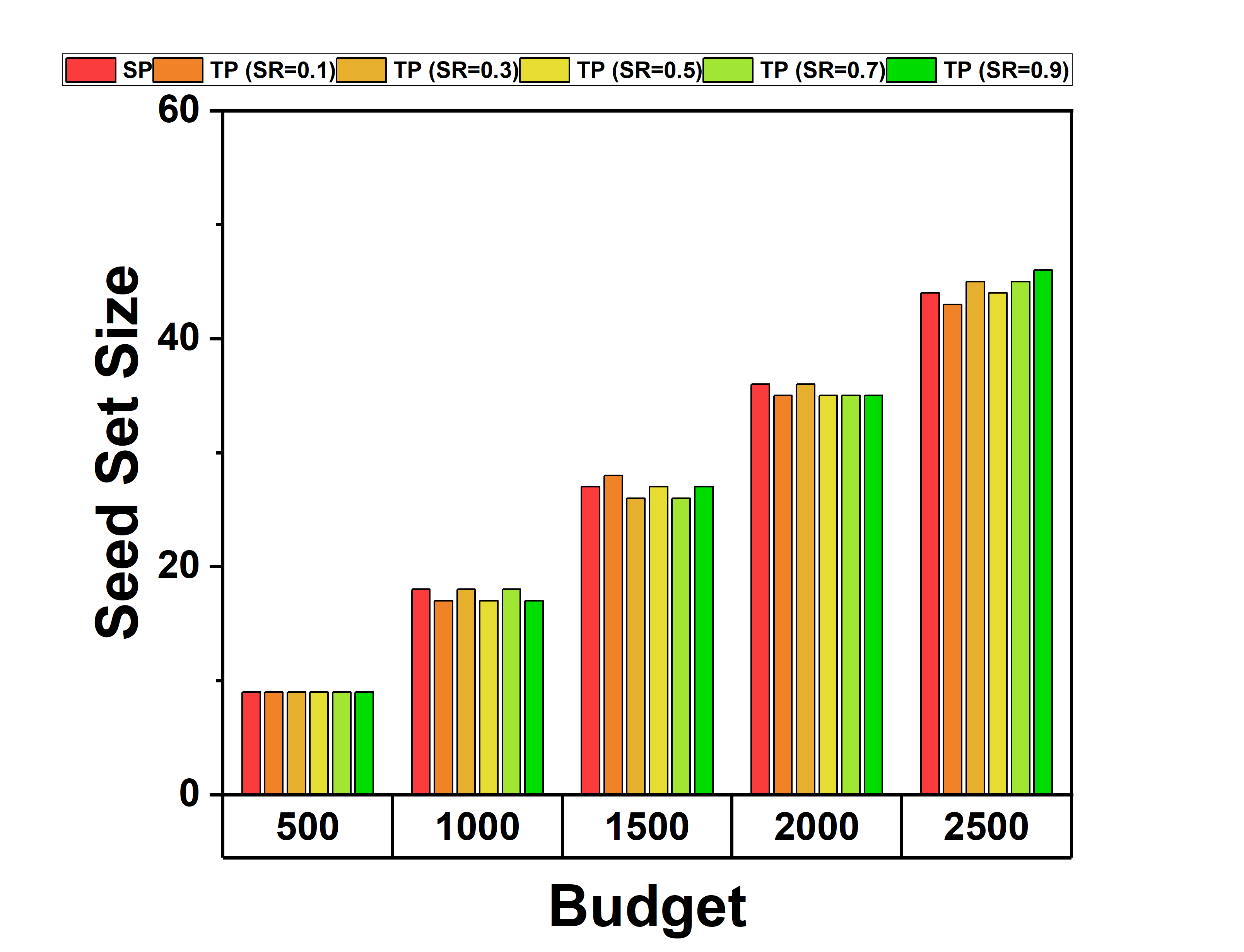}
        \caption{Timestep 2}
    \end{subfigure}
    \hspace{0.05\linewidth}
    \begin{subfigure}[t]{0.3\linewidth}
        \centering
        \includegraphics[width=\linewidth]{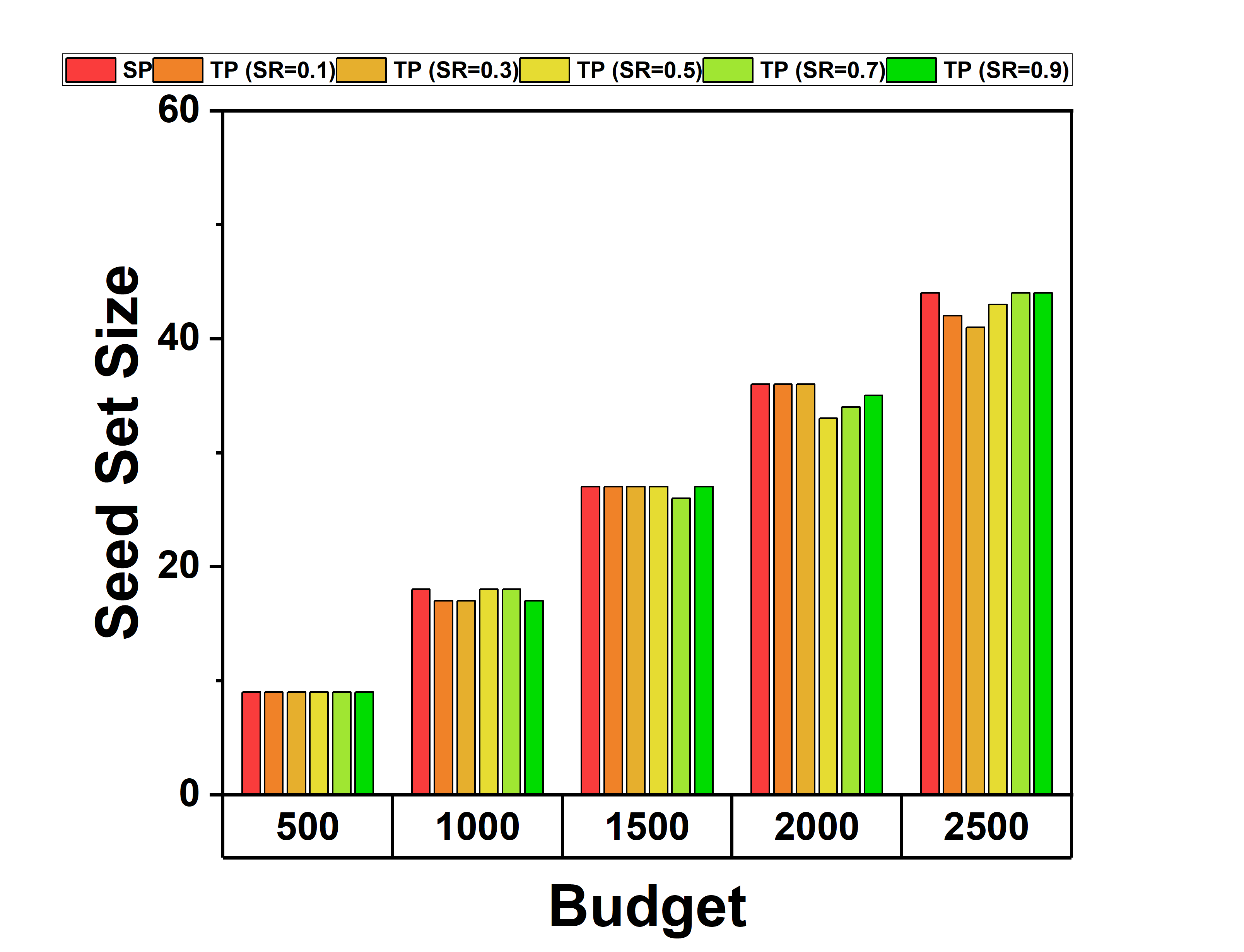}
        \caption{Timestep 4}
    \end{subfigure}

    \vspace{0.5cm}

    \begin{subfigure}[t]{0.3\linewidth}
        \centering
        \includegraphics[width=\linewidth]{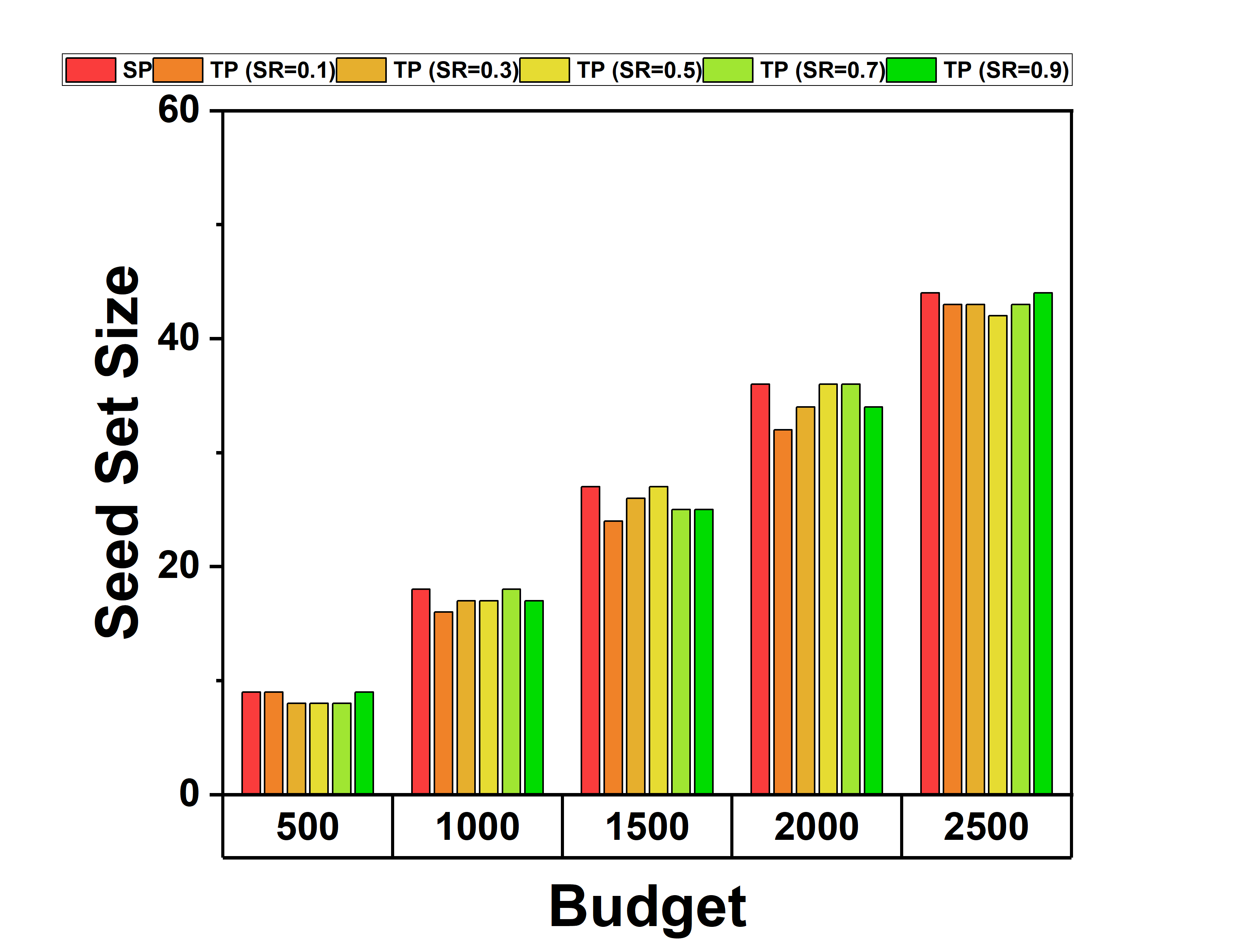}
        \caption{Timestep 6}
    \end{subfigure}
    \hfill
    \begin{subfigure}[t]{0.3\linewidth}
        \centering
        \includegraphics[width=\linewidth]{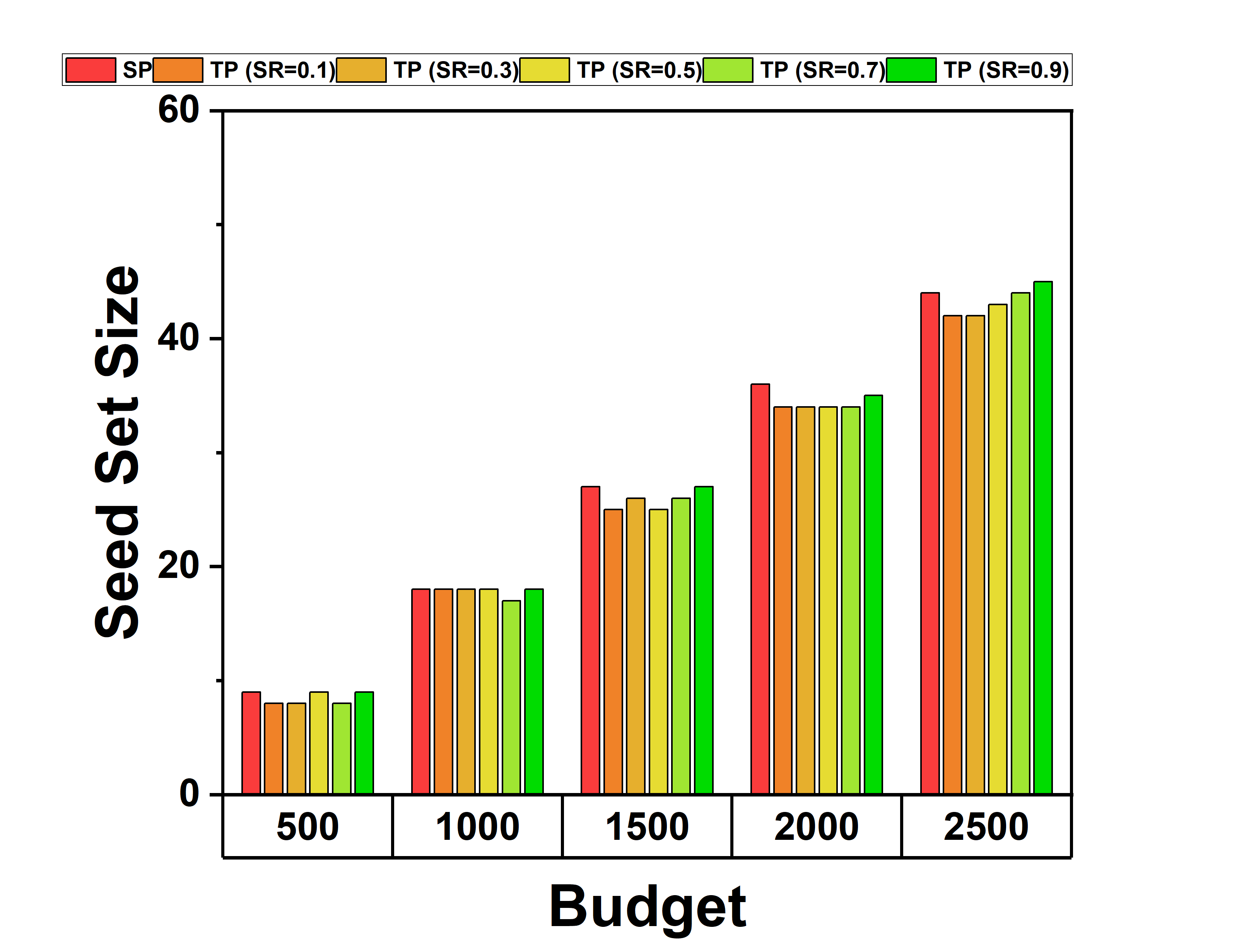}
        \caption{Timestep 8}
    \end{subfigure}
    \hfill
    \begin{subfigure}[t]{0.3\linewidth}
        \centering
        \includegraphics[width=\linewidth]{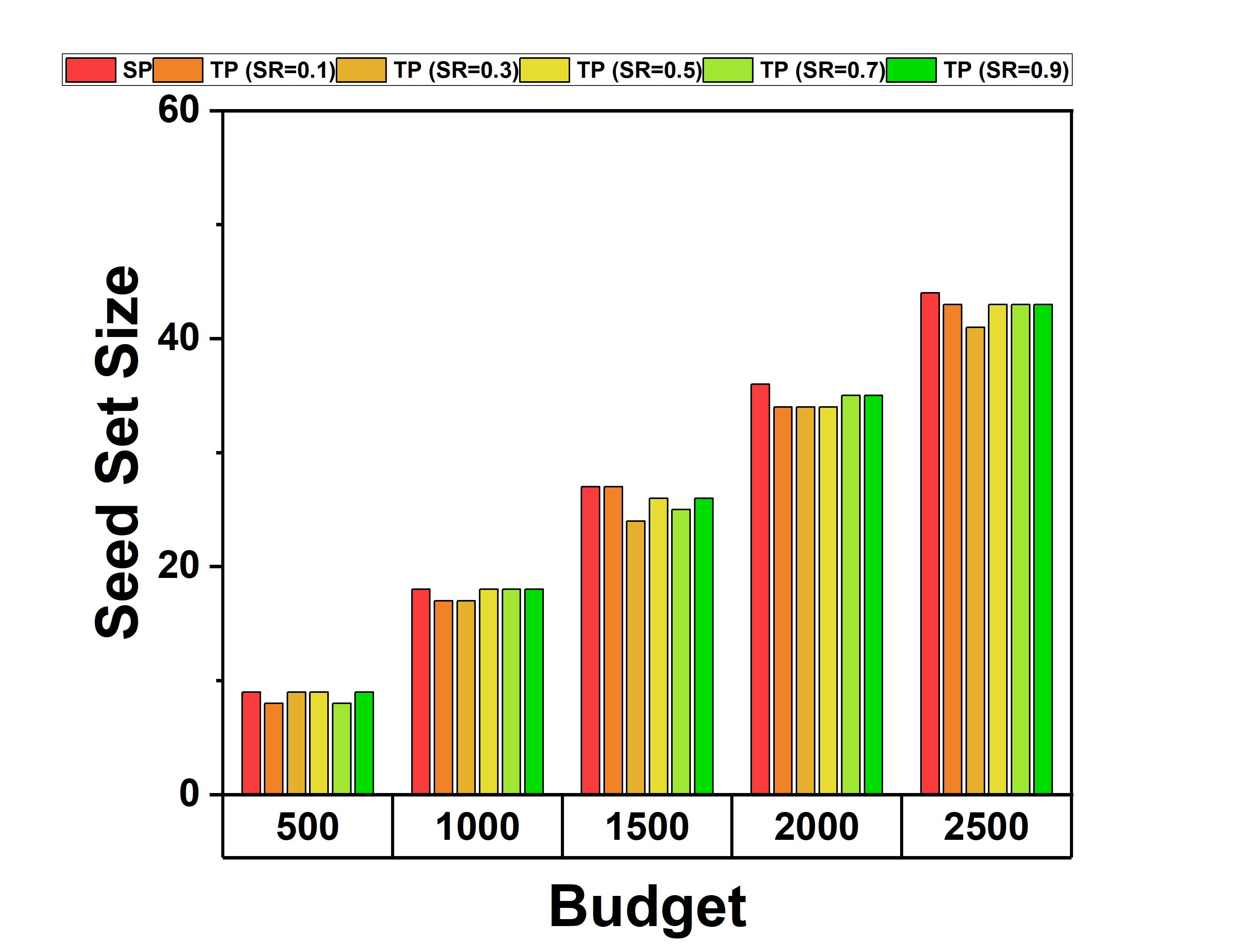}
        \caption{Timestep 10}
    \end{subfigure}

    \caption{Seed Set Size Distribution of Single Phase Vs. Two Phase (Simple Greedy Algorithm, \textit{Email-Eu-Core} Dataset, Probability Setting - Trivalency)}
    \label{RQ4_T6}
\end{figure}

\begin{figure}[htbp]
    \centering

    \begin{subfigure}[t]{0.3\linewidth}
        \centering
        \includegraphics[width=\linewidth]{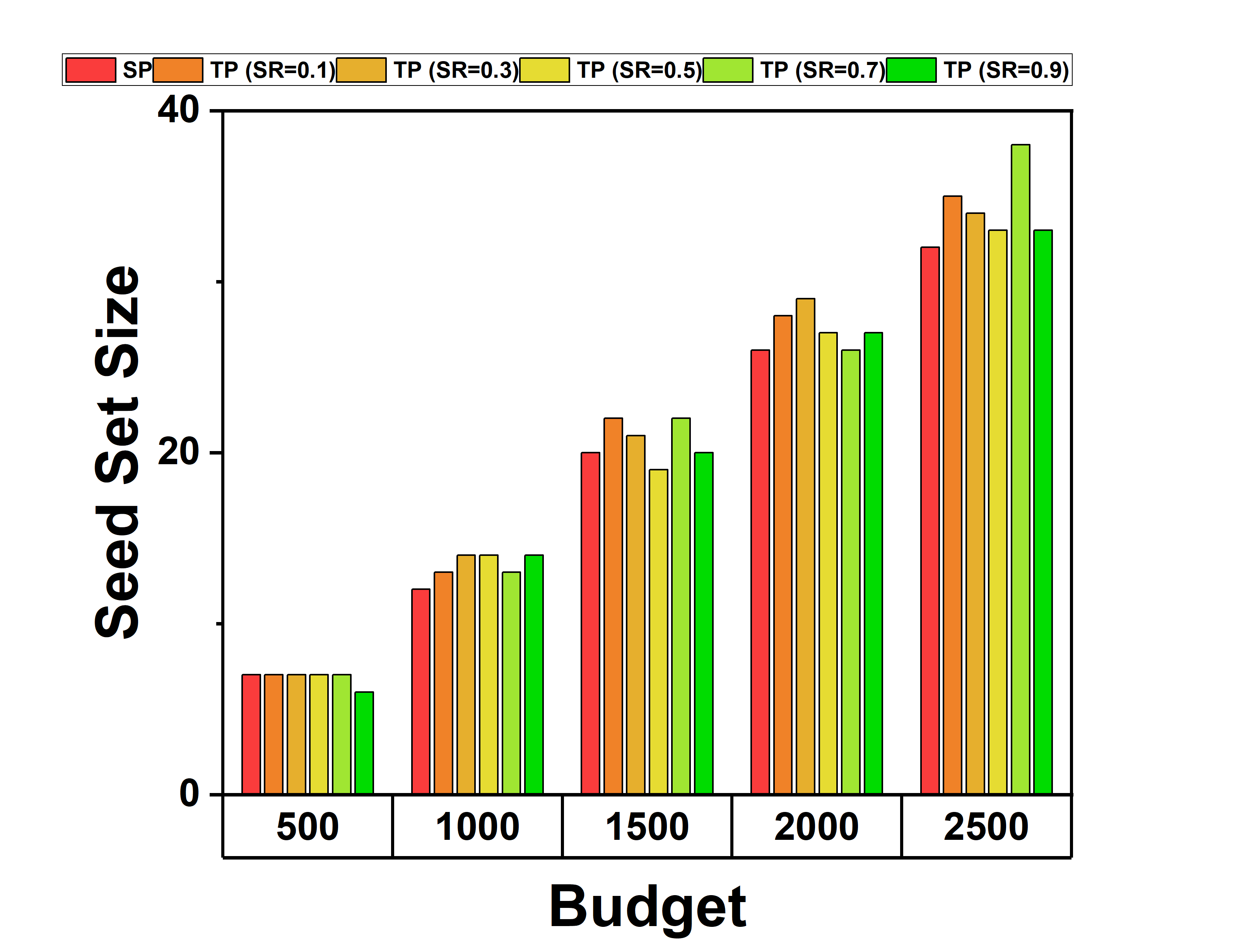}
        \caption{Timestep 2}
    \end{subfigure}
    \hspace{0.05\linewidth}
    \begin{subfigure}[t]{0.3\linewidth}
        \centering
        \includegraphics[width=\linewidth]{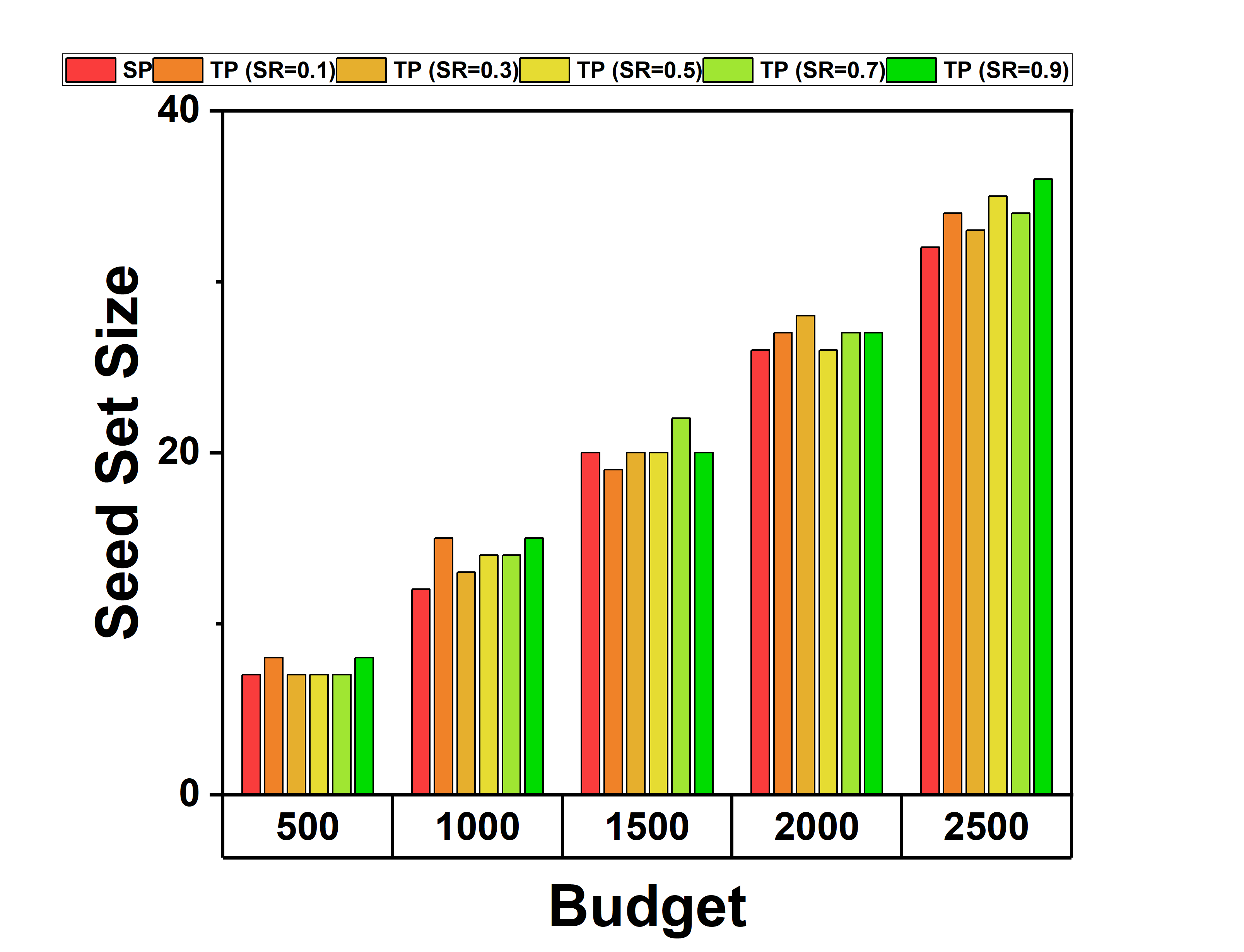}
        \caption{Timestep 4}
    \end{subfigure}

    \vspace{0.5cm}

    \begin{subfigure}[t]{0.3\linewidth}
        \centering
        \includegraphics[width=\linewidth]{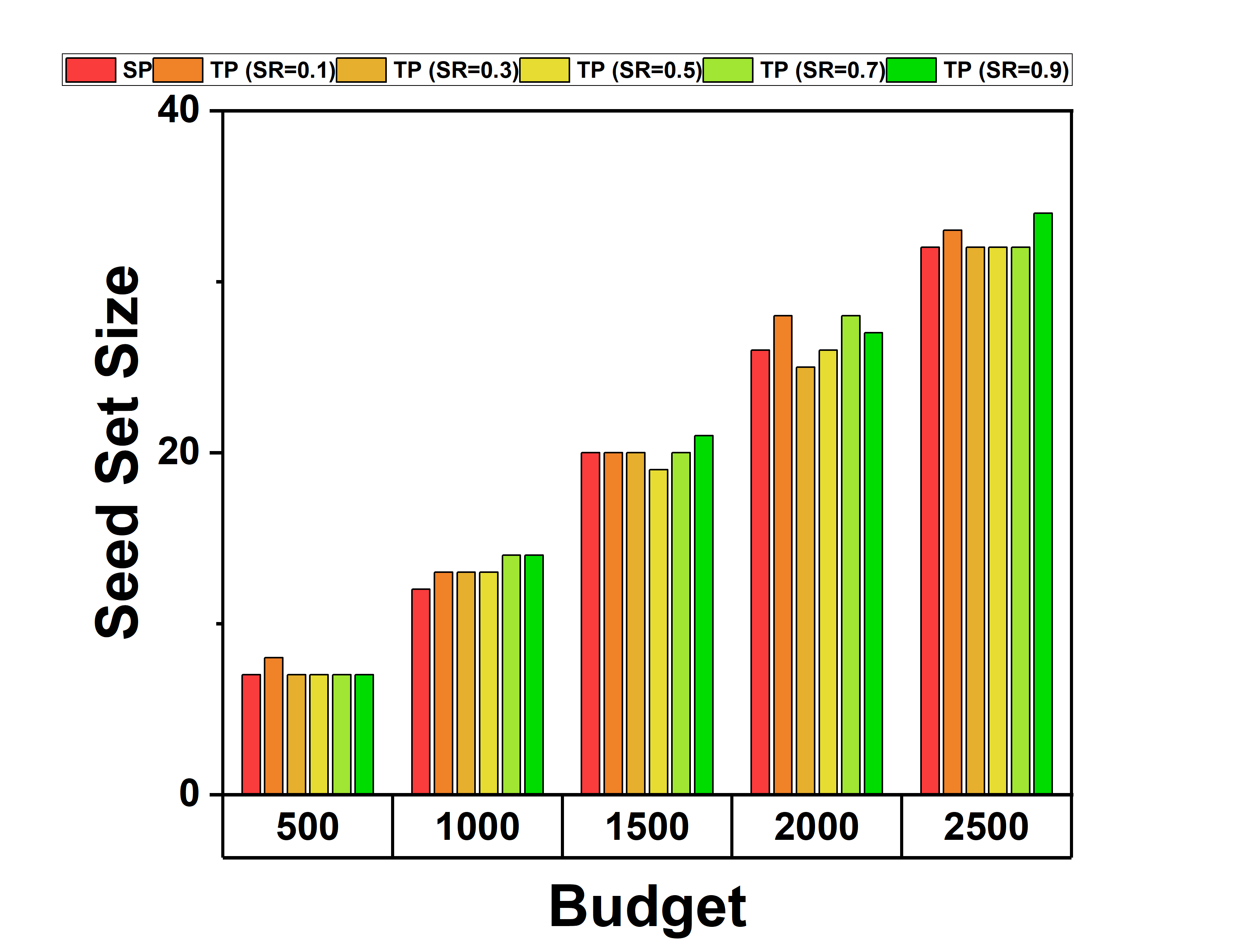}
        \caption{Timestep 6}
    \end{subfigure}
    \hfill
    \begin{subfigure}[t]{0.3\linewidth}
        \centering
        \includegraphics[width=\linewidth]{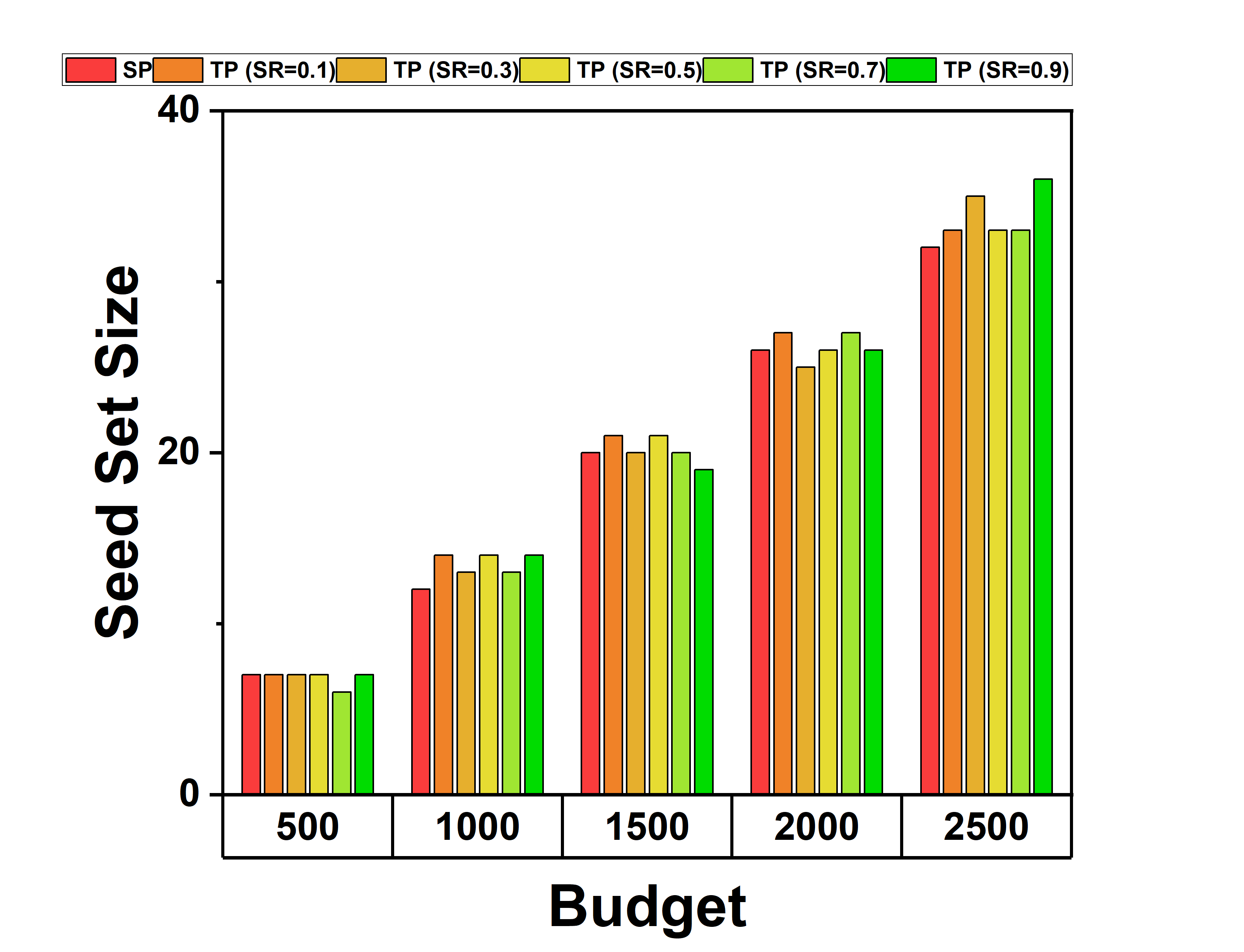}
        \caption{Timestep 8}
    \end{subfigure}
    \hfill
    \begin{subfigure}[t]{0.3\linewidth}
        \centering
        \includegraphics[width=\linewidth]{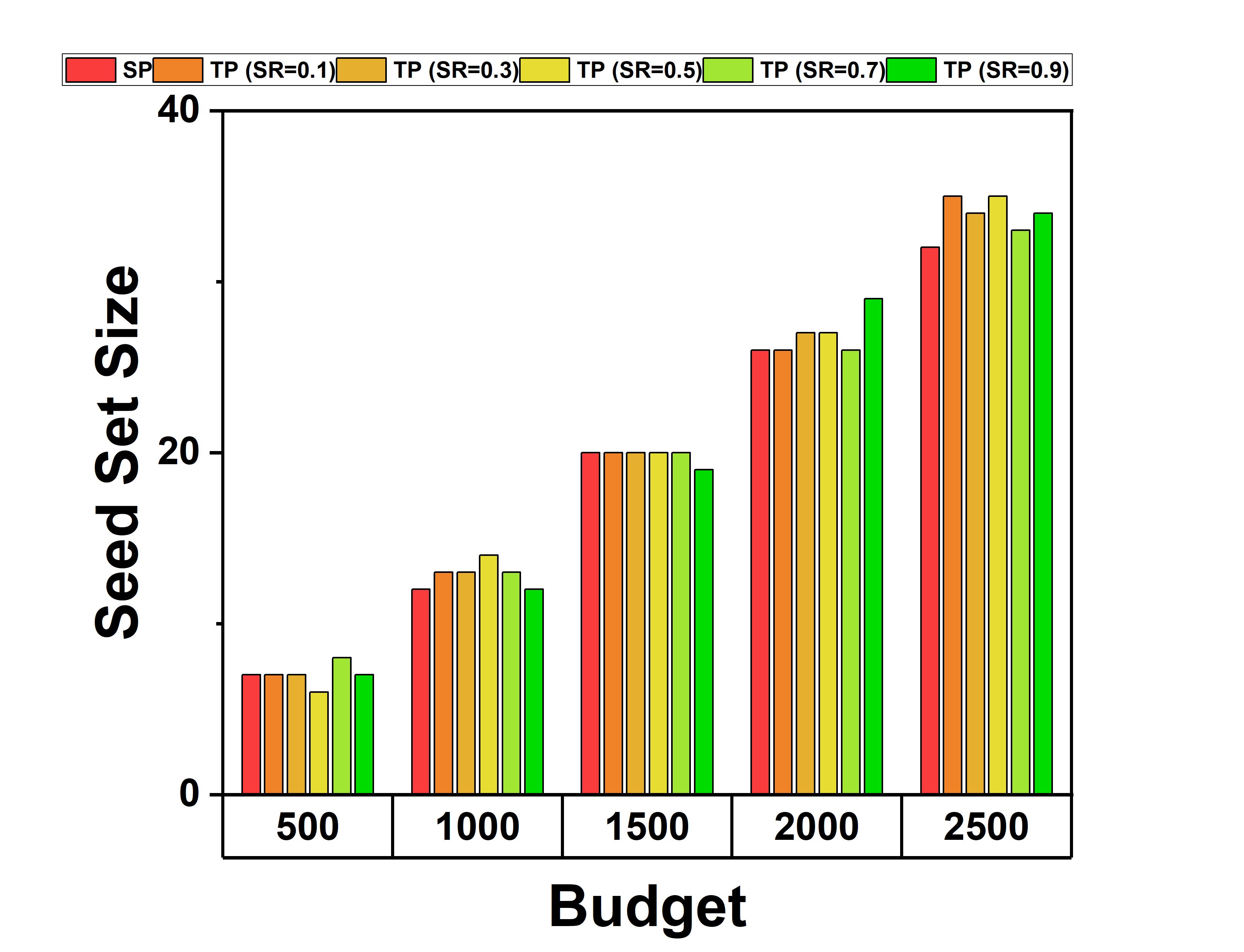}
        \caption{Timestep 10}
    \end{subfigure}

    \caption{Seed Set Size Distribution of Single Phase Vs. Two Phase (Double Greedy Algorithm, \textit{Email-Eu-Core} Dataset, Probability Setting - Trivalency)}
    \label{RQ4_T7}
\end{figure}

\begin{figure}[htbp]
    \centering

    \begin{subfigure}[t]{0.3\linewidth}
        \centering
        \includegraphics[width=\linewidth]{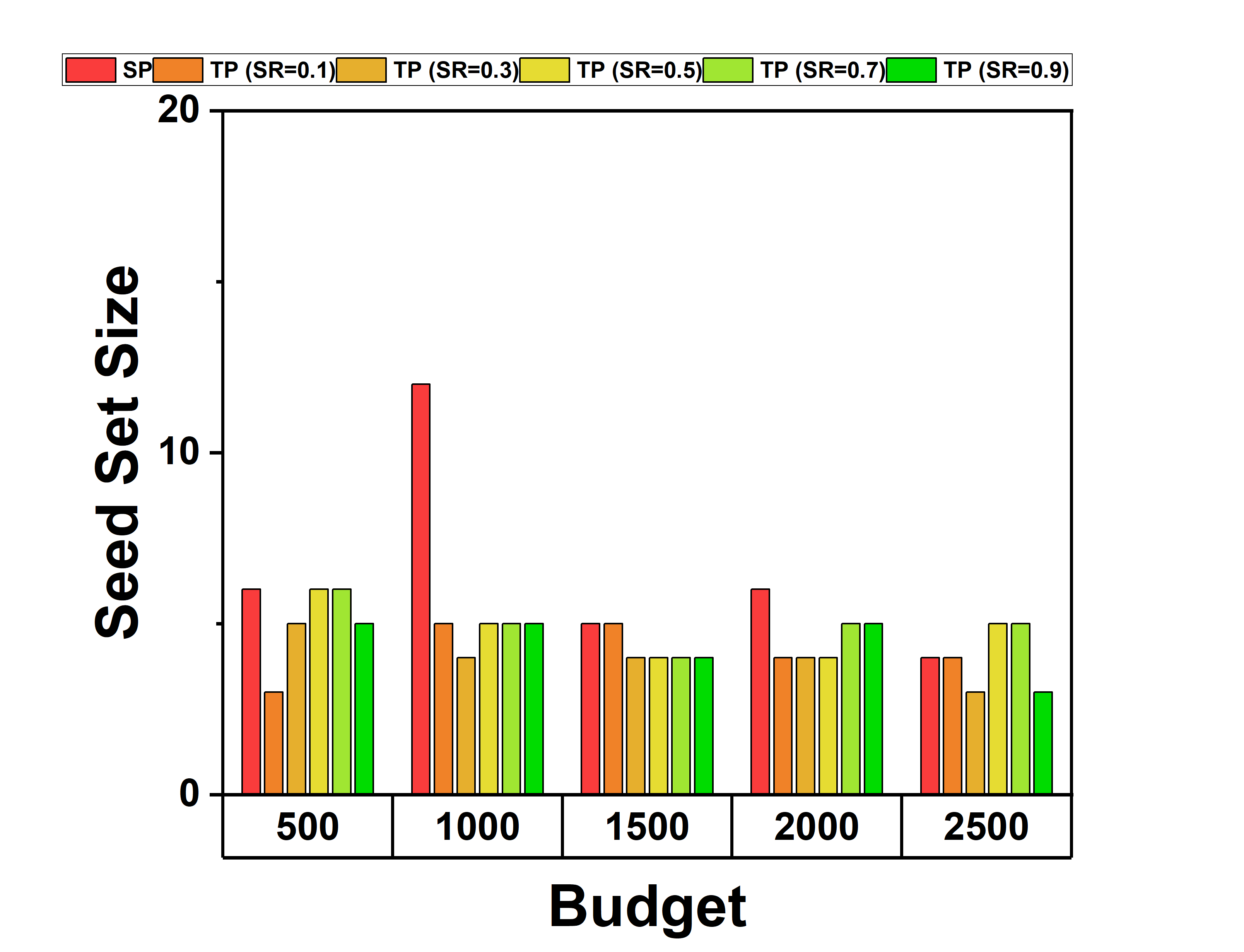}
        \caption{Timestep 2}
    \end{subfigure}
    \hspace{0.05\linewidth}
    \begin{subfigure}[t]{0.3\linewidth}
        \centering
        \includegraphics[width=\linewidth]{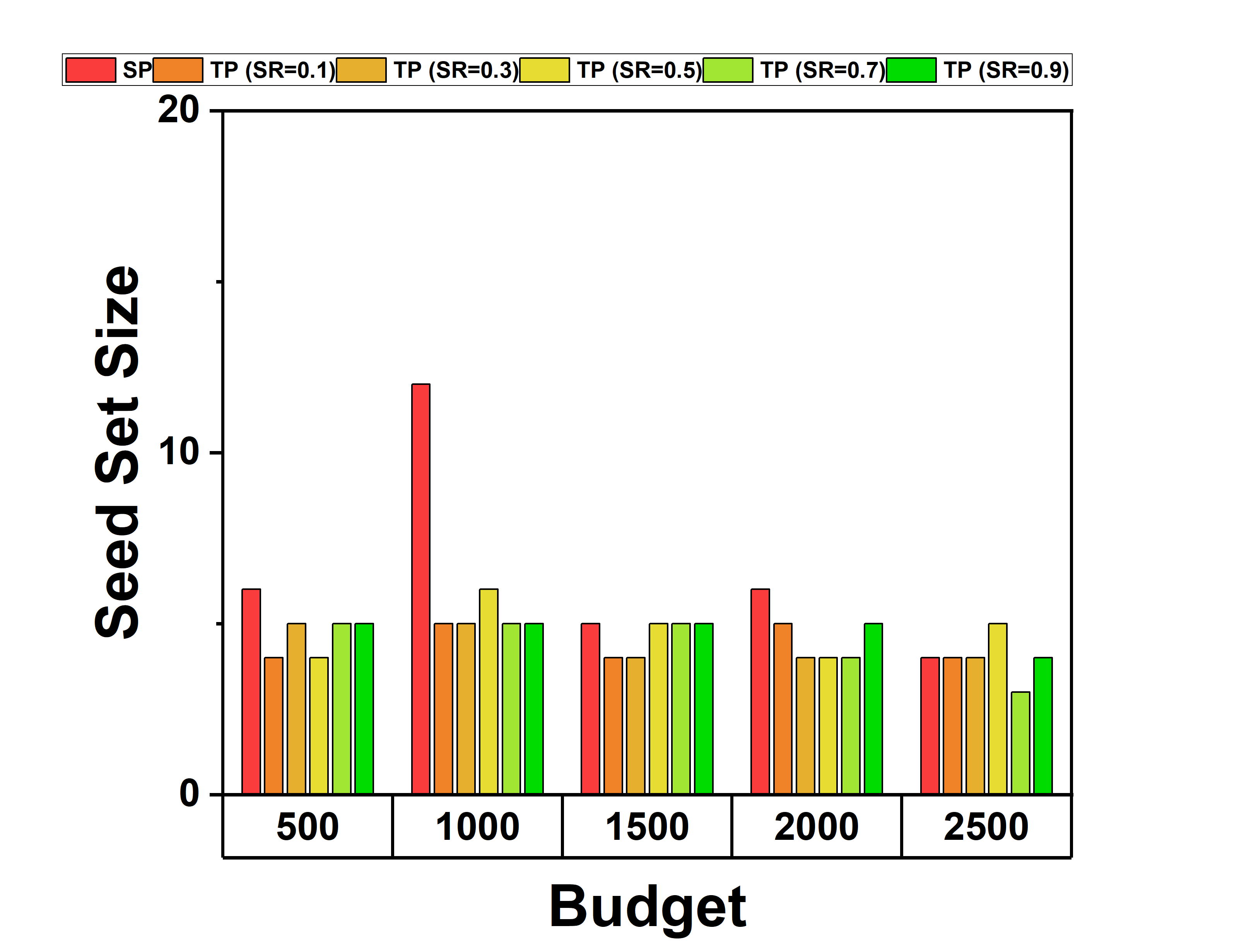}
        \caption{Timestep 4}
    \end{subfigure}

    \vspace{0.5cm}

    \begin{subfigure}[t]{0.3\linewidth}
        \centering
        \includegraphics[width=\linewidth]{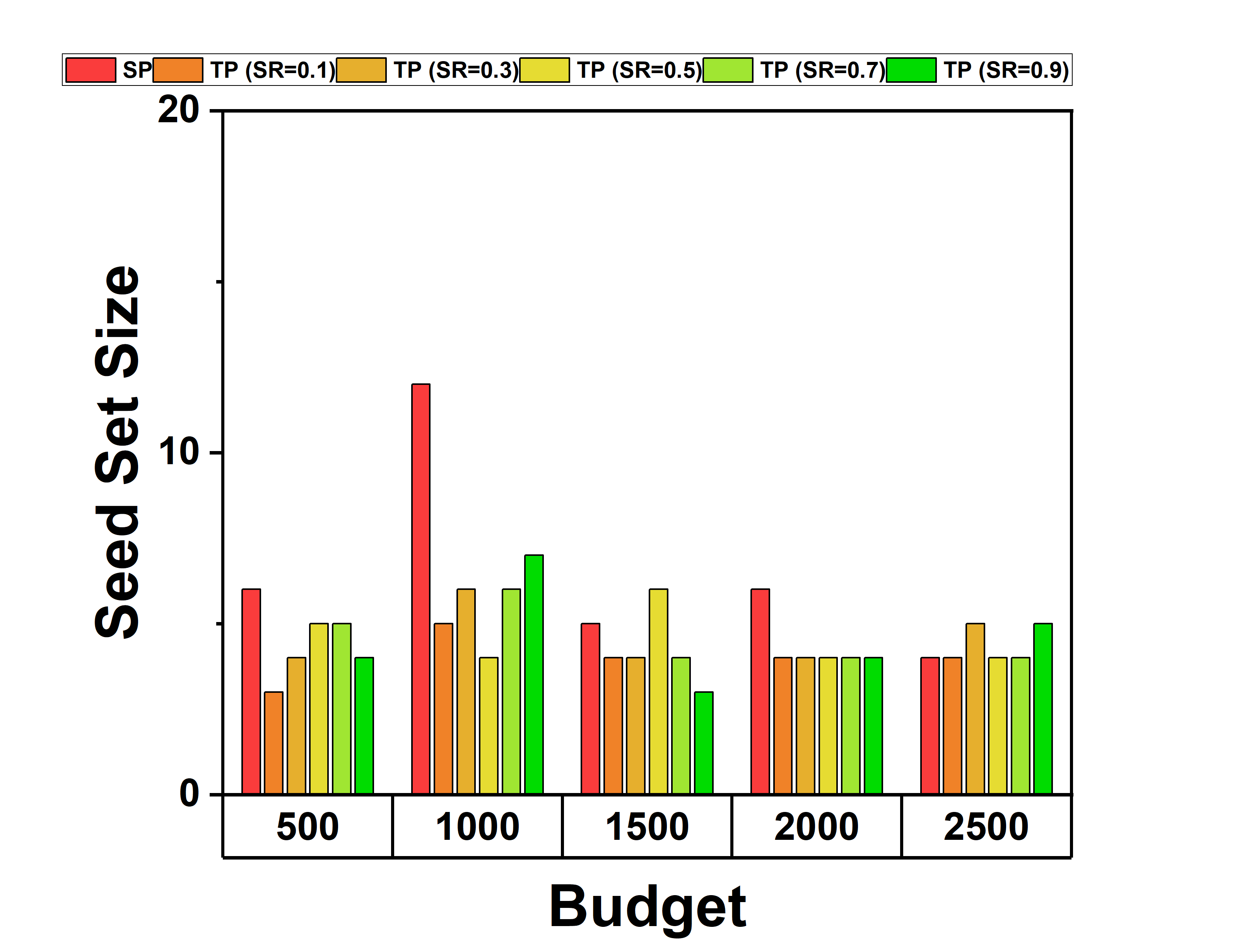}
        \caption{Timestep 6}
    \end{subfigure}
    \hfill
    \begin{subfigure}[t]{0.3\linewidth}
        \centering
        \includegraphics[width=\linewidth]{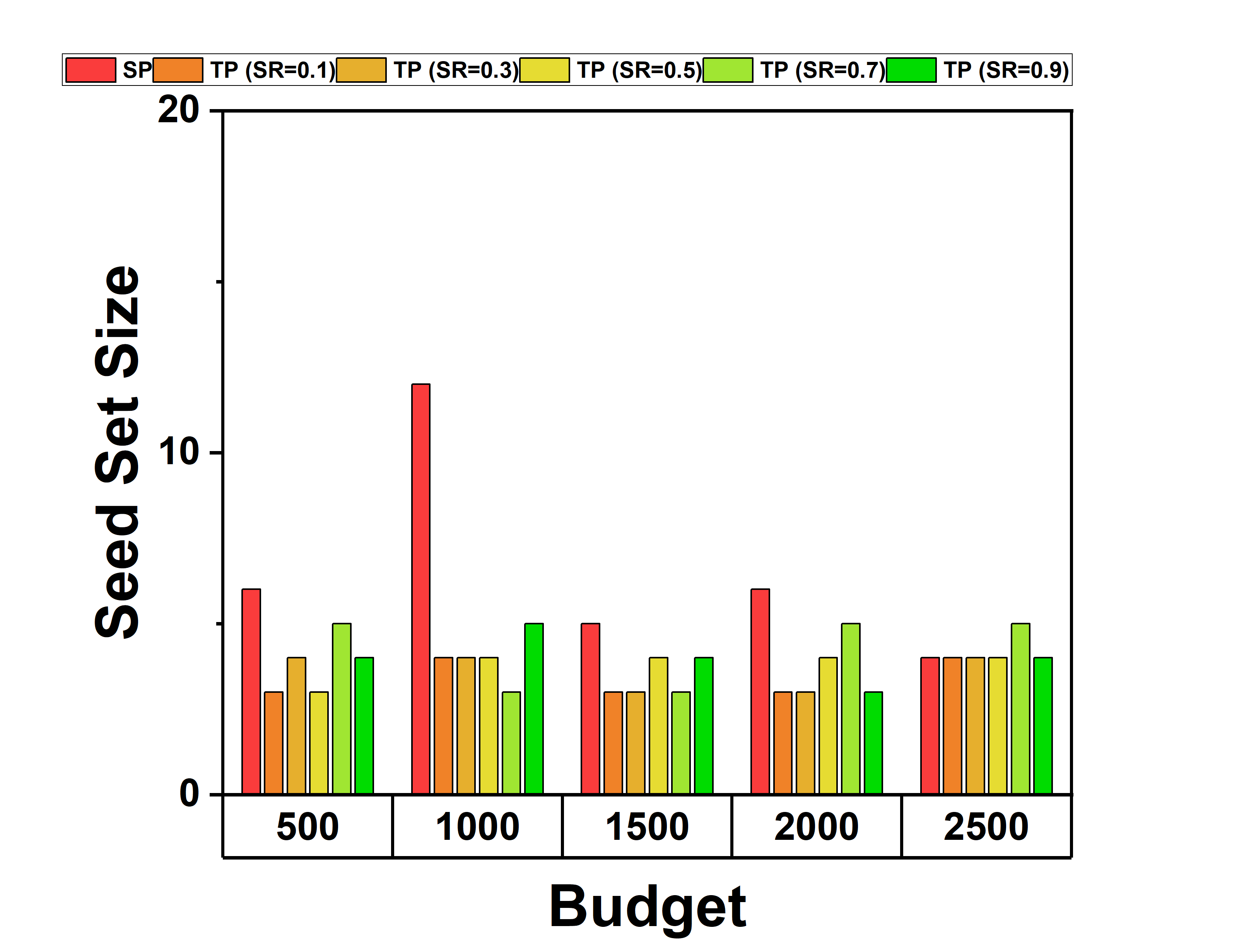}
        \caption{Timestep 8}
    \end{subfigure}
    \hfill
    \begin{subfigure}[t]{0.3\linewidth}
        \centering
        \includegraphics[width=\linewidth]{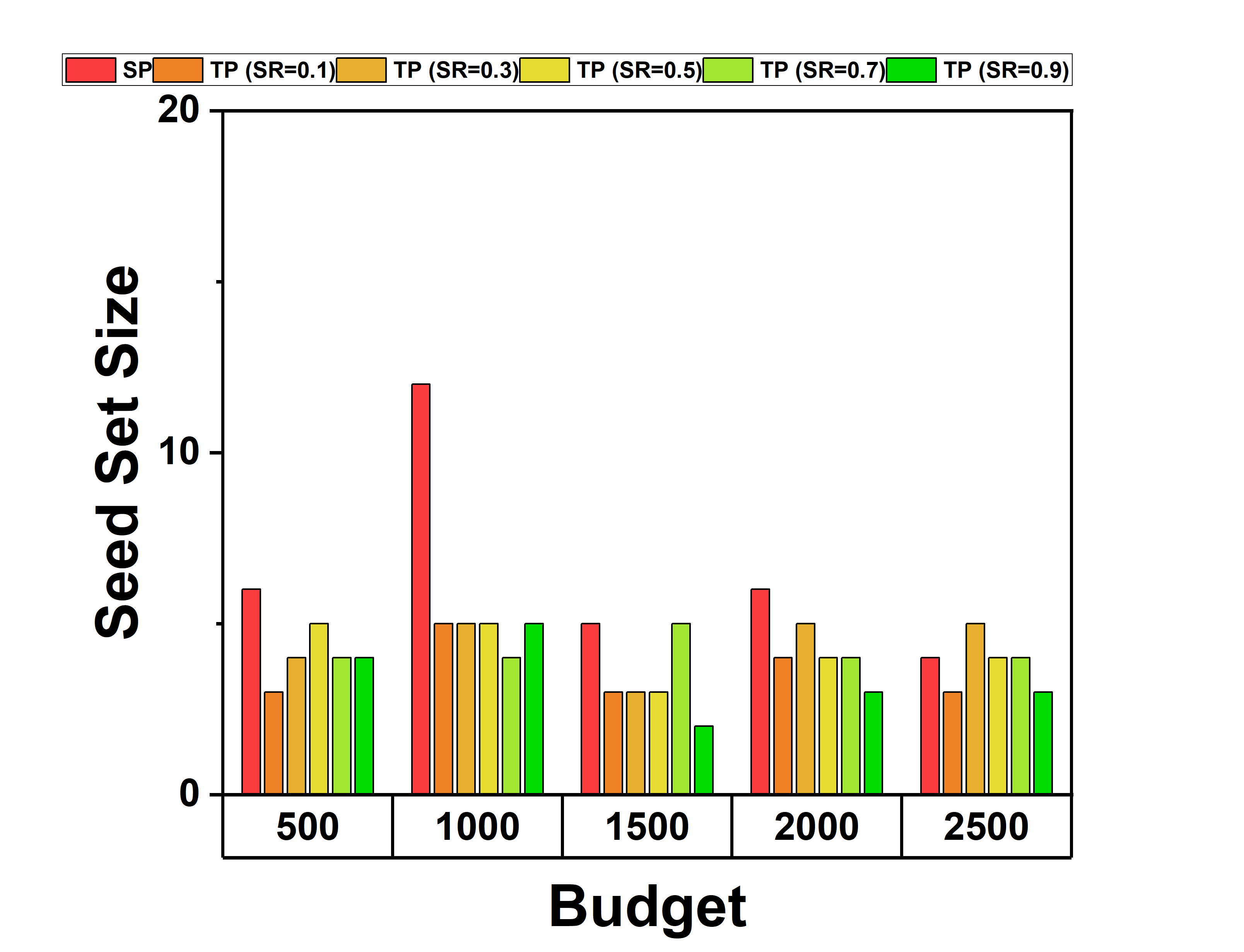}
        \caption{Timestep 10}
    \end{subfigure}

    \caption{Seed Set Size Distribution of Single Phase Vs. Two Phase (Stochastic Greedy Algorithm, \textit{Email-Eu-Core} Dataset, Probability Setting - Trivalency)}
    \label{RQ4_T8}
\end{figure}

\subsubsection{Analysis of Diffusion Rounds in Single vs. Two phase Influence Propagation}
The analysis for the \textit{LM} dataset under the trivalency probability setting explores how well different algorithms convert additional diffusion rounds into profit. There’s a noticeable difference between baseline and proposed methods in terms of round-wise efficiency. Among the baselines, both \textbf{Random} and \textbf{HD} show good returns for extra rounds. For example, at budget $500$, split ratio $0.1$, and timestep $2$, \textbf{Random} gained $2295.73$ from $2$ extra rounds (Figure~\ref{RQ5LM_T1}(a)), while \textbf{HD} gained $1228.26$ in the same setting (Figure~\ref{RQ5LM_T2}(a)). This shows that allowing more rounds can be valuable for these algorithms. On the other hand, some baselines performed poorly. \textbf{HighCC} gave only $23.78$ extra profit from $1$ additional round in the same configuration (Figure~\ref{RQ5LM_T3}(a)). \textbf{DD} achieved a gain of $152.53$ from $9$ extra rounds at budget $500$, split ratio $0.3$, and timestep $10$ (Figure~\ref{RQ5LM_T4}(e)), which is quite low. \textbf{SD} showed a modest gain of $216.74$ for $1$ extra round at budget $500$, split ratio $0.1$, and timestep $2$ (Figure~\ref{RQ5LM_T5}(a)). In contrast, our proposed algorithms converted extra rounds into significantly higher profit. \textbf{DG} was particularly efficient. In one case (budget $1500$, split ratio $0.3$, timestep $8$), \textbf{DG} gained $8203.02$ from $8$ extra rounds (Figure~\ref{RQ5LM_T7}(d)). Another example at budget $2000$, split ratio $0.9$, and timestep $4$ shows a profit increase of $7688.63$ from just $3$ additional rounds (Figure~\ref{RQ5LM_T7}(b)). \textbf{SG} also performed well. At budget $1000$, split ratio $0.9$, and timestep $10$, \textbf{SG} gained $4544.14$ from $9$ extra rounds (Figure~\ref{RQ5LM_T6}(e)). In another case (budget $500$, split ratio $0.9$, timestep $4$), $3$ additional rounds led to $4175.68$ in profit (Figure~\ref{RQ5LM_T6}(b)). \textbf{StG0.1} followed a similar pattern. At budget $1000$, split ratio $0.7$, and timestep $4$, it gained $4402.66$ from $3$ additional rounds (Figure~\ref{RQ5LM_T8}(b)). In conclusion, the efficiency of extra diffusion rounds depends greatly on the algorithm used. While some baseline methods (like \textbf{Random} and \textbf{HD}) show decent returns, others (like \textbf{HighCC} and \textbf{DD}) perform poorly. On the other hand, our proposed methods—\textbf{DG}, \textbf{SG}, and \textbf{StG0.1}—consistently turn additional rounds into significant profit, proving the value of allowing deeper propagation.

\begin{figure}[htbp]
    \centering
    \begin{subfigure}[t]{0.3\linewidth}
        \centering
        \includegraphics[width=\linewidth]{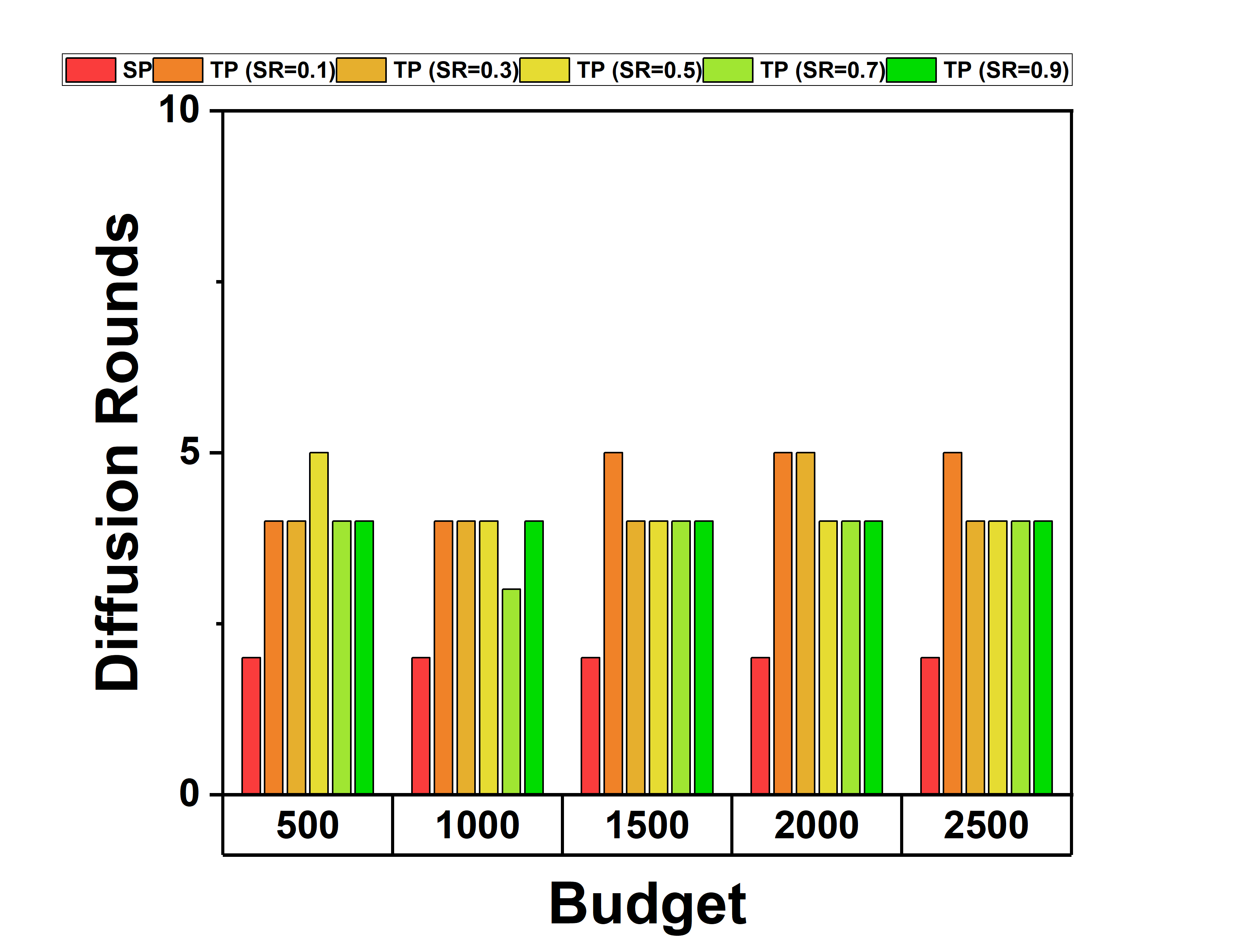}
        \caption{Timestep 2}
    \end{subfigure}
    \hspace{0.05\linewidth}
    \begin{subfigure}[t]{0.3\linewidth}
        \centering
        \includegraphics[width=\linewidth]{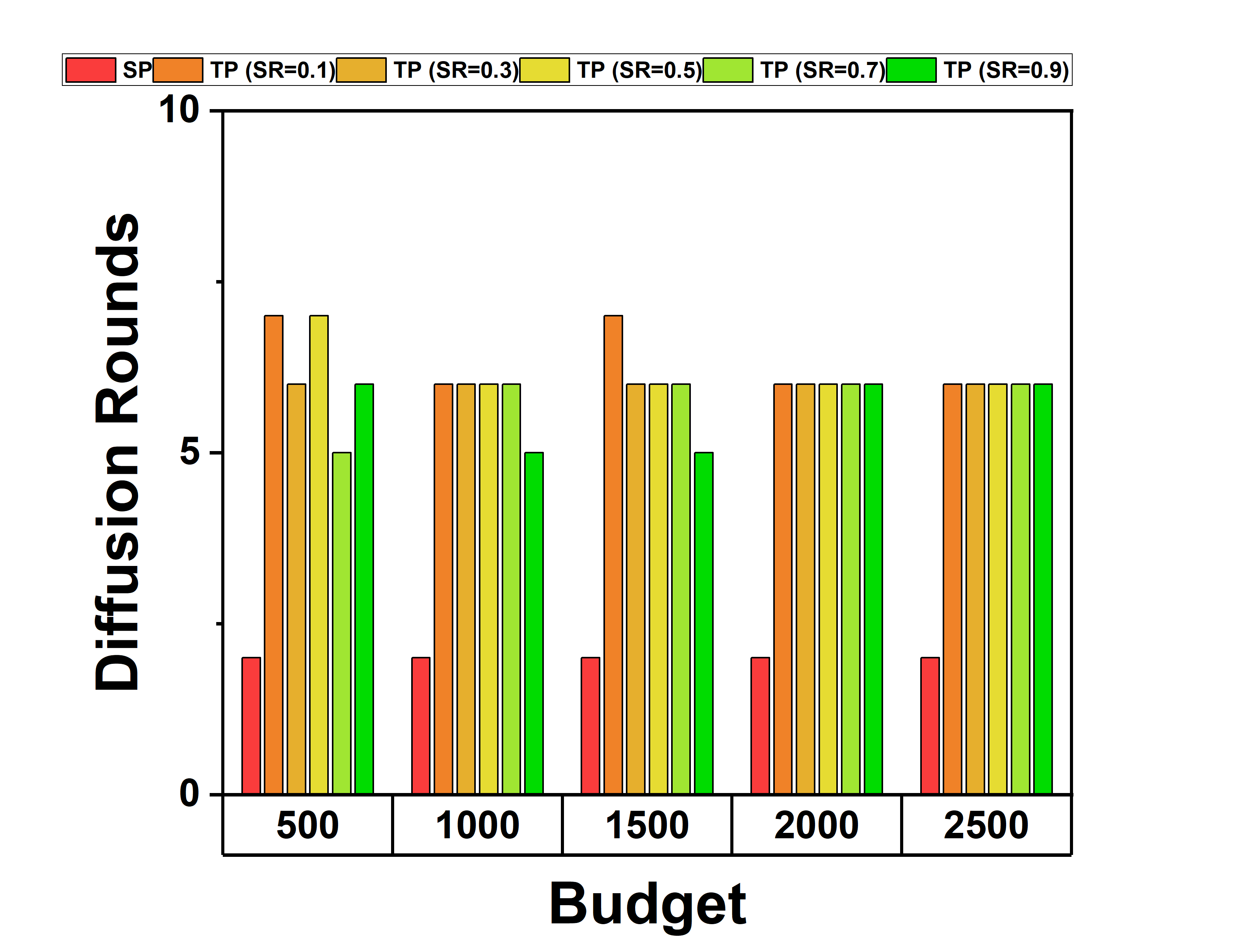}
        \caption{Timestep 4}
    \end{subfigure}

    \vspace{0.5cm}

    \begin{subfigure}[t]{0.3\linewidth}
        \centering
        \includegraphics[width=\linewidth]{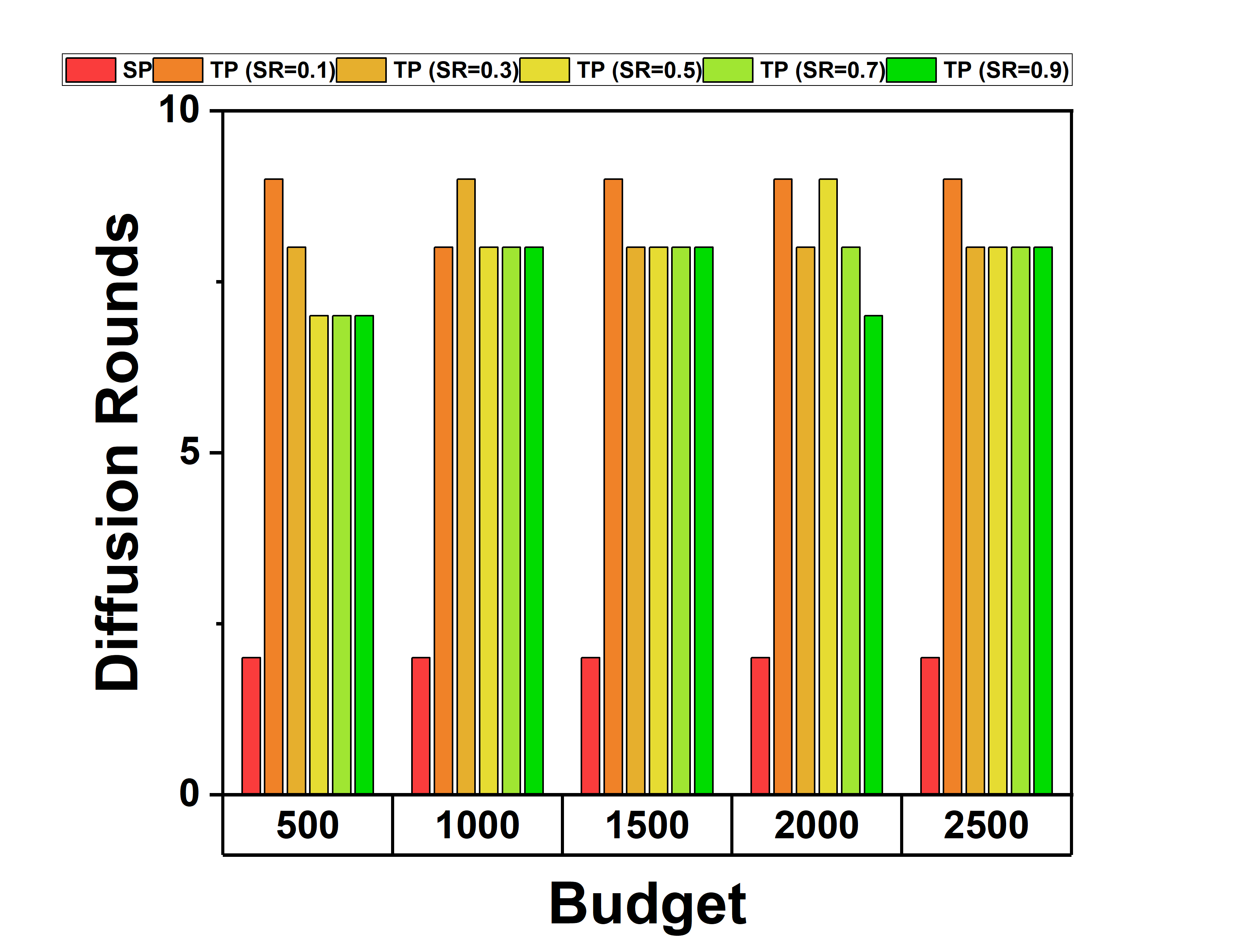}
        \caption{Timestep 6}
    \end{subfigure}
    \hfill
    \begin{subfigure}[t]{0.3\linewidth}
        \centering
        \includegraphics[width=\linewidth]{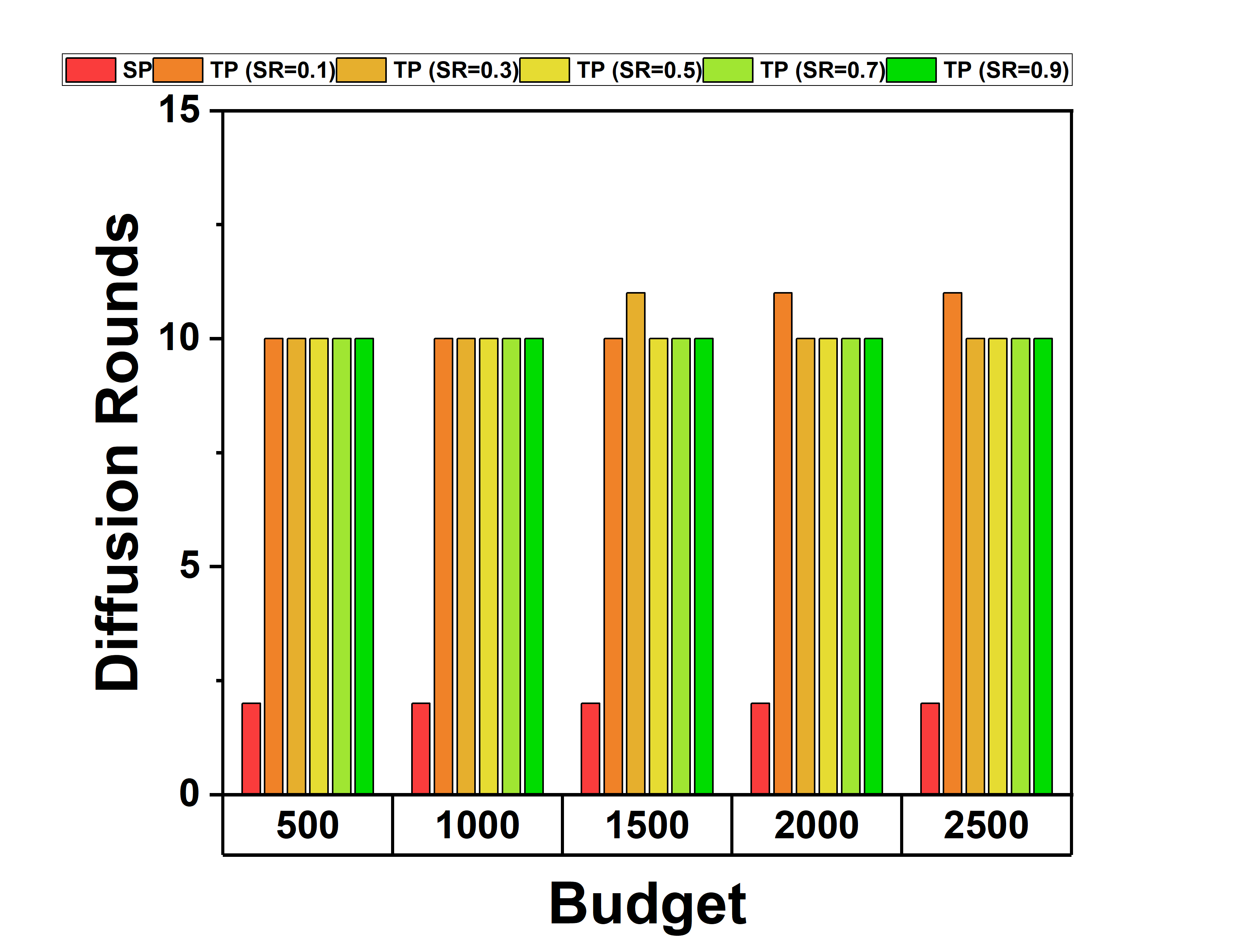}
        \caption{Timestep 8}
    \end{subfigure}
    \hfill
    \begin{subfigure}[t]{0.3\linewidth}
        \centering
        \includegraphics[width=\linewidth]{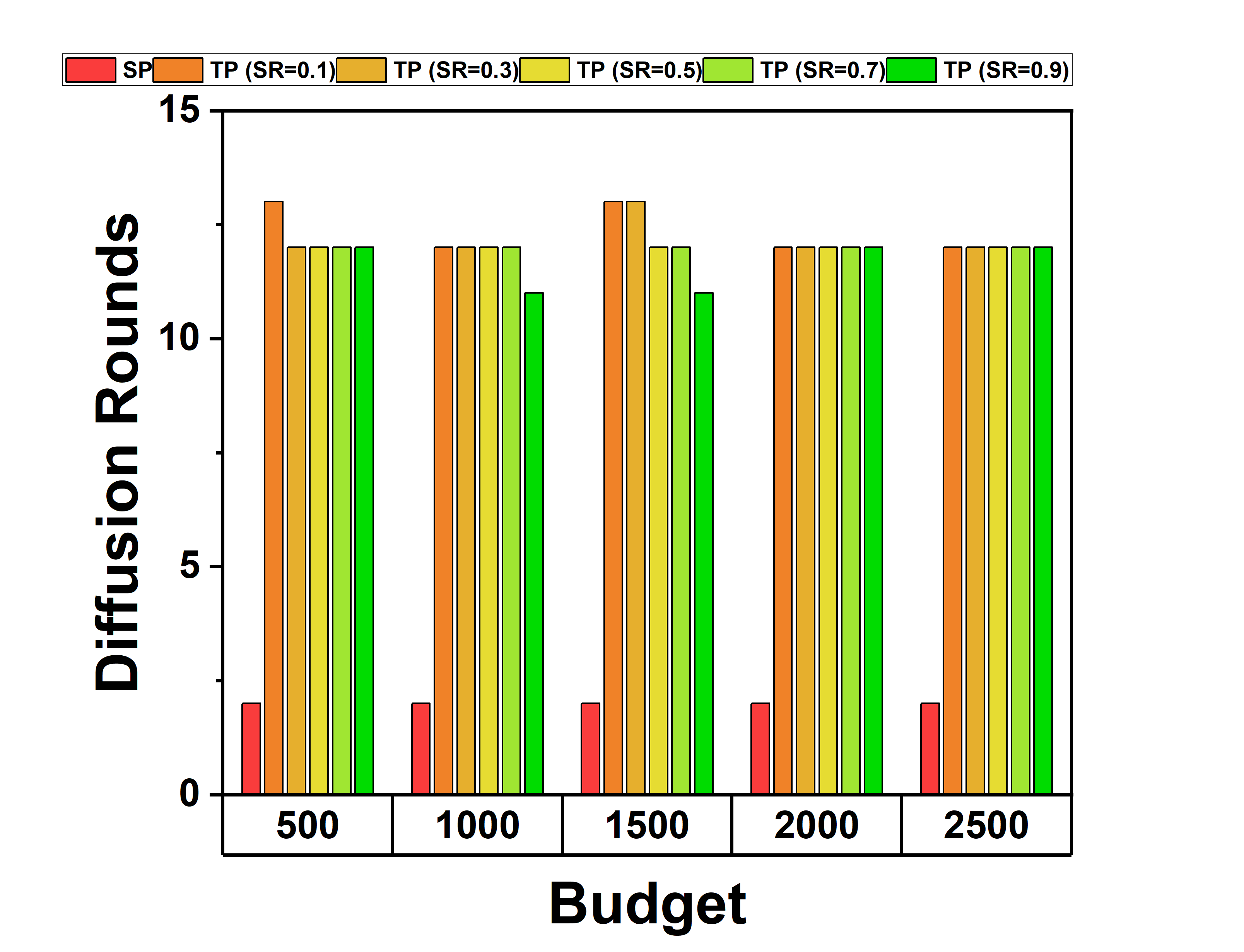}
        \caption{Timestep 10}
    \end{subfigure}

    \caption{Diffusion Rounds in Single Phase Vs. Two Phase (Random Algorithm, \textit{LM} Dataset, Probability Setting - Trivalency)}
    \label{RQ5LM_T1}
\end{figure}

\begin{figure}[htbp]
    \centering
    \begin{subfigure}[t]{0.3\linewidth}
        \centering
        \includegraphics[width=\linewidth]{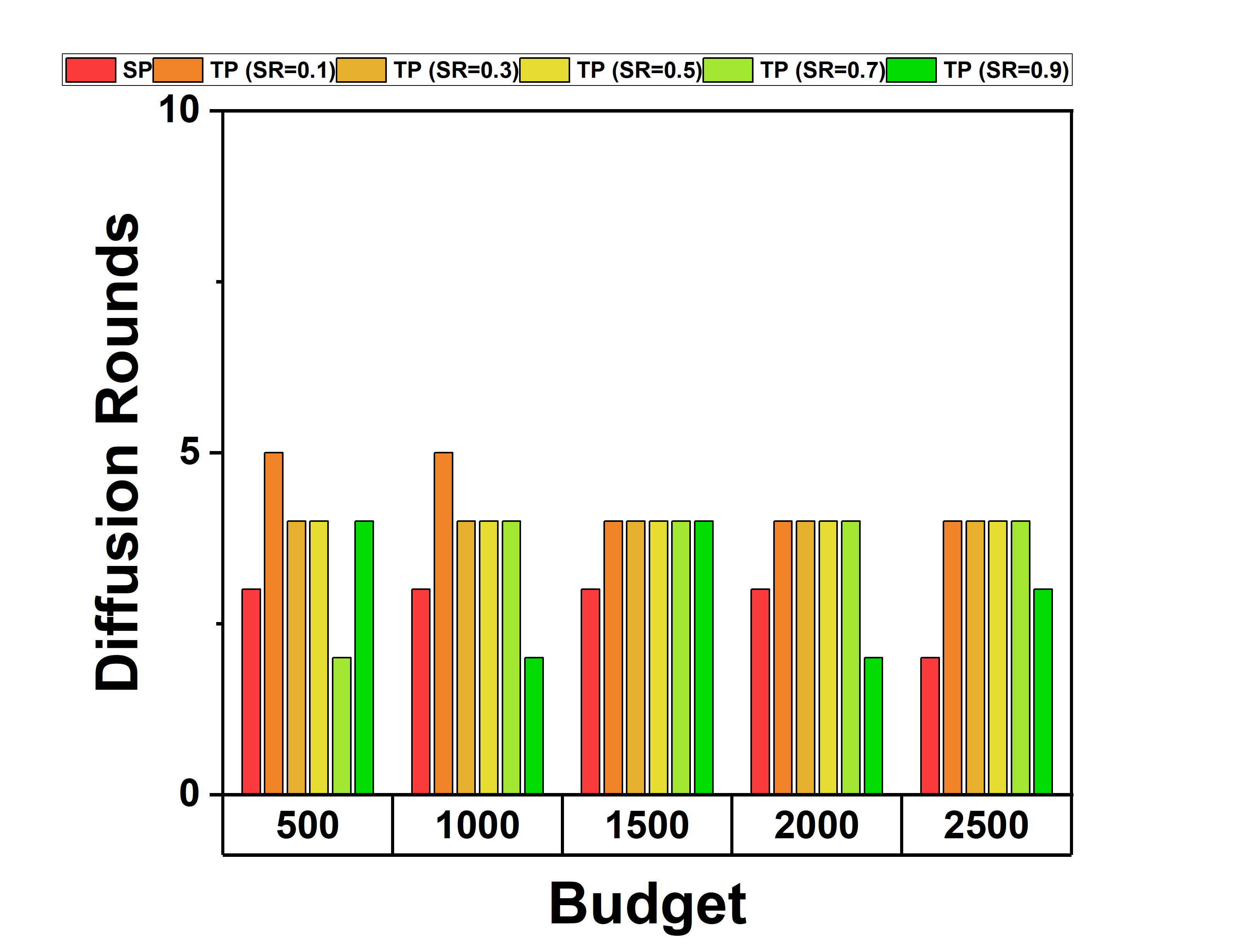}
        \caption{Timestep 2}
    \end{subfigure}
    \hspace{0.05\linewidth}
    \begin{subfigure}[t]{0.3\linewidth}
        \centering
        \includegraphics[width=\linewidth]{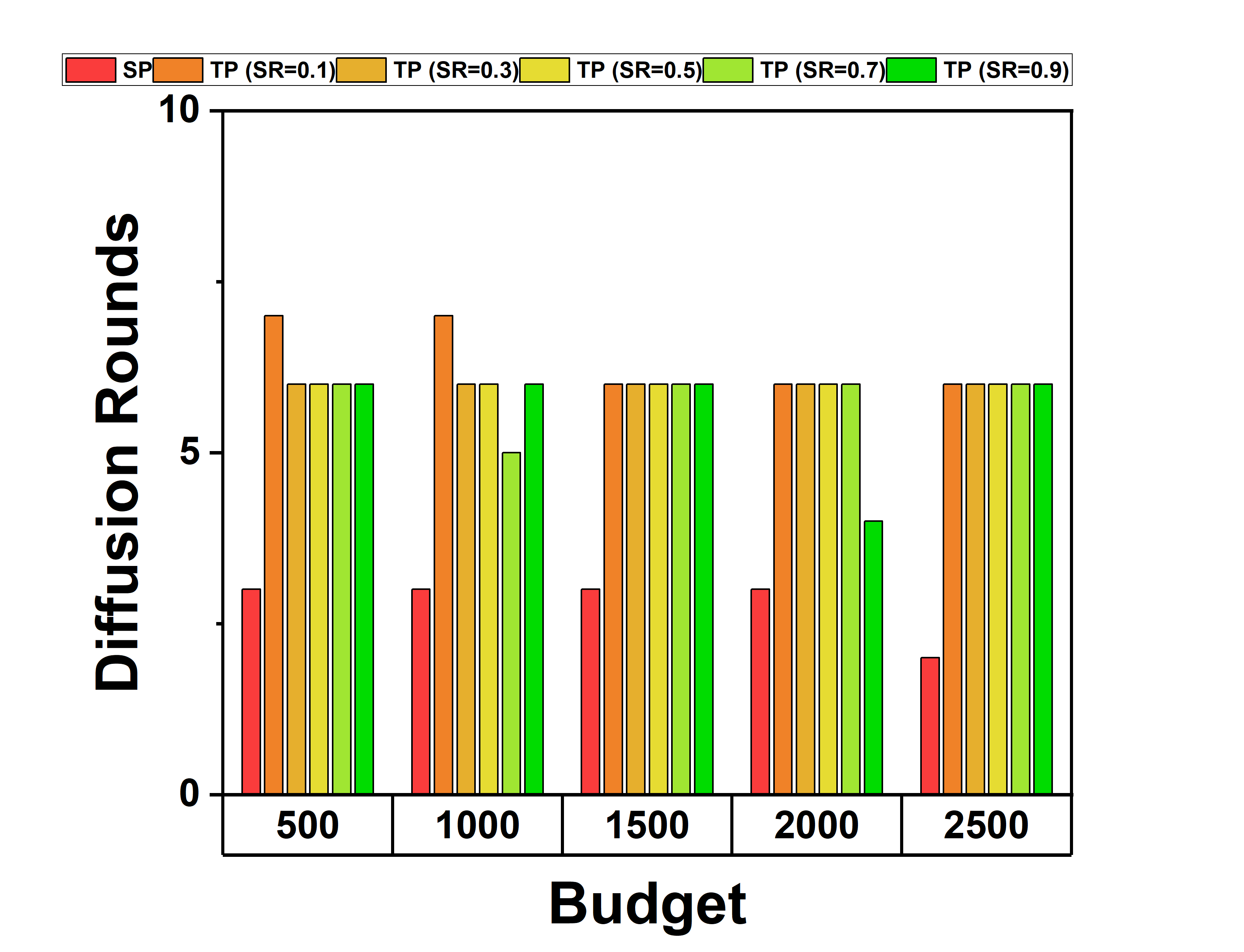}
        \caption{Timestep 4}
    \end{subfigure}

    \vspace{0.5cm}

    \begin{subfigure}[t]{0.3\linewidth}
        \centering
        \includegraphics[width=\linewidth]{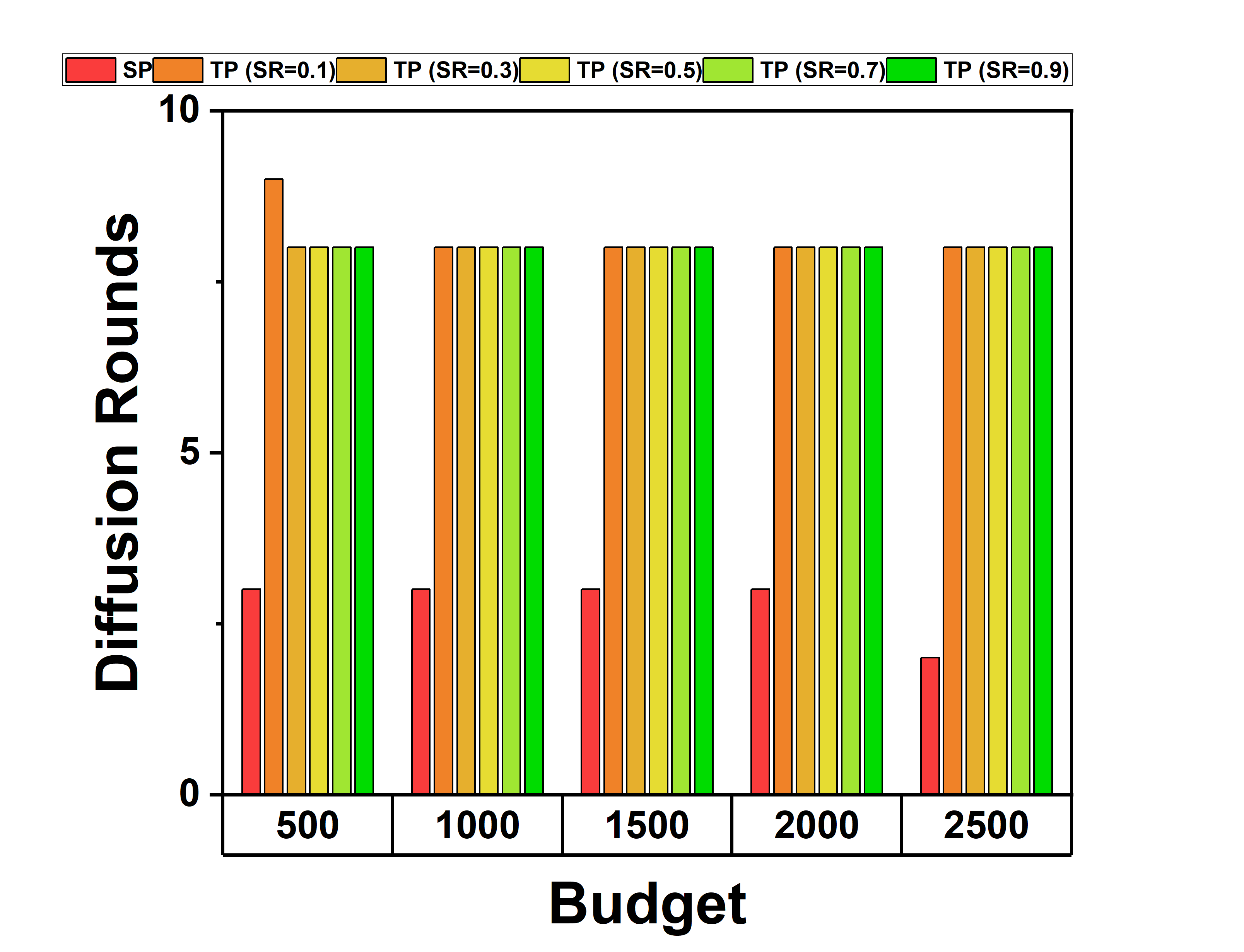}
        \caption{Timestep 6}
    \end{subfigure}
    \hfill
    \begin{subfigure}[t]{0.3\linewidth}
        \centering
        \includegraphics[width=\linewidth]{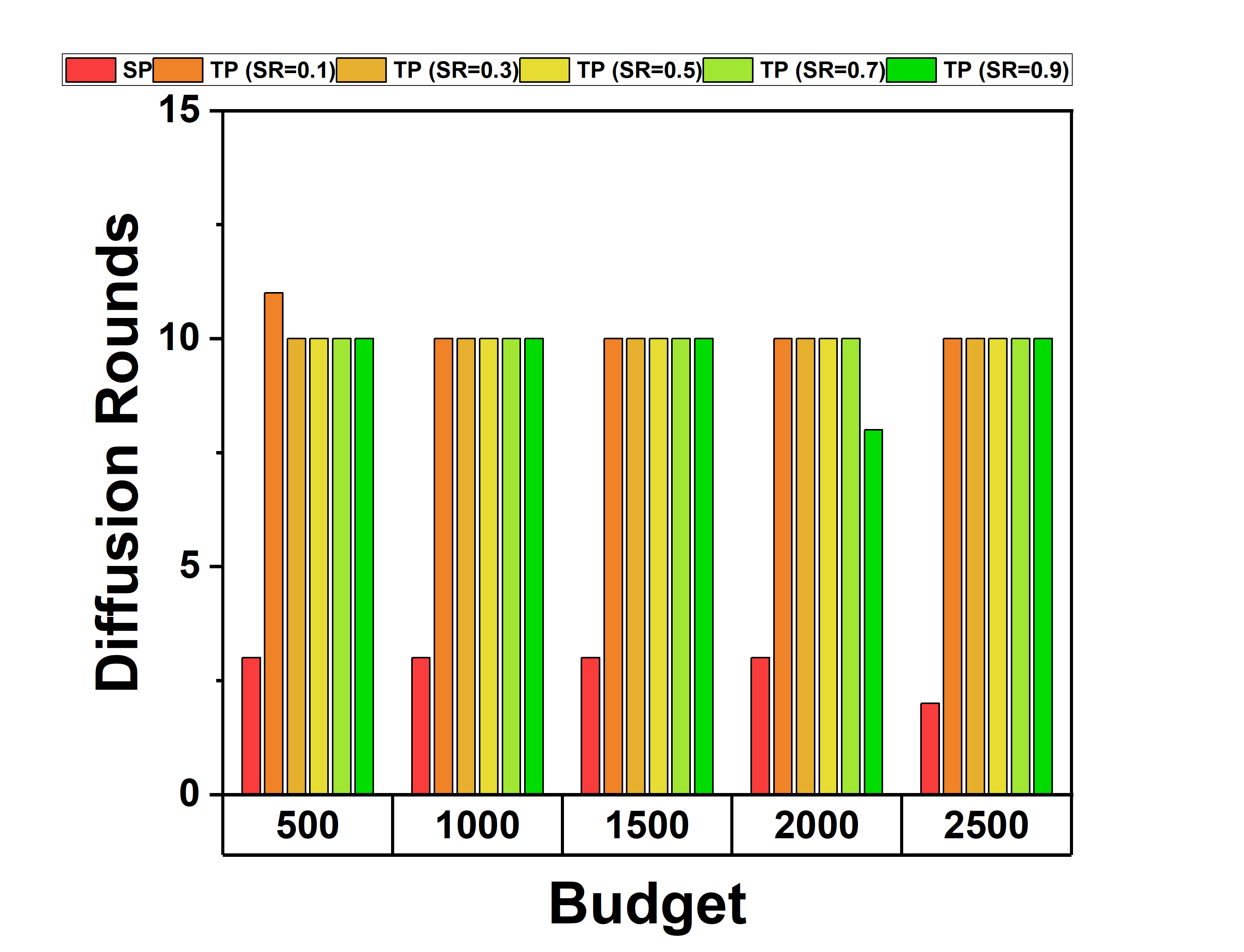}
        \caption{Timestep 8}
    \end{subfigure}
    \hfill
    \begin{subfigure}[t]{0.3\linewidth}
        \centering
        \includegraphics[width=\linewidth]{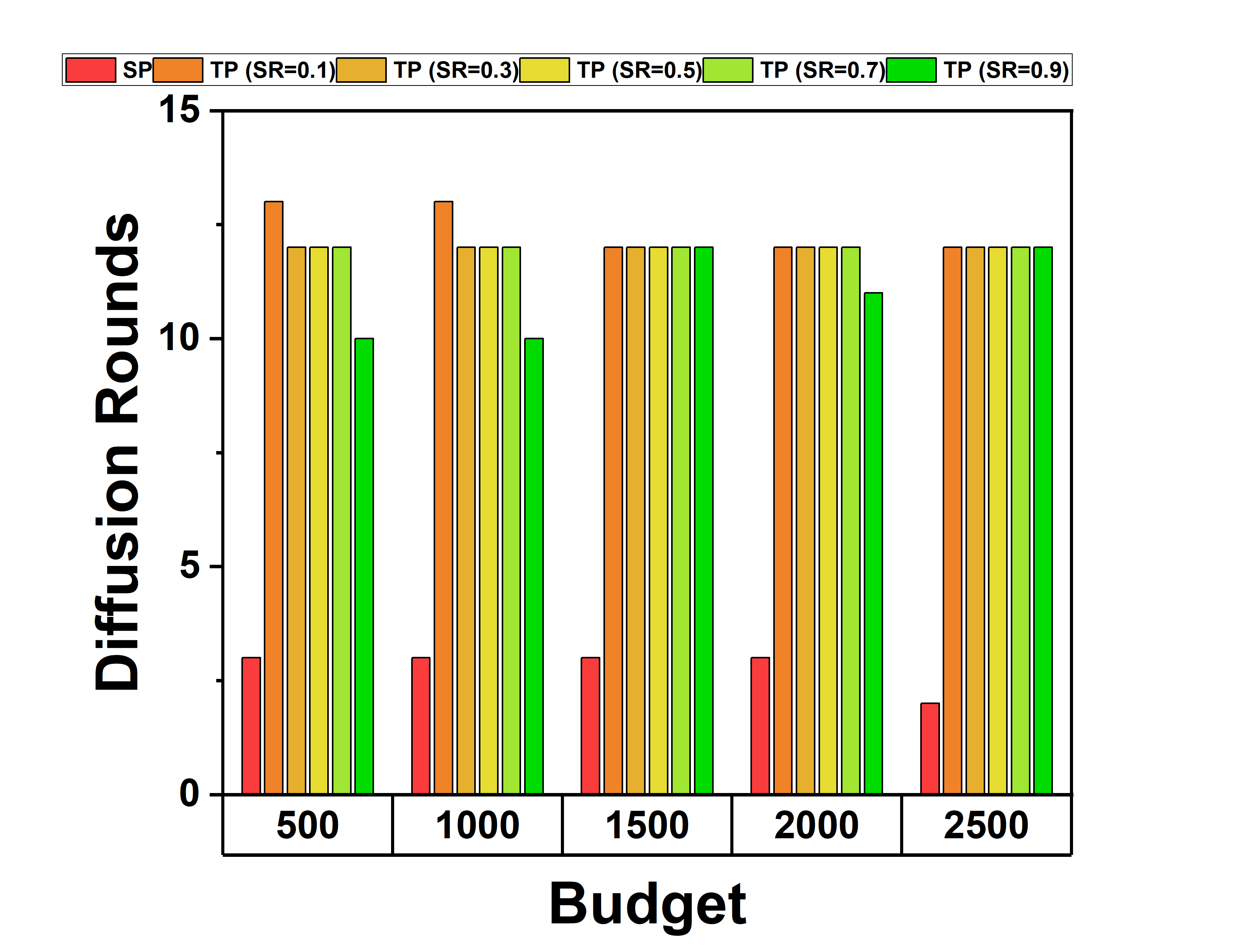}
        \caption{Timestep 10}
    \end{subfigure}

    \caption{Diffusion Rounds in Single Phase Vs. Two Phase (High Degree Algorithm, \textit{LM} Dataset, Probability Setting - Trivalency}
    \label{RQ5LM_T2}
\end{figure}

\begin{figure}[htbp]
    \centering
    \begin{subfigure}[t]{0.3\linewidth}
        \centering
        \includegraphics[width=\linewidth]{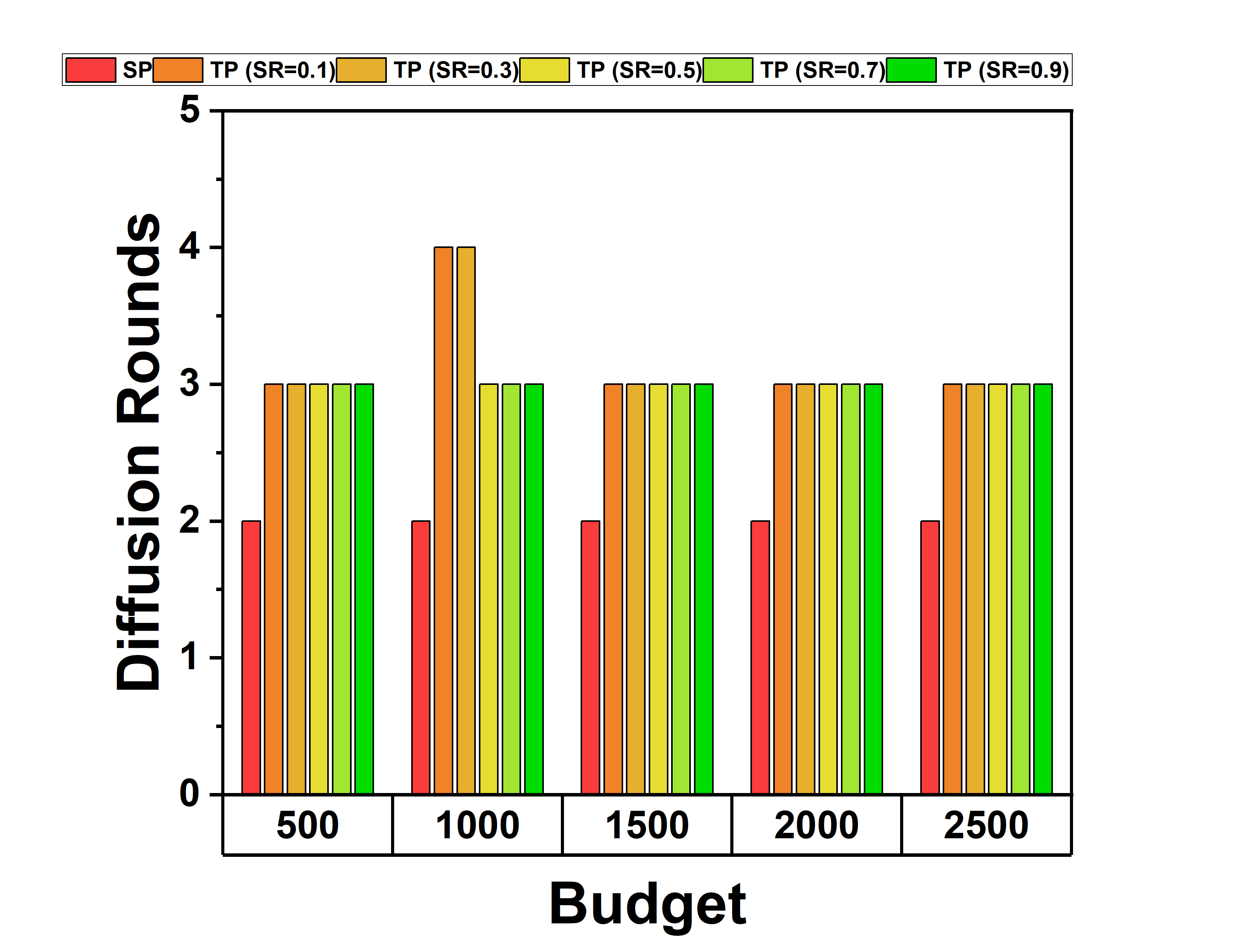}
        \caption{Timestep 2}
    \end{subfigure}
    \hspace{0.05\linewidth}
    \begin{subfigure}[t]{0.3\linewidth}
        \centering
        \includegraphics[width=\linewidth]{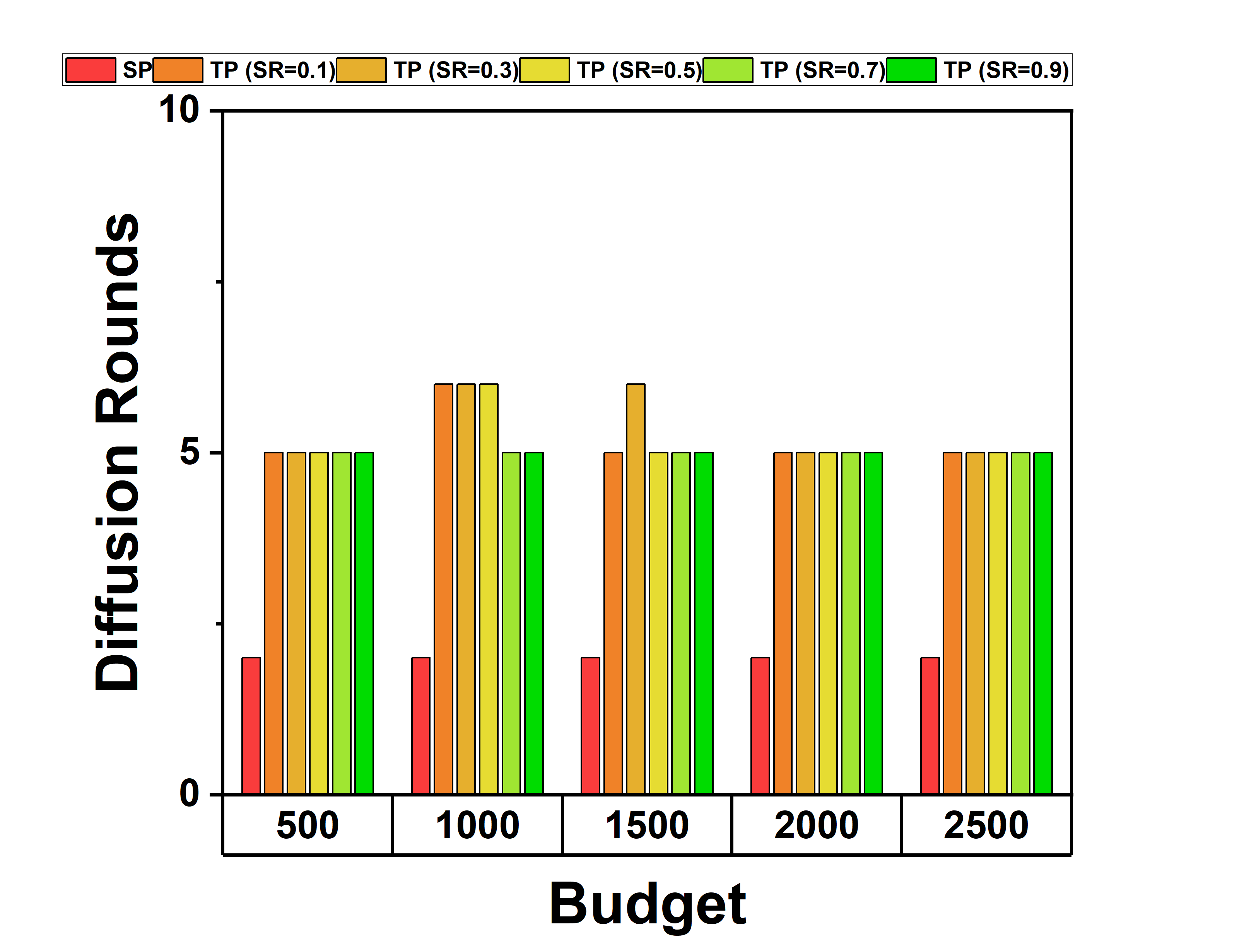}
        \caption{Timestep 4}
    \end{subfigure}

    \vspace{0.5cm}

    \begin{subfigure}[t]{0.3\linewidth}
        \centering
        \includegraphics[width=\linewidth]{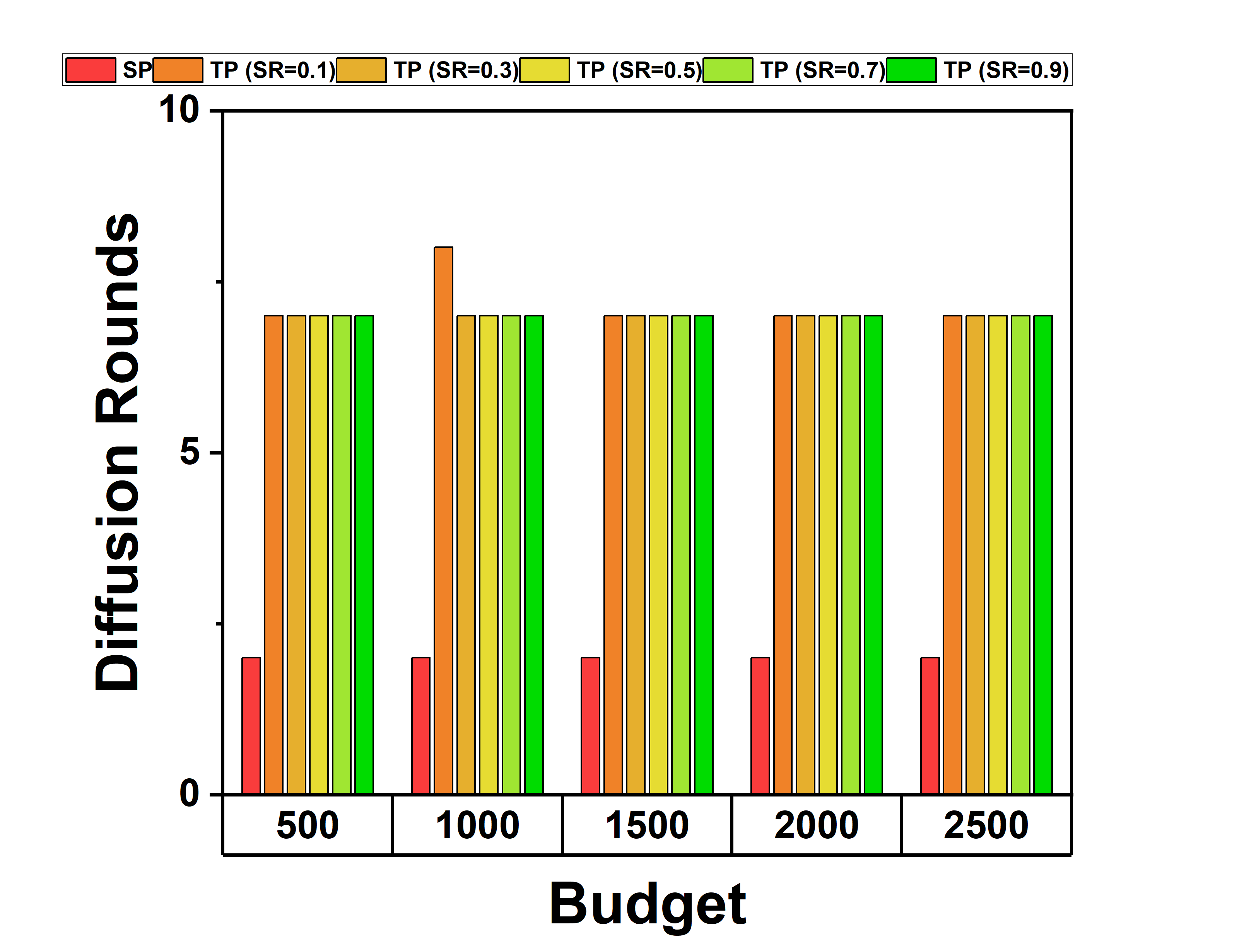}
        \caption{Timestep 6}
    \end{subfigure}
    \hfill
    \begin{subfigure}[t]{0.3\linewidth}
        \centering
        \includegraphics[width=\linewidth]{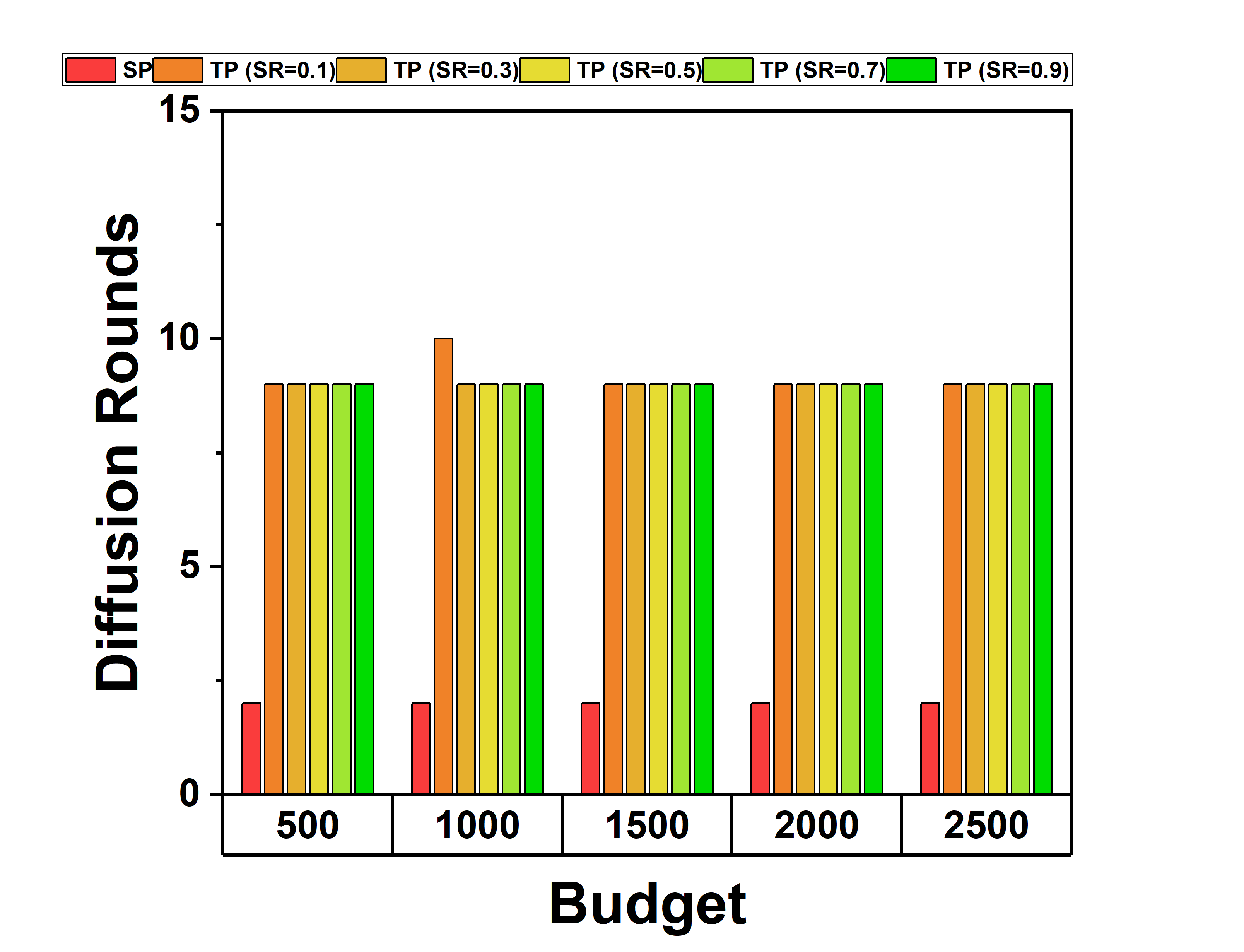}
        \caption{Timestep 8}
    \end{subfigure}
    \hfill
    \begin{subfigure}[t]{0.3\linewidth}
        \centering
        \includegraphics[width=\linewidth]{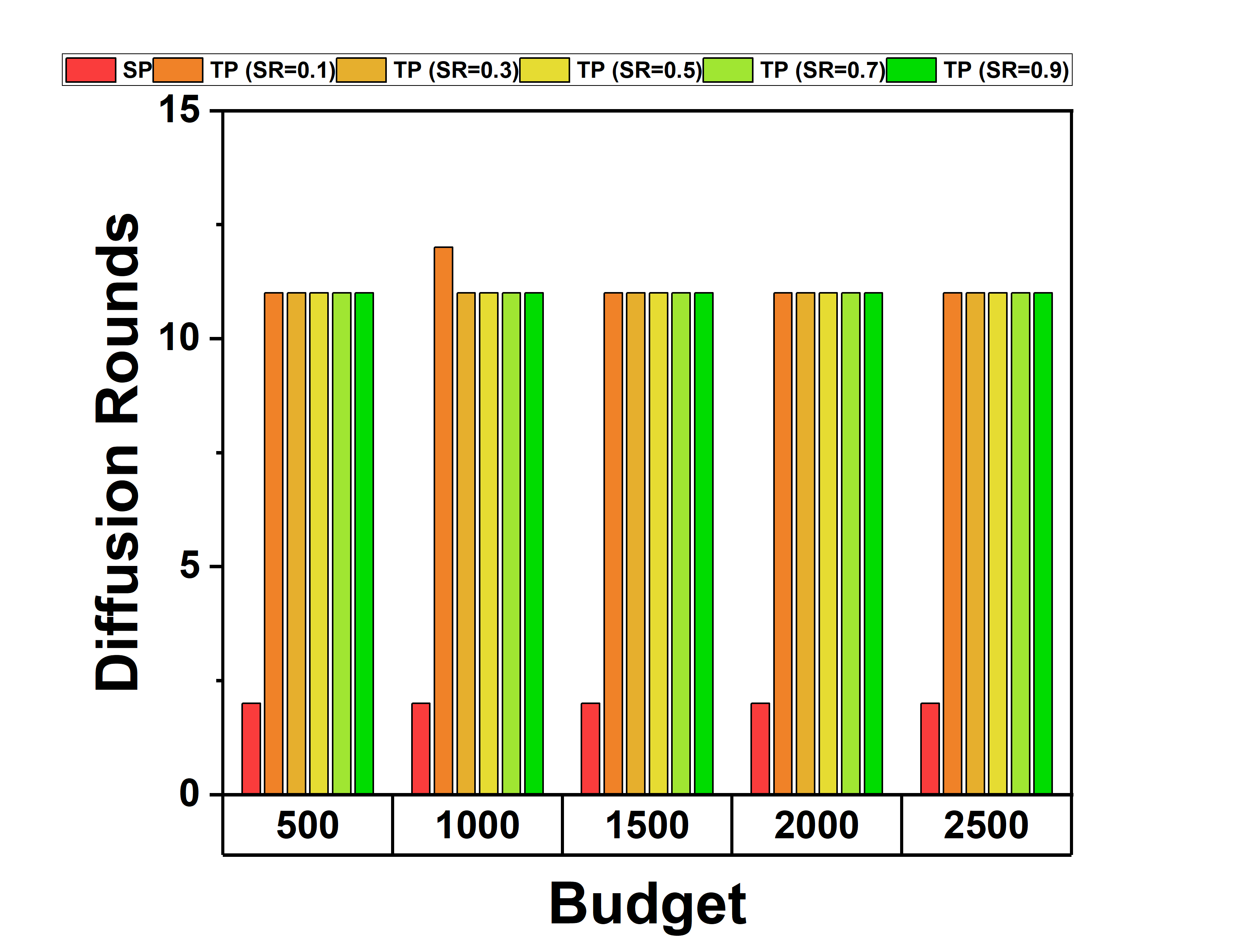}
        \caption{Timestep 10}
    \end{subfigure}

    \caption{Diffusion Rounds in Single Phase Vs. Two Phase (Clustering Coefficient Algorithm, \textit{LM} Dataset, Probability Setting - Trivalency)}
    \label{RQ5LM_T3}
\end{figure}

\begin{figure}[htbp]
    \centering
    \begin{subfigure}[t]{0.3\linewidth}
        \centering 
        \includegraphics[width=\linewidth]{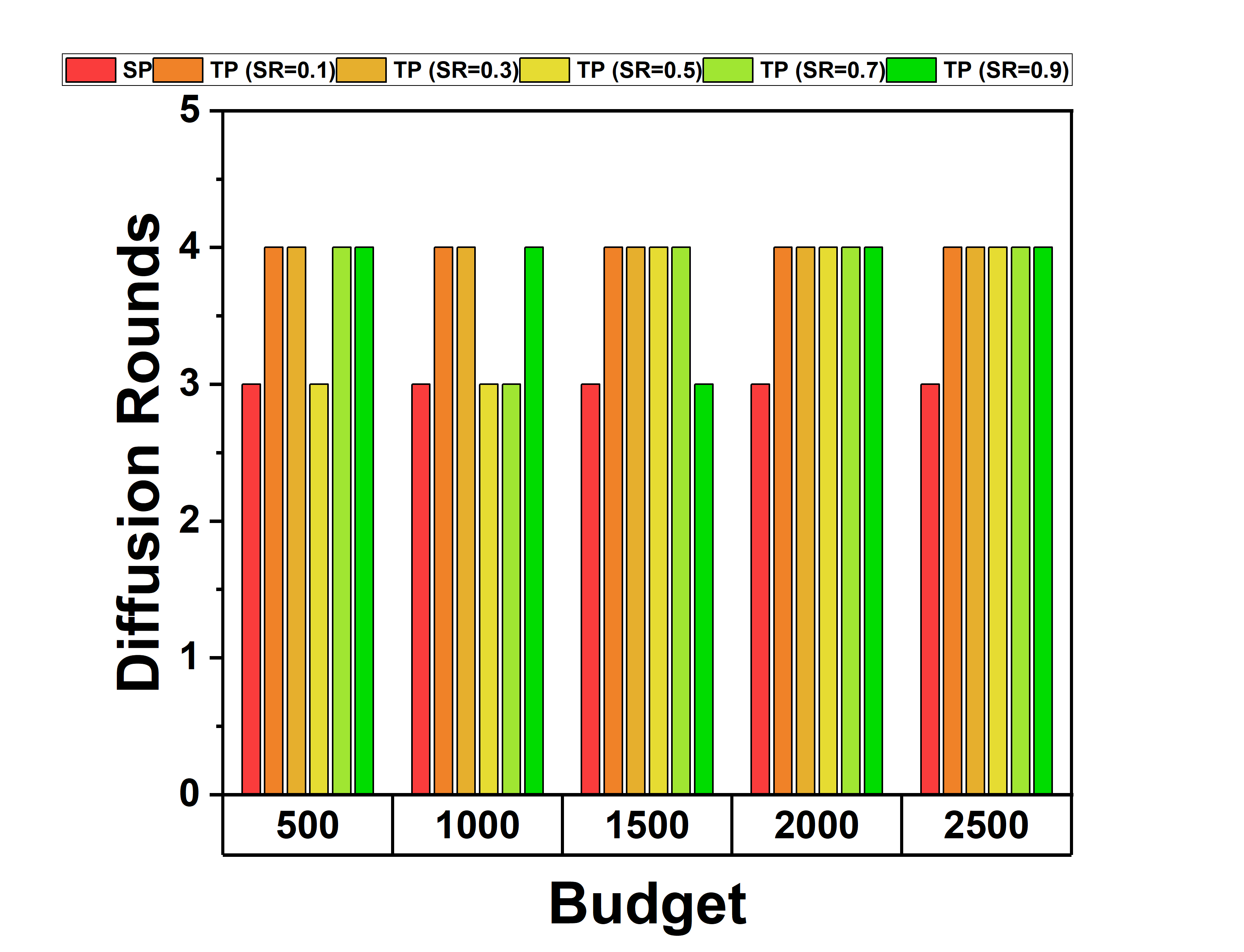}
        \caption{Timestep 2}
    \end{subfigure}
    \hspace{0.05\linewidth}
    \begin{subfigure}[t]{0.3\linewidth}
        \centering
        \includegraphics[width=\linewidth]{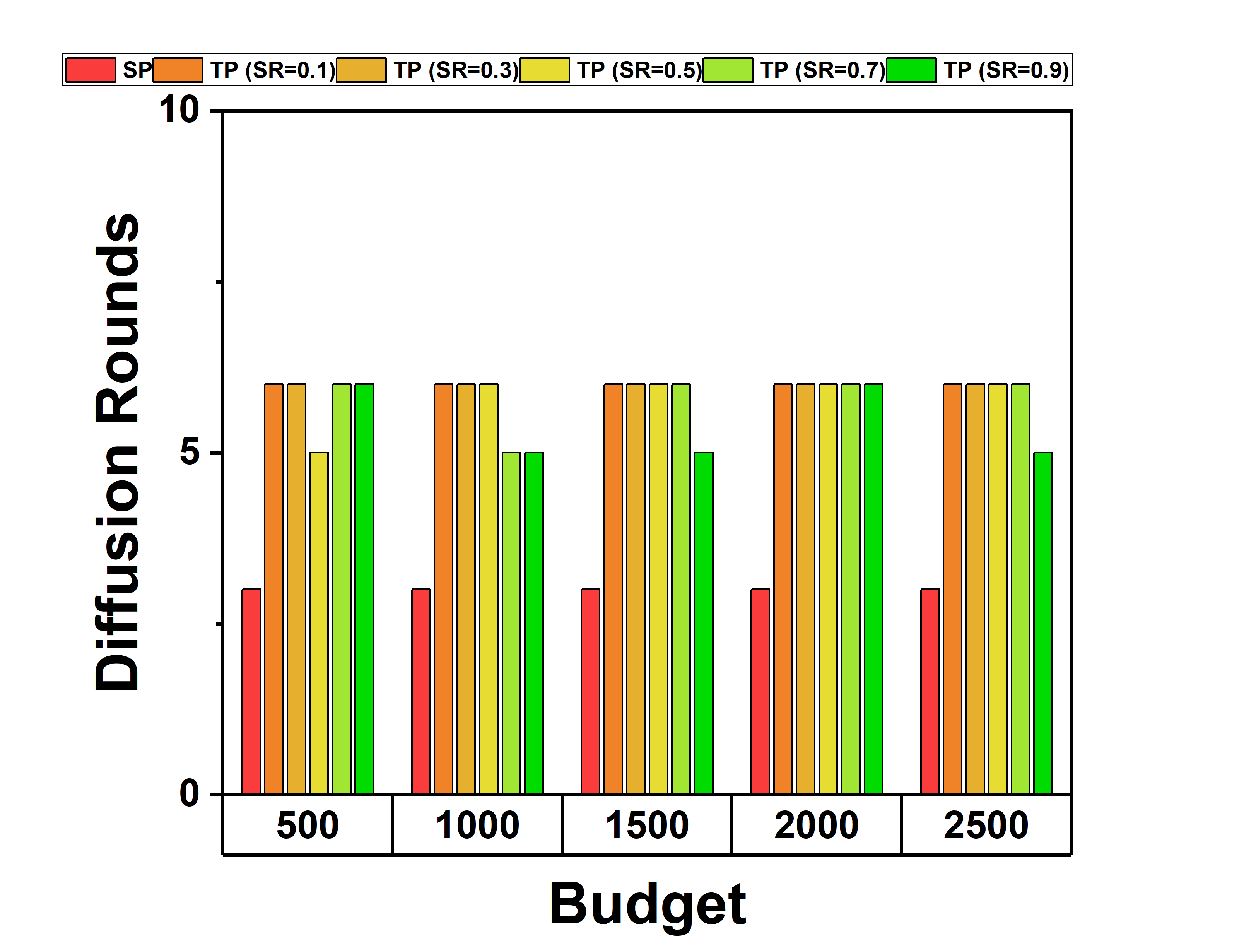}
        \caption{Timestep 4}
    \end{subfigure}

    \vspace{0.5cm}

    \begin{subfigure}[t]{0.3\linewidth}
        \centering
        \includegraphics[width=\linewidth]{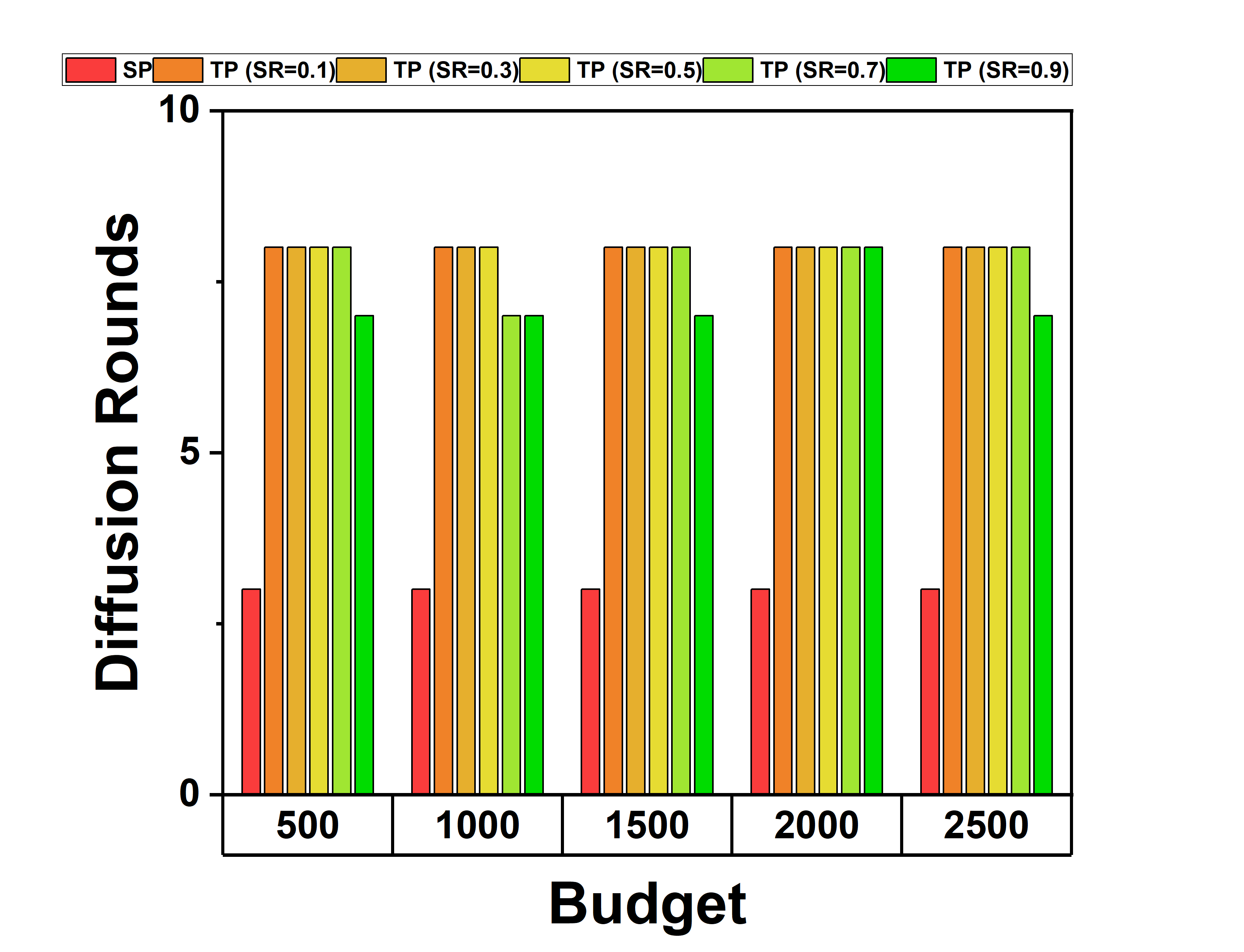}
        \caption{Timestep 6}
    \end{subfigure}
    \hfill
    \begin{subfigure}[t]{0.3\linewidth}
        \centering
        \includegraphics[width=\linewidth]{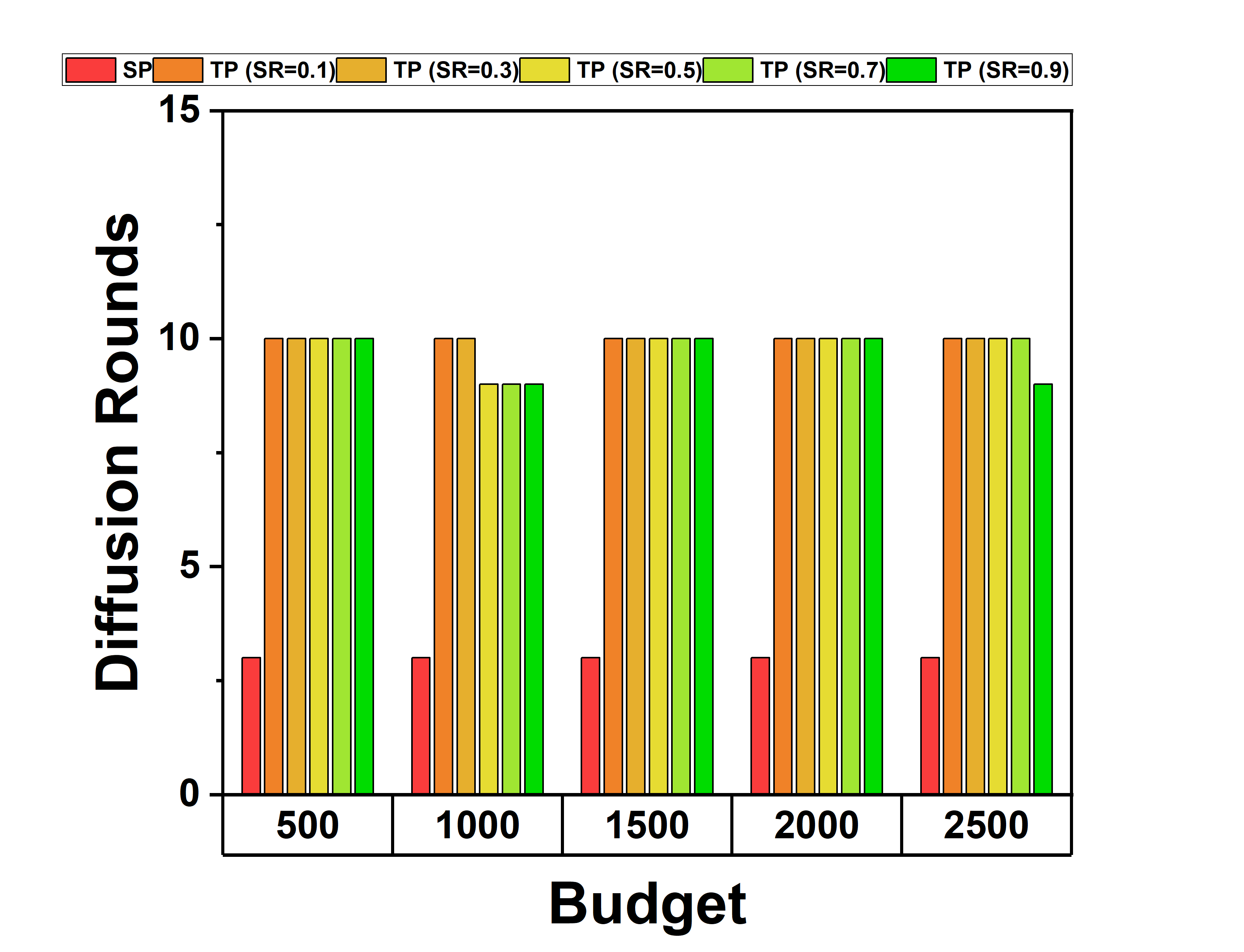}
        \caption{Timestep 8}
    \end{subfigure}
    \hfill
    \begin{subfigure}[t]{0.3\linewidth}
        \centering
        \includegraphics[width=\linewidth]{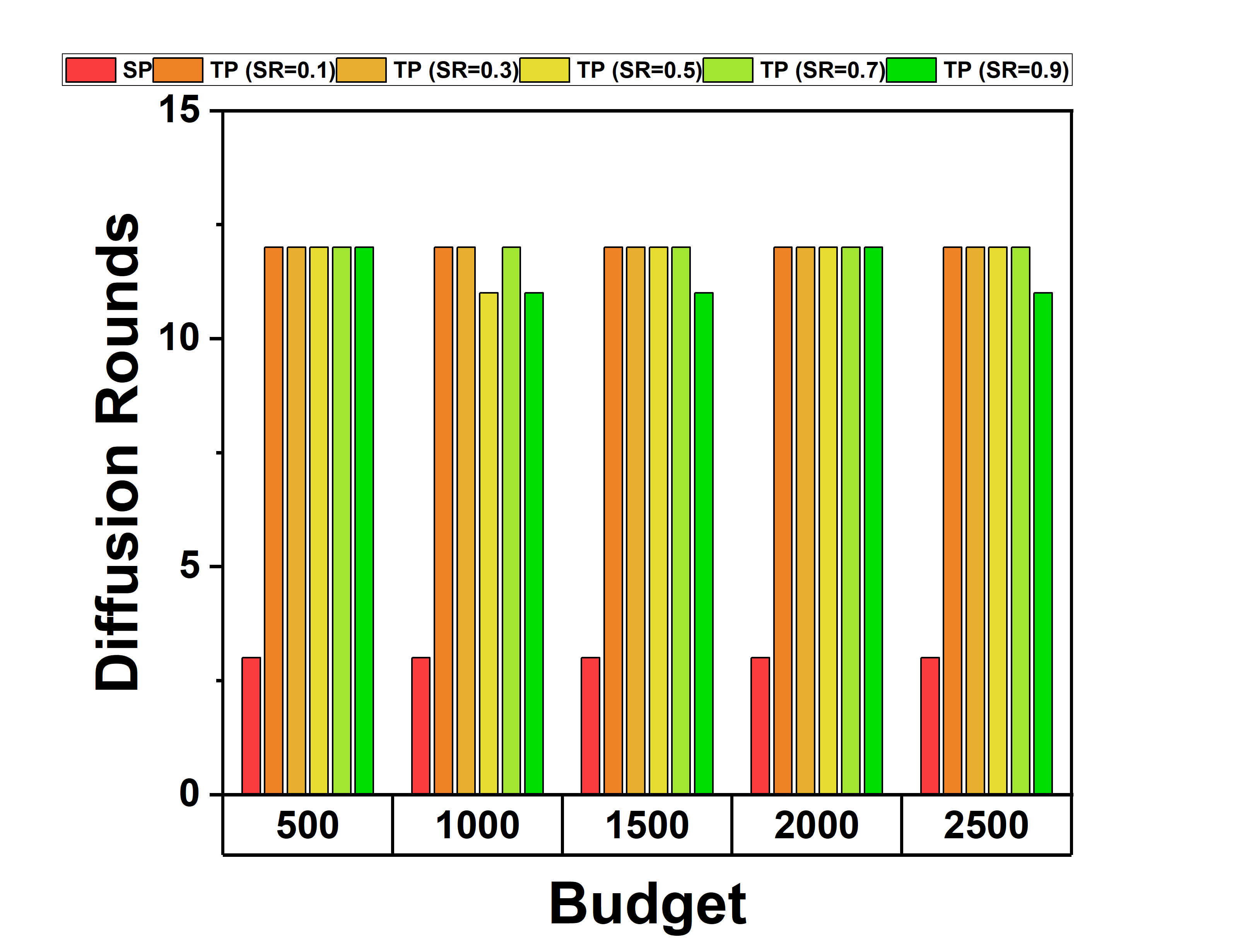}
        \caption{Timestep 10}
    \end{subfigure}
    \caption{Diffusion Rounds in Single Phase Vs. Two Phase (Degree Discount Algorithm, \textit{LM} Dataset, Probability Setting - Trivalency)}
    \label{RQ5LM_T4}
\end{figure}

\begin{figure}[htbp]
    \centering
    \begin{subfigure}[t]{0.3\linewidth}
        \centering
        \includegraphics[width=\linewidth]{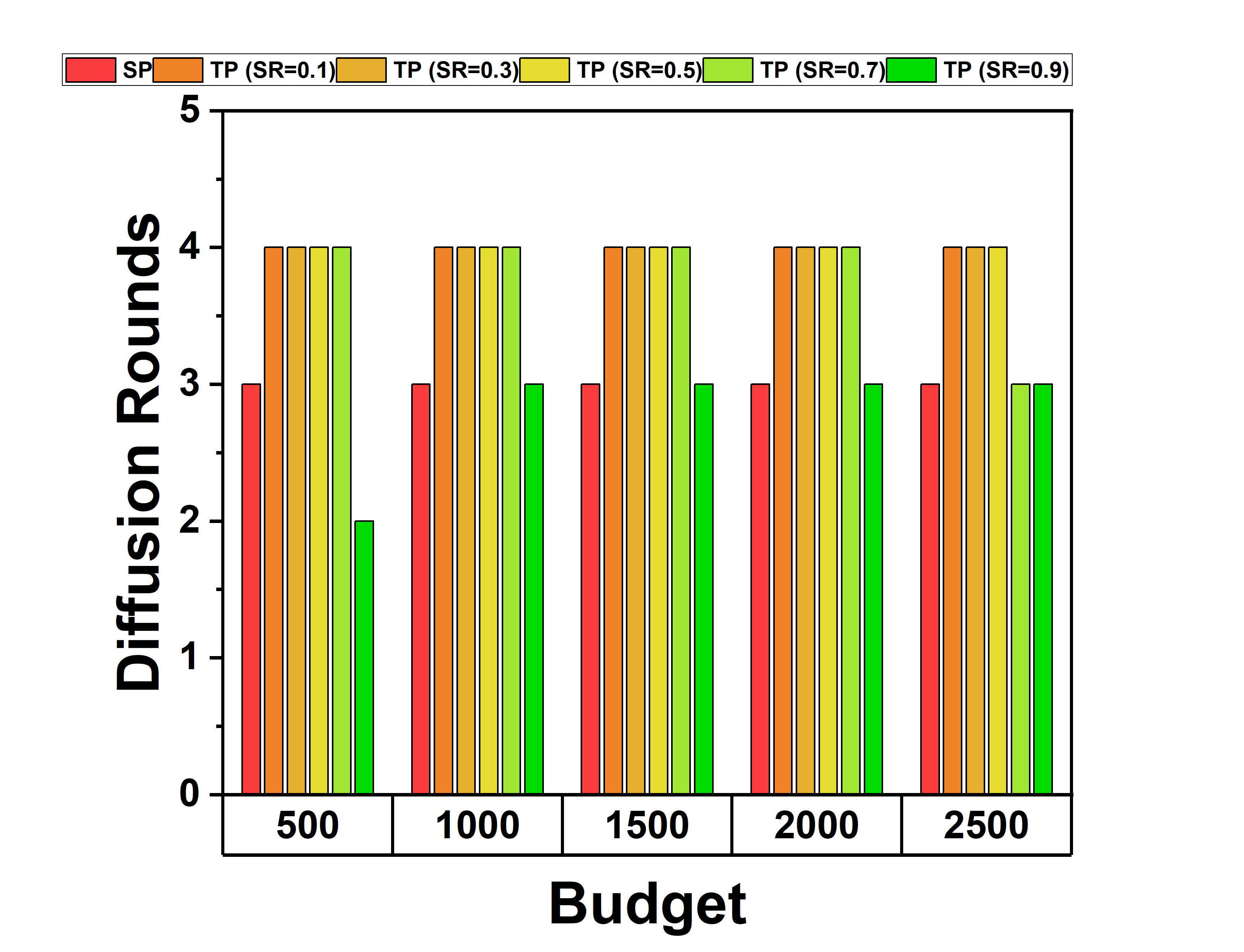}
        \caption{Timestep 2}
    \end{subfigure}
    \hspace{0.05\linewidth}
    \begin{subfigure}[t]{0.3\linewidth}
        \centering
        \includegraphics[width=\linewidth]{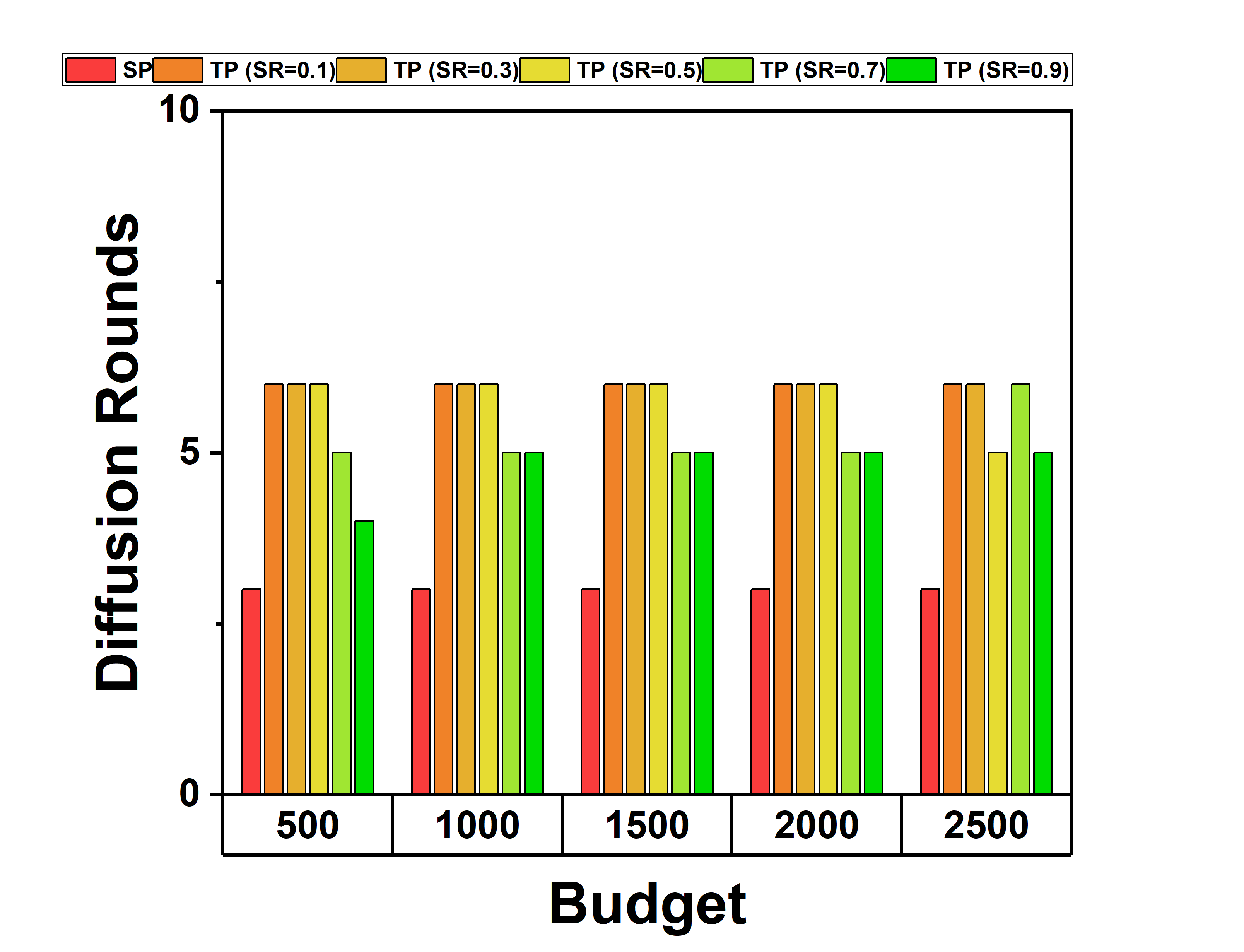}
        \caption{Timestep 4}
    \end{subfigure}

    \vspace{0.5cm}

    \begin{subfigure}[t]{0.3\linewidth}
        \centering
        \includegraphics[width=\linewidth]{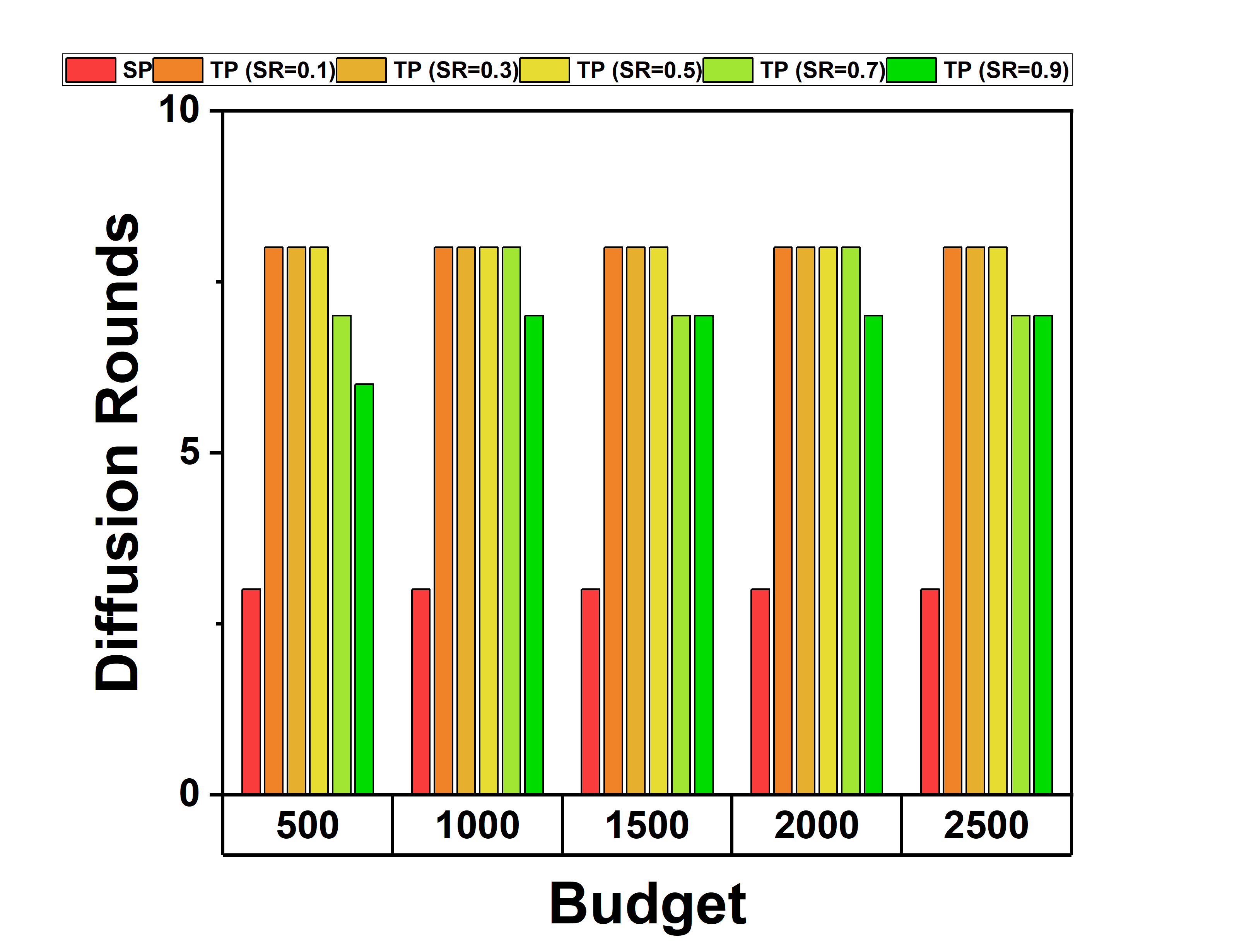}
        \caption{Timestep 6}
    \end{subfigure}
    \hfill
    \begin{subfigure}[t]{0.3\linewidth}
        \centering
        \includegraphics[width=\linewidth]{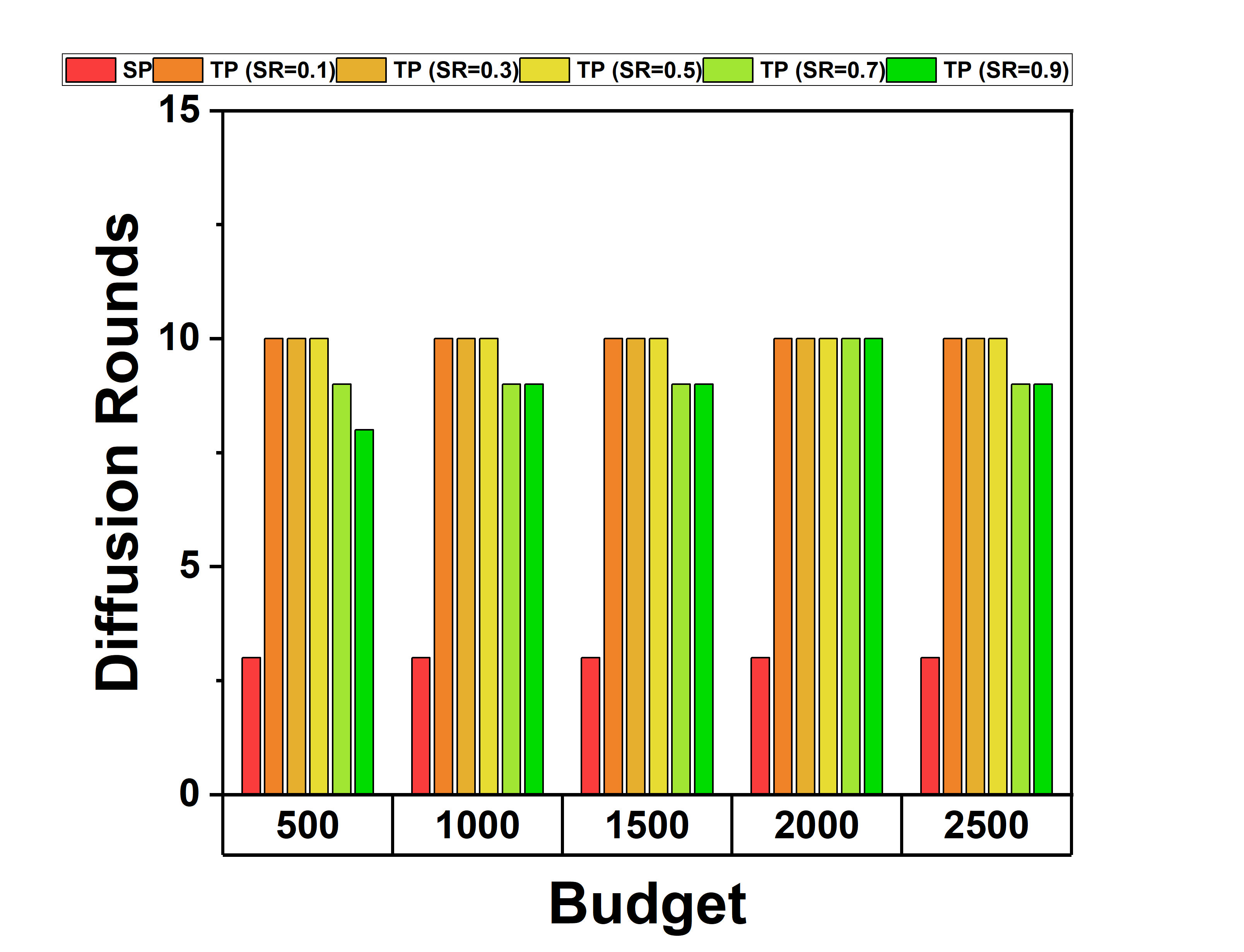}
        \caption{Timestep 8}
    \end{subfigure}
    \hfill
    \begin{subfigure}[t]{0.3\linewidth}
        \centering
        \includegraphics[width=\linewidth]{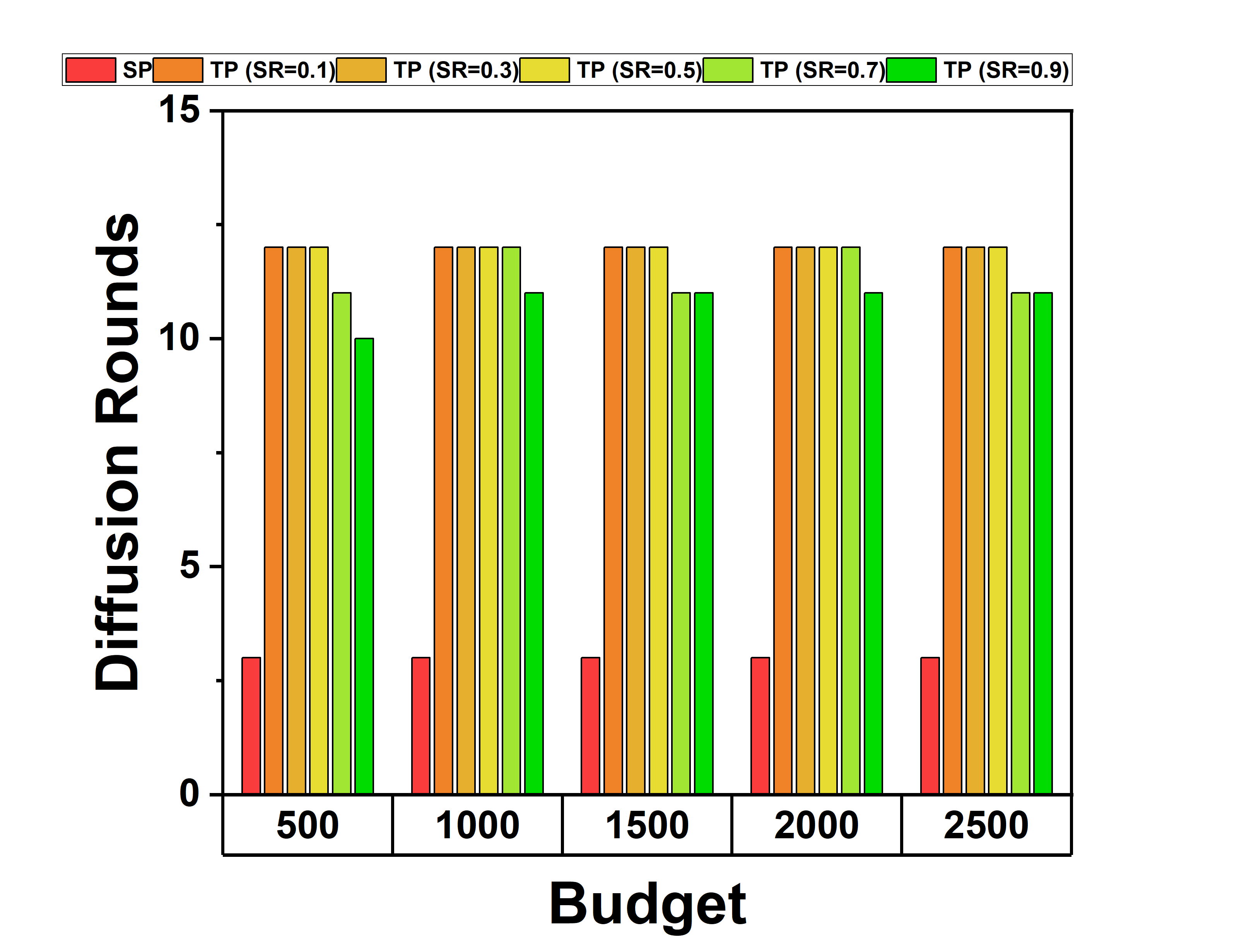}
        \caption{Timestep 10}
    \end{subfigure}

    \caption{Diffusion Rounds in Single Phase Vs. Two Phase (Single Discount Algorithm, \textit{LM} Dataset, Probability Setting - Trivalency)}
    \label{RQ5LM_T5}
\end{figure}

\begin{figure}[htbp]
    \centering
    \begin{subfigure}[t]{0.3\linewidth}
        \centering
        \includegraphics[width=\linewidth]{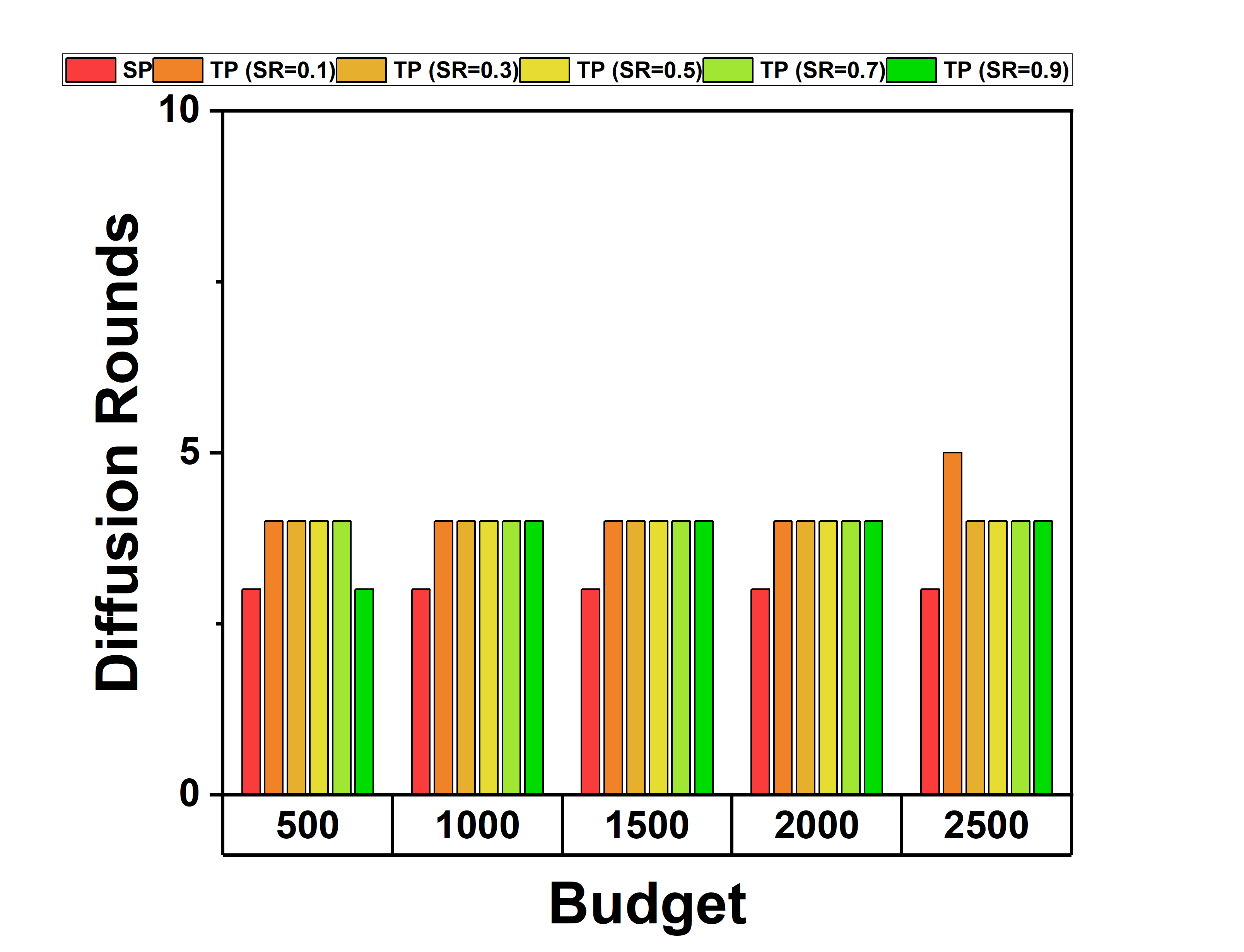}
        \caption{Timestep 2}
    \end{subfigure}
    \hspace{0.05\linewidth}
    \begin{subfigure}[t]{0.3\linewidth}
        \centering
        \includegraphics[width=\linewidth]{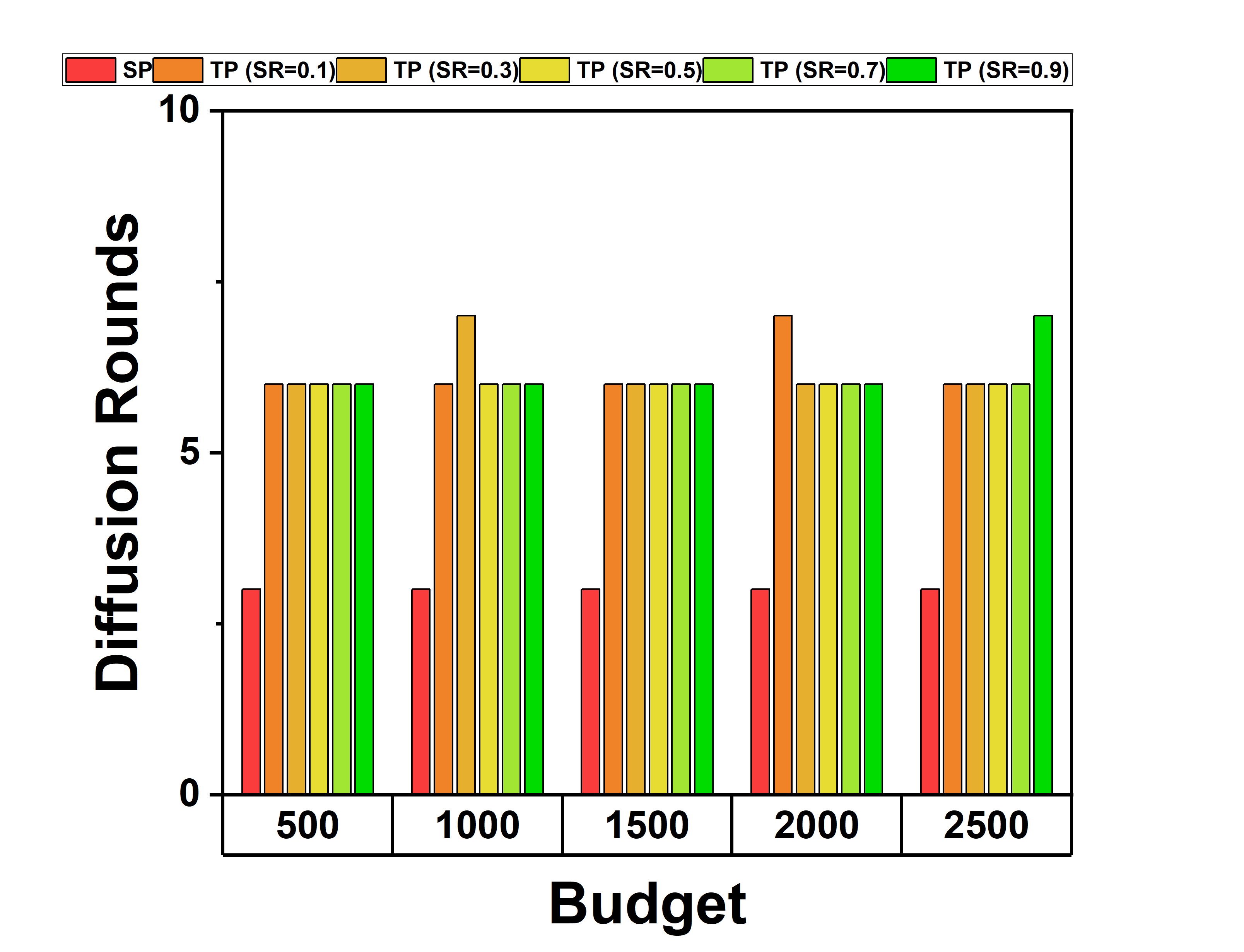}
        \caption{Timestep 4}
    \end{subfigure}

    \vspace{0.5cm}

    \begin{subfigure}[t]{0.3\linewidth}
        \centering
        \includegraphics[width=\linewidth]{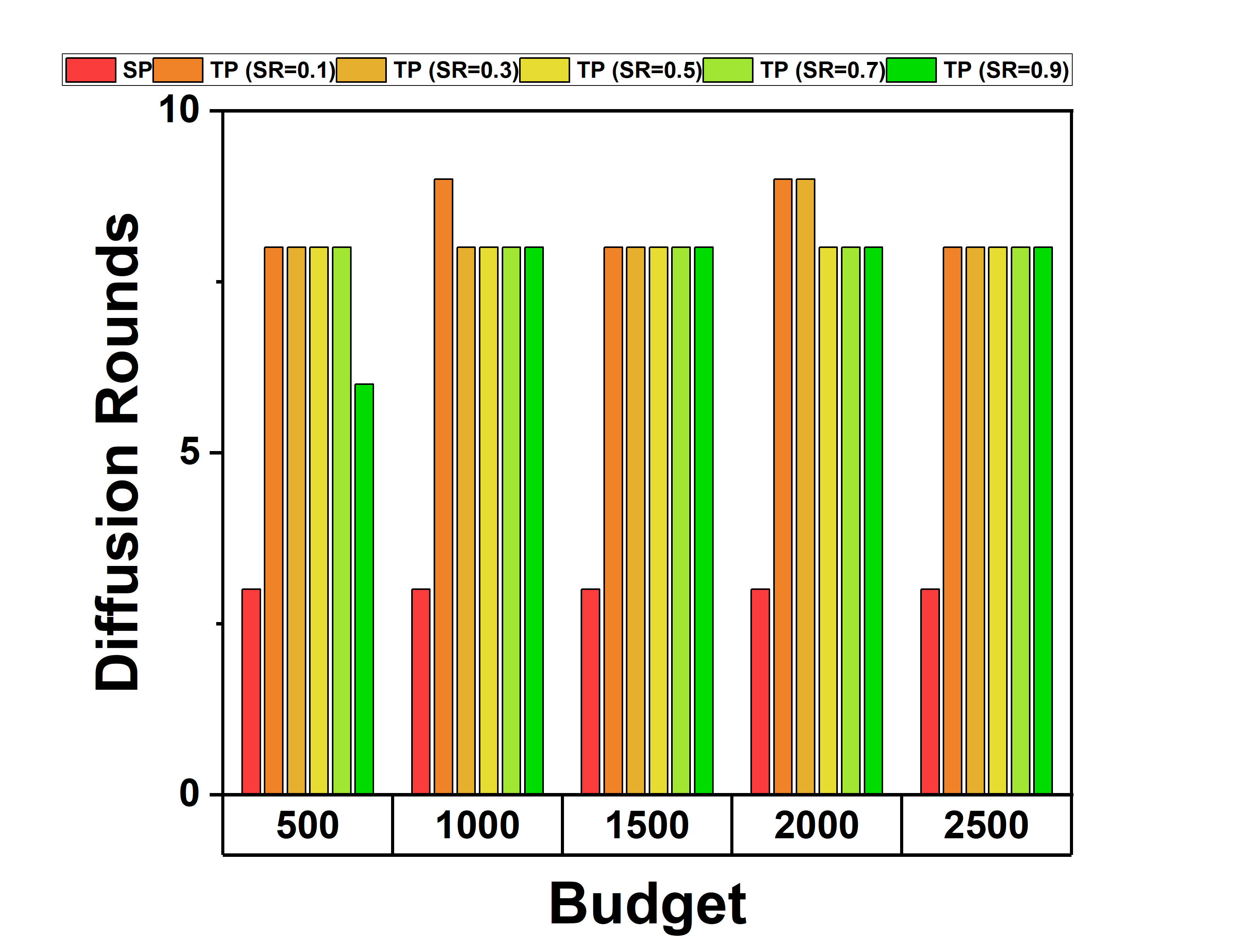}
        \caption{Timestep 6}
    \end{subfigure}
    \hfill
    \begin{subfigure}[t]{0.3\linewidth}
        \centering
        \includegraphics[width=\linewidth]{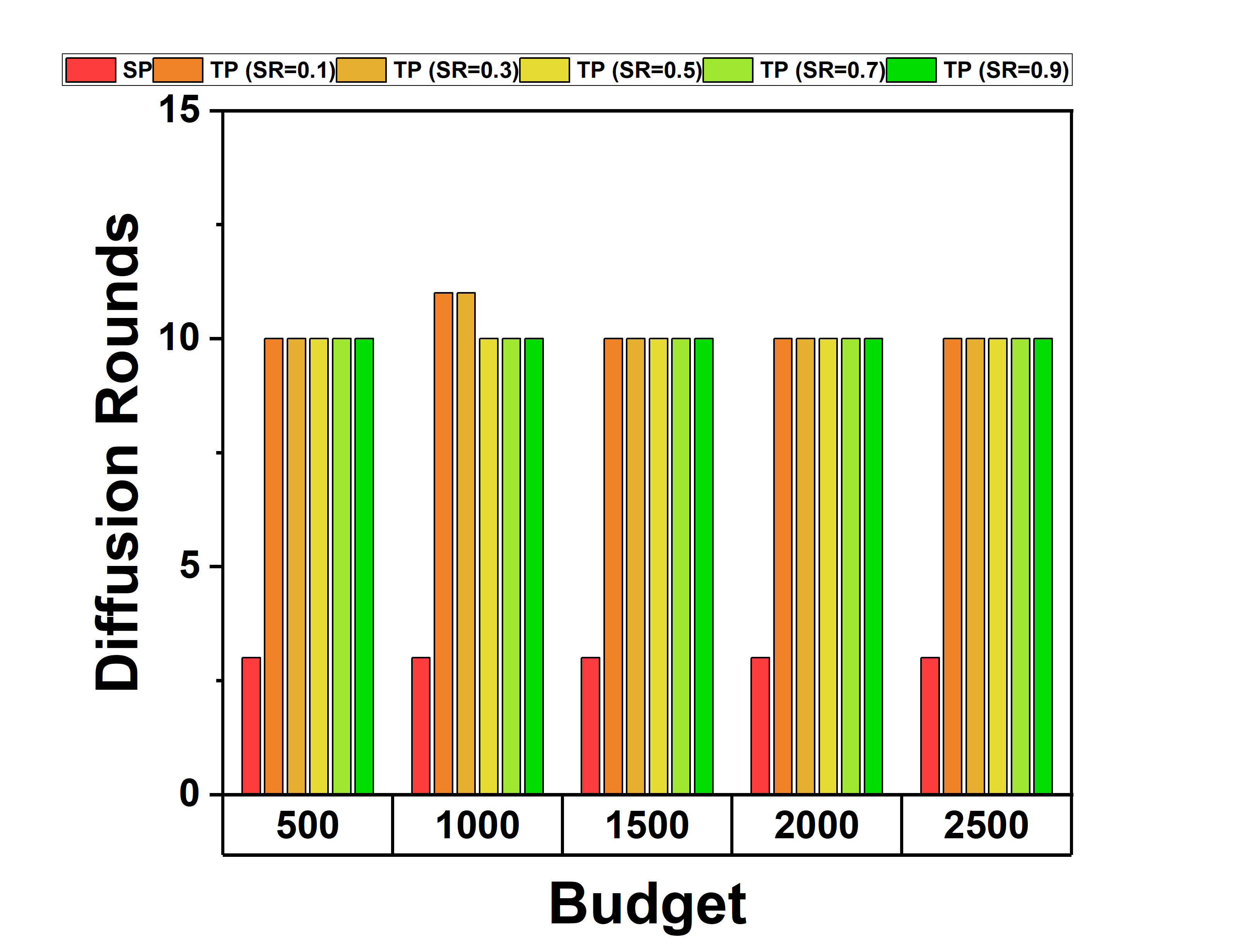}
        \caption{Timestep 8}
    \end{subfigure}
    \hfill
    \begin{subfigure}[t]{0.3\linewidth}
        \centering
        \includegraphics[width=\linewidth]{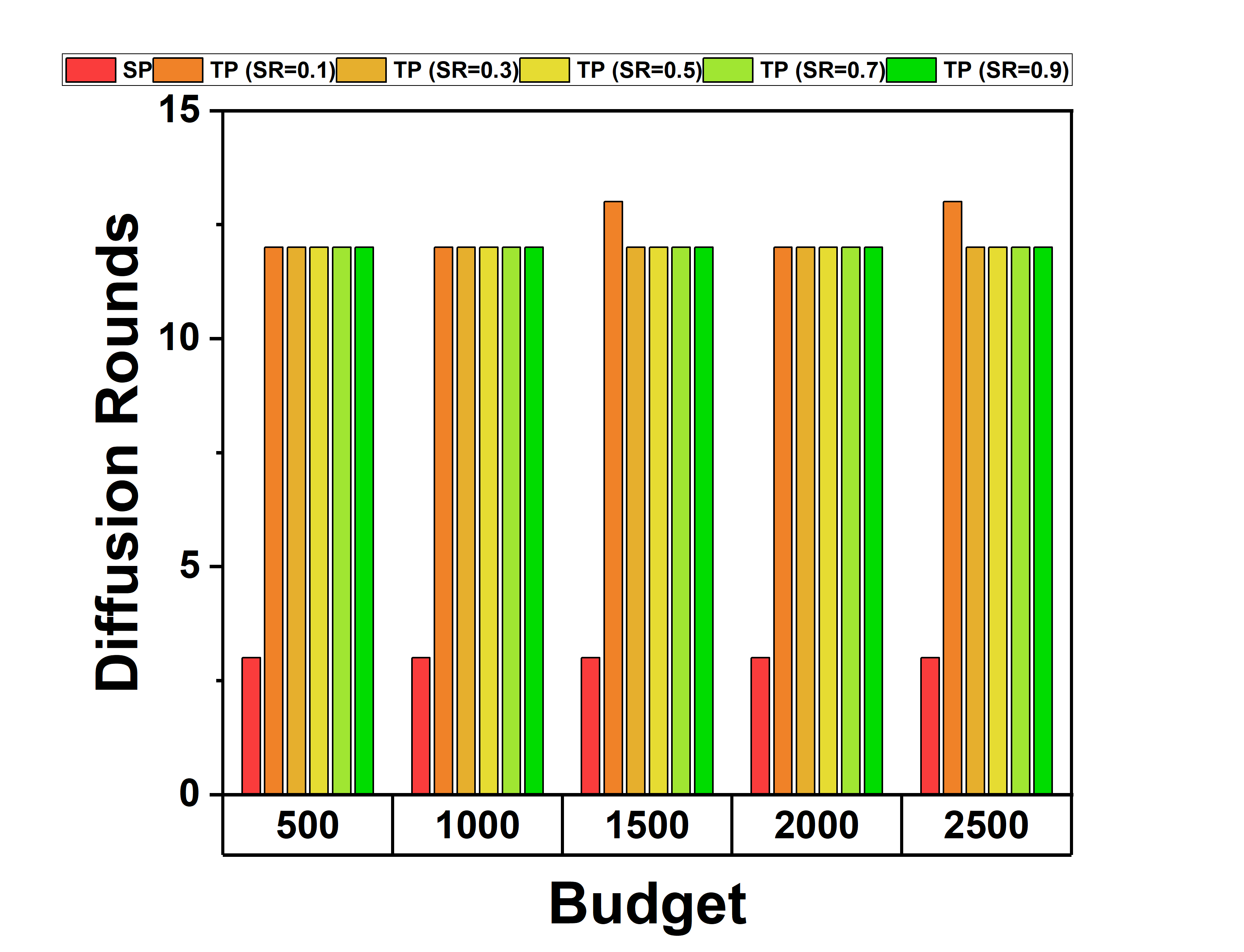}
        \caption{Timestep 10}
    \end{subfigure}

    \caption{Diffusion Rounds in Single Phase Vs. Two Phase (Simple Greedy Algorithm, \textit{LM} Dataset, Probability Setting - Trivalency)}
    \label{RQ5LM_T6}
\end{figure}

\begin{figure}[htbp]
    \centering
    \begin{subfigure}[t]{0.3\linewidth}
        \centering
        \includegraphics[width=\linewidth]{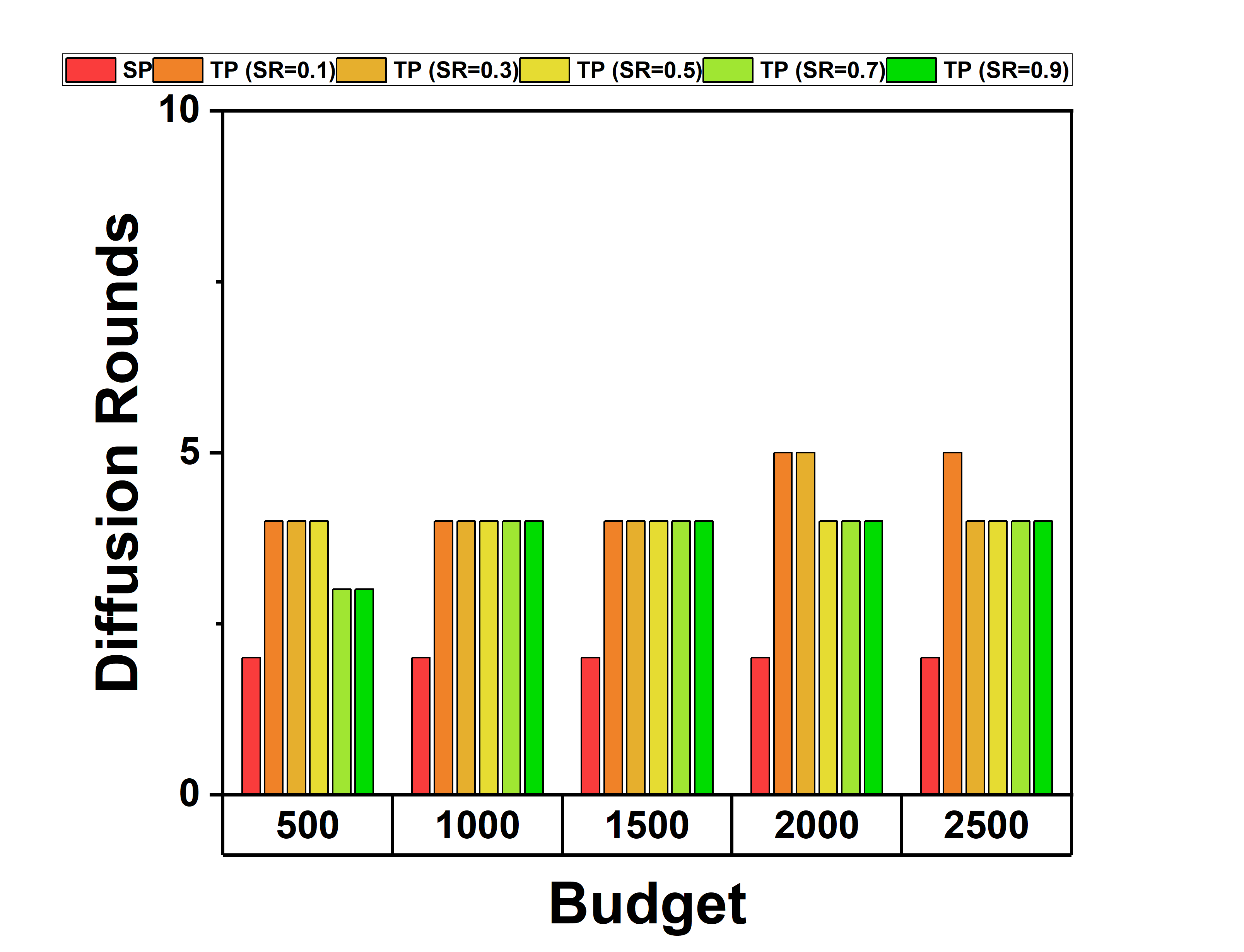}
        \caption{Timestep 2}
    \end{subfigure}
    \hspace{0.05\linewidth}
    \begin{subfigure}[t]{0.3\linewidth}
        \centering
        \includegraphics[width=\linewidth]{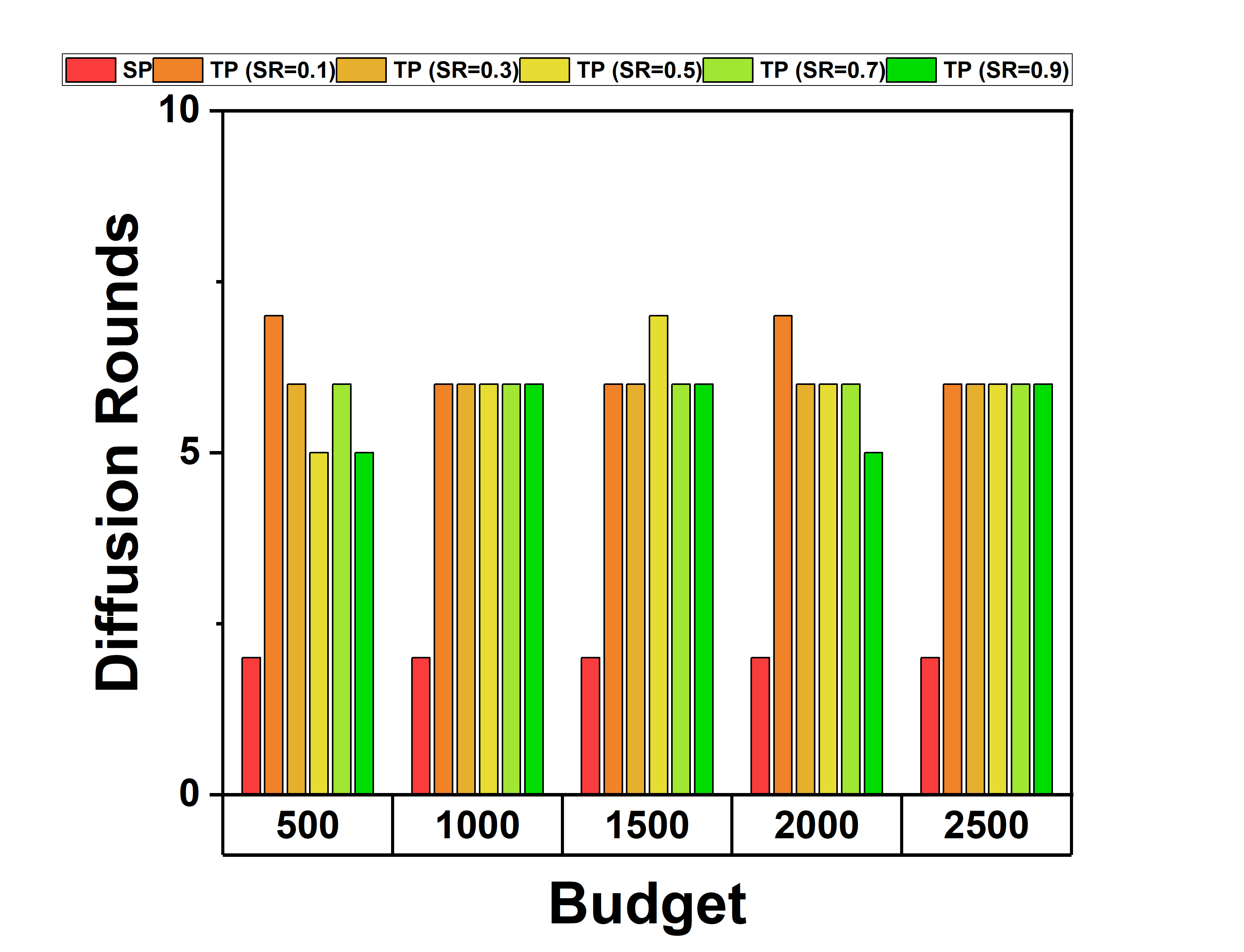}
        \caption{Timestep 4}
    \end{subfigure}

    \vspace{0.5cm}

    \begin{subfigure}[t]{0.3\linewidth}
        \centering
        \includegraphics[width=\linewidth]{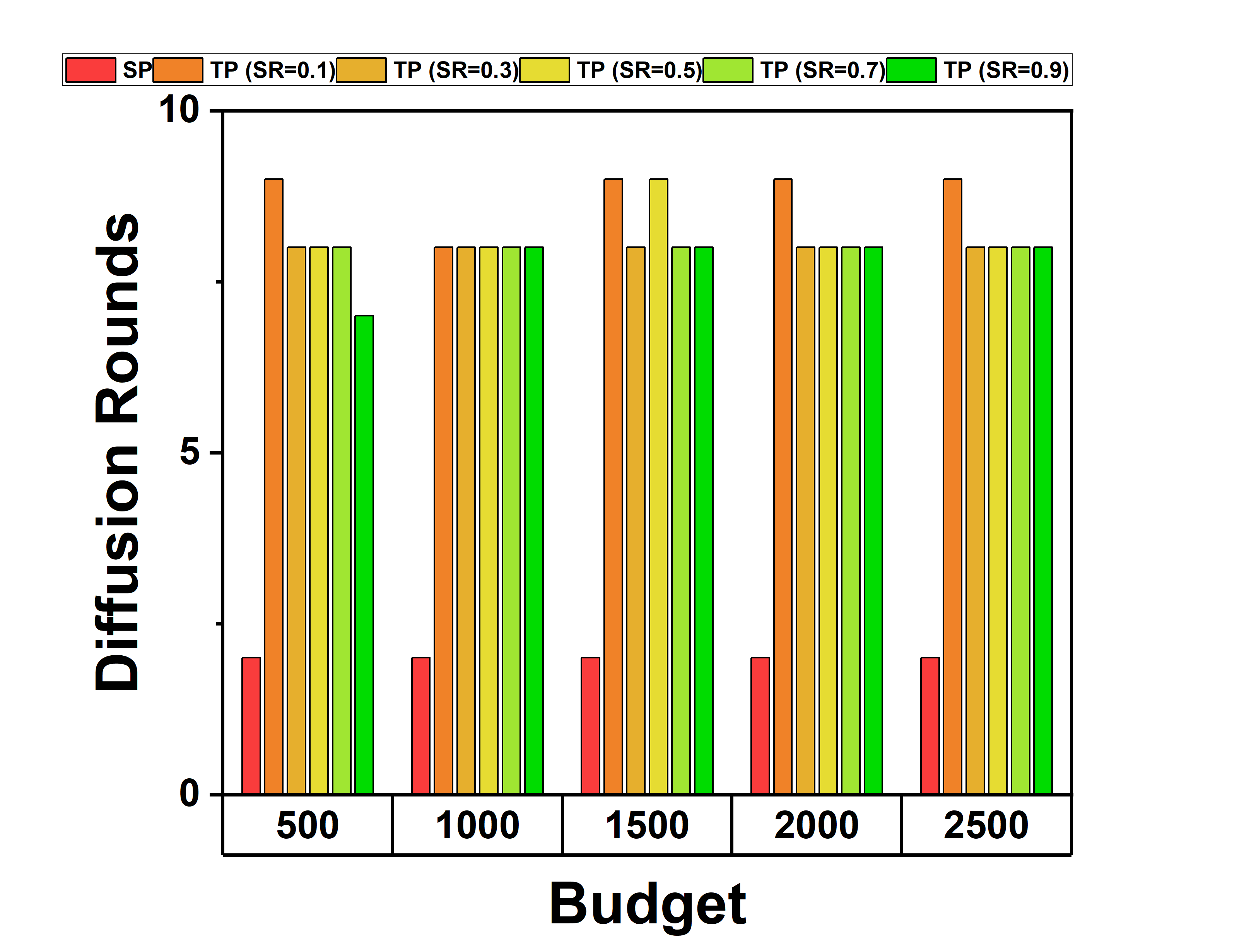}
        \caption{Timestep 6}
    \end{subfigure}
    \hfill
    \begin{subfigure}[t]{0.3\linewidth}
        \centering
        \includegraphics[width=\linewidth]{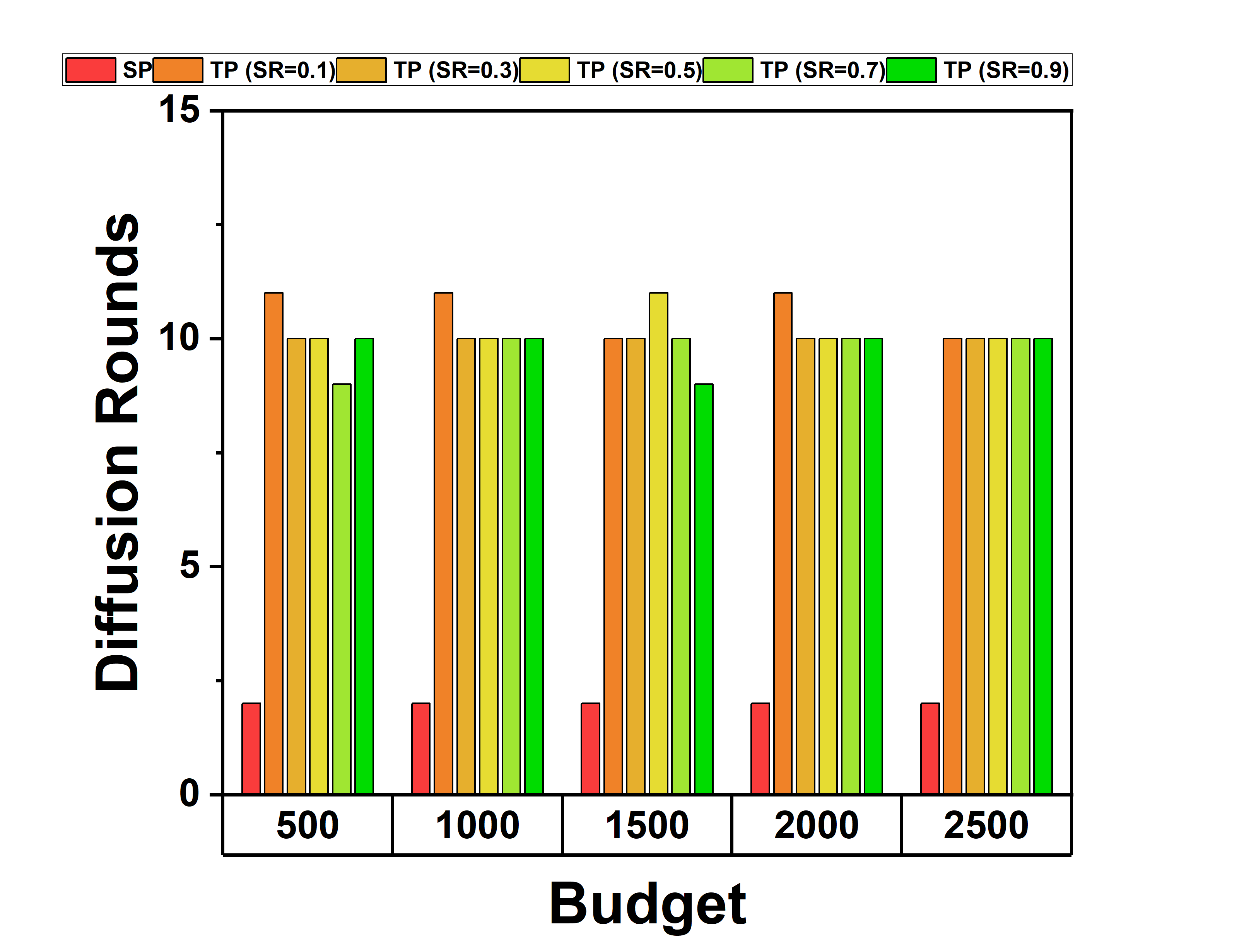}
        \caption{Timestep 8}
    \end{subfigure}
    \hfill
    \begin{subfigure}[t]{0.3\linewidth}
        \centering
        \includegraphics[width=\linewidth]{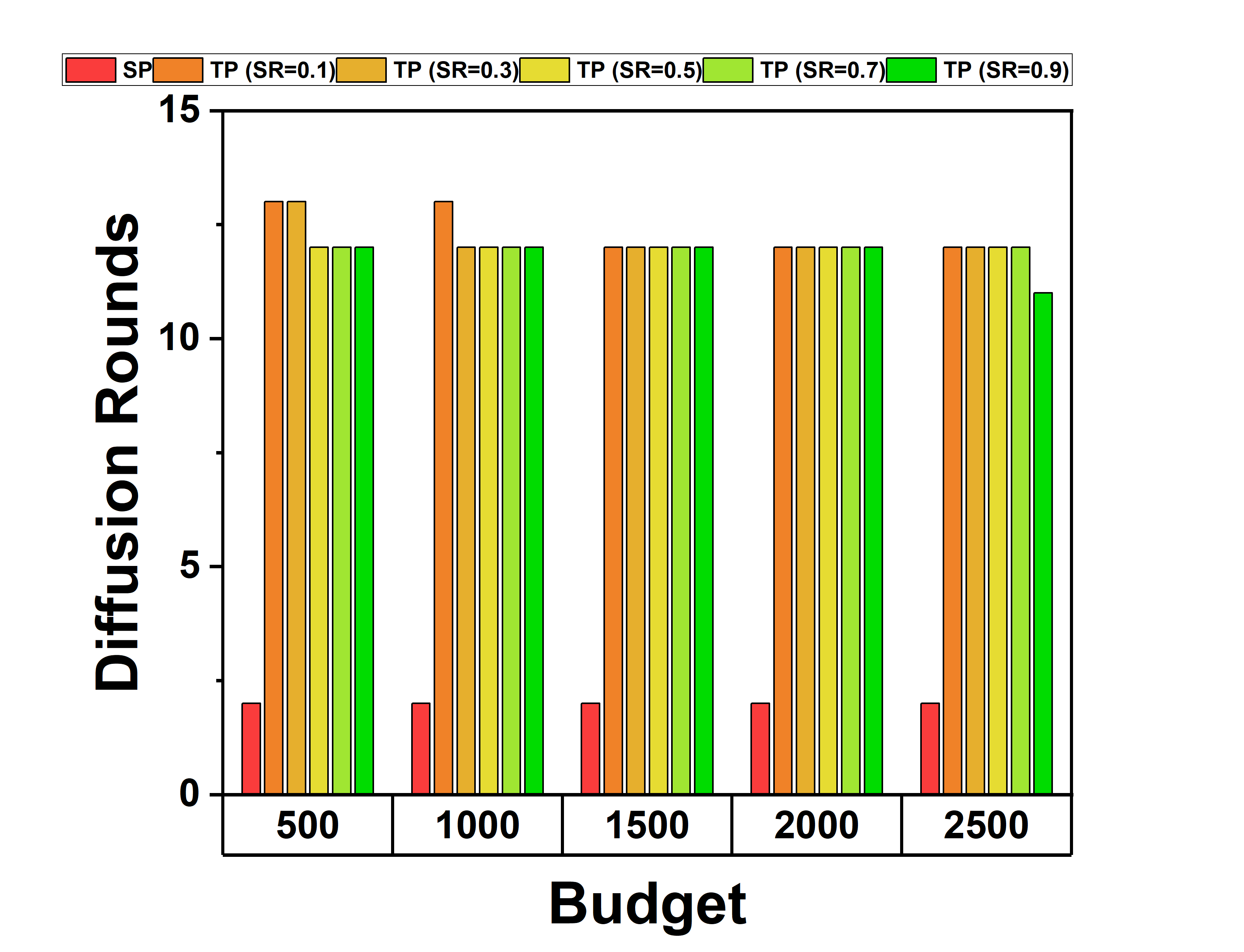}
        \caption{Timestep 10}
    \end{subfigure}

    \caption{Diffusion Rounds in Single Phase Vs. Two Phase (Double Greedy Algorithm, \textit{LM} Dataset, Probability Setting - Trivalency)}
    \label{RQ5LM_T7}
\end{figure}

\begin{figure}[htbp]
    \centering
    \begin{subfigure}[t]{0.3\linewidth}
        \centering
        \includegraphics[width=\linewidth]{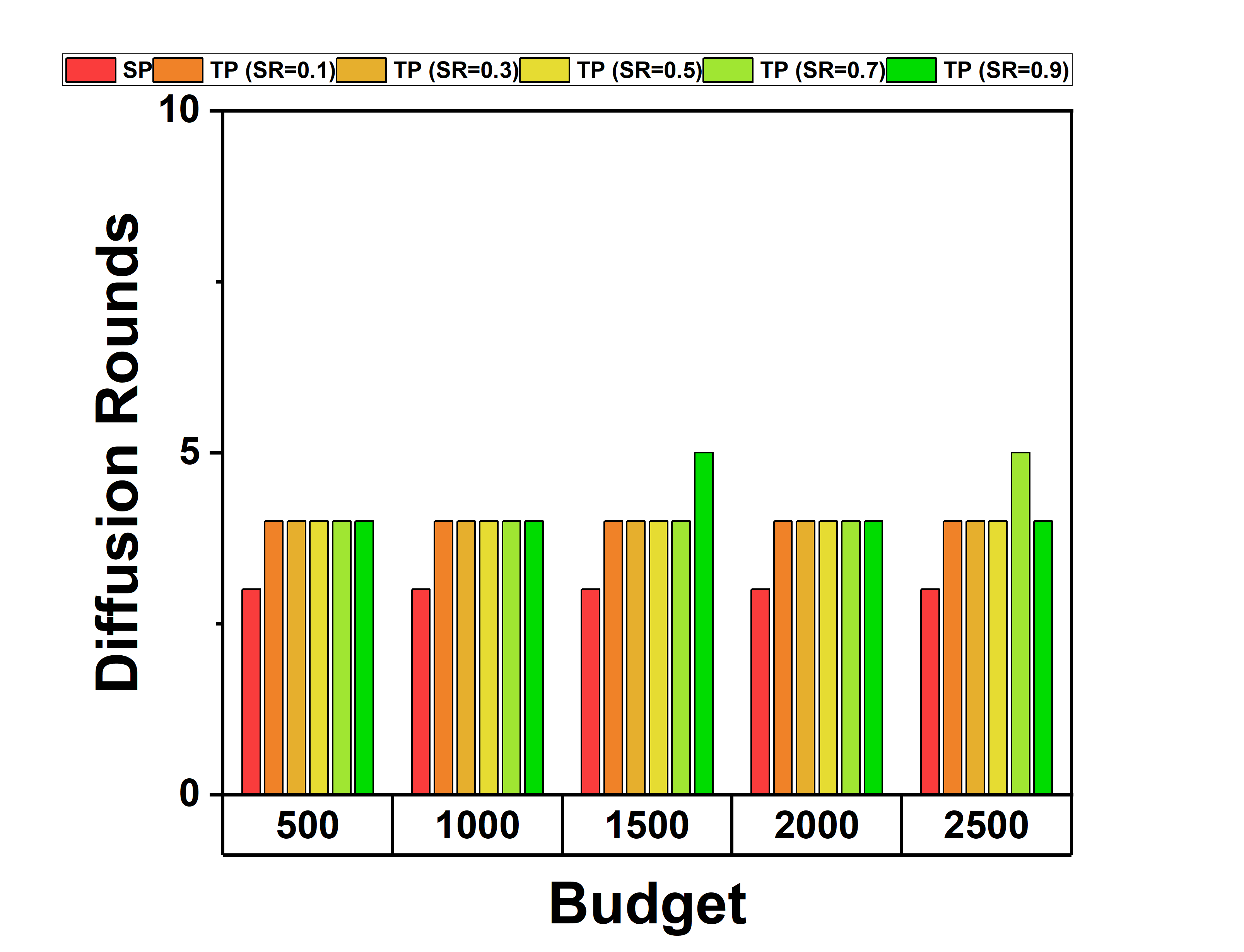}
        \caption{Timestep 2}
    \end{subfigure}
    \hspace{0.05\linewidth}
    \begin{subfigure}[t]{0.3\linewidth}
        \centering
        \includegraphics[width=\linewidth]{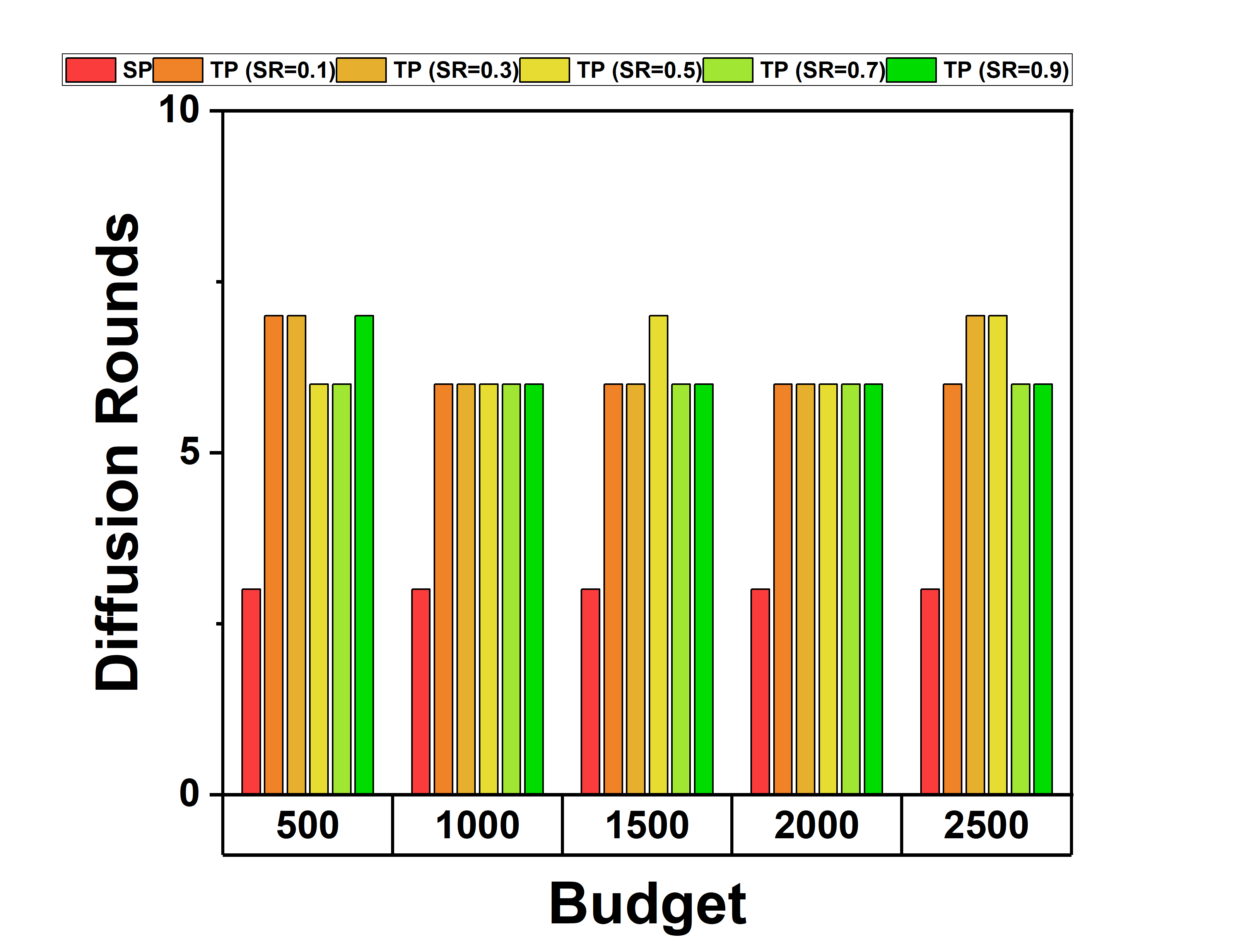}
        \caption{Timestep 4}
    \end{subfigure}

    \vspace{0.5cm}

    \begin{subfigure}[t]{0.3\linewidth}
        \centering
        \includegraphics[width=\linewidth]{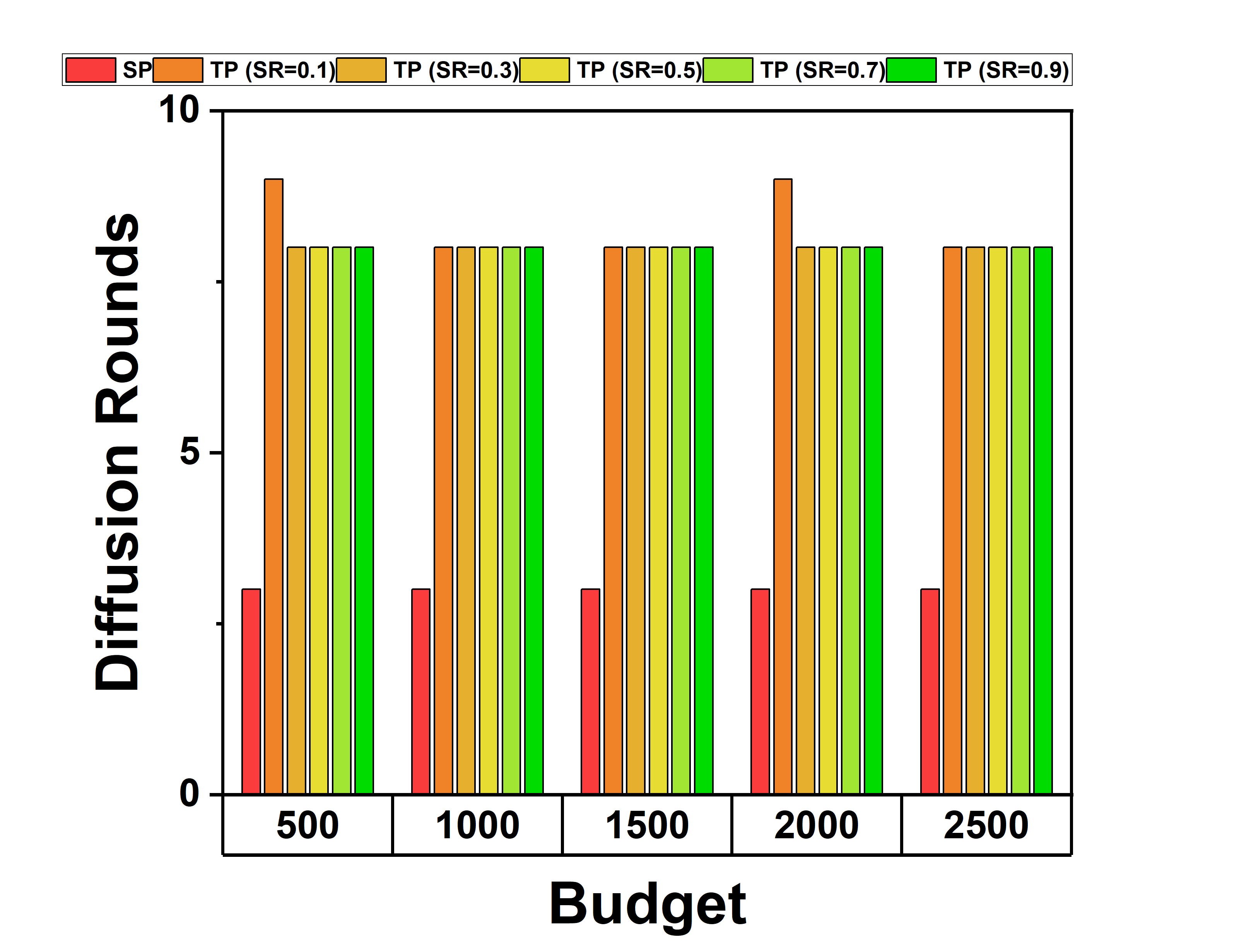}
        \caption{Timestep 6}
    \end{subfigure}
    \hfill
    \begin{subfigure}[t]{0.3\linewidth}
        \centering
        \includegraphics[width=\linewidth]{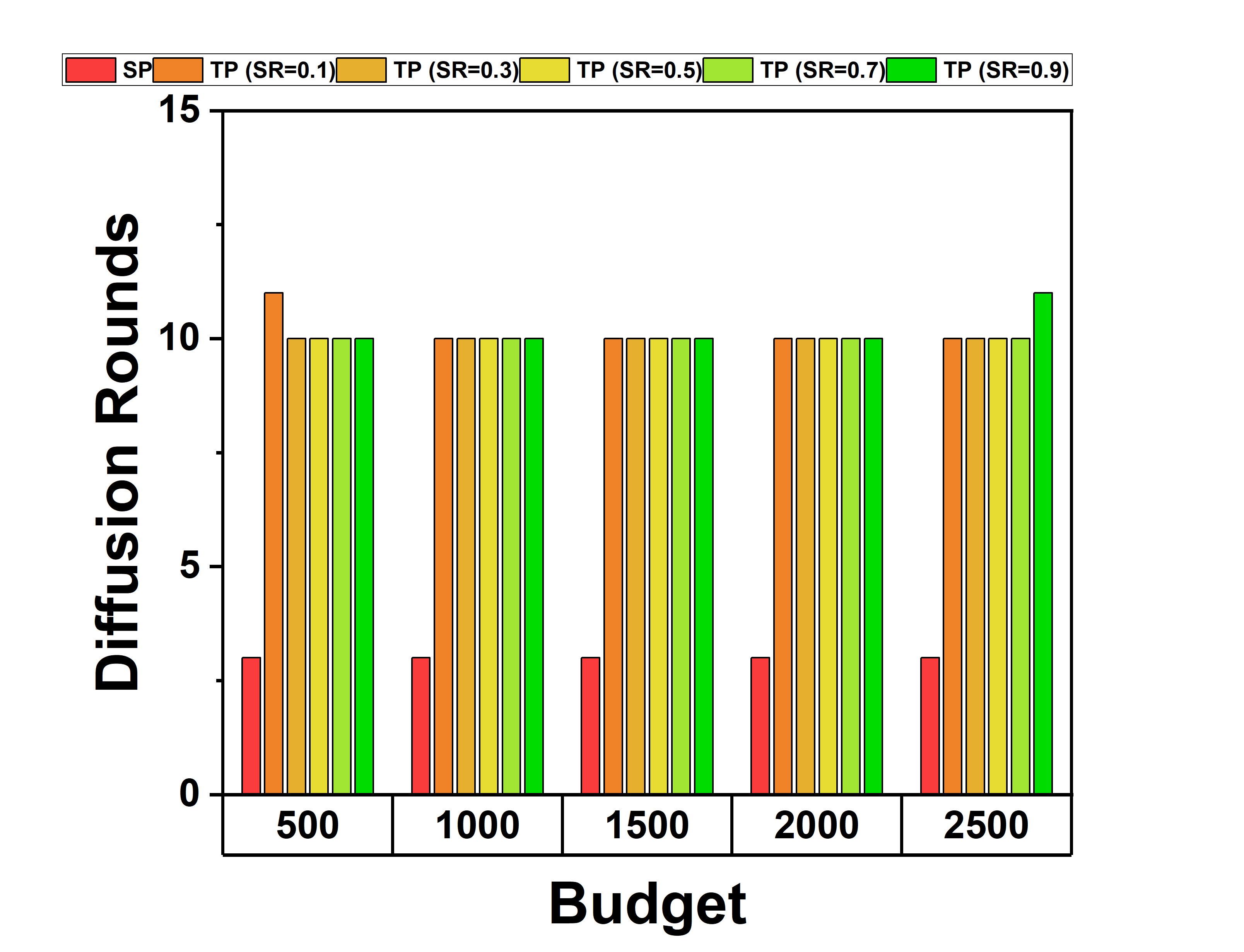}
        \caption{Timestep 8}
    \end{subfigure}
    \hfill
    \begin{subfigure}[t]{0.3\linewidth}
        \centering
        \includegraphics[width=\linewidth]{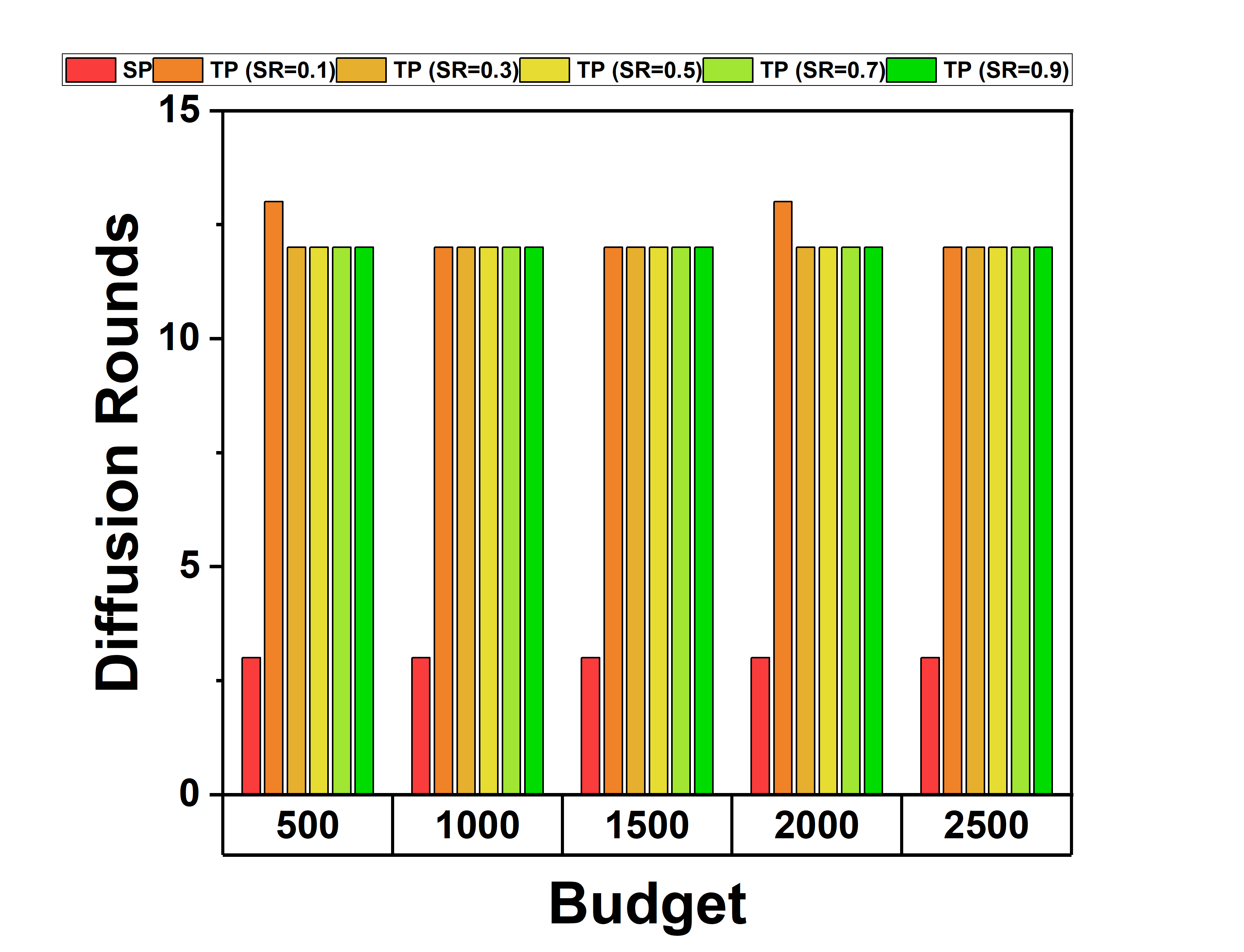}
        \caption{Timestep 10}
    \end{subfigure}

    \caption{Diffusion Rounds in Single Phase Vs. Two Phase (Stochastic Greedy Algorithm, \textit{LM} Dataset, Probability Setting - Trivalency)}
    \label{RQ5LM_T8}
\end{figure}


\begin{figure}[htbp]
    \centering

    \begin{subfigure}[t]{0.3\linewidth}
        \centering
        \includegraphics[width=\linewidth]{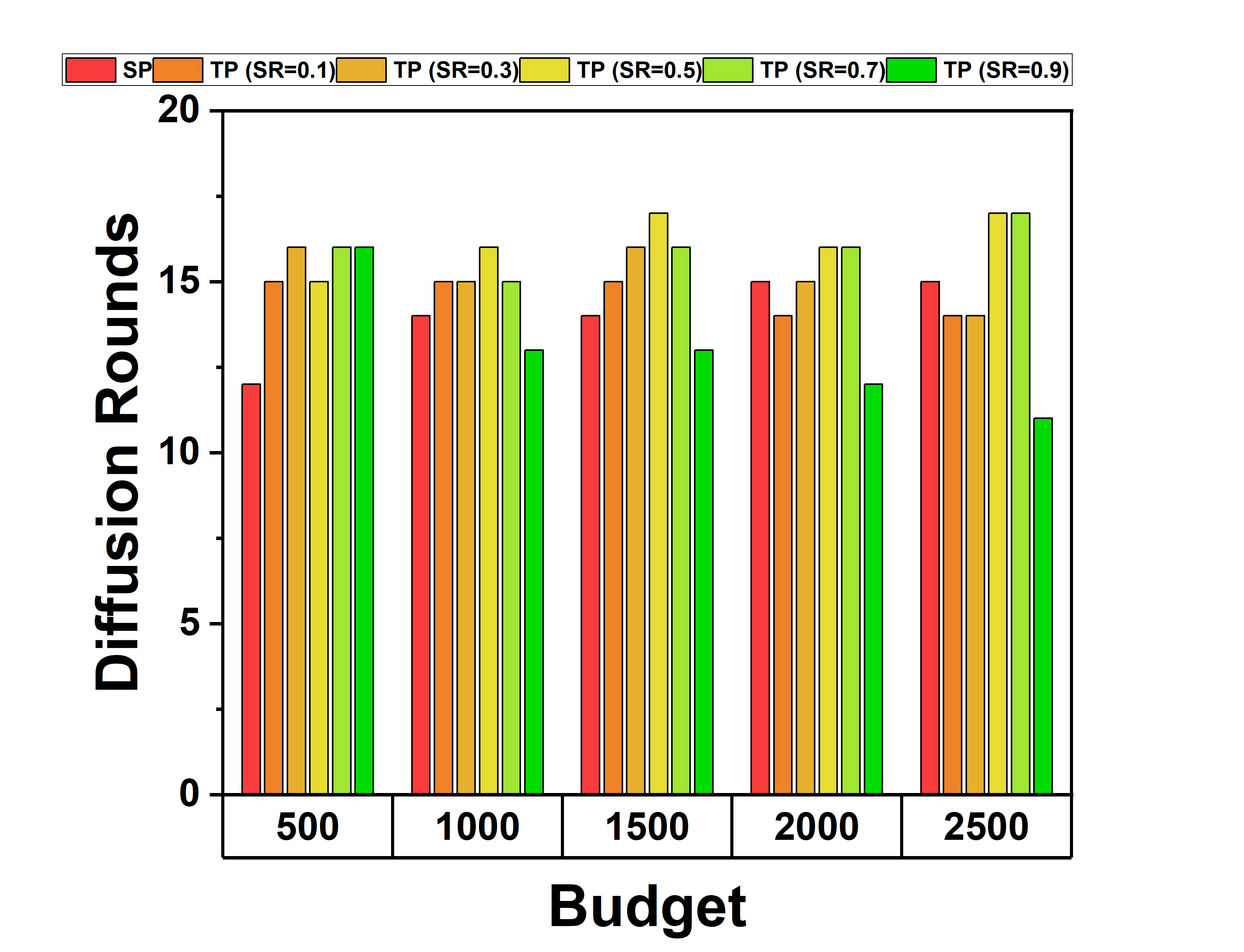}
        \caption{Timestep 2}
    \end{subfigure}
    \hspace{0.05\linewidth}
    \begin{subfigure}[t]{0.3\linewidth}
        \centering
        \includegraphics[width=\linewidth]{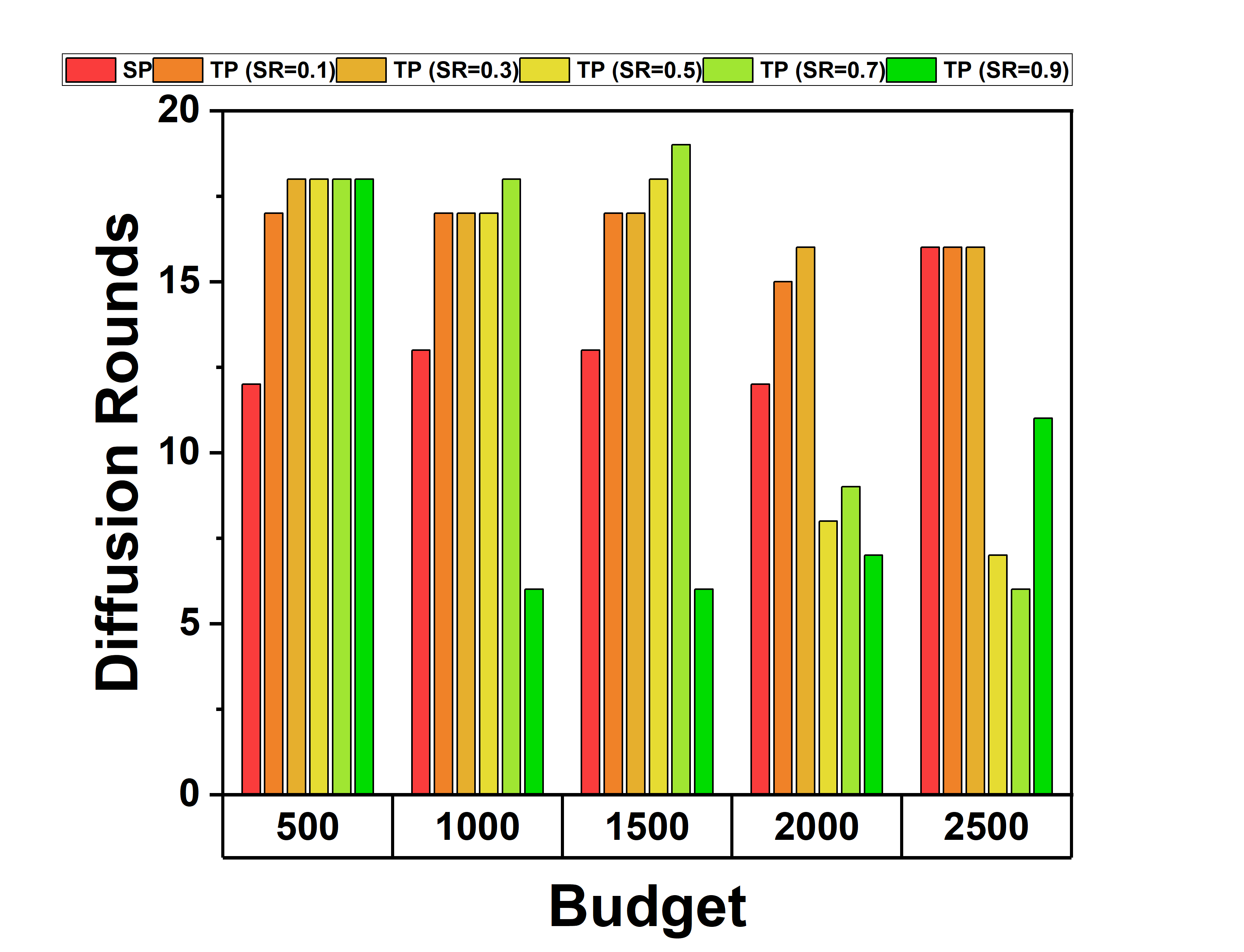}
        \caption{Timestep 4}
    \end{subfigure}

    \vspace{0.5cm}

    \begin{subfigure}[t]{0.3\linewidth}
        \centering
        \includegraphics[width=\linewidth]{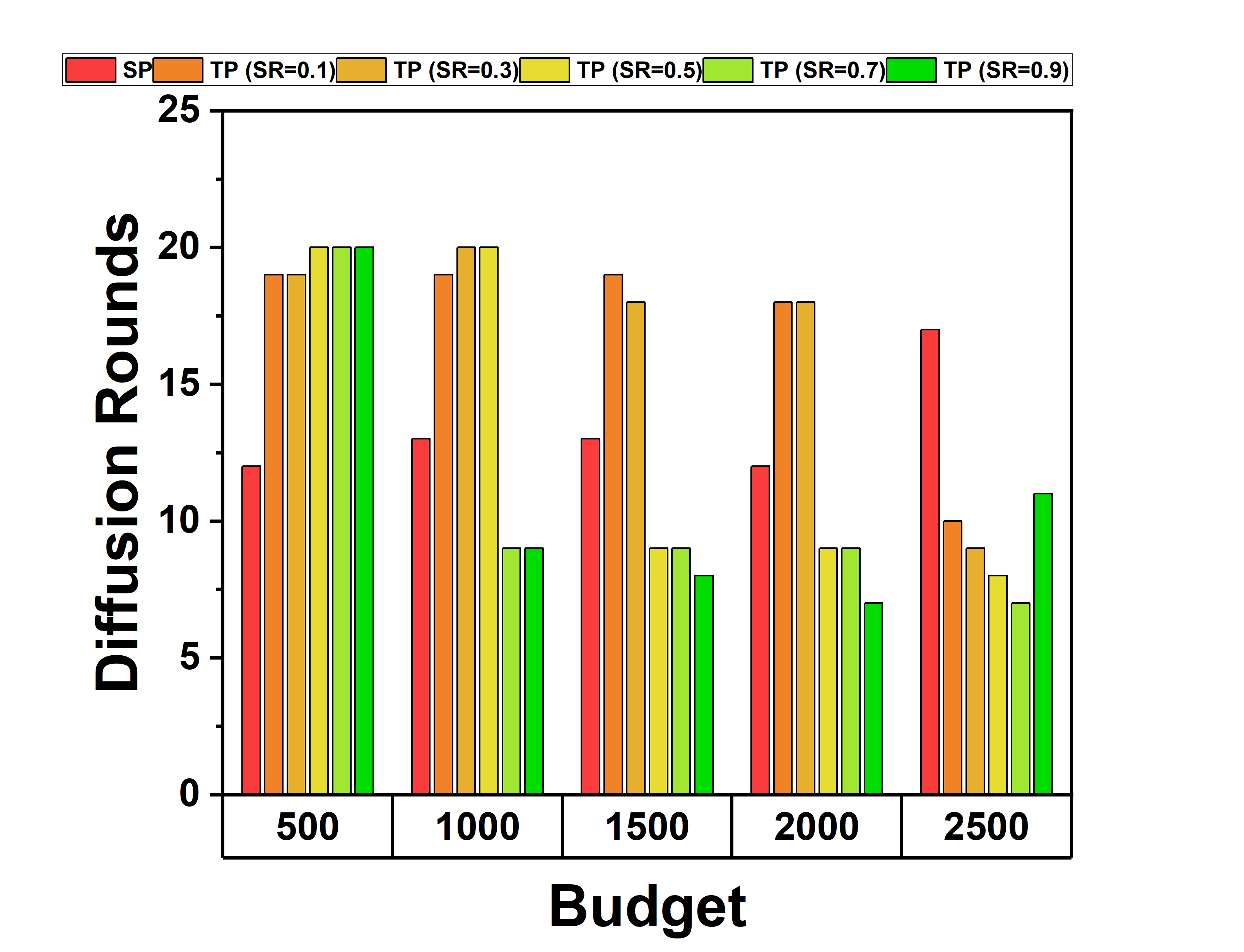}
        \caption{Timestep 6}
    \end{subfigure}
    \hfill
    \begin{subfigure}[t]{0.3\linewidth}
        \centering
        \includegraphics[width=\linewidth]{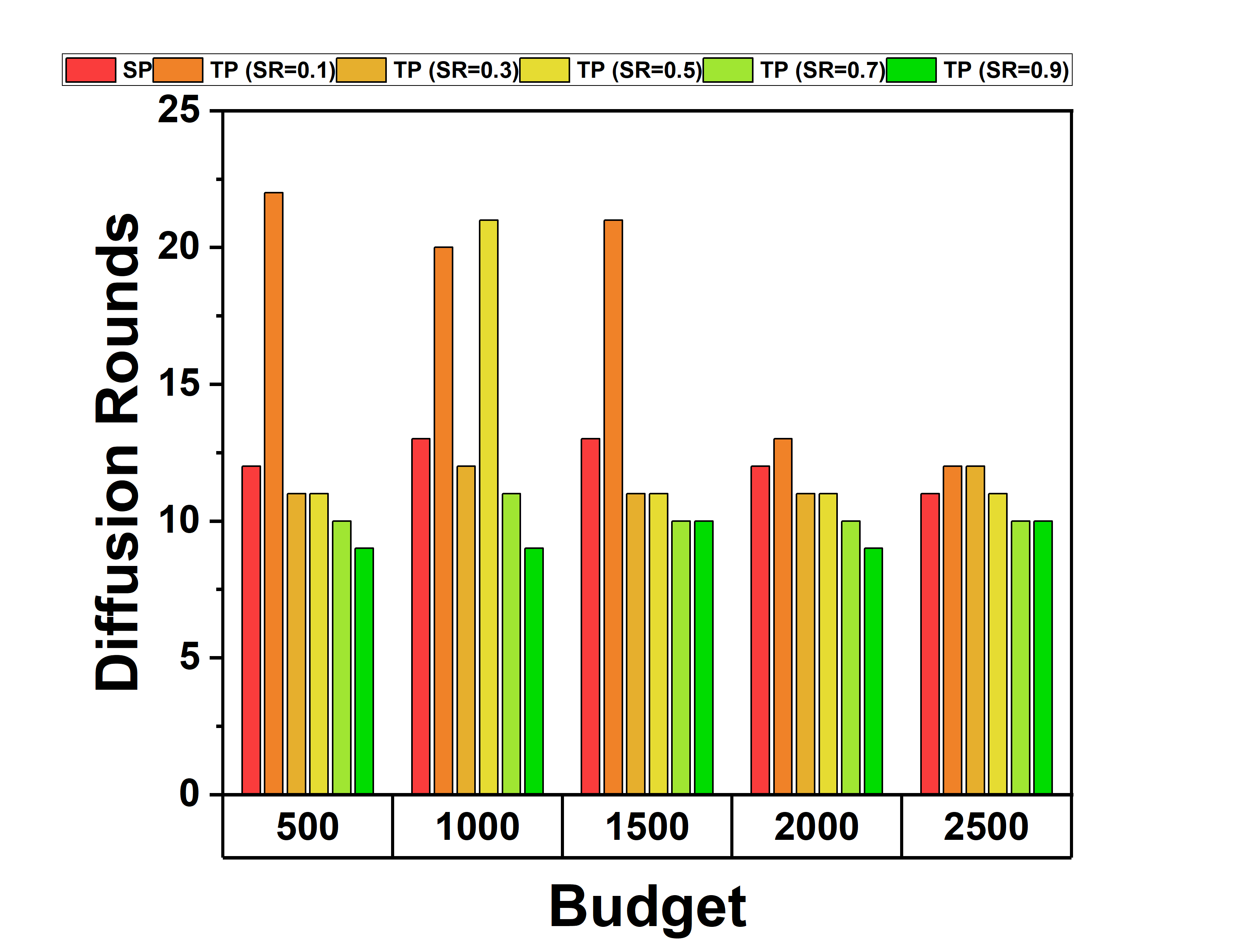}
        \caption{Timestep 8}
    \end{subfigure}
    \hfill
    \begin{subfigure}[t]{0.3\linewidth}
        \centering
        \includegraphics[width=\linewidth]{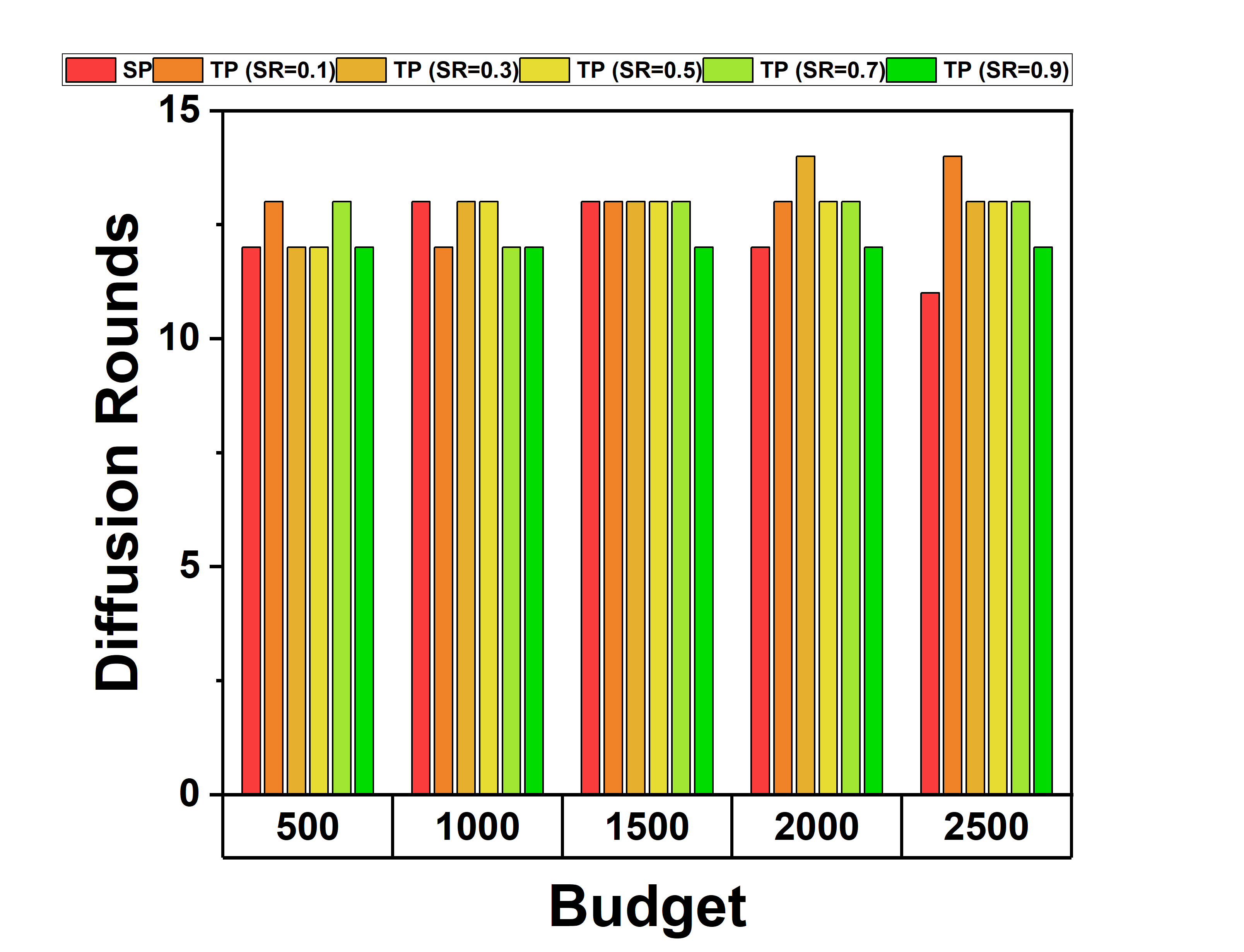}
        \caption{Timestep 10}
    \end{subfigure}

    \caption{Diffusion Rounds in Single Phase Vs. Two Phase (Random Algorithm, \textit{Email-Eu-Core} Dataset, Probability Setting - Trivalency)}
    \label{RQ5_T1}
\end{figure}

\begin{figure}[htbp]
    \centering

    \begin{subfigure}[t]{0.3\linewidth}
        \centering
        \includegraphics[width=\linewidth]{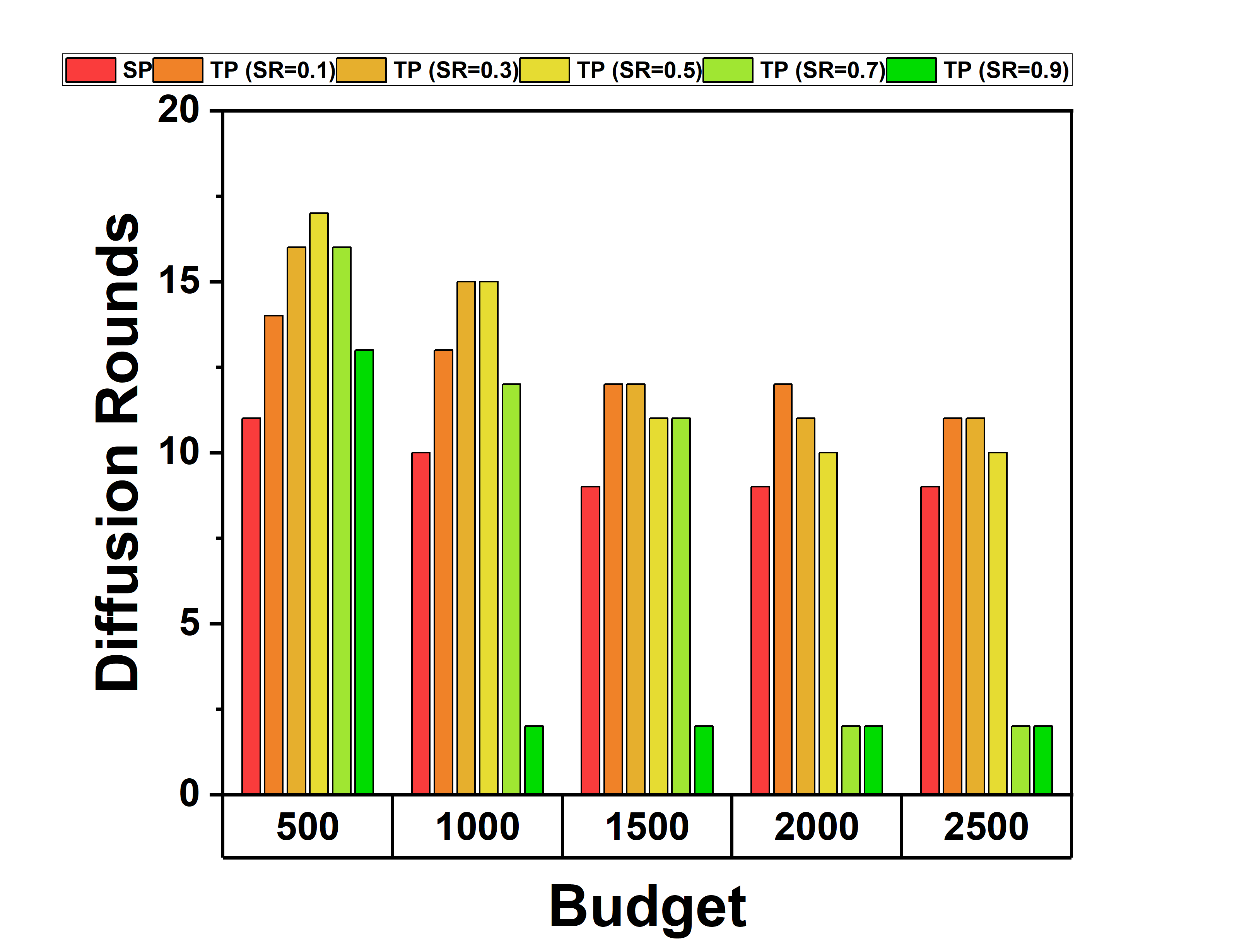}
        \caption{Timestep 2}
    \end{subfigure}
    \hspace{0.05\linewidth}
    \begin{subfigure}[t]{0.3\linewidth}
        \centering
        \includegraphics[width=\linewidth]{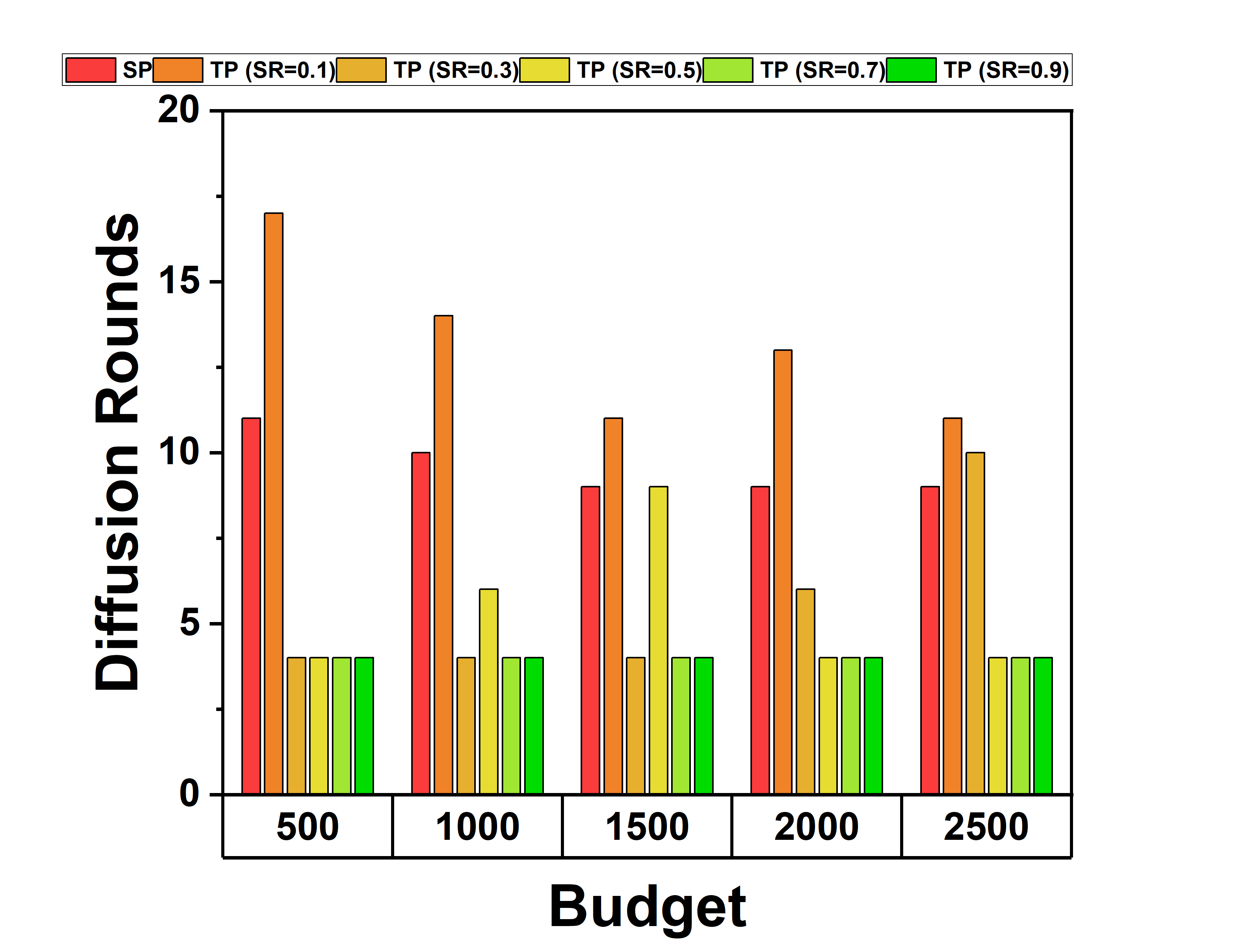}
        \caption{Timestep 4}
    \end{subfigure}

    \vspace{0.5cm}

    \begin{subfigure}[t]{0.3\linewidth}
        \centering
        \includegraphics[width=\linewidth]{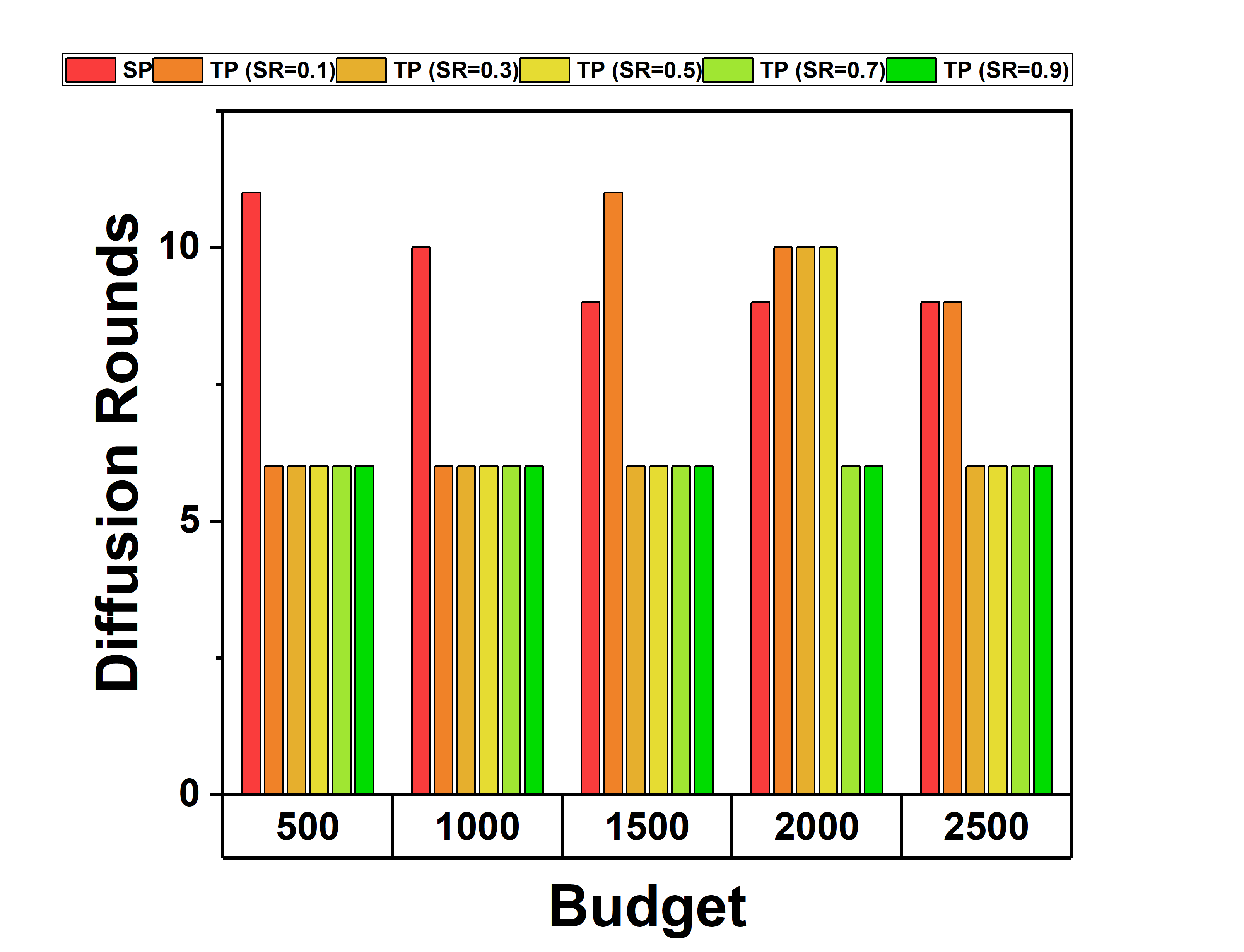}
        \caption{Timestep 6}
    \end{subfigure}
    \hfill
    \begin{subfigure}[t]{0.3\linewidth}
        \centering
        \includegraphics[width=\linewidth]{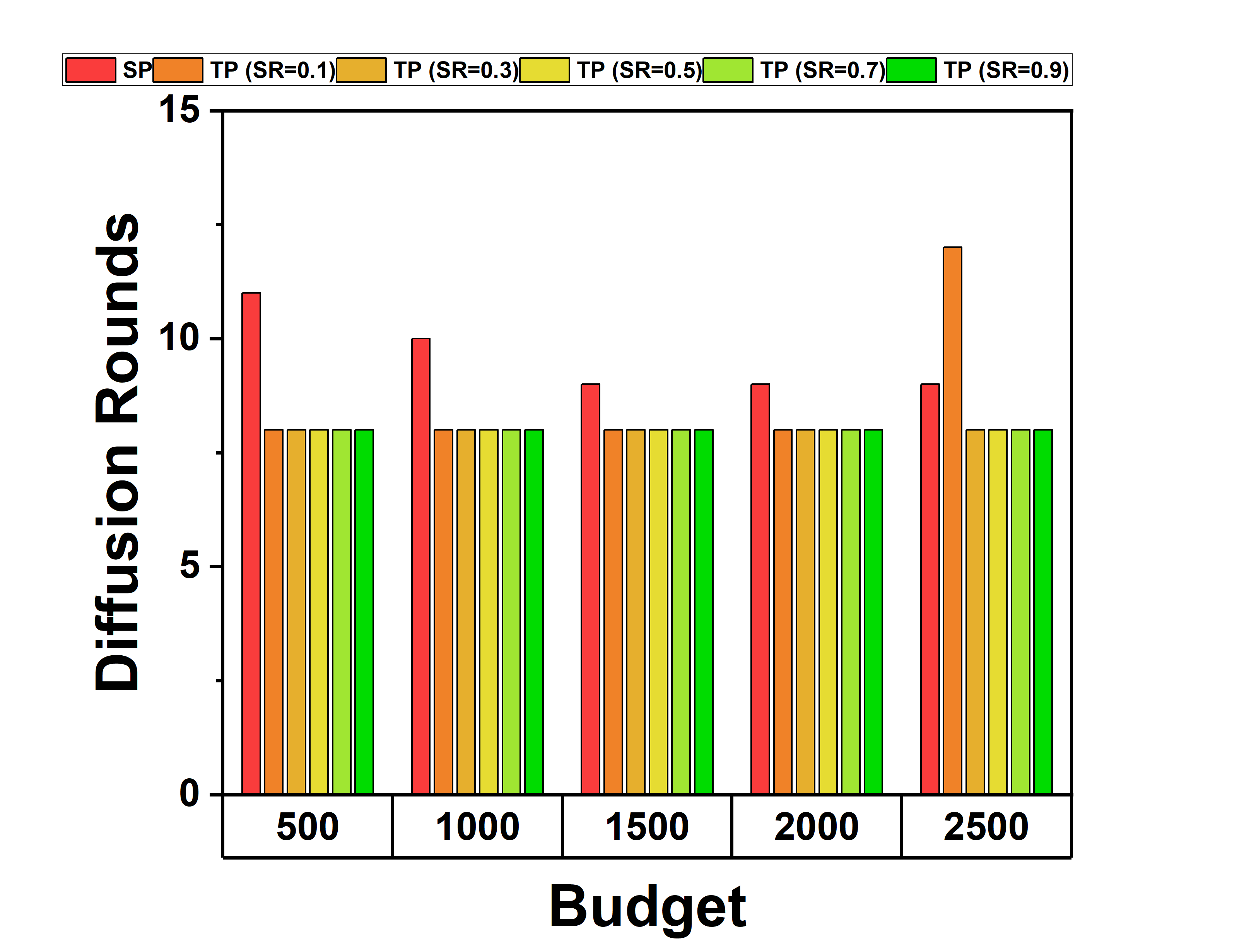}
        \caption{Timestep 8}
    \end{subfigure}
    \hfill
    \begin{subfigure}[t]{0.3\linewidth}
        \centering
        \includegraphics[width=\linewidth]{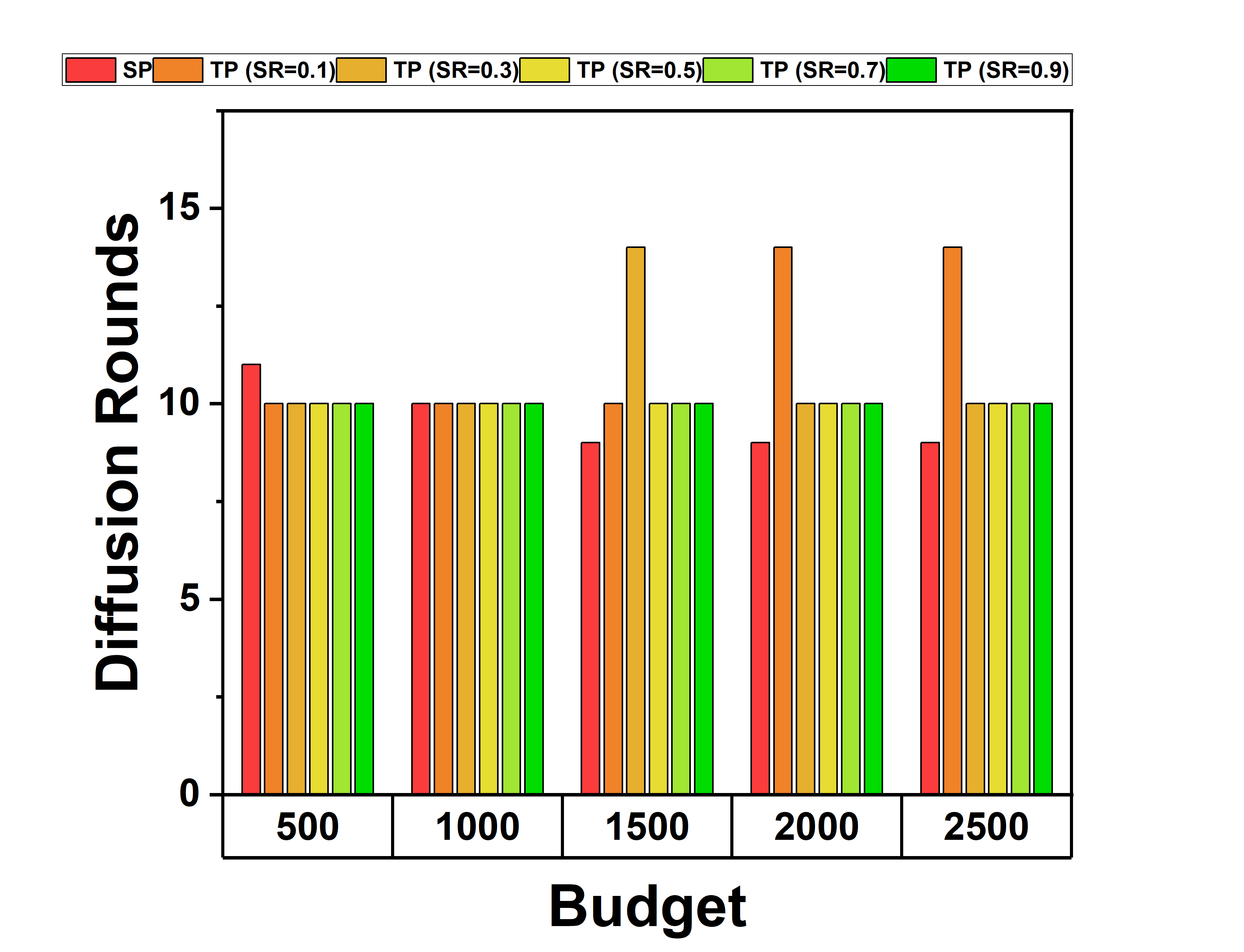}
        \caption{Timestep 10}
    \end{subfigure}

    \caption{Diffusion Rounds in Single Phase Vs. Two Phase (High Degree Algorithm, \textit{Email-Eu-Core} Dataset, Probability Setting - Trivalency)}
    \label{RQ5_T2}
\end{figure}

\begin{figure}[htbp]
    \centering

    \begin{subfigure}[t]{0.3\linewidth}
        \centering
        \includegraphics[width=\linewidth]{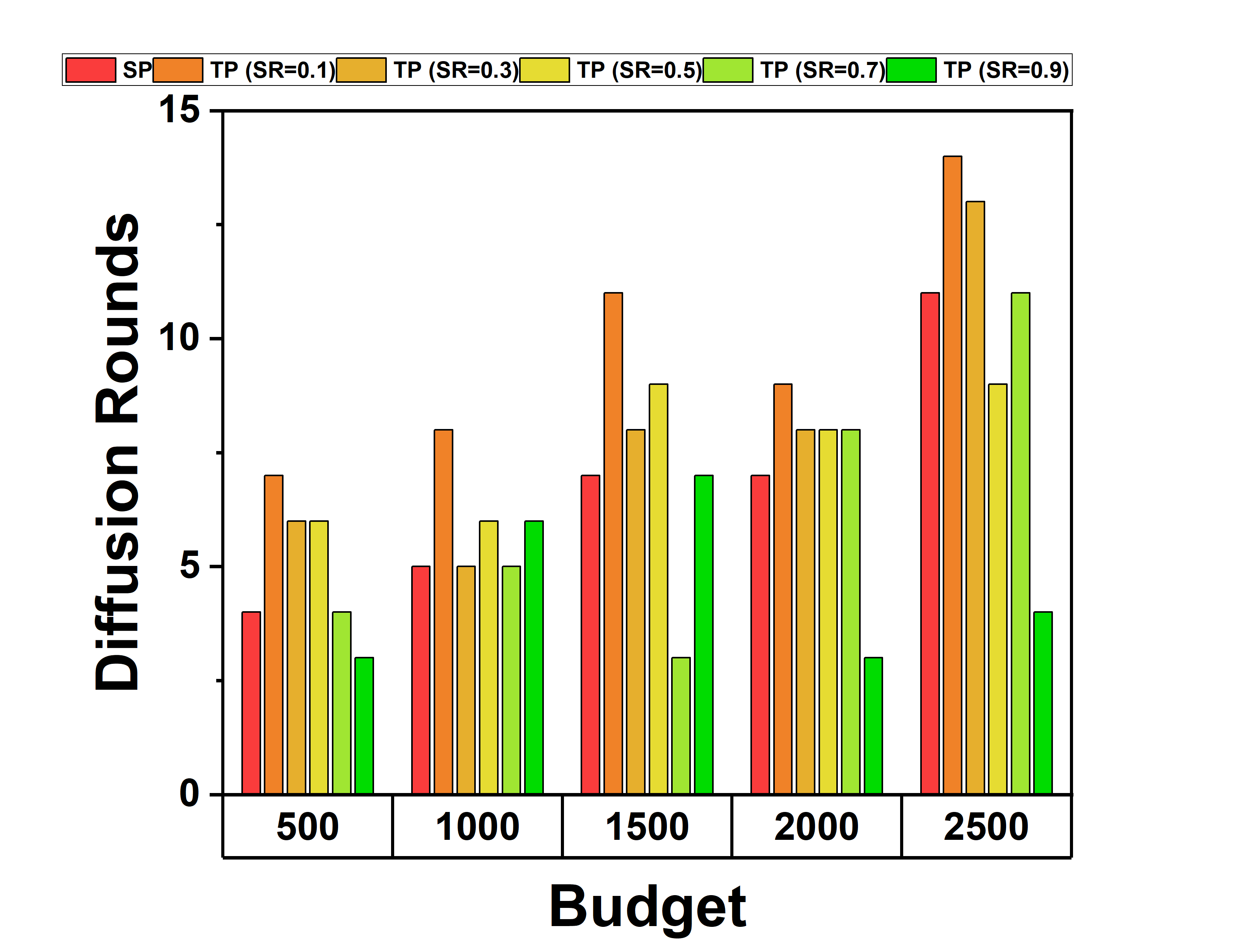}
        \caption{Timestep 2}
    \end{subfigure}
    \hspace{0.05\linewidth}
    \begin{subfigure}[t]{0.3\linewidth}
        \centering
        \includegraphics[width=\linewidth]{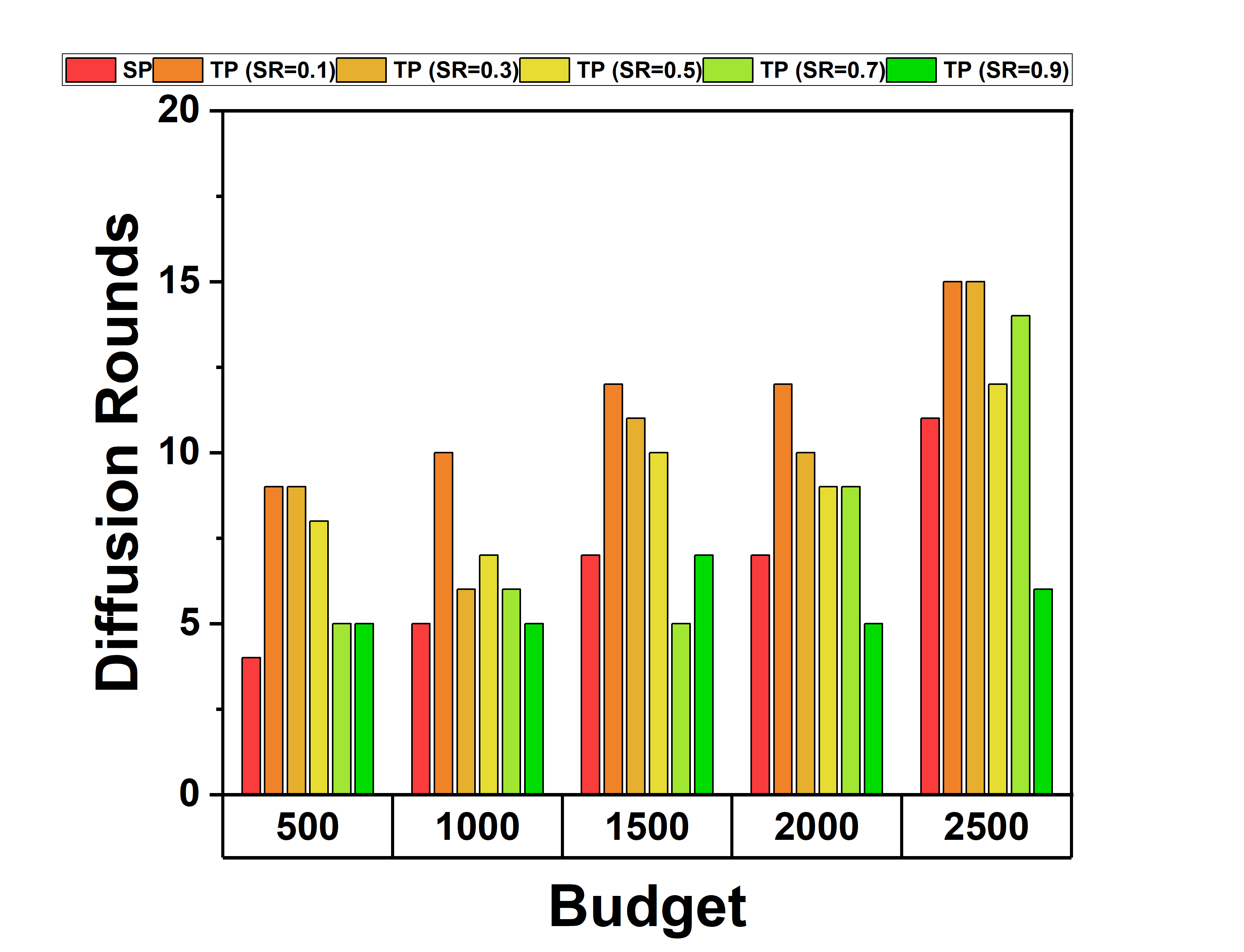}
        \caption{Timestep 4}
    \end{subfigure}

    \vspace{0.5cm}

    \begin{subfigure}[t]{0.3\linewidth}
        \centering
        \includegraphics[width=\linewidth]{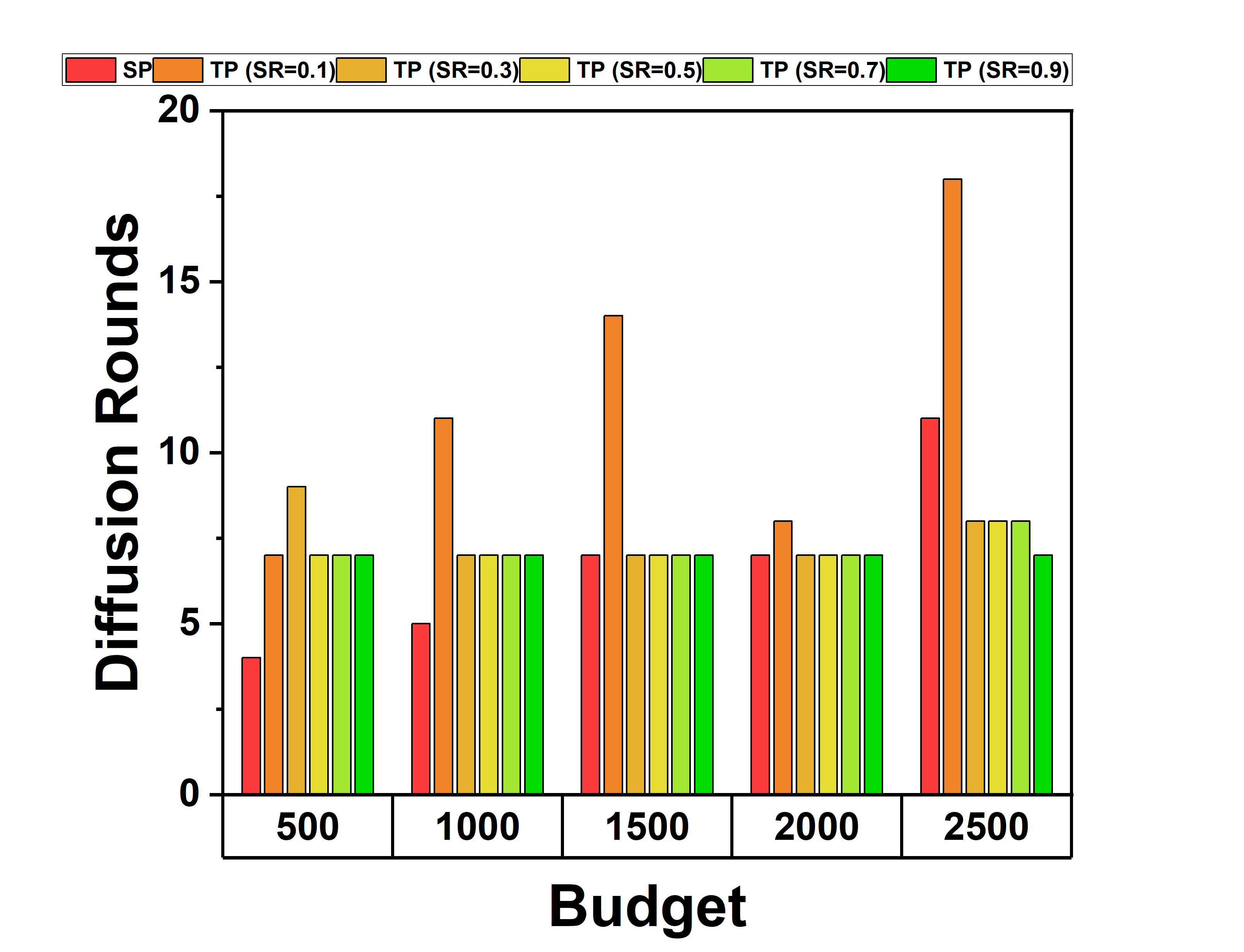}
        \caption{Timestep 6}
    \end{subfigure}
    \hfill
    \begin{subfigure}[t]{0.3\linewidth}
        \centering
        \includegraphics[width=\linewidth]{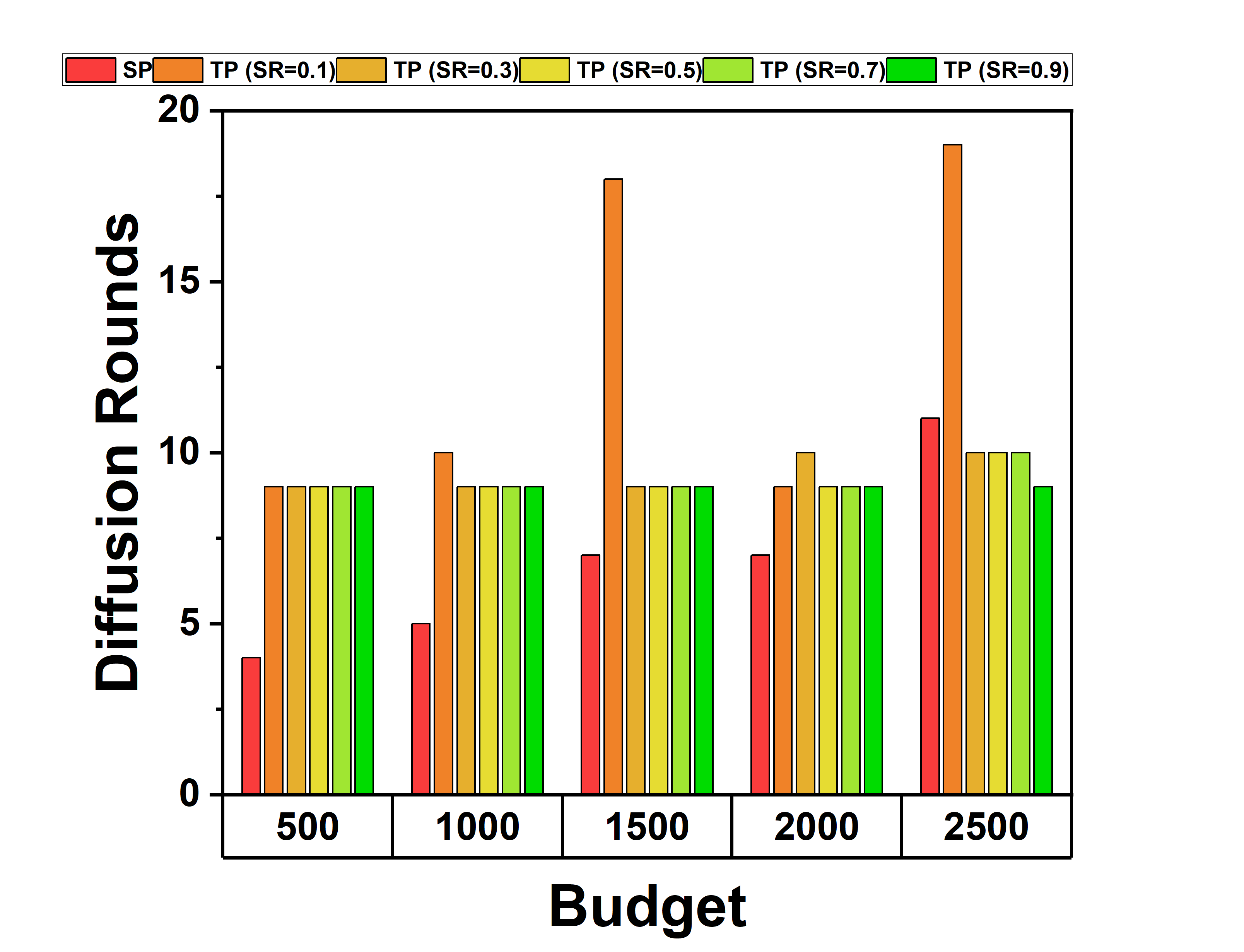}
        \caption{Timestep 8}
    \end{subfigure}
    \hfill
    \begin{subfigure}[t]{0.3\linewidth}
        \centering
        \includegraphics[width=\linewidth]{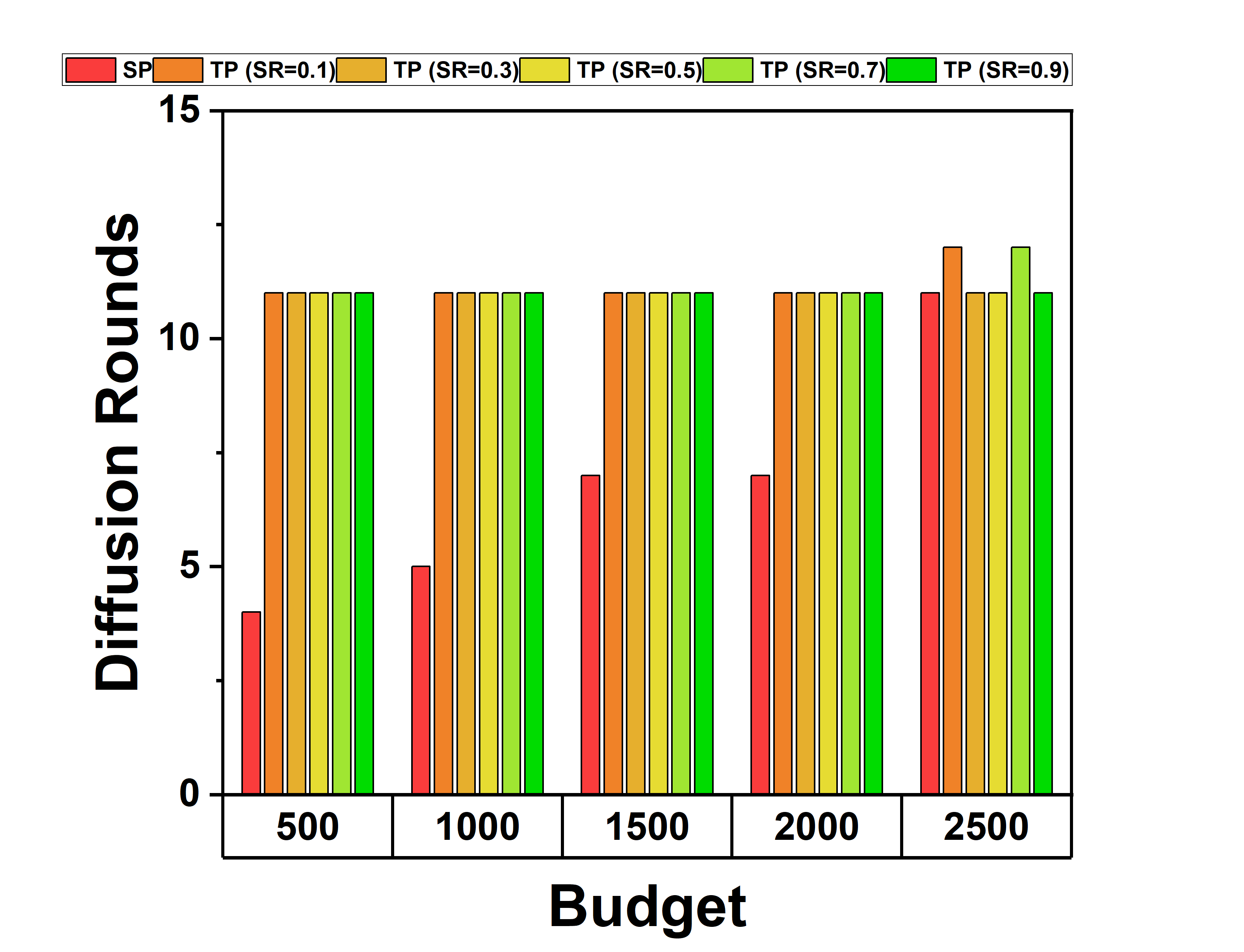}
        \caption{Timestep 10}
    \end{subfigure}

    \caption{Diffusion Rounds in Single Phase Vs. Two Phase (Clustering Coefficient Algorithm, \textit{Email-Eu-Core} Dataset, Probability Setting - Trivalency)}
    \label{RQ5_T3}
\end{figure}

\begin{figure}[htbp]
    \centering

    \begin{subfigure}[t]{0.3\linewidth}
        \centering
        \includegraphics[width=\linewidth]{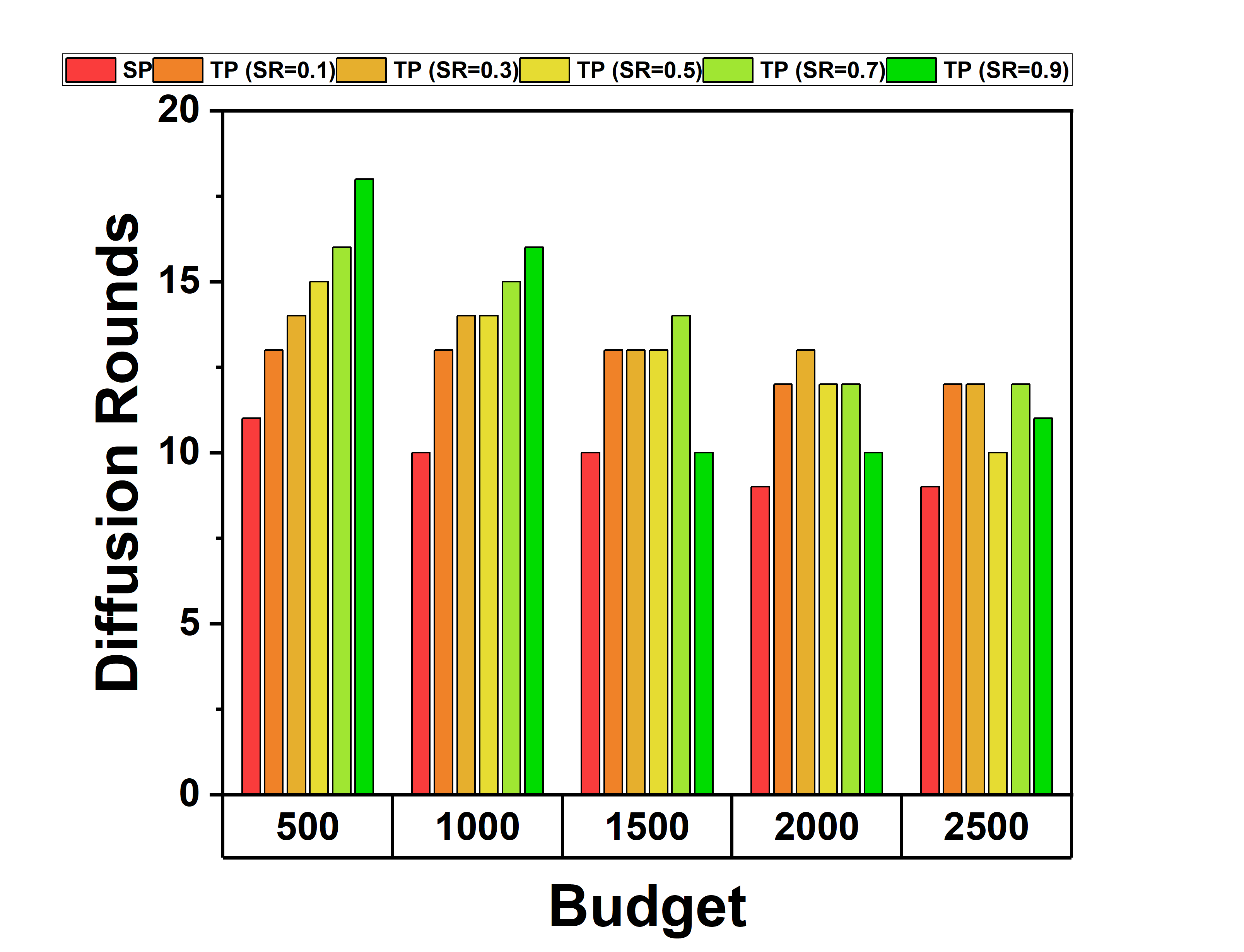}
        \caption{Timestep 2}
    \end{subfigure}
    \hspace{0.05\linewidth}
    \begin{subfigure}[t]{0.3\linewidth}
        \centering
        \includegraphics[width=\linewidth]{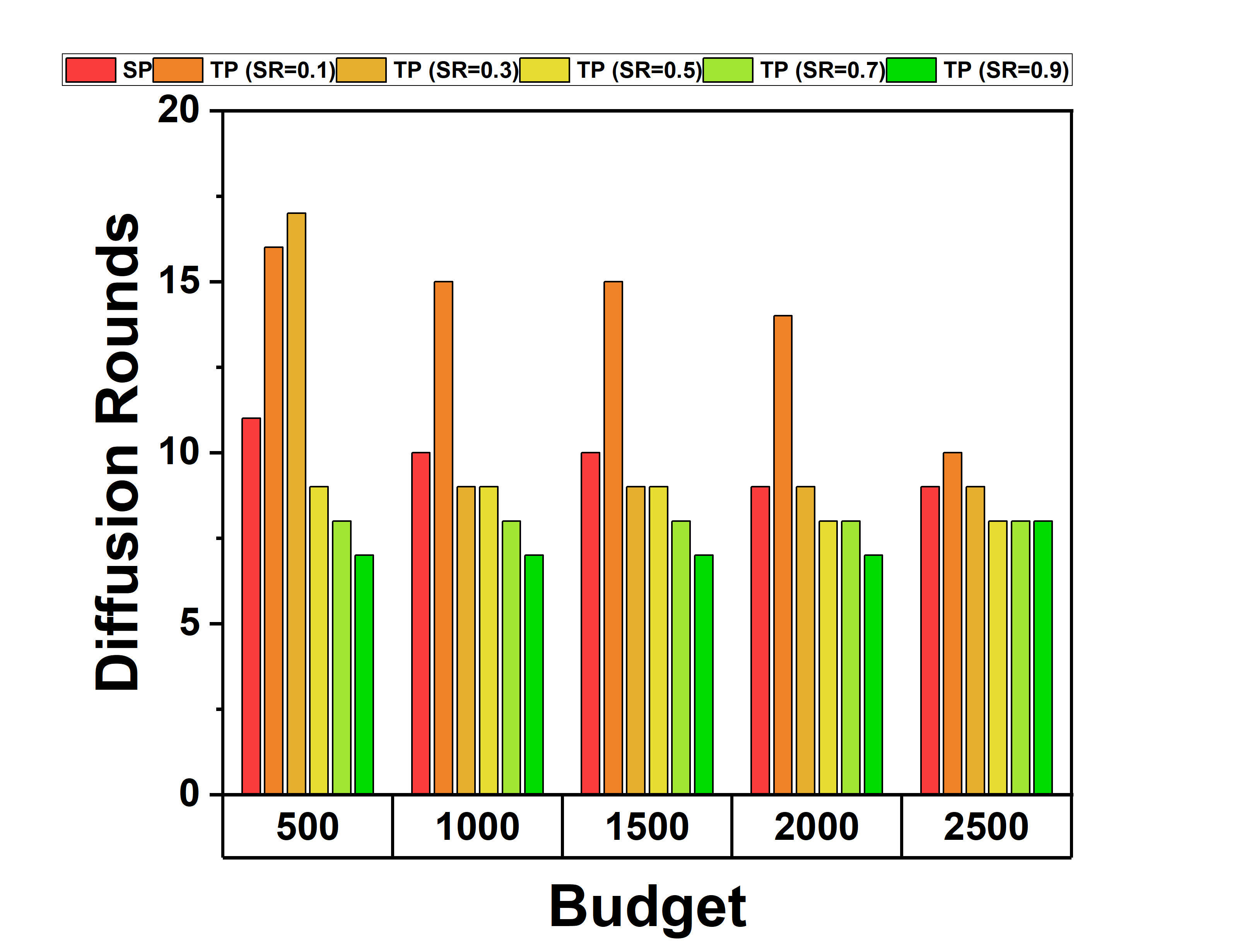}
        \caption{Timestep 4}
    \end{subfigure}

    \vspace{0.5cm}

    \begin{subfigure}[t]{0.3\linewidth}
        \centering
        \includegraphics[width=\linewidth]{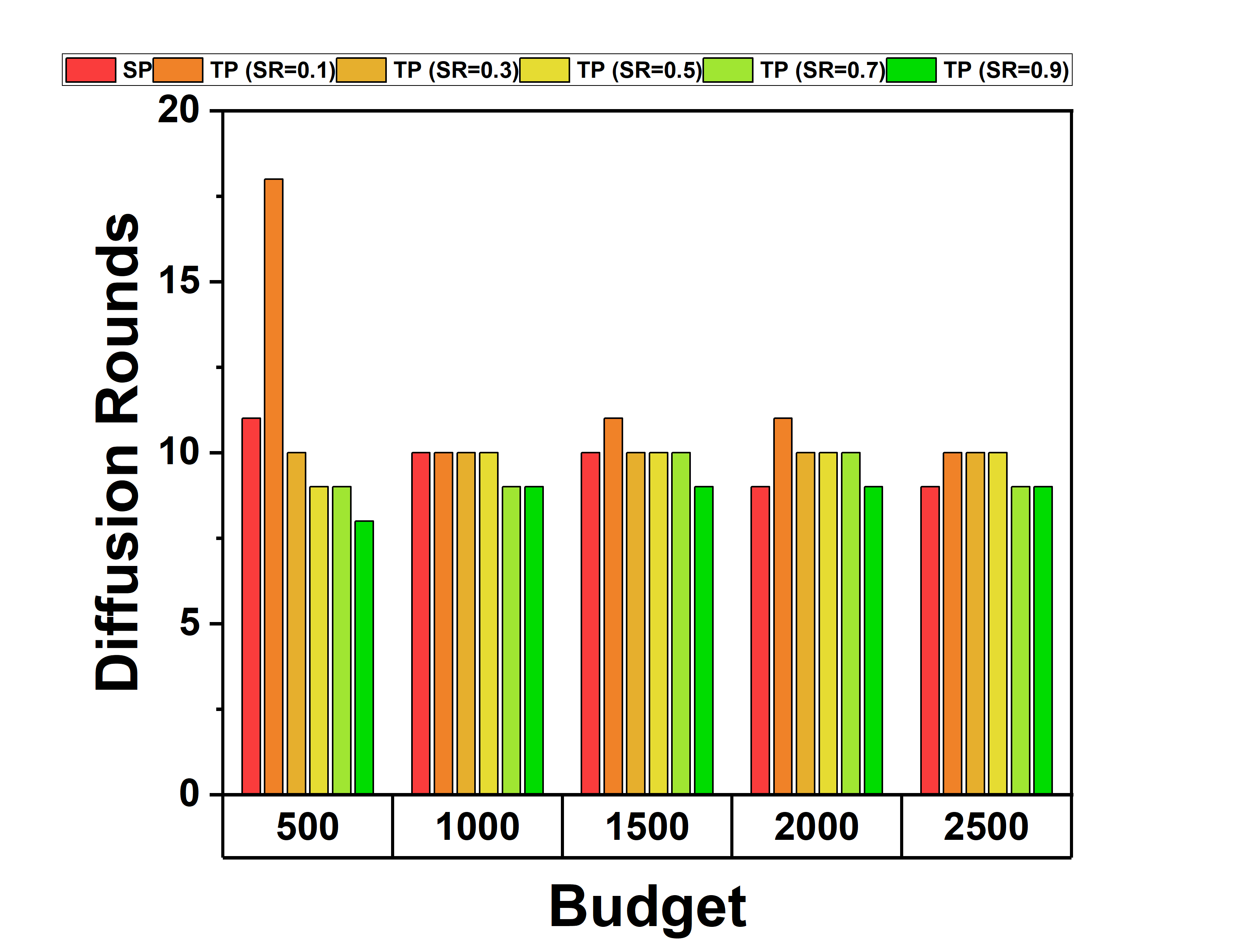}
        \caption{Timestep 6}
    \end{subfigure}
    \hfill
    \begin{subfigure}[t]{0.3\linewidth}
        \centering
        \includegraphics[width=\linewidth]{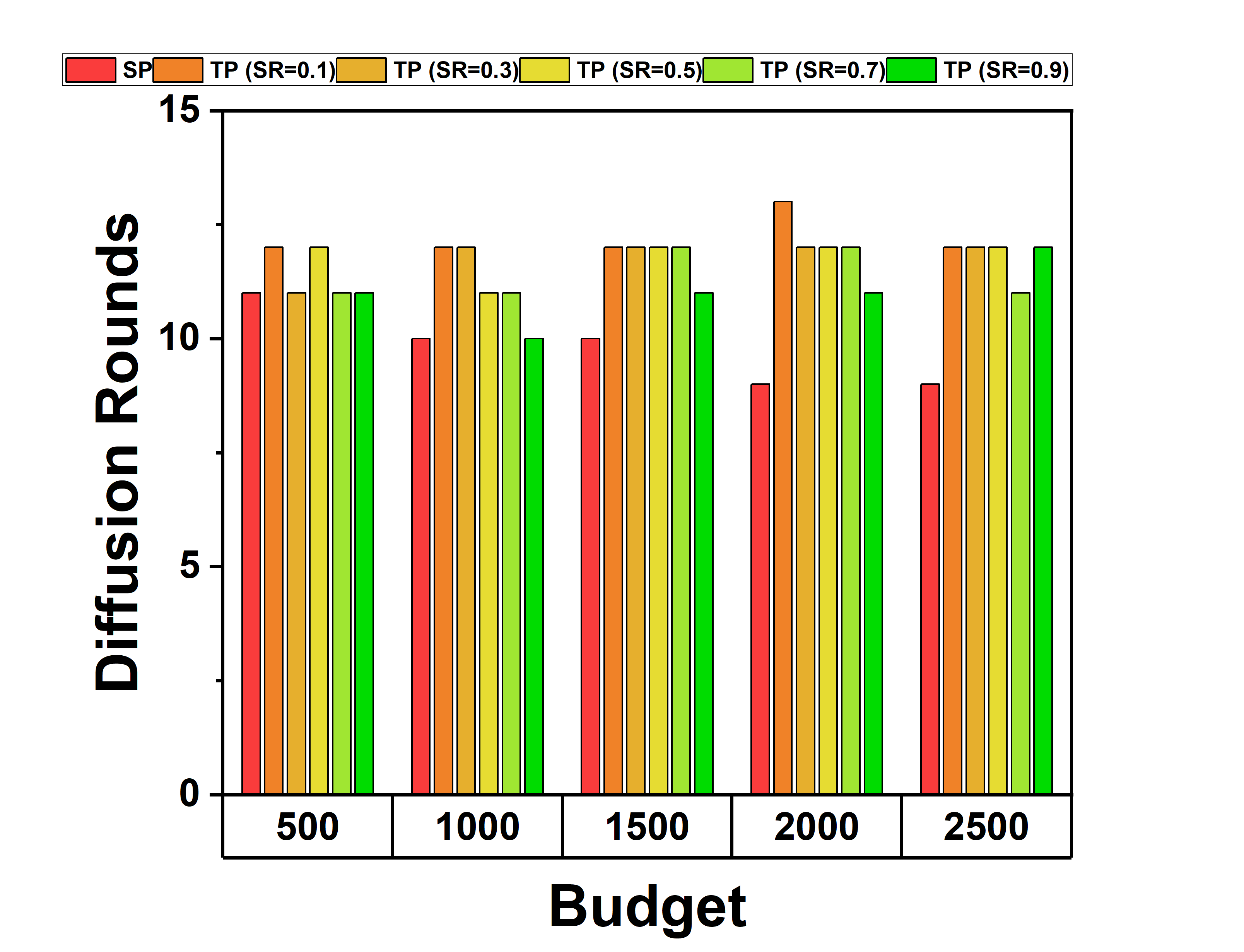}
        \caption{Timestep 8}
    \end{subfigure}
    \hfill
    \begin{subfigure}[t]{0.3\linewidth}
        \centering
        \includegraphics[width=\linewidth]{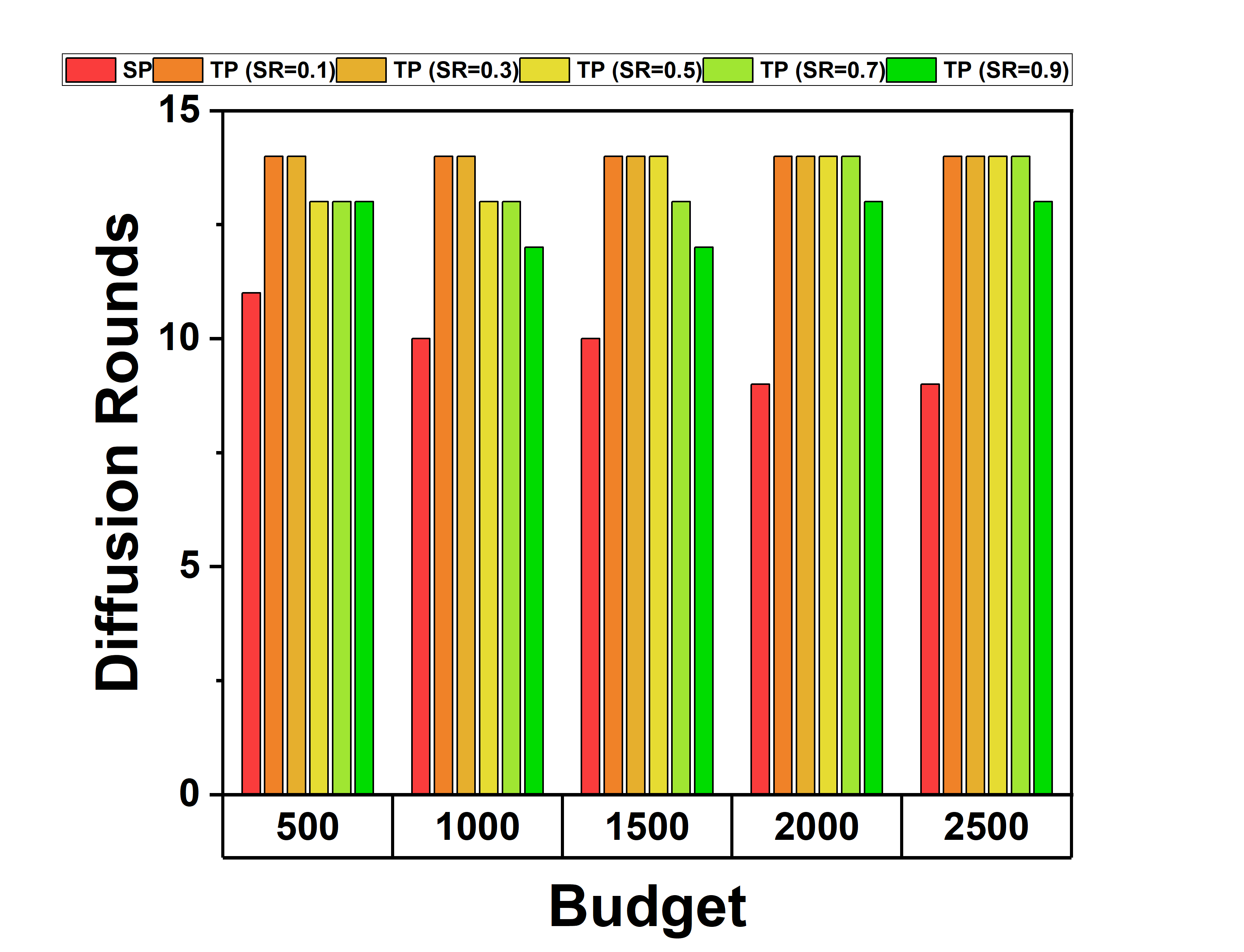}
        \caption{Timestep 10}
    \end{subfigure}

    \caption{Diffusion Rounds in Single Phase Vs. Two Phase (Degree Discount Algorithm, \textit{Email-Eu-Core} Dataset, Probability Setting - Trivalency)}
    \label{RQ5_T4}
\end{figure}

\begin{figure}[htbp]
    \centering

    \begin{subfigure}[t]{0.3\linewidth}
        \centering
        \includegraphics[width=\linewidth]{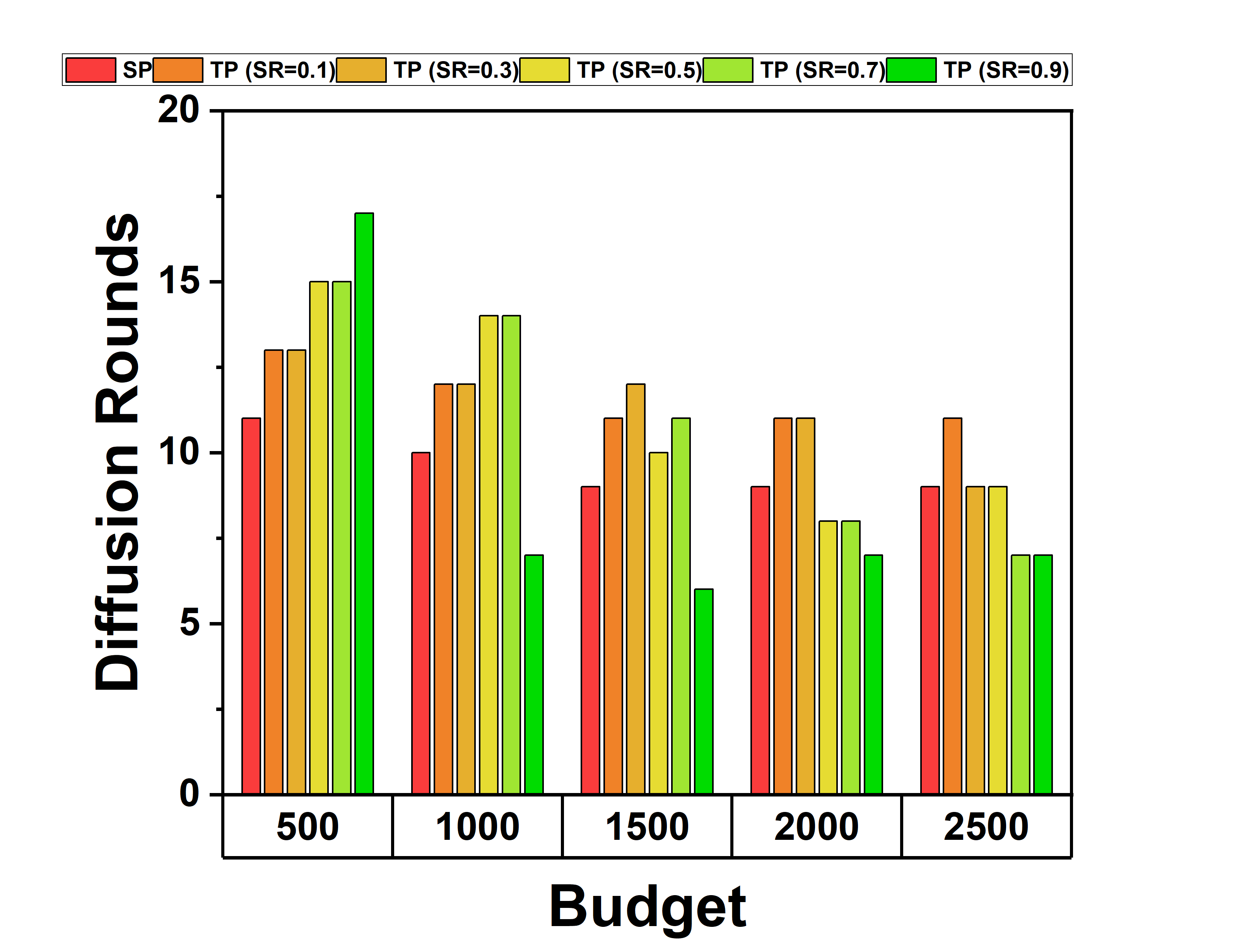}
        \caption{Timestep 2}
    \end{subfigure}
    \hspace{0.05\linewidth}
    \begin{subfigure}[t]{0.3\linewidth}
        \centering
        \includegraphics[width=\linewidth]{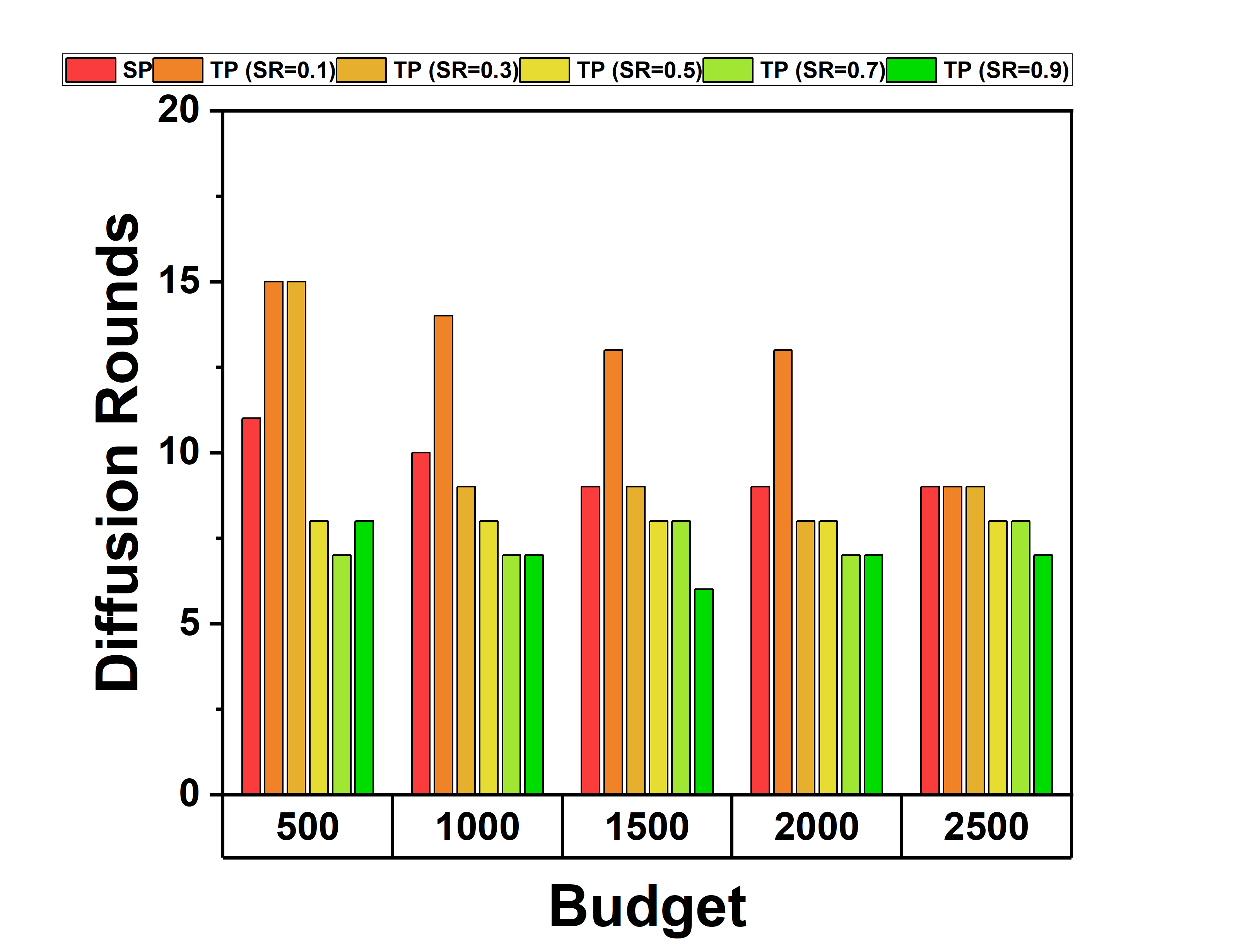}
        \caption{Timestep 4}
    \end{subfigure}

    \vspace{0.5cm}

    \begin{subfigure}[t]{0.3\linewidth}
        \centering
        \includegraphics[width=\linewidth]{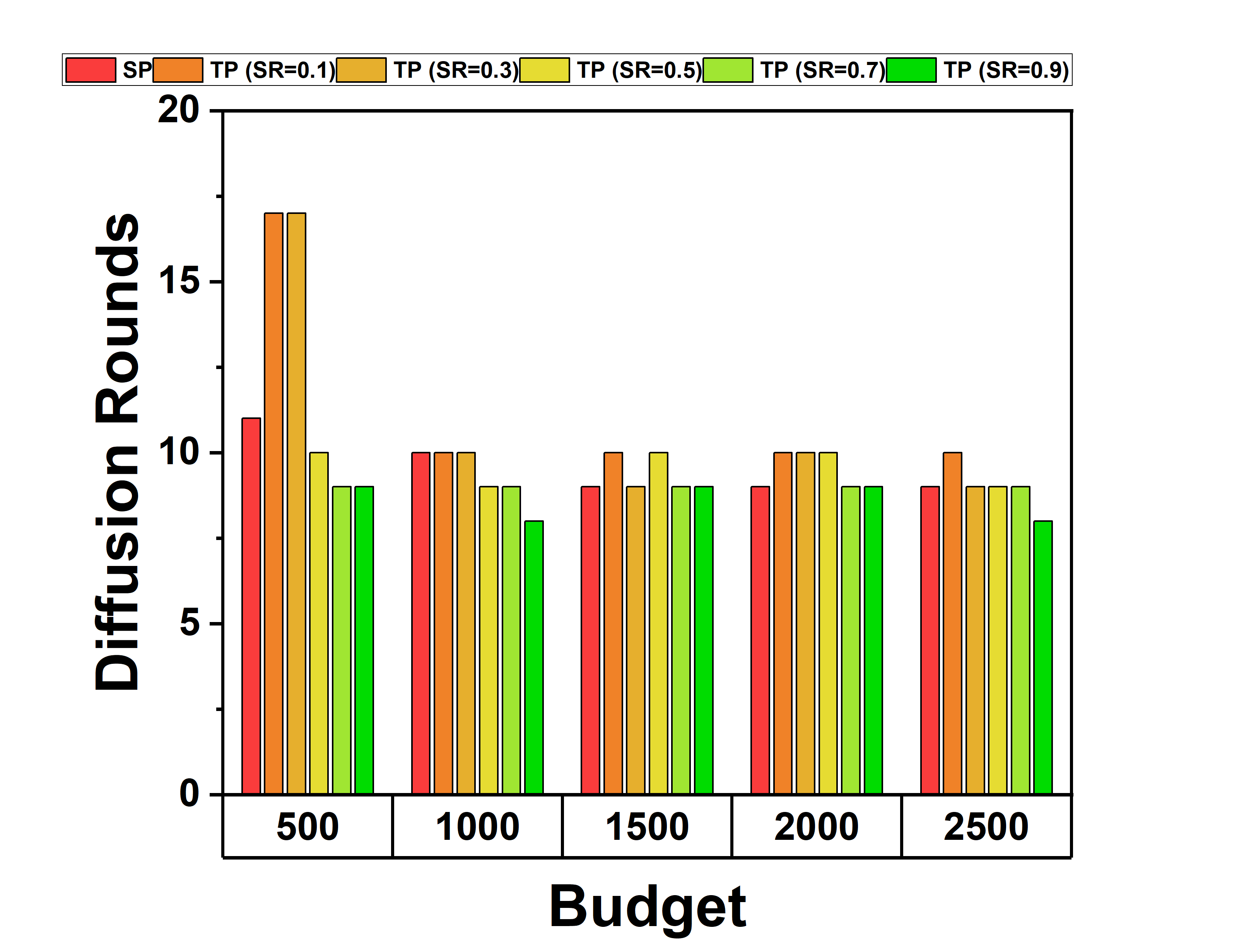}
        \caption{Timestep 6}
    \end{subfigure}
    \hfill
    \begin{subfigure}[t]{0.3\linewidth}
        \centering
        \includegraphics[width=\linewidth]{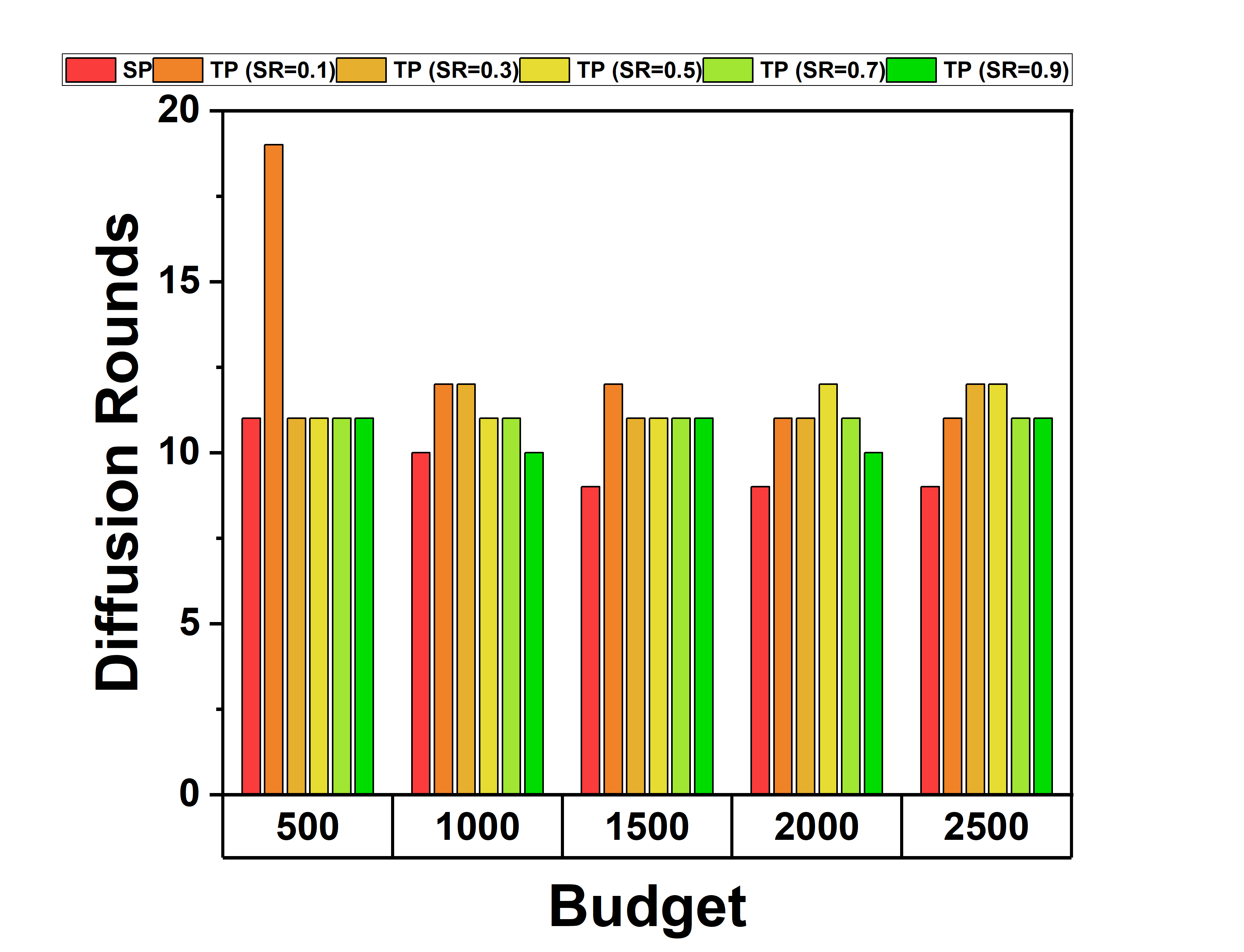}
        \caption{Timestep 8}
    \end{subfigure}
    \hfill
    \begin{subfigure}[t]{0.3\linewidth}
        \centering
        \includegraphics[width=\linewidth]{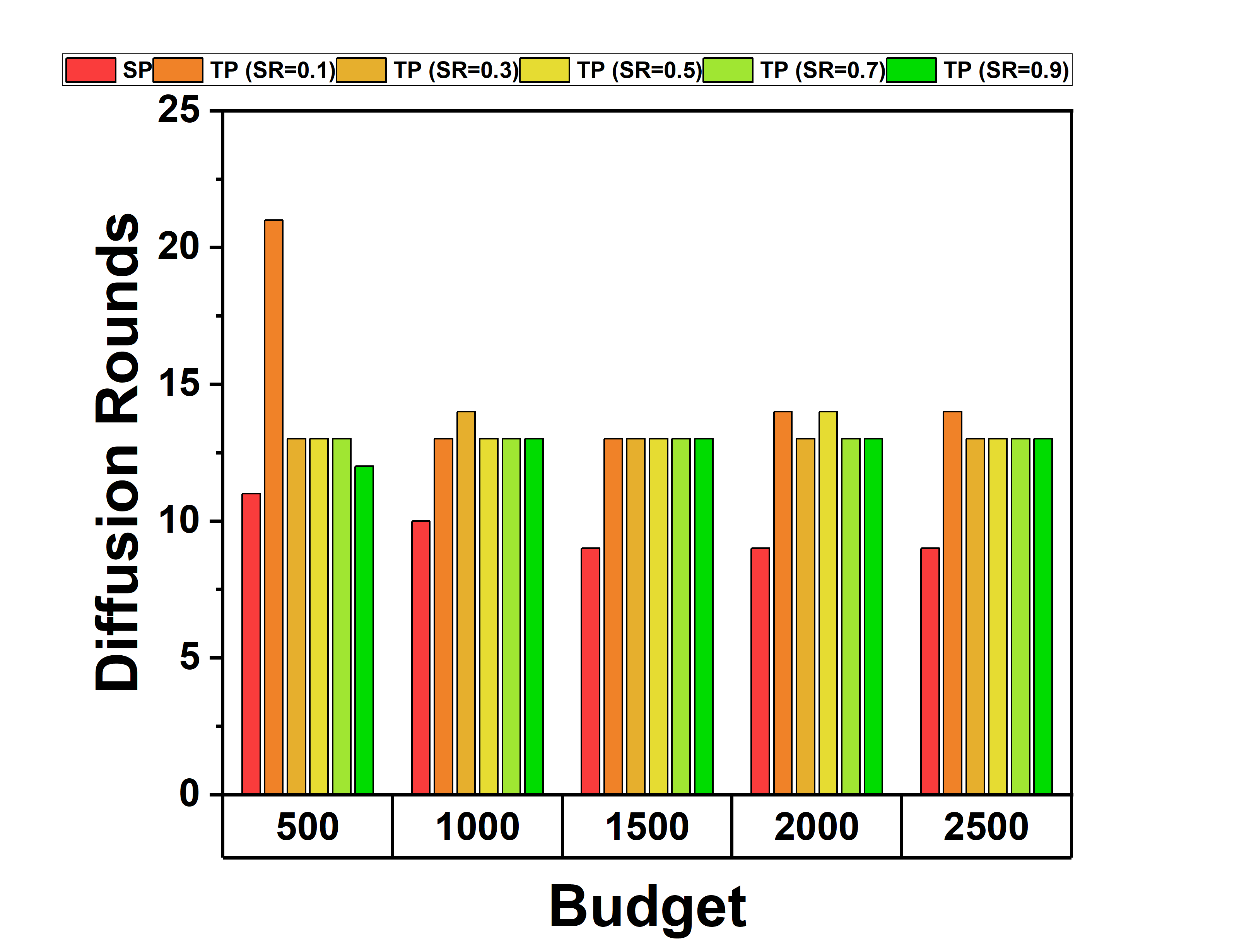}
        \caption{Timestep 10}
    \end{subfigure}

    \caption{Diffusion Rounds in Single Phase Vs. Two Phase (Single Discount Algorithm, \textit{Email-Eu-Core} Dataset, Probability Setting - Trivalency)}
    \label{RQ5_T5}
\end{figure}

\begin{figure}[htbp]
    \centering

    \begin{subfigure}[t]{0.3\linewidth}
        \centering
        \includegraphics[width=\linewidth]{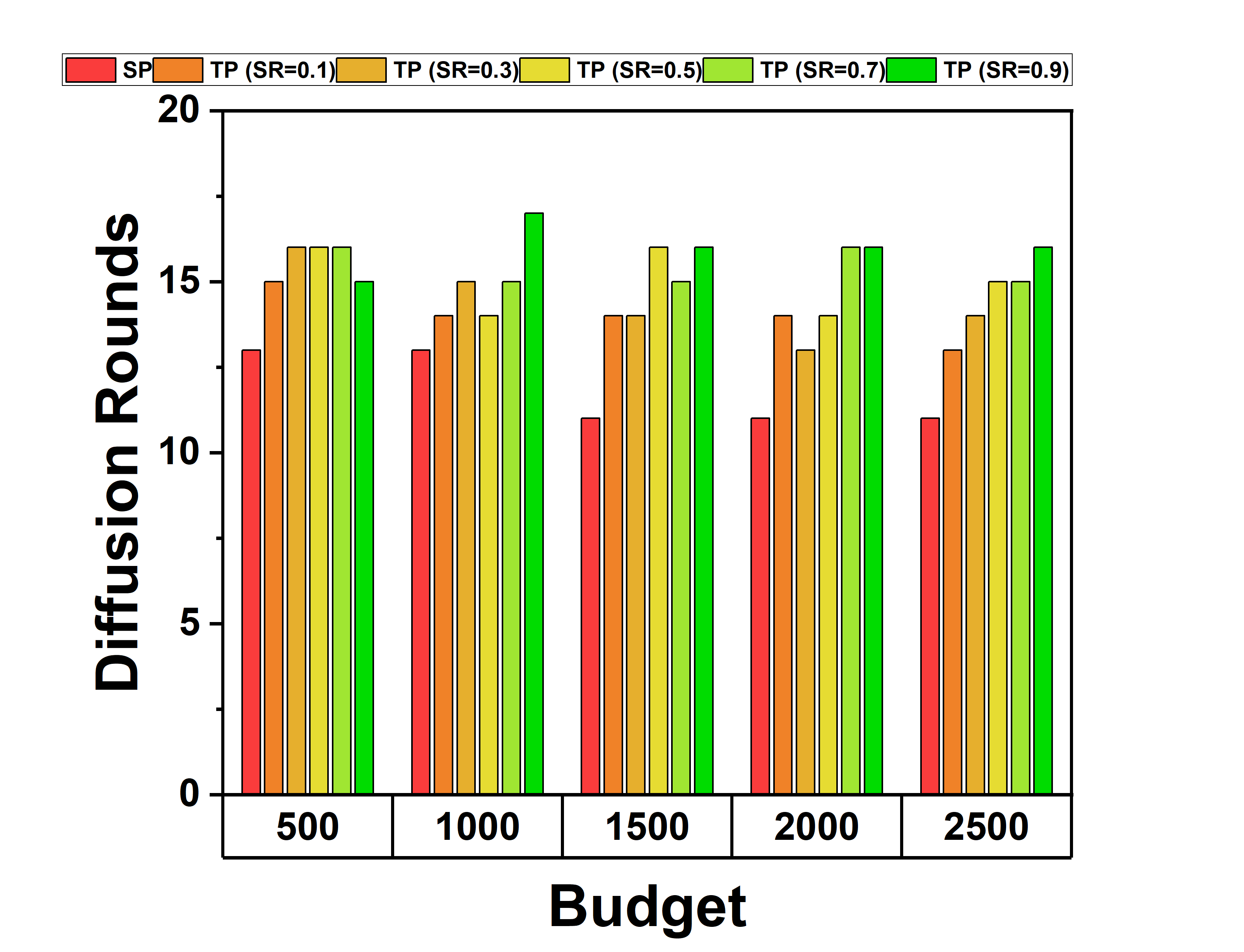}
        \caption{Timestep 2}
    \end{subfigure}
    \hspace{0.05\linewidth}
    \begin{subfigure}[t]{0.3\linewidth}
        \centering
        \includegraphics[width=\linewidth]{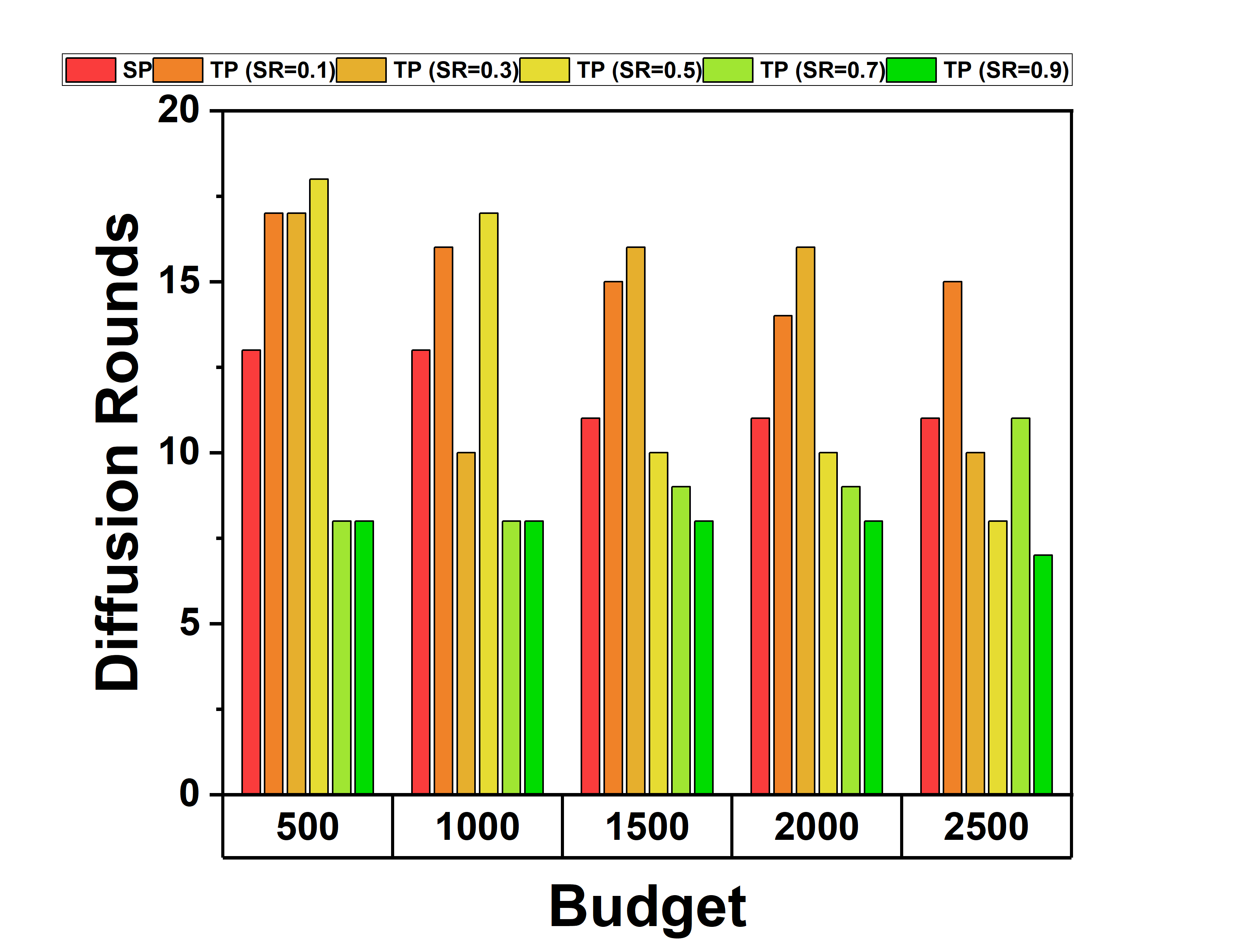}
        \caption{Timestep 4}
    \end{subfigure}

    \vspace{0.5cm}

    \begin{subfigure}[t]{0.3\linewidth}
        \centering
        \includegraphics[width=\linewidth]{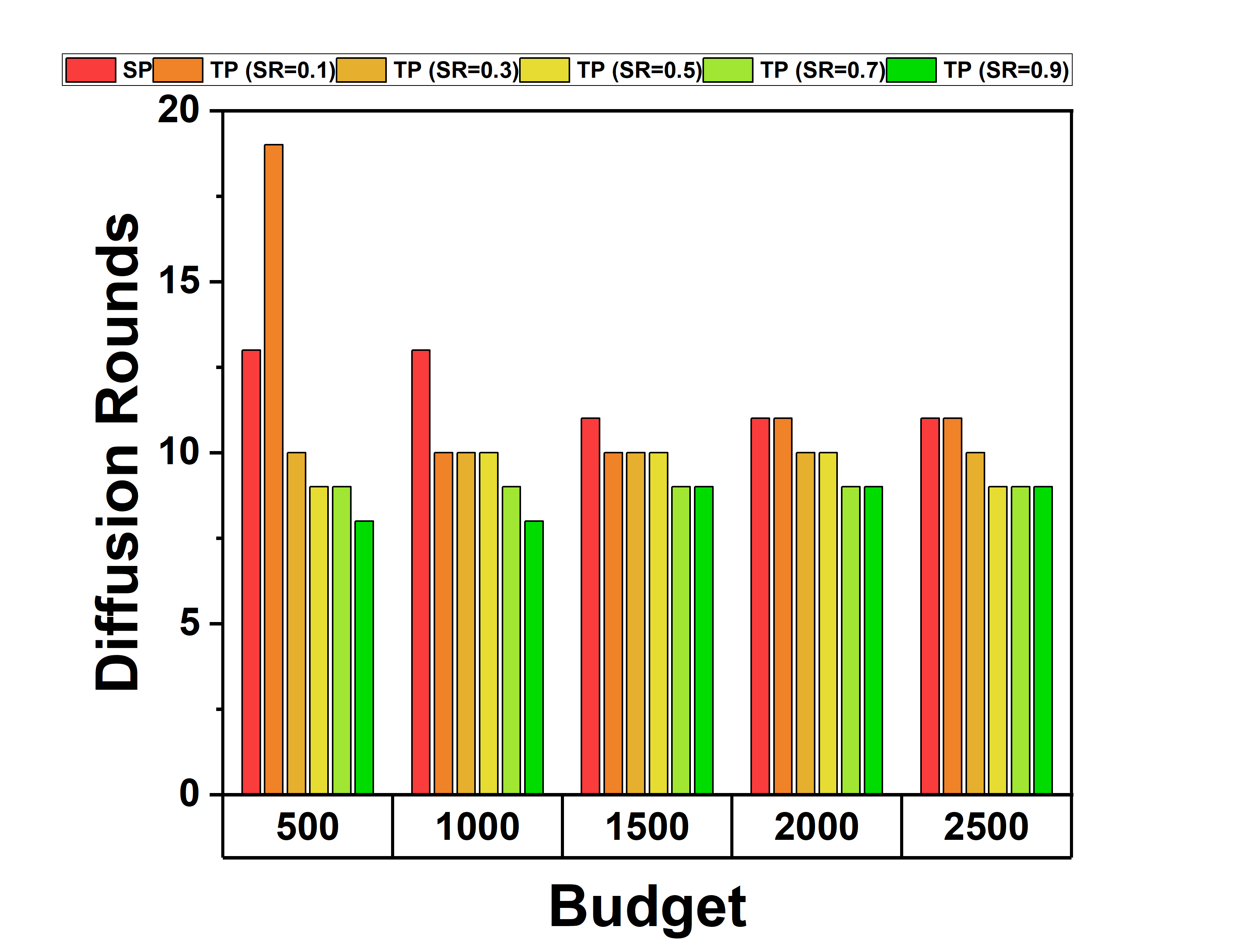}
        \caption{Timestep 6}
    \end{subfigure}
    \hfill
    \begin{subfigure}[t]{0.3\linewidth}
        \centering
        \includegraphics[width=\linewidth]{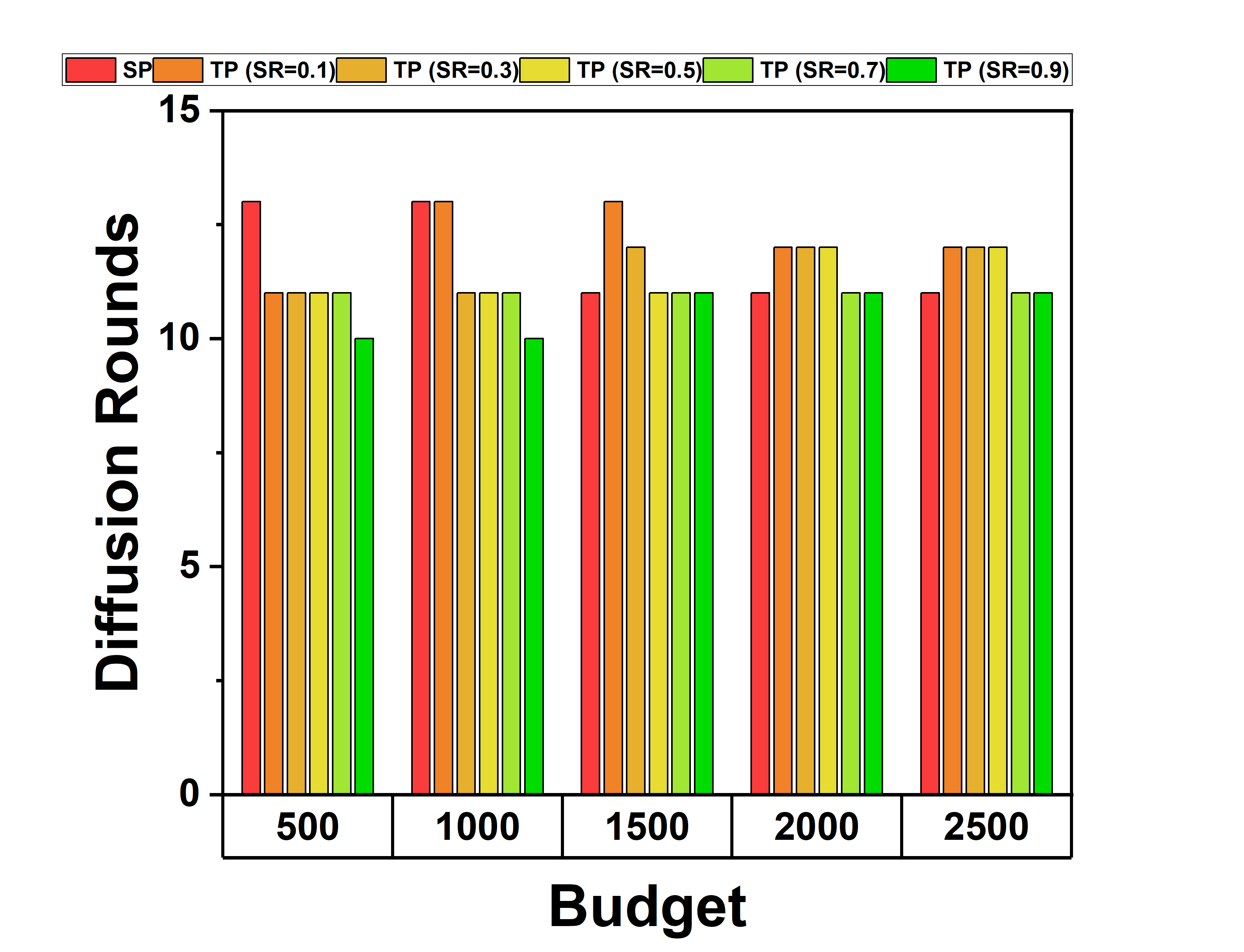}
        \caption{Timestep 8}
    \end{subfigure}
    \hfill
    \begin{subfigure}[t]{0.3\linewidth}
        \centering
        \includegraphics[width=\linewidth]{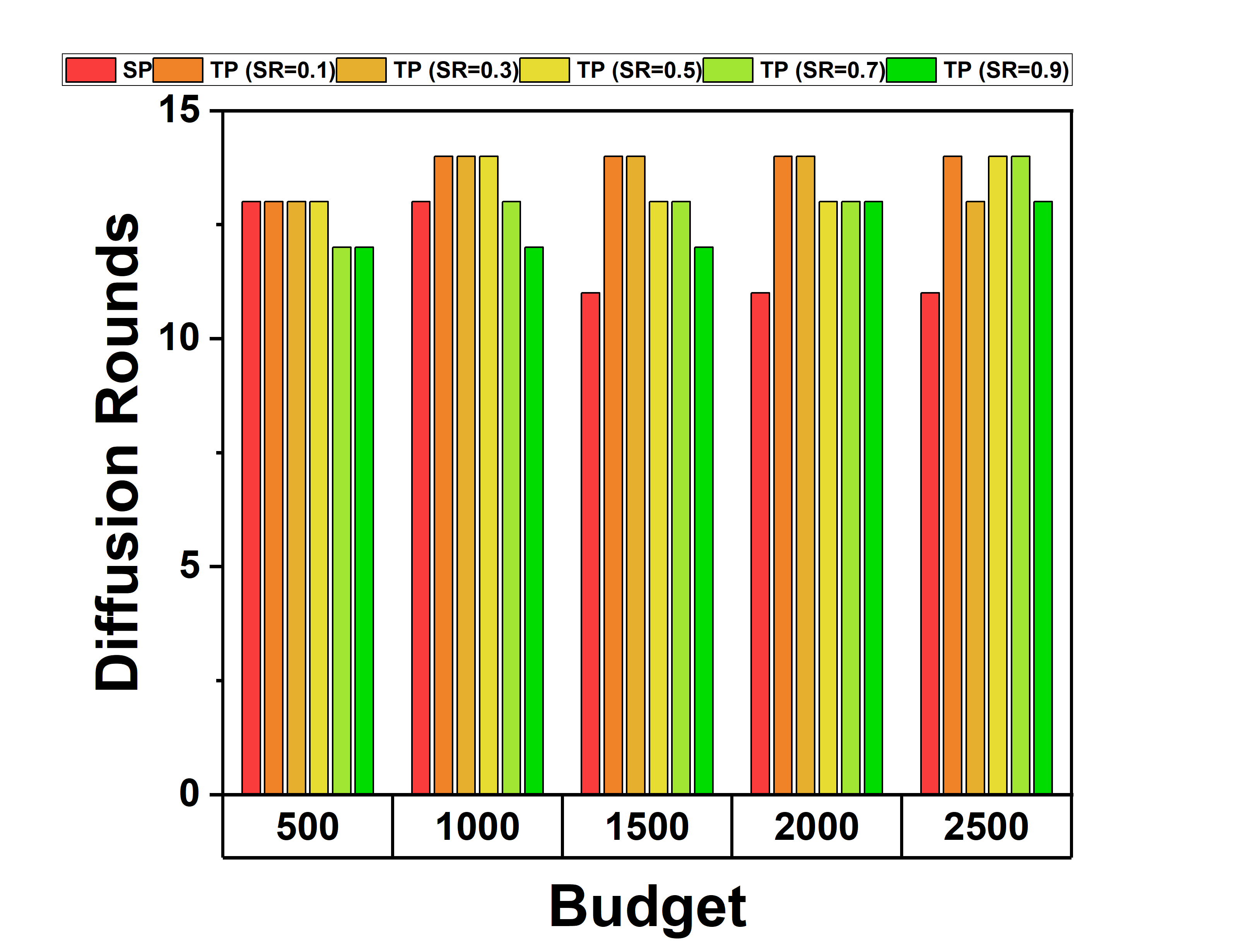}
        \caption{Timestep 10}
    \end{subfigure}

    \caption{Diffusion Rounds in Single Phase Vs. Two Phase (Simple Greedy Algorithm, \textit{Email-Eu-Core} Dataset, Probability Setting - Trivalency)}
    \label{RQ5_T6}
\end{figure}

\begin{figure}[htbp]
    \centering

    \begin{subfigure}[t]{0.3\linewidth}
        \centering
        \includegraphics[width=\linewidth]{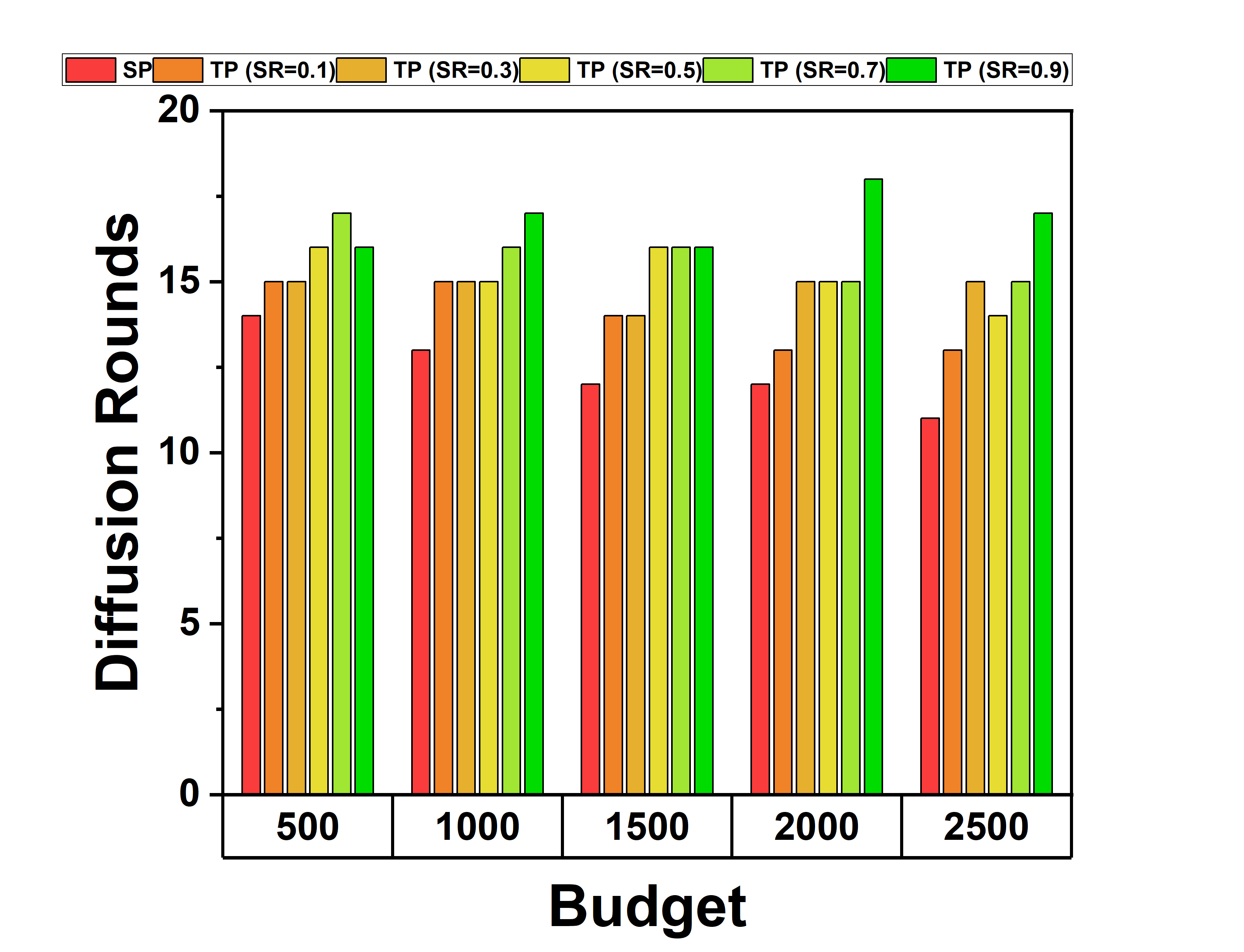}
        \caption{Timestep 2}
    \end{subfigure}
    \hspace{0.05\linewidth}
    \begin{subfigure}[t]{0.3\linewidth}
        \centering
        \includegraphics[width=\linewidth]{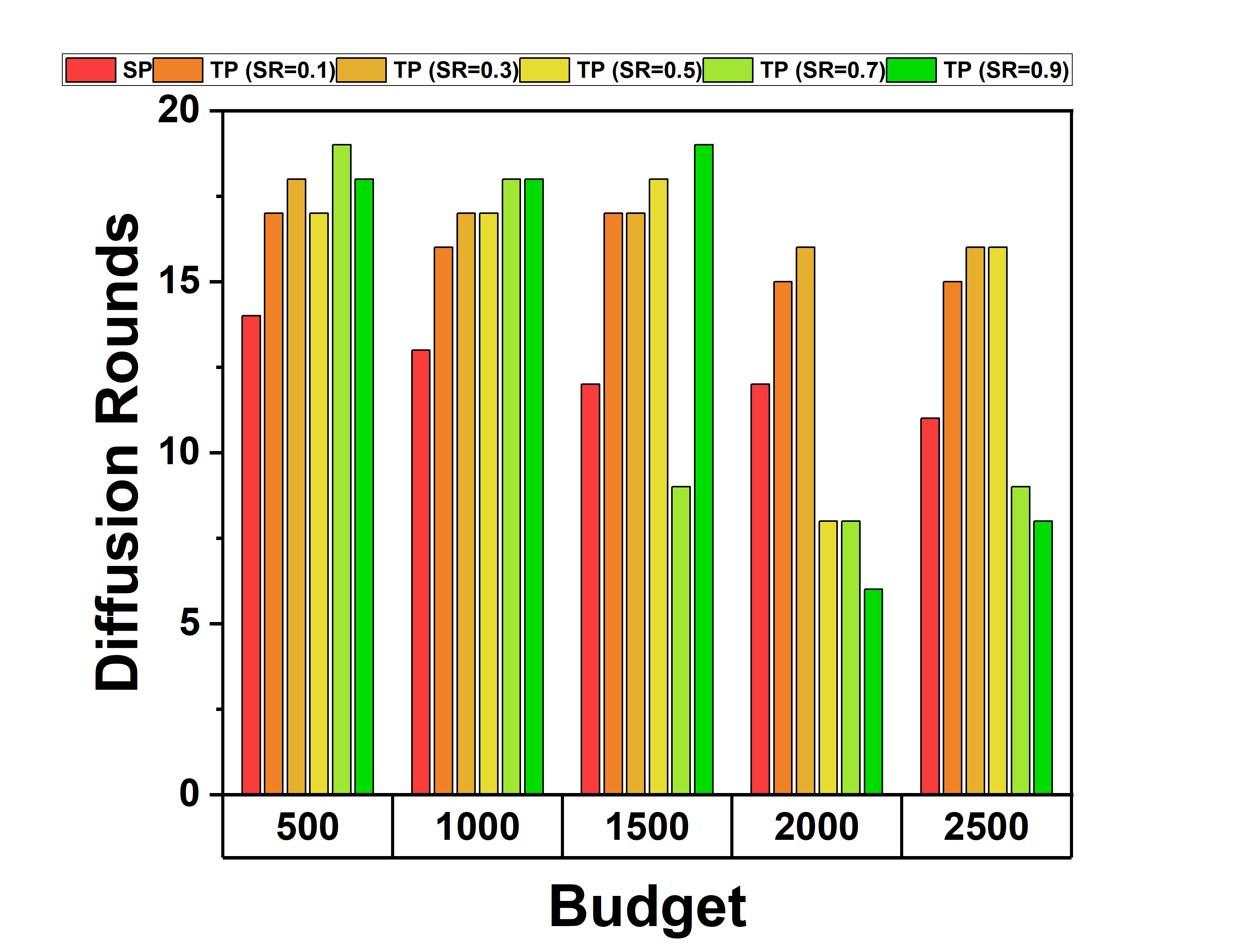}
        \caption{Timestep 4}
    \end{subfigure}

    \vspace{0.5cm}

    \begin{subfigure}[t]{0.3\linewidth}
        \centering
        \includegraphics[width=\linewidth]{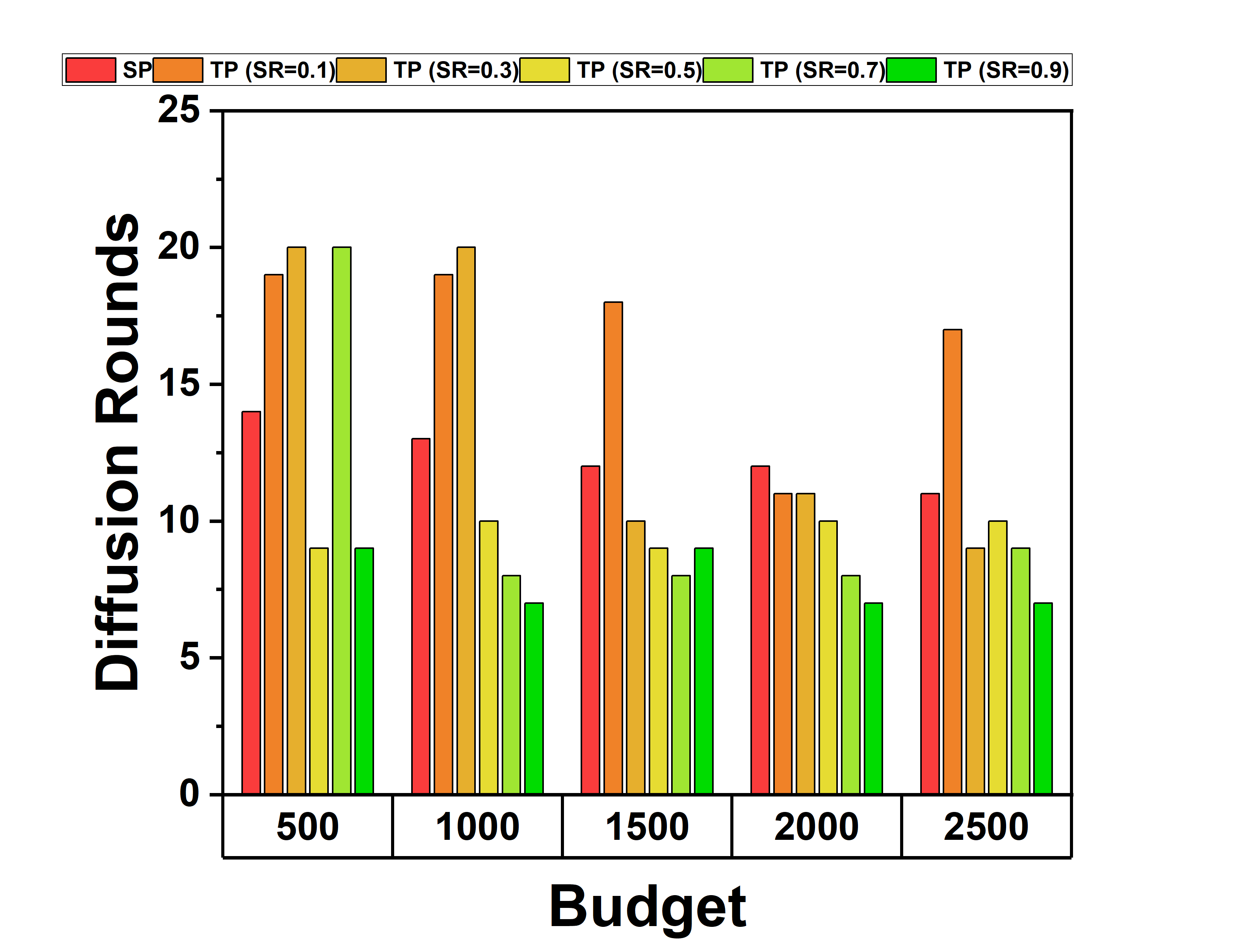}
        \caption{Timestep 6}
    \end{subfigure}
    \hfill
    \begin{subfigure}[t]{0.3\linewidth}
        \centering
        \includegraphics[width=\linewidth]{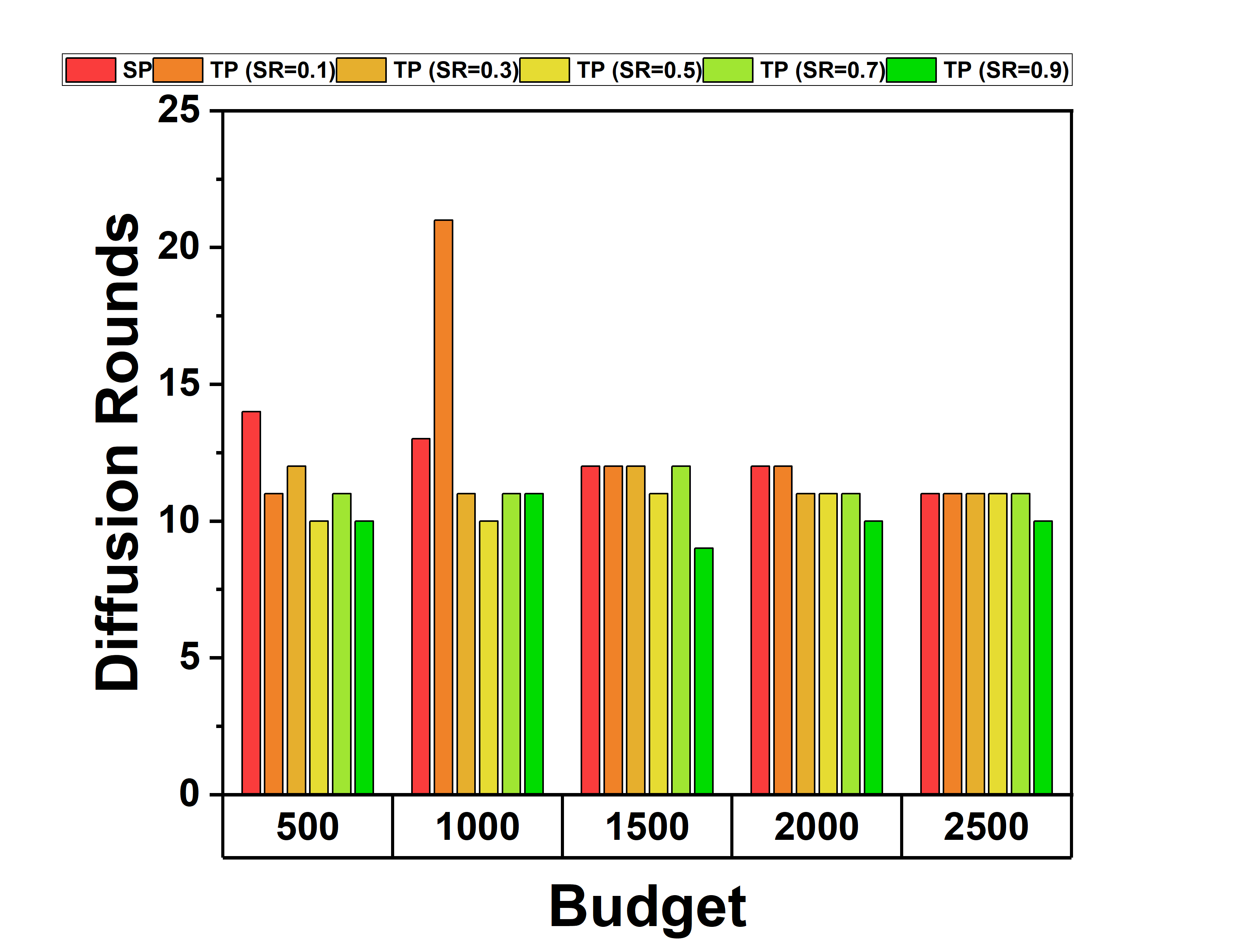}
        \caption{Timestep 8}
    \end{subfigure}
    \hfill
    \begin{subfigure}[t]{0.3\linewidth}
        \centering
        \includegraphics[width=\linewidth]{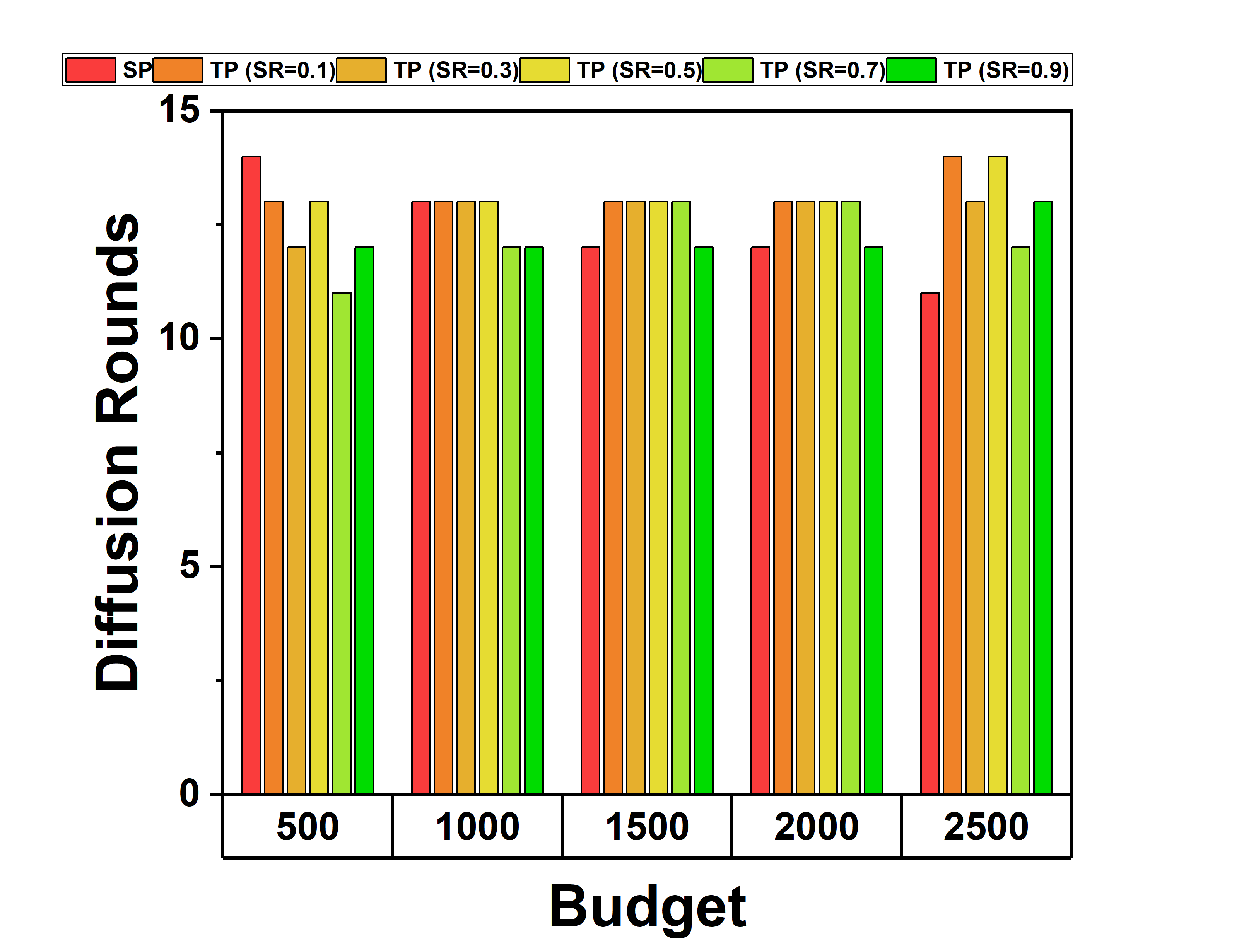}
        \caption{Timestep 10}
    \end{subfigure}

    \caption{Diffusion Rounds in Single Phase Vs. Two Phase (Double Greedy Algorithm, \textit{Email-Eu-Core} Dataset, Probability Setting - Trivalency)}
    \label{RQ5_T7}
\end{figure}

\begin{figure}[htbp]
    \centering

    \begin{subfigure}[t]{0.3\linewidth}
        \centering
        \includegraphics[width=\linewidth]{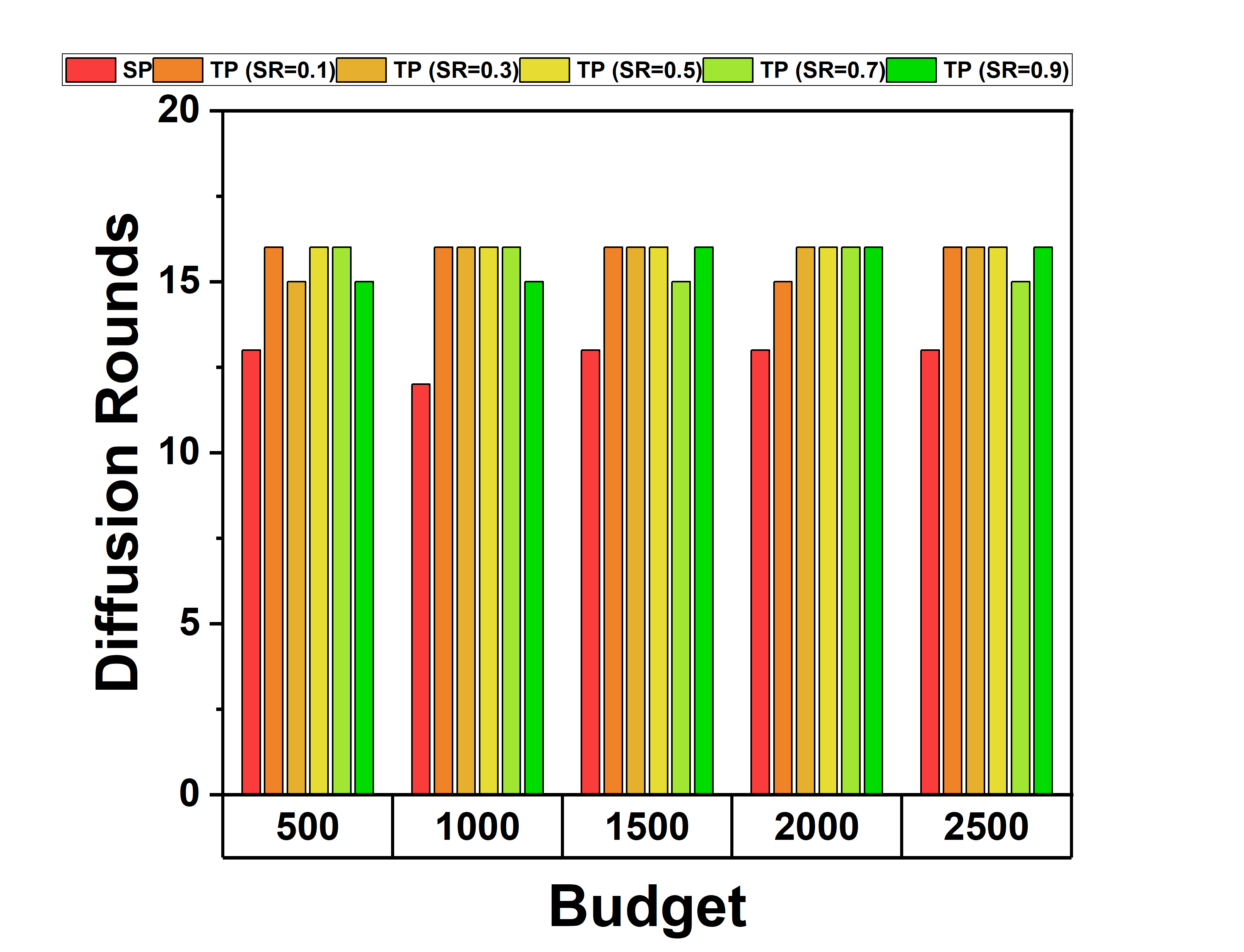}
        \caption{Timestep 2}
    \end{subfigure}
    \hspace{0.05\linewidth}
    \begin{subfigure}[t]{0.3\linewidth}
        \centering
        \includegraphics[width=\linewidth]{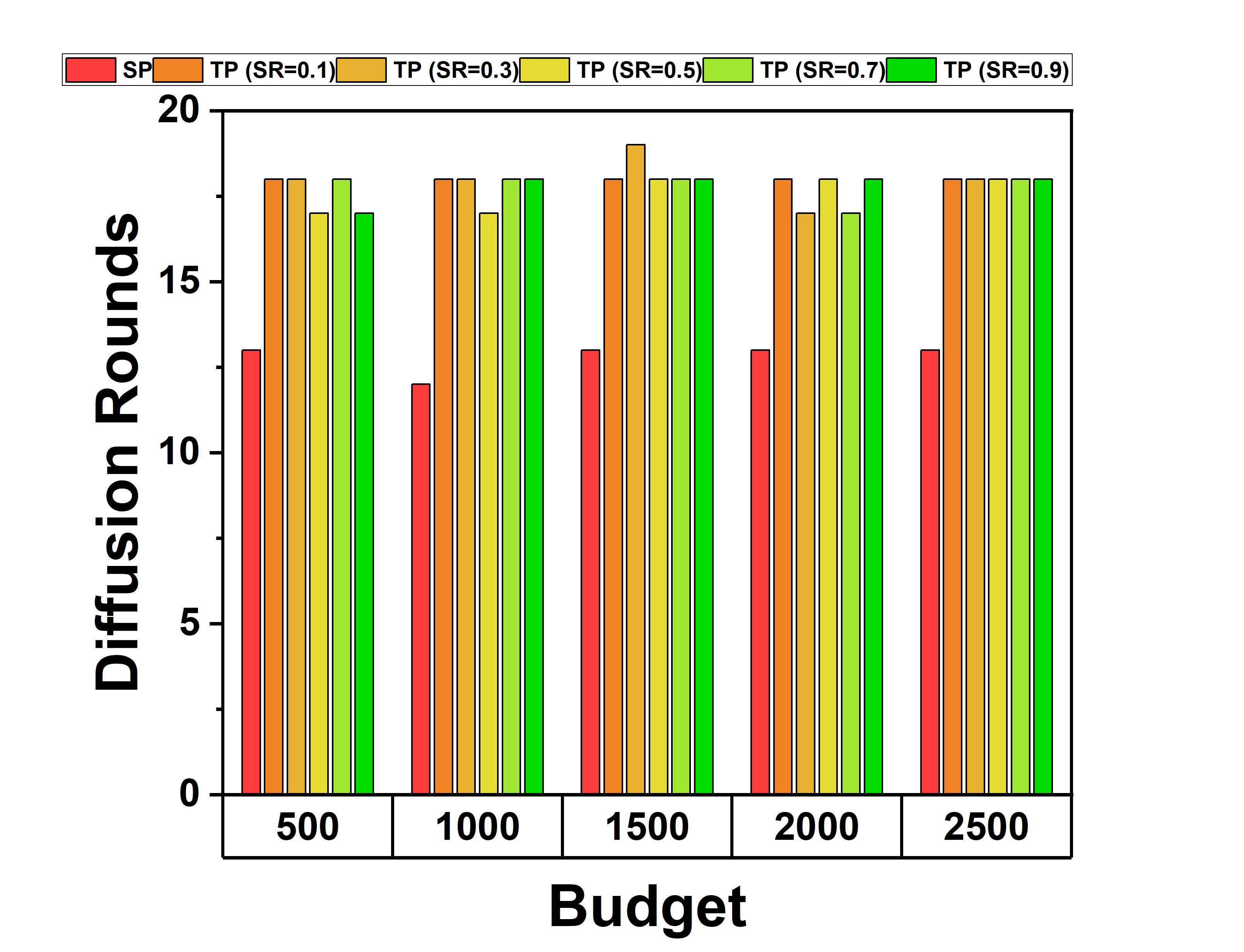}
        \caption{Timestep 4}
    \end{subfigure}

    \vspace{0.5cm}

    \begin{subfigure}[t]{0.3\linewidth}
        \centering
        \includegraphics[width=\linewidth]{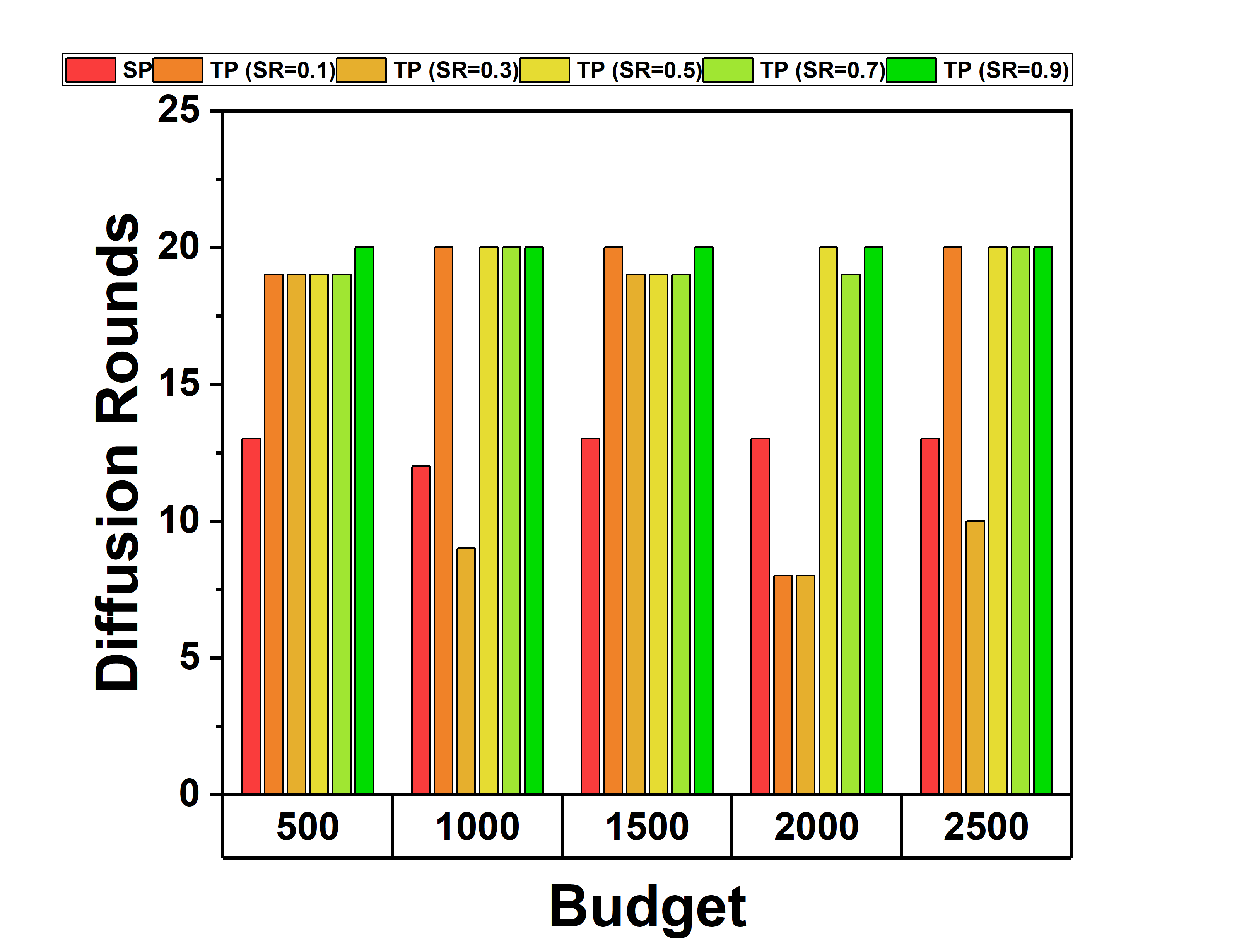}
        \caption{Timestep 6}
    \end{subfigure}
    \hfill
    \begin{subfigure}[t]{0.3\linewidth}
        \centering
        \includegraphics[width=\linewidth]{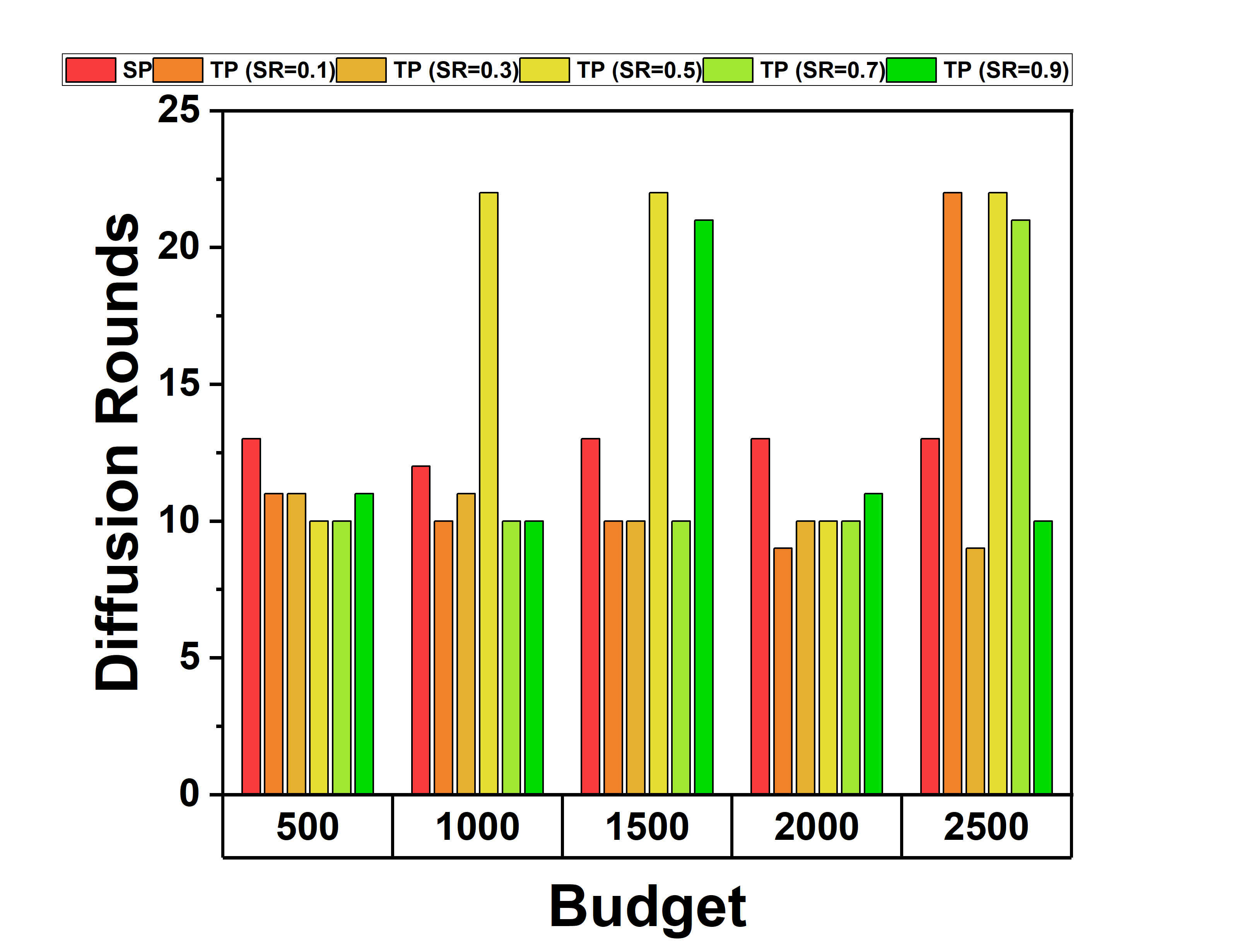}
        \caption{Timestep 8}
    \end{subfigure}
    \hfill
    \begin{subfigure}[t]{0.3\linewidth}
        \centering
        \includegraphics[width=\linewidth]{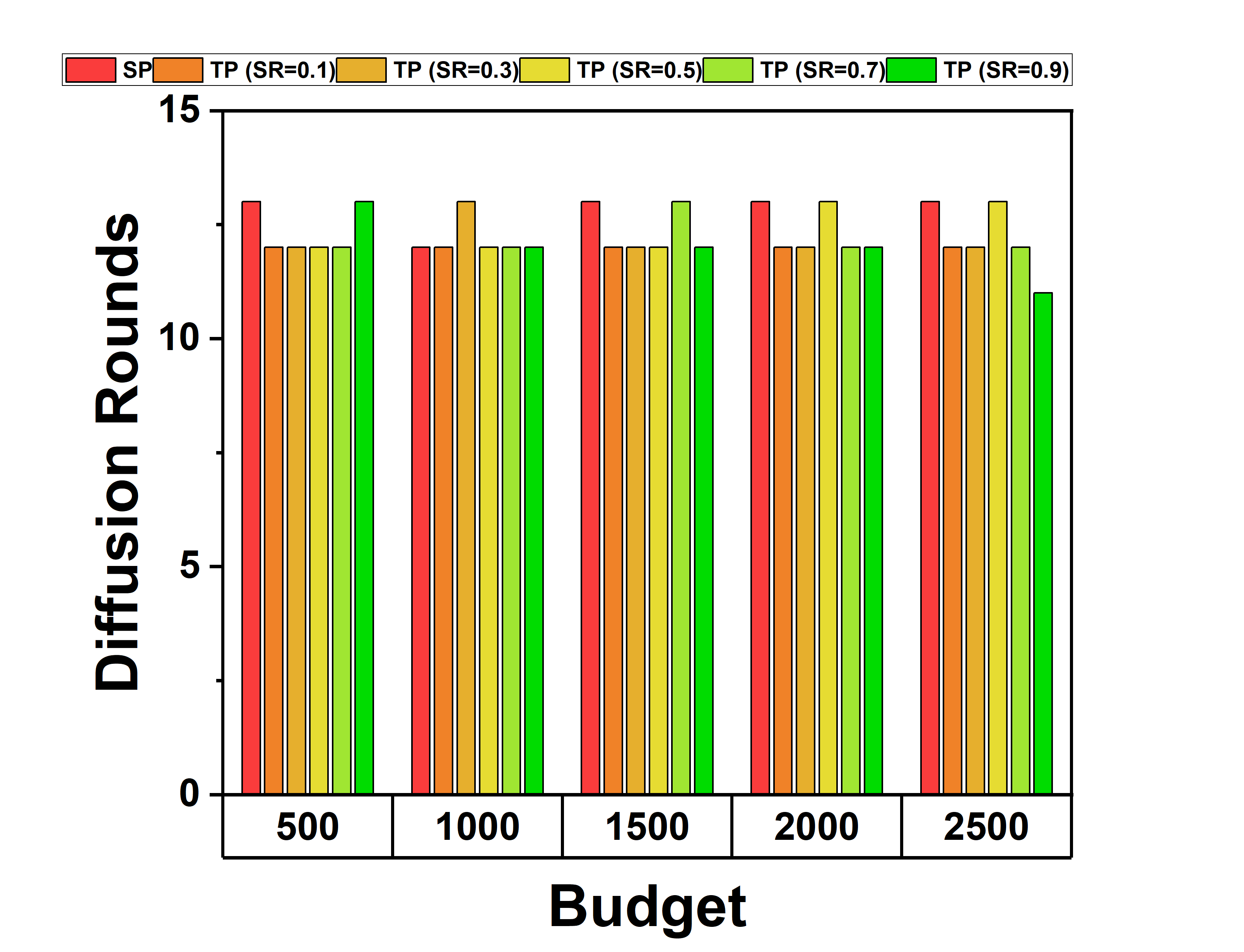}
        \caption{Timestep 10}
    \end{subfigure}

    \caption{Diffusion Rounds in Single Phase Vs. Two Phase (Stochastic Greedy Algorithm, \textit{Email-Eu-Core} Dataset, Probability Setting - Trivalency)}
    \label{RQ5_T8}
\end{figure}

Figures~\ref{RQ5_T1} to \ref{RQ5_T8} show how the number of diffusion rounds differs between single-phase and two phase settings in the \textit{Email-Eu-Core} dataset under trivalency probability setting. For the \textbf{Random} algorithm, diffusion rounds in the single-phase increase with budget—from $12$ rounds at budget $500$ to $20$ rounds at budget $2500$. In the two phase setting, this rises further. At budget $2500$, diffusion rounds ranged from $22$ to $24$, with $24$ rounds recorded at split ratio $0.5$ and timestep $6$ (Figure~\ref{RQ5_T1}(c)), which is higher than the single-phase value of $20$. \textbf{HD} shows a varied but generally increased number of diffusion rounds in the two phase setup. At budget $500$, two phase rounds range from $6$ to $14$. One instance had $14$ rounds at split ratio $0.1$ and timestep $2$ (Figure~\ref{RQ5_T2}(a)), compared to $11$ rounds in single-phase. However, another setting with split ratio $0.7$ and timestep $8$ had only $8$ rounds, showing some variation. For \textbf{HighCC}, two phase diffusion rounds are consistently higher. At budget $500$, they range from $6$ to $12$. For example, at split ratio $0.1$ and timestep $2$, \textbf{HighCC} recorded $7$ rounds—higher than the $4$ rounds in single-phase. At split ratio $0.9$ and timestep $10$, it reached $12$ rounds (Figure~\ref{RQ5_T3}(e)). \textbf{DD} also shows higher or equal diffusion rounds in the two phase setup. At budget $500$, rounds range from $11$ to $18$. One configuration shows $13$ rounds at split ratio $0.1$ and timestep $2$ (Figure~\ref{RQ5_T4}(a)), up from $11$ in single-phase. \textbf{SD} shows a similar pattern, with rounds ranging from $11$ to $19$ in the two phase setting. In one example (Figure~\ref{RQ5_T5}(a)), \textbf{SD} used $13$ rounds at split ratio $0.1$ and timestep $2$, compared to $11$ in single-phase. \textbf{SG} gives mixed results but usually shows equal or slightly more rounds. At budget $500$, two phase diffusion rounds range from $11$ to $19$. One setting shows $15$ rounds at split ratio $0.1$ and timestep $2$ (Figure~\ref{RQ5_T6}(a)), which is higher than $13$ in single-phase. At budget $2500$, rounds range from $20$ to $24$. For example, $20$ rounds were recorded at split ratio $0.1$ and timestep $2$, matching single-phase, and $24$ rounds were seen at split ratio $0.9$ and timestep $10$ (Figure~\ref{RQ5_T6}(e)). \textbf{DG} also performs well in this aspect. At budget $500$, rounds in two phase range from $11$ to $19$. One example (Figure~\ref{RQ5_T7}(a)) shows $15$ rounds at split ratio $0.1$ and timestep $2$, compared to $14$ in single-phase. For \textbf{StG0.1}, diffusion rounds at budget $500$ range from $11$ to $19$. One case in Figure~\ref{RQ5_T8}(a) shows $16$ rounds at split ratio $0.1$ and timestep $2$, which is higher than $13$ in single-phase. Overall, these findings suggest that the two phase setting helps extend the influence process by allowing more rounds of diffusion. This deeper spread generally leads to a broader reach and higher potential impact in the network.

\subsubsection{Comparative Analysis of Computational Time Across Algorithms}
The computational analysis for the \textit{LM} dataset (under the trivalency probability setting) focuses on configurations where the timestep is fixed at $2$, reflecting early intervention in the diffusion process. The results in Figure~\ref{RQ6LM_T1}(a) show that most algorithms strike a reasonable balance between profit gain and the extra computation time required in the two phase framework. For example, the \textbf{Random} algorithm at budget $500$ and split ratio $0.1$ achieved a profit improvement of $39.43\%$ while requiring only $7.06$ \textit{seconds} of additional time. Similarly, in the same configuration, the \textbf{HD} algorithm gained $15.38\%$ and used $7.20$ \textit{seconds} more, as shown in Figure~\ref{RQ6LM_T1}(a). In that same figure, the \textbf{HighCC} algorithm showed a profit gain of just $0.34\%$ for an extra $7.31$ \textit{seconds} of time. The \textbf{DD} algorithm, under a split ratio of $0.3$, improved profit by $0.25\%$ with an added time of $7.65$ \textit{seconds}. \textbf{SD} delivered a profit increase of $2.68\%$ for a time cost of $9.12$ \textit{seconds} (Figure~\ref{RQ6LM_T1}(a)), while the more computationally demanding \textbf{SG} showed a $1.38\%$ gain for $9.49$ \textit{seconds} of extra time. In contrast, the \textbf{DG} algorithm provided a strong trade-off, achieving a $9.18\%$ profit improvement with just $8.29$ \textit{seconds} of additional computation time, as seen in Figure~\ref{RQ6LM_T1}(a). Finally, the \textbf{StG0.1} algorithm, even with higher diffusion and computation, remained efficient. At split ratio $0.3$, it delivered a $15.75\%$ gain in $8.01$ \textit{seconds} (Figure~\ref{RQ6LM_T1}(a)). In summary, for a fixed timestep of $2$, most algorithms manage to justify the additional computational effort by showing a positive profit gain, especially our proposed methods—\textbf{DG}, \textbf{SG}, and \textbf{StG0.1} which maintain good performance even under constrained timing.

\begin{figure}[htbp]
    \centering
    \begin{subfigure}[t]{0.3\linewidth}
        \centering
        \includegraphics[width=\linewidth]{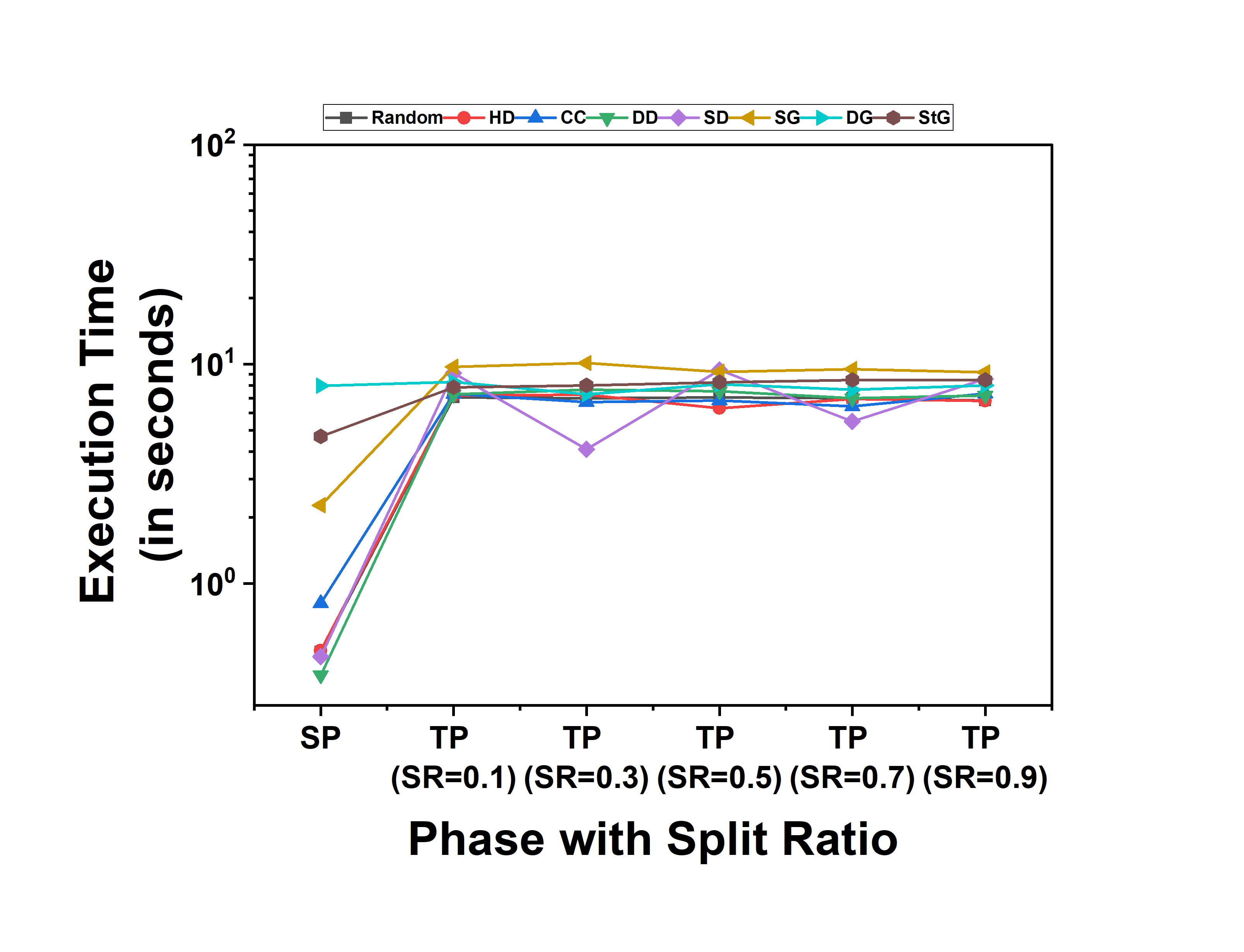}
        \caption{Budget 500}
    \end{subfigure}
    \begin{subfigure}[t]{0.3\linewidth}
        \centering
        \includegraphics[width=\linewidth]{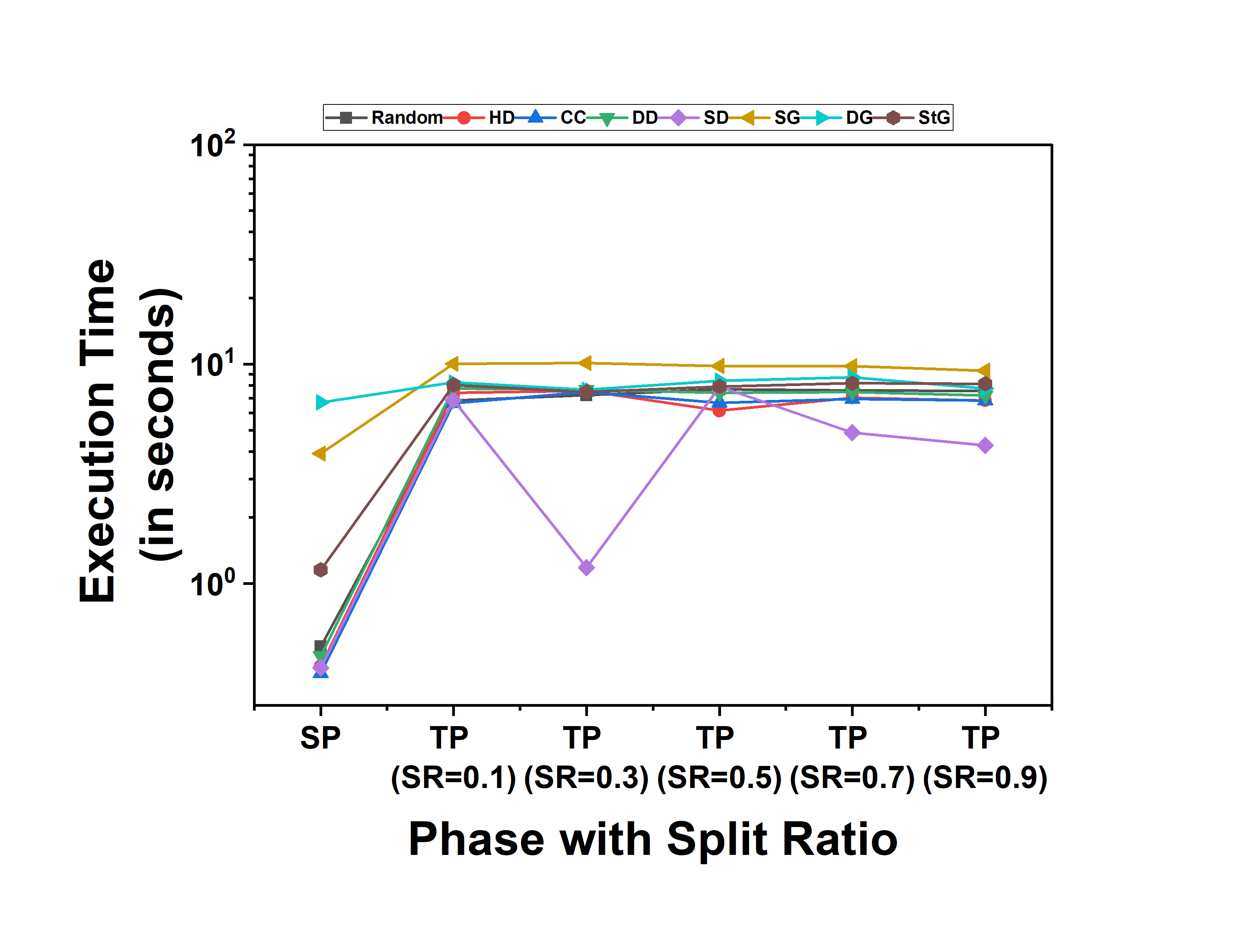}
        \caption{Budget 1000}
    \end{subfigure}

    \vspace{0.5cm}

    \begin{subfigure}[t]{0.3\linewidth}
        \centering
        \includegraphics[width=\linewidth]{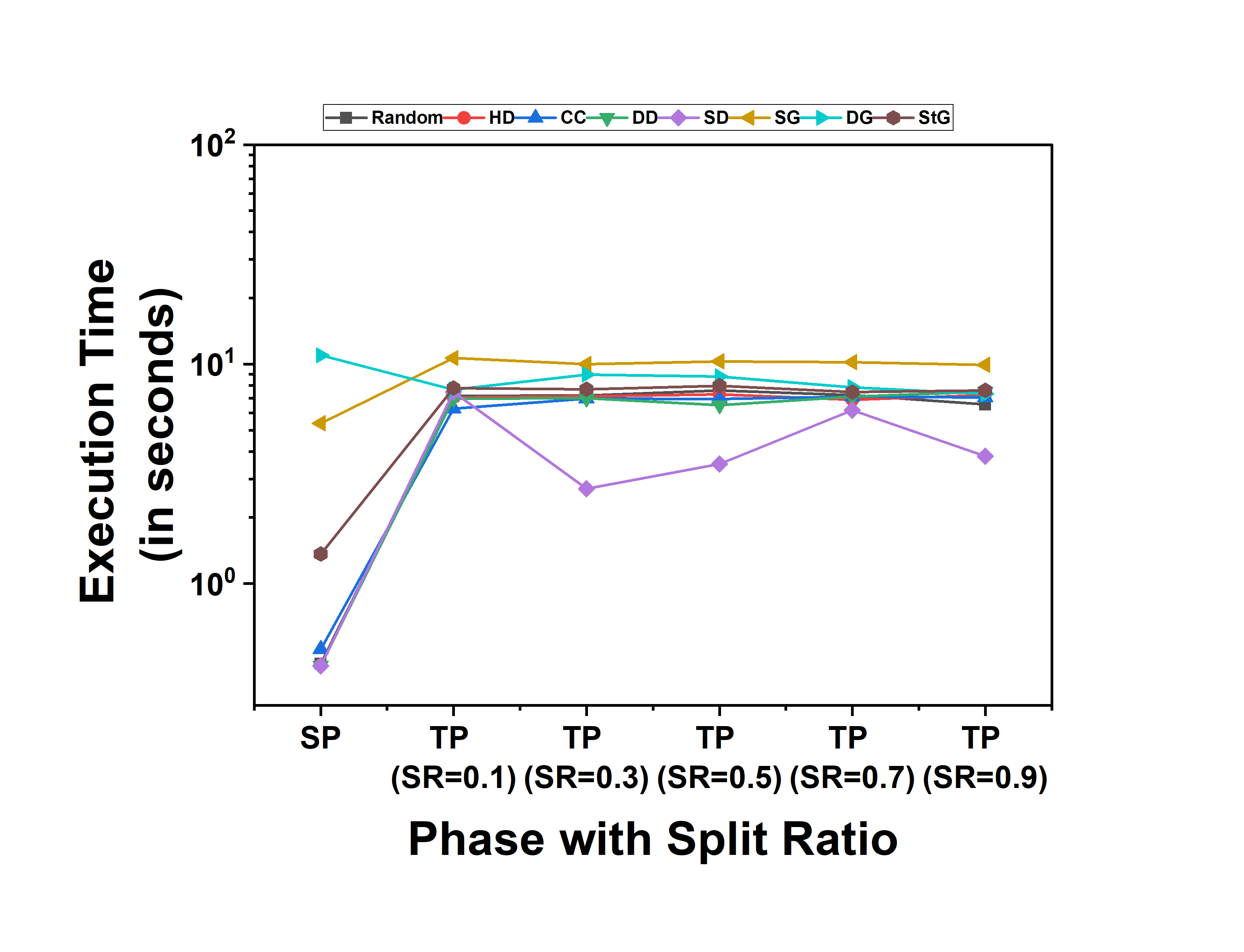}
        \caption{Budget 1500}
    \end{subfigure}
    \hfill
    \begin{subfigure}[t]{0.3\linewidth}
        \centering
        \includegraphics[width=\linewidth]{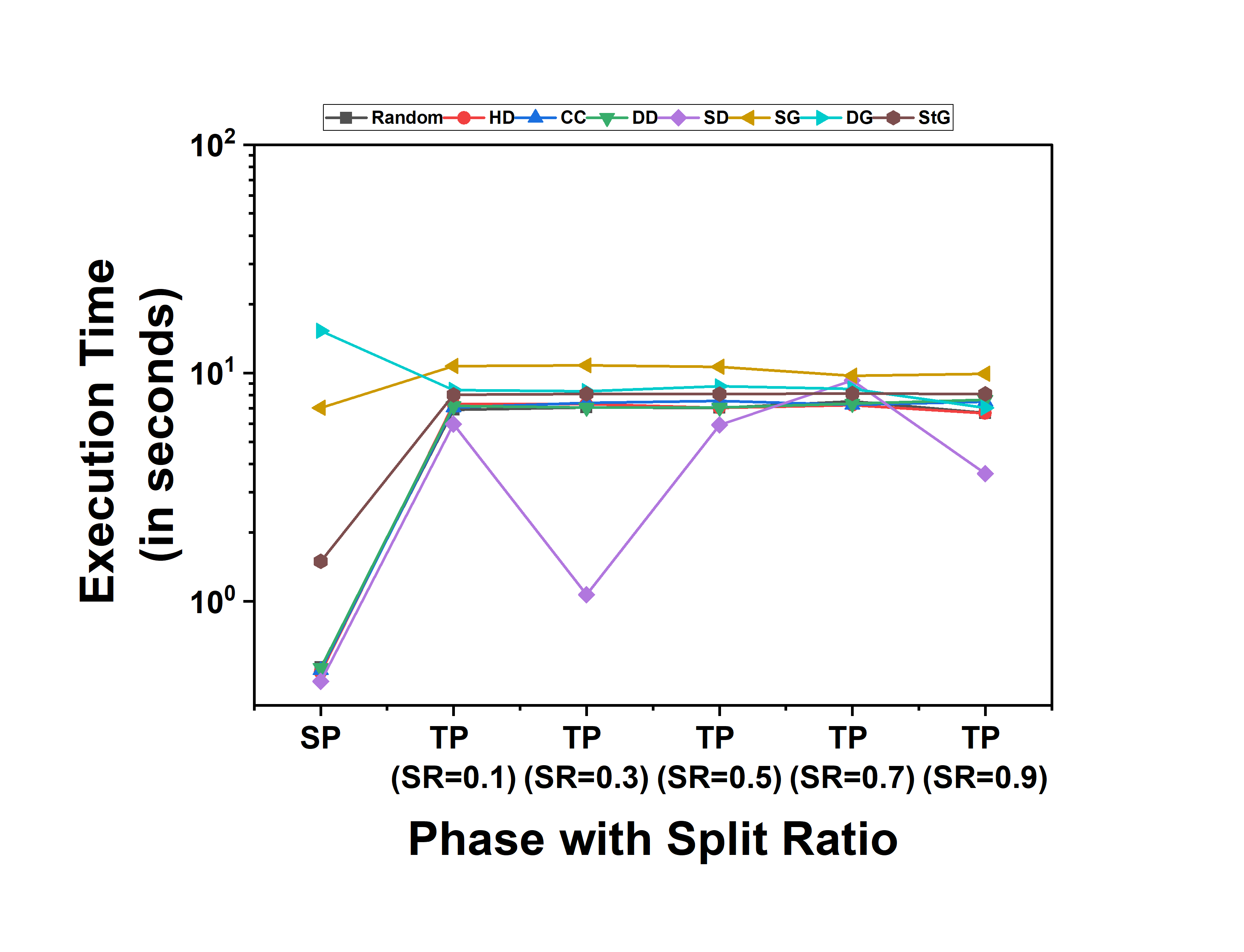}
        \caption{Budget 2000}
    \end{subfigure}
    \hfill
    \begin{subfigure}[t]{0.3\linewidth}
        \centering
        \includegraphics[width=\linewidth]{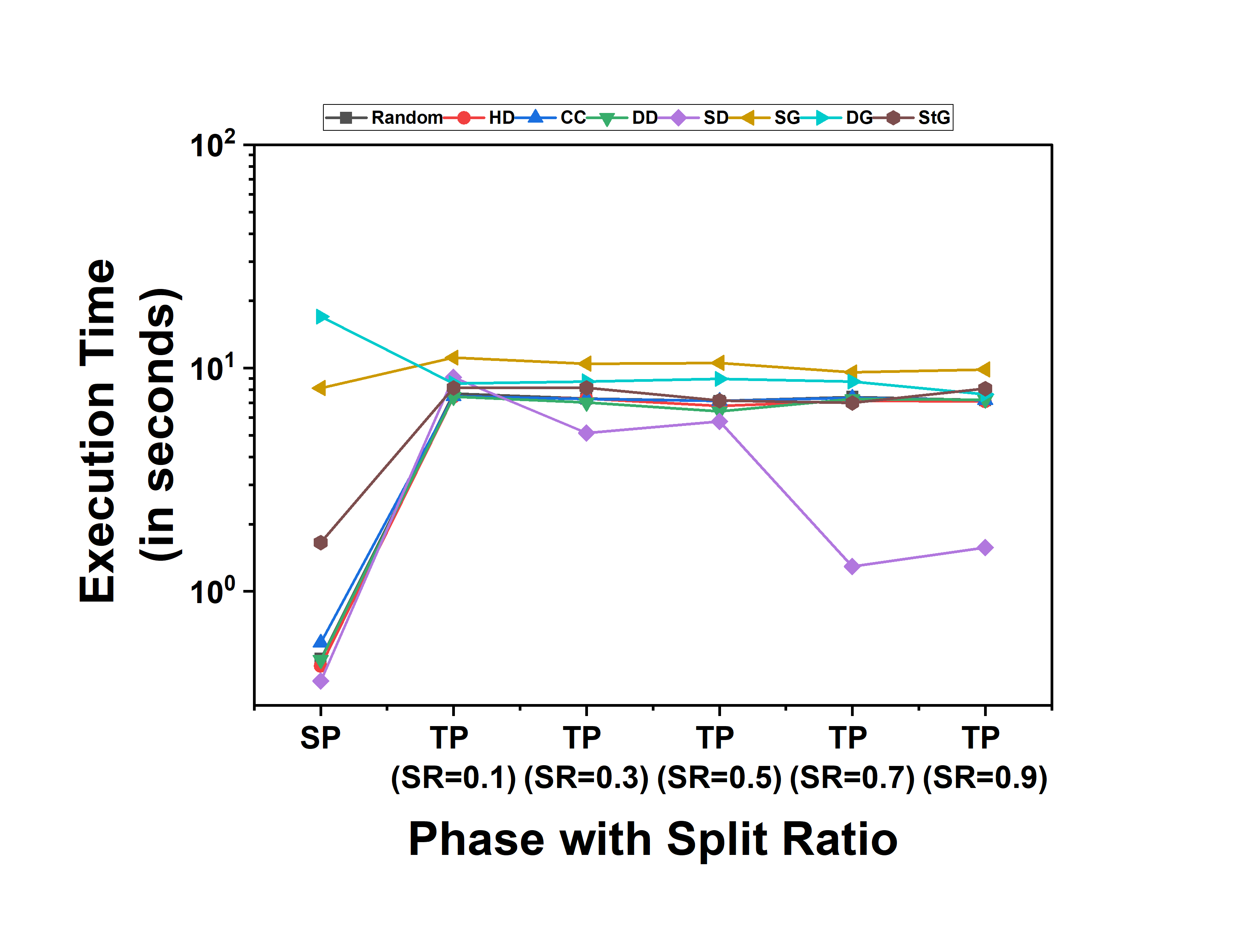}
        \caption{Budget 2500}
    \end{subfigure}

    \caption{Execution Time for Single Phase Vs.\ Two Phase across all Algorithms for Timestep 2, \textit{LM} Dataset, Probability Setting - Trivalency}
    \label{RQ6LM_T1}
\end{figure}

Figure~\ref{RQ6_T1} shows the computational time required by each algorithm for all budget values in both single-phase and two phase settings for the \textit{Email-Eu-Core} dataset under trivalency probability setting. For the \textbf{Random} algorithm, the time taken in the two phase setting is consistently and significantly higher than in the single-phase setting. For example, at a budget of $500$, the single-phase time was $1.10$ \textit{seconds}, while the two phase times ranged between $18.20$ \textit{seconds} and $18.32$ \textit{seconds}. One such instance with a split ratio of $0.1$ and timestep $2$ resulted in a time of $18.32$ \textit{seconds} (Figure~\ref{RQ6_T1}(a)). \textbf{HD} shows a similar trend. As seen in Figure~\ref{RQ6_T1}(a), for a budget of $500$, the single-phase time was $1.28$ \textit{seconds}, and the two phase times ranged from $18.66$ \textit{seconds} to $19.10$ \textit{seconds}. For the configuration with a split ratio of $0.1$ and timestep $2$, the time reached $19.10$ \textit{seconds}. In the case of \textbf{HighCC}, the difference is also notable. For a budget of $500$, the single-phase time was $1.18$ \textit{seconds}. In contrast, two phase times ranged from $18.69$ \textit{seconds} to $19.34$ \textit{seconds}, with one configuration at split ratio $0.1$ and timestep $2$ taking $18.91$ \textit{seconds} (Figure~\ref{RQ6_T1}(a)). \textbf{DD} also demonstrates much higher computation times in the two phase setup. For a budget of $500$, the single-phase time was $1.14$ \textit{seconds}. In one two phase instance with a split ratio of $0.1$ and timestep $2$, the time rose to $20.01$s, within the overall range of $19.07$ \textit{seconds} to $20.04$ \textit{seconds}. As for \textbf{SD}, the two phase times show greater variation. At a budget of $500$, \textbf{SD} took $1.23$ \textit{seconds} in the single-phase setting. Two phase times ranged from $1.19$ \textit{seconds} to $16.14$ \textit{seconds}, with $16.14$ \textit{seconds} recorded at split ratio $0.1$ and timestep $2$. The \textbf{SG} algorithm shows a dramatic increase in two phase time. While the single-phase execution took $16.59$ \textit{seconds}, the two phase setting required between $102.97$ \textit{seconds} and $131.08$ \textit{seconds}. In Figure~\ref{RQ6_T1}(a), one configuration with split ratio $0.1$ and timestep $2$ took $131.08$ \textit{seconds}. \textbf{DG} also exhibits substantially longer two phase times. At a budget of $500$, the single-phase execution lasted $7.80$ \textit{seconds}, while two phase times ranged from $20.93$ \textit{seconds} to $24.89$ \textit{seconds}. For split ratio $0.1$ and timestep $2$, \textbf{DG} took $24.89$ \textit{seconds} (Figure~\ref{RQ6_T1}(a)). At a higher budget of $2500$, \textbf{DG}’s single-phase time was $34.19$ \textit{seconds}, whereas two phase values ranged from $39.46$ \textit{seconds} to $49.33$ \textit{seconds} occurring at split ratio $0.9$ and timestep $2$ (Figure~\ref{RQ6_T1}(e)). In the case of \textbf{StG0.1}, the computational time was also higher in two phase. As seen in Figure~\ref{RQ6_T1}(a), the single-phase setting took $4.76$ \textit{seconds}. Two phase values ranged from $18.76$ \textit{seconds} to $20.18$ \textit{seconds}, with the highest value recorded at split ratio $0.1$ and timestep $2$. Across all algorithms, the two phase method demands significantly more computation time than the single-phase version. For algorithms with low single-phase times like \textbf{Random, HD, CC,} and \textbf{DD} the two phase execution often takes ten times longer. In the case of already heavy algorithms like \textbf{SG} and \textbf{DG}, two phase execution further increases the time cost, which might limit their use in large-scale or real-time scenarios. Interestingly, \textbf{SD} sometimes shows comparable times in both settings, which indicates that its computational cost is sensitive to parameter selection. Lastly, \textbf{StG0.1} while slower in two phase is still relatively efficient compared to \textbf{SG} and \textbf{DG}.

\begin{figure}[htbp]
    \centering

    \begin{subfigure}[t]{0.3\linewidth}
        \centering
        \includegraphics[width=\linewidth]{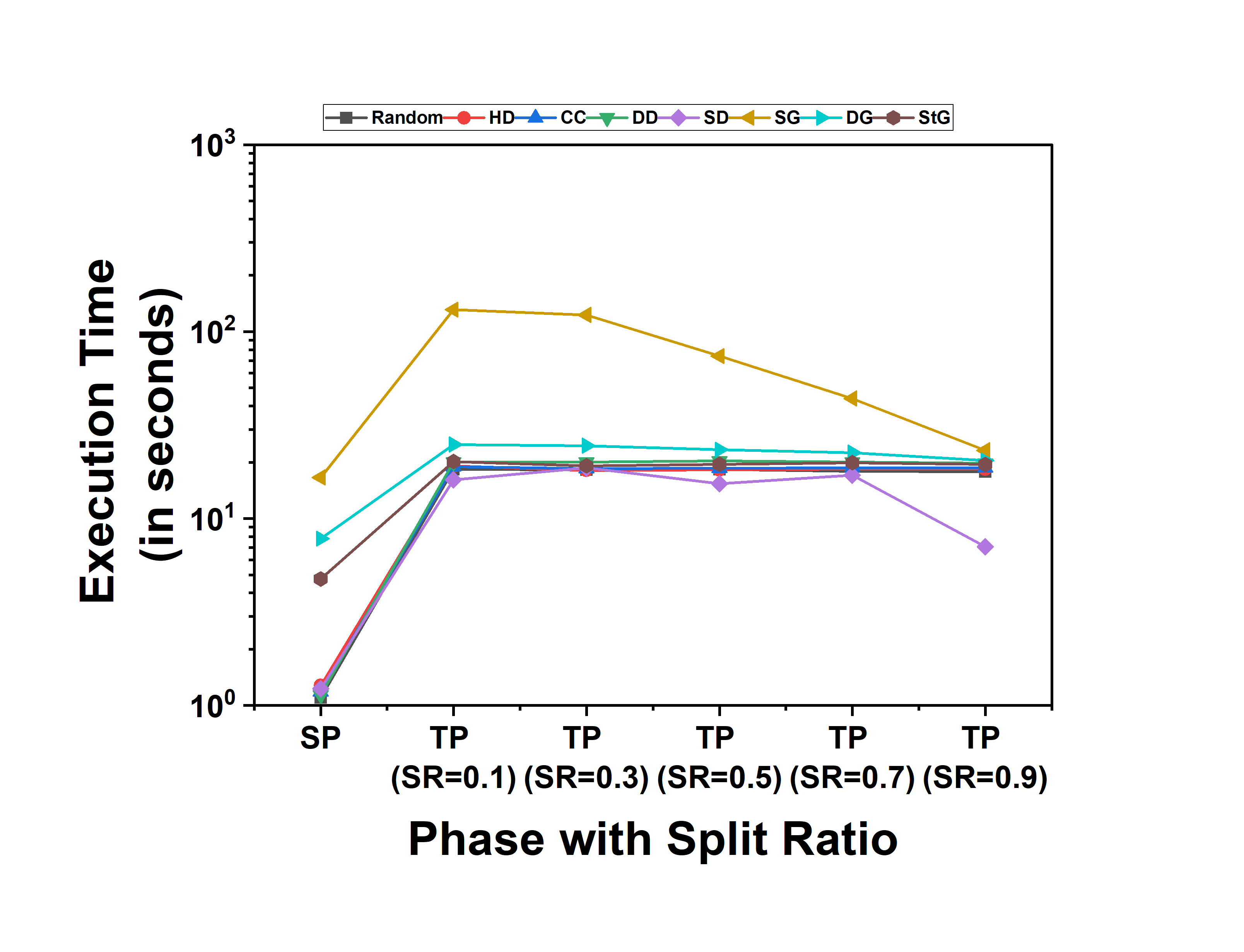}
        \caption{Budget 500}
    \end{subfigure}
    \hspace{0.05\linewidth}
    \begin{subfigure}[t]{0.3\linewidth}
        \centering
        \includegraphics[width=\linewidth]{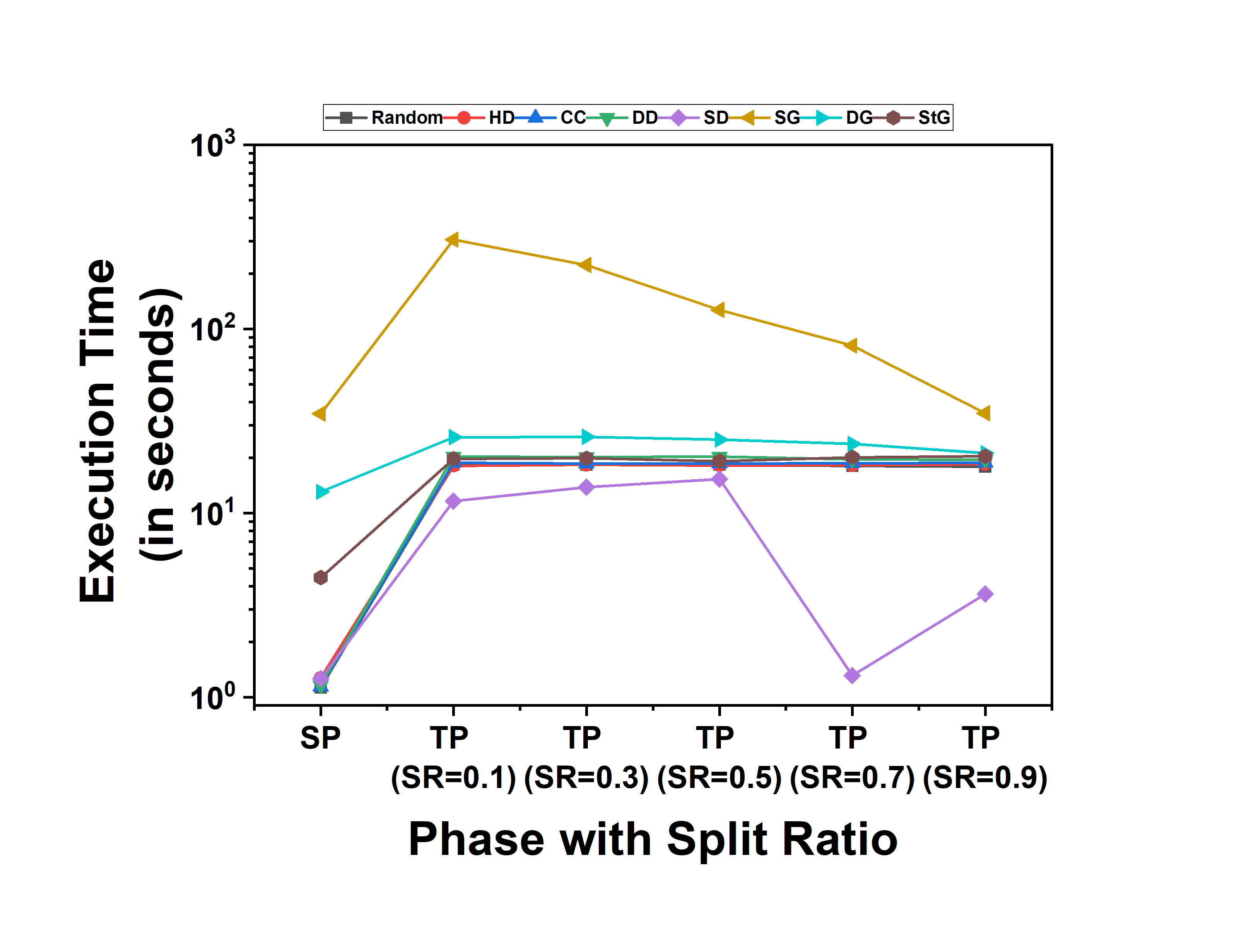}
        \caption{Budget 1000}
    \end{subfigure}

    \vspace{0.5cm}

    \begin{subfigure}[t]{0.3\linewidth}
        \centering
        \includegraphics[width=\linewidth]{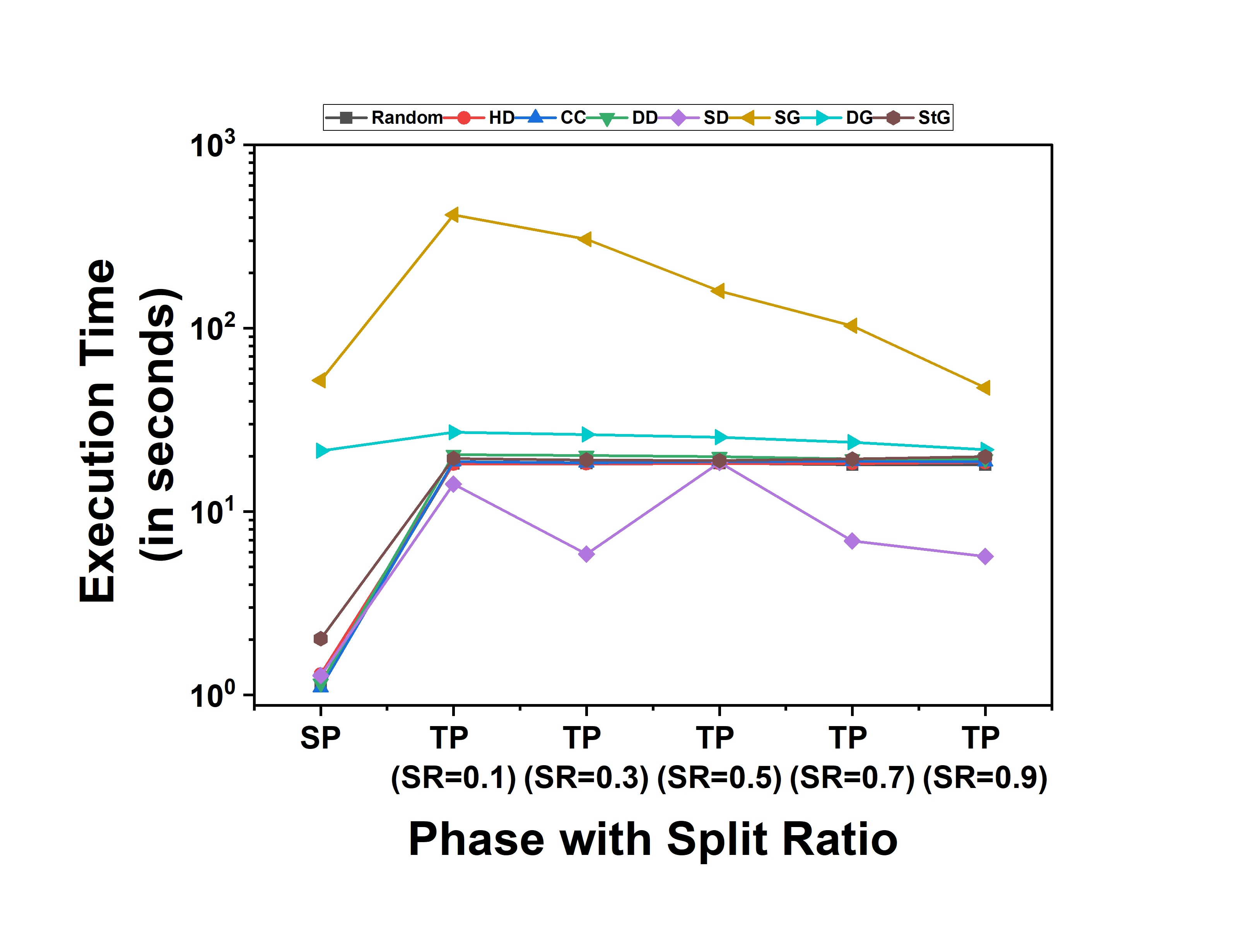}
        \caption{Budget 1500}
    \end{subfigure}
    \hfill
    \begin{subfigure}[t]{0.3\linewidth}
        \centering
        \includegraphics[width=\linewidth]{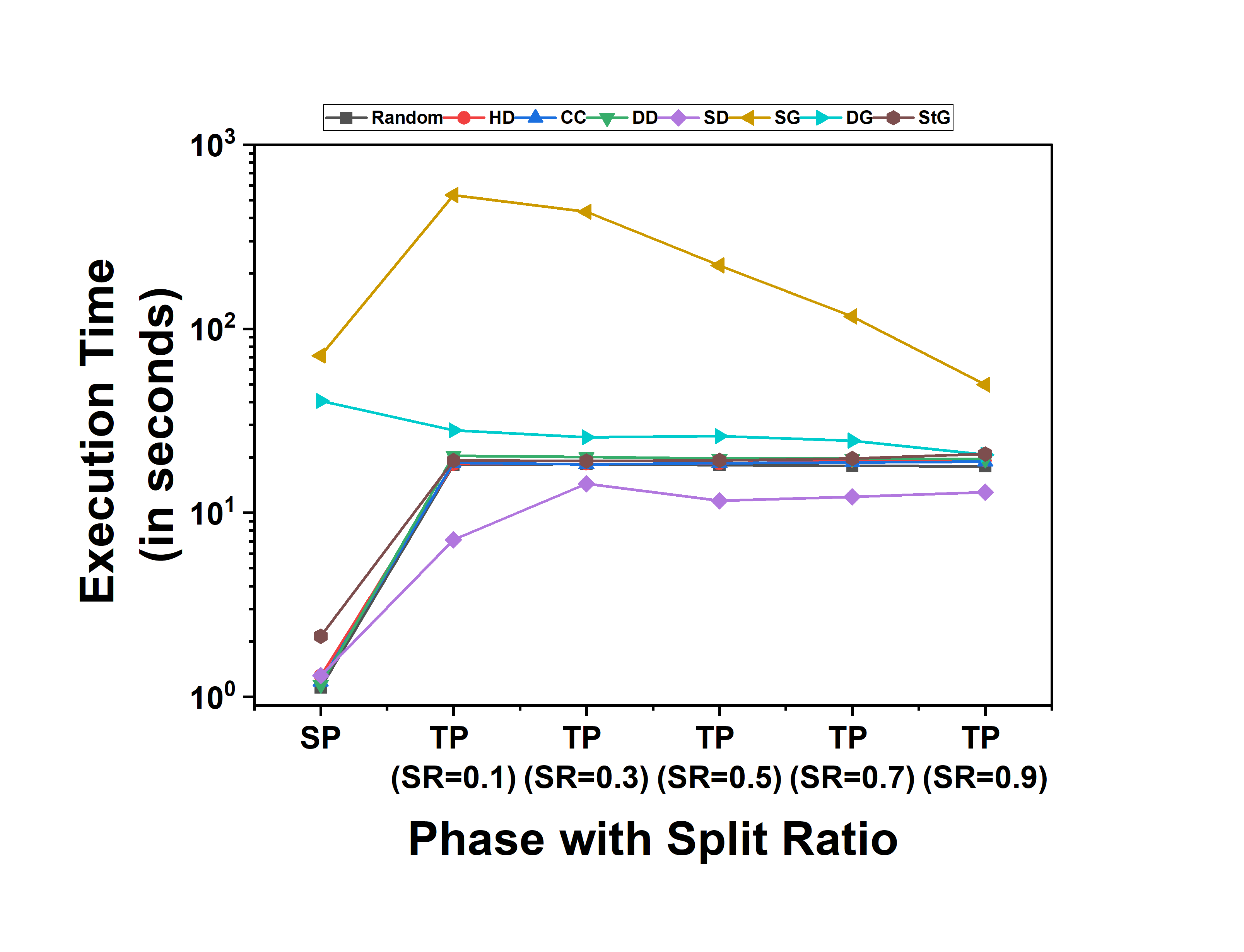}
        \caption{Budget 2000}
    \end{subfigure}
    \hfill
    \begin{subfigure}[t]{0.3\linewidth}
        \centering
        \includegraphics[width=\linewidth]{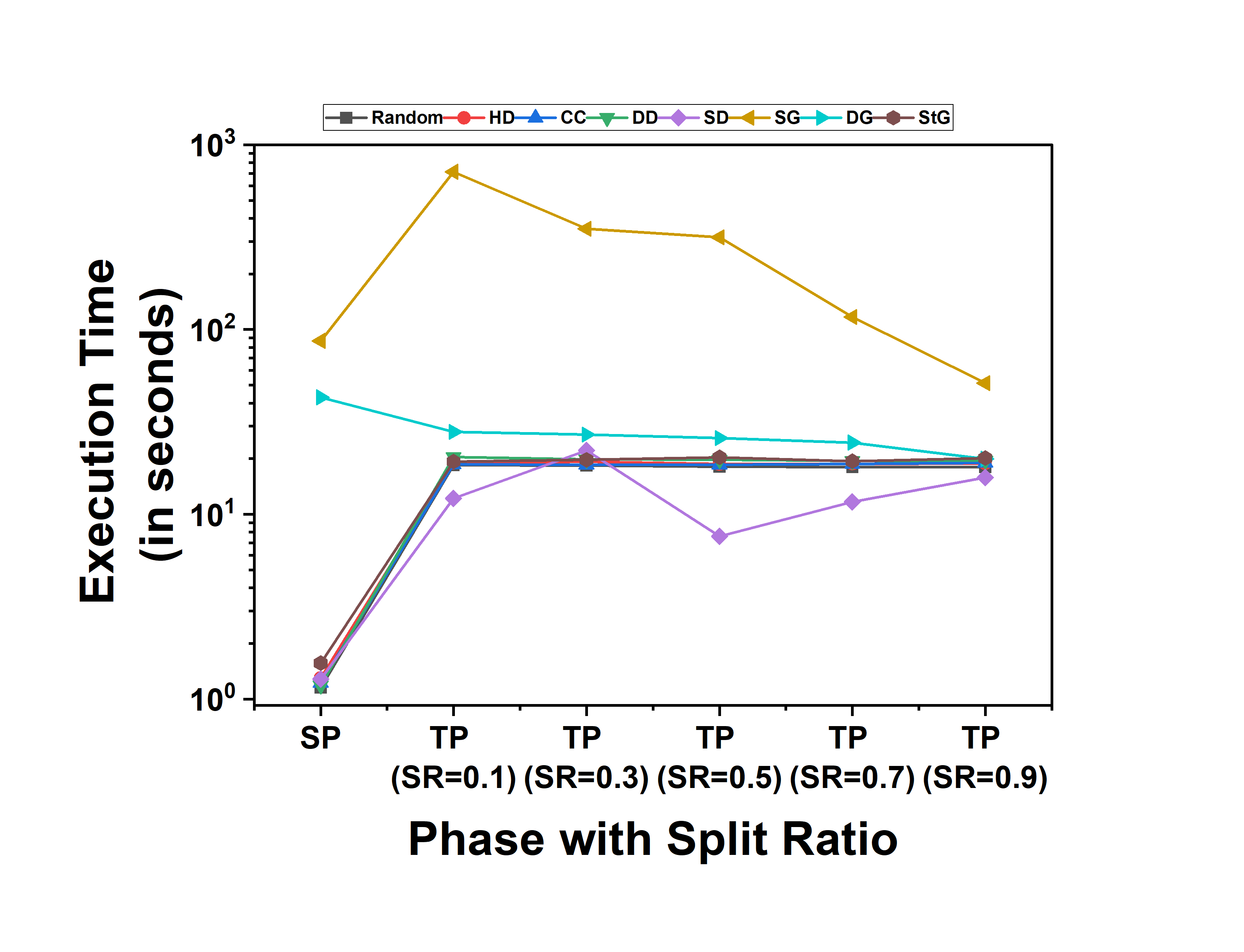}
        \caption{Budget 2500}
    \end{subfigure}

    \caption{Execution Time for Single Phase Vs. Two Phase across all Algorithms for Timestep 2, \textit{Email-Eu-Core} Dataset, Probability Setting - Trivalency}
    \label{RQ6_T1}
\end{figure}

\subsubsection{Impact of $\epsilon$ on Profit of Stochastic Single-Phase and Two phase Settings}

To understand the effect of the parameter $\epsilon$ on profit in the stochastic greedy algorithm for the \textit{LM} dataset, we compare the exact results of \textbf{StG0.01}, \textbf{StG0.1}, \textbf{StG0.3}, and \textbf{StG0.6} in both single-phase and two phase settings at a fixed budget of $500$ in Figure~\ref{Fig:RQ7LM_T1}. In the single-phase setting, the best result is observed for $\epsilon = 0.01$ with a profit of $9526.36$. This clearly outperforms $\epsilon = 0.1$ ($8320.35$), $\epsilon = 0.3$ ($8820.46$), and $\epsilon = 0.6$ ($7378.80$), confirming that lower $\epsilon$ values are more effective in the single-phase scenario. In the two phase setting, we observe variation across split ratios. For instance, at timestep $2$ and split ratio $0.7$, $\epsilon = 0.01$ achieves the highest profit of $11553.10$ in Figure~\ref{Fig:RQ7LM_T1}(m), outperforming all other variants. At the same configuration, $\epsilon = 0.6$ gives $10677.08$, $\epsilon = 0.1$ gives $10631.02$, and $\epsilon = 0.3$ performs the worst with $7301.4$. At timestep $4$ and split ratio $0.3$, $\epsilon = 0.01$ once again gives the best result of $10171.89$ (Figure~\ref{Fig:RQ7LM_T1}(e)), while $\epsilon = 0.6$ trails at $9969.78$, $\epsilon = 0.1$ gives $8326.98$, and $\epsilon = 0.3$ is lowest with $8109.62$. At timestep $8$, the trend continues. For example, with split ratio $0.9$, $\epsilon = 0.01$ achieves a peak profit of $11709.75$, while $\epsilon = 0.6$ also does well at $7435.43$. However, $\epsilon = 0.1$ and $\epsilon = 0.3$ remain behind at $10343.97$ and $9362.31$ respectively. These instance-based results confirm that $\epsilon = 0.01$ consistently leads to higher profits across split ratios and timesteps in the two phase setting. However, since lower $\epsilon$ leads to higher computational cost due to greedy behavior and repeated simulations, it may not be the best choice for large networks. In such cases, $\epsilon = 0.6$ appears as a reasonable trade-off as it gives competitive profits while offering better runtime efficiency.
\begin{figure}[htbp]
\centering
\captionsetup[sub]{font=footnotesize}
\begin{tabular}{cccc}
    \begin{subfigure}[t]{0.22\textwidth}
        \includegraphics[width=\linewidth]{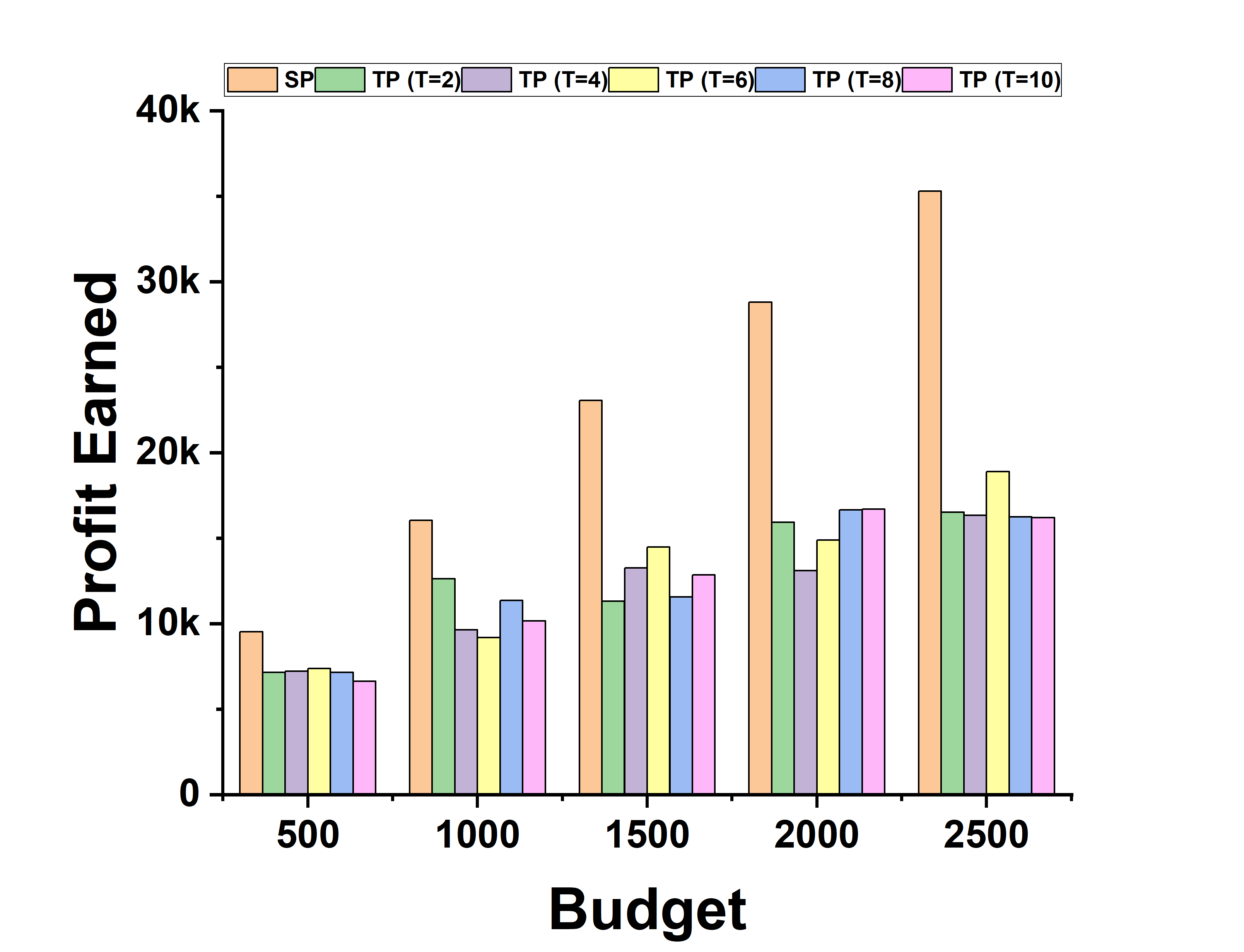}
        \caption{Epsilon=0.01, Split Ratio=10\%}
    \end{subfigure} &
    \begin{subfigure}[t]{0.22\textwidth}
        \includegraphics[width=\linewidth]{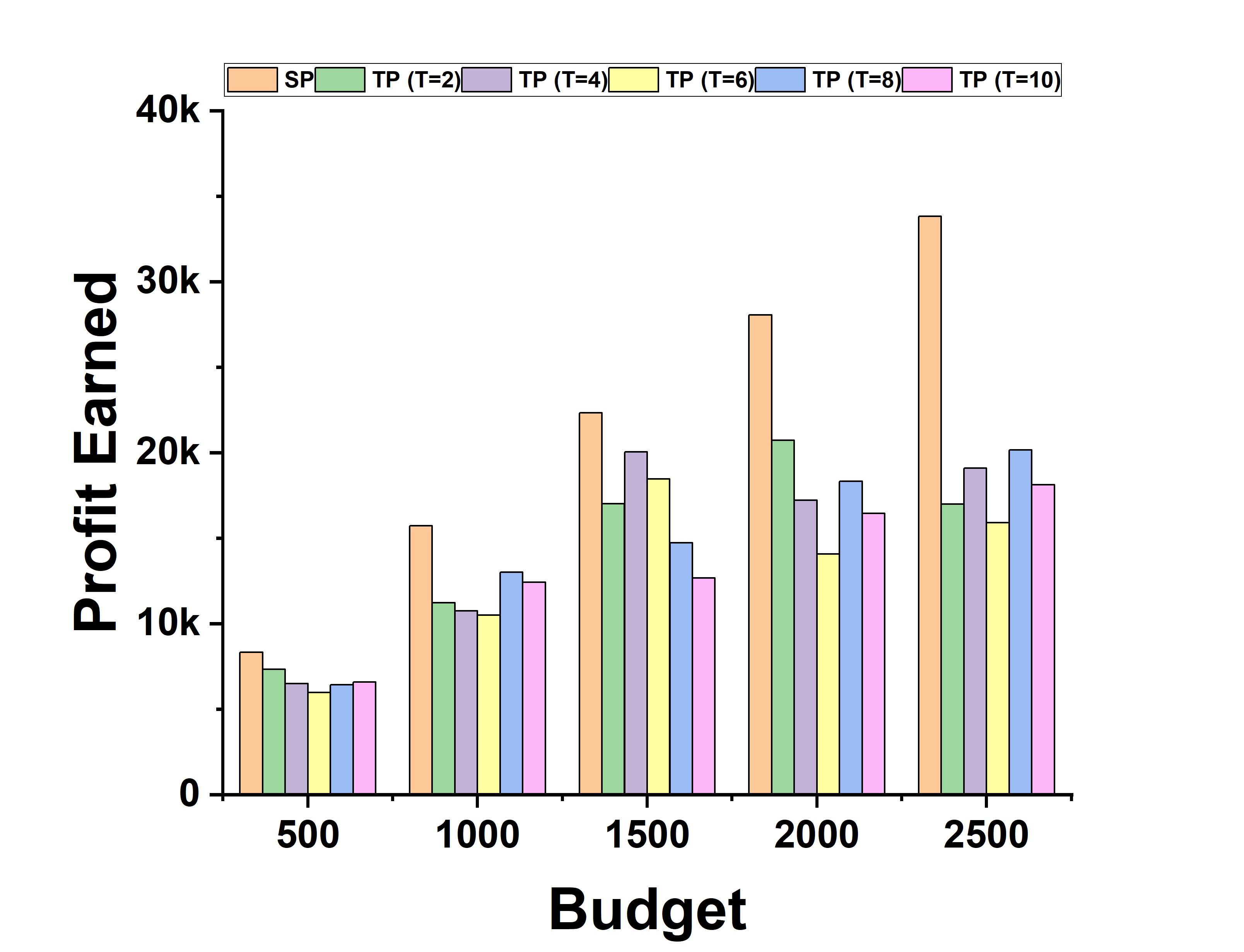}
        \caption{Epsilon=0.1, Split Ratio=10\%}
    \end{subfigure} &
    \begin{subfigure}[t]{0.22\textwidth}
        \includegraphics[width=\linewidth]{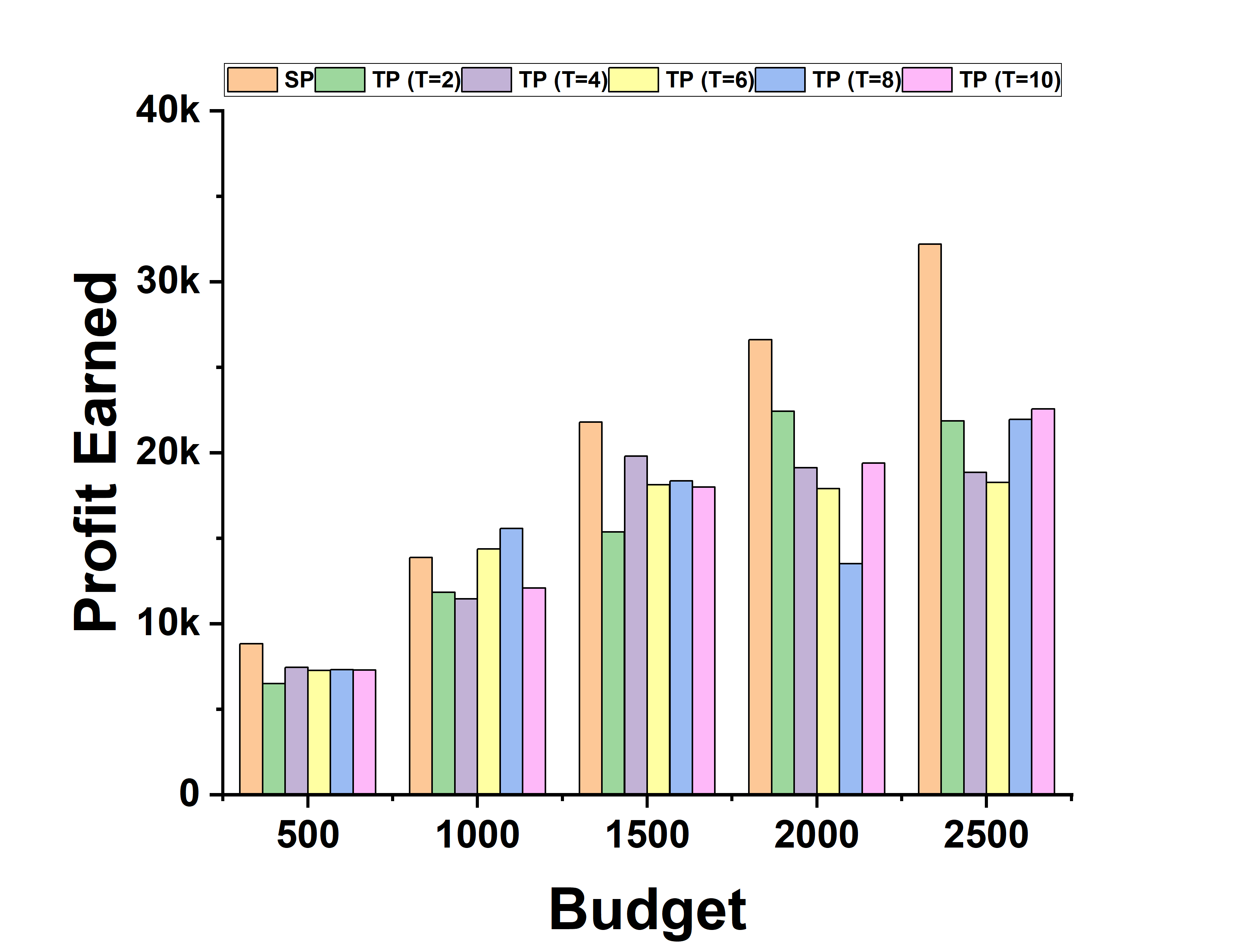}
        \caption{Epsilon=0.3, Split Ratio=10\%}
    \end{subfigure} &
    \begin{subfigure}[t]{0.22\textwidth}
        \includegraphics[width=\linewidth]{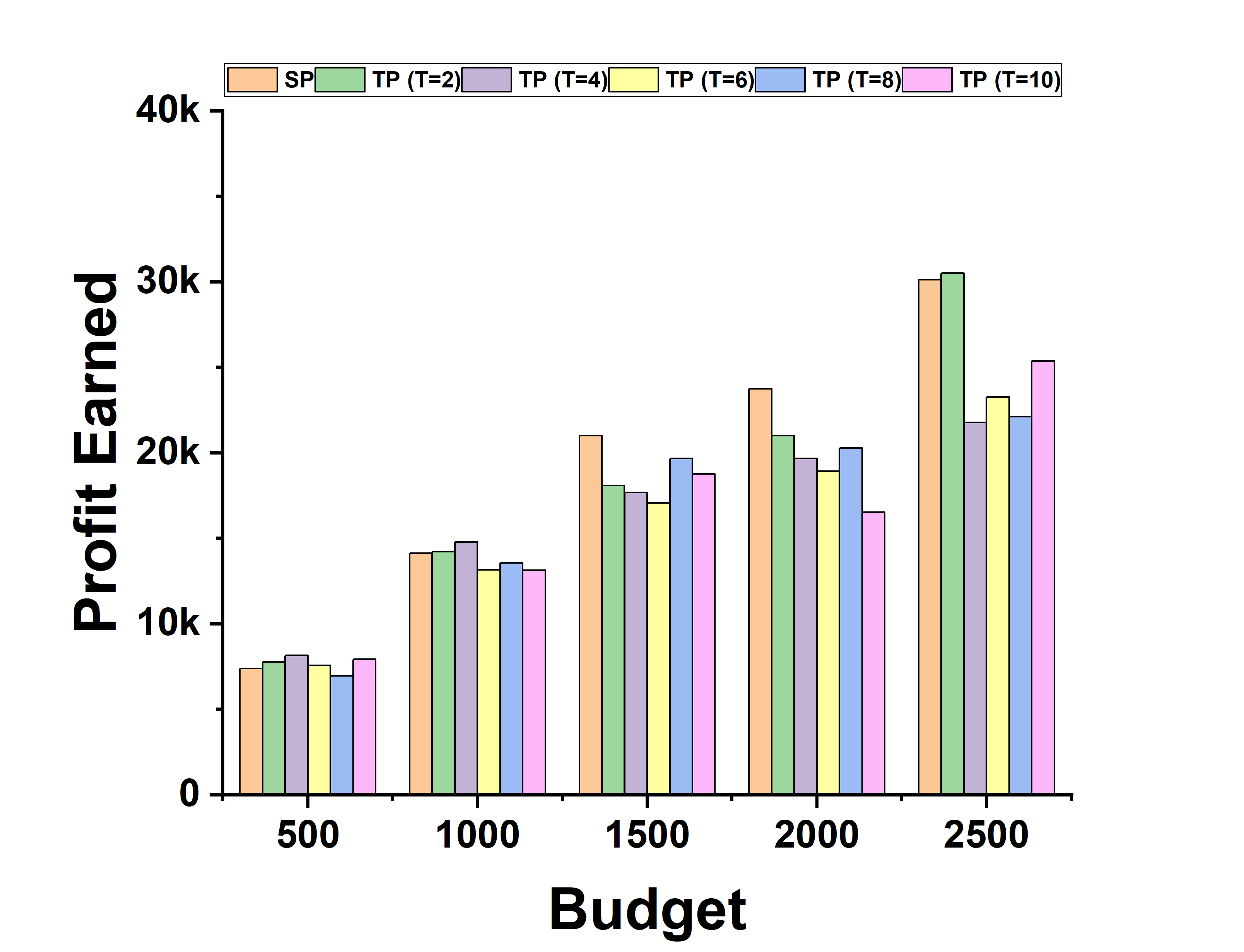}
        \caption{Epsilon=0.6, Split Ratio=10\%}
    \end{subfigure} \\[6pt]

    \begin{subfigure}[t]{0.22\textwidth}
        \includegraphics[width=\linewidth]{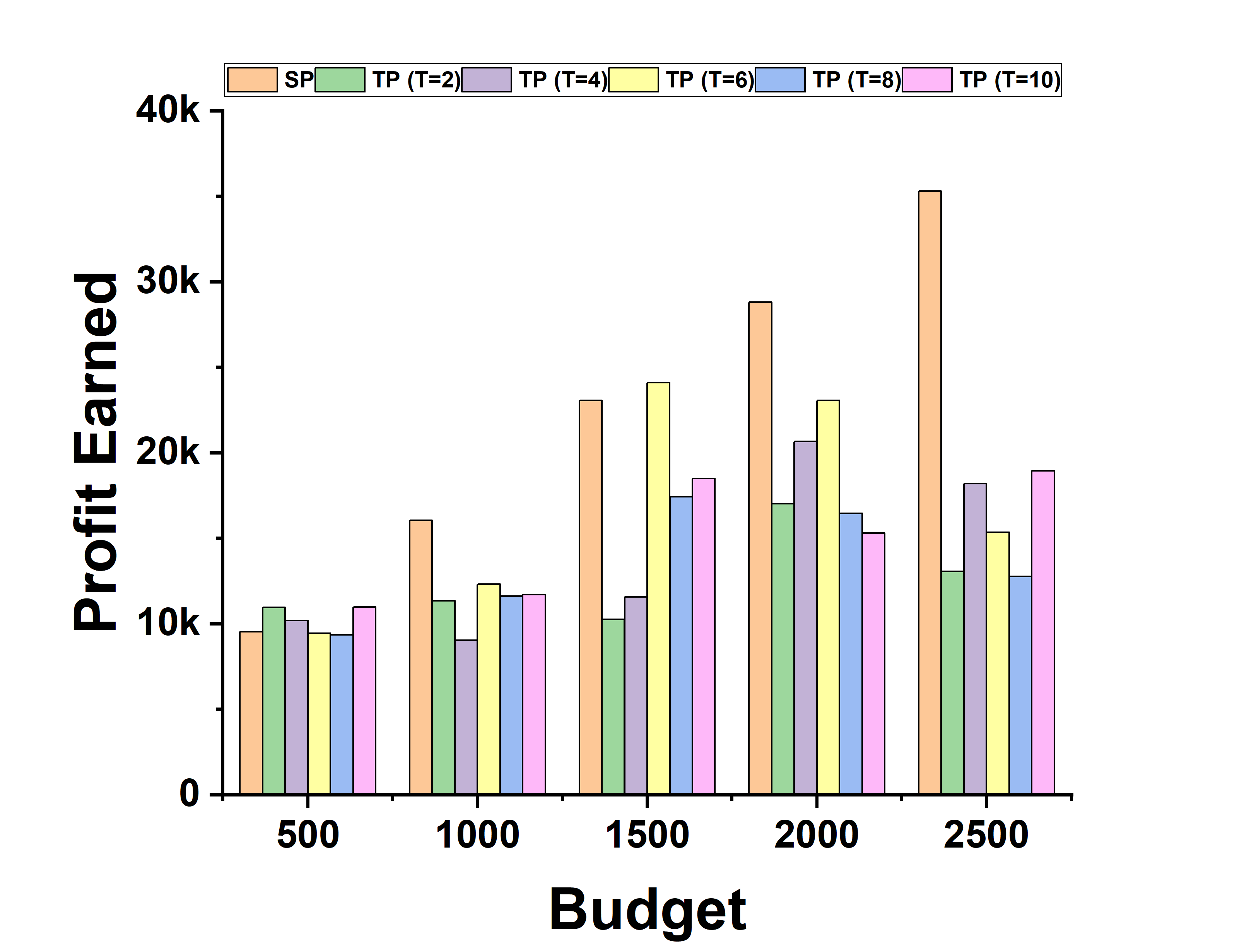}
        \caption{Epsilon=0.01, Split Ratio=30\%}
    \end{subfigure} &
    \begin{subfigure}[t]{0.22\textwidth}
        \includegraphics[width=\linewidth]{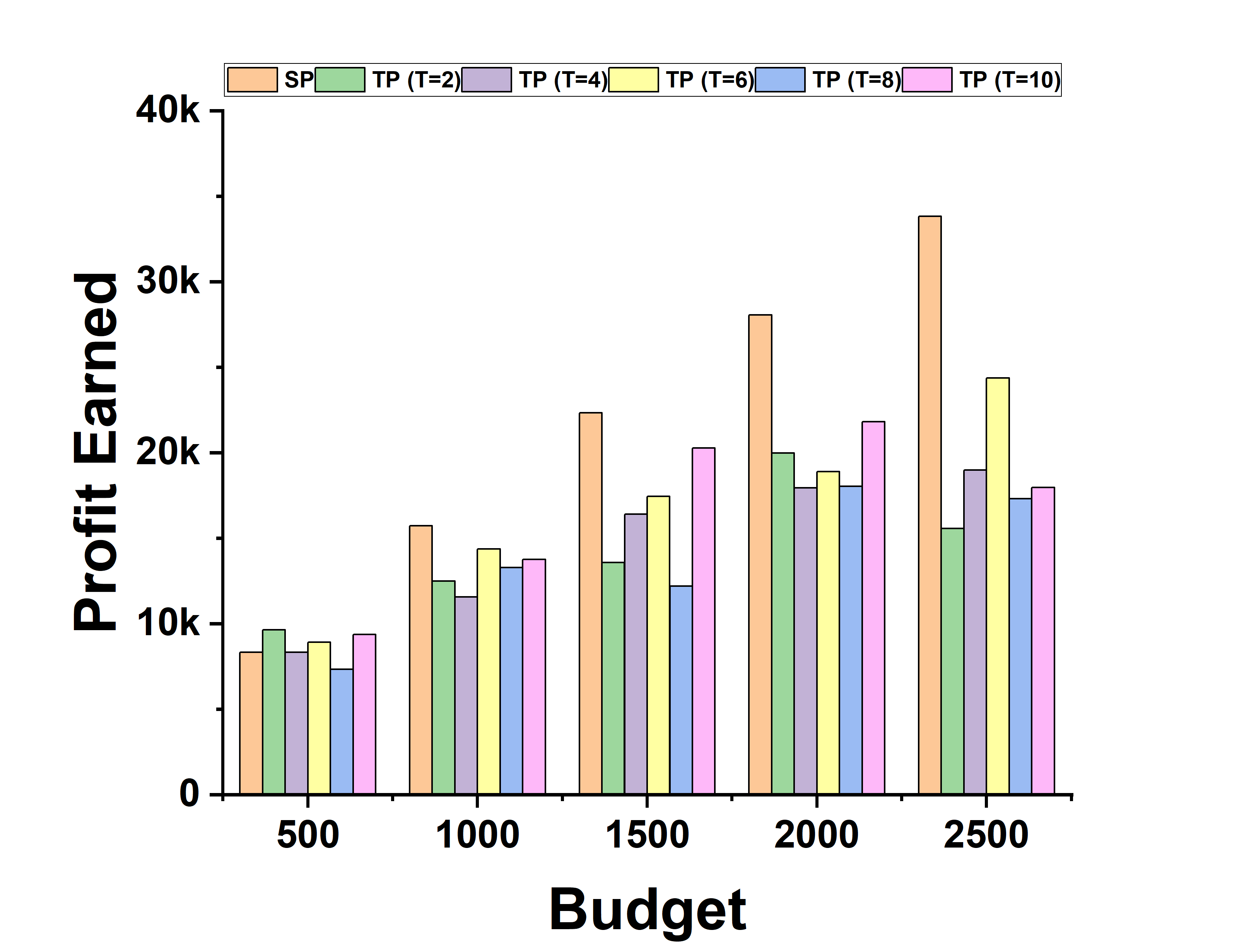}
        \caption{Epsilon=0.1, Split Ratio=30\%}
    \end{subfigure} &
    \begin{subfigure}[t]{0.22\textwidth}
        \includegraphics[width=\linewidth]{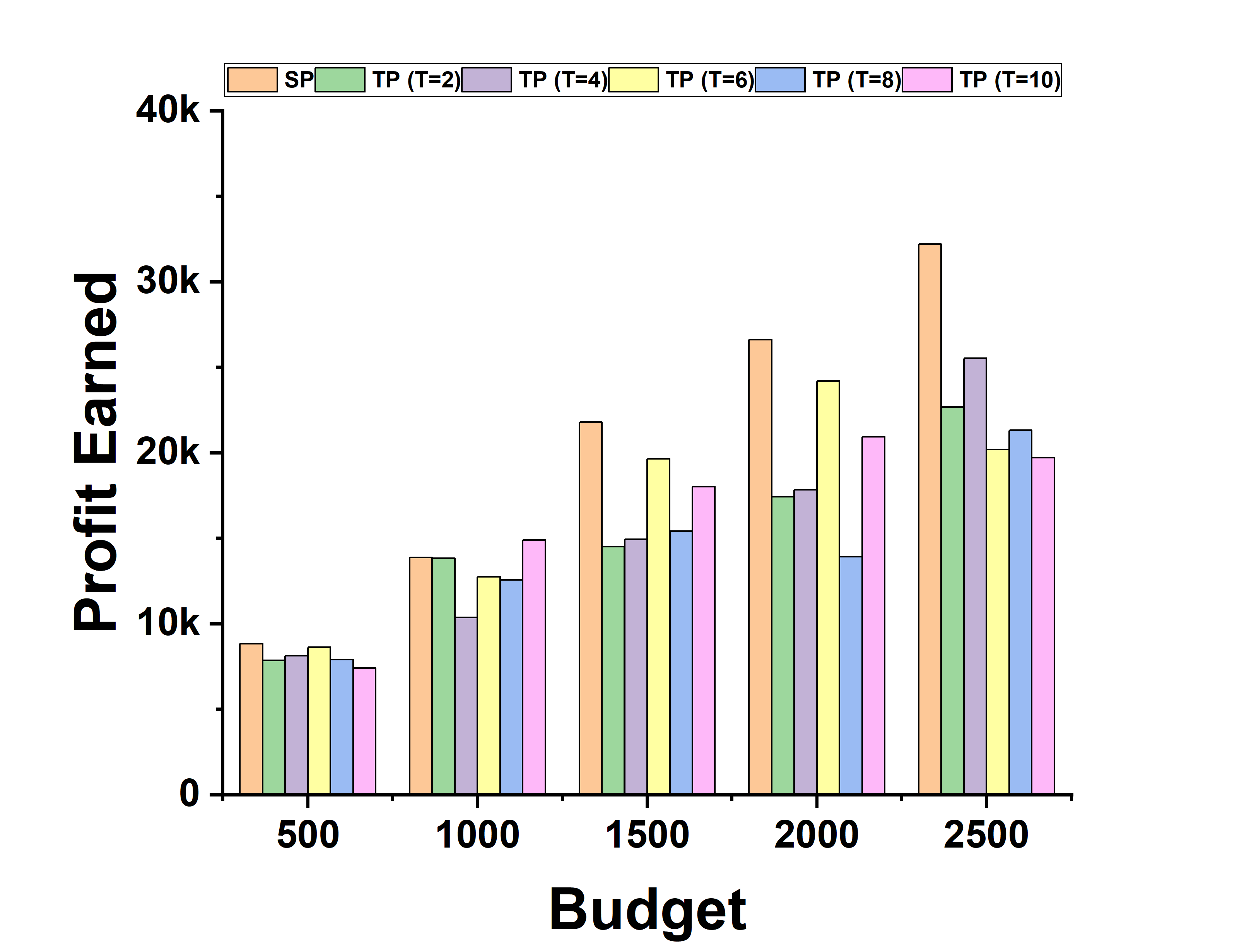}
        \caption{Epsilon=0.3, Split Ratio=30\%}
    \end{subfigure} &
    \begin{subfigure}[t]{0.22\textwidth}
        \includegraphics[width=\linewidth]{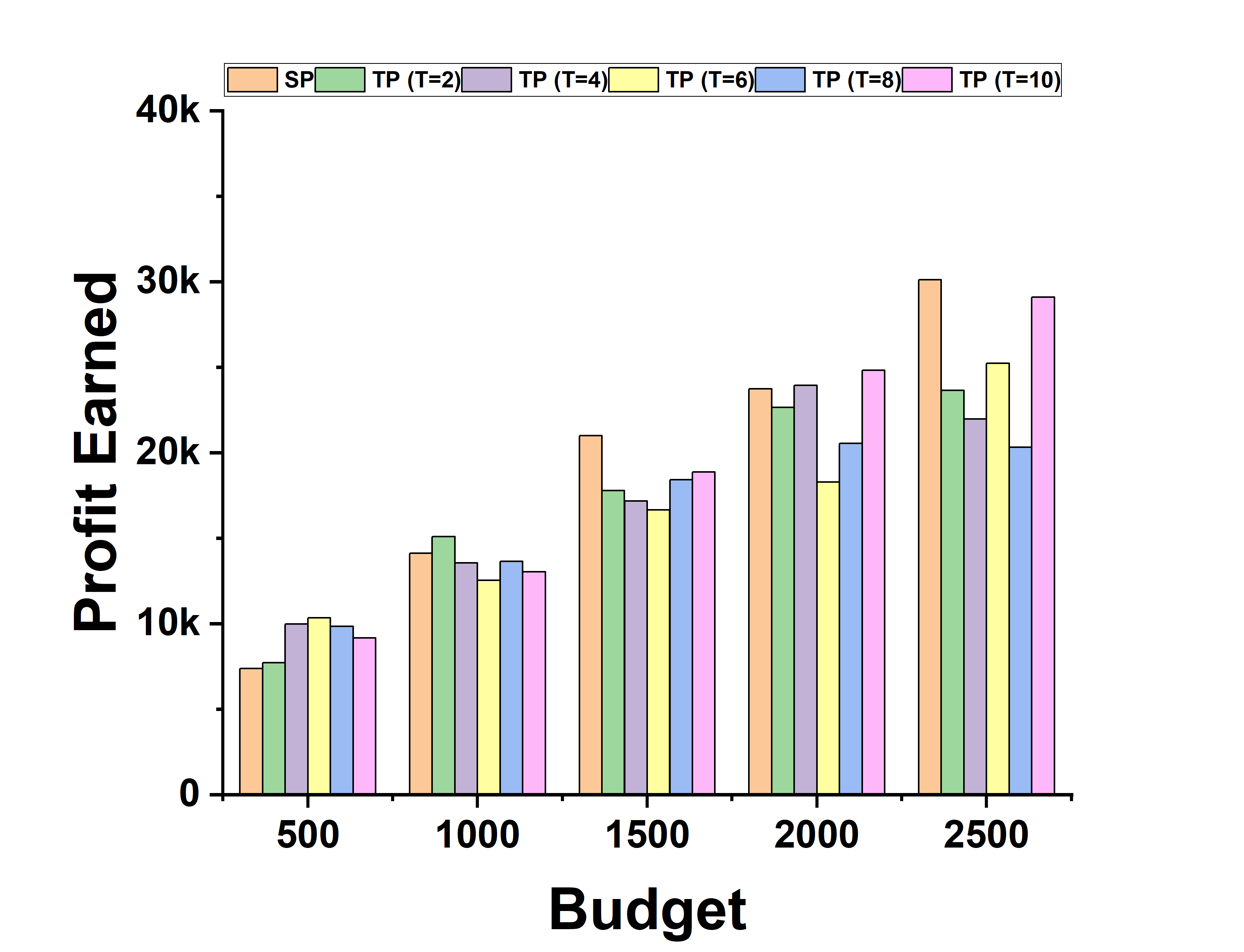}
        \caption{Epsilon=0.6, Split Ratio=30\%}
    \end{subfigure} \\[6pt]

    \begin{subfigure}[t]{0.22\textwidth}
        \includegraphics[width=\linewidth]{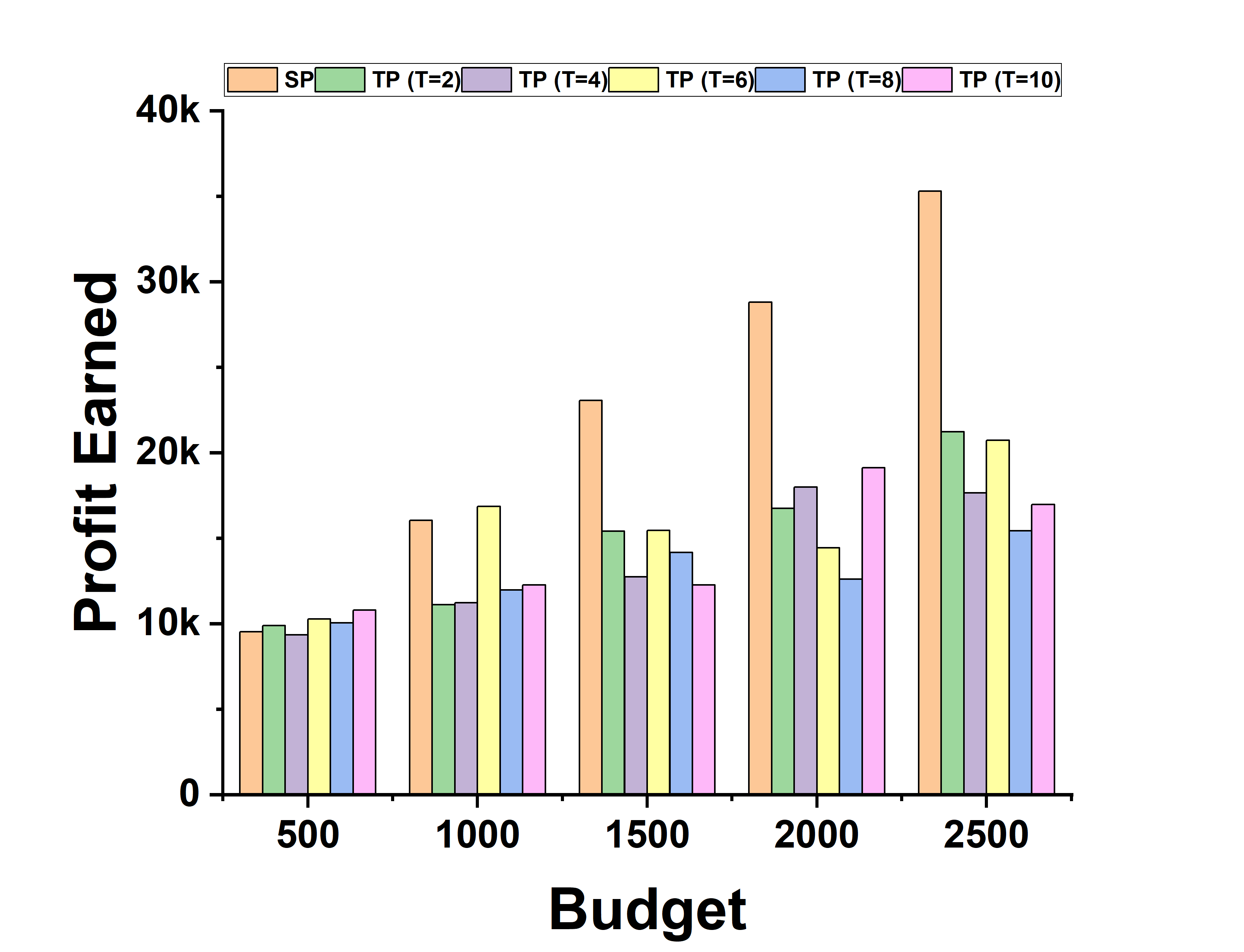}
        \caption{Epsilon=0.01, Split Ratio=50\%}
    \end{subfigure} &
    \begin{subfigure}[t]{0.22\textwidth}
        \includegraphics[width=\linewidth]{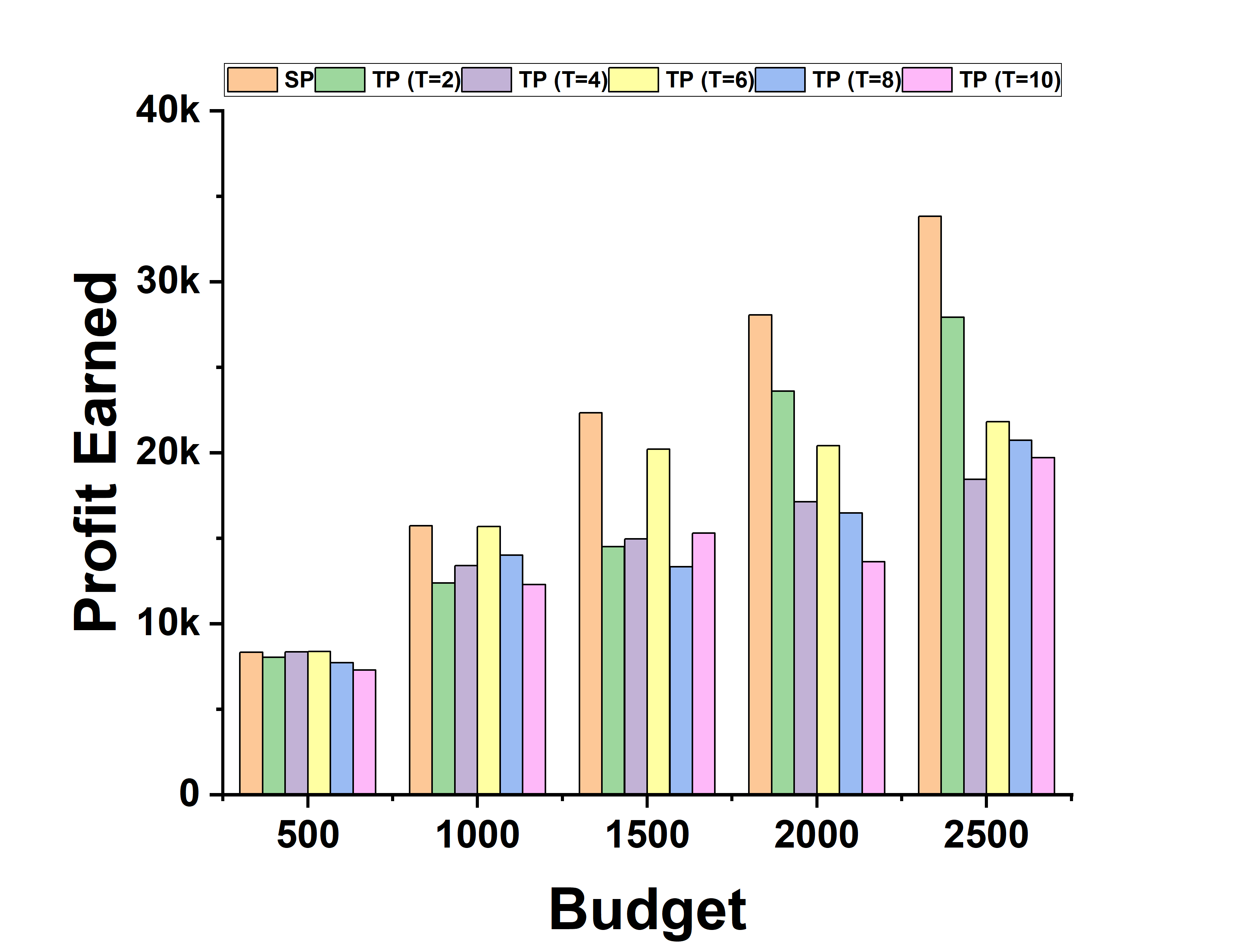}
        \caption{Epsilon=0.1, Split Ratio=50\%}
    \end{subfigure} &
    \begin{subfigure}[t]{0.22\textwidth}
        \includegraphics[width=\linewidth]{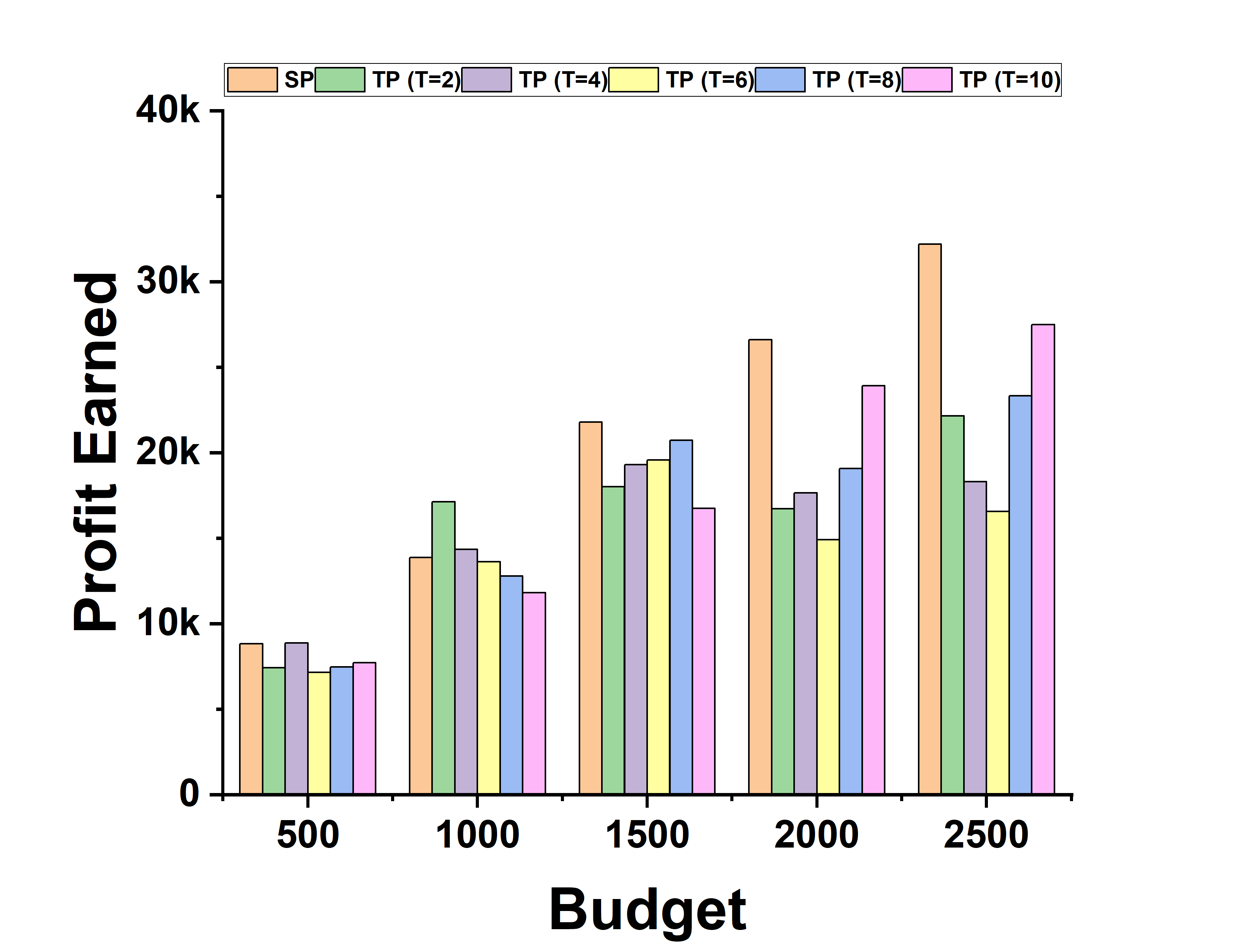}
        \caption{Epsilon=0.3, Split Ratio=50\%}
    \end{subfigure} &
    \begin{subfigure}[t]{0.22\textwidth}
        \includegraphics[width=\linewidth]{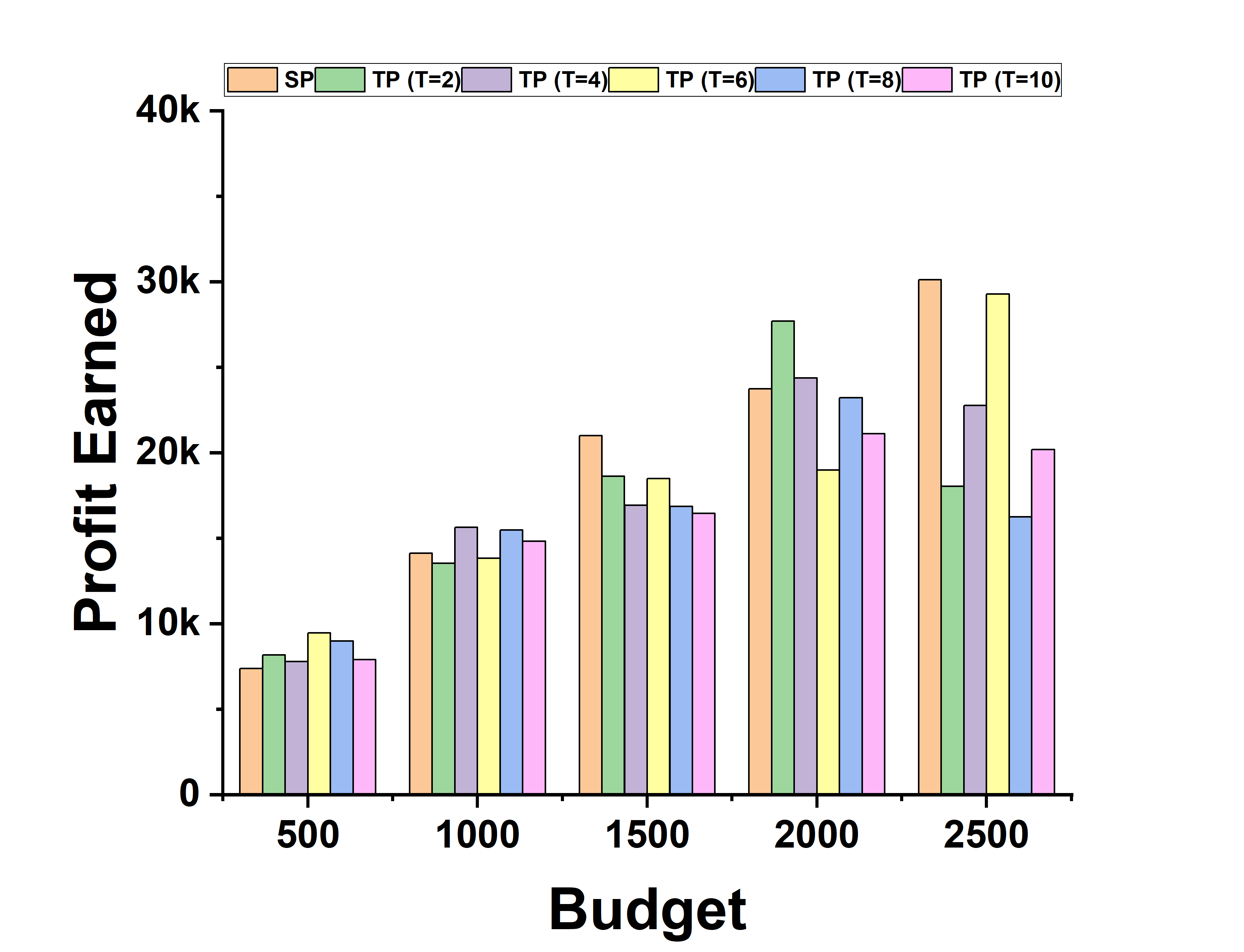}
        \caption{Epsilon=0.6, Split Ratio=50\%}
    \end{subfigure} \\[6pt]

    \begin{subfigure}[t]{0.22\textwidth}
        \includegraphics[width=\linewidth]{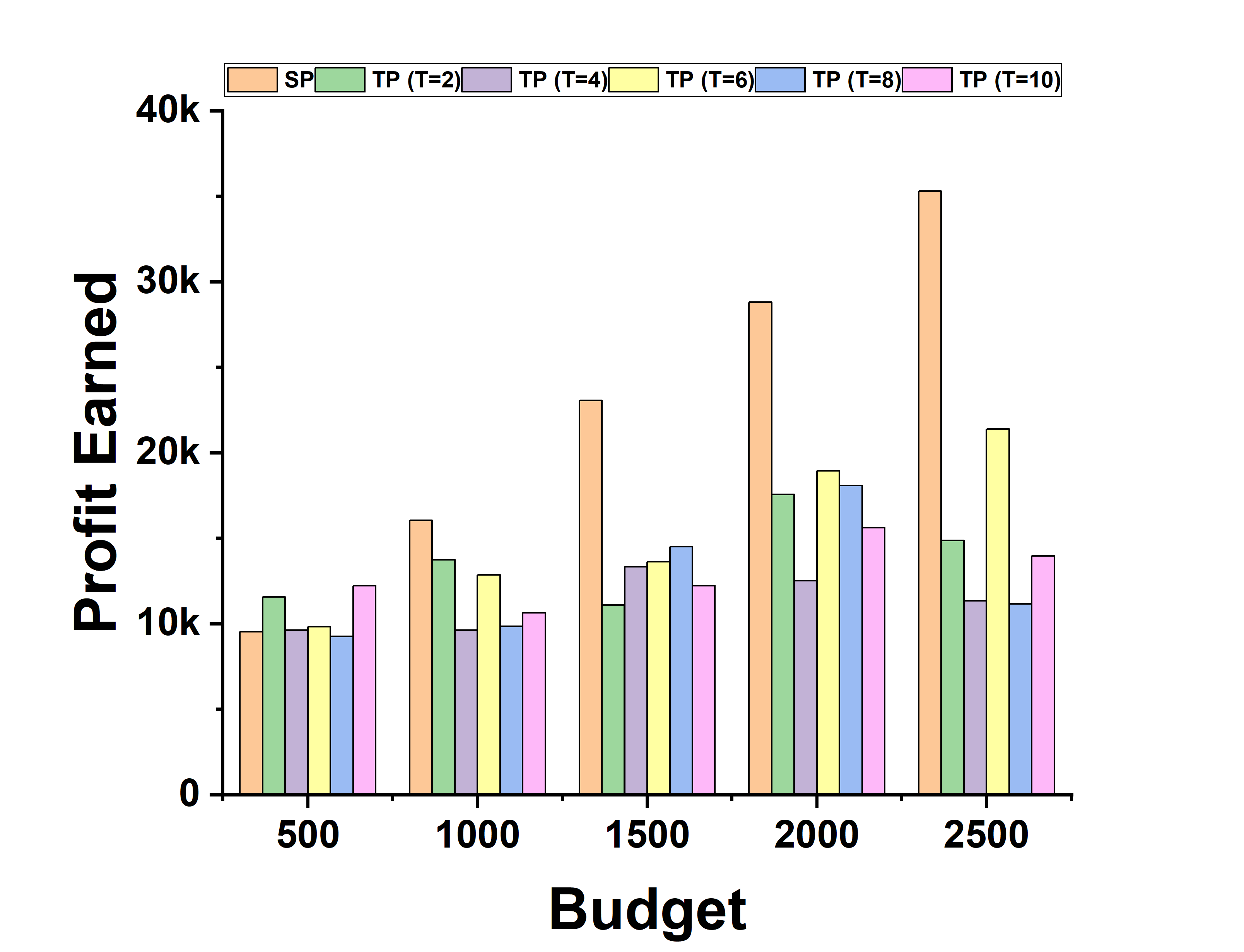}
        \caption{Epsilon=0.01, Split Ratio=70\%}
    \end{subfigure} &
    \begin{subfigure}[t]{0.22\textwidth}
        \includegraphics[width=\linewidth]{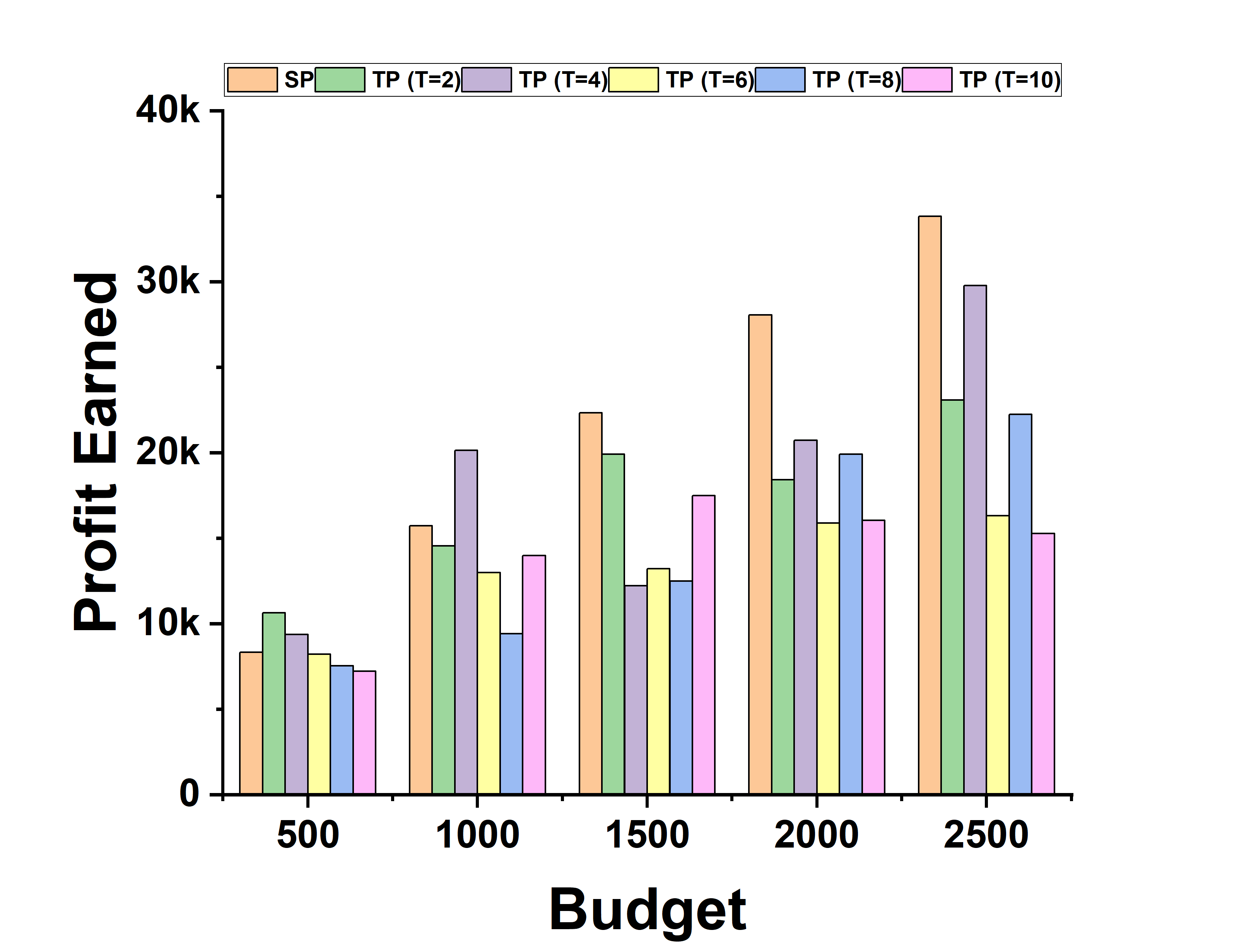}
        \caption{Epsilon=0.1, Split Ratio=70\%}
    \end{subfigure} &
    \begin{subfigure}[t]{0.22\textwidth}
        \includegraphics[width=\linewidth]{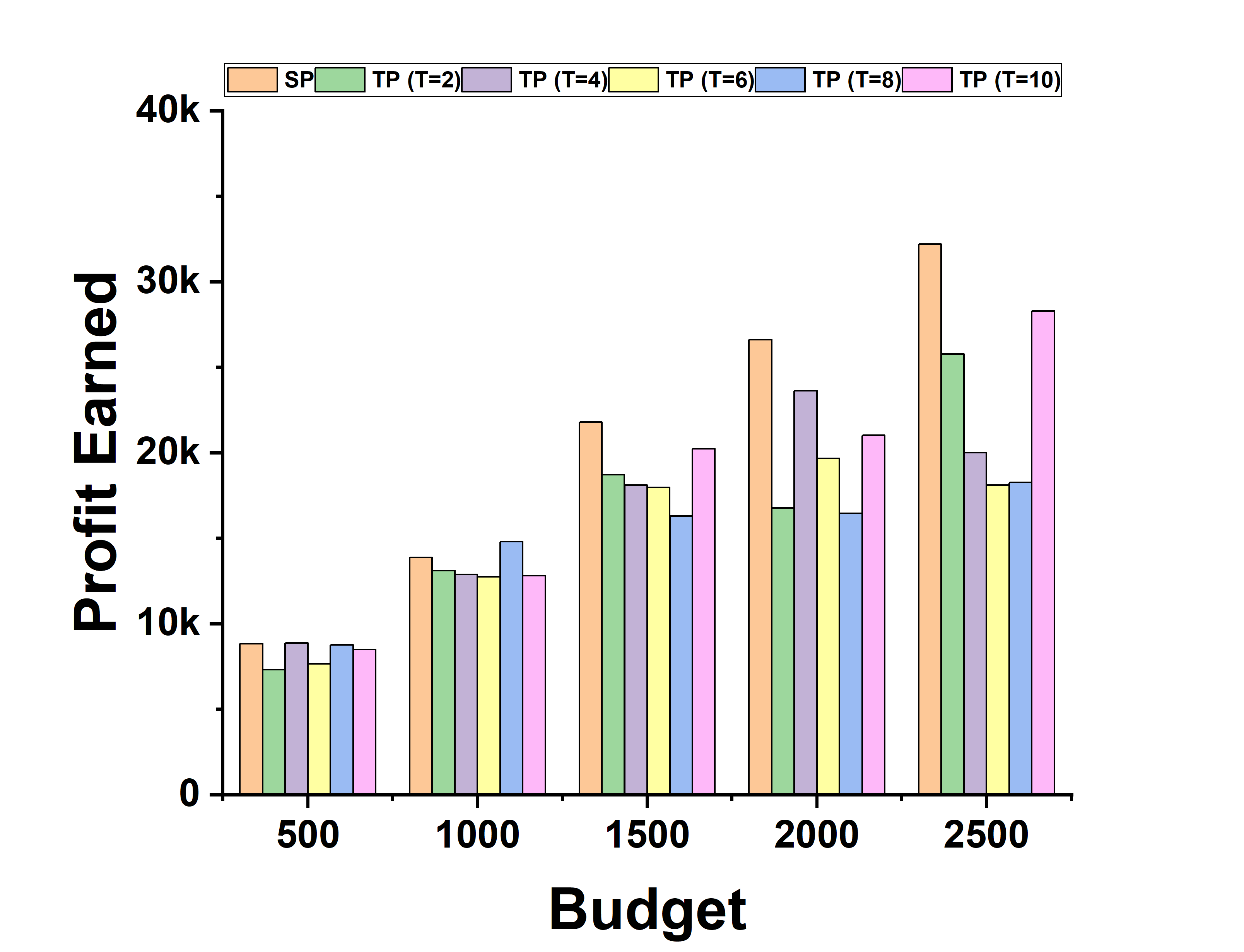}
        \caption{Epsilon=0.3, Split Ratio=70\%}
    \end{subfigure} &
    \begin{subfigure}[t]{0.22\textwidth}
        \includegraphics[width=\linewidth]{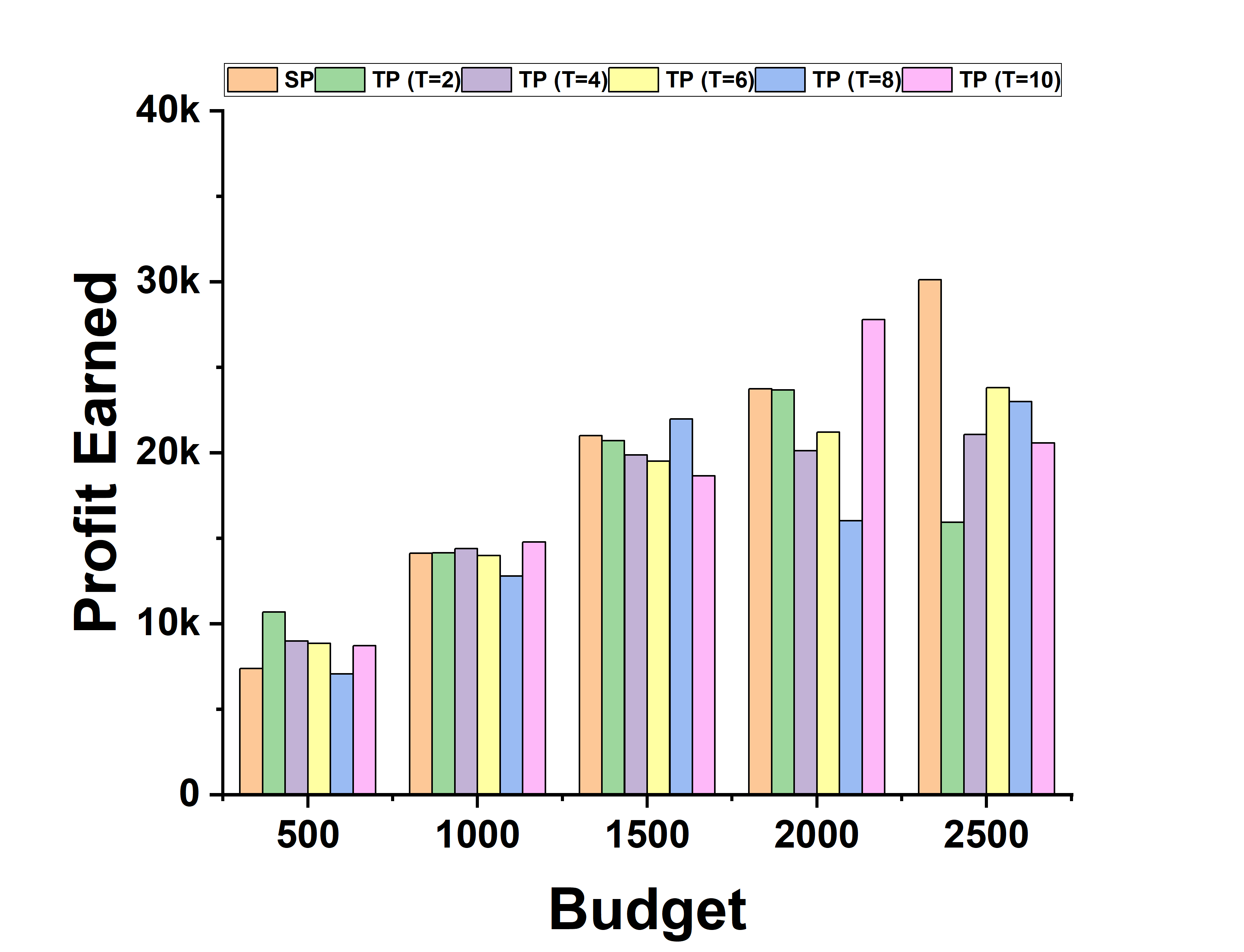}
        \caption{Epsilon=0.6, Split Ratio=70\%}
    \end{subfigure} \\[6pt]

    \begin{subfigure}[t]{0.22\textwidth}
        \includegraphics[width=\linewidth]{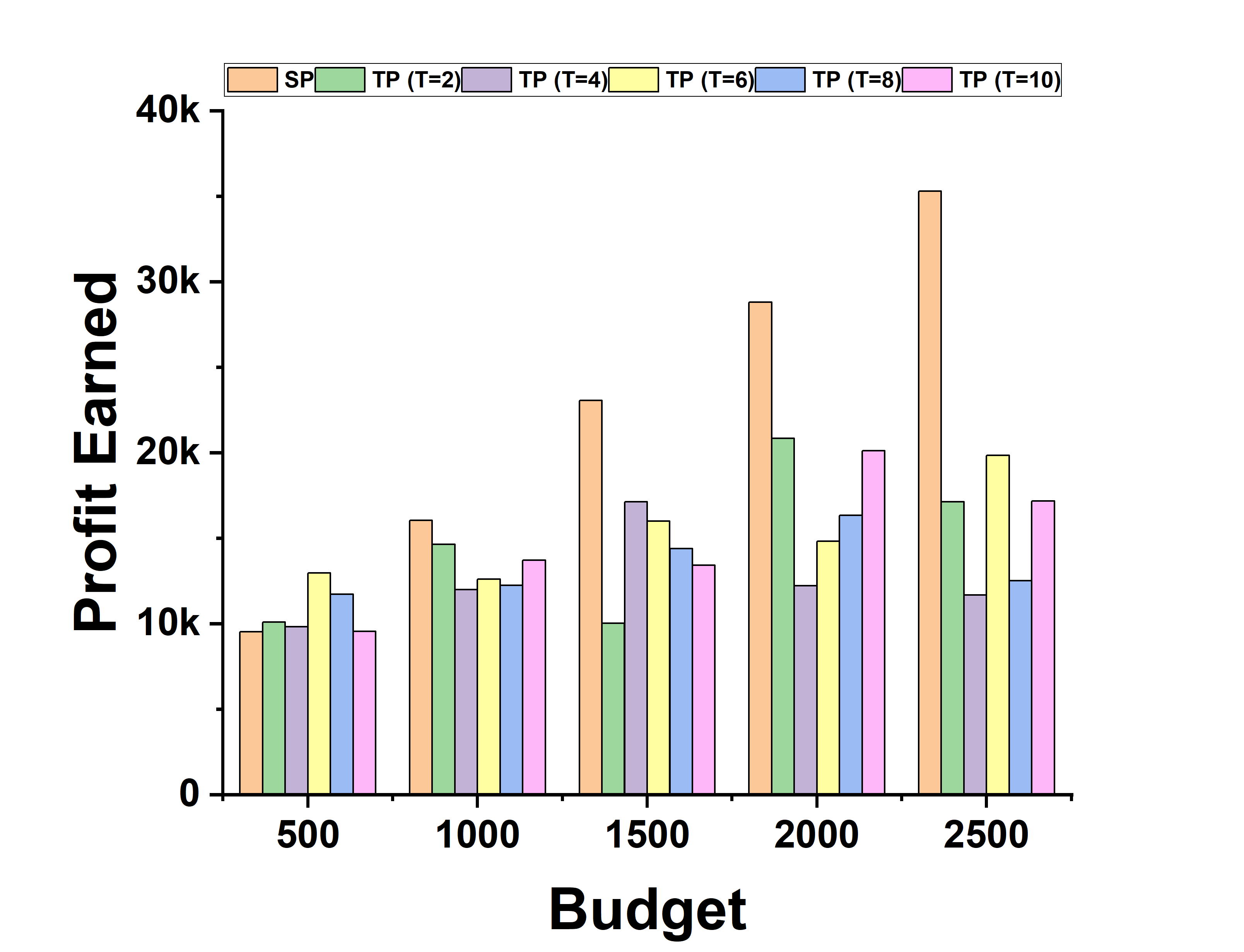}
        \caption{Epsilon=0.01, Split Ratio=90\%}
    \end{subfigure} &
    \begin{subfigure}[t]{0.22\textwidth}
        \includegraphics[width=\linewidth]{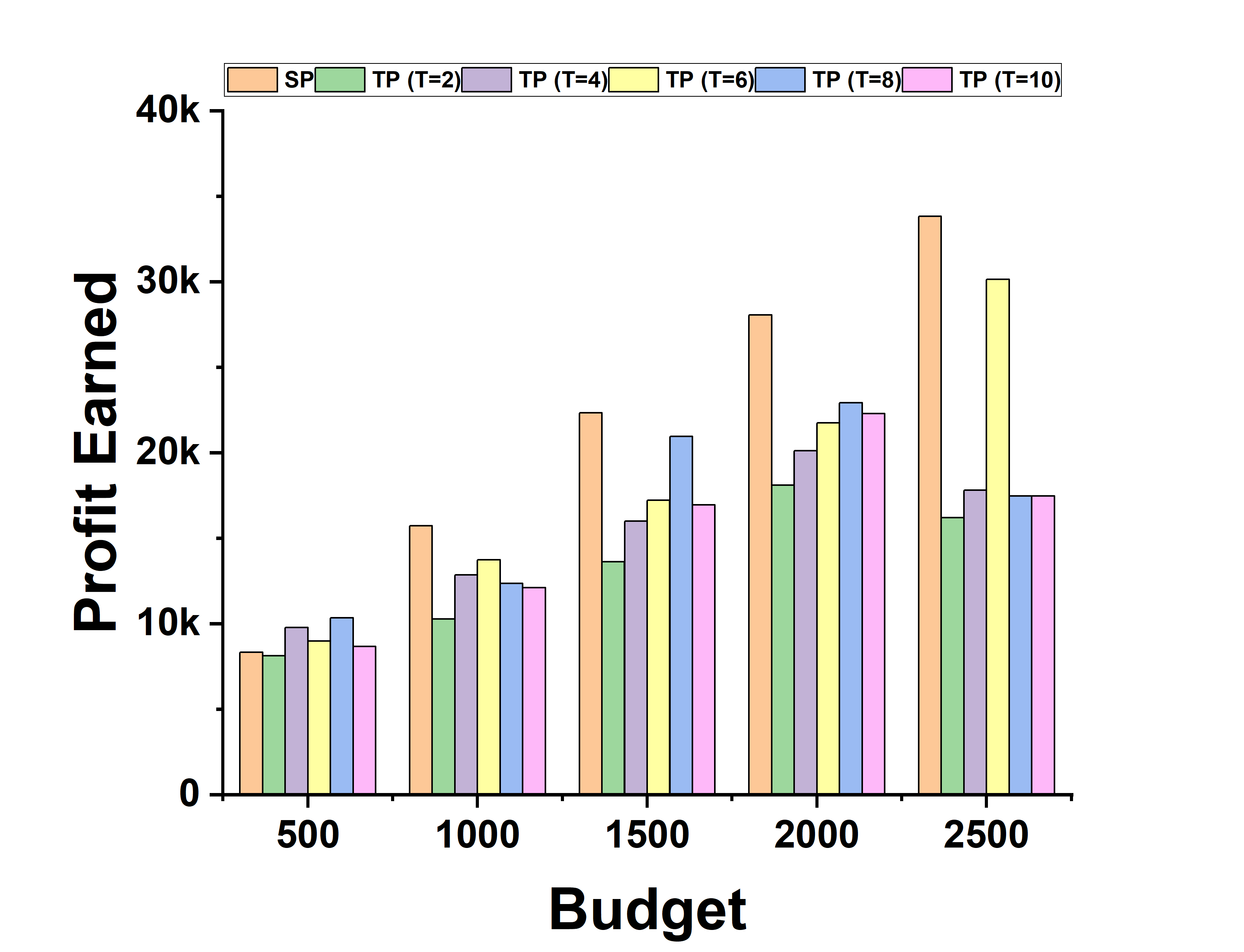}
        \caption{Epsilon=0.1, Split Ratio=90\%}
    \end{subfigure} &
    \begin{subfigure}[t]{0.22\textwidth}
        \includegraphics[width=\linewidth]{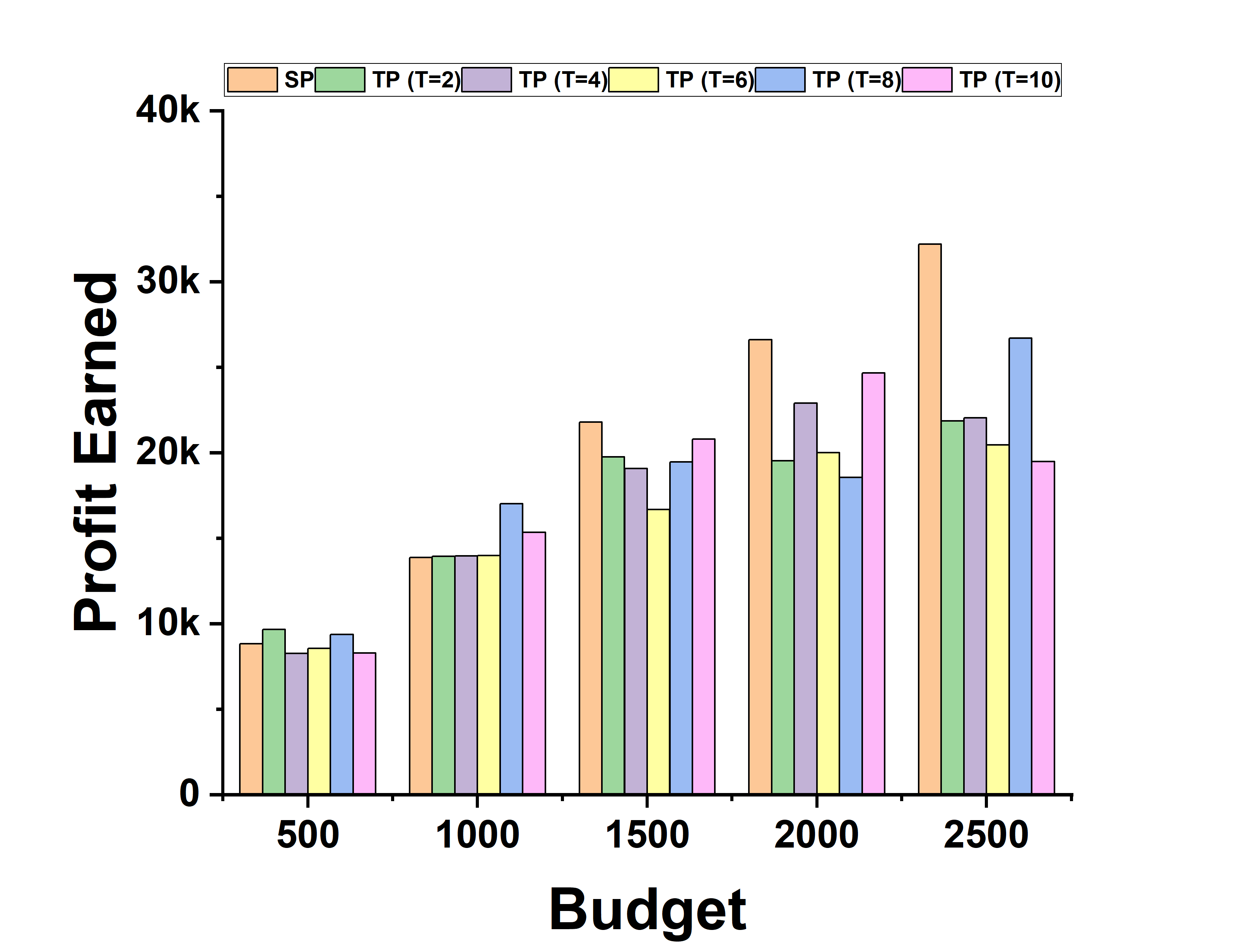}
        \caption{Epsilon=0.3, Split Ratio=90\%}
    \end{subfigure} &
    \begin{subfigure}[t]{0.22\textwidth}
        \includegraphics[width=\linewidth]{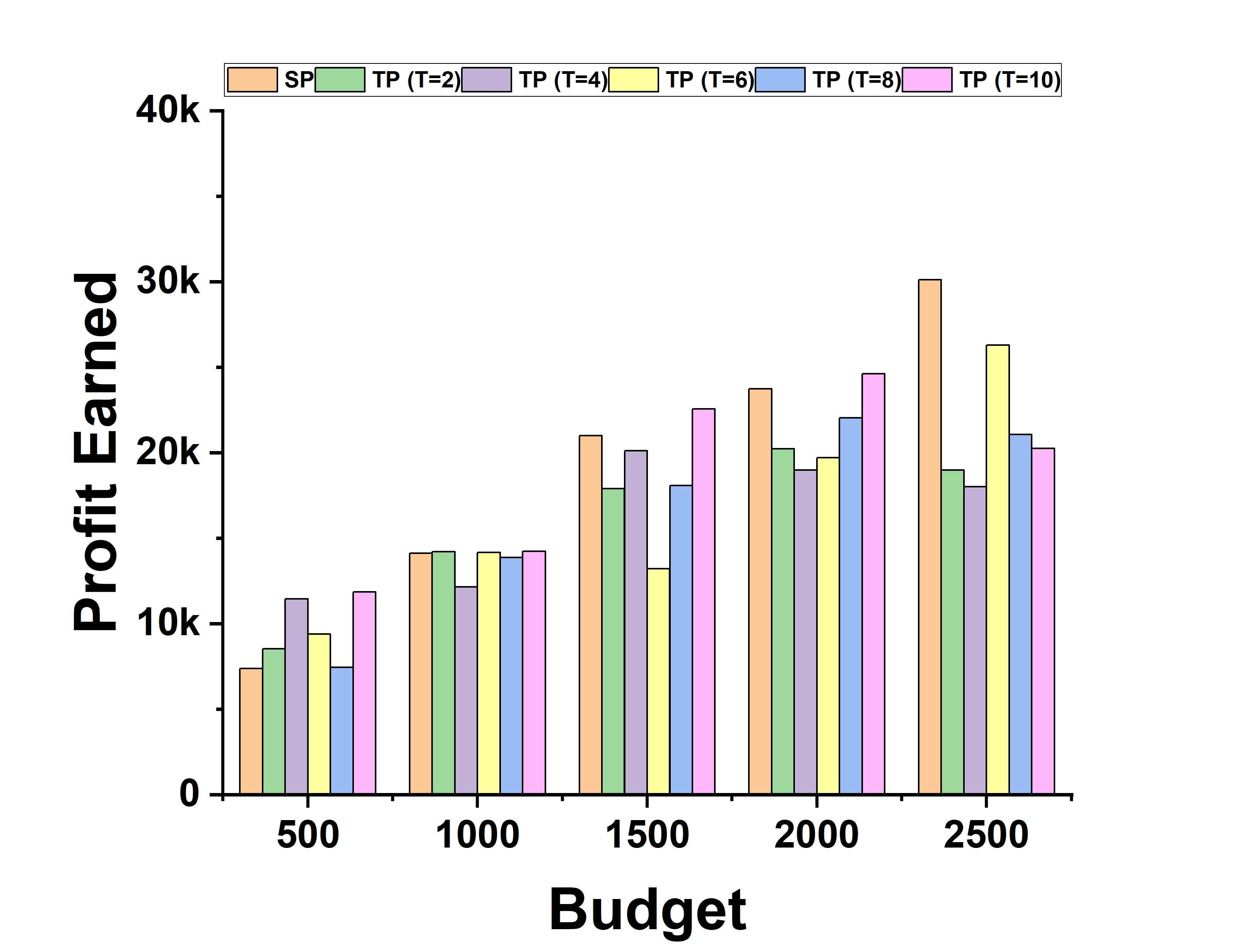}
        \caption{Epsilon=0.6, Split Ratio=90\%}
    \end{subfigure}
\end{tabular}
\caption{Profit Earned in Single Phase Vs.\ Two Phase for different Epsilon values for Stochastic Greedy, Probability Setting - Trivalency, \textit{LM} Dataset}
\label{Fig:RQ7LM_T1}
\end{figure}

To analyze the effect of the $\epsilon$ parameter on the profit earned by stochastic greedy algorithms, we examine the exact profits of \textbf{StG0.01}, \textbf{StG0.1}, \textbf{StG0.3}, and \textbf{StG0.6} on the \textit{Email-Eu-Core} dataset in both single-phase and two phase settings for a fixed budget $500$. In the single-phase, all variants perform similarly, with $\epsilon = 0.6$ achieving the highest profit of $324362.13$ shown in Figure~\ref{Fig:RQ7_T1}(d), followed by $\epsilon = 0.01$ ($323968.21$ in Figure~\ref{Fig:RQ7_T1}(a)), $\epsilon = 0.1$ ($323467.17$ in Figure~\ref{Fig:RQ7_T1}(b)), and 
$\epsilon = 0.3$ ($321682.97$ in Figure~\ref{Fig:RQ7_T1}(c)). This suggests low sensitivity to $\epsilon$ in single-phase. In contrast, the two phase setting shows a strong influence of $\epsilon$. At timestep $2$ and split ratio $0.3$, the highest profit is obtained by $\epsilon = 0.1$  seen in Figure~\ref{Fig:RQ7_T1}(f) with $324900$, while $\epsilon = 0.01$ yields $324163.42$ in Figure~\ref{Fig:RQ7_T1}(e). Meanwhile, $\epsilon = 0.3$ and $\epsilon = 0.6$ fall behind with $311771.94$ (Figure~\ref{Fig:RQ7_T1}(g)) and $323547.97$ (Figure~\ref{Fig:RQ7_T1}(h)), respectively. At timestep $6$ and split ratio $0.3$, $\epsilon = 0.1$ again leads with $325421.49$, followed closely by $\epsilon = 0.01$ ($324935.12$) and $\epsilon = 0.3$ ($321109.68$), whereas $\epsilon = 0.6$ drops to $311174.89$.
A significant pattern is observed at deeper timesteps. For example, at timestep $10$ and split ratio $0.1$, the best result is achieved by $\epsilon = 0.3$ with a profit of $377713.42$, followed by $\epsilon = 0.1$ ($366271.11$), $\epsilon = 0.6$ ($336840.62$), 
and $\epsilon = 0.01$ ($317006.24$). Similarly, at split ratio $0.5$, $\epsilon = 0.3$ again outperforms others with $368856.84$ in Figure~\ref{Fig:RQ7_T1}(k), while $\epsilon = 0.01$ yields only $328783.34$ (Figure~\ref{Fig:RQ7_T1}(j)). These results show that while $\epsilon = 0.01$ performs well at early to mid timesteps (e.g., timestep $6$), moderate values like $\epsilon = 0.3$ deliver the highest profits in later stages. Moreover, $\epsilon = 0.6$ gives competitive results with significantly less computation.


\begin{figure}[t]
\centering
\captionsetup[sub]{font=footnotesize}  
\begin{tabular}{cccc}
    \begin{subfigure}[t]{0.22\textwidth}
        \includegraphics[width=\linewidth]{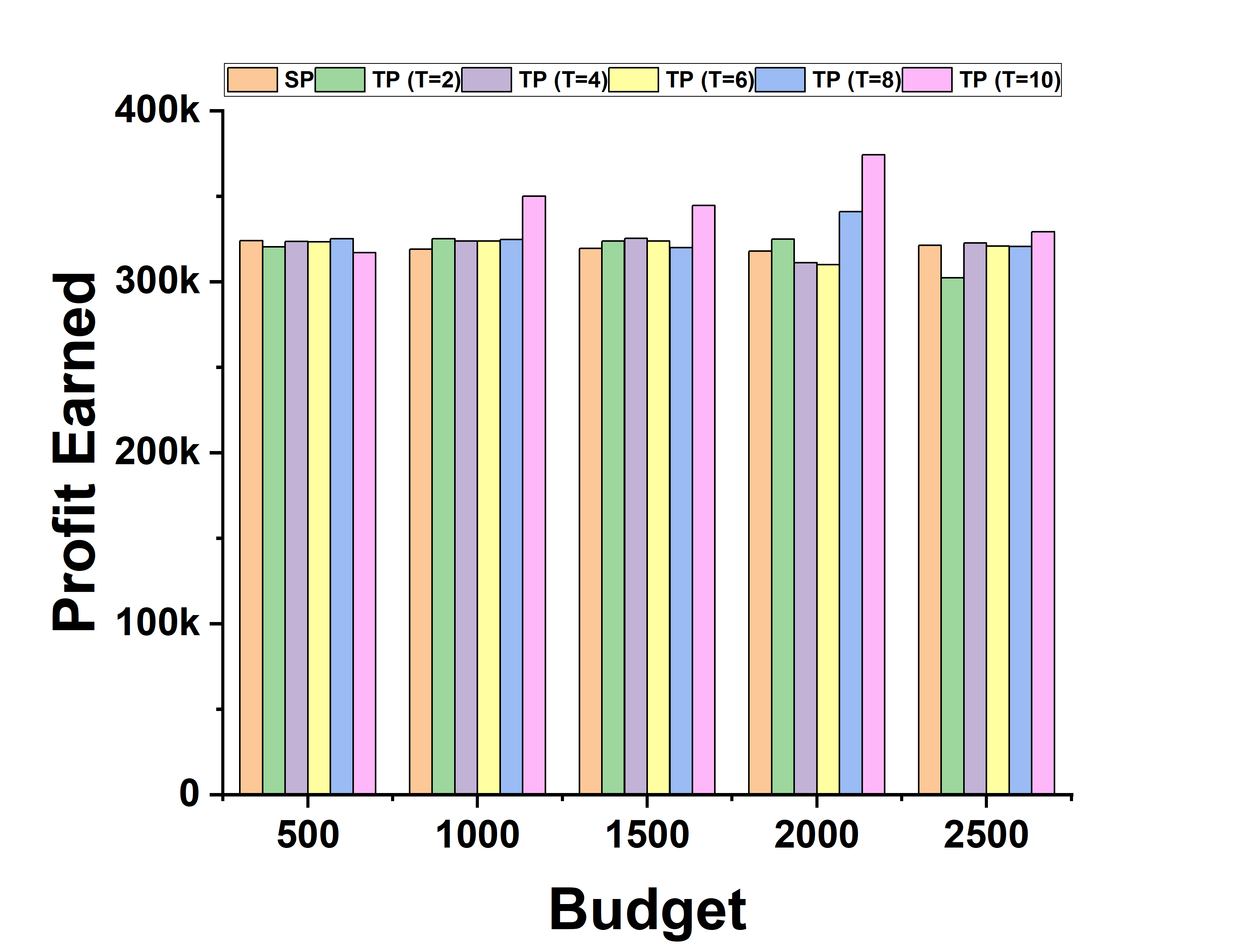}
        \caption{Epsilon=0.01, Split Ratio=10\%}
    \end{subfigure} &
    \begin{subfigure}[t]{0.22\textwidth}
        \includegraphics[width=\linewidth]{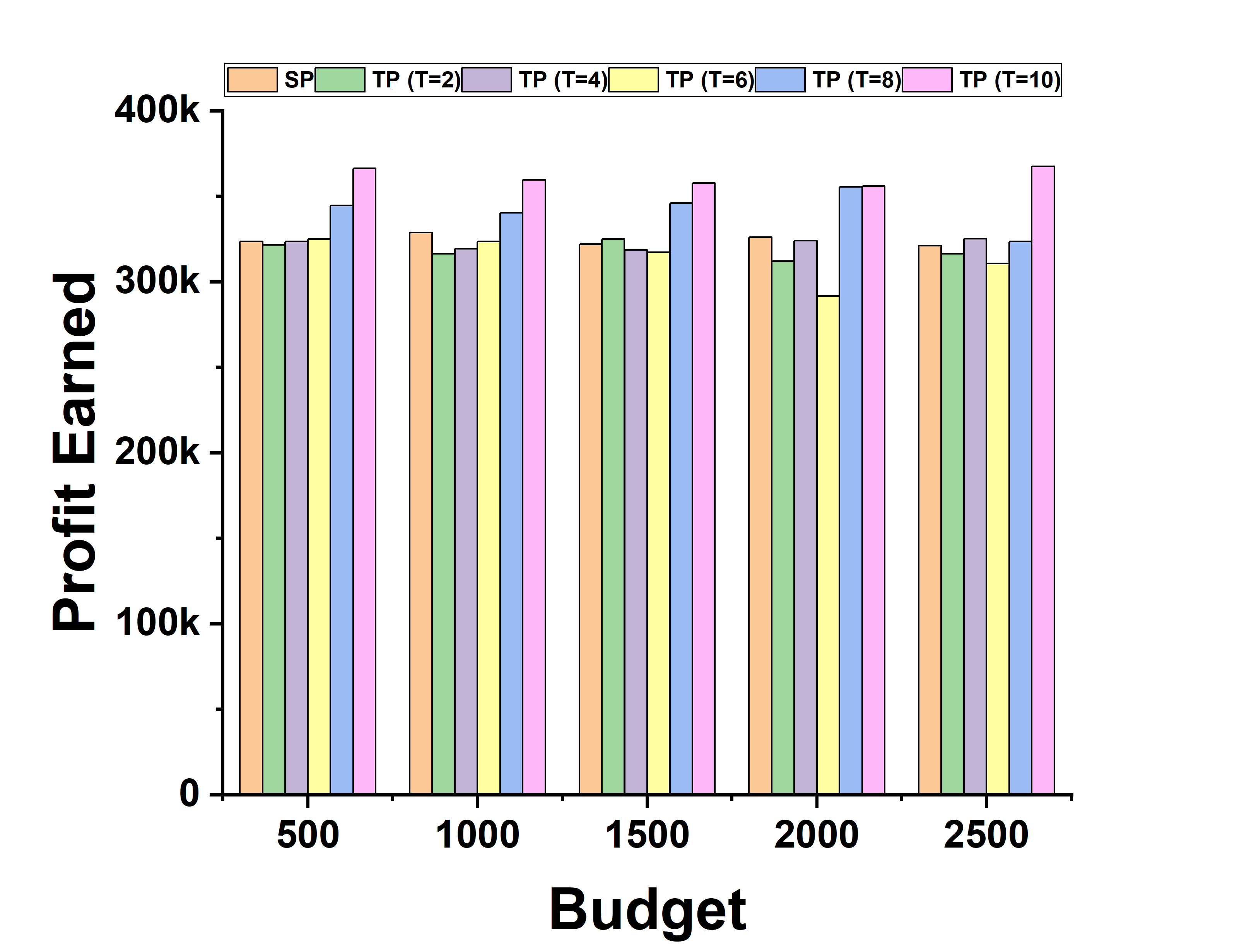}
        \caption{Epsilon=0.1, Split Ratio=10\%}
    \end{subfigure} &
    \begin{subfigure}[t]{0.22\textwidth}
        \includegraphics[width=\linewidth]{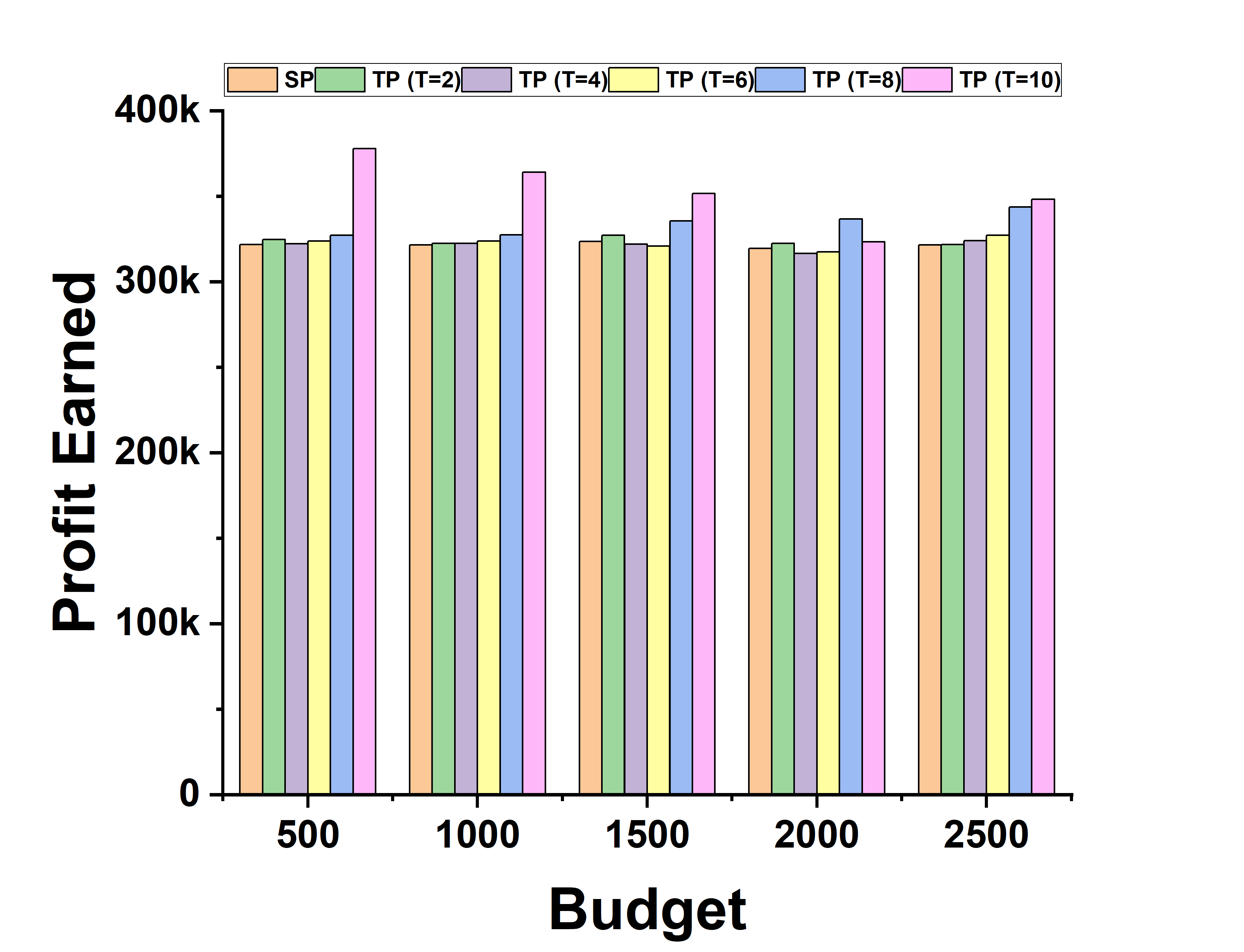}
        \caption{Epsilon=0.3, Split Ratio=10\%}
    \end{subfigure} &
    \begin{subfigure}[t]{0.22\textwidth}
        \includegraphics[width=\linewidth]{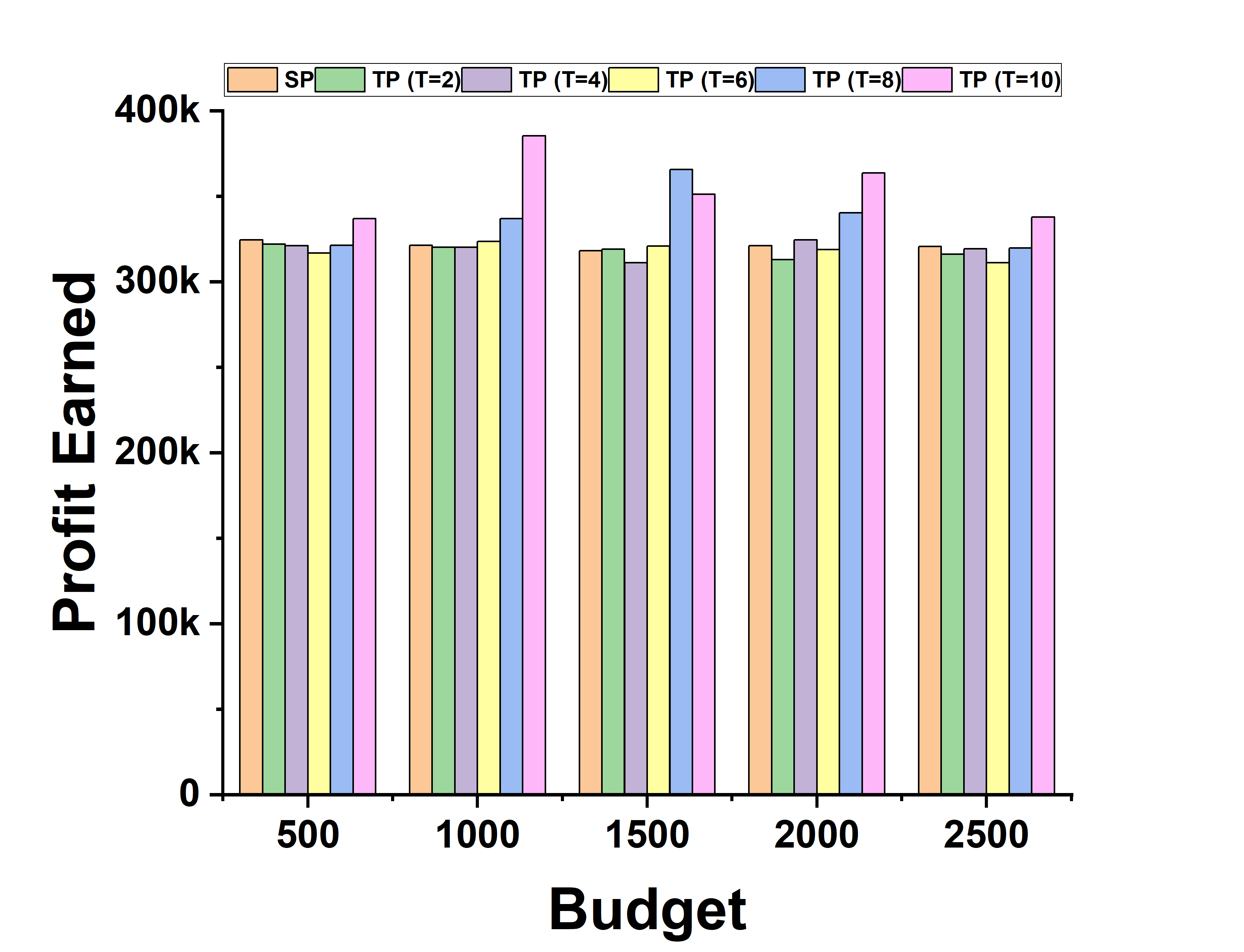}
        \caption{Epsilon=0.6, Split Ratio=10\%}
    \end{subfigure} \\[6pt]

    \begin{subfigure}[t]{0.22\textwidth}
        \includegraphics[width=\linewidth]{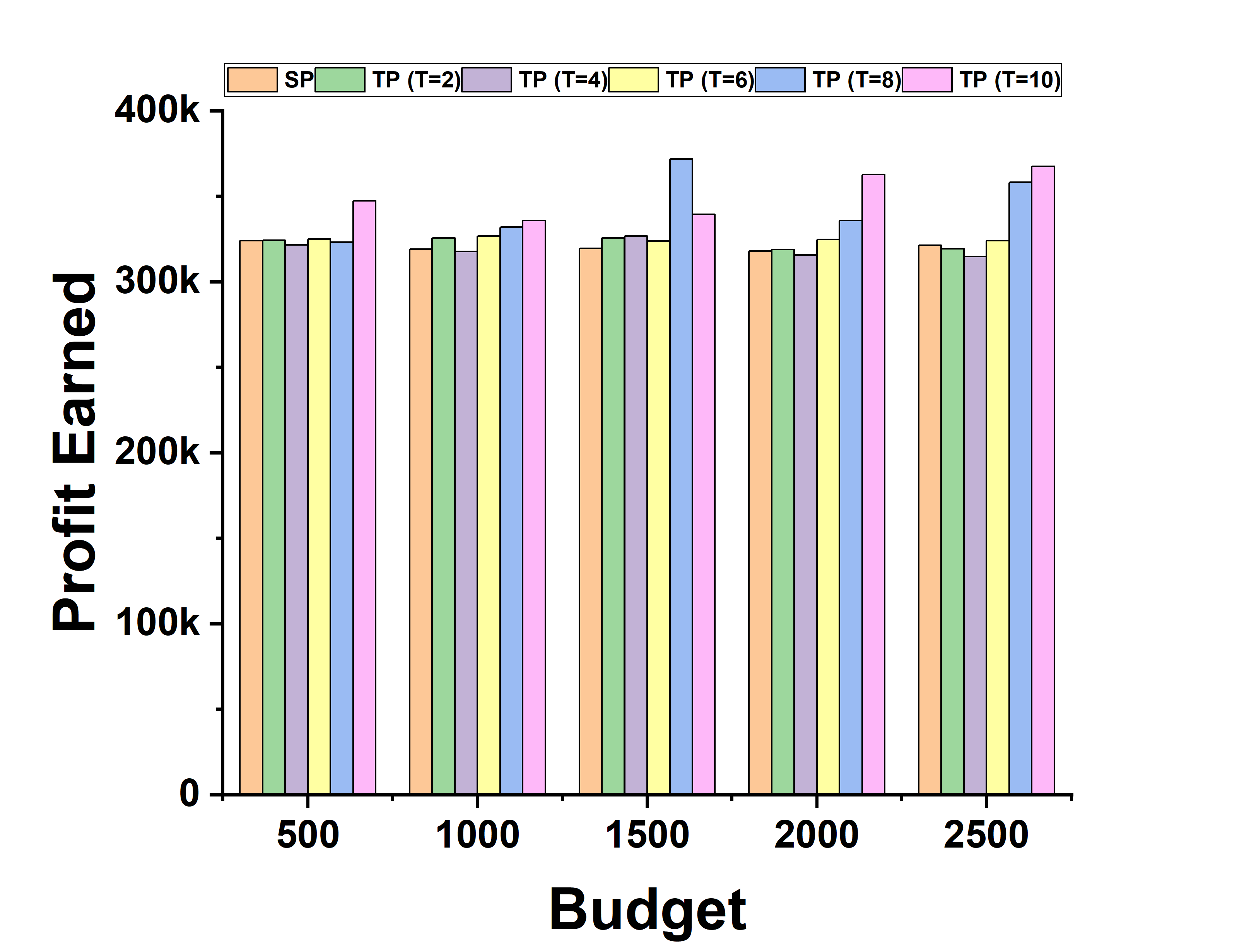}
        \caption{Epsilon=0.01, Split Ratio=30\%}
    \end{subfigure} &
    \begin{subfigure}[t]{0.22\textwidth}
        \includegraphics[width=\linewidth]{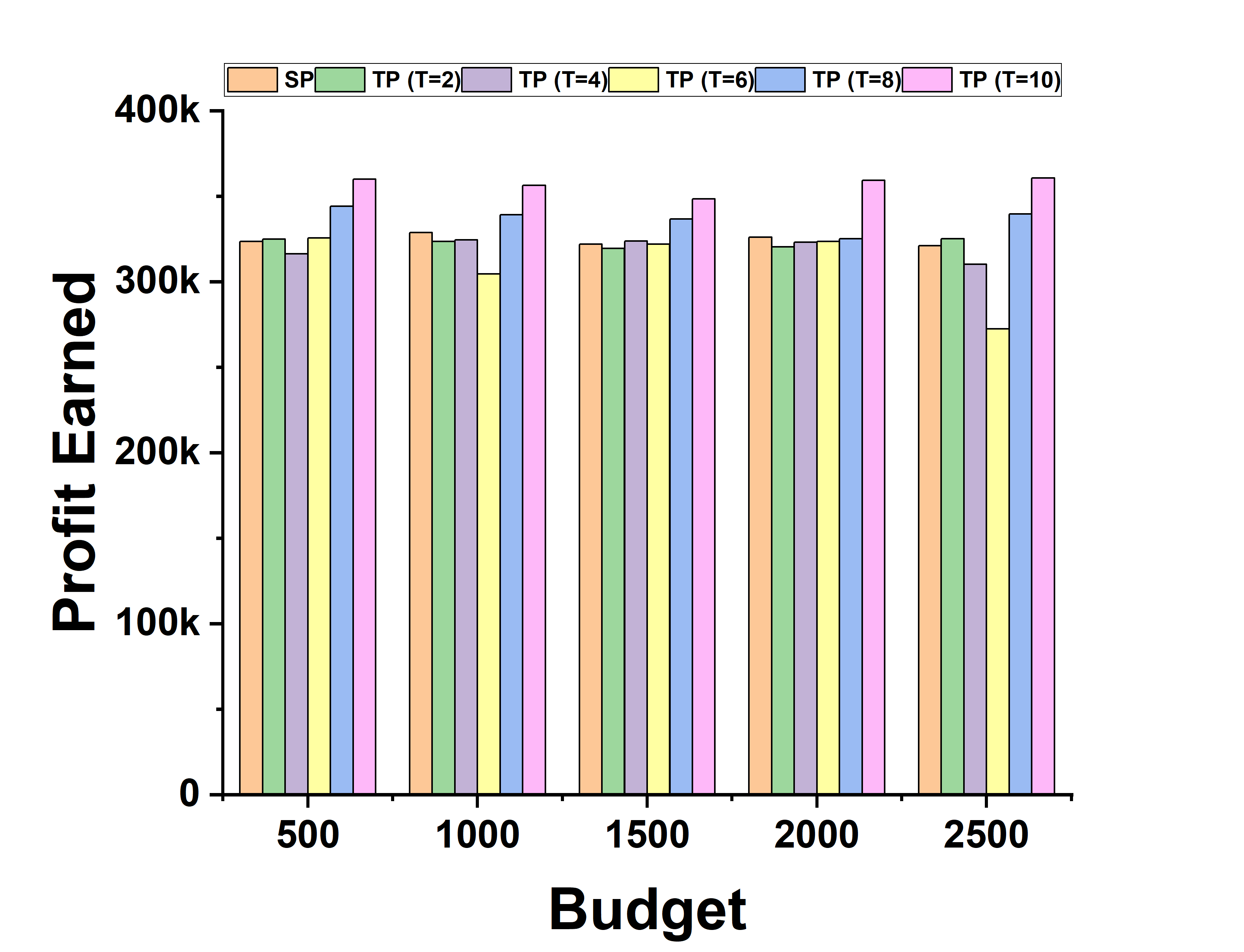}
        \caption{Epsilon=0.1, Split Ratio=30\%}
    \end{subfigure} &
    \begin{subfigure}[t]{0.22\textwidth}
        \includegraphics[width=\linewidth]{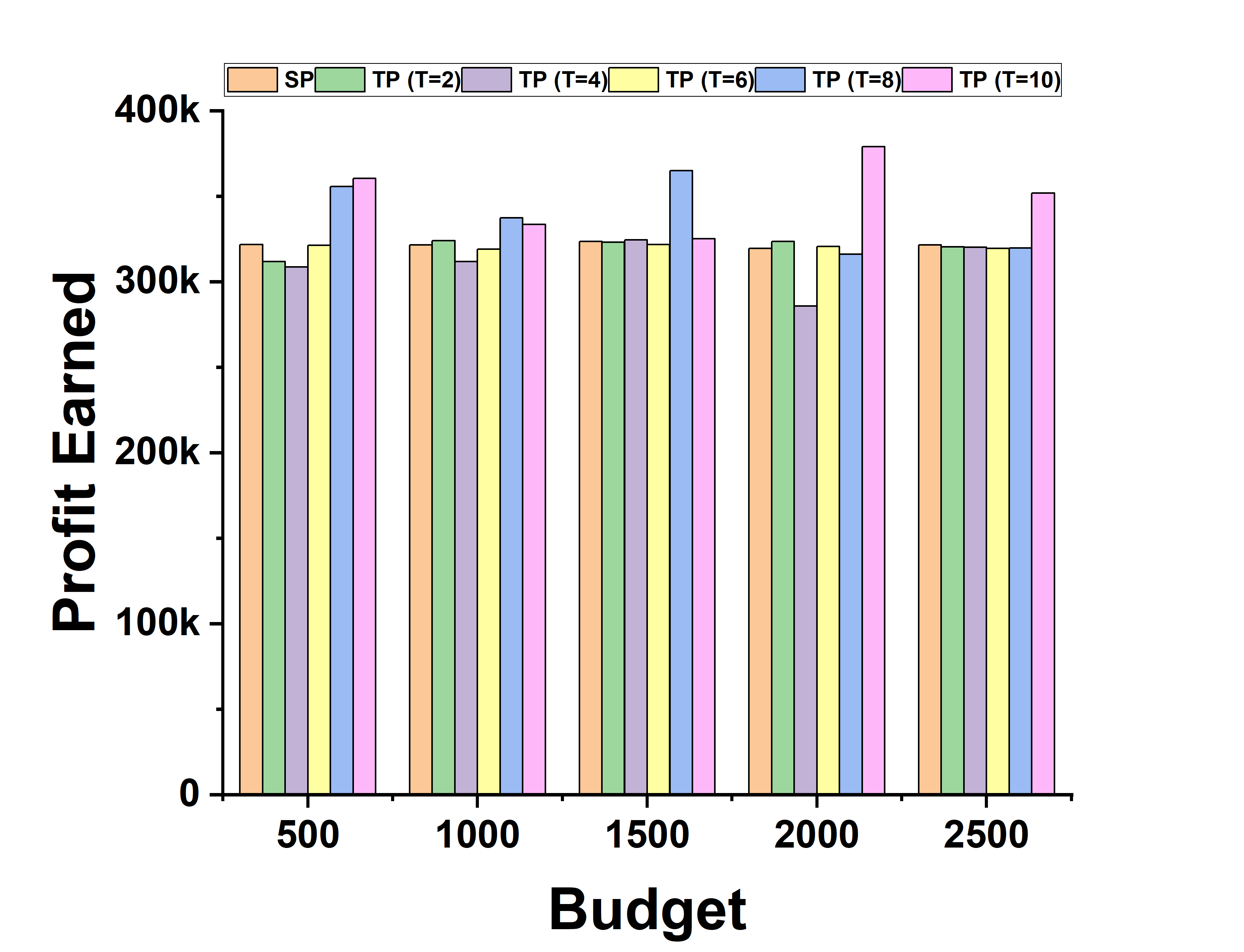}
        \caption{Epsilon=0.3, Split Ratio=30\%}
    \end{subfigure} &
    \begin{subfigure}[t]{0.22\textwidth}
        \includegraphics[width=\linewidth]{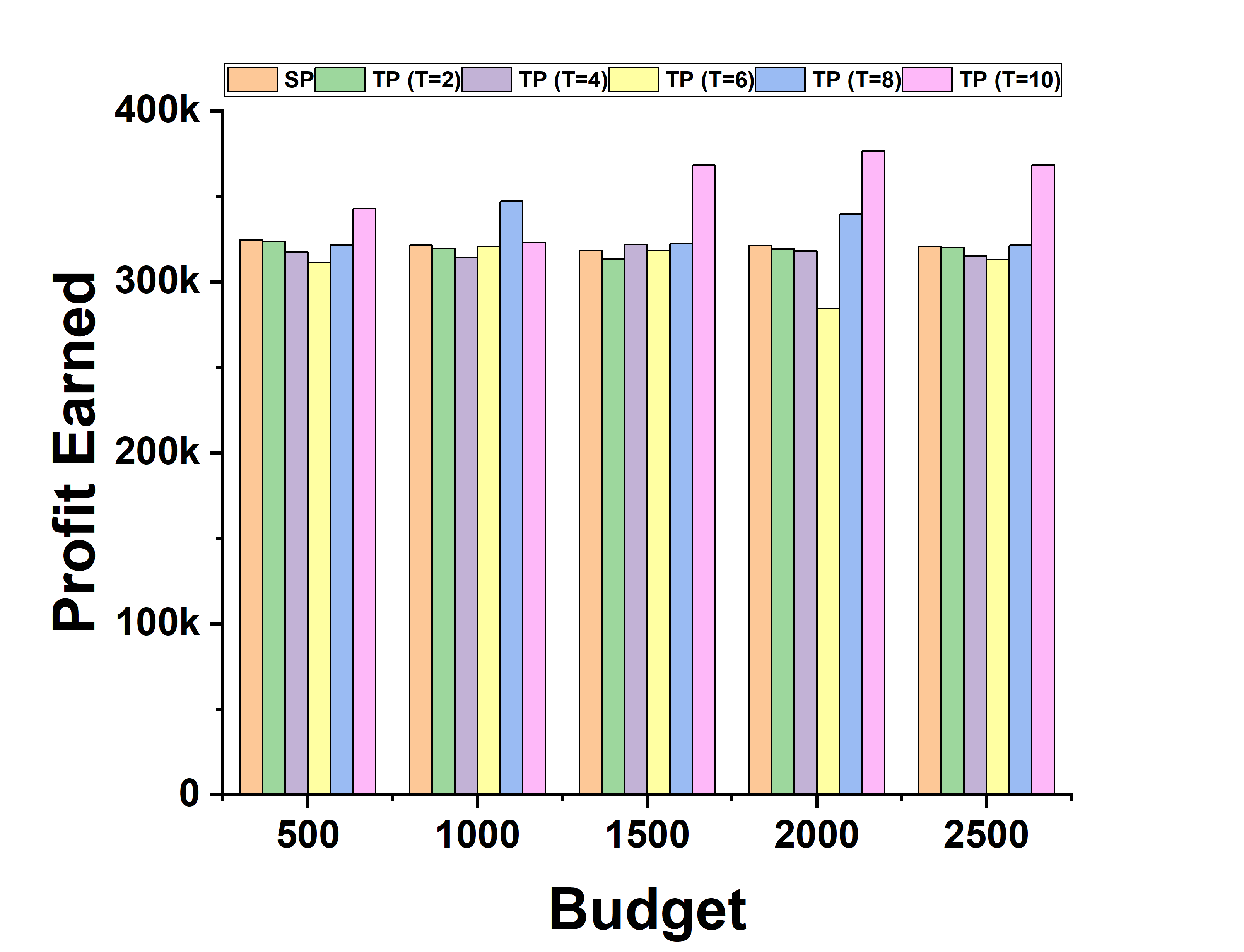}
        \caption{Epsilon=0.6, Split Ratio=30\%}
    \end{subfigure} \\[6pt]

	    \begin{subfigure}[t]{0.22\textwidth}
        \includegraphics[width=\linewidth]{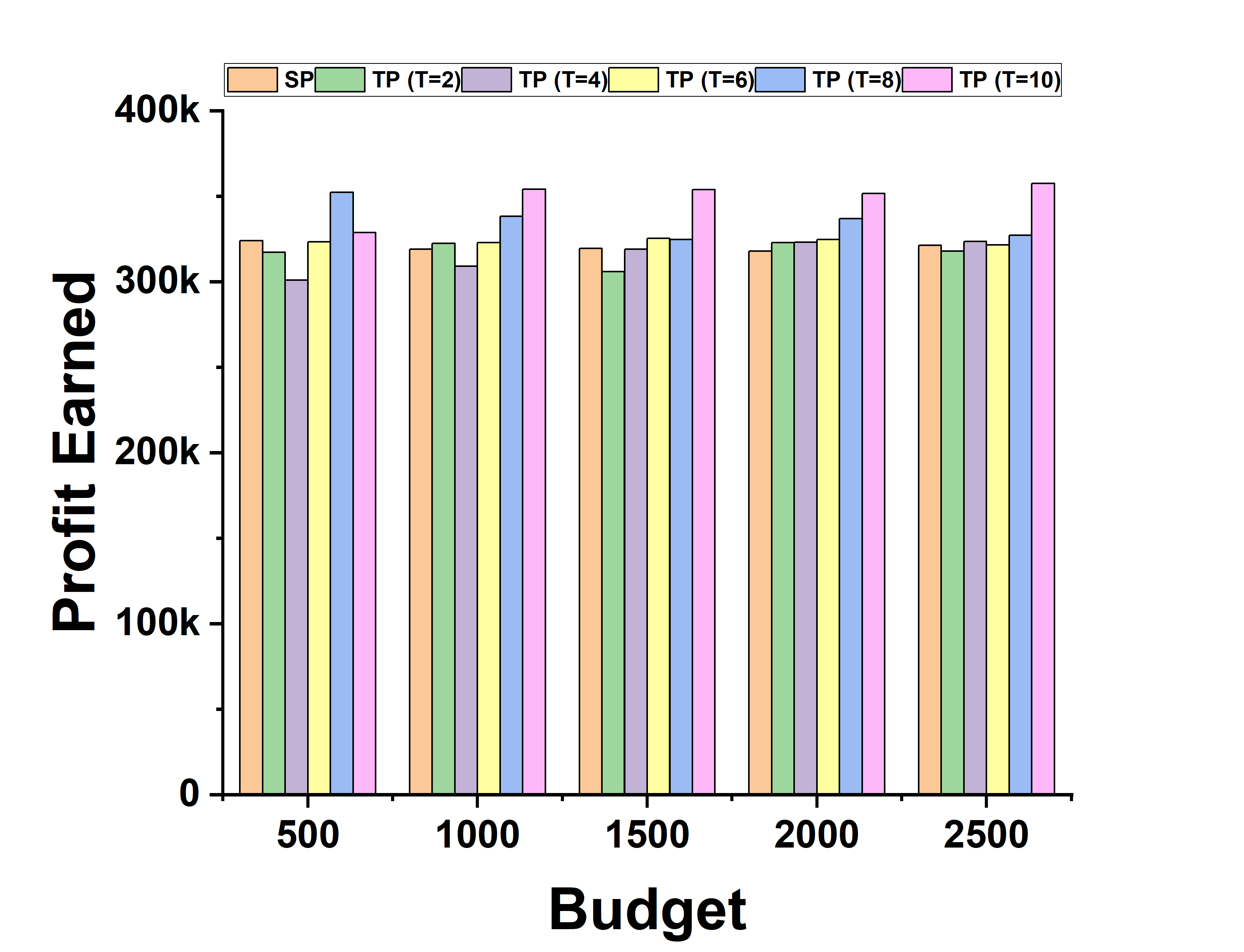}
        \caption{Epsilon=0.01, Split Ratio=50\%}
    \end{subfigure} &
    \begin{subfigure}[t]{0.22\textwidth}
        \includegraphics[width=\linewidth]{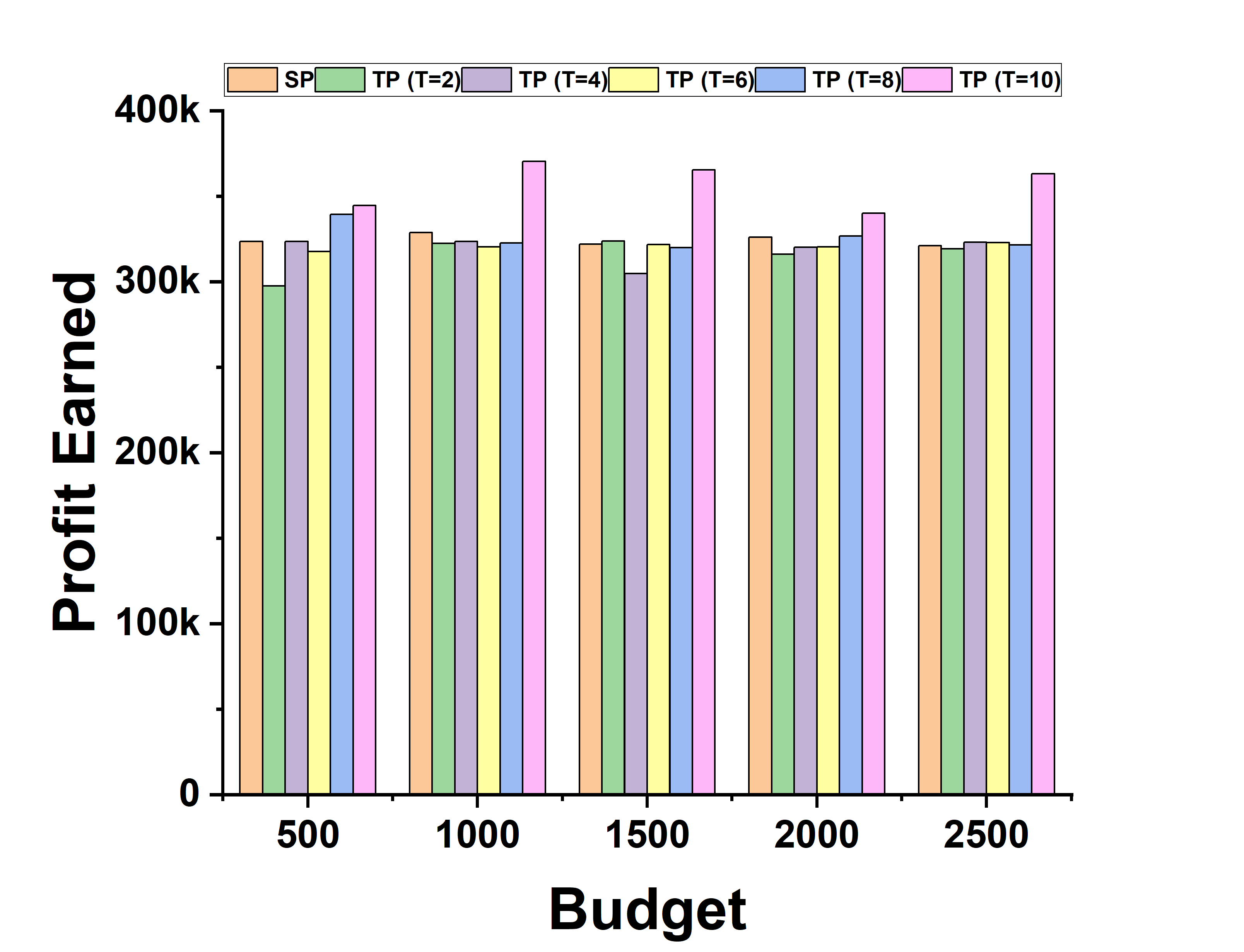}
        \caption{Epsilon=0.1, Split Ratio=50\%}
    \end{subfigure} &
    \begin{subfigure}[t]{0.22\textwidth}
        \includegraphics[width=\linewidth]{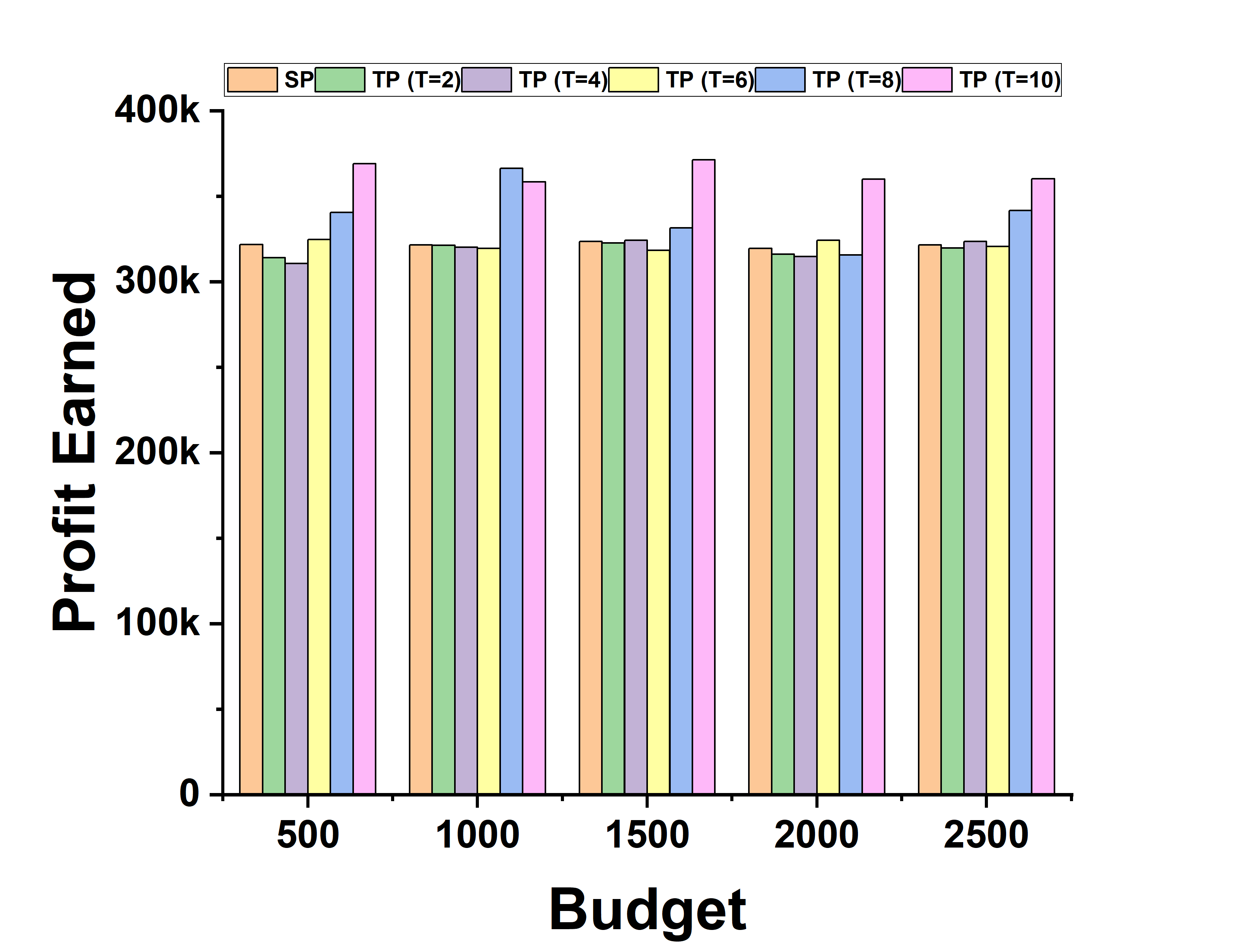}
        \caption{Epsilon=0.3, Split Ratio=50\%}
    \end{subfigure} &
    \begin{subfigure}[t]{0.22\textwidth}
        \includegraphics[width=\linewidth]{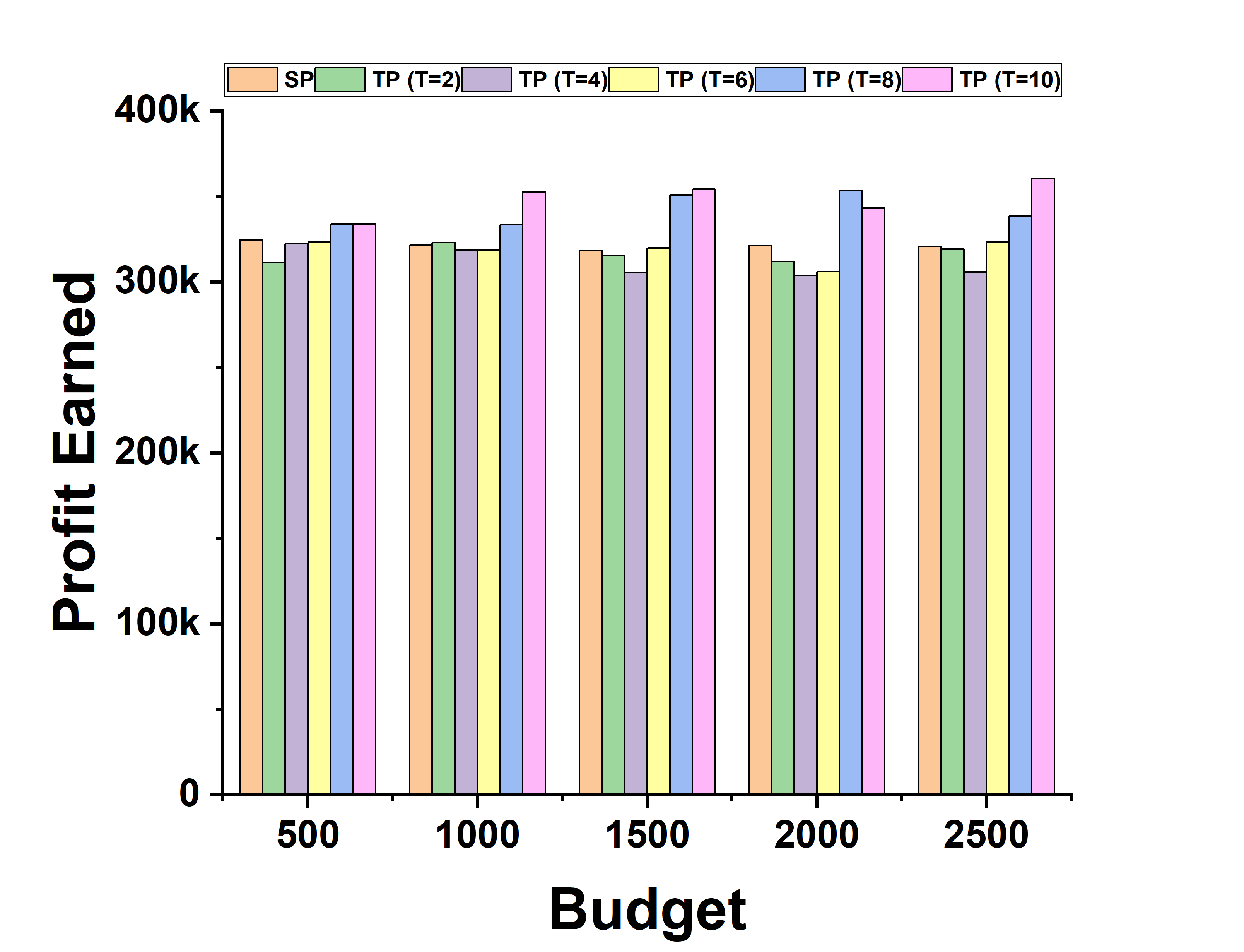}
        \caption{Epsilon=0.6, Split Ratio=50\%}
    \end{subfigure} \\[6pt]
    
        \begin{subfigure}[t]{0.22\textwidth}
        \includegraphics[width=\linewidth]{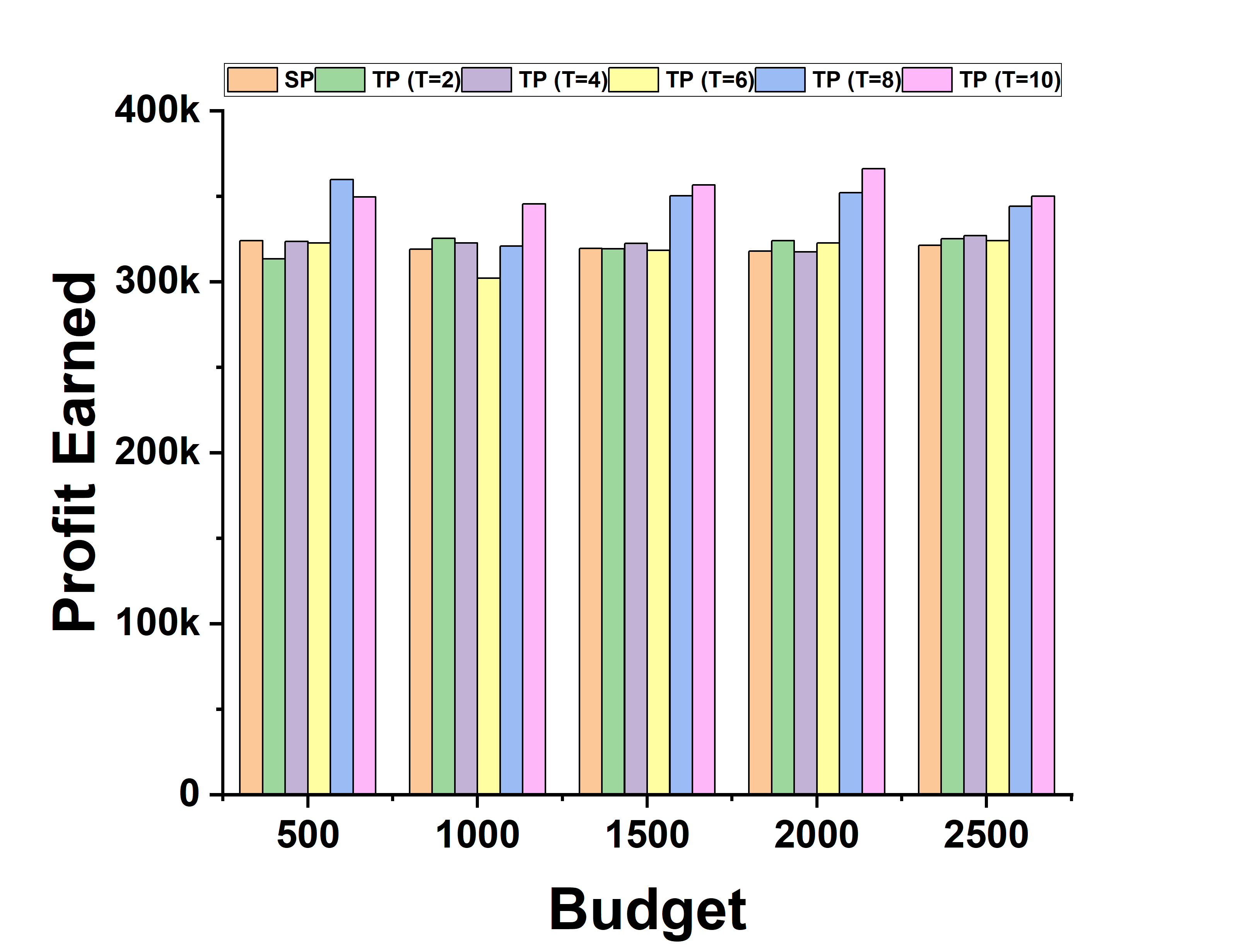}
        \caption{Epsilon=0.01, Split Ratio=70\%}
    \end{subfigure} &
    \begin{subfigure}[t]{0.22\textwidth}
        \includegraphics[width=\linewidth]{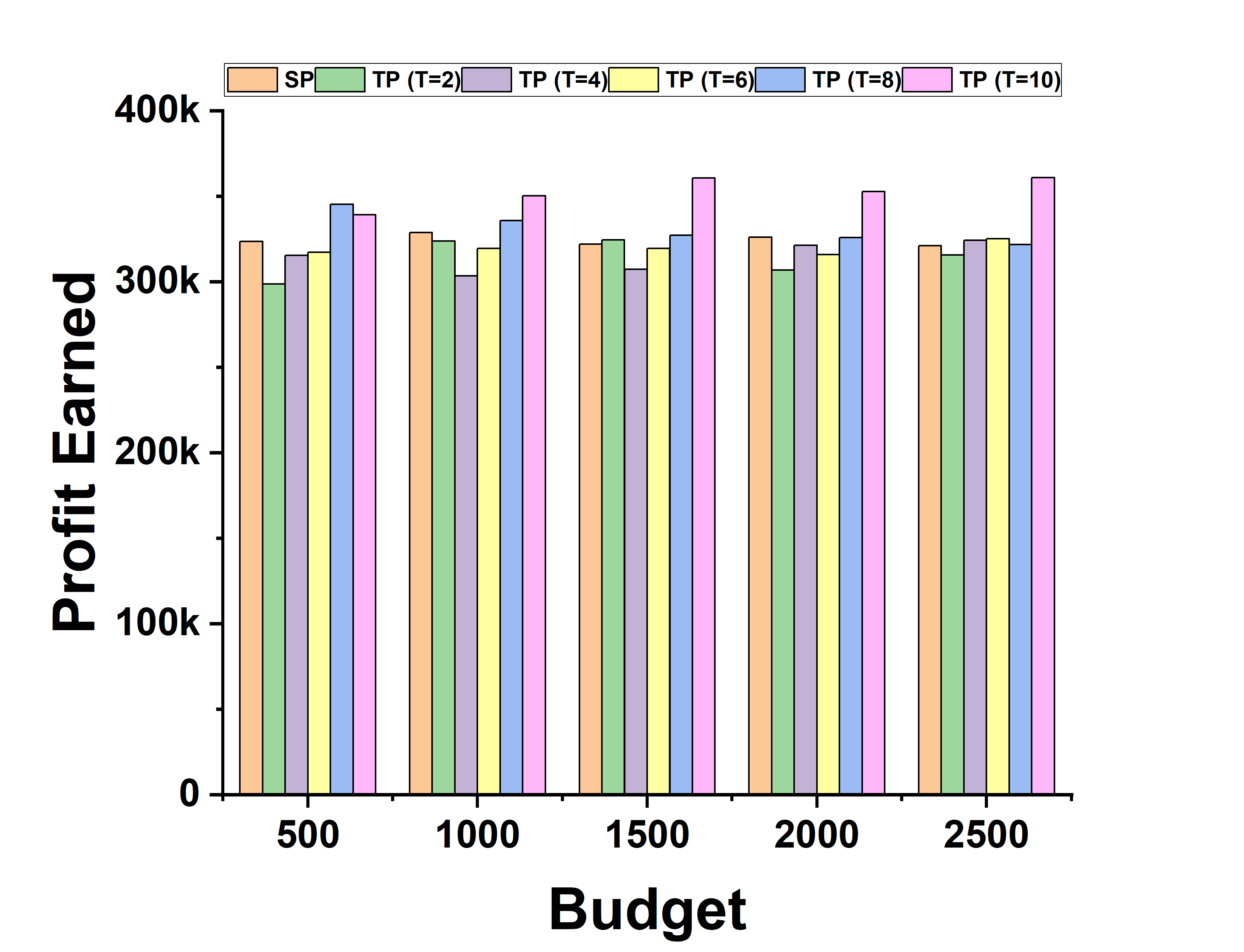}
        \caption{Epsilon=0.1, Split Ratio=70\%}
    \end{subfigure} &
    \begin{subfigure}[t]{0.22\textwidth}
        \includegraphics[width=\linewidth]{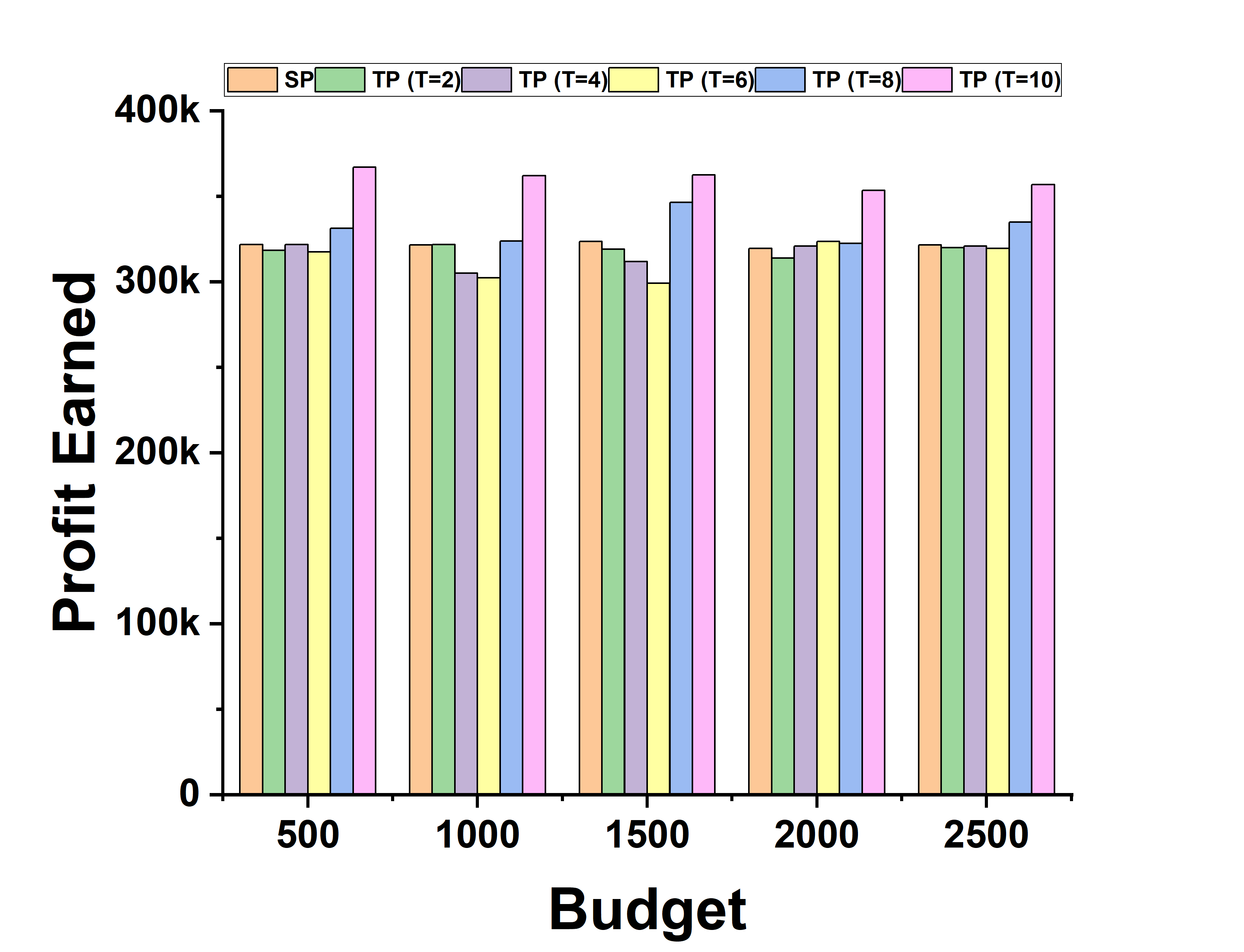}
        \caption{Epsilon=0.3, Split Ratio=70\%}
    \end{subfigure} &
    \begin{subfigure}[t]{0.22\textwidth}
        \includegraphics[width=\linewidth]{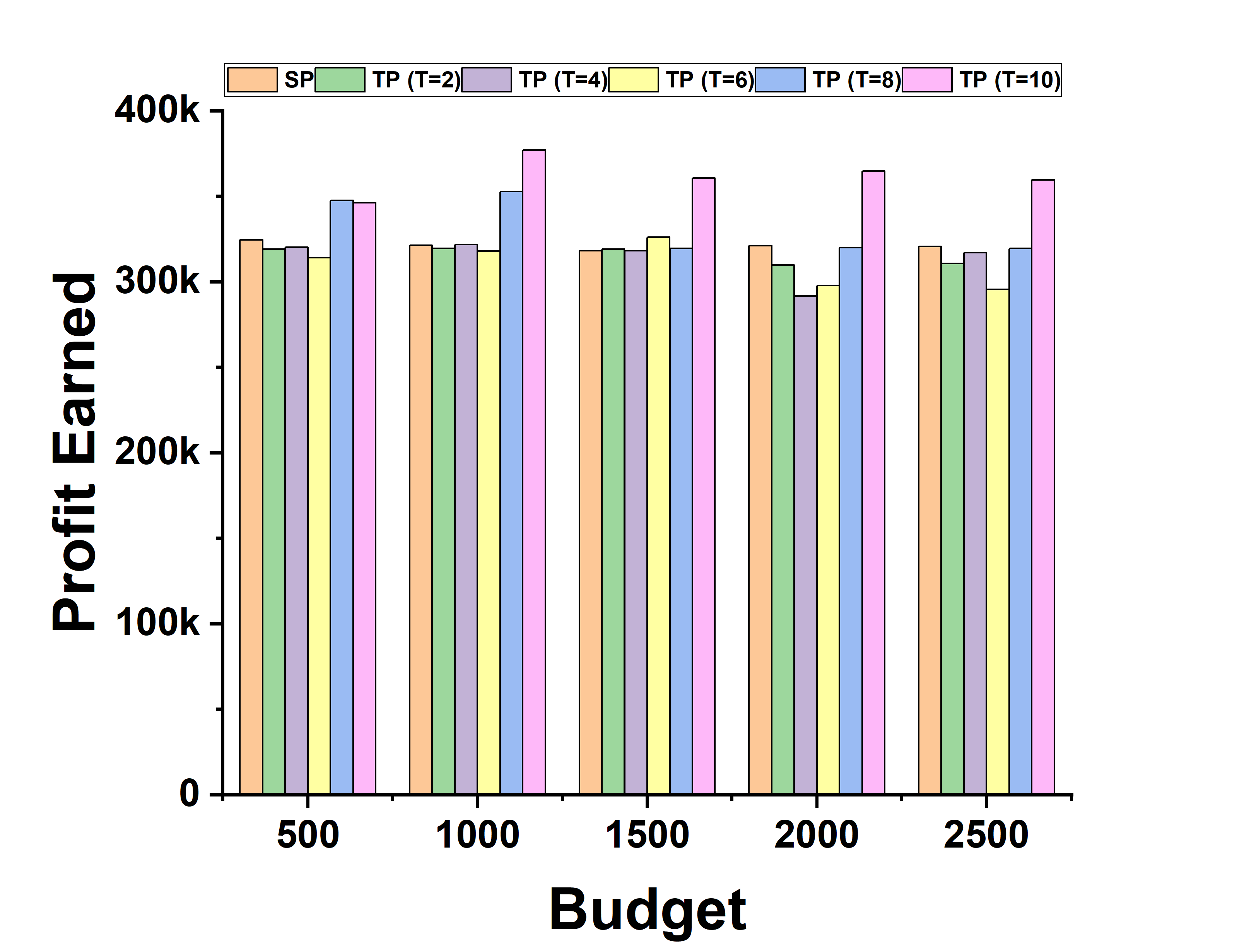}
        \caption{Epsilon=0.6, Split Ratio=70\%}
    \end{subfigure} \\[6pt]
    
        \begin{subfigure}[t]{0.22\textwidth}
        \includegraphics[width=\linewidth]{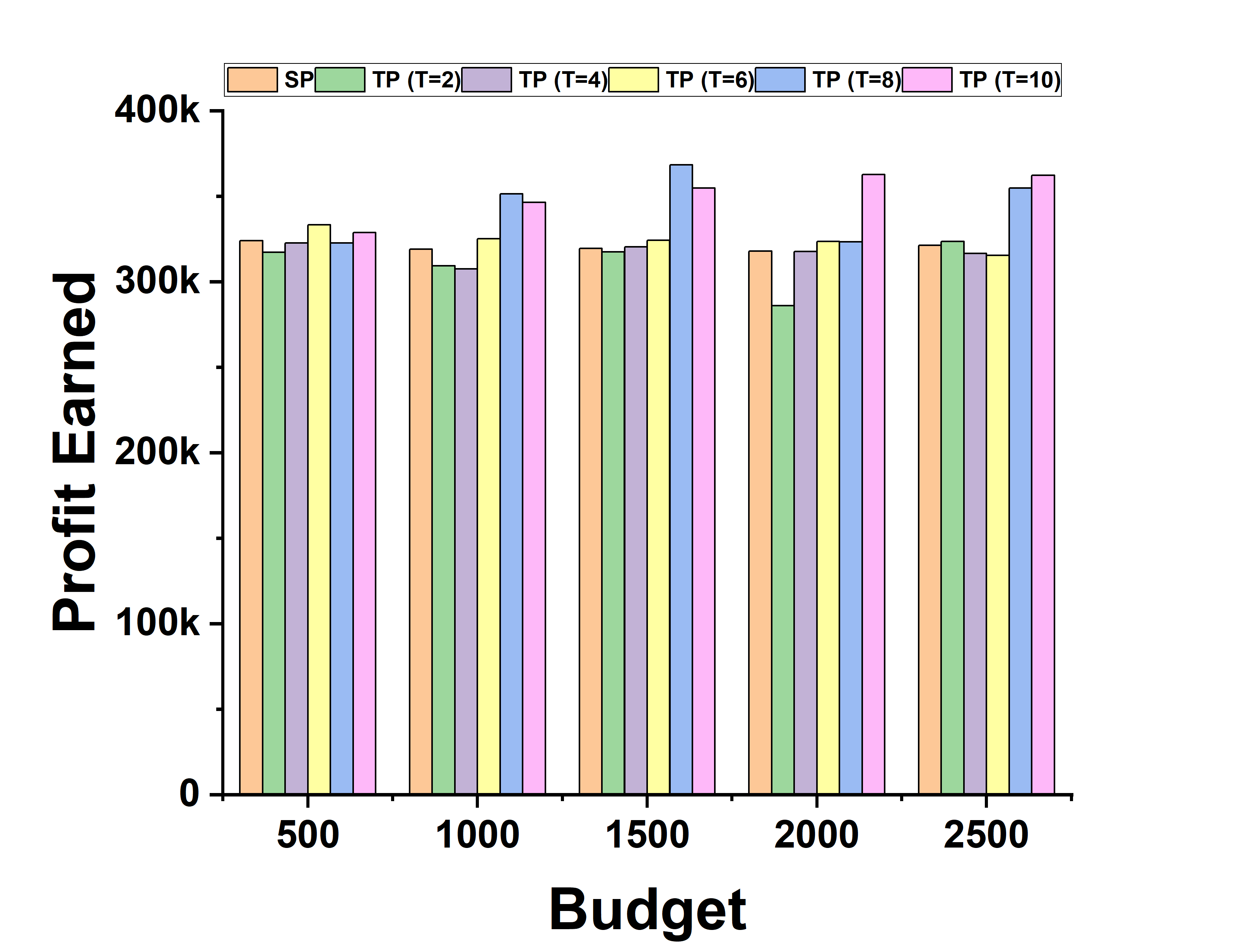}
        \caption{Epsilon=0.01, Split Ratio=90\%}
    \end{subfigure} &
    \begin{subfigure}[t]{0.22\textwidth}
        \includegraphics[width=\linewidth]{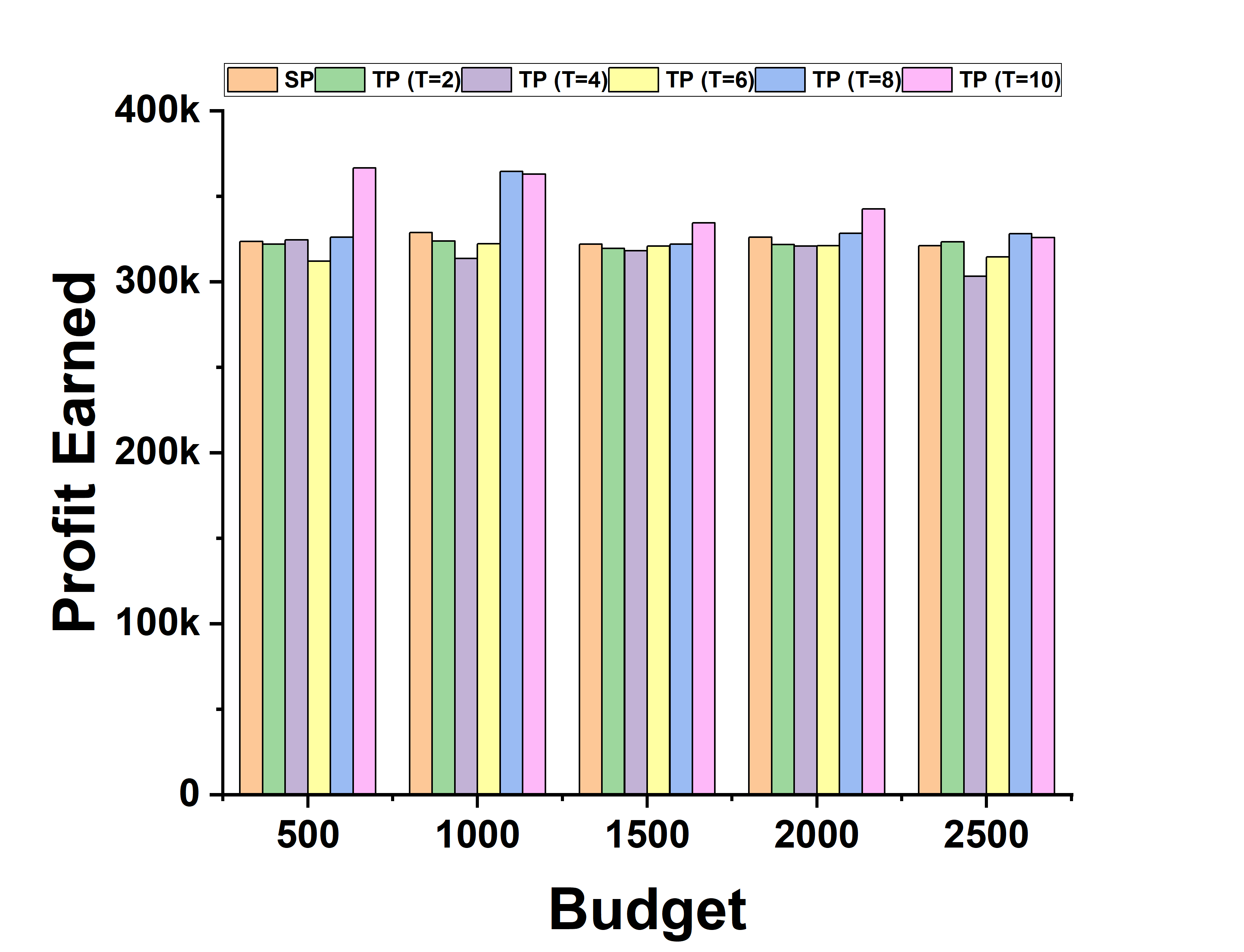}
        \caption{Epsilon=0.1, Split Ratio=90\%}
    \end{subfigure} &
    \begin{subfigure}[t]{0.22\textwidth}
        \includegraphics[width=\linewidth]{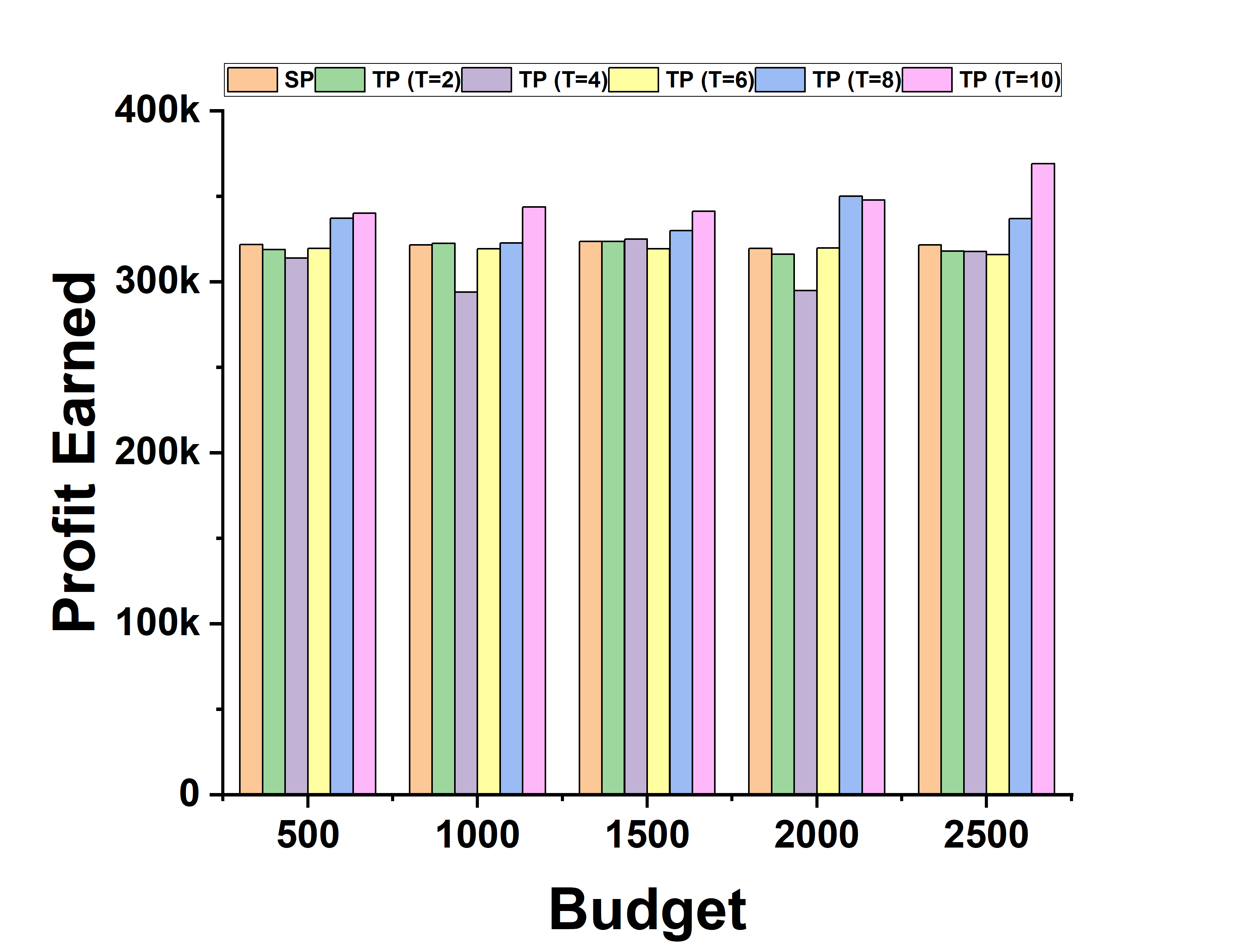}
        \caption{Epsilon=0.3, Split Ratio=90\%}
    \end{subfigure} &
    \begin{subfigure}[t]{0.22\textwidth}
        \includegraphics[width=\linewidth]{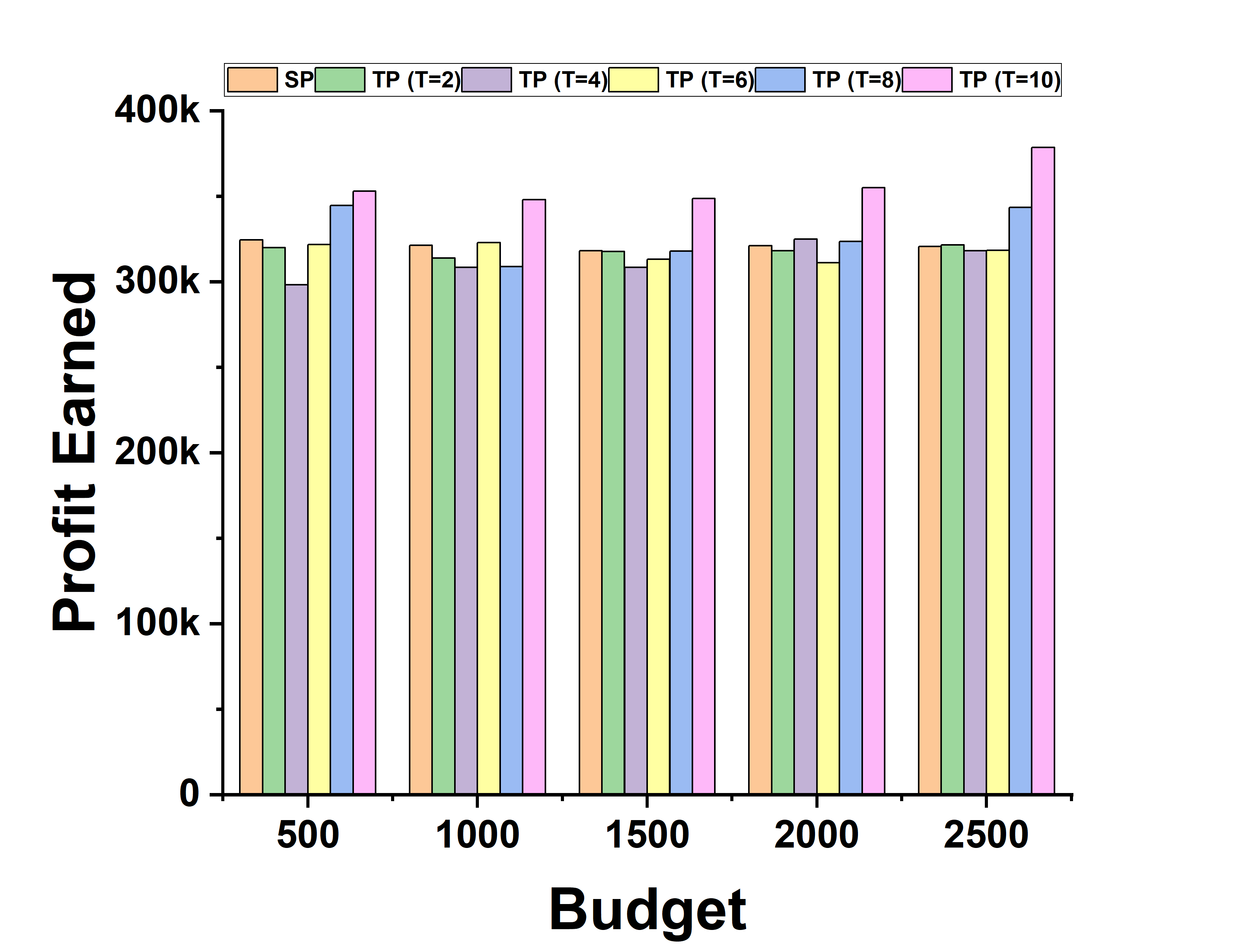}
        \caption{Epsilon=0.6, Split Ratio=90\%}
    \end{subfigure}    
\end{tabular}
\caption{Profit Earned in Single Phase Vs. Two Phase for different Epsilon values for Stochastic Greedy, Probability Setting - Trivalency, \textit{Email-Eu-Core} Dataset}
\label{Fig:RQ7_T1}
\end{figure}

\textbf{Slashdot Dataset}


\begin{figure}[]
    \centering
    \begin{subfigure}[t]{0.45\linewidth}
        \centering
        \includegraphics[width=\linewidth]{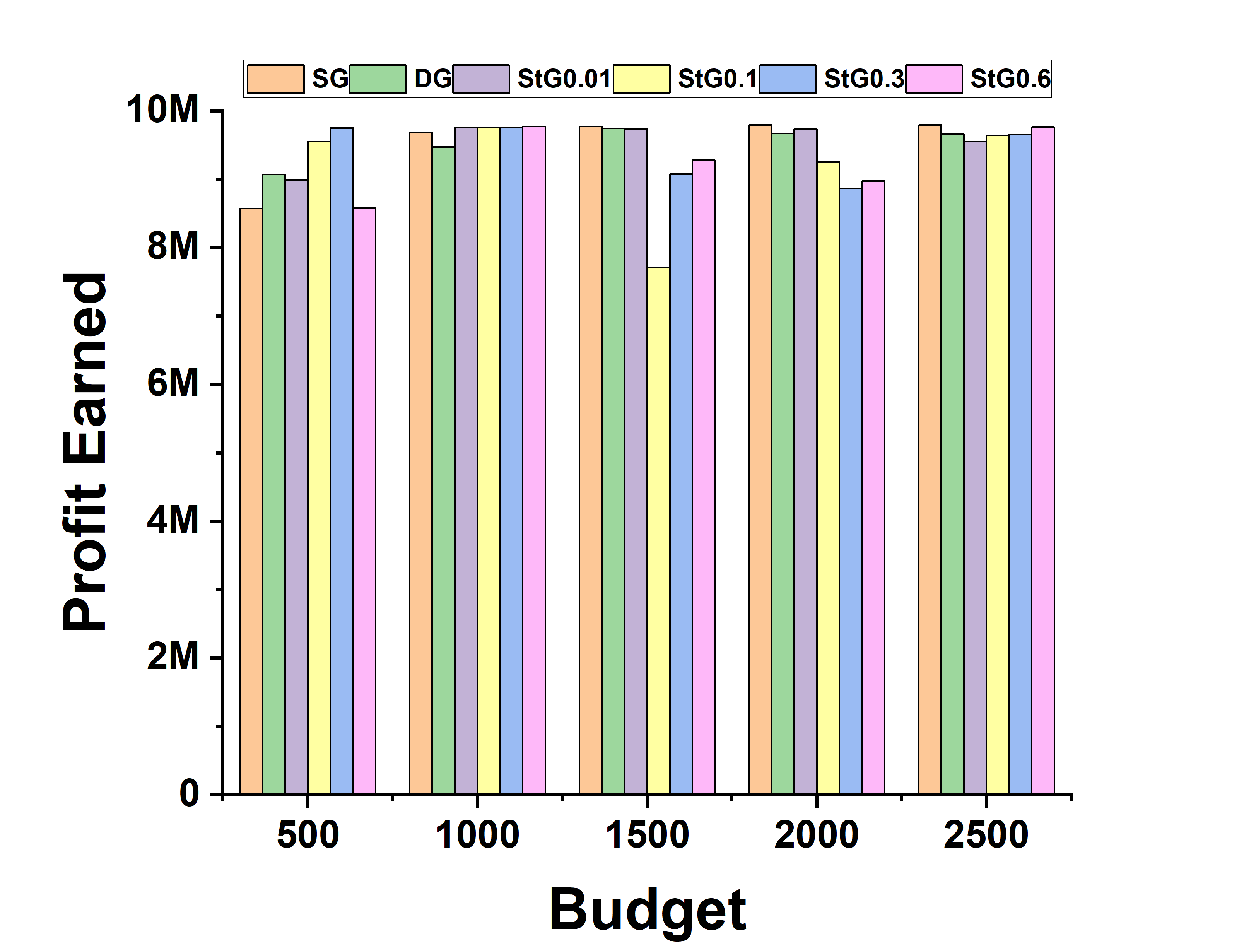}
        \caption{Budget Vs. Profit Earned}
        \label{fig:image1}
    \end{subfigure}
    \hspace{0.05\linewidth}
    \begin{subfigure}[t]{0.45\linewidth}
        \centering
        \includegraphics[width=\linewidth]{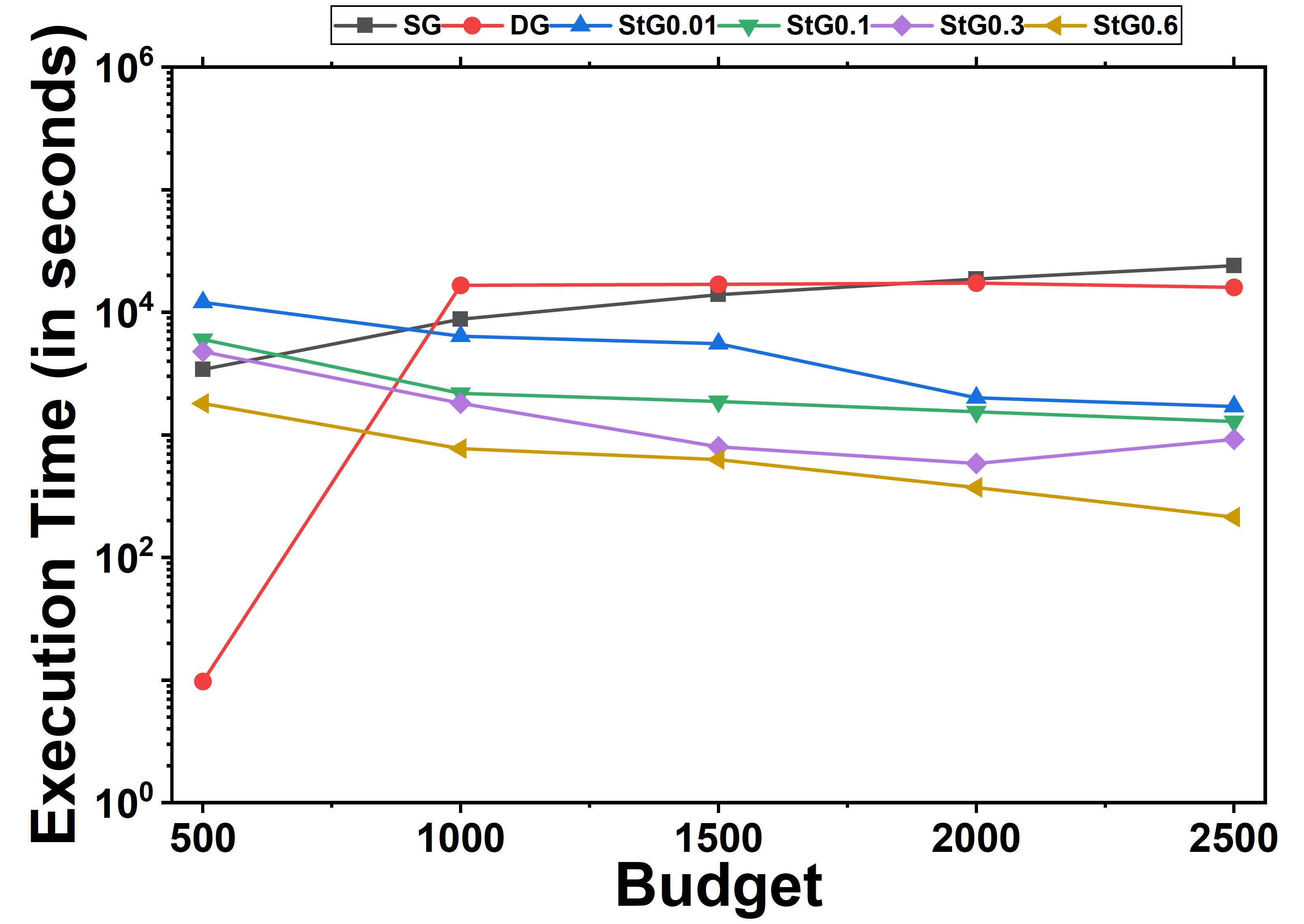}
        \caption{Budget Vs. Execution Time (in seconds)}
        \label{fig:image2}
    \end{subfigure}

    \caption{Single-Phase setting Plots for \textit{Slashdot} Dataset under Trivalency Probability Setting}
    \label{Fig:SlashdotGraphs}
\end{figure}

To evaluate the impact of the parameter $\epsilon$ on the profit earned by stochastic greedy algorithms, we analyze the single-phase results of \textbf{StG0.01}, \textbf{StG0.1}, \textbf{StG0.3}, and \textbf{StG0.6} on the \textit{Slashdot} dataset in Figure~\ref{Fig:SlashdotGraphs}. These algorithms, which differ by the amount of randomness introduced during seed selection, show distinct behaviors based on the value of $\epsilon$. As $\epsilon$ decreases, the algorithm becomes greedier, generally improving profit but at the cost of increased computational time. This trend is visible in the results at a budget of 2500. The most greedy variant, \textbf{StG0.01}, achieves a profit of $9549657.01$ using only $2$ seed nodes but requires $1706.89$ seconds. \textbf{StG0.1} and \textbf{StG0.3} deliver slightly lower profits of $9642142.53$ and $9649247.19$, requiring $1283.12$ and $921.29$ seconds, respectively. Surprisingly, \textbf{StG0.6} achieves a strong profit of $9755685.06$ the highest among all stochastic variants with a remarkably low runtime of just $213.09$ seconds, showing that even high $\epsilon$ variants can produce competitive results efficiently. When compared with other well-known algorithms, the \textbf{SG} algorithm secures the highest profit overall at $9791626.15$ but at a much higher computational cost of $24023.84$ seconds and a seed set size of $47$. The \textbf{DG} approach also performs well, earning $9655314.67$ in $15993.44$ seconds using $11$ seeds. In contrast, the stochastic greedy algorithms especially \textbf{StG0.3} and \textbf{StG0.6} achieve nearly comparable profits with significantly fewer seed nodes ($2$ to $4$) and reduced time. These observations highlight the practical trade-off while low $\epsilon$ values like $0.01$ offer better profits due to more deterministic behavior, their high time cost makes them less ideal for large-scale datasets. This is especially relevant for the Slashdot dataset, which is large and would require substantial computation in a two phase setup. In such cases, moderate to high $\epsilon$ values (such as $0.3$ or $0.6$) present a good balance, delivering strong profits with much lower computational effort, making them more suitable for scalable two phase deployment.

\section{Conclusion and Future Research Directions} \label{Sec:Conclusions}
In this paper, we have studied the Profit Maximization Problem in Social Networks in Two phase Setting, where the goal is to split the budget into two halves, and the respective halves will be used for their respective phases. We develop a mathematical model that expresses the expected influence at the end of the second phase. We show that even if we consider that the optimal seed set will be chosen in the second phase selecting the optimal seed set in the first phase for the profit maximization problem is an NP-hard and also hard to approximate beyond $(1-\frac{1}{e})$ factor. To solve this problem in the two phase setting, we have proposed three solution methodologies. These methodologies have been analyzed to understand their time and space requirements. We perform an extensive set of experiments with real-world social network datasets and compare the performance with many existing solution approaches. By analyzing our experimental results, we explain several research questions in and around this problem. Lastly, we believe that this work will give better exposure to upcoming researchers in this area.

\bibliography{sn-bibliography}

\end{document}